\documentclass[a4paper, 10pt, twoside, notitlepage]{book}

\usepackage[usenames,svgnames]{xcolor}
\usepackage{amsmath, amsthm, amssymb, amsfonts, courier, fancyhdr, tikz, changepage, rotating, mathdots, multirow, bigdelim, booktabs, soul, comment, graphicx, lipsum, environ}
\usepackage{hyperref}
\usepackage[font=small, margin=10mm, labelfont=bf]{caption}
\usetikzlibrary{decorations.pathreplacing}
\usetikzlibrary{through,arrows,decorations.markings,patterns,decorations.pathmorphing,backgrounds,positioning,fit,patterns,decorations.text,calc,plotmarks}

\usetikzlibrary{external}
\tikzset{external/system call={lualatex \tikzexternalcheckshellescape -halt-on-error -interaction=batchmode -jobname "\image" "\texsource"}}
\makeatletter
\providecommand{\env@tikzpicture@save@env}{}
\providecommand{\env@tikzpicture@process}{}
\newcounter{tikzfigcntr}
\RenewEnviron{tikzpicture}[1][]{ \includegraphics{FiguresSmallPDF/SpinChains-figure\thetikzfigcntr} \stepcounter{tikzfigcntr}} 
\makeatother
	
\newcommand{\smallgraph}{6.4cm}
\newcommand{\stndgraph}{12.8cm}
\newcommand{\largegraph}{19cm}
\newcommand{\smallgraphmk}{0.05cm}
\newcommand{\smallgraphmks}{0.025cm}
\newcommand{\asratio}{0.75}
\newcommand{\invasratio}{1.333333}
\newcommand{\spacing}{1.2}

\newcommand{\halfgraphwidth}{6.4cm}
\newcommand{\quartergraphwidth}{3.2cm}
\newcommand{\pival}{3.14159}

\tikzset{redmk/.style={mark=*, mark options={scale=0.5}}}
\tikzset{orangemk/.style={mark=square*, mark options={scale=0.5}}}
\tikzset{yellowmk/.style={mark=square*, mark options={scale=0.5, rotate=45}}}
\tikzset{greenmk/.style={mark=triangle*, mark options={scale=0.7}}}
\tikzset{bluemk/.style={mark=triangle*, mark options={scale=0.7, rotate=180}}}
\tikzset{purplemk/.style={mark=x, mark options={scale=0.8}}}

\newcommand{\p}{p.}

\newenvironment{code}
  {\begin{adjustwidth}{9mm}{9mm} \ttfamily \small}
  {\normalfont \normalsize \end{adjustwidth}}
  
\newcommand*\cleartoleftpage{\clearpage\ifodd\value{page}\hbox{}\newpage\fi}

\newenvironment{changemargin}[2]{
	\begin{list}{}{
		\setlength{\topsep}{0pt}
		\setlength{\leftmargin}{#1}
		\setlength{\rightmargin}{#2}
		\setlength{\listparindent}{\parindent}
		\setlength{\itemindent}{\parindent}
		\setlength{\parsep}{\parskip}
	}
	\item[]
}{\end{list}}

\newtheorem{theorem}{Theorem}
\newtheorem{conjecture}{Conjecture}
\newtheorem{corollary}{Corollary}
\newtheorem{lemma}{Lemma}
\newtheorem*{Gaudin*}{Gaudin's lemma}
\newtheorem*{CD*}{Christoffel-Darboux formula}

\numberwithin{equation}{section}

\DeclareMathOperator\Ip{Im}
\DeclareMathOperator\Rp{Re}
\DeclareMathOperator\im{i}
\DeclareMathOperator\jm{j}
\DeclareMathOperator\km{k}
\DeclareMathOperator\e{e}
\DeclareMathOperator\di{d}
\DeclareMathOperator\Tr{Tr}
\DeclareMathOperator\sgn{sgn}

\DeclareMathOperator\iid{i.i.d.}

\newcommand{\histfill}{blue!10}
\definecolor{orange}{RGB}{255,180,0}

\renewcommand{\a}{{\check{a}}}
\renewcommand{\b}{{\check{b}}}
\renewcommand{\c}{{\check{c}}}
\newcommand{\cof}{{c}}
\renewcommand{\r}{{m}}
\newcommand{\s}{{s}}

\title{Quantum Spin Chains and Random Matrix Theory}
\author{Huw J Wells}

\begin{document}

\pagestyle{plain}
\setcounter{page}{1}
\pagenumbering{roman}

\begin{titlepage}
\begin{center}
\vspace*{20mm}
\textsc{\LARGE University of Bristol}\\[1.5cm]
\rule{\linewidth}{0.5mm} \\[0.3cm]{\huge \bfseries Quantum Spin Chains and Random\\[0.2cm]Matrix Theory}\\[0.4cm]
\rule{\linewidth}{0.5mm} \\[1.5cm]
\emph{Author}\\
{\large{\textsc{Huw J. Wells}}}
\\[1.5cm]
\emph{Supervisors}\\
{\large{\textsc{Professor J.P. Keating}}}\\
{\large{\textsc{Professor N. Linden}}}
\vfill {\emph{A dissertation submitted to the University of Bristol in accordance with the requirements for the degree of Doctor of Philosophy in the School of Mathematics}}\\
\vspace{3mm}
March 2014
\end{center}
\end{titlepage}

\cleardoublepage
\chapter*{Abstract}
\begin{changemargin}{10mm}{10mm}
The spectral statistics and entanglement within the eigenstates of generic spin chain Hamiltonians are analysed.

A class of random matrix ensembles is defined which include the most general nearest-neighbour qubit chain Hamiltonians.  For these ensembles, and their generalisations, it is seen that the long chain limiting spectral density is a Gaussian and that this convergence holds on the level of individual Hamiltonians.  The rate of this convergence is numerically seen to be slow.  Higher eigenvalue correlation statistics are also considered, the canonical nearest-neighbour level spacing statistics being numerically observed and linked with ensemble symmetries.  A heuristic argument is given for a conjectured form of the full joint probability density function for the eigenvalues of a wide class of such ensembles.  This is numerically verified in a particular case.

For many translationally-invariant nearest-neighbour qubit Hamiltonians it is shown that there exists a complete orthonormal set of eigenstates for which the entanglement present in a generic member, between a fixed length block of qubits and the rest of the chain, approaches its maximal value as the chain length increases.  Many such Hamiltonians are seen to exhibit a simple spectrum so that their eigenstates are unique up to phase.

The entanglement within the eigenstates contrasts the spectral density for such Hamiltonians, which is that seen for a non-interacting chain of qubits.  For such non-interacting chains, their always exists a basis of eigenstates for which there is no entanglement present.
\end{changemargin}

\cleardoublepage
\chapter*{Acknowledgements}
\begin{changemargin}{10mm}{10mm}
I am very grateful for the support of my friends and family over the last three and a half years.  Without their continuing encouragement and advice I would not have been able to complete this research.

Both my supervisors have patiently guided me through the last few years and I wish to thank them for finding the time to do so from their crowded schedules.  Anna Maltsev has also generously spared her time for many useful discussions which have been both engaging and enlightening.  I would also like to thank Dale Smith for a detailed proof reading of the text in a draft of this dissertation.

Finally I wish to acknowledge the financial support from the Engineering and Physical Sciences Research Council who funded this research.
\end{changemargin}

\cleardoublepage
\chapter*{Author's declaration}
\begin{changemargin}{10mm}{10mm}
I declare that the work in this dissertation was carried out in accordance with the requirements of the University's Regulations and Code of Practice for Research Degree Programmes and that it has not been submitted for any other academic award. Except where indicated by specific reference in the text, the work is the candidate's own work. Work done in collaboration with, or with the assistance of, others is indicated as such. Any views expressed in the dissertation are those of the author.
\\[20pt]
SIGNED \dotfill DATE \dotfill
\end{changemargin}

\cleardoublepage
\chapter*{Notation}
\begin{changemargin}{10mm}{10mm}
\begin{tabular}{p{20mm} p{104mm}}
\bf{Symbol} & \bf{Definition}\\
$\mathbb{N}$ & The positive integers\\
$\mathbb{N}_0$ & The non-negative integers\\
$\mathbb{R}$ & The real numbers\\
$\mathbb{C}$ & The complex numbers\\
$n$ & An integer from $\{2,3,\dots\}$ unless otherwise stated\\
$\mathcal{N}\left(\mu,\sigma^2\right)$ & The normal distribution with mean $\mu$ and variance $\sigma^2$\\
$\mathcal{U}\left(a,b\right)$ & The uniform distribution supported on $[a,b]\subset\mathbb{R}$\\
$\iid$ & Independently and identically distributed\\
$\check{\cdot}$ & Denotes Fermi operators\\
$\hat{\cdot}$ & Denotes quantities related to ensembles\\
$\mathbb{E}(\cdot)$ & The expectation of a random variable or random variables\\
$\langle\cdot\rangle$ & The integral with respect to the measure associated to the parameters in the argument\\
$\cdot^T$ & The transpose of a vector or matrix\\
$\overline{\cdot}$ & The complex conjugate of a complex scalar, vector or matrix in the standard basis\\
$\cdot^\dagger$ & The complex conjugate transpose of a vector or matrix\\
$|\cdot\rangle$ & Dirac notation for a column vector\\
$\langle\cdot|$ & Complex conjugate transpose of $|\cdot\rangle$\\
$I_{n}$ & The $n\times n$ identity matrix\\
$I$ & The identity operator\\
GOE & Gaussian orthogonal ensemble\\
GUE & Gaussian unitary ensemble\\
GSE & Gaussian symplectic ensemble
\end{tabular}
\end{changemargin}

\cleardoublepage

\pagestyle{fancy}
\renewcommand{\sectionmark}[1]{\markboth{\MakeUppercase{\thechapter \,\, #1}}{}}
\renewcommand{\sectionmark}[1]{\markright{\thesection \,\, #1}}
\fancyhead{}
\fancyfoot{}
\fancyhead[RO]{\rightmark}
\fancyhead[LE]{\leftmark}
\fancyfoot[C]{\thepage}

\tableofcontents 
\listoffigures
\cleardoublepage
\setcounter{page}{1}
\pagenumbering{arabic}
 
\chapter{Introduction}\label{Introduction}
\section{Solvable quantum spin chains}
\subsection{Motivation}
In 1911, Niels Bohr proved that
\begin{quote}
  \emph{``At any finite temperature, and in all finite applied electrical or magnetic fields, the net magnetisation of a collection of [classical non-relativistic] electrons in thermal equilibrium vanishes identically.''} \hfill\cite[\p21]{Mattis1988}
\end{quote}
Hendrika Johanna van Leeuwen also independently discovered this fact in 1919 and today it is known as the Bohr-van Leeuwen theorem.  In particular this theorem does not allow for ferromagnetism, the underlying mechanics by which certain materials, such as iron, can form permamagnents.

The quantum mechanical spin chain was simultaneously used by Dirac and Heisenberg to address this problem.  The Heisenberg (anti)-ferromagnet is such a quantum mechanical model describing a line of quantum spins, for example distinguishable electrons, interacting with their neighbours.  Varying the single parameter in this model allows for a sharp phase transition in its ground (lowest energy) state from a ferromagnetic state (all spins aligned) to an anti-ferromagnetic state (in which spins tend to anti-align with their neighbours) \cite[\p138]{Mattis1988}, see Figure \ref{Align}.

\begin{figure}
\centering
\input{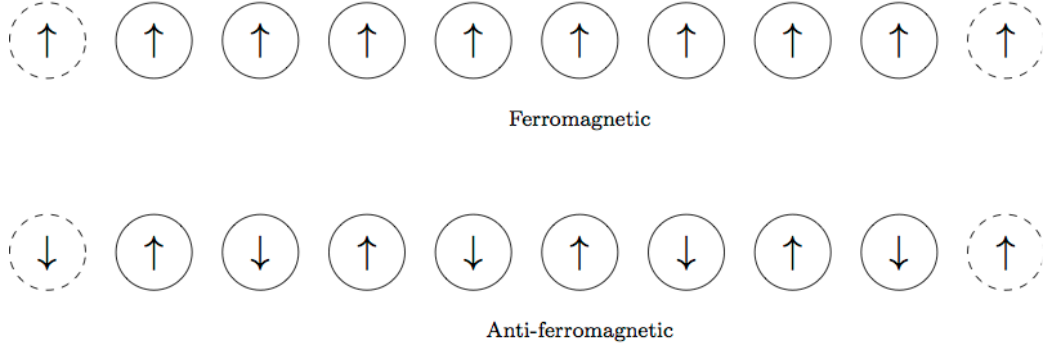}
\caption[Ferromagnetic and anti-ferromagnetic states]{A section of ferromagnetic and anti-ferromagnetic  states of a quantum spin chain.  The circles represent the quantum spins arranged in a line and the arrows their `spin direction'.  The ferromagnetic state (top) has all the spins aligned and the anti-ferromagnetic state (bottom) has all the spins anti-aligned.}
\label{Align}
\end{figure}

Quantum mechanical spin chains also have a wealth of other applications in physics.  For example, they can be used to transfer quantum states, from one end of a chain to the other, with high fidelity  \cite{Bose2003,Linden2004}.  They have been theoretically used as a model for a quantum computer, fault tolerant universal gates being implemented on sections of a chain \cite{Loss2011}.  It has also been suggested that the entanglement in the ground state of some Heisenberg quantum spin chains can be used as a means to generate entangled pairs of quantum particles \cite{Mohktari2012}.

Little is yet known about the eigenvalues and eigenstates of generic quantum spin chains.  The work presented here sheds light on the distribution of the eigenvalues of generic quantum spin chain Hamiltonians, and the amount of bipartite entanglement present in their eigenstates.

\subsection{Quantum mechanical background}\label{Quantum}
A detailed description of the quantum mechanics of finite dimension quantum systems, needed to describe such spin chains, is given in \cite{Chuang2000}.  The key concepts that will be required here are summarised as follows:

\subsubsection{Tensor product}
The tensor product \cite[\p71]{Chuang2000}, denoted $\otimes$, is an important tool in the description of quantum spin chains.  It is defined as the map $\mathbb{C}^{r_{a}\times c_{a}}\times\mathbb{C}^{r_{b}\times c_{b}}\longrightarrow\mathbb{C}^{r_{a}r_{b},c_{a}c_{b}}$ denoted $\left(A,B\right)\longmapsto A\otimes B$ where
\begin{equation}
  (A\otimes B)_{r_{b}(j-1)+l,c_{b}(k-1)+m}=A_{j,k}B_{l,m}
\end{equation}
for the matrix indices $j$, $k$, $l$ and $m$ ranging from $1$ to $r_{a}$, $c_{a}$, $r_{b}$ and $c_{b}$ respectively.

For any matrices (including row and column vectors) $A,A^\prime\in\mathbb{C}^{r_{a}\times c_{a}}$, $B,B^\prime\in\mathbb{C}^{r_{b}\times c_{b}}$, $C\in\mathbb{C}^{r_{c}\times c_{c}}$ and $D\in\mathbb{C}^{r_{d}\times c_{d}}$ with $c_{a}=r_{c}$ and $c_{b}=r_{d}$, and complex number $z$, the tensor product satisfies the conditions
\begin{align}
  \Big(A+A^\prime\Big)\otimes B&=A\otimes B+A^\prime\otimes B\nonumber\\
  A\otimes\Big(B+B^\prime\Big)&=A\otimes B+A \otimes B^\prime\nonumber\\
  z\Big(A\otimes B\Big)&=\Big(zA\Big)\otimes B=A\otimes\Big(zB\Big)\nonumber\\
  (A\otimes B)(C\otimes D)&=(AC)\otimes(BD)\nonumber\\
  (A\otimes B)^\dagger&=A^\dagger\otimes B^\dagger
\end{align}

Furthermore, by the definition of the tensor product, the trace of $A\otimes B$ satisfies
\begin{equation}
  \Tr\Big( A\otimes B\Big) = \sum_{j}\Big(A\otimes B\Big)_{j,j}=\sum_{j,k}A_{j,j}B_{k,k}=\Tr\left( A\right)\Tr \left(B\right)
\end{equation}
for the square matrices $A$ and $B$.

Additionally, let $\mathcal{A}=\mathbb{C}^{d_{\mathcal{A}}\times 1}\equiv\mathbb{C}^{d_{\mathcal{A}}}$ and $\mathcal{B}=\mathbb{C}^{d_{\mathcal{B}}\times 1}\equiv\mathbb{C}^{d_{\mathcal{B}}}$ be two Hilbert spaces (endowed with the standard complex Euclidean inner-product) with orthonormal bases $\{|a\rangle\}_{a=1}^{d_\mathcal{A}}$ and $\{|b\rangle\}_{b=1}^{d_\mathcal{B}}$ respectively.  The Hilbert space $\mathcal{A}\otimes\mathcal{B}$ is defined to have the basis $\{ |a\rangle\otimes|b\rangle\}_{a,b}$ (with elements equivalently written as $|a\rangle|b\rangle$), and therefore $\mathcal{A}\otimes\mathcal{B}=\mathbb{C}^{d_{\mathcal{A}}}\otimes\mathbb{C}^{d_{\mathcal{B}}}\equiv\mathbb{C}^{d_{\mathcal{A}}d_{\mathcal{B}}}$.

\subsubsection{Quantum states}
Finite length quantum spin chains are finite dimensional quantum systems made up of many individual, distinguishable, finite dimensional quantum spin particles.  The state of an individual $d\in\mathbb{N}$ dimensional quantum system (for example a single spin particle) is defined as a normalised column vector $|\psi\rangle$ in the Hilbert space $\mathbb{C}^{d}\equiv\mathbb{C}^{d\times 1}$ (endowed with the standard complex Euclidean inner-product).  If the spin chain is made of $n\in\mathbb{N}$ distinguishable spin particles and the $j^{th}$ particle has the corresponding Hilbert space $\mathcal{H}_{j}$, then by a postulate of quantum mechanics the state of the whole chain is defined as a normalised column vector $|\phi\rangle$ in the Hilbert space
\begin{equation}
  \mathcal{H} = \bigotimes_{j=1}^{n}\mathcal{H}_{j}
\end{equation}
of all $n$ spin particles \cite[\p80]{Chuang2000}.

In this work, chains formed from distinguishable spin-half particles (qubits), will be the main focus of study.  The associated Hilbert space of each qubit is $\mathbb{C}^{2}$ and therefore the associated Hilbert space of $n$ qubits is $\left(\mathbb{C}^{2}\right)^{\otimes n}\equiv\mathbb{C}^{2^n}$.

\subsubsection{Partial trace}
A key tool for looking at subsystems of quantum systems is the partial trace \cite[\p107]{Chuang2000}.  If the Hilbert space $\mathcal{H}$ of a quantum system can be written as the tensor product of two smaller spaces $\mathcal{A}$ and $\mathcal{B}$ each with an orthonormal basis $\{|a\rangle_\mathcal{A}\}_{a=1}^{d_\mathcal{A}}$ and $\{|b\rangle_\mathcal{B}\}_{b=1}^{d_\mathcal{B}}$ respectively, then any  operator $M$ (a complex $d_\mathcal{A}d_\mathcal{B}\times d_\mathcal{A}d_\mathcal{B}$ matrix) acting on $\mathcal{H}=\mathcal{A}\otimes\mathcal{B}$ has the form
\begin{equation}
  M=\sum_{a,b}\sum_{a^\prime,b^\prime}c_{a,b,a^\prime,b^\prime}\Big(|a\rangle_{\mathcal{A}}|b\rangle_{\mathcal{B}}\Big)\Big(\,_{\mathcal{A}}\langle a^\prime|\,_{\mathcal{B}}\langle b^\prime|\Big)
\end{equation}
where $c_{a,b,a^\prime,b^\prime}$ are some complex coefficients, as the $|a\rangle_{\mathcal{A}}|b\rangle_{\mathcal{B}}$ form an orthonormal basis of $\mathcal{H}$.  The partial trace over $\mathcal{B}$ of the matrix $M$ acting on $\mathcal{H}$ is then defined to be
\begin{equation}
  \Tr_{\mathcal{B}}\left(M\right)=\sum_{a,b}\sum_{a^\prime,b^\prime}c_{a,b,a^\prime,b^\prime}|a\rangle_\mathcal{A}\,_\mathcal{A}\langle a^\prime|\Tr\Big(|b\rangle_\mathcal{B}\,_\mathcal{B}\langle b^\prime|\Big)
\end{equation}
This is a matrix acting on $\mathcal{A}$.

\subsubsection{Schmidt decomposition}\label{Schmidt}
The Schmidt decomposition \cite[\p109]{Chuang2000} allows states, in a Hilbert space of the form $\mathcal{A}\otimes\mathcal{B}$, of some quantum system to be decomposed in a minimal way over some basis elements of $\mathcal{A}$ and $\mathcal{B}$.  It asserts that for any state $|\phi\rangle\in\mathcal{A}\otimes\mathcal{B}$ there exists an orthonormal basis $\{|a\rangle_\mathcal{A}\}_{a=1}^{d_\mathcal{A}}$ of $\mathcal{A}$ and $\{|b\rangle_\mathcal{B}\}_{b=1}^{d_\mathcal{B}}$ of $\mathcal{B}$ and scalars $0\leq s_{j}\leq1$ with
\begin{equation}
  \sum_{j=1}^{\min(d_{\mathcal{A}},d_{\mathcal{B}})}s_{j}^{2}=1
\end{equation}
such that 
\begin{equation}
  |\phi\rangle=\sum_{j=1}^{\min(d_{\mathcal{A}},d_{\mathcal{B}})}s_{j}|j\rangle_\mathcal{A}|j\rangle_\mathcal{B}
\end{equation}

\subsubsection{Superposition}
Given two orthogonal quantum states, $|0\rangle$ and $|1\rangle$, the state
\begin{equation}
  \frac{c_{0}|0\rangle+c_{1}|1\rangle}{\sqrt{|c_{0}|^{2}+|c_{1}|^{2}}}
\end{equation}
where $c_{0},c_{1}\in\mathbb{C}$, is their superposition (a further normalised state in the Hilbert space). This definition extends to more than two states.  In particular, for any orthonormal basis of the Hilbert space, an arbitrary state $|\phi\rangle$ may be considered as a superposition over all of the basis states \cite[\p81]{Chuang2000}.

\subsubsection{Quantum Hamiltonians}
The Hamiltonian of a finite dimensional quantum system, with corresponding Hilbert space $\mathcal{H}$, is a Hermitian matrix $H$, acting on $\mathcal{H}$, describing the energy of the system \cite[\p83]{Chuang2000}.  An example of such a Hamiltonian is seen in Section \ref{1.1.5} and the physical relevance of such a matrix (measurement) is seen in next subsection.

As $H$ is a Hermitian matrix (say with dimension $N$) there exist $N$ eigenstates of $H$, $|\psi_{k}\rangle$ for $k=1,\dots,N$, which are orthonormal and form a basis of the Hilbert space $\mathcal{H}$,  with corresponding real eigenvalues $\lambda_{k}$ so that
\begin{equation}
  H=\sum_{k=1}^{N}\lambda_{k}|\psi_{k}\rangle\langle\psi_{k}|
\end{equation}

As the $|\psi_{k}\rangle$ form an orthonormal basis of $\mathcal{H}$, any state $|\phi\rangle$ of the system may be written in this basis as
\begin{equation}
  |\phi\rangle=\sum_{k=1}^{N}c_{k}|\psi_{k}\rangle
\end{equation}
for some complex coefficients $c_{k}$, where $\sum_{k}|c_{k}|^{2}=1$ as $|\phi\rangle$ is normalised.

\subsubsection{Measurement}
Given that a quantum system is in some state $|\phi\rangle$, any Hermitian matrix $H$, acting on the system's Hilbert space, corresponds to some physical measurement of the system.  For the Hamiltonian this is the system's energy.  One of the postulates of quantum mechanics states that the outcomes of a physical measurement can only be one of the eigenvalues of the matrix $H$.  The probability of observing the eigenvalue $\lambda$ is $\langle\phi|P_{\lambda}|\phi\rangle$ where $P_{\lambda}$ is the projector onto the subspace spanned by eigenstates of $H$ with the eigenvalue $\lambda$, that is for orthonormal eigenstates $|\psi_{k}\rangle$ of $H$,
\begin{equation}
  P_{\lambda}=\sum_{|\psi_{k}\rangle:H|\psi_{k}\rangle=\lambda|\psi_{k}\rangle}|\psi_{k}\rangle\langle\psi_{k}|
\end{equation}
After the measurement a further postulate asserts that the state of the system collapses to the normalised projection of that state onto the eigenspace corresponding to the eigenvalue observed \cite[\p84]{Chuang2000}.

\subsubsection{Density matrix}
Alternatively, the state of a quantum system can be represented by the density matrix $\rho=|\phi\rangle\langle\phi|$ rather than the normalised vector $|\phi\rangle$.  By definition, $\rho=\rho^\dagger$, and so $\rho$ is a Hermitian matrix.  Measurement outcome probabilities are then equivalently given by $\Tr\left(\rho P_{\lambda}\right)$ \cite[\p99]{Chuang2000}.

\subsubsection{Classical mixtures}
Moreover, the density matrix description allows for classical ensembles of quantum systems.  If a system is either in state $\rho_{1}$ or $\rho_{2}$ with probability $p_{1}$ and $p_{2}$ respectively (with $p_{1}+p_{2}=1$) then the classical mixture
\begin{equation}
 \rho=p_{1}\rho_{1}+p_{2}\rho_{2}
\end{equation}
provides a description of this, so that the measurement probabilities $\Tr\left(\rho P_{\lambda}\right)$ are algebraically accurate.  This interpretation generalises to any number of states and classical probabilities.  From this definition it follows that a Hermitian matrix $\rho$ is a valid density matrix if and only if $\Tr\left(\rho\right)=1$ and $\langle\phi|\rho|\phi\rangle\geq0$ for all normalised states $|\phi\rangle$ \cite[\p99]{Chuang2000}.

A pure state is defined to be one that can be written in the form $\rho=|\phi\rangle\langle\phi|$ for some normalised vector $|\phi\rangle$ in the Hilbert space, if this is not the case then $\rho$ is called a mixed state.  

In particular, for the orthonormal basis of $\mathcal{H}$ consisting of the elements $|\psi_{k}\rangle$ for $k=1,\dots,N$, the state
\begin{equation}
  \sum_{k=1}^{N}\frac{1}{N}|\psi_{k}\rangle\langle\psi_{k}|
\end{equation}
is called maximally mixed.  That is, it is the classical mixture of the most orthogonal states possible with as little information as possible known about which state is present (see Section \ref{entropy} on entropy).

\subsubsection{Reduced density matrix}
Consider a quantum system formed of two subsystems $A$ and $B$ with Hilbert spaces $\mathcal{A}$ and $\mathcal{B}$ respectively, so that the Hilbert space of the full system is $\mathcal{H}=\mathcal{A}\otimes\mathcal{B}$.  If the whole system is in the state $\rho$ then the reduced state $\rho_\mathcal{A}$, of subsystem $A$ on its own, is defined to be
\begin{equation}
  \rho_\mathcal{A}=\Tr_\mathcal{B}\left(\rho\right)
\end{equation}
By applying the Schmidt decomposition to an expansion of $\rho$, over pure states and probabilities, it is seen that $\rho_{\mathcal{A}}$ remains a valid density matrix.  Moreover, $\rho_\mathcal{A}$ provides the correct mathematical description (measurement outcomes) of the state of subsystem $A$ given that no access is available to subsystem $B$ and subsystem $B$ remains undisturbed \cite[\p105]{Chuang2000}.

\subsection{The standard basis}\label{stndBasis}
Let the standard basis for $\mathbb{C}^{2}$ be denoted by
\begin{equation}
	|0\rangle=\begin{pmatrix}1\\0\end{pmatrix}\qquad\qquad
	|1\rangle=\begin{pmatrix}0\\1\end{pmatrix}
\end{equation}
so that for any vector $|\phi\rangle\in\mathbb{C}^{2}$ there exist complex coefficients $c_{0},c_{1}\in\mathbb{C}$ such that
\begin{equation}
	|\phi\rangle=c_{0}|0\rangle+c_{1}|1\rangle
\end{equation}
Define the standard basis for $\left(\mathbb{C}^{2}\right)^{\otimes n}$, where $n\in\mathbb{N}$, to be the vectors
\begin{equation}
	|\boldsymbol{x}\rangle=|x_{1}\rangle\otimes\dots\otimes|x_{n}\rangle
\end{equation}
for the multi-indices $\boldsymbol{x}=(x_{1},\dots,x_{n})\in\{0,1\}^{n}$.

\subsection{Pauli matrix basis}\label{PauliMatrixBasis}
The Pauli matrices \cite[\p65]{Chuang2000} provide a framework for describing quantum mechanical spin chain Hamiltonians, they are defined (in the standard basis) as
\begin{equation}
  \sigma^{(0)}=I_{2}=\begin{pmatrix}1&0\\0&1\end{pmatrix}\qquad
  \sigma^{(1)}=\begin{pmatrix}0&1\\1&0\end{pmatrix}\qquad
  \sigma^{(2)}=\begin{pmatrix}0&-\im\\\im&0\end{pmatrix}\qquad
  \sigma^{(3)}=\begin{pmatrix}1&0\\0&-1\end{pmatrix}
\end{equation}
By direct computation they satisfy
\begin{alignat}{3}
  \Tr\Big(\sigma^{(a)}\sigma^{(b)}\Big)&=2\delta_{a,b}\,\,\,\qquad\qquad\qquad &&\text{for }a,b=0,1,2,3\nonumber\\
  \sigma^{(a)}\sigma^{(b)}&=-\sigma^{(b)}\sigma^{(a)}\,\qquad\qquad  &&\text{for }a\neq b,\,\,a,b=1,2,3\nonumber\\
  \sigma^{(a)}\sigma^{(a)}&=I_{2}\,\qquad\qquad  &&\text{for }a=0,1,2,3\nonumber\\
  \sigma^{(1)}\sigma^{(2)}\sigma^{(3)}&=\im I_{2}\qquad\qquad\qquad
\end{alignat}
where $\delta_{a,b}$ is the Kronecker delta symbol.

\subsubsection{Parametrisation of $2^{n}\times 2^{n}$ Hermitian matrices}
The space of $2^{n}\times 2^{n}$ complex matrices admits the (scaled) Hilbert-Schmidt inner-product \cite[\p76]{Chuang2000}
\begin{equation}
  (X,Y) = \frac{1}{2^{n}}\Tr(XY^\dagger )
\end{equation}
as seen in Appendix \ref{Norms}. The $4^{n}$, $2^{n}\times2^{n}$ Hermitian matrices
\begin{equation}
  P_{\boldsymbol{a}}=\sigma^{( a_{1} )}\otimes\dots\otimes\sigma^{( a_{n} )}
\end{equation}
for $\boldsymbol{a}=(a_{1},\dots,a_{n})\in\{0,1,2,3\}^{n}$ form an orthonormal basis of the $2^{n}\times 2^{n}$ Hermitian matrices as
\begin{align}
  (P_{\boldsymbol{a}},P_{\boldsymbol{b}})&=\frac{1}{2^{n}}\Tr\Big(\Big(\sigma^{( a_{1} )}\otimes\dots\otimes\sigma^{( a_{n} )}\Big)\Big(\sigma^{( b_{1} )}\otimes\dots\otimes\sigma^{( b_{n} )}\Big)^\dagger\Big)\nonumber\\
  &=\frac{1}{2}\Tr\Big(\sigma^{( a_{1} )}\sigma^{(b_{1})}\Big)\dots\frac{1}{2}\Tr\Big(\sigma^{( a_{n} )}\sigma^{(b_{n})}\Big)
\end{align}
by applying the properties of the tensor product.  Then by the properties of the Pauli matrices, $(P_{\boldsymbol{a}},P_{\boldsymbol{b}})=\delta_{a_{1},b_{1}}\dots\delta_{a_{n},b_{n}}=\delta_{\boldsymbol{a},\boldsymbol{b}}$.  Therefore it must be concluded that every $2^{n}\times2^{n}$ Hermitian matrix $H$ (which has $4^{n}$ real parameters) can be written in this basis, that is
\begin{equation}
  H=\sum_{a_{1},\dots,a_{n}=0}^{3}c_{\boldsymbol{a}}P_{\boldsymbol{a}}
\end{equation}
where $c_{\boldsymbol{a}}=(H,P_{\boldsymbol{a}})$ and $c_{\boldsymbol{a}}\in\mathbb{R}$ as $H=H^\dagger$.

\subsubsection{Notation}\label{notation}
For the description of spin chain Hamiltonians, the following notation will be adopted:
\begin{equation}
  \sigma_{  j  }^{(a)}\equiv I_{2}^{\otimes (j-1)}\otimes\sigma^{(a)}\otimes I_{2}^{\otimes n-j}
\end{equation}
This matrix acts on the Hilbert space of $n$ qubits, $\left(\mathbb{C}^{2}\right)^{\otimes{n}}$, acting on the $j^{th}$ qubit (that is the $j^{th}$ tensor factor in $\left(\mathbb{C}^{2}\right)^{\otimes{n}}$) as $\sigma^{(a)}$ for $a=0,1,2,3$ and acting as the identity on the remaining qubits.  Cyclic boundary conditions are taken so that $\sigma_{n+j}^{(a)}$ is identified with $\sigma_{j}^{(a)}$.

\subsection{Quantum spin chains}\label{1.1.5}
Quantum spin chains are a collection of distinguishable quantum particles arranged in a line or ring where only neighbouring particles are allowed to interact.  Higher dimensional analogous include interactions on lattices or more complicated geometries.

The simplest case is that of $n$ distinguishable qubits labelled $1$ to $n$ where qubit $j$ is only allowed to interact with qubits $j\pm1$ (cyclically).  The associated Hilbert space is then $\left(\mathbb{C}^{2}\right)^{\otimes n}$, the $n$ fold tensor product of the individual qubit Hilbert spaces.   It will be seen in Section \ref{Random matrix model} that the most general Hamiltonian for such a system can be written as
\begin{equation}
	H_{n}^{(gen)}=\sum_{j=1}^{n}\sum_{a,b=1}^{3}\alpha_{a,b,j}\sigma_{  j  }^{(a)}\sigma_{  j+1  }^{(b)}+\sum_{j=1}^{n}\sum_{a=1}^{3}\alpha_{a,0,j}\sigma_{  j  }^{(a)}+\alpha_{0,0,1} I_{2^{n}}
\end{equation}
for some real coefficients $\alpha_{a,b,j}$.  Figure \ref{chain} gives a graphical representation of such a chain.

Two well studied chains are the $XY$ and $XYZ$ models with the Hamiltonians
\begin{align}
	H_{n}^{(XY)}&=\frac{J}{2}\sum_{j=1}^{n}\left((1+\gamma)\sigma_{j}^{(1)}\sigma_{j+1}^{(1)}+(1-\gamma)\sigma_{j}^{(2)}\sigma_{j+1}^{(2)}\right)+\frac{h}{2}\sum_{j=1}^{n}\sigma_{j}^{(3)}\nonumber\\
	H_{n}^{(XYZ)}&=\frac{1}{2}\sum_{j=1}^{n}\left(J_{x}\sigma_{j}^{(1)}\sigma_{j+1}^{(1)}+J_{y}\sigma_{j}^{(2)}\sigma_{j+1}^{(2)}+J_{z}\sigma_{j}^{(3)}\sigma_{j+1}^{(3)}\right)+\frac{h}{2}\sum_{j=1}^{n}\sigma_{j}^{(3)}
\end{align}
respectively for the real coefficients $\gamma$ (anisotropy parameter), $h$ (relating to an external magnetic field) and $J$, $J_{x}$, $J_{y}$ and $J_{z}$ (coupling parameters).  The boundary terms $\sigma_{n}^{(a)}\sigma_{1}^{(a)}$ are omitted in some definitions.  The $XY$ model reduces to the $XX$ model when $\gamma=0$ and the Ising model when $\gamma=1$.  The $XYZ$ model reduces to the $XXZ$ model when $J_{x}=J_{y}=J$ and $J_{z}=\Delta J$ for some real coefficient $\Delta$ and the $XXX$ model when $\Delta=1$.  Examples of these models, and references thereof, will be seen in the subsequent sections.

\begin{figure}
	\centering
	\input{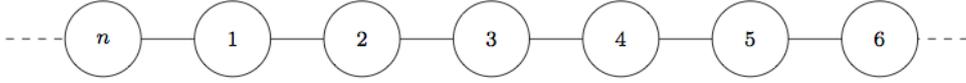}
	\caption[Graphical representation of a spin chain]{A graphical representation of the spin chain Hamiltonian $H_{n}^{(gen)}$.  The circles represent the qubits labelled $1$ to $n$ and the links the interactions terms $\sum_{a,b=1}^{3}\alpha_{a,b,j}\sigma_{j}^{(a)}\sigma_{j+1}^{(b)}$ which act only on the neighbouring qubits labelled $j$ and $j+1$.  The local terms $\sum_{a=1}^{3}\alpha_{a,0,j}\sigma_{  j  }^{(a)}$ are associated to the qubit labelled $j$.  The identity term represents a global energy shift.}
	\label{chain}
\end{figure}

\subsection{Jordan-Wigner transform}\label{JW}
Fermi operators are a collection of operators $\a_{1},\dots,\a_{n}$ that satisfy the canonical commutation relations
\begin{align}
	\a_{j}\a_{k}^{\dagger}&=-\a_{k}^{\dagger}\a_{j}+\delta_{j,k}I\nonumber\\
	\a_{j}\a_{k}&=-\a_{k}\a_{j}
\end{align}
Nielsen's notes \cite{Nielsen2005} give an excellent introduction to these operators and the Jordan-Wigner transform.

The spin chain Hamiltonian
\begin{equation}
  H_{n}^{(JW)}=\sum_{j=1}^{n-1}\sum_{a,b=1}^{2}\alpha_{a,b,j}\sigma_{j}^{(a)}\sigma_{j+1}^{(b)}+\sum_{j=1}^{n}\alpha_{3,0,j}\sigma_{j}^{(3)}
\end{equation}
can be exactly solved (diagonalised) with the Jordan-Wigner transform \cite{Nielsen2005}, which maps the Pauli matrices onto Fermi annihilation and creation operators, see Appendix \ref{JWT}

This model is seen to contain the $XY$, $XX$ and Ising models.  The procedure is explicitly used in Section \ref{DOSrate}.

\subsection{Bethe ansatz}\label{Bethe}
The Heisenberg $XXZ$ chain is an example of a quantum spin chain Hamiltonian that cannot be solved (diagonalised) by the standard Jordan-Wigner transform method.  It is defined by the Hamiltonian
\begin{equation}
  H_{n}^{(XXZ)}=\frac{J}{2}\sum_{j=1}^{n}\left(\sigma_{  j  }^{(1)}\sigma_{  j+1  }^{(1)}+\sigma_{  j  }^{(2)}\sigma_{  j+1  }^{(2)}+\Delta\sigma_{  j  }^{(3)}\sigma_{  j+1  }^{(3)}\right)+\frac{h}{2}\sum_{j=1}^{n}\sigma_{  j  }^{(3)}
\end{equation}
Here $\Delta$ and $J$ are real coupling coefficients and $h\in\mathbb{R}$ gives the external magnetic field strength.  The sign of the constant $J$ defines the ferromagnetic and anti-ferromagnetic Heisenberg chain as mentioned in the introduction.

The $XXX$ model with Hamiltonian $H_{n}^{(XXX)}$ is recovered if $\Delta=1$.  This model can be exactly solved with the Bethe ansatz \cite{Muller2004}, due to Bethe in 1931.  An outline of this process will be given here for $h=0$, as the symmetry subspaces used will be useful later on.

\subsubsection{Symmetry}
The matrix
\begin{equation}
  H_{n}^{(Z)}=\sum_{j=1}^{n}\sigma_{  j  }^{(3)}
\end{equation}
is diagonal in the standard basis (Section \ref{stndBasis}) as by definition
\begin{equation}
	H_{n}^{(Z)}|\boldsymbol{x}\rangle
	=\sum_{j=1}^{n}(-1)^{x_{j}}|\boldsymbol{x}\rangle
	=\sum_{j=1}^{n}(1-2x_{j})|\boldsymbol{x}\rangle
	=(n-2\s)|\boldsymbol{x}\rangle
\end{equation}
where $\s=\sum_{j=1}^{n}x_{j}$.   By considering the action on the standard basis, the Hamiltonian $H_{n}^{(XXX)}$ is seen to be block diagonal in the eigenbasis of $H_{n}^{(Z)}$, each block corresponding to one of the eigenvalues $n-2\s$, for $\s=0,\dots,n$, of $H_{n}^{(Z)}$.  This symmetry and the translational symmetry along the chain are the two key properties of the model that allow the Bethe ansatz to apply.  The diagonalisation now proceeds within the eigensubspaces labelled by $\s$:

\subsubsection{The $\s=0$ and $\s=n$ subspaces}
The eigenspace for $\s=0$ contains only one eigenstate, $|\boldsymbol{0}\rangle=|0\rangle\dots|0\rangle$.  Likewise the eigenspace for $\s=n$ contains only $|\boldsymbol{1}\rangle=|1\rangle\dots|1\rangle$.

\subsubsection{The general $\s$ subspaces}
The eigenspace corresponding to the eigenvalue $n-2\s$ of $H_{n}^{(Z)}$ for some fixed value of $\s$ contains $\genfrac{(}{)}{0pt}{}{n}{\s}$ eigenstates of $H_{n}^{(Z)}$, given by the standard basis elements $|\boldsymbol{y}\rangle$ for which $\sum_{j=1}^{n}y_{j}=\s$.  These states can be represented by the vector $|n_{1},\dots,n_\s\rangle$ where $n_{1},\dots,n_\s$ are the distinct positions of the $\s$ factors $|1\rangle$ in $|\boldsymbol{y}\rangle=|y_{1}\rangle\otimes\dots\otimes|y_{n}\rangle$.

An eigenstate of $H_{n}^{(XXX)}$ in this subspace then has the form
\begin{equation}
  |\psi\rangle=\sum_{\genfrac{}{}{0pt}{}{n_{1},\dots,n_{\s}=1}{\text{distinct}}}^{n}c_{n_{1},\dots,n_\s}|n_{1},\dots,n_\s\rangle
\end{equation}
for some complex coefficients $c_{n_{1},\dots,n_\s}$, as the vectors $|n_{1},\dots,n_\s\rangle$ span it.  The (co-ordinate) Bethe ansatz for the solution is given by
\begin{equation}
 c_{n_{1},\dots,n_\s}=\sum_{\tau\in \mathcal{S}_\s}\e^{\im\sum_{j=1}^{\s}\xi_{\tau(j)}n_{j}+\frac{i}{2}\sum_{1\leq j<k\leq \s}\theta_{\tau(j),\tau(k)}}
\end{equation}
where $\mathcal{S}_\s$ is the permutation group on $\s$ elements.  For $|\phi\rangle$ to be eigenstates of $H_{n}^{(XXX)}$ it can be shown that the momenta of the Bethe ansatz $\xi_{j}$ and the phase angles $\theta_{j,k}$ simultaneously satisfy the $\s(\s+1)$ equations
\begin{align}\label{Betheeq}
  2\cot\frac{\theta_{j,k}}{2}&=\cot\frac{\xi_{j}}{2}-\cot\frac{\xi_{k}}{2}\qquad\text{for }j,k=1,2,\dots,\s\nonumber\\
 n\xi_{j}&=2\pi\lambda_{j}+\sum_{j\neq k}\theta_{j,k}\qquad\text{for }j=1,\dots,\s
\end{align}
for some fixed values of $\lambda_{1},\lambda_{2},\dots,\lambda_\s=0,\dots,n-1$ such that a solution exists.  The Bethe quantum numbers $\lambda_{1},\dots,\lambda_\s$ then index the complete set of solutions.

Finding solutions of (\ref{Betheeq}) for given Bethe quantum numbers in principle enables a complete description of the eigenstates, and therefore eigenvalues, of $H_{n}^{(XXX)}$.  This is still a non-trivial problem though and in Karhach's and M\"uller's view
\begin{quote}
  \emph{``the eigenvalues and eigenvectors for a finite dimensional system may be obtained with less effort from a brute force numerical diagonalisation''} \hfill\cite{Muller2004}
\end{quote}
although they do list the general advantages that the Bethe ansatz provides (calculation in the large chain limit for example).
\section{Entanglement}\label{etang}
Entanglement is one of the most striking features of quantum mechanics and an outline is given in \cite[\p95]{Chuang2000}. 

Consider a composite quantum system, of two subsystems $A$ and $B$, with the Hilbert space $\mathcal{H}=\mathcal{A}\otimes\mathcal{B}$.  Let $\{|a\rangle\}_{a=1}^{d_{\mathcal{A}}}$ and $\{|b\rangle\}_{b=1}^{d_{\mathcal{B}}}$ be orthonormal bases for the Hilbert spaces $\mathcal{A}$ of subsystem $A$ and $\mathcal{B}$ of subsystem $B$ respectively.  Any state $|\phi\rangle$ of the composite system may then be written as
\begin{equation}
  |\phi\rangle=\sum_{a,b}c_{a,b}|a\rangle|b\rangle
\end{equation}
for some complex coefficients $c_{a,b}$ such that $\sum_{a,b}|c_{a,b}|^{2}=1$.

Such a quantum state $|\phi\rangle$ is called a product state across systems $A$ and $B$ if it can be written as
\begin{equation}
  |\phi\rangle=|\phi_{A}\rangle\otimes|\phi_{B}\rangle
\end{equation}
where $|\phi_{A}\rangle\in\mathcal{A}$ is some state of subsystem $A$ and $|\phi_{B}\rangle\in\mathcal{B}$ is some state of subsystem $B$.  If this is not the case then the quantum state $|\phi\rangle$ is called entangled across subsystems $A$ and $B$.

In terms of density matrices, separable quantum states $\rho$ across subsystems $A$ and $B$ are those that can be written as
\begin{equation}
  \rho=\sum_{j}p_{j}\rho_{\mathcal{A}}\otimes\rho_{\mathcal{B}}
\end{equation}
where $\rho_{\mathcal{A}}$ and $\rho_{\mathcal{B}}$ are some density matrices of subsystems $A$ and $B$ respectively and the $p_{j}$ are probabilities such that $\sum_{j}p_{j}=1$.  If this is not the case then the quantum state $\rho$ is called entangled across subsystems $A$ and $B$.  This can be thought of as a generalisation of the vector description.  Here separable states are classical mixtures of pure product states.

If $\mathcal{A}=\mathcal{B}=\mathbb{C}^{2}$, for example the case of two qubits, the singlet state \cite[\p95]{Chuang2000}
\begin{equation}
  \frac{|0\rangle|1\rangle-|1\rangle|0\rangle}{\sqrt{2}}
\end{equation}
where $|0\rangle$ and $|1\rangle$ form the standard basis of $\mathbb{C}^{2}$, is an example of an entangled state between the two qubits.

\subsection{Applications}\label{app}
The phenomenon of quantum entanglement is responsible for some highly `non-classical' effects in quantum mechanics.  For example, in 1992 Bennett and Wiesner constructed the super-dense coding protocol \cite[\p97]{Chuang2000}.  Here, two parties, say Alice and Bob, each have one of two qubits in the (entangled) singlet state as above.  Alice encodes one of four values by performing one of four local operations on her particle and then sends her qubit to Bob who is then able to determine exactly which value Alice encoded.  This, in effect, has enabled Alice to transmit two classical bits of information to Bob via only sending one quantum bit (qubit) of information.  It relies on the entanglement in the original qubit pair.  By sending one qubit alone, Bob would only be able to distinguish two states (that is two values from Alice) with certainty.

Entanglement also allows for the perfect teleportation of an unknown quantum state \cite[\p26]{Chuang2000}.  This enables one party, say Alice, to teleport the state of a qubit to a second party, say Bob, by only sending classical information.  Again, Alice and Bob must each have one of two qubits in the singlet state as above.  Alice then makes a measurement on both her qubit with the unknown state and her qubit from the singlet pair, then classically sends the result to Bob.  Bob can then convert the state of his single particle into that of the unknown state Alice had before she destroyed it by measurement.

Quantum algorithms, those using quantum states and operations, can outperform their classical counterparts exponentially.  Jozsa and Linden \cite{Linden2003} have shown that multi-partite entanglement (a generalisation of entanglement across two parties to multiple parties) is necessary for this speed up for algorithms operating on pure states.  They give the comparison of Shor's quantum algorithm, which can factor an integer of $n$ digits in a running time bounded by $O(n^{3})$, against the best known classical algorithm, which has a running time bounded by $O(\e^{n^{\frac{1}{3}}\left(\log(n)\right)^{\frac{2}{3}}})$, as a dramatic example.

\subsection{Entanglement measures}\label{Entanglement}
Given the apparent importance of entanglement, a suitable measure of `how much' entanglement a state contains should be identified.  To this end, Plenio and Virmani \cite[\p8]{Virmani2006} review the following desirable axioms for an entanglement measure:

\subsubsection{Axioms}
\begin{enumerate}
 \item Given a bipartite quantum system, that is one which is split into two subsystems $A$ and $B$ with Hilbert spaces $\mathcal{A}$ and $\mathcal{B}$ respectively, an entanglement measure $E(\rho)$ should assign a non-negative real number to each state $\rho$ of the system.
 \item If $\{|a\rangle_\mathcal{A}\}_{a=1}^{d_{\mathcal{A}}}$ and $\{|b\rangle_\mathcal{B}\}_{b=1}^{d_{\mathcal{B}}}$ are orthonormal bases of $\mathcal{A}$ and $\mathcal{B}$ respectively then
\begin{equation}
    E\left(\left(\frac{\sum_{j=1}^{\min\left(d_{\mathcal{A}},d_{\mathcal{B}}\right)}|j\rangle_\mathcal{A}|j\rangle_\mathcal{B}}{\sqrt{\min\left(d_{\mathcal{A}},d_{\mathcal{B}}\right)}}\right)\left(\frac{\sum_{k=1}^{\min\left(d_{\mathcal{A}},d_{\mathcal{B}}\right)}\,_\mathcal{A}\langle k|\,_\mathcal{B}\langle k|}{\sqrt{\min\left(d_{\mathcal{A}},d_{\mathcal{B}}\right)}}\right)\right)=\log_{2}\left(\min\left(d_{\mathcal{A}},d_{\mathcal{B}}\right)\right)
  \end{equation}
  should be the maximal value of $E$.
 \item If $\rho$ is separable then $E(\rho)=0$.
 \item The value of $E$ should not increase under deterministic local-operations-and-classical-communication (LOCC) protocols.
 \item If $\rho=|\phi\rangle\langle\phi|$ is a pure state, then
  \begin{equation}
    E(\rho)=E(|\phi\rangle\langle\phi|)=S\Big(\Tr_\mathcal{B}\left(|\phi\rangle\langle\phi|\right)\Big)
  \end{equation}
  where $S$ denotes the von-Neumann entropy, $S(\rho)=-\Tr\left(\rho\log_{2}(\rho)\right)$.
\end{enumerate}

\subsection{Entropy}\label{entropy}\label{Von}
The von-Neumann entropy was used in the last axiom proposed for an entanglement measure. It is the case that the von-Neumann entropy of the reduced state $\Tr_\mathcal{B}\left(\rho\right)$, or equivalently $\Tr_\mathcal{A}\left(\rho\right)$ by the Schmidt decomposition, is indeed the unique measure on pure states $\rho$ that satisfies the first four axioms above \cite[\p7]{Virmani2006}.

\subsubsection{Von-Neumann entropy}
The von-Neumann entropy of the state $\rho$ is defined in \cite[\p510]{Chuang2000} to be
\begin{equation}
  S(\rho)=-\Tr\left(\rho\log_{2}(\rho)\right)
\end{equation}
or equivalently if $\lambda_{k}$ are the eigenvalues of $\rho$
\begin{equation}
  S(\rho)=-\sum_{k}\lambda_{k}\log_{2}(\lambda_{k})
\end{equation}
where the value of $0\log_{2}(0)$ is taken to be zero.

Two examples of both the von-Neumann entropy and an entanglement measure $E$ concern the product state $|0\rangle|0\rangle$ and the (maximally entangled) singlet state $\frac{1}{\sqrt{2}}\Big(|0\rangle|1\rangle-|1\rangle|0\rangle\Big)$ on two qubits.  The reduced density matrices on a single qubit corresponding to each of these states are $|0\rangle\langle0|$ and $\frac{1}{2}|0\rangle\langle0|+\frac{1}{2}|1\rangle\langle1|$ respectively.  The respective values of the von-Neumann entropy are then $-1\log_{2}(1)=0$ and $-\frac{1}{2}\log_{2}\left(\frac{1}{2}\right)-\frac{1}{2}\log_{2}\left(\frac{1}{2}\right)=\log_{2}(2)$, that is the extremal values of an entanglement measure $E$.

\subsubsection{Purity}\label{purity}
The purity of a state $\rho$ is defined as $\Tr\left(\rho^{2}\right)$.  Taking the definition from Section \ref{Quantum}, $\rho$ has the decomposition	$\rho=\sum_{j}p_{j}|\psi_{j}\rangle\langle\psi_{j}|$ for the probabilities $p_{j}$ which sum to $1$ and some orthonormal basis with elements $|\psi_{j}\rangle$.  Therefore the purity of $\rho$ for a $d$ dimensional system is equal to its maximal value of $1$ if and only if $\rho$ is a pure state and equal to its minimal value of $\frac{1}{d}$ if and only if $\rho$ is maximally mixed.

For a composite quantum system of two subsystems $A$ and $B$ with the Hilbert spaces $\mathcal{A}$ and $\mathcal{B}$ respectively, the purity of the reduced density matrix $\rho_\mathcal{A}$ of a pure state $\rho$ of the whole system can also be used to deduce properties of the entanglement between systems $A$ and $B$.  If $\Tr\left(\rho_\mathcal{A}^{2}\right)=1$ then $\rho$ must be a separable (product) state across subsystems $A$ and $B$.  If $\Tr\left(\rho_\mathcal{A}^{2}\right)=\frac{1}{d}$, where $d$ is the minimum dimension of $\mathcal{A}$ and $\mathcal{B}$, then $\rho$ must be maximally entangled across subsystems $A$ and $B$.

These facts are proved in Section \ref{EigenDef} where use will be made of them.

\subsection{Entanglement in spin chains}\label{etangChain}
\subsubsection{Ground state}
In the last decade there has been a lot of effort in studying the entanglement between a continuous block of $l$ spins, and the rest of the chain, in the ground states of certain spin chain Hamiltonians.  The main tool used for this is the entropy of entanglement, the von-Neumann entropy of the reduced ground state on the block of $l$ spins.

Vidal, Latorre, Rico and Kitaev \cite{Vidal2002} first numerically studied this in the $XY$ and $XXZ$ chains.  They used the Jordan-Wigner transformation or Bethe ansatz to find the ground state of each chain, then from this, numerically calculated the entropy of entanglement for the block of $l$ spins.  Having used the Jordan-Wigner transformation and Bethe ansatz initially, this numerical procedure was only polynomially hard in $l$, and as such, relatively large values of $l$ were numerically accessible.  They predicted that the entropy of entanglement was linearly proportional to $\log_{2}(l)$ near critical regions (linked with phase transitions in the ground state) in the limit of large chain length and saturated elsewhere.

Jin and Korepin \cite{Korepin2003} first analytically confirmed this for the $XX$ model.  They represented the entropy of entanglement by a Toeplitz determinant and computed the asymptotics with the Fisher-Hartwig conjecture.  With Its \cite{Korepin2004} they then employed Fredholm operators and Riemann-Hilbert problems to determine the entropy of entanglement as a function of $\gamma$ and $h$ for the $XY$ chain in the large chain, followed by the large $l$ limit and observed logarithmic singularities at the critical regions, with saturation elsewhere.

Calabrese and Cardy \cite{Cardy2004} and Korepin \cite{Korepin2005} used conformal field theory to argue that many integrable one dimensional chains have, to leading order, such a logarithmic singularity in the entropy of entanglement of their ground states at their critical regions (gapless models).

Keating and Mezzadri \cite{Keating2005} expressed the entropy of entanglement of the ground state of a wide range of models (related to quadratic forms of Fermi operators) as averages over classical compact groups.  They were then able to compute these averages, which were either Toeplitz determinants or Hankel matrices, to leading and next to leading order, using a generalisation of the Fisher-Hartwig conjecture.

\subsubsection{Area law for the ground state}
Away from the critical regions, the entropy of entanglement of the ground state for a block of $l$ spins is not singular as $l\to\infty$ in the limit of large chain length for the models mentioned above.  In fact, Masanes \cite{Masanes2010} shows that for any selection of $l$ spins the entropy of entanglement is proportional to the boundary of that selection with the rest of the spin chain, up to a logarithmic factor and some general assumptions which are commonly met by such spin chain models.  The result is applicable to finite dimensional lattices and other more elaborate interaction geometries too.

This is consistent with the previous results.  Here saturation of the entropy of entanglement was seen away from the critical regions, as in the one dimensional chain the boundary of the block of $l$ qubits is constant.

\subsubsection{Matrix product states}
The existence of this area law suggests that the ground states in question do not look like arbitrary states in the Hilbert space, but are somehow closer to product or other less entangled states.  In fact for a large class of spin chain models the ground states can be well approximated by a matrix product state \cite[\p8]{Cirac2008}.  For a $d^{n}$ dimensional system ($n$ subsystems of dimension $d$ for example) the $D\in\mathbb{N}$ dimensional matrix product states are defined in \cite[\p16]{Cirac2008} to be 
\begin{equation}
	\sum_{a_{1},\dots,a_{n}}^{d}\Tr\left(C^{(1,a_{1})}\dots C^{(n,a_{n})}\right)|a_{1}\rangle\otimes\dots\otimes|a_{n}\rangle
\end{equation}
where the $C^{(j,a_{j})}$ are $D\times D$ matrices for each value of $j$ (site index) and $a_{j}$.   Each matrix $C^{(j,a_{j})}$ contains separate parameters of the state.  The case where $D=1$ reduces to a product state over the $n$ subsystems.  The value of $D$ and $d$ may be site dependent to generate the most general matrix product states.

The number of parameters in these states is linear in $n$ for fixed $D$, unlike an arbitrary state for which there are $2d^{n}$ (real) parameters.  Although, for large enough $D$ any state may be represented in this way.  For many spin chain models, relatively low values of $D$ suffice to give good approximations to their ground states.  Therefore any numerical techniques for finding such approximations to ground states are often made tractable.

\subsubsection{Higher states}
Far fewer results exist concerning the entanglement in higher energy eigenstates of spin chains \cite{Calabrese2009}.  In relation to the models already seen, the following two results show that the entanglement in these higher states behaves very differently to that in the ground state.

Entanglement between two spins in the XXX model has been numerical studied by Arnesen, Bose and Verdral \cite{Vedral2000}.  They found that by varying the system's temperature (to move away from the ground state) and magnetic field (a model parameter), the entanglement between two spins could be increased, even in the case where none was initially present.

Alba, Fagotti and Calabrese \cite{Calabrese2009} later treated the entropy of entanglement of eigenstates of both the critical and non-critical $XY$ and $XXZ$ models for a continuous block of $l$ spins in the large chain limit.  They used the Jordan-Wigner transform and the Bethe ansatz to diagonalise the models.  A full analytic description was given for the $XY$ model whereas numerical methods had to be resorted to for the $XXZ$ model.  They found a distinct class of eigenstates with an entropy of entanglement proportional to the block length $l$.
\section{Non-integrable qubit chains}\label{nonInt}
Non-integrability in quantum spin chains may be defined though the observation of chaotic effects in the thermodynamic (large chain length) limit, as seen from the following examples.  Roughly speaking, the limited range of models which are solved by the Jordan-Wigner transformation or Bethe ansatz are integrable.  The defining signature of non-integrability is seen in the neighbouring level (eigenvalue) spacing distributions of the Hamiltonians.  Integrable models tend to show a Poisson distribution and non-integrable models tend to exhibit a distribution with repulsion between the energy levels.  This will be seen in the following examples and references, which illustrate the ease with which non-integrability arises.

Before this, note that much of the initial work into integrable spin chains was based on numerical investigation.  M\"uller, Bonner and Parkinson \cite{Parlinson1987} argue however that the reliability of these numerical methods and approximation techniques is greatly reduced when applied to non-integrable models.  In integrable models the thermodynamic behaviour is qualitatively seen in relatively short chains whereas with non-integrable models qualitative differences in behaviour are seen at different (numerically accessible) chain lengths.  This must reduce the strength of any conclusions made from numerical analysis in these cases.

\subsection{Ising model with tilted magnetic field}
Karthik, Sharma and Lakshminarayan \cite{Lakshminarayan2007} consider the Ising model in a tilted magnetic field with the Hamiltonian
\begin{equation}
  H_{n}^{(KSL)}=\frac{J}{2}\sum_{j=1}^{n-1}\sigma_{  j  }^{(3)}\sigma_{  j+1  }^{(3)}+\frac{h}{2}\sum_{j=1}^{n}\Big(\sin(\theta)\sigma_{  j  }^{(1)}+\cos(\theta)\sigma_{  j  }^{(3)}\Big)
\end{equation}
with the real constants $J$ (coupling strength) and $h$ (the magnetic field strength) and the tilting parameter $\theta\in\left[0,\frac{\pi}{2}\right]$.  For $\theta=0$ the model is trivially diagonal and for $\theta=\frac{\pi}{2}$ the Ising model\footnote{This is the Ising model seen in Section \ref{1.1.5} up to a local unitary transformation that maps $\sigma^{(3)}\to\sigma^{(1)}$ and $\sigma^{(1)}\to\sigma^{(3)}$.} is recovered, which is solvable via the Jordan-Wigner transformation.

The Hamiltonian commutes with the bit reversal matrix $B$ which interchanges the qubits labelled $j$ and $n-j+1$ so that both matrices share a joint orthonormal eigenbasis.  As $B^{2}=I$ and $B$ is Hermitian, this basis can be split into two subspaces on which the eigenvalues of $B$ are either $\pm1$, labelled odd ($-1$) and even ($+1$).  Karthik, Sharma and Lakshminarayan numerically plot the graph reproduced in Figure \ref{avoid}, showing part of the even energy spectrum as the angle $\theta$ is varied.  The energy levels do not cross each other within this region but closely approach and then veer away.  This repulsion is a key signature of a non-integrable system.

\begin{figure}
\centering
\input{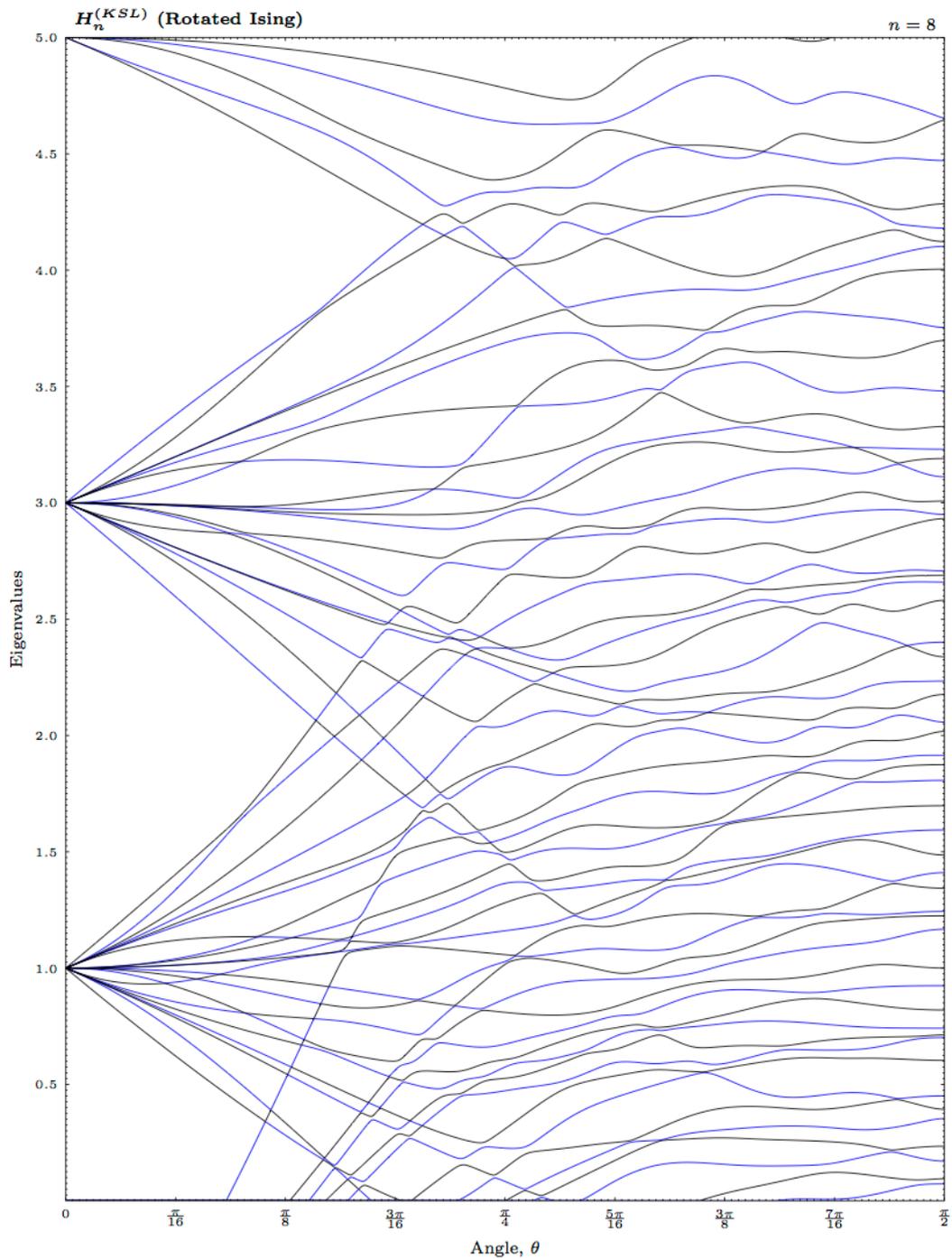}
\caption[Avoided crossings in the Ising model with a tilted magnetic field]{Adapted from \cite{Lakshminarayan2007}.  Part of the even energy spectrum of the Ising model with tilted magnetic field, $H_{n}^{(KSL)}$, on $n=8$ qubits for $\theta\in\left[0,\frac{\pi}{2}\right]$ and $J=h=2$.  The levels are coloured alternatively.  The energy levels are seen not to cross as $\theta$ is increased.  They only closely approach and then veer away.}
\label{avoid}
\end{figure}

For the even symmetry subspace, the authors also rescale the eigenvalues so that they have an approximately constant density (unfolding, see Section \ref{nnnumerics}) and then plot the histogram of the spacings between neighbouring (unfolded) eigenvalues seen in Figure \ref{Hus}.  This density is seen to drop at small spacings, indicating the reluctance of the eigenvalues to group closely.

\subsection{Chain defects}
Gubin and Santos \cite{Santos2011} consider the $XXZ$ chain with a defect at the $d^{th}$ site, that is
\begin{equation}\label{defect}
  H_{n}^{(GS)}=\frac{J}{2}\sum_{j=1}^{n-1}\Big(\sigma_{  j  }^{(1)}\sigma_{  j+1  }^{(1)}+\sigma_{  j  }^{(2)}\sigma_{  j+1  }^{(2)}+\Delta\sigma_{  j  }^{(3)}\sigma_{  j  }^{(3)}\Big)+\frac{h}{2}\sum_{j=1}^{n}\sigma_{j}^{(3)}+\frac{\epsilon}{2}\sigma_{  d  }^{(3)}
\end{equation}
with the real constants $J$ and $\Delta$ (coupling constants), $h$ (the magnetic field strength) and $\epsilon$.  For $\epsilon=0$, this is the $XXZ$ model\footnote{This is the $XXZ$ model seen in Section \ref{1.1.5} without the boundary terms $\sigma_{n}^{(a)}\sigma_{1}^{(a)}$.} and is solvable by the Bethe ansatz.

For $n=15$, $h=0$, $J=0.5$ and $\Delta=0.5$, adding the defect $\epsilon=0.5$ at the beginning of the chain ($d=1$) produces approximately Poisson statistics (an integrable characteristic) for the unfolded nearest-neighbour level spacings in the symmetry subspace $H_{n}^{(Z)}=-5$ of the unfolded spectrum (see Section \ref{Bethe}).  Adding the defect in the middle of the chain ($d=7$) however produces approximately Wigner statistics (a non-integrable characteristic) in this subspace as seen in Figure \ref{Santos}.

\begin{figure}
  \centering
  \input{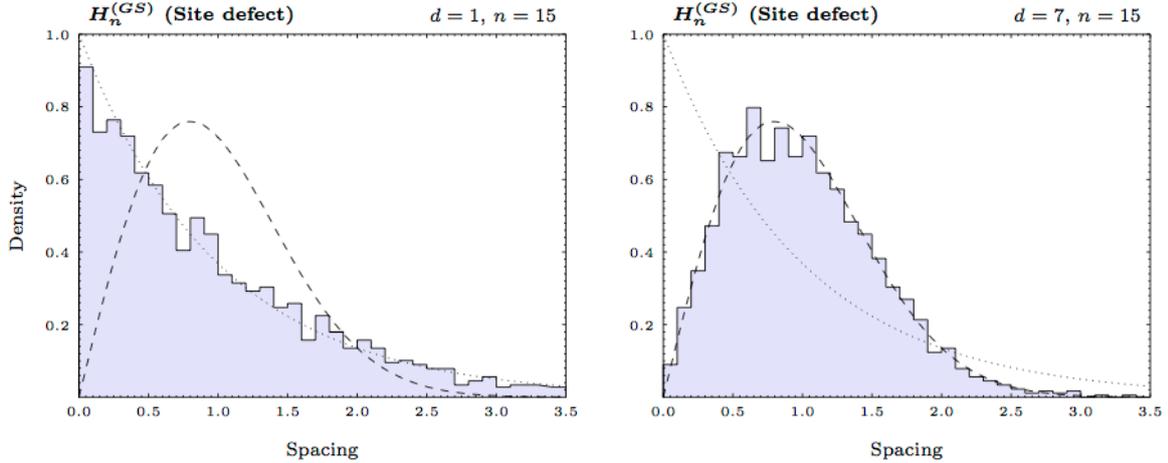}
  \caption[Crossover from Poisson to Wigner statistics]{Adapted from \cite{Santos2011}:  The numerical nearest-neighbour level spacing distributions of the unfolded spectrum of the $XXZ$ chain with a site defect, $H_{n}^{(GS)}$, within the $H_{n}^{(Z)}=-5$ symmetry subspace, for $n=15$, $h=0$, $J=0.5$, $\Delta=0.5$, $\epsilon=0.5$, $d=1$ (left) and $d=7$ (right).  The change from Poisson (dotted line) to Wigner (dashed line) statistics is observed.}
  \label{Santos}
\end{figure}

\subsection{Next to nearest-neighbour interactions}
Kudo and Deguchi \cite{Deguchi2004} studied the similar model
\begin{align}\label{DeguchiNNN}
  H_{n}^{(KD)}&=\frac{J_{1}}{2}\sum_{j=1}^{n}\Big(\sigma_{  j  }^{(1)}\sigma_{  j+1  }^{(1)}+\sigma_{  j  }^{(2)}\sigma_{  j+1  }^{(2)}+\Delta_{1}\sigma_{  j  }^{(3)}\sigma_{  j+1  }^{(3)}\Big)\nonumber\\
  &\qquad\qquad+\frac{J_{2}}{2}\sum_{j=1}^{n}\Big(\sigma_{  j  }^{(1)}\sigma_{  j+2  }^{(1)}+\sigma_{  j  }^{(2)}\sigma_{  j+2  }^{(2)}+\Delta_{2}\sigma_{  j  }^{(3)}\sigma_{  j+2  }^{(3)}\Big)
\end{align}
where $J_{1},J_{2},\Delta_{1}$ and $\Delta_{2}$ are real coupling constants.  They numerically calculated the (unfolded) nearest-neighbour level spacing distribution in the symmetry subspace $H_{n}^{(Z)}=0$ for $n=18$.  In this subspace there is a further symmetry subspace with a corresponding symmetry matrix $H_{n}^{(K)}$. Figure \ref{Deguchi} shows the Wigner statistics in the $H_{n}^{(Z)}=0$, $H_{n}^{(K)}=\frac{2\pi}{n}$ subspace that they observed.  

If the system is not fully desymmetrised, that is the whole $H_{n}^{(Z)}=0$ subspace is taken, the statistics change their qualitative behaviour as seen in Figure \ref{Deguchi}.  This highlights the importance of the symmetries of these types of model.

\begin{figure}
	\centering
	\input{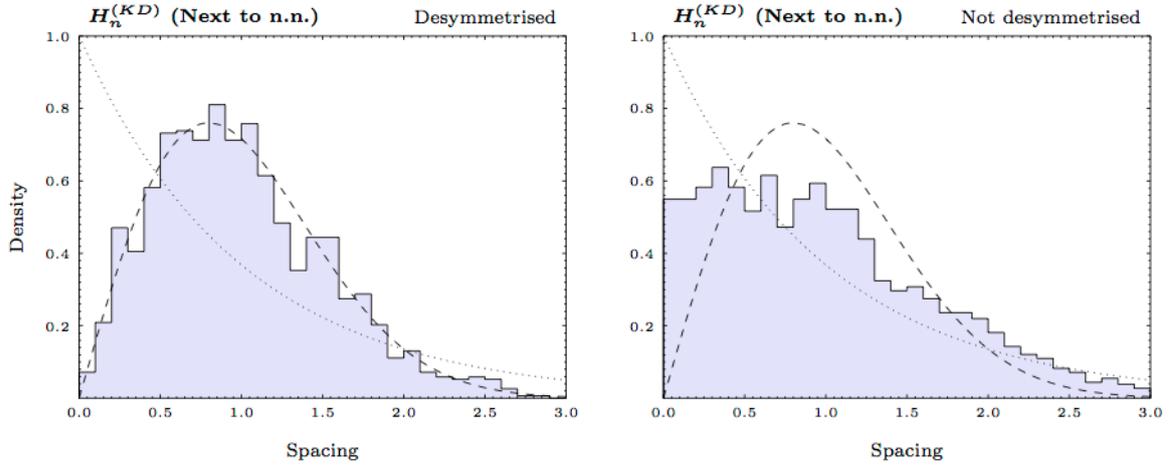}
	\caption[Wigner statistics in a fully desymmetrised spectral sector]{Adapted from \cite{Deguchi2004}:  (Left) The numerical nearest-neighbour level spacing distribution of the unfolded spectrum of the next to nearest-neighbour Hamiltonian $H_{n}^{(KD)}$, within the $H_{n}^{(Z)}=0$, $H_{n}^{(K)}=\frac{2\pi}{n}$ symmetry subspace, for $n=18$, $J_{1}=2J_{2}$, $\Delta_{1}=\Delta_{2}=0.5$ .  (Right) The same quantity for the entire (that is, incompletely desymmetrised) $H_{n}^{(Z)}=0$ subspace.  The dashed and dotted lines give the Wigner spacing distribution and Poisson distribution respectively.}
	\label{Deguchi}
\end{figure}

\subsection{Coupled chains}
Hus and Angl\`es d'Auriac \cite{Angles1992} have studied two coupled $XXX$ chains described by
\begin{align}\label{coupled}
  H_{n}^{(HA)}&=\frac{J_{1}}{2}\sum_{j=1}^{n-1}\left(\sigma_{j,1}^{(1)}\sigma_{j+1,1}^{(1)}+\sigma_{j,1}^{(2)}\sigma_{j+1,1}^{(2)}+\sigma_{j,1}^{(3)}\sigma_{j+1,1}^{(3)}\right)\nonumber\\
  &\qquad\qquad+\frac{J_{1}}{2}\sum_{j=1}^{n-1}\left(\sigma_{j,2}^{(1)}\sigma_{j+1,2}^{(1)}+\sigma_{j,2}^{(2)}\sigma_{j+1,2}^{(2)}+\sigma_{j,2}^{(3)}\sigma_{j+1,2}^{(3)}\right)\nonumber\\
  &\qquad\qquad\qquad\qquad+\frac{J_{2}}{2}\sum_{j=1}^{n}\left(\sigma_{j,1}^{(1)}\sigma_{j,2}^{(1)}+\sigma_{j,1}^{(2)}\sigma_{j,2}^{(2)}+\sigma_{j,1}^{(3)}\sigma_{j,2}^{(3)}\right)
\end{align}
where the qubits are labelled by the pair of labels $j=1\dots,n$ and $c=1,2$ and where $J_{1}$ and $J_{2}$ are real coupling constants.

They again see a transition from Poisson statistics to Wigner statistics in the (unfolded) nearest-neighbour level spacing distributions, in the $H_{n}^{(Z)}=0$ symmetry subspace, as the coupling between the chains increases, see Figure \ref{Hus}.

\begin{figure}
  \centering
  \input{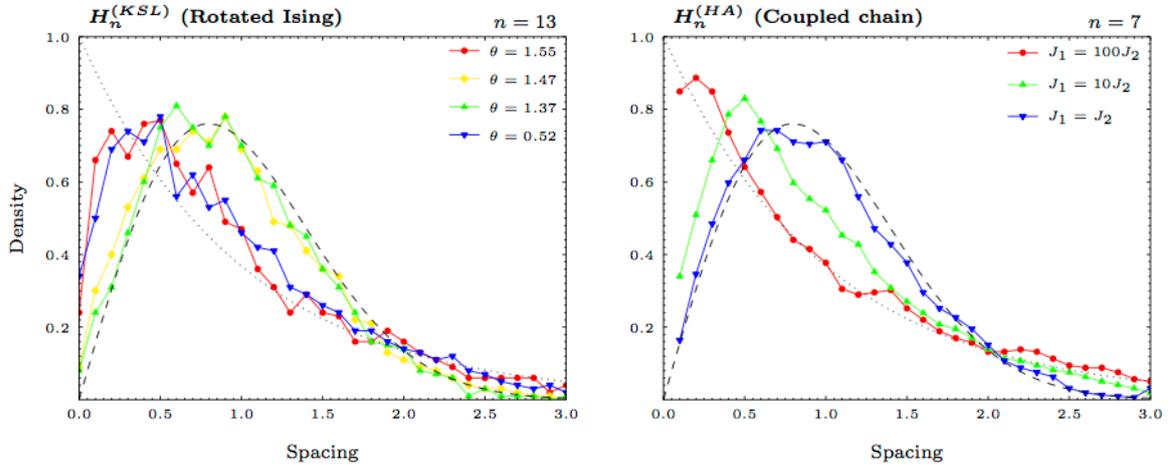}
  \caption[Nearest-neighbour level spacing distribution of a $XXZ$ chain with a defect]{(left) Adapted from  \cite{Lakshminarayan2007}:  The numerical nearest-neighbour level spacing distributions of the unfolded spectrum of the Ising chain with a tilted magnetic field, $H_{n}^{(KSL)}$, within the even symmetry subspace, for $n=13$, $J=h=2$ and different angles $\theta$.   (Right) Adapted from  \cite{Angles1992}:  The numerical nearest-neighbour level spacing distributions of the unfolded spectrum of the coupled chain, $H_{n}^{(HA)}$, within each symmetry subspace, for $n=7$ qubits and different ratios of $J_{1}$ and $J_{2}$.  The crossover from Poisson (dotted line) to Wigner (dashed line) statistics is observed in both cases.}
  \label{Hus}
\end{figure}

\subsection{Energy distributions of generic spin chains}\label{Mahler}
For the vast majority of spin chain models there do not exist analytic techniques to solve (diagonalise) them.  Hartmann, Mahler and Hess \cite{Mahler,Mahler2005} however start to analyse a generic spin chain without the need for diagonalisation.  They explicitly consider the Hamiltonian of $n$ qubits
\begin{equation}
  H_{n}^{(HMH)}=\sum_{j=1}^{n-1}h_{j}
\end{equation}
where $h_{j}$ are Hermitian matrices acting only on the qubits labelled $j$ and $j+1$ non-trivially.  A term $h_{n}$ acting on the qubits labelled $n$ and $1$ can be added so that $H_{n}^{(HMH)}$ includes all of the nearest-neighbour models already considered.

They look at product states, $|\phi\rangle=\bigotimes_{j}|\phi_{j}\rangle$, for states $|\phi_{j}\rangle$ of the individual qubits, for which
\begin{equation}
	\nu_{\phi}=\langle\phi|H_{n}^{(HMH)}|\phi\rangle\qquad\text{and}\qquad
	\Delta_{\phi}^{2}=\langle\phi|{H_{n}^{(HMH)}}^{2}|\phi\rangle-\nu_{\phi}^{2}
\end{equation}
satisfy $\Delta_{\phi}\geq nC_{1}$ for some constant $C_{1}$ and $\langle\chi|h_{j}|\chi\rangle\leq C_{2}$ for some constant $C_{2}$ for all states $|\chi\rangle$.

The product state $|\phi\rangle$ can be decomposed over the eigenstates $|\psi_{k}\rangle$ of $\frac{H_{n}^{(HMH)}-\nu_{\phi}}{\Delta_{\phi}}$ as
\begin{equation}
  |\phi\rangle=\sum_{k=1}^{2^{n}}c_{k}|\psi_{k}\rangle
\end{equation}
for some complex coefficients $c_{k}$ such that $|\phi\rangle$ is normalised. These values can then be used to define a probability distribution on the real line so that
\begin{equation}
  \mathbb{P}(\lambda\leq x)=\sum_{k:\lambda_{k}\leq x}|c_{k}|^{2}
\end{equation}
where the $\lambda_{k}$ is the corresponding (real) eigenvalue to the eigenstate $|\psi_{k}\rangle$.

Hartmann, Mahler and Hess prove that this distribution tends weakly to that of a standard normal distribution in the limit of large chain length.  Furthermore, for an orthonormal basis of such product states $|\phi_{k}\rangle$ they show in \cite{Mahler2005} that the spectral density $\rho_{n,1}^{(HMH)}(\lambda)$ of $H_{n}^{(HMH)}$ in the limit $n\to\infty$ weakly converges to
\begin{equation}
  \rho_{n,1}^{(HMH)}(\lambda)=\frac{1}{2^{n}}\sum_{k=1}^{2^{n}}\frac{1}{\sqrt{2\pi\Delta^{2}_{\phi_{k}}}}\e^{-\frac{(\lambda-\nu_{\phi_{k}})^{2}}{2\Delta_{\phi_{k}}^{2}}}
\end{equation}
so that the proportion of the eigenvalues of $H_{n}^{(HMH)}$ in $(-\infty,x)$ is $\int_{-\infty}^{x}\rho_{n,1}^{(HMH)}(\lambda)\di\lambda$ in the limit $n\to\infty$ for all fixed real values $x$.
\section{Random matrix theory}\label{RMT}
Hartmann, Mahler and Hess started to address the issue of the analysis of generic spin chain models without the need to solve (diagonalise) them explicitly.  This is very  much the spirit of random matrix theory, where statistical properties of matrix ensembles are studied.  The Wigner nearest-neighbour distribution, as seen throughout the last section, is well understood in the random matrix theory community.  The connection to chaos is also widely studied in this framework.

\subsection{Motivation}
The Oxford Handbook of Random Matrix Theory \cite{Handbook} offers a comprehensive background to random matrix theory.  A short summary of the relevant history therein is given here:

In the 1930's narrow energy resonances were observed in the scattering of slow neutrons by large nuclei.  This lead Bohr to formulate the notion of the compound nucleus with many particle interactions in 1937.  There was no exact way to determine these interactions in a nucleus so Wigner and Eisenbud formally used the `R-matrix' to model the results of this process.  As a function of energy, the R-matrix was defined to be singular at each eigenvalue of the Hamiltonian of the nucleus (resonance of the nucleus).

The Wishart ensemble of random matrices had existed since the 1920's.  Wigner was prompted by this ensemble to model the spectral characteristics of the R-matrix by a random matrix (a real symmetric matrix with independent Gaussian distributed entries) as the specific details of the R-matrix were too complicated to explicitly determine.  He states in a related paper that
\begin{quote}
  \emph{``The present problem arose from the consideration of the properties of the wave function of quantum mechanical systems which are assumed to be so complicated that statistical considerations can be applied to them.''}\hfill\cite{Wigner1955}
\end{quote}

Wigner determined that the limiting spectral density (the probability density function describing the distribution of the eigenvalues) of this ensemble had the shape of a semicircle for large matrix size.  This was deduced by calculating the individual moments of the spectral density.  Furthermore he predicted that the nearest-neighbour level spacing distribution for the unfolded eigenvalues (that is, transformed to have a constant density on some compact interval) of this ensemble for large matrix size should have the form $c_{1}s\e^{-c_{2}s^{2}}$ where $s$ is a spacing parameter and $c_{1}$ and $c_{2}$ are real constants.  This is the so called Wigner distribution seen throughout the last section where $c_{1}$ and $c_{2}$ were chosen such that this spacing probability density has unit mean.  This spacing density shows the tendency of the eigenvalues to repel one another, as it tends to zero as $s\to0$.

This `level repulsion' had already be observed in Hermitian matrices by Wigner and von Neumann in 1929.  By the 1960's this level repulsion had also been observed in physical systems by Rosenzweig and Porter, lending weight to Wigner's spacing prediction and the effectiveness of the predictions of random matrix models.  Many further examples will be seen in the next section.

In 1960 Mehta and Gaudin employed orthogonal polynomials to rigorously reproduce the semicircle  as the limiting spectral density for Wigner's ensemble.  These techniques also allowed them to determine the exact nearest-neighbour level spacing distribution, finding it to be extremely close to that of Wigner's prediction.

\subsection{Symmetry}
In 1962 Dyson introduced the three circular ensembles characterised by the values $\beta=1,2,4$.  They were the circular orthogonal ensemble (COE), circular unitary ensemble (CUE) and circular symplectic ensemble (CSE) respectively.  These ensembles consist of real orthogonal, complex unitary and quaternion unitary matrices, respectively, endowed with the Haar probability measure over each compact group.  Dyson calculated the nearest-neighbour level spacing distributions of these ensembles non-rigorously and then the spectral point correlation functions (that is, the spectral density and higher correlation functions) and re-derived Gaudin's results for their nearest-neighbour level spacings.

Dyson matched these ensembles to physical systems corresponding to their symmetries.  He also showed that only the $\beta=1,2,4$ ensembles were necessary to describe any Hilbert space symmetry.  He states that
\begin{quote}
  \emph{``the most general kind of matrix ensemble, defined with a symmetry group which may be completely arbitrary, reduces to a direct product of independent irreducible ensembles each of which belongs to one of the three known types.''}\hfill\cite[\p43]{Handbook}
\end{quote}

This is commonly referred to as Dyson's three fold way.  These results are also applicable to the Gaussian invariant ensembles characterised by the values $\beta=1,2,4$, the Gaussian orthogonal ensemble (GOE), Gaussian unitary ensemble (GUE) and Gaussian symplectic ensemble (GSE) respectively.  These ensembles consist of real symmetric, complex Hermitian and quaternion Hermitian matrices, respectively, where the real parameters of the matrices' elements are independently distributed Gaussian random variables.

In the 1990's other symmetries were also found to be required for a complete description of relevant physical systems, such as chiral symmetry, see \cite[\p52]{Handbook} for a review.

\subsection{Simple nearest-neighbour level spacing statistics}\label{spacingStat}
The canonical nearest-neighbour level spacing statistics of the GOE, GUE and GSE can been seen for small matrix dimension through an enlightening calculation \cite{Keating}\cite[\p70]{Haake}.  Consider the $2\times 2$ matrix
\begin{equation}
	G=\begin{pmatrix}w+z&x-\im y_{1}-\jm y_{2}-\km y_{3}\\x+\im y_{1}+\jm y_{2}+\km y_{3}&w-z\end{pmatrix}
\end{equation}
where $\im$, $\jm$ and $\km$ are the quaternion basis elements, $\im$ is identified with the standard imaginary unit and $w,x,y_{1},y_{2},y_{3}$ and $z$ are real parameters.  Equipped with the Gaussian probability measure 
\begin{equation}
	\e^{-\frac{1}{2}\Tr\left(H^{2}\right)}\di w\di x\di y_{1}\di y_{2}\di y_{3}\di z=\e^{-w^{2}-x^{2}-y_{1}^{2}-y_{2}^{2}-y_{3}^{2}-z^{2}}\di w\di x\di y_{1}\di y_{2}\di y_{3}\di z
\end{equation}
these matrices, up to scaling, form the $2\times 2$ GSE, with $y_{2}$ and $y_{3}$ removed the GUE, with $y_{1}$, $y_{2}$ and $y_{3}$ removed the GOE and with $x$, $y_{1}$, $y_{2}$ and $y_{3}$ removed an ensemble of two independent points.  The eigenvalues of the matrix $G$ are given by
\begin{equation}
	\lambda=w\pm\sqrt{x^{2}+y_{1}^{2}+y_{2}^{2}+y_{3}^{2}+z^{2}}
\end{equation}
Setting $r^{2}=x^{2}+y_{1}^{2}+y_{2}^{2}+y_{3}^{2}+z^{2}$ allows the spacing, $s$, between the eigenvalues to be written as $2r$.  Changing variables from $x,y_{1},y_{2},y_{3}$ and $z$ to the polar coordinates $r$ and some further angular parameters leads to the following probability density functions for the spacing $s$ after integrating out the angular parameters
\begin{equation}\label{1.4.4}
	\rho(s)=Cs^\beta\e^{-cs^{2}}
\end{equation}
where $\beta=0,1,2,4$ for the independent points, GOE, GUE and GSE respectively and $C$ and $c$ are constants.  The factor $s^\beta$ arises from the Jacobian of the transform to polar coordinates in $\beta+1$ dimensions.

These are the approximate forms of the nearest-neighbour level spacing distributions seen previously in Section \ref{nonInt}.  Figure \ref{simpleSpacings} shows these distributions scaled to have unit area and mean.  The curves for $\beta=1,2,4$ closely approximate the limiting nearest-neighbour level spacing distributions for the unfolded eigenvalues of the GOE, GUE and GSE respectively as matrix size increases \cite[\p14]{Mehta}.  They will be used to approximate these distributions from now on.

\begin{figure}
  \centering
  \input{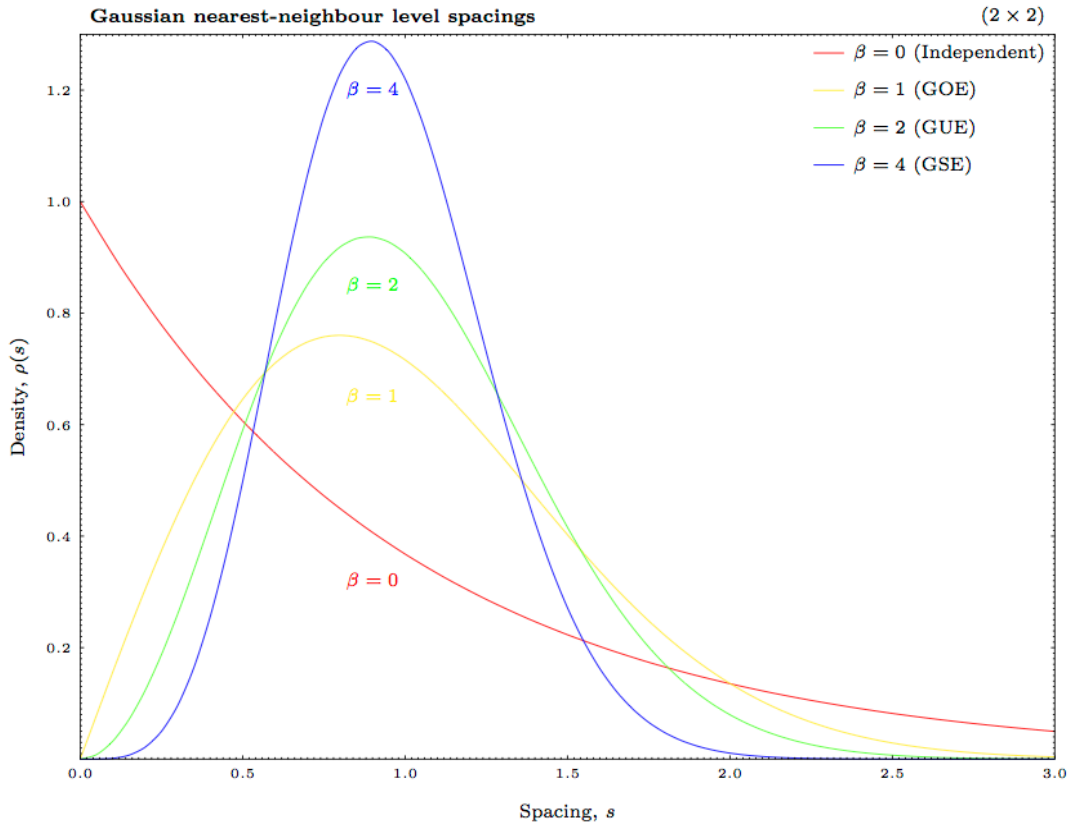}      
  \caption[Simple random matrix nearest-neighbour level spacing distributions]{The nearest-neighbour level spacing distributions for $2\times2$ GSE, GUE, GOE and independent point ensemble scaled to have unit area and mean.}
  \label{simpleSpacings}
\end{figure}

\subsection{Experimental data}\label{Data}
The Wigner spacing distribution (that is (\ref{1.4.4}) for $\beta=1$) is seen for the spacings between neighbouring unfolded eigenvalues in a wide range of physical systems.  Figure \ref{GOE} shows a range of systems for which this has been observed.  The nearest-neighbour level spacing distribution for the GUE was even observed by Krb\'alek and \v Seba \cite{Seba2000} in the bus arrival times in Cuernavaca Mexico.

The property leading to the correspondence of these statistics to those of random matrix ensembles with the appropriate symmetries is conjectured to be chaos:

\begin{figure}
  \centering
  \input{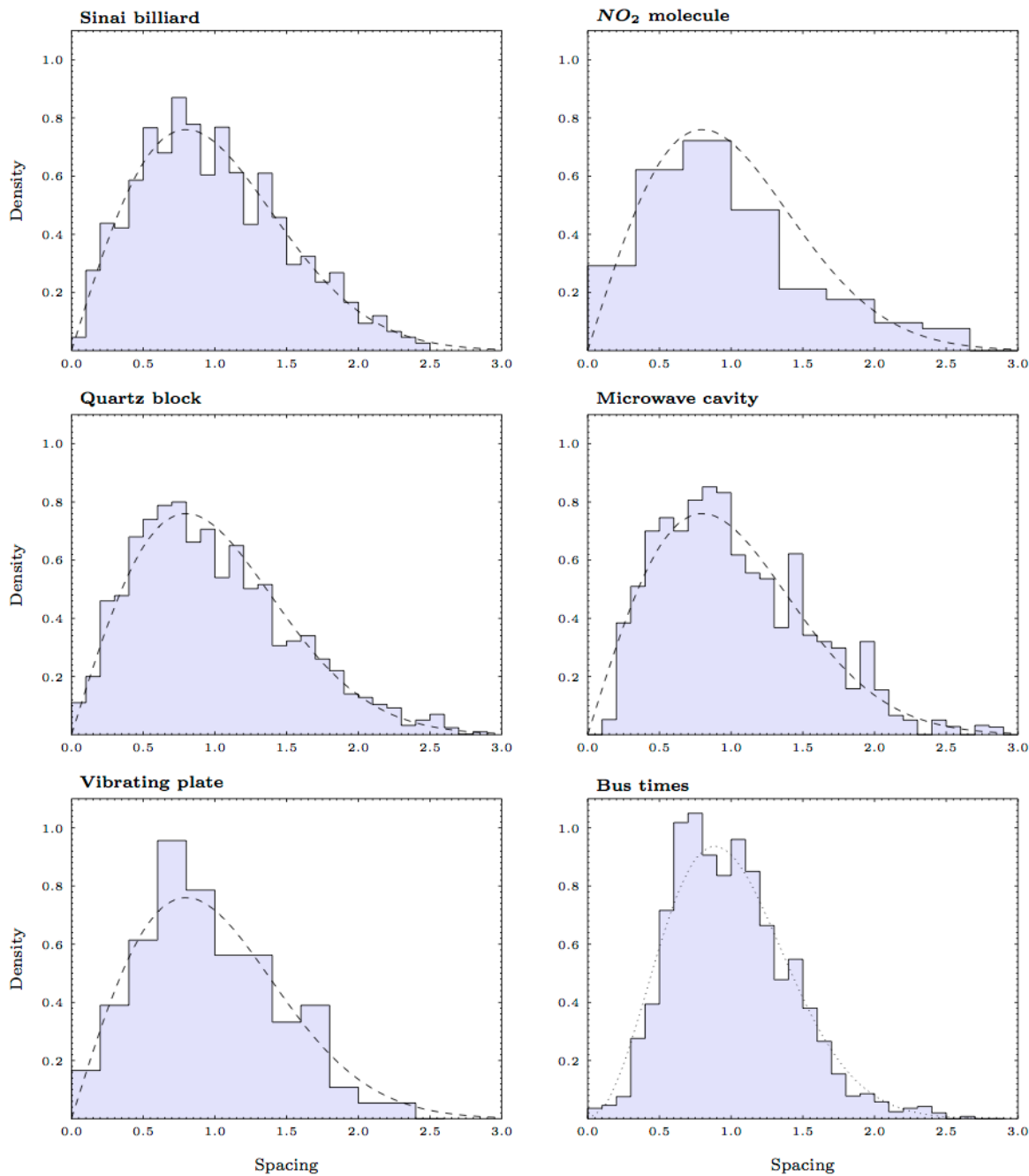}      
  \caption[Physical examples of GOE and GUE nearest-neighbour level spacing distributions]{Adapted from \cite[\p6]{Haake}:  Nearest-neighbour level spacing distributions for the unfolded quantum spectrum of a Sinai billiard with time reversal symmetry, quantum spectrum of an $NO_{2}$ molecule, energy spectrum of a vibrating quartz block in the shape of a three dimensional Sinai billiard, microwave spectrum of a chaotic volume, energy spectrum of a vibrating plate in the shape of a stadium billiard and bus arrival times in Cuernavaca Mexico \cite{Seba2000}.  All the distributions are normalised and scaled to have unit mean, approximate GOE (dashed) and GUE (dotted) limiting  nearest-neighbour level spacing distributions are overlaid.}
  \label{GOE}
\end{figure}

\subsection{Chaos}
Chaos plays a central role in the application of a random matrix model to a physical system \cite[\p23]{Handbook}.  Billiard systems provide an example of this with the benefit that their chaotic interpretation is relatively simple.   There are many other widely ranging examples though, as seen in the previous section.

The (classical) stadium billiard \cite[\p165]{Snaith} is a two dimensional area consisting of a rectangular section with semicircular sections joined to each end, see Figure \ref{Billiard}. A point particle is free to move within the billiard, travelling in straight lines with constant speed. When contact is made with the billiard's boundary, elastic reflection occurs according to Snell's law.  The billiard is classically chaotic in the sense that most close trajectories exponentially separate in time.

\begin{figure}
  \centering
  \input{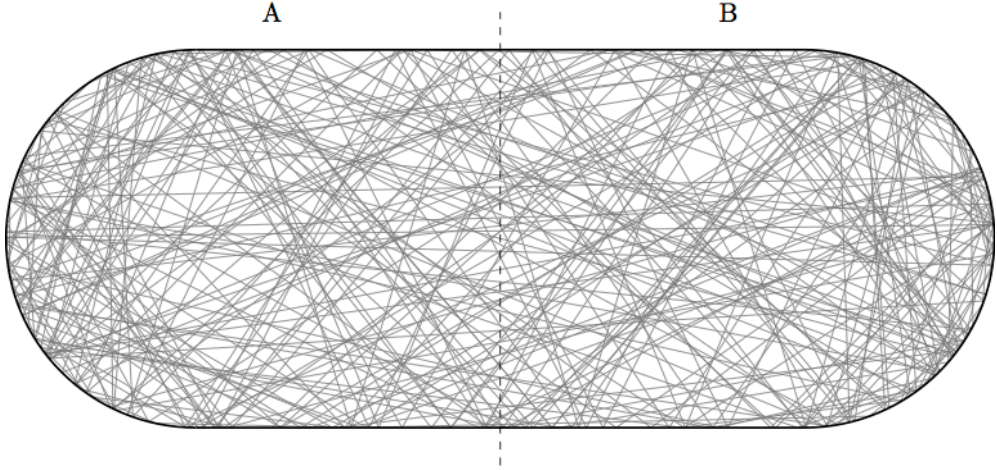}
  \caption[The classical stadium billiard]{The classical stadium billiard.  The trajectory of a free particle within the billiard is shown, reflections occurring at the billiard's boundary according to Snell's law.  Over a long time period the particle spends roughly equal time in either half A or B for a generic path, a consequence of this billiard's chaotic nature \cite[\p165]{Snaith}.}
  \label{Billiard}
\end{figure}

The quantum version of this billiard consists of solving the two dimensional Schr\"{o}dinger (wave) equation with zero potential inside the billiard, infinite potential outside and a zero wave condition on the billiard's boundary.  The solutions to Schr\"{o}dinger's equation produces (unfolded) energy eigenvalues with a Wigner nearest-neighbour distribution.  This is just one of the quantum features of this system well modelled by the GOE \cite[\p166]{Snaith}.

The connection between such classically chaotic systems and the local spectral statistics of random matrices is strong and results from the system's and ensemble's symmetries being matched.  The Bohigas, Giannoni and Schmit (BGS) conjecture states that:
\begin{quote}
\emph{``the spectral fluctuation measures of a generic classically chaotic system coincide with those of the canonical random matrix ensemble that has the same symmetries (unitary orthogonal or symplectic)"}\hfill\cite[\p24]{Handbook}
\end{quote}
For the billiard above, GOE nearest-neighbour level spacing statistics are seen and this connects with the fact that the system has a time reversal symmetry, which matches it to this ensemble.  If a charged particle is considered in the billiard, with a perpendicular magnetic field applied, GUE nearest-neighbour spacing statistics are observed, relating to the fact that time reversal symmetry has been broken so that the system's symmetries match that of the GUE \cite[\p167]{Snaith}.

In relation to the spin chains seen earlier, the non-integrable chains (roughly those that display level repulsion) can be interpreted as chaotic in the sense of matching random matrix predictions.  An analogue to time reversal symmetry is also available for these spin chains.  In section \ref{nnnumerics} it will be seen that this symmetry, or lack thereof, again corresponds to GOE or GUE level statistics being observed.

\subsection{Key ensemble properties}
There are three ensemble properties that make a random matrix model physically useful if they hold.  These are:

\subsubsection{Universality}
An ensemble is universal \cite[\p103]{Handbook} if the ensemble's average local spectral fluctuations are independent of the ensemble's probability measure in the limit of large matrix size.  This has been shown to hold for a wide range of ensembles.  Of particular note are the Hermitian Wigner ensembles (Hermitian matrices with the real parameters of their entries independently distributed) and the Hermitian unitarily invariant ensembles (Hermitian matrices $H$ with a distribution proportional to $\e^{-\Tr\left(V(H)\right)}$ for some suitable function $V$).  The GUE is at the intersection of the two.  Universality also holds for the related orthogonal and symplectic ensembles.

Heuristically, this property shows that the average limiting local spectral fluctuations are a result of the ensemble's symmetry, that is the symmetry of the system it is modelling, rather that an artefact of the probability measure chosen for the ensemble.

\subsubsection{Ergodicity}
An ensemble is ergodic \cite[\p20]{Handbook} if the ensemble's average local spectral fluctuations are equal to almost every member's individual
local spectral average, in the limit of large matrix size.  In particular the $\beta=1,2,4$ Gaussian and circular ensembles are ergodic.

Heuristically, this property states that the ensemble average of some local spectral fluctuation is a good predictor for that of a specific member of the ensemble, perhaps describing a physical system of interest.

\subsubsection{Stationarity}
An ensemble is stationary \cite[\p20]{Handbook} if the ensemble's average local spectral fluctuations are independent of spectral position, in the limit of large matrix size.  In particular the $\beta=1,2,4$ Gaussian and circular ensembles are stationary, within the bulk of their spectra.

Heuristically, this property states that the ensemble average of some local spectral fluctuation is constant throughout the bulk of the spectrum.
\section{K-body ensembles}
Random matrix models make satisfying predictions about the local spectra fluctuations of many relevant physical systems, as seen above.  However this is not the case for the global spectral statistics.  In 1971 Dyson quotes Professor G.E. Uhlenbeck as saying:
\begin{quote}
  \emph{``If you admit that the Wigner ensemble gives a completely wrong answer for the level density, why do you believe any of the other predictions of random matrix theory?''}\hfill\cite{Dyson1971e}
\end{quote}
This concern led to the search for ensembles whose wider spectral statistics more closely match that of physical systems of interest.

\subsection{Embedded ensembles}
One prominent such ensemble is the $2$-body random ensemble (TBRE) or, more generally, the $k$-body embedded Gaussian orthogonal ensemble (EGOE($k$)).  Ensembles of this type were first suggested by French and Wong \cite{French1970,French1971} and Bohigas and Flores \cite{Bohigas1971,Bohigas1971a} in the 1970's.  Here the goal was to determine an ensemble which modelled only the $k<n$ body interactions between the $n$ indistinguishable particles within a nucleus, so that the spectral density (numerically) was a more physically realistic Gaussian distribution.

In these papers, the transition of the Gaussian spectral density to the semicircle of the GOE as $k$  increased from $2$ to $n$ was numerically studied, see Figure \ref{French}.  Also, the unfolded nearest-neighbour level spacing distributions of these ensembles were numerically studied and similarities to that of the GOE observed, see Figure \ref{FrenchSpacings}.

The models resisted much further analysis though due to their complicated structure.

\begin{figure}
  \centering
  \input{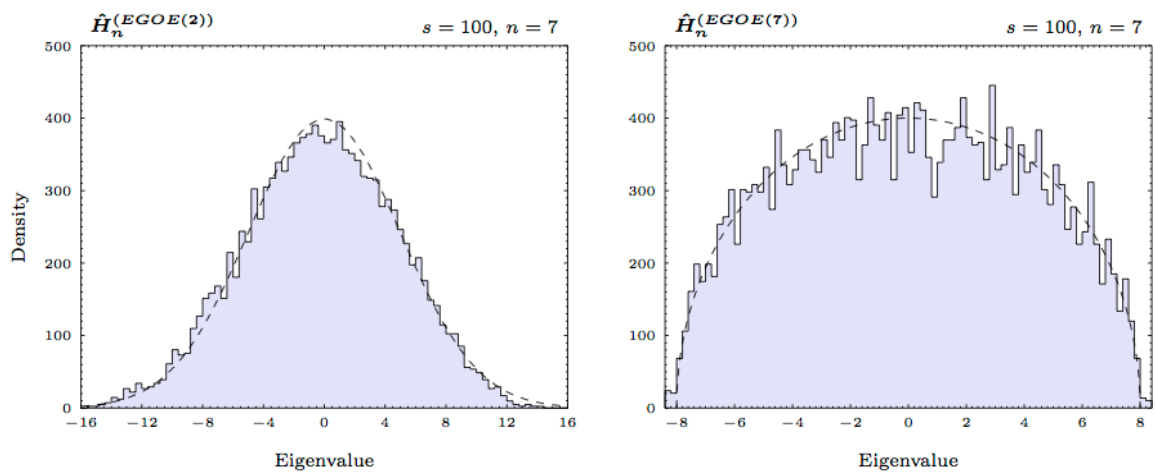}
  \caption[Spectral density for the EGOE($k$)]{Adapted from \cite{French1970}.  The average spectral density over $s=100$ samples from a realisation of the EGOE($k=2$) (left) and EGOE($k=7$) (right) on $n=7$ indistinguishable particles.  A clear change from a Gaussian distribution (left dashed) to a semi-circular distribution (right dashed) is seen.}
  \label{French}
\end{figure}

\begin{figure}
  \centering
  \input{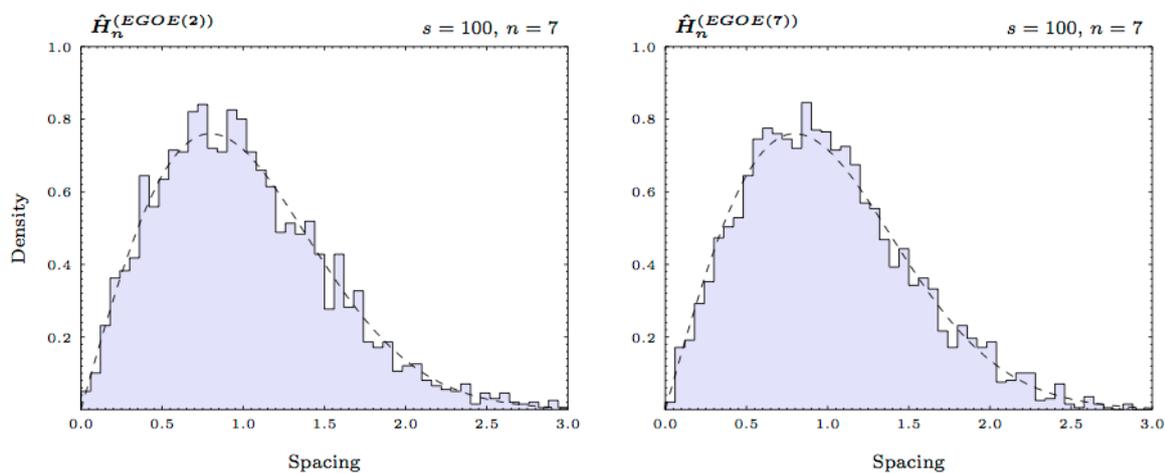}
  \caption[Nearest-neighbour level spacing statistics for the EGOE($k$)]{Adapted from \cite{French1971}.  The average unfolded nearest-neighbour level spacing distributions over $s=100$ samples from a realisation of the EGOE($k=2$) (left) and EGOE($k=7$) (right) on $n=7$ indistinguishable particles.  Only spacings near the centre of the spectrum were used.  A clear similarity to the approximate limiting GOE nearest-neighbour level spacing distribution (dashed) is seen. }
  \label{FrenchSpacings}
\end{figure}

\subsection{Explicit ensemble construction}
\subsubsection{Hilbert space}
The EGOE($k$) is an ensemble of random matrices which act on the space of $n$ identical particles (fermions or bosons) each with $d$ individual single particle states.  In the case of fermions (an analogous construction holds for bosons) the $n$-particle states
\begin{equation}
  \a_{s_{1}}^\dagger\dots \a_{s_{n}}^\dagger|\boldsymbol{0}\rangle_{\a}=|s_{1},\dots,s_{n}\rangle_{n}=|\boldsymbol{s}\rangle_{n}
\end{equation}
where $\boldsymbol{s}\in\{1,\dots,d\}^{n}$, $s_{1}<\dots<s_{n}$, $n<d$ and $\a_{j}^\dagger$ are the Fermi creation matrices of Appendix \ref{JWT} with the vacuum state $|\boldsymbol{0}\rangle_{\a}$, form an orthonormal basis of this space.  Here the values of $s_1,\dots,s_n$ can be regarded as state occupation labels of the $n$ fermions.

The $n$-particle states $\a_{t_{1}}^\dagger\dots \a_{t_{n}}^\dagger|\boldsymbol{0}\rangle_{\a}=|\boldsymbol{t}\rangle_{n}$ for any $\boldsymbol{t}\in\{1,\dots,d\}^{n}$ can be related to the basis states by the canonical commutation relations for the Fermi matrices $\a_{j}$ and $\a_{j}^\dagger$, see Appendix \ref{JWT}.

\subsubsection{Ensemble elements}
The elements of the EGOE($k$) for fermions are defined \cite{Kota2001} to be
\begin{equation}
  \hat{H}_{n}^{(EGOE(k))}=\sum_{\genfrac{}{}{0pt}{}{1<s_{1}<\dots<s_{n}\leq d}{1<s_{1}^\prime<\dots<s_{n}^\prime\leq d}}\hat{c}_{\boldsymbol{s}^\prime,\boldsymbol{s}}|\boldsymbol{s}^\prime\rangle_{n}\,_{n}\langle \boldsymbol{s}|
\end{equation}
For $k=2$, the random variables $\hat{c}_{\boldsymbol{s}^\prime,\boldsymbol{s}}$ are defined through the random $k=2$ body interaction Hamiltonian
\begin{equation}
  \hat{H}_{2}^{(int)}=\sum_{\genfrac{}{}{0pt}{}{1<r_{1}<r_{2}\leq d}{1<r_{1}^\prime<r_{2}^\prime\leq d}}\hat{d}_{\boldsymbol{r}^\prime,\boldsymbol{r}}|\boldsymbol{r}^\prime\rangle_{2}\,_{2}\langle \boldsymbol{r}|
\end{equation}
where the real coefficients $\hat{d}_{\boldsymbol{r},\boldsymbol{r}^\prime}=\hat{d}_{\boldsymbol{r}^\prime,\boldsymbol{r}}$ are independent normal random variables.  This matrix describes a `random' interaction between two fermions.

This $k=2$ body interaction is then built into the Hamiltonian $\hat{H}_{n}^{(EGOE(k))}$ in the following way.  For any $n$-particle states $|\boldsymbol{t}\rangle_{n}$ and $|\boldsymbol{t}^\prime\rangle_{n}$ as above,  if $\boldsymbol{t}=\boldsymbol{t}^\prime$ then
\begin{equation}
  \,_{n}\langle\boldsymbol{t}^\prime|\hat{H}_{n}^{(EGOE(2))}|\boldsymbol{t}\rangle_{n}=\sum_{1\leq j<k\leq n}\,_{2}\langle t_{j},t_{k}|\hat{H}^{(int)}_{2}|t_{j},t_{k}\rangle_{2}
\end{equation}
If $t_{1}\neq t_{1}^\prime$ and $t_{j}=t_{j}^\prime$ for all $j=2,\dots,n$ then
\begin{equation}
   \,_{n}\langle\boldsymbol{t}^\prime|\hat{H}_{n}^{(EGOE(2))}|\boldsymbol{t}\rangle_{n}=\sum_{1<j\leq n}\,_{2}\langle t_{1}^\prime,t_{j}|\hat{H}^{(int)}_{2}|t_{1},t_{j}\rangle_{2}
\end{equation}
If $t_{1}\neq t_{1}^\prime$, $t_{2}\neq t_{2}^\prime$ and $t_{j}=t_{j}^\prime$ for all $j=3,\dots,n$ then
\begin{equation}
  \,_{n}\langle\boldsymbol{t}^\prime|\hat{H}_{n}^{(EGOE(2))}|\boldsymbol{t}\rangle_{n}=\,_{2}\langle t_{1}^\prime,t_{2}^\prime|\hat{H}^{(int)}_{2}|t_{1},t_{2}\rangle_{2}
\end{equation}

By permuting the Fermi matrices in the definition of the states $|\boldsymbol{t}\rangle_{n}$, these three cases give many of the random variables $\hat{c}_{\boldsymbol{s}^\prime,\boldsymbol{s}}$.  The rest are defined to be zero.  This construction only models the two-body interactions described by $\hat{H}_{2}^{(int)}$.  Any higher-body interactions are not present by construction, that is, many of the $ \hat{c}_{\boldsymbol{s}^\prime,\boldsymbol{s}}$ are zero as such terms require higher-body interactions to go between the corresponding states in state space.

An analogous procedure is used to define the EGOE($k$) for $k\neq2$ and the $k$-body embedded Gaussian unitary ensembles (EGUE($k$)) by either enlarging the dimension ($2$ to $k$) or changing the symmetry (orthogonal to unitary) of the interaction Hamiltonian $\hat{H}_{2}^{(int)}$ above, respectively.

\subsection{Distinguishable particles}
\subsubsection{Pi\v zorn, Prosen, Mossmann and Seligman}
Distinguishable two level particles (qubits) have also been considered.  Pi\v zorn, Prosen, Mossmann and Seligman \cite{Prosen2007} numerically study the following Hamiltonian on a chain of $n$ (even, to avoid degeneracy issues) qubits
\begin{equation}\label{Pro}
  \hat{H}_{n}^{(PPMS)}=\sum_{j=1}^{n-1}\sum_{a,b=1}^{3}\hat{\alpha}_{a,b,j}\sigma_{  j  }^{(a)}\sigma_{  j+1  }^{(b)}+\lambda\sum_{j=1}^{n}\sum_{a=1}^{3}\hat{\alpha}_{a,0,j}\sigma_{  j  }^{(a)}
\end{equation}
where for each $n$ separately, the coefficients $\hat{\alpha}_{a,b,j}$ are real identically distributed Gaussian random variables with zero mean and where $\lambda$ is a real control parameter.  Here, only $k=2$ particle interactions  between neighbouring qubits on the chain are present.  This model may be also be generalised from a chain to any graph of $n$ qubits.

Pi\v zorn, Prosen, Mossmann and Seligman observe a new type of phase transition in the ground state of this model at $\lambda=0$.  Two well defined effects of this are seen which are relevant here.  One in the spectral gap between the lowest energy level and the second lowest energy level and the second in the entropy of entanglement between two equal halves of the chain.

Pi\v zorn and his collaborators observe that critical chains correspond to a vanishing spectral gap, $g$, in the large chain limit.  For the chain above they find that for $\lambda\neq0$, $\langle g\rangle\sim n^{-\eta}$ where $\eta\approx0.39\pm0.01$ and for $\lambda=0$, that perhaps $\langle g\rangle\sim \e^{-\zeta n}$ where $\zeta\approx0.07\pm0.02$.  Both values are strictly critical but a sharp transition is seen between them.  Polynomial decay (as opposed to exponential decay) in the spectral gap can lead to the model exhibiting non-critical behaviours too.

Moreover, the distribution of the spectral gap in the $\lambda=0$ case numerically approximates a Poisson distribution for $n=10,14,16$ whereas for $\lambda=1$ the spectral gap closely matches that of the approximate limiting GUE nearest-neighbour level spacing distribution for $n=10,16$, see Figure \ref{Prosen}.

For a bi-partition of the chain into two equal lengths, it has been seen in Section \ref{etangChain} that for critical chains the entropy of entanglement, $S$, grows logarithmically with $n$ whereas for non-critical chains $S$ saturates at some value.  Pi\v zorn and his collaborators observe this in the model above.  With $\lambda\neq0$ the entropy numerically saturates and for $\lambda=0$ the entropy is approximately $c\log_{2}(n)$ up to an additive constant, where $c=0.17\pm0.02$, see Figure \ref{Prosen}.

\begin{figure}
  \centering
  \input{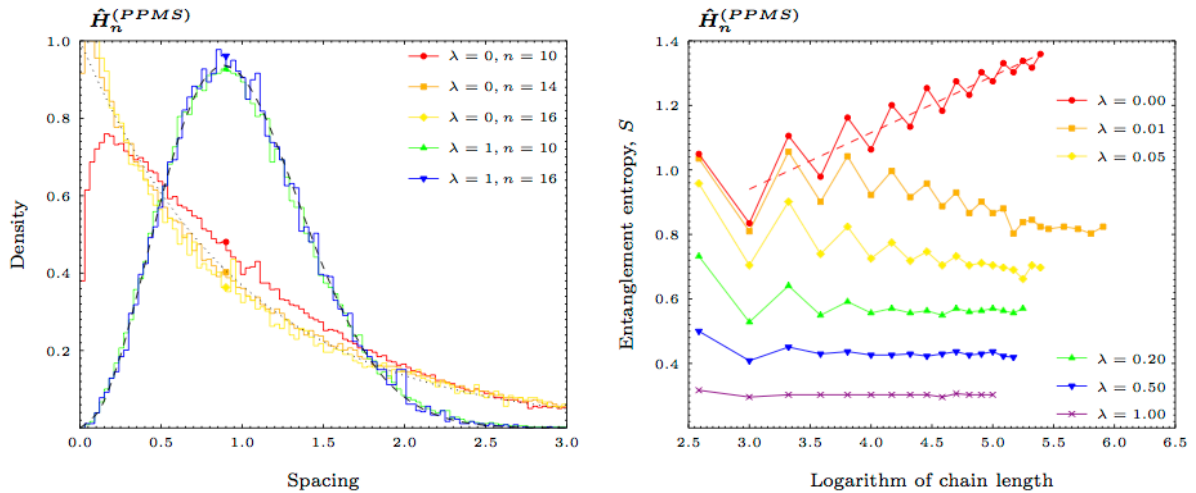}
  \caption[Phase transition in a qubit chain]{(Left) Adapted from \cite{Prosen2007}.  The distribution of the spectral gap in $\hat{H}_{n}^{(PPMS)}$.  A phase transition from approximate limiting GUE (dashed) nearest-neighbour level spacing statistics to Poisson (dotted) statistics is seen.  (Right) Adapted from \cite{Prosen2007}. The entropy of entanglement on half the chain of the ground state of $\hat{H}_{n}^{(PPMS)}$ against the logarithm of chain length.  A phase transition from saturation to logarithmic growth of $S$ is seen at $\lambda=0$.}
  \label{Prosen}
\end{figure}

\subsubsection{Movassagh and Edelman}
Movassagh and Edelman \cite{Movassagh2011} also propose a method to approximate the spectral density of such systems, albeit with numerically verified accuracy.  They consider models such as $\hat{H}_{n}^{(PPMS)}$.  This Hamiltonian may be written as $\hat{H}_{n}^{(PPMS)}=\hat{A}+\hat{B}$ where the matrices
\begin{align}
 	\hat{A}&=\sum_{\genfrac{}{}{0pt}{}{j=1}{j\text{ even}}}^{n-1}\sum_{a,b=1}^{3}\hat{\alpha}_{a,b,j}\sigma_{j}^{(a)}\sigma_{j+1}^{(b)}+\lambda \sum_{\genfrac{}{}{0pt}{}{j=1}{j\text{ even}}}^{n}\sum_{a=1}^{3}\hat{\alpha}_{a,0,j}\sigma_{j}^{(a)}\nonumber\\
 	\hat{B}&=\sum_{\genfrac{}{}{0pt}{}{j=1}{j\text{ odd}}}^{n-1}\sum_{a,b=1}^{3}\hat{\alpha}_{a,b,j}\sigma_{j}^{(a)}\sigma_{j+1}^{(b)}+\lambda\sum_{\genfrac{}{}{0pt}{}{j=1}{j\text{ odd}}}^{n}\sum_{a=1}^{3}\hat{\alpha}_{a,0,j}\sigma_{j}^{(a)}
\end{align}
are both computationally easy to diagonalise (by local diagonalisation).  The true spectral density of $\hat{H}_{n}^{(PPMS)}$ is then approximated by linear interpolation between the spectral densities obtained from first treating $\hat{A}$ and $\hat{B}$ as commuting operators and then as operators with a faithful spectrum but with Haar distributed eigenstates (free convolution).  The interpolation is chosen to match the first four moments of the true spectral density.  Movassagh and Edelman find, numerically, that this approximation is remarkably good.
\section{Gaussian unitary ensemble}\label{GUE}
To highlight some of the key techniques commonly used in random matrix theory, the Gaussian unitary ensemble (GUE) will be briefly analysed.  An overview of these methods can be found in {\cite[\p31]{Snaith}}, \cite[\p66,103]{Handbook} and \cite[\p354]{Mehta}, with their historical significance noted in Section \ref{RMT}.  The techniques here can be extended and applied to many other random matrix ensembles.

The GUE consists of all $N\times N$ Hermitian matrices $H$ endowed, up to scaling, with the probability measure
\begin{align}
  \boldsymbol{\di}\hat{\mu}_{N}^{(GUE)}(H)&=C\e^{-\Tr \left(H^{2}\right)}\boldsymbol{\di}H\nonumber\\
  &\equiv C\e^{-\sum_{j}h_{j,j}^{2}-2\sum_{k<l}\left( h_{k,l}^{2}+ {h^\prime_{k,l}}^{2}\right)}\prod_{j=1}^{N}\di h_{j,j}\prod_{1\leq k<l\leq N}\di  h_{k,l}\di  h^\prime_{k,l}
\end{align}
where $C$ is some normalisation constant and $h_{j,j}=H_{j,j}$, $h_{k,l}=\Rp\left(H_{k,l}\right)$ and $h^\prime_{k,l}=\Ip\left(H_{k,l}\right)$ are $N^{2}$ real parameters.  That is, the real parameters of the elements of $H$ are independent Gaussian random variables.  By this definition, the GUE is a Wigner ensemble, its probability measure is also invariant under the unitary conjugation of $H$ (see Appendix \ref{Lebesgue}), so that it is also a unitarily invariant ensemble.

\subsection{Diagonalisation and the Vandermonde determinant}\label{GUEdiag}
Any $N\times N$ Hermitian matrix $H$ can be diagonalised as $H=U\Lambda U^\dagger$, where $U$ is a unitary matrix with $N^{2}$ parameters (see Appendix \ref{unitaryParam}) whose columns are the $N$ orthonormal eigenstates of $H$ in some arbitrary order and $\Lambda$ is a diagonal matrix formed from the $N$ corresponding real eigenvalues of $H$.

Parameterising the Hermitian matrices via the matrices $U$ and $\Lambda$ provides some `over covering' of the Hermitian matrices due to the arbitrary order chosen for the eigenvalues of $H$, arbitrary complex phase factors of the eigenvalues and possible degeneracies in the spectrum.  For example, if $H$ has the eigenvalues $\lambda_{k}$ and eigenvectors $|\psi_{k}\rangle$ (listed in some arbitrary order) so that $H=\sum_{k}\lambda_{k}|\psi_{k}\rangle\langle\psi_{k}|=U\Lambda U^\dagger$, the complex phases of $|\psi_{k}\rangle$ and order of eigenvalues and associated eigenstates can be chosen arbitrarily.

Uniformly choosing over all these redundancies allows for the $N^{2}$ dimensional probability measure $\boldsymbol{\di}\hat{\mu}_{N}^{(GUE)}(H)$ to be written equivalently, up to a normalisation constant $C_{1}$, as the $N^{2}+N$ dimensional probability measure\footnote{As the complex phases of the columns of $U$ are redundant in the parametrisation of $H$, the related $N$ parameters may be dropped from the parametrisation so that only $N^{2}-N$ parameters need be used to describe the $N$ orthonormal eigenstates of $H$, see Appendix \ref{unitaryParam}.}
\begin{equation}
  C_{1}\e^{-\sum_{j}\lambda_{j}^{2}}\prod_{1\leq k<l\leq N}(\lambda_{l}-\lambda_{k})^{2}\di\lambda_{1}\dots\di\lambda_{N}\boldsymbol{\di}U
\end{equation}
Here $\boldsymbol{\di}U$ is the normalised Haar measure over all the $N\times N$ unitary matrices.  The appearance of Haar measure stems from the fact that the original probability measure on $H$ is invariant under the conjugation $VHV^\dagger$ by any unitary matrix $V$, see Appendix \ref{Lebesgue}.  The product $\prod_{k<l}(\lambda_{l}-\lambda_{k})^{2}$ is the squared Vandermonde determinant, the Jacobian of this variable transformation.

It is seen that the eigenvalues and eigenstates are statistically independent.  Integrating out the parameters corresponding to the eigenstates ($\boldsymbol{\di}U$) leaves the joint probability density function of the (unordered) eigenvalues of $H$ as,
\begin{equation}
  \hat{\rho}^{(GUE)}_{N,N}(\lambda_{1},\dots,\lambda_{N})=C_{1}\e^{-\sum_{j}\lambda_{j}^{2}}\prod_{1\leq k<l\leq N}(\lambda_{l}-\lambda_{k})^{2}
\end{equation}
The $r$-point correlation functions are then defined as
\begin{equation}
  \hat{\rho}^{(GUE)}_{N,r}(\lambda_{1},\dots,\lambda_{r})=\int_{\mathbb{R}^{N-r}}\hat{\rho}^{(GUE)}_{N,N}(\lambda_{1},\dots,\lambda_{N})\di\lambda_{r+1}\dots\di\lambda_{N}
\end{equation}
for $r=1,\dots,N-1$.  Here $1$-point correlation function corresponds to the spectral density for the GUE.

\subsection{Orthogonal polynomials}
To calculate spectral statistics of the GUE, Mehta pioneered the used of orthogonal polynomials.  As sketched in Appendix \ref{stdRMT}, any joint probability density function of the form
\begin{equation}
  C_{2}\left(\prod_{j=1}^{N}\omega(\lambda_{j})\right)\prod_{1\leq k<l\leq N}(\lambda_{l}-\lambda_{k})^{2}
\end{equation}
where $w(\lambda)$ is a real positive function on the real interval $J$ such that $\int_{J} x^{m}\omega(x)\di x<\infty$ for all $m\in\mathbb{N}_{0}$, can be written as
\begin{equation}\label{polyJDOS}
  C_{2}\left(\prod_{j=1}^{N-1}(p_{j},p_{j})\right)\det_{1\leq k,l\leq N}\Big(K_{N}(\lambda_{k},\lambda_{l})\Big)
\end{equation}
where the kernel $K_{N}$ is defined to be
\begin{equation}
  K_{N}(x,y)=\omega^{\frac{1}{2}}(x)\omega^{\frac{1}{2}}(y)\sum_{j=0}^{N-1}\frac{p_{j}(x)p_{j}(y)}{(p_{j},p_{j})}
\end{equation}
for the monic orthogonal polynomials $p_{j}$ of degree $j$, which are orthogonal with respect to the inner-product
\begin{equation}
  (p_{j},p_{k})=\int_{J} w(x)p_{j}(x)p_{k}(x)\di x
\end{equation}
For the GUE the corresponding orthogonal polynomials on $J=\mathbb{R}$ are the monic Hermite polynomials,
\begin{equation}
  q_{j}(x)=\frac{(-1)^{j}\e^{x^{2}}}{2^{j}}\frac{\di^{j}}{\di x^{j}}\e^{-x^{2}}
\end{equation}

\subsection{Correlation functions}\label{GUEcorr}
A restricted version of Gaudin's lemma {\cite[\p126]{Snaith}} states that if $J$ is an interval on the real line and $f:J\times J\to\mathbb{R}$ such that
  \begin{equation}
  	\int_{J}f(x,y)f(y,z)\di y=f(x,z)\qquad\text{and}\qquad\int_{J}f(x,x)\di x=N
	\end{equation}
	for some real constant $N$, then
	\begin{equation}
		\int_{J}\det_{1\leq k,l\leq r+1}\Big(f(\lambda_{k},\lambda_{l})\Big)\di\lambda_{r+1}=(N-r)\det_{1\leq k,l\leq r}\Big(f(\lambda_{k},\lambda_{l})\Big)
	\end{equation}
As $f=K_{N}^{(GUE)}$ (the corresponding kernel for the GUE) satisfies the conditions of Gaudin's lemma (see Appendix \ref{stdRMT}), this now allows the expression for the joint density function in (\ref{polyJDOS}) to be integrated, and therefore gives the $r$-point correlation functions as
\begin{equation}
  \frac{(N-r)!}{N!}\det_{1\leq k,l\leq r}\big(K_{N}(\lambda_{k},\lambda_{l})\big)
\end{equation}
The $1$-point correlation function (or spectral density) is then given by $\frac{1}{N}K_{N}(\lambda,\lambda)$, which, in the particular case of the GUE is asymptotically equal to the semicircle law
\begin{equation}
  \hat{\rho}_{N,1}^{(GUE)}(\lambda)\sim\frac{\sqrt{2}}{\pi\sqrt{N}}\begin{cases}
    \sqrt{1-\frac{\lambda^{2}}{2N}}\qquad\qquad&\text{if }|\lambda|\leq\sqrt{2N}\\
    0&\text{else }
  \end{cases}
\end{equation}
as $N\to\infty$.  Figure \ref{GUEspec} shows the plot of the function $\hat{\rho}^{(GUE)}_{N,1}$ for $N=2^{3}=8$ and $N=2^{4}=16$.  The oscillation arising from the orthogonal polynomials that the densities are constructed from are clearly seen.  Similar oscillations are seen for the numerical spectral densities in the spin chain ensembles given in Appendix \ref{specHist}.

\begin{figure}
	\centering
	\input{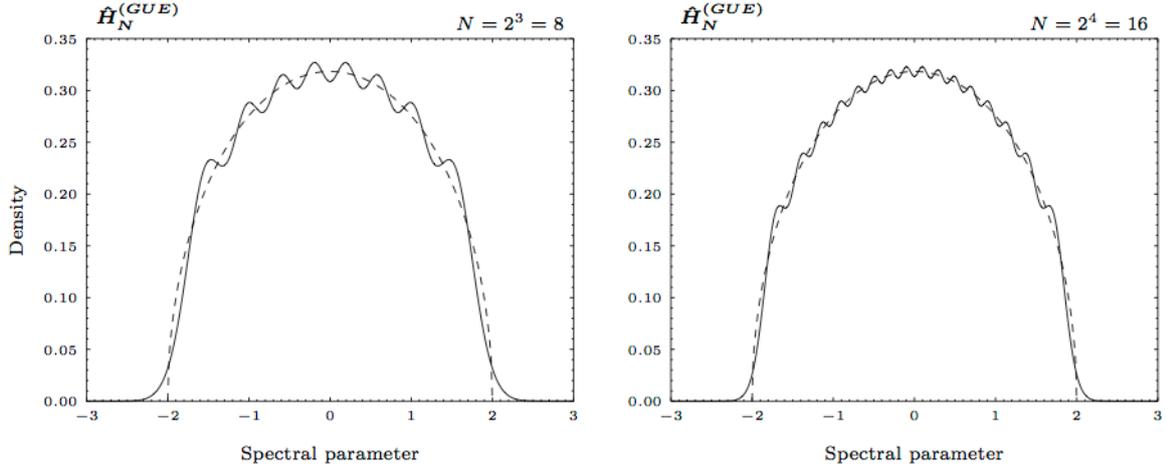}
	\caption[The spectral density for the GUE for $N=2^{3}$ and $N=2^{4}$]{The spectral density of the GUE for $N=2^{3}$ and $N=2^{4}$.  The plot is of the rescaled scaled density $\sqrt{\frac{N}{2}}\hat{\rho}^{(GUE)}_{N,1}\left(\lambda\sqrt{\frac{N}{2}}\right)$, so that it has unit variance.  The dashed line gives the semicircle law, similarly rescaled.}
	\label{GUEspec}
\end{figure}

\subsection{Level repulsion}\label{GUElevel}
The results from the orthogonal polynomial method also enable the tenancy of the eigenvalues of the GUE to repel each other to be investigated.  The Christoffel-Darboux formula \cite[\p52]{Snaith}
\begin{equation}
		\sum_{j=0}^{N-1}\frac{p_{j}(x),p_{j}(y)}{(p_{j},p_{j})}=\frac{p_{N}(x)p_{N-1}(y)-p_{N-1}(x)p_{N}(y)}{(p_{N-1},p_{N-1})(x-y)}
\end{equation}
for the general monic orthogonal polynomials $p_{j}$, together with the asymptotic form of the monic orthogonal Hermite polynomials $q_{j}$ \cite[\p194]{Szego},
\begin{equation}
	\e^{-\frac{x^{2}}{2}}q_{j}(x)=\frac{1}{\sqrt{\pi}}\Gamma\left(\frac{j+1}{2}\right)\cos\left(x\sqrt{2j+1}-\frac{j\pi}{2}\right)+O\left(\frac{1}{\sqrt{j}}\right)
	\end{equation}
as $j\to\infty$, allows the kernel for the GUE to be approximated by
\begin{equation}
  K_{N}^{(GUE)}(x,y)\sim\frac{\sin\left(\sqrt{2N}\left(x-y\right)\right)}{\pi(x-y)}
\end{equation}
for fixed values of $x$ and $y$, as outlined in Appendix \ref{stdRMT}.  This is known as the sine kernel.  It in fact holds for the bulk of the spectrum and is found in a variety of different ensembles.  A different kernel holds at the edge of the support of the limiting spectral density.

By the formula for the $2$-point correlation function above, it then follows that
\begin{equation}
  \hat{\rho}_{N,2}^{(GUE)}(\lambda_{1},\lambda_{2})=\frac{1}{N(N-1)}\det_{1\leq k,l\leq 2}\big(K_{N}^{(GUE)}(\lambda_{k},\lambda_{l})\big)
\end{equation}
Using the previous asymptotic formula for the kernel yields for suitably large $N$ and suitably small spacings $\lambda_{1}-\lambda_{2}$ in some fixed window, that $\hat{\rho}_{N,2}^{(GUE)}(\lambda_{1},\lambda_{2})$ is approximately quadratic in the small spacing $\lambda_{1}-\lambda_{2}$.  Again, a sketch of this calculation is given in Appendix \ref{stdRMT}.  This precisely shows the quadratic repulsion of eigenvalues that are close to each other.  A similar quadratic repulsion has already been seen in Section \ref{spacingStat} where the case $N=2$ was explicitly calculated.

\subsection{Entanglement of eigenstates}
Due to the unitary invariance of the GUE, the orthonormal eigenstates of an ensemble member are uniformly distributed (Haar measure) over the sphere of unit radius in $\mathbb{C}^{N}$.

Let $N=n_{1}n_{2}$ so that the Hilbert space $\mathbb{C}^{N}$ can be decomposed into two portions, $\mathcal{A}=\mathbb{C}^{n_{1}}$ and $\mathcal{B}=\mathbb{C}^{n_{2}}$ with $\mathbb{C}^{N}\equiv\mathcal{A}\otimes\mathcal{B}$.  Assuming, without loss of generality, that $n_{1}\leq n_{2}$, the reduced density matrix on $\mathcal{A}$ of an eigenstate of a matrix $H$, with non-degenerate eigenvalue, from the GUE can be considered.  \v{Z}nidari\v{c} \cite{Znidaric2007} shows that for large $n_{2}$ the expected purity of this reduced density matrix is almost minimal with deviations going to zero as $n_{2}$ increases for fixed $n_{1}$.  This shows a high likelihood of the eigenstates of a member of the GUE being close to maximally entangled across $\mathcal{A}$ and $\mathcal{B}$ in this regime as $n\to\infty$.

\section{Plan}
The rest of this work will be split into the following sections:

In Chapter \ref{DOS} the spectral density for generic quantum spin chains will be considered.  First in Section \ref{Random matrix model} the ensembles describing the most general qubit nearest-neighbour chains will be constructed.  In Section \ref{PointCorrFunc} the definitions of the spectral probability measures, corresponding spectral densities and characteristic functions, for the ensembles constructed, and their specific members, will be given.  The $r$-point densities or correlation functions will also be defined here.  The results of numerical simulations of the spectral densities of relevant ensembles will be given in Section \ref{Numerics}.  In Section \ref{Moment method} the moments of the spectral probability measure in the long chain limit are calculated for a particular ensemble.  This is extended to the complete limiting spectral density for a wide range of ensembles in Section \ref{Separating odd and even sites}.  Finally the convergence to this limiting spectral density will be seen to hold on the level of individual ensemble members, that is specific sequences of spin chain Hamiltonians, in Section \ref{Splitting method}.

Entanglement within the eigenstates of some general quantum spin chains will be studied in Chapter \ref{Eigenstate entanglement}.  The definition of state purity and of the translation matrix shall be given in Section \ref{EigenDef} as these will be the main tools for this analysis.  Section \ref{eigenNum} contains the results of numerical simulations of the relevant ensembles.  The entanglement in each eigenstate of a variety of fixed long quantum spin chains, between one chain site and the rest of the chain, is then analysed in Section \ref{SingleEnt}.  In Section \ref{EntLblock} this calculation is extended to more sites, in the case of long translationally-invariant chains.

The eigenstatistics of finite length chains are then considered in Chapter \ref{finite}.  First, in Section \ref{DOSrate}, the rate of convergence of the spectral density is analysed to complement the limiting density results of Chapter \ref{DOS}.  Furthermore, as the entanglement results of Chapter \ref{Eigenstate entanglement} are made especially relevant if spectral degeneracies are not present, the occurrence of such degeneracies is studied in Section \ref{degensec}.  This analysis provides intuition for the form of the nearest-neighbour level spacing statistics, observed numerically and presented in Section \ref{nnnumerics}.  Going beyond the spectral density and spacing statistics for the chain ensembles, the form of the general joint spectral densities for more general ensembles is conjectured, and a heuristic argument given, in Section \ref{JDOS}.  This conjectured form is seen to closely predict the numerical results in an accessible ensemble.  Finally, in Section \ref{eigenstatebounds}, the tightness of the purity bounds from the eigenstate analysis in Chapter \ref{Eigenstate entanglement} are confirmed, in at least some cases.

In the last chapter, the preceding results are collated and their relevance to one another summarised.  In doing so, some open problems for further research are highlighted.

\clearpage
   
\chapter{Spectral density}\label{DOS}
In this chapter the spectral density for an ensemble of random matrices, and its individual members, consisting of the most general Hamiltonians describing a nearest-neighbour qubit chain (see Figure \ref{10QubitRing}) will be studied.

\begin{figure}
\centering
	\input{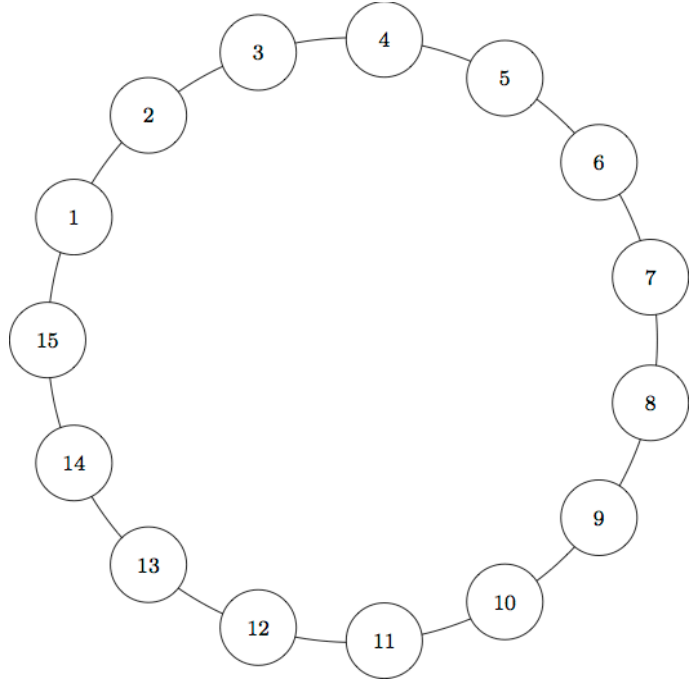}
	\caption[A nearest-neighbour qubit chain]{A nearest-neighbour qubit chain of $n=15$ qubits.  The fifteen qubits, denoted by circles labelled $1$ to $15$, are arranged in a ring.  Interactions (connecting lines) are only allowed between nearest-neighbours.}
	\label{10QubitRing}
\end{figure}

First a candidate for such an ensemble will be defined in Section \ref{Random matrix model} and its spectral density defined in Section \ref{PointCorrFunc} along with the related distribution or probability measure to this density.  Next, in Section \ref{Numerics}, the results of numerical simulations will be presented.  These results suggest that the limiting spectral density for large $n$ (large spin chain length) is a Gaussian distribution.  Then, as Wigner did for the real symmetric Wigner matrices with independent Gaussian entries, each moment of the spectral density will be calculated in the large $n$ limit.  These moments will be seen to be those of a Gaussian distribution, in Section \ref{Moment method}.

To recover the limiting distribution however, another technique will be used.  The random spin chain Hamiltonian will be split into many statistically independent portions in Section \ref{Separating odd and even sites}, to which the central limit theorem will be applied.  Here, the convergence of the characteristic function of the spectral probability measure will be determined and the continuity theorem used to imply the weak convergence of the spectral probability measure to that of a Gaussian.  The use of the central limit theorem will then enable a universality result to be proved in Section \ref{UniDOS}, that is, the limiting distribution will be shown to be independent of the distributions used to define the ensemble, up to some reasonable conditions.  The extension to the most general ensemble and to more general interaction geometries will also be made in Sections \ref{momentslocalterms} and \ref{Ccolour} respectively.

In Section \ref{Splitting method} it will be shown that this convergence holds on the level of individual Hamiltonians.  The techniques presented by Hartmann, Mahler and Hess \cite{Mahler,Mahler2005} will be adapted for this.  Again the extension to more general interactions and, in this case, increased local dimension are made in Section \ref{MGen} and \ref{qudits} respectively.

The work presented in this chapter is based on that given by the author in \cite{KLW2014_RMT} and \cite{KLW2014_FIXED}.  Ideas initiated from discussions with \cite{Keating} and \cite{Linden} are highlighted where appropriate.
\section{Ensembles of generic qubit chain Hamiltonians}\label{Random matrix model}
Ensembles of generic qubit chain Hamiltonians will now be constructed.  All the Hamiltonians seen in Chapter \ref{Introduction} describing nearest-neighbour spin chains of $n$ qubits, labelled from $1$ to $n$, can be written in the form
\begin{equation}\label{genH}
  H_{n}^{(gen)}=\sum_{j=1}^{n}\sum_{a,b=1}^{3}\alpha_{a,b,j}\sigma_{  j  }^{(a)}\sigma_{  j+1  }^{(b)}+\sum_{j=1}^{n}\sum_{a=1}^{3}\beta_{a,j}\sigma_{  j  }^{(a)}+\gamma I_{2^{n}}
\end{equation}
where the $\alpha_{a,b,j}$, $\beta_{a,j}$ and $\gamma$ are some real coefficients, see Figure \ref{nnPauliChain}.  Let $\alpha^{(j)}$ denote the $3\times 3$ real matrix $\left(\alpha_{a,b,j}\right)_{a,b}$ and $\boldsymbol{\beta}^{(j)}$ denote the vector $\left(\beta_{1,j},\beta_{2,j},\beta_{3,j}\right)^{T}$.

\begin{figure}
  \centering
  \input{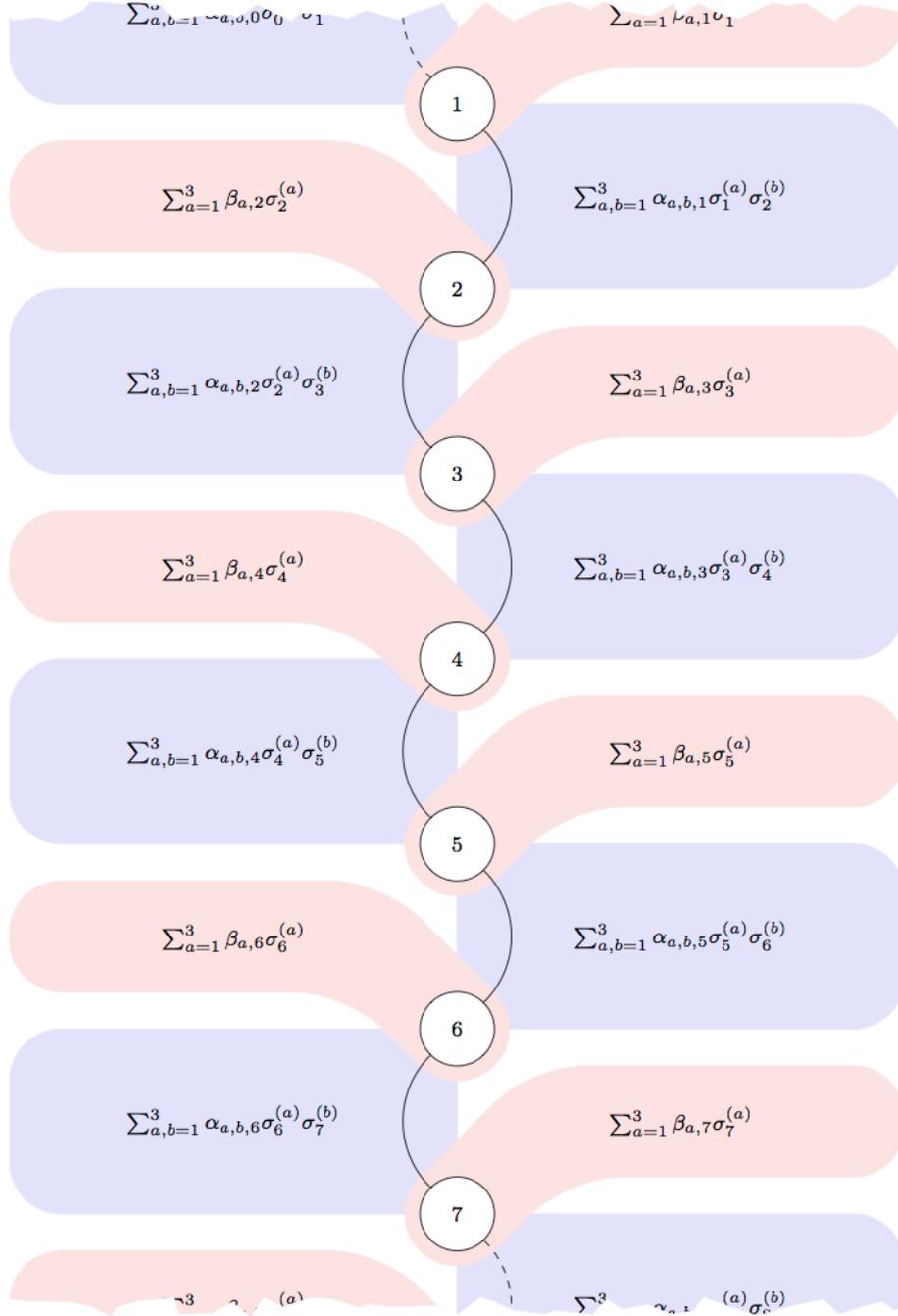}
  \caption[Pauli terms in a nearest-neighbour qubit chain]{A representation of the Pauli terms in a Hamiltonian describing a nearest-neighbour qubit chain.  The circles represent $7$ distinguishable qubits out of the $n$ in the chain.  The links represent the interactions between nearest-neighbours. There exist interaction terms $\sum_{a,b=1}^{3}\alpha_{a,b,j}\sigma_{  j  }^{(a)}\sigma_{  j+1  }^{(b)}$ between the qubits $j$ and $j+1$, local terms $\sum_{a=1}^{3}\beta_{a,j}\sigma_{  j  }^{(a)}$ corresponding to each qubit $j$ and a global energy shifting term $\gamma I_{2^{n}}$.}
  \label{nnPauliChain}
\end{figure}

The Heisenberg $XY$ chain
\begin{equation}
  \frac{J}{2}\sum_{j=1}^{n}\left(\left(1+\gamma\right)\sigma_{  j  }^{(1)}\sigma_{  j+1  }^{(1)}+\left(1-\gamma\right)\sigma_{  j  }^{(2)}\sigma_{  j+1  }^{(2)}\right)+\frac{h}{2}\sum_{j=1}^{n}\sigma_{  j  }^{(3)}
\end{equation}
for example, may be written in this form with
\begin{equation}
  \alpha^{(j)}=\frac{J}{2}\begin{pmatrix}
			      1+\gamma & 0 & 0 \\
			      0 & 1-\gamma & 0 \\
			      0 & 0 & 0
		\end{pmatrix}
  \qquad\qquad\beta^{(j)}=\frac{h}{2}\begin{pmatrix}
			      0  \\
			      0  \\
			      1
		\end{pmatrix}
  \qquad\qquad\gamma=0
\end{equation}

\subsection{Parametrisation of the most general qubit chain}\label{GenChain}
In the most general setting, let the interaction between the neighbouring qubits $j$ and $j+1$ (with cyclic labelling) be described by the arbitrary Hamiltonian $h_{j}$  (a $4\times 4$ Hermitian matrix acting on the Hilbert space of two qubits, $\left(\mathbb{C}^{2}\right)^{\otimes 2}$).  In Section \ref{PauliMatrixBasis} it has been seen that $h_{j}$ may be decomposed as
\begin{equation}\label{2.6}
  h_{j}=\frac{1}{4}\sum_{a,b=0}^{3}\Tr\left(h_{j}\sigma^{(a)}\otimes\sigma^{(b)}\right)\sigma^{(a)}\otimes\sigma^{(b)}
\end{equation}
The complete Hamiltonian of the $n$-qubit system (a $2^{n}\times2^{n}$ Hermitian matrix acting on the Hilbert space of all $n$ qubits, $\left(\mathbb{C}^{2}\right)^{\otimes n}$) is then the sum of all these individual Hamiltonians acting on the two corresponding qubits and leaving the rest undisturbed.  The most general Hamiltonian describing a nearest-neighbour qubit chain, see Figure \ref{nnChain}, is therefore written as
\begin{equation}
  \sum_{j=1}^{n-1}I_{2}^{\otimes(j-1)}\otimes h_j\otimes I_{2}^{\otimes {n-j-1}}+\text{boundary term}
\end{equation}
where the boundary term consists of $h_n$ acting non-trivially on the qubits labelled $n$ and $1$, in a cyclic fashion compared with the preceding terms.

\begin{figure}
  \centering
  \input{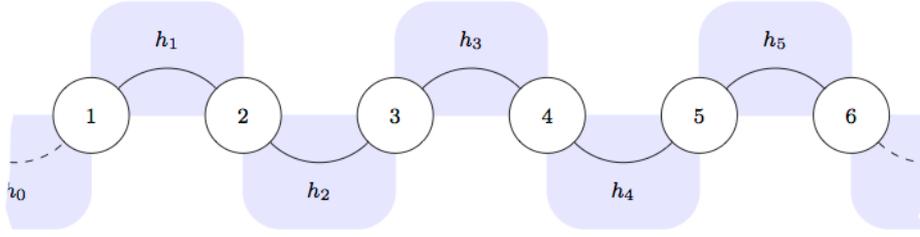}
  \caption[The most general nearest-neighbour qubit chain]{A representation of the most general nearest-neighbour qubit chain.  The circles represent $6$ qubits out of the $n$ in the chain and the links represent the interactions between nearest-neighbours. Each two-party interaction is described by the arbitrary Hamiltonian $h_{j}$, that is an arbitrary $4\times4$ Hermitian matrix.}
  \label{nnChain}
\end{figure}

By reordering the sums here and using the decomposition (\ref{2.6}), this Hamiltonian may be rewritten in the form (\ref{genH}),
\begin{align}
  &\frac{1}{4}\sum_{j=1}^{n}\sum_{a,b=1}^{3}\Tr\left(h_{j}\sigma^{(a)}\otimes\sigma^{(b)}\right)\sigma_{  j  }^{(a)}\sigma_{  j+1  }^{(b)}\nonumber\\
  &\qquad+\frac{1}{4}\sum_{j=1}^{n}\sum_{a=1}^{3}\left(\Tr\left(h_{j}\sigma^{(a)}\otimes I_{2}\right)+\Tr\left(h_{j-1}I_{2}\otimes\sigma^{(a)} \right)\right)\sigma_{  j  }^{(a)}\nonumber\\
  &\qquad\qquad+\frac{1}{4}\sum_{j=1}^{n}\Tr\left(h_{j}I_2\otimes I_2\right)I_{2^{n}}
\end{align}
where
\begin{align}
  \alpha^{(j)}&=\frac{1}{4}\begin{pmatrix}
			      \Tr\left(h_{j}\sigma^{(1)}\otimes\sigma^{(1)}\right) & \Tr\left(h_{j}\sigma^{(1)}\otimes\sigma^{(2)}\right) & \Tr\left(h_{j}\sigma^{(1)}\otimes\sigma^{(3)}\right) \\
			      \Tr\left(h_{j}\sigma^{(2)}\otimes\sigma^{(1)}\right) & \Tr\left(h_{j}\sigma^{(2)}\otimes\sigma^{(2)}\right) & \Tr\left(h_{j}\sigma^{(2)}\otimes\sigma^{(3)}\right) \\
			      \Tr\left(h_{j}\sigma^{(3)}\otimes\sigma^{(1)}\right) & \Tr\left(h_{j}\sigma^{(3)}\otimes\sigma^{(2)}\right) & \Tr\left(h_{j}\sigma^{(3)}\otimes\sigma^{(3)}\right)
		\end{pmatrix}\nonumber\\
  \boldsymbol{\beta}^{(j)}&=\frac{1}{4}\begin{pmatrix}
			      \Tr\left(h_{j}\sigma^{(1)}\otimes I_{2}\right)+\Tr\left(h_{j-1}I_{2}\otimes\sigma^{(1)} \right)  \\
			      \Tr\left(h_{j}\sigma^{(2)}\otimes I_{2}\right)+\Tr\left(h_{j-1}I_{2}\otimes\sigma^{(2)} \right)  \\
			      \Tr\left(h_{j}\sigma^{(3)}\otimes I_{2}\right)+\Tr\left(h_{j-1}I_{2}\otimes\sigma^{(3)} \right)
		\end{pmatrix}\nonumber\\
  \gamma&=\frac{1}{4}\sum_{j=1}^{n}\Tr\left(h_j\right)
\end{align}
for a cyclic labelling of the $h_j$.

\subsection{The model}\label{model}
Since the most general nearest-neighbour qubit chain may be described by a Hamiltonian of the form
\begin{equation}
  \sum_{j=1}^{n}\sum_{a,b=1}^{3}\alpha_{a,b,j}\sigma_{  j  }^{(a)}\sigma_{  j+1  }^{(b)}+\sum_{j=1}^{n}\sum_{a=1}^{3}\beta_{a,j}\sigma_{  j  }^{(a)}+\gamma I_{2^{n}}
\end{equation}
the random matrix model to be considered will be the ensemble of all such matrices where the coefficients are all taken to be independently identically distributed Gaussian random variables $\hat{\alpha}_{a,b,j}$, $\hat{\beta}_{a,j}$ and $\hat{\gamma}$ respectively, with zero mean and some fixed variance $s_{n}^{2}$.

To simplify the upcoming calculations further, attention will, at first, focus on the smaller ensemble
\begin{equation}
	\hat{H}_{n}= \sum_{j=1}^{n}\sum_{a,b=1}^{3}\hat{\alpha}_{a,b,j}\sigma_{  j  }^{(a)}\sigma_{  j+1  }^{(b)}
\end{equation}
where all the random variables $\hat{\beta}_{a,j}$ and $\hat{\gamma}$ have been removed.  Let specific members of this reduced ensemble be denoted by
\begin{equation}
	H_{n}= \sum_{j=1}^{n}\sum_{a,b=1}^{3}\alpha_{a,b,j}\sigma_{  j  }^{(a)}\sigma_{  j+1  }^{(b)}
\end{equation}
for the real parameters $\alpha_{a,b,j}$ corresponding to the random variables $\hat{\alpha}_{a,b,j}$.  This simply reduces the dimension of the space of matrices under consideration by removing the `local' terms, along with the global energy shift term.  The reinclusion of these terms will be seen not to affect the calculations made throughout this chapter.

\subsection{Variance}
The variance $s_{n}^{2}$ of the random variables $\hat{\alpha}_{a,b,j}$ will be set so that the variance of the random eigenvalues $\hat{\lambda}_{k}$ of $\hat{H}_{n}$, is equal to unity for all $n$.  That is
\begin{equation}
  \mathbb{E}\left(\frac{1}{2^{n}}\sum_{k=1}^{2^{n}}\hat{\lambda}_{k}^{2}\right)=\left\langle\frac{1}{2^{n}}\sum_{k=1}^{2^{n}}\lambda_{k}^{2}\right\rangle=1
\end{equation}
where $\left\langle\cdot\right\rangle$ denotes the ensemble average over all members $H_{n}$ of $\hat{H}_{n}$, with eigenvalues $\lambda_{k}$ respectively.  The trace of the square of $H_{n}$ gives precisely the sum of the squared eigenvalues of $H_{n}$, therefore it is required that
\begin{equation}
  \left\langle\frac{1}{2^{n}}\Tr \left(H_{n}^{2}\right)\right\rangle=1
\end{equation}
Substituting the definition of $H_{n}$ and taking the average inside the trace gives the equivalent expression of
\begin{equation}\label{conVar}
 \frac{1}{2^{n}}\Tr\left(\sum_{j,j^\prime=1}^{n}\sum_{a,a^\prime,b,b^\prime=1}^{3}
      \left\langle\alpha_{a,b,j}
      \alpha_{a^\prime,b^\prime,j^\prime}\right\rangle
      \sigma_{  j  }^{(a)}
      \sigma_{  j+1  }^{(b)}
      \sigma_{  j^\prime  }^{( a^\prime )}
      \sigma_{  j^\prime+1  }^{( b^\prime )} \right) =1
\end{equation}
Only the diagonal terms in this expression are non-zero as the matrices $\sigma_{  j  }^{(a)}\sigma_{  j+1  }^{(b)}$ are all orthogonal under the Hilbert-Schmidt inner-product (see Section \ref{PauliMatrixBasis}).  This leaves the equivalent condition of
\begin{equation}
 \frac{1}{2^{n}}\Tr \left(s_{n}^{2}\sum_{j=1}^{n}\sum_{a,b=1}^{3}   
      \sigma_{  j  }^{(a)}
      \sigma_{  j+1  }^{(b)}
      \sigma_{  j  }^{( a )}
      \sigma_{  j+1  }^{( b )} \right) =1
\end{equation}
as $\left\langle\alpha_{a,b,j}^2\right\rangle=s_{n}^{2}$.  The product $\sigma_{  j  }^{(a)}\sigma_{  j+1  }^{(b)}\sigma_{  j  }^{( a )}\sigma_{  j+1  }^{( b )}$ is equal to the identity as $\sigma_{  j+1  }^{(b)}$ and $\sigma_{  j  }^{( a )}$ commute and $\left(\sigma_{  j  }^{(a)}\right)^{2}=I_{2^{n}}$.  Therefore, $\Tr\left(\sigma_{  j  }^{(a)}\sigma_{  j+1  }^{(b)}\sigma_{  j  }^{( a )}\sigma_{  j+1  }^{( b )}\right)=\Tr(I_{2^{n}})=2^{n}$ and the resulting condition requires that $s_{n}^{2}=\frac{1}{9n}$.

\section{Definitions and required results}\label{PointCorrFunc}
The spectral probability measure for ensembles of Hermitian matrices will now be defined, along with the related spectral density and characteristic function, using the example of the ensemble $\hat{H}_{n}$.   The $r$-point correlation functions will also be constructed.  The related quantities for specific ensemble members will also be presented.

\subsection{Spectral and \texorpdfstring{$r$}{r}-point probability measures for ensembles}\label{DOSDef}
Any member $H_{n}$ of the ensemble $\hat{H}_{n}$ may be diagonalised as
\begin{equation}
  H_{n}=U\Lambda U^\dagger
\end{equation}
where $\Lambda$ is a diagonal matrix of the form $\Lambda=\text{Diag}(\lambda_{1},\dots,\lambda_{2^{n}})$, $\lambda_{1},\dots,\lambda_{2^{n}}$ are the eigenvalues (in any order) of $H_{n}$ and $U$ is a unitary matrix whose columns are corresponding eigenstates of $H_{n}$.

The probability measure associated to $\hat{H}_{n}$ must induce a probability measure on the space of these unordered real numbers $\lambda_{1},\dots,\lambda_{2^{n}}$ and the $4^{n}$ real parameters $v_{1},\dots,v_{4^{n}}$ of the compact unitary group $\mathcal{U}(2^{n})$.  Any arbitrary choice of $U$ or ordering of the $\lambda_{k}$ being uniformly distributed over, so that reordering the eigenvalues has no effect on this measure.   This is precisely the case for the GUE, as seen in Section \ref{GUE}.  Let this induced probability measure be denoted by
\begin{equation}
  \boldsymbol{\di}\hat{\mu}_{n,2^{n}+4^{n}}\left(\lambda_{1},\dots,\lambda_{2^{n}},v_{1},\dots,v_{4^{n}}\right)
\end{equation}
Let the joint probability measure of the unordered eigenvalues of $H_{n}$ (equivalently the $2^{n}$-point probability measure) then be defined as the marginal probability measure
\begin{equation}
  \boldsymbol{\di}\hat{\mu}_{n,2^{n}}(\lambda_{1},\dots,\lambda_{2^{n}})=\int_{v_{1},\dots,v_{4^{n}}}\boldsymbol{\di}\hat{\mu}_{n,2^{n}+4^{n}}\left(\lambda_{1},\dots,\lambda_{2^{n}},v_{1},\dots,v_{4^{n}}\right)
\end{equation}
The $r$-point probability measures for $r=1,\dots,2^{n}-1$ are then defined to be the marginals of this probability measure 
\begin{equation}
	\boldsymbol{\di}\hat{\mu}_{n,r}(\lambda_{1},\dots,\lambda_{r})=\int_{\lambda_{r+1},\dots,\lambda_{2^{n}}}\boldsymbol{\di}\hat{\mu}_{n,2^{n}}(\lambda_{1},\dots,\lambda_{2^{n}})
\end{equation}
where the integration over any set of $2^{n}-r$ variables $\lambda_{k}$ gives and equivalent result due to the inbuilt symmetry of the measure $\boldsymbol{\di}\hat{\mu}_{n,2^{n}}$.  The $1$-point probability measure $\di\hat{\mu}_{n,1}$ will be referred to as the spectral probability measure.

The probability of finding an eigenvalue of a random sample $H_{n}$, from the ensemble $\hat{H}_{n}$, in the (real) interval $[a,b]$ (that is, the expected proportion of eigenvalues in $[a,b]$) is therefore given by
\begin{equation}
  \mathbb{P}\left(\lambda\in[a,b]\right)=\int_{a}^{b}\di\hat{\mu}_{n,1}(\lambda)
\end{equation}
Similar, higher dimensional, analogues exist for the other $r$-point probability measures.

\subsection{Spectral density and \texorpdfstring{$r$}{r}-point correlation functions for ensembles}\label{DOSformal}
Formally, the $2^{n}$-point probability measure could be formulated in terms of a probability density function so that
\begin{equation}
	\boldsymbol{\di}\hat{\mu}_{n,2^{n}}(\lambda_{1},\dots,\lambda_{2^{n}})=\hat{\rho}_{n,2^{n}}(\lambda_{1},\dots,\lambda_{2^{n}})\di\lambda_{1}\dots\di\lambda_{2^{n}}
\end{equation}
where $\hat{\rho}_{n,2^{n}}$ is the generalised function given by
\begin{equation}\label{2.2.7}
	\hat{\rho}_{n,2^{n}}(\lambda_{1},\dots,\lambda_{2^{n}})=\left\langle\frac{1}{2^{n}!}\sum_{\tau\in \mathcal{S}_{2^{n}}}\prod_{k=1}^{2^{n}}\delta\left(\lambda_{k}-\lambda^\prime_{\tau(k)}\right)\right\rangle
\end{equation}
where the $\lambda^\prime_{k}$ are the eigenvalues of the matrix $H_{n}$ given in some fixed order.  The sum over $S_{2^{n}}$ is over the permutation group on $2^{n}$ elements, so that $\hat{\rho}_{n,2^{n}}$ is explicitly invariant under permutation of its arguments.  In fact this is not needed due to the presence of the ensemble average, as will be seen next.

The validity of (\ref{2.2.7}) can be seen by writing out the definition of the ensemble average $\langle\cdot\rangle$.  Doing so yields that the right hand side of (\ref{2.2.7}) is equal to
\begin{align}
	\int\frac{1}{2^{n}!}\sum_{\tau\in \mathcal{S}_{2^{n}}}\prod_{k=1}^{2^{n}}\delta\left(\lambda_{k}-\lambda^\prime_{\tau(k)}\right)\boldsymbol{\di}\hat{\mu}_{n,2^{n}+4^{n}}\left(\lambda_{1}^\prime,\dots,\lambda^\prime_{2^{n}},v^\prime_{1},\dots,v^\prime_{4^{n}}\right)
\end{align}
The parameters $v_{1}^\prime,\dots,v_{4^{n}}^\prime$ may be integrated out to leave the $4^{n}$-point probability measure.  As the $4^n$-point probability measure is invariant under permutations of its arguments, the sum over permutations in the previous expression may be dropped to leave
\begin{equation}
	\int\prod_{k=1}^{2^{n}}\delta\left(\lambda_{k}-\lambda^\prime_{k}\right)\boldsymbol{\di}\hat{\mu}_{n,2^{n}}\left(\lambda_{1}^\prime,\dots,\lambda^\prime_{2^{n}}\right)
\end{equation}
On substituting $\hat{\rho}_{n,2^n}\left(\lambda_{1}^\prime,\dots,\lambda^\prime_{2^{n}}\right)\di\lambda^\prime_1\dots\di\lambda^\prime_{2^n}$ for $\boldsymbol{\di}\hat{\mu}_{n,2^{n}}\left(\lambda_{1}^\prime,\dots,\lambda^\prime_{2^{n}}\right)$, the integration may now be formally computed to leave the intended result, $\hat{\rho}_{n,2^n}\left(\lambda_{1},\dots,\lambda_{2^{n}}\right)$.

This then allows the formal $r$-point correlation functions to be defined as
\begin{equation}
	\hat{\rho}_{n,r}(\lambda_{1},\dots,\lambda_{r})=\int\hat{\rho}_{n,2^{n}}(\lambda_{1},\dots,\lambda_{2^{n}})\di\lambda_{r+1}\dots\di\lambda_{2^{n}}
\end{equation}
if $\boldsymbol{\di}\hat{\mu}_{n,r}(\lambda_{1},\dots,\lambda_{r})=\hat{\rho}_{n,r}(\lambda_{1},\dots,\lambda_{r})\di\lambda_{1}\dots\di\lambda_{r}$ for $r=1,\dots,2^{n}-1$.  In particular the $1$-point correlation function or equivalently the spectral density is given by
\begin{equation}
	\hat{\rho}_{n,1}(\lambda)=\left\langle\frac{1}{2^{n}}\sum_{k=1}^{2^{n}}\delta\left(\lambda-\lambda^\prime_{k}\right)\right\rangle
\end{equation}
if $\di\hat{\mu}_{n,1}(\lambda)=\hat{\rho}_{n,1}(\lambda)\di\lambda$.  These expressions can be verified in the same way for the $2^{n}$-point correlation function above.

\subsection{Characteristic function and moments for ensembles}\label{CharacteristicFunction}
The characteristic function of a real random variable $\hat{x}$ is defined as $\psi(t)=\mathbb{E}\left(\e^{\im t\hat{x}}\right)$ \cite[\p342]{Bill}.  Provided that the $m^{th}$ moment of $\hat{x}$ exists then $\psi(t)$ will be $m$ times differentiable and the $m^{th}$ moment will be given by
\begin{equation}
  \im^{-m}\frac{\di^{m}}{\di t^{m}}\psi(t)\bigg|_{t=0}=\mathbb{E}\left(\hat{x}^{m}\right)
\end{equation}
Intuitively the $m^{th}$ moment of the spectral probability measure for the ensemble $\hat{H}_{n}$ should be the ensemble average over the $m^{th}$ moments of the eigenvalues, $\lambda_{1},\dots,\lambda_{2^{n}}$, of the ensemble's members, that is
\begin{equation}
	\left\langle\frac{1}{2^{n}}\sum_{k=1}^{2^{n}}\lambda_{k}^{m}\right\rangle=\left\langle\frac{1}{2^{n}}\Tr\left(H_{n}^{m}\right)\right\rangle
\end{equation}
This will now be seen to be the case:

The characteristic function $\hat{\psi}_{n}(t)$ associated to the spectral probability measure of the ensemble $\hat{H}_{n}$ is by definition
\begin{equation}
  \hat{\psi}_{n}(t)=\int_{-\infty}^\infty\e^{\im t \lambda}\di\hat{\mu}_{n,1}(\lambda)
\end{equation}
Substituting the definition of the probability measure $\di\hat{\mu}_{n,1}(\lambda)$ gives that
\begin{equation}
  \hat\psi_{n}(t)=\int_{-\infty}^\infty\e^{\im t \lambda}\int_{\lambda_{2},\dots,\lambda_{2^{n}},v_{1},\dots,v_{4^{n}}}\boldsymbol{\di}\hat{\mu}_{n,2^{n}+4^{n}}(\lambda,\lambda_{2},\dots,\lambda_{2^{n}},v_{1},\dots,v_{4^{n}})
\end{equation}
Relabelling $\lambda\to\lambda_{1}$ and using the invariance of $\boldsymbol{\di}\hat{\mu}_{n,2^{n}+4^{n}}$ under permutation of the spectral parameters $\lambda_{k}$ then yields that
\begin{equation}
  \hat\psi_{n}(t)=\left\langle\frac{1}{2^{n}}\sum_{k=1}^{2^{n}}\e^{\im t\lambda_{k}}\right\rangle=\left\langle\frac{1}{2^{n}}\Tr\left(\e^{\im tH_{n}}\right)\right\rangle
\end{equation}
The moments are then given by
\begin{equation}
  \im^{-m}\frac{\di^{m}}{\di t^{m}}\left\langle\frac{1}{2^{n}}\Tr\left(\e^{\im tH_{n}}\right)\right\rangle\bigg|_{t=0}
  =\left\langle\frac{1}{2^{n}}\Tr \left(H_{n}^{m}\right)\right\rangle
\end{equation}
as desired.

\subsection{Specific ensemble members}\label{MembersDef}
The $r$-point correlation functions and characteristic function above can be defined for each member $H_{n}$ of the ensemble $\hat{H}_{n}$ individually.  In this case, for the fixed (unordered) eigenvalues $\lambda_{1}^\prime,\dots,\lambda_{2^{n}}^\prime$ of the fixed matrix $H_{n}$, the $r$-point correlation functions are defined to be
\begin{align}
	\rho_{n,2^{n}}(\lambda_{1},\dots,\lambda_{2^{n}})&=\frac{1}{2^{n}!}\sum_{\tau\in \mathcal{S}_{2^{n}}}\prod_{k=1}^{2^{n}}\delta\left(\lambda_{k}-\lambda^\prime_{\tau(k)}\right)\nonumber\\
	\rho_{n,r}(\lambda_{1},\dots,\lambda_{r})&=\int\rho_{n,2^{n}}(\lambda_{1},\dots,\lambda_{2^{n}})\di\lambda_{r+1}\dots\di\lambda_{2^{n}}\nonumber\\
	\rho_{n,1}(\lambda)&=\frac{1}{2^{n}}\sum_{k=1}^{2^{n}}\delta(\lambda-\lambda_{k}^\prime)
\end{align}
(so that $\hat{\rho}_{n,r}=\left\langle\rho_{n,r}\right\rangle$) and the characteristic function associated to $\rho_{n,1}$ reads
\begin{equation}
	\psi_{n}(t)=\frac{1}{2^{n}}\Tr\left(\e^{\im tH_{n}}\right)
\end{equation}
(so that $\hat{\psi}_{n}=\left\langle\psi_{n}\right\rangle$).  In particular the integral of $\rho_{n,1}$ over an interval counts the proportion of eigenvalues of $H_{n}$ falling within that interval,  $\hat{\rho}_{n,1}$ correspondingly gives the ensemble average of this proportion by its definition.

\subsection{L\'{e}vy's continuity theorem}\label{ContinuityTheorem}
L\'{e}vy's continuity theorem will play a crucial role in the forthcoming results of this chapter.  It states that:

\begin{theorem}[L\'{e}vy's continuity theorem {\cite[\p349]{Bill}}]
  For $n\in\mathbb{N}$ let $\di\mu$ and $\di\mu_{n}$ be probability measures on the real line with characteristic functions $\psi(t)$ and $\psi_{n}(t)$ respectively. The sequence of probability measures $\di\mu_{n}$ converges to $\di\mu$ in distribution (weakly) as $n\to\infty$, that is for any $x\in\mathbb{R}$
\begin{equation}
 \lim_{n\to\infty}\int_{-\infty}^{x}\di\mu_{n}=\int_{-\infty}^{x}\di\mu
\end{equation}
if and only if the corresponding sequence $\psi_{n}(t)$ converges point-wise to $\psi(t)$.
\end{theorem}

\subsection{Lyapunov's central limit theorem}\label{Lyapunov}
Lyapunov's central limit theorem will also play a crucial role in the forthcoming results of this chapter.  It states that:

\begin{theorem}[Lyapunov's central limit theorem for triangular arrays {\cite[\p362]{Bill}}]
Let $r:\mathbb{N}\to\mathbb{N}$ be a strictly increasing function.  For each $n\in\mathbb{N}$ let $\hat{x}_{n,1},\dots,\hat{x}_{n,r(n)}$ be independent random variables (not necessarily identically distributed) each with finite mean $\mathbb{E}\left(\hat{x}_{n,j}\right)$ and variance $\mathbb{E}\left(\hat{x}_{n,j}^{2}\right)$ and let
\begin{equation}
	S_{n}^{2}=\sum_{j=1}^{r(n)}\mathbb{E}\left(\hat{x}_{n,j}^{2}\right)
\end{equation}
If the Lyapunov condition
\begin{equation}
	\lim_{n\to\infty}\frac{1}{S_{n}^{2+\delta}}\sum_{j=1}^{r(n)}\mathbb{E}\left(\left|\hat{x}_{n,j}-\mathbb{E}\left(\hat{x}_{n,j}\right)\right|^{2+\delta}\right)=0
\end{equation}
is satisfied for some $\delta>0$ then the distribution of the sum
\begin{equation}
	\frac{1}{S_{n}}\sum_{j=1}^{r(n)}\Big(\hat{x}_{n,j}-\mathbb{E}\left(\hat{x}_{n,j}\right)\Big)
\end{equation}
converges in distribution to a standard normal random variable.
\end{theorem}

\section{Numerics}\label{Numerics}
The methodology and results of numerical simulations of the spectral density for a range of relevant ensembles will now be presented.

\subsection{Ensembles}\label{DOSModels}
The spectral density of the following ensembles of Hamiltonians of $n=2,3,\dots$ qubits will be numerically approximated:  The first is exactly the ensemble defined in the previous section with $9n$ independent identically distributed (i.i.d.) random varables\footnote{Here $\mathcal{N}\left(\mu,s^2\right)$ denotes the normal distribution with mean $\mu$ and variance $s^2$ and $\mathcal{U}\left(x,y\right)$ the uniform distribution supported on the real interval $[x,y]$.  },
\begin{equation}
	\hat{H}_{n}=\sum_{j=1}^{n}\sum_{a,b=1}^{3}\hat{\alpha}_{a,b,j}\sigma_{  j  }^{(a)}\sigma_{  j+1  }^{(b)}\qquad\qquad\hat{\alpha}_{a,b,j}\sim\mathcal{N}\left(0,\frac{1}{9n}\right) \iid
\end{equation}
To see the effect that the distribution of the random variables $\hat{\alpha}_{a,b,j}$ have on the spectral density the related ensemble
\begin{equation}
	\hat{H}^{(uniform)}_{n}=\sum_{j=1}^{n}\sum_{a,b=1}^{3}\hat{\alpha}_{a,b,j}\sigma_{  j  }^{(a)}\sigma_{  j+1  }^{(b)}\qquad\qquad\hat{\alpha}_{a,b,j}\sim\mathcal{U}\left(-\frac{\sqrt{3}}{\sqrt{9n}},\frac{\sqrt{3}}{\sqrt{9n}}\right) \iid
\end{equation}
that is with uniformly distributed random variables, will also be treated.  To investigate the effect local terms proportional to $\sigma_{j}^{(a)}$ have, the ensemble
\begin{equation}
	\hat{H}^{(local)}_{n}=\sum_{j=1}^{n}\sum_{a=1}^{3}\sum_{b=0}^{3}\hat{\alpha}_{a,b,j}\sigma_{  j  }^{(a)}\sigma_{  j+1  }^{(b)}\qquad\qquad\hat{\alpha}_{a,b,j}\sim\mathcal{N}\left(0,\frac{1}{12n}\right) \iid
\end{equation}
will be looked at.  Translationally-invariant Hamiltonians, with respect to translation along the chain, will be of particular interest in Chapter \ref{Eigenstate entanglement}.  To this end the following two ensembles of translationally invariant Hamiltonians will be considered
\begin{alignat}{3}
	\hat{H}^{(inv)}_{n}&=\sum_{j=1}^{n}\sum_{a,b=1}^{3}\hat{\alpha}_{a,b}\sigma_{  j  }^{(a)}\sigma_{  j+1  }^{(b)}\qquad\qquad&&\hat{\alpha}_{a,b}\sim\mathcal{N}\left(0,\frac{1}{9n}\right) \iid\nonumber\\
	\hat{H}^{(inv,local)}_{n}&=\sum_{j=1}^{n}\sum_{a=1}^{3}\sum_{b=0}^{3}\hat{\alpha}_{a,b}\sigma_{  j  }^{(a)}\sigma_{  j+1  }^{(b)}\qquad\qquad&&\hat{\alpha}_{a,b}\sim\mathcal{N}\left(0,\frac{1}{12n}\right) \iid
\end{alignat}
one with and one without local terms.  In Chapter \ref{finite}, Hamiltonians that may be solved via the Jordan-Wigner transform will be explicitly used, therefore the related ensemble
\begin{equation}
	\hat{H}^{(JW)}_{n}=\sum_{j=1}^{n-1}\sum_{a,b=1}^{2}\hat{\alpha}_{a,b,j}\sigma_{  j  }^{(a)}\sigma_{  j+1  }^{(b)}+\sum_{j=1}^{n}\alpha_{3,0,j}\sigma_{j}^{(3)}\qquad\qquad\hat{\alpha}_{a,b,j}\sim\mathcal{N}\left(0,\frac{1}{5n-4}\right) \iid
\end{equation}
will also be numerically investigated.  Finally the ensemble of Heisenberg $XYZ$ type Hamiltonians
\begin{equation}
	\hat{H}^{(Heis)}_{n}=\sum_{j=1}^{n}\sum_{a=1}^{3}\hat{\alpha}_{a,a,j}\sigma_{  j  }^{(a)}\sigma_{  j+1  }^{(a)}\qquad\qquad\hat{\alpha}_{a,a,j}\sim\mathcal{N}\left(0,\frac{1}{3n}\right) \iid
\end{equation}
will provide some particularly exotic spectral densities which converge very slowly.

\subsection{Numerical methodology}
The numerical simulation of the ensembles above was carried out on a computer with a quad-core Intel i$5$ processor running at $2.9$GHz with $6$MB L$3$ cache and $8$GB of RAM.  The code was written in C++ and used the GNU Scientific Library (GSL) version $1.15$, see \cite{GSL} for full documentation.

\subsubsection{Matrix generation}
First, many matrices were sampled from each ensemble.  The GSL contains the following functions to do this
\begin{code}
gsl\_{c}omplex\_{m}atrix$^{*}$ \textbf{gsl\_{m}atrix\_{c}omplex\_{a}lloc}(int \textbf{r}, int \textbf{c})\\
void \textbf{gsl\_{m}atrix\_{c}omplex\_{f}ree}(gsl\_{c}omplex\_{m}atrix$^{*}$ \textbf{m})
\end{code}
The first allocates memory for a $r\times c$ complex matrix and returns a pointer to it, the second frees memory allocated by the first.  The value of the $r^{th}$ row and $c^{th}$ column of a matrix at $m$ are set to the complex number $z$ by the GSL function
\begin{code}
void \textbf{gsl\_{m}atrix\_{c}omplex\_{s}et}(gsl\_{c}omplex\_{m}atrix$^{*}$ \textbf{m}, int \textbf{r}, int \textbf{c},
\begin{flushright}gsl\_{c}omplex \textbf{z})\end{flushright}
\end{code}

\subsubsection{Random number generation}
The entries for each matrix were then filled, in the correct locations, with suitable pseudo-random numbers. Using the GSL, a default random number generator $r$ is created by
\begin{code}
gsl\_{r}ng$^{*}$ \textbf{r} = \textbf{gsl\_{r}ng\_{a}lloc}(gsl\_{r}ng\_{d}efault)
\end{code}
and seeded, so that the pseudo-random numbers generated will be unpredictable, with the processor time via
\begin{code}
\textbf{gsl\_{r}ng\_{s}et}(\textbf{r}, (unsigned long int)\textbf{time}(NULL))
\end{code}

This then allowed the normal distribution with zero mean and variance $\sigma^{2}$ to be sampled with the GSL function
\begin{code}
double \textbf{gsl\_{r}an\_{g}aussian\_{z}iggurat}(const gsl\_{r}ng \textbf{r}, double $\boldsymbol{\sigma}$)
\end{code}
which implements the alternative Marsaglia-Tsang ziggurat method, the fastest algorithm available in the GSL.  The uniform distribution on $[a,b]$ is sampled with the GSL function
\begin{code}
double \textbf{gsl\_{r}an\_{f}lat}(const gsl\_{r}ng \textbf{r}, double \textbf{a}, double \textbf{b})
\end{code}

\subsubsection{Diagonalisation}
The resulting $N\times N$ Hermitian matrices (with $N=2^{n}$) were then diagonalised.  The GSL function for this needs a workspace in memory, this is allocated and deallocated by the following GSL functions
\begin{code}
gsl\_{e}igen\_{h}erm\_{w}orkspace$^{*}$ \textbf{gsl\_{e}igen\_{h}erm\_{a}lloc}(int \textbf{N})\\
void \textbf{gsl\_{e}igen\_{h}ermv\_{f}ree}(gsl\_{e}igen\_{h}erm\_{w}orkspace$^{*}$ \textbf{w})
\end{code}
The $N\times N$ matrix at $m$ could then be diagonalised with the GSL function
\begin{code}
int \textbf{gsl\_{e}igen\_{h}erm}(gsl\_{c}omplex\_{m}atrix$^{*}$ \textbf{m}, gsl\_{v}ector$^{*}$ \textbf{eval}, \begin{flushright}gsl\_{m}atrix\_{c}omplex$^{*}$ \textbf{evec}, gsl\_{e}igen\_{h}erm\_{w}orkspace$^{*}$ \textbf{w})\end{flushright}
\end{code}
where the resulting eigenvalues are stored in the vector at $eval$ and the resulting eigenvectors are stored as the columns of the matrix at $evec$.  Diagonalisation was the most computationally time consuming step in the numerical simulation of the ensembles.  Table $\ref{times}$ gives an indication of the approximate time it took to diagonalise a $2^{n}\times 2^{n}$ matrix with the GSL function above.  Note that four matrices may be diagonalised in parallel, utilising all four cores of the processor available.

\begin{table}[position specifier]
  \centering
  \begin{tabular}{c|c|c}
    \toprule
    Number of qubits, $n$ & Dimension of $H$ & Approximate time to diagonalise (s) \\
    \midrule
    $8$ & $2^{8}=256$ & $1\times10^{0}\pm10^{0}$\\
    $9$ & $2^{9}=512$ & $6\times10^{0}\pm10^{0}$\\
    $10$ & $2^{10}=1024$ & $8\times10^{1}\pm10^{1}$\\
    $11$ & $2^{11}=2048$ & $8\times10^{2}\pm10^{2}$\\
    $12$ & $2^{12}=4096$ & $8\times10^{3}\pm10^{3}$\\
    $13$ & $2^{13}=8192$ & $8\times10^{4}\pm10^{4}$\\
    \bottomrule
  \end{tabular}
  \caption[The approximate time taken to diagonalise matrices with the GSL]{The approximate time it took to diagonalise a $2^{n}\times 2^{n}$ Hermitian matrix with the GSL function {\small \texttt{\textbf{gsl\_{e}igen\_{h}erm}()}}.}
  \label{times}
\end{table}

\subsection{Spectral density}
The approximate spectral density was calculated by splitting the interval $[-3,3]$ into $l$ sub-intervals of equal length $[x_{j},x_{j+1})$ where
\begin{equation}
  -3=x_{0}<x_{1}<\dots<x_{l}=3
\end{equation}
For all the $2^{n}\times 2^{n}$ matrices generated for a particular ensemble, the number of eigenvalues in each interval from all the samples was calculated and a single normalised histogram produced.  The GSL provides the following functions to allocate and deallocate memory to a histogram with $l$ bins, set the histogram's range to $[a,b]$, to count a data point $x$ into the correct bin and to scale a histogram by a factor $s$;
\begin{code}
gsl\_{h}istogram$^{*}$ \textbf{gsl\_{h}istogram\_{a}lloc}(int \textbf{l})\\
void \textbf{gsl\_{h}istogram\_{f}ree}(gsl\_{h}istogram$^{*}$ \textbf{h})\\
int \textbf{gsl\_{h}istogram\_{s}et\_{r}anges\_{u}niform}(gsl\_{h}istogram$^{*}$ \textbf{h}, double \textbf{a}, double \textbf{b})\\
int \textbf{gsl\_{h}istogram\_{i}ncrement}(gsl\_{h}istogram$^{*}$ \textbf{h}, double \textbf{x})\\
int \textbf{gsl\_{h}istogram\_{s}cale}(gsl\_{h}istogram$^{*}$ \textbf{h}, double \textbf{s})
\end{code}
respectively.

All the ensembles considered in Section \ref{DOSModels} are symmetric under sign conjugation, for example $\hat{H}_{n}$ is invariant under $\hat{H}_{n}\to-\hat{H}_{n}$.  Therefore the spectral density of each ensemble must be an even function, for example $\hat{\rho}_{n,1}(\lambda)=\hat{\rho}_{n,1}(-\lambda)$ in the case of $\hat{H}_{n}$.  This symmetry was used to further smooth statistical fluctuations in the spectral histograms of the ensembles considered, by averaging each histogram's bin values symmetrically around zero.

\subsection{Results}
The normalised spectral histogram over $s=2^{6}$ matrices sampled from the ensemble $\hat{H}_{13}$ is shown in Figure \ref{model1}.   This histogram approximates the spectral density of the ensemble $\hat{H}_{13}$ by construction.  A slight discrepancy is seen from the density of a standard normal random variable.  The corresponding figures for the normalised spectral histograms over $s=2^{19-n}$ matrices sampled from the ensembles $\hat{H}_{n}$ for $n=2,\dots,13$ are given in Appendix \ref{specHist}.  Fluctuations from the standard normal distribution are seen to shrink as $n$ increases. 

The oscillations around the standard Gaussian density seen in Appendix \ref{specHist} have conceivably $2^{n}-1$ peaks and troughs.  It is possible that these could correspond to a degree $2^{n}$ polynomial, similarly to the case for the oscillations seen in the GUE's spectral density for finite matrix size to which the method of orthogonal polynomials is applicable, see Section \ref{GUE}.

\begin{figure}
  \centering
  \input{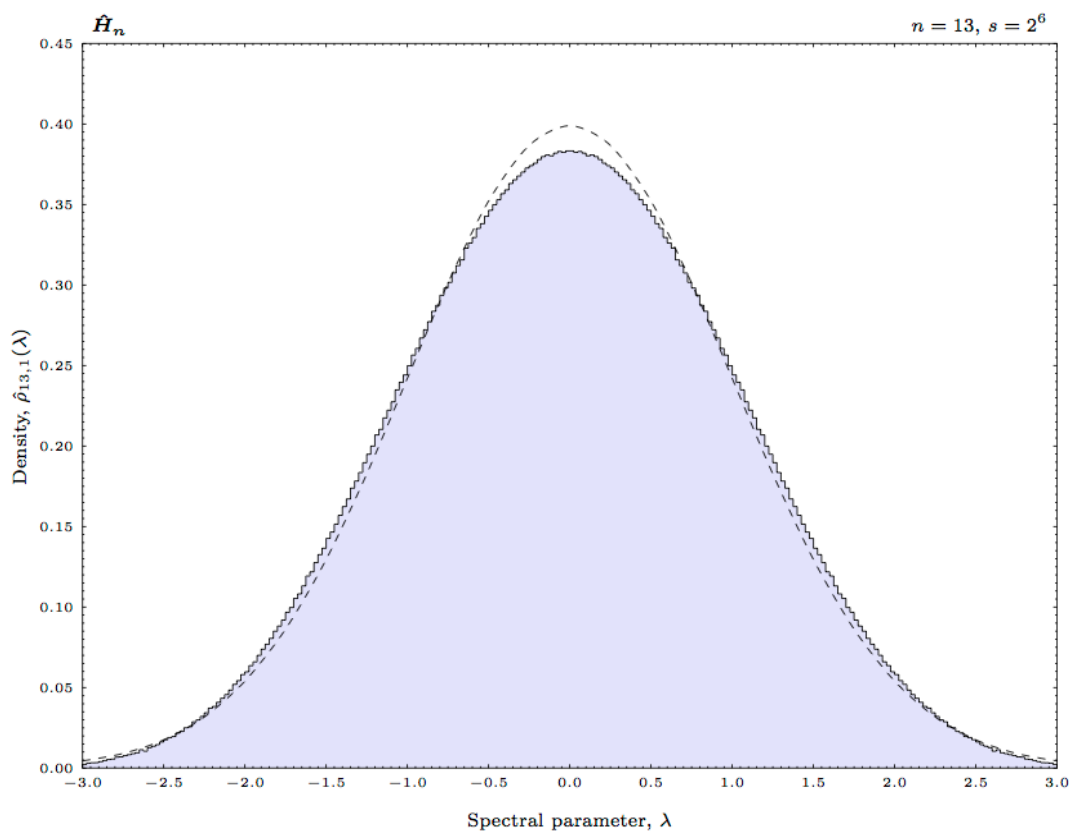}
  \caption[Spectral histogram for $\hat{H}_{13}$]{The normalised spectral histogram for $\hat{H}_{13}$.  The spectra used were numerically obtained from each of $s=2^{6}$ random samples of $\hat{H}_{13}$.  The dashed line gives the standard Gaussian density.}
  \label{model1}
\end{figure}

The corresponding figures for all the other ensembles defined in Section \ref{DOSModels} are also included in Appendix \ref{specHist}.  For $\hat{H}_{n}^{(uniform)}$ the convergence to the standard normal distribution, as $n$ increases, is directly comparable to that for $\hat{H}_{n}$, although the fluctuations decrease more slowly.  This  indicates some universally in the ensemble's probability measure, that is a possible independence of the limiting distribution from the ensemble's distribution.

The remaining ensembles all show a convergence to the standard Gaussian density for their spectral densities.  In particular this is true for the ensembles $\hat{H}_{n}^{(inv)}$, $\hat{H}_{n}^{(inv,local)}$ and $\hat{H}_{n}^{(Heis)}$ which each contain a fixed number, independent of $n$, of random variables.  The ensemble $\hat{H}_{n}^{(Heis)}$ has the fewest (only $3$) independent random variables.  The conjectured convergence of its spectral density to the standard Gaussian is notably slower than for the other ensembles, perhaps as a result.

\section{Limiting moments of the ensemble spectral density}\label{Moment method}
This section is devoted to determining the limiting moments of the spectral density of the sequence of ensembles
\begin{equation}\label{nnEnsemble}
  \hat{H}_{n}=\sum_{j=1}^{n}\sum_{a,b=1}^{3}\hat{\alpha}_{a,b,j}\sigma_{  j  }^{(a)}\sigma_{  j+1  }^{(b)}\qquad\qquad \hat{\alpha}_{a,b,j}\sim\mathcal{N}\left(0,\frac{1}{9n}\right) \iid
  \end{equation}
for $n=2,3,\dots$.  That is, as defined in Section \ref{PointCorrFunc}, the limiting values of the quantity
\begin{equation}
  \left\langle\frac{1}{2^{n}}\Tr\left(H_{n}^{m}\right)\right\rangle=\left\langle\frac{1}{2^{n}}\sum_{k=1}^{2^{n}}\lambda_{k}^{m}\right\rangle
\end{equation}
as $n\to\infty$ for $m\in\mathbb{N}$, where $\lambda_{k}$ are the eigenvalues of the member $H_{n}$ of the ensemble $\hat{H}_{n}$.

The following theorem addresses this:

\subsection{Limiting moments of the spectral density for the ensembles \texorpdfstring{$\hat{H}_{n}$}{Hn}}
\begin{theorem}[Limiting moments of the spectral density for the ensembles $\hat{H}_{n}$]
  For the ensembles $\hat{H}_{n}$ the $m^{th}$ moments of their spectral probability measures tend to the $m^{th}$ moment of a standard normal distribution, in the large chain limit.  That is, as $n\to\infty$
  \begin{equation}
    \left\langle\frac{1}{2^{n}}\Tr\left( H_{n}^{m}\right)\right\rangle\rightarrow\frac{1}{\sqrt{2\pi}}\int_{-\infty}^{\infty}x^{m}\e^{-\frac{x^{2}}{2}}\di x
  \end{equation}
  for all fixed $m\in\mathbb{N}$.
\end{theorem}

\begin{proof}
The proof will be split into four parts.  First, using the definition of $H_{n}$ the powers $H_{n}^{m}$ will be expanded.  Second, the coefficients in this expansion will be compared to that of the multinomial coefficients for commuting variables.  Then, the terms in the expansion of $H_{n}^{m}$ will be grouped by certain properties such that the expansion can be seen to be approximated by a multinomial expansion.  Finally the multinomial coefficients will be shown to give the correct limiting moments:

\subsubsection{Relabelling}
To simplify the notation let the parameters $\alpha_{a,b,j}$ be relabelled as $\alpha_{k}$ and likewise the matrices $\sigma_{  j  }^{(a)}\sigma_{  j+1  }^{(b)}$ as $P_{k}$.  A suitable relabelling would map the indices $(a,b,j)$ to $k$ as $k=9(j-1)+3(b-1)+a$, that is identifying
\begin{equation}
  \alpha_{a,b,j}\equiv\alpha_{9(j-1)+3(b-1)+a}\qquad\text{and}\qquad\sigma_{  j  }^{(a)}\sigma_{  j+1  }^{(b)}\equiv P_{9(j-1)+3(b-1)+a}
\end{equation}
for $a,b=1,2,3$ and $j=1,\dots,n$.

\subsubsection{Expansion}
Under this relabelling members $H_{n}$ of the ensemble $\hat{H}_{n}$ are now expressed as
\begin{equation}
  H_{n}=\sum_{k=1}^{9n}\alpha_{k}P_{k}
\end{equation}
Substituting this definition of $H_{n}$ into $H_{n}^{m}$ and expanding gives that
\begin{equation}
  H_{n}^{m}=\left(\sum_{k=1}^{9n}\alpha_{k}P_{k}\right)^{m}=\sum_{k_{1},\dots,k_{m}=1}^{9n}\alpha_{k_{1}}\dots\alpha_{k_{m}}P_{k_{1}}\dots P_{k_{m}}
\end{equation}
The ensemble average may now be performed term-wise, that is the quantities $\left\langle\alpha_{k_{1}}\dots\alpha_{k_{m}}\right\rangle$ can be computed for each term.  To this end, the products $\alpha_{k_{1}}\dots\alpha_{k_{m}}$ may be rewritten by collecting all repeated variables, that is for each of the $(9n)^{m}$ instances of the vector $(k_{1},\dots, k_{m})\in\{1,\dots,9n\}^{m}$ there exists a vector $(m_{1}, \dots, m_{9n})\in\{0,\dots,m\}^{9n}$ such that $m_{1}+ \dots +m_{9n}=m$ and
\begin{equation}
  \alpha_{k_{1}}\dots\alpha_{k_{m}}=\alpha_{1}^{m_{1}}\dots\alpha_{9n}^{m_{9n}}
\end{equation}
For example, if $m=4$
\begin{equation}
  \alpha_{1}\alpha_{2}\alpha_{1}\alpha_{1}=\alpha_{1}^{3}\alpha_{2}^{1}\alpha_{3}^{0}\dots\alpha_{9n}^{0}
\end{equation}

As the $\alpha_{k}$ in $H_{n}$ correspond to independent random variables in $\hat{H}_{n}$, it follows directly that
\begin{equation}
  \left\langle\alpha_{k_{1}}\dots\alpha_{k_{m}}\right\rangle=\left\langle\alpha_{1}^{m_{1}}\right\rangle\dots\left\langle\alpha_{9n}^{m_{9n}}\right\rangle
\end{equation}
For any $k$ and odd value of $m_{k}$, 
\begin{equation}
  \left\langle\alpha_{k}^{m_{k}}\right\rangle=\sqrt{\frac{9n}{2\pi}}\int_{-\infty}^{\infty}\alpha_{k}^{m_{k}}\e^{-\frac{9n\alpha_{k}^{2}}{2}}\di \alpha_{k} = 0
\end{equation}
due to the parity of the integrand.  Therefore only terms for which each value of $m_{k}$ is even, hence only even values of $m$, will now be considered.  The terms for which each $m_{k}$ is even must also contain an even number of each Pauli matrix $P_{k}$, as each factor $\alpha_{k}$ comes with exactly one $P_{k}$.

Given any two matrices $P_{k}$ and $P_{k^\prime}$, that is $\sigma_{  j  }^{(a)}\sigma_{  j+1  }^{(b)}$ and $\sigma_{  j^\prime  }^{( a^\prime )}\sigma_{  j^\prime+1  }^{( b^\prime )}$, they either commute or anti-commute as the matrices $\sigma^{(1)}, \sigma^{(2)}$ and $\sigma^{(3)}$ are pairwise anti-commuting.  For example

\begin{center}
\begin{tabular}{ l l l l }
  $\sigma_{  1  }^{(3)}\sigma_{  2  }^{(3)}$ & and & $\sigma_{  5  }^{(2)}\sigma_{  6  }^{(2)}$ & commute since all the $\sigma_{  j  }^{(a)}$ present commute\\
  $\sigma_{  1  }^{(3)}\sigma_{  2  }^{(3)}$ & and & $\sigma_{  2  }^{(3)}\sigma_{  3  }^{(3)}$ & commute since all the $\sigma_{  j  }^{(a)}$ present commute \\
  $\sigma_{  1  }^{(3)}\sigma_{  2  }^{(3)}$ & and & $\sigma_{  2  }^{(2)}\sigma_{  3  }^{(2)}$ & anti-commute since $\sigma_{  2  }^{(3)}$ and $\sigma_{  2  }^{(2)}$ anti-commute \\
\end{tabular}
\end{center}
Then in any product $P_{k_{1}}\dots P_{k_{m}}$ the matrices may be reordered (acquiring a possible negative sign if some of the matrices $P_{k}$ anti-commute with each other) to a product $\pm P_{1}^{m_{1}}\dots P_{9n}^{m_{9n}}$.  For any $P_{k}$, $P_{k}^{2}=I_{2^{n}}$ by the definition of the Pauli matrices, so that the terms for which each $m_{k}$ is even must be equal to $\pm I_{2^{n}}$.  For example in the case of $m=4$, 
\begin{align}
  P_{1}P_{2}P_{1}P_{2} &= \left(\sigma_{  1  }^{(1)}\sigma_{  2  }^{(1)}\right)
		 \left(\sigma_{  1  }^{(2)}\sigma_{  2  }^{(1)}\right)
		 \left(\sigma_{  1  }^{(1)}\sigma_{  2  }^{(1)}\right)
		 \left(\sigma_{  1  }^{(2)}\sigma_{  2  }^{(1)}\right)\nonumber\\
		&=-\left(\sigma_{  1  }^{(1)}\sigma_{  2  }^{(1)}\right)
		\left(\sigma_{  1  }^{(1)}\sigma_{  2  }^{(1)}\right)
		 \left(\sigma_{  1  }^{(2)}\sigma_{  2  }^{(1)}\right)
		 \left(\sigma_{  1  }^{(2)}\sigma_{  2  }^{(1)}\right)\nonumber\\
		&=-P_{1}^{2}P_{2}^{2}\nonumber\\
		&=-I_{2^{n}}
\end{align}
Therefore
\begin{align}\label{signs}
  \left\langle\frac{1}{2^{n}}\Tr\left(H_{n}^{m}\right)\right\rangle
  &=\frac{1}{2^{n}}\Tr\left(I_{2^{n}}\right) \sum_{k_{1},\dots,k_{m}=1}^{9n}\pm\left\langle\alpha_{k_{1}}\dots\alpha_{k_{m}}\right\rangle\nonumber\\
  &=\sum_{k_{1},\dots,k_{m}=1}^{9n}\pm\left\langle\alpha_{k_{1}}\dots\alpha_{k_{m}}\right\rangle
\end{align}
where each term has a, as yet undetermined, sign (here the terms without an even number of each factor $\alpha_{k}$ have also been included but their contribution to the average remains zero).

\subsubsection{Multinomial coefficients}
If each of the signs in this expression were positive, progress could be made with the multinomial theorem which states that
\begin{align}\label{multi}
  \left(\sum_{k=1}^{9n}\alpha_{k}\right)^{m}
  &=\sum_{k_{1},\dots,k_{m}=1}^{9n}\alpha_{k_{1}}\dots\alpha_{k_{m}}\nonumber\\
  &=\sum_{\genfrac{}{}{0pt}{}{m_{1},\dots,m_{9n}\in\mathbb{N}_{0}}{m_{1}+\dots+m_{9n}=m}}\frac{m!}{m_{1}!\dots m_{9n}!}\alpha_{1}^{m_{1}}\dots\alpha_{9n}^{m_{9n}}
\end{align}
Including the positive and negative signs in (\ref{signs}) however produce some other coefficients, $\cof_{m_{1},\dots,m_{9n}}$, so that collecting terms of the form $\alpha_{1}^{m_{1}}\dots\alpha_{9n}^{m_{9n}}$ gives
\begin{equation}
  \left\langle\frac{1}{2^{n}}\Tr\left(H_{n}^{m}\right)\right\rangle=\sum_{\genfrac{}{}{0pt}{}{m_{1},\dots,m_{9n}\in2\mathbb{N}_{0}}{m_{1}+\dots+m_{9n}=m}}\cof_{m_{1},\dots,m_{9n}}\left\langle\alpha_{1}^{m_{1}}\dots\alpha_{9n}^{m_{9n}}\right\rangle
\end{equation}
(where only the terms for which all the $m_{k}$ are even have been included here, as all the others do not contribute to the average, as seen before).  It will now be shown that replacing the coefficients $\cof_{m_{1},\dots,m_{9n}}$ with the multinomial coefficients
\begin{equation}
  \frac{m!}{m_{1}!\dots m_{9n}!}
\end{equation}
leads to a good approximation of $\left\langle\frac{1}{2^{n}}\Tr \left(H_{n}^{m}\right)\right\rangle$ in the large $n$ limit.

\subsubsection{Coefficients relating to commuting matrices}
The coefficients $\cof_{m_{1},\dots,m_{9n}}$, such that each pair of matrices in the set $\{P_{k}|m_{k}\neq0\}$ commute and each $m_{k}$ is even, will now be determined. The terms containing the product $\alpha_{1}^{m_{1}}\dots\alpha_{9n}^{m_{9n}}$ in
\begin{equation}
  \sum_{k_{1},\dots,k_{m}=1}^{9n}\pm\left\langle\alpha_{k_{1}}\dots\alpha_{k_{m}}\right\rangle
\end{equation}
must necessarily come with a positive sign as the corresponding matrices $P_{k}$ all commute pairwise.  Let these terms be denoted by $\mathcal{C}$.  By comparing the $\alpha_{1}^{m_{1}}\dots\alpha_{9n}^{m_{9n}}$ coefficients in the multinomial theorem,  there are seen to be precisely $\frac{m!}{m_{1}!\dots m_{9n}!}$ such terms.  As each term comes with a coefficient of $+1$
\begin{equation}
  \cof_{m_{1},\dots,m_{9n}}=\sum_{\mathcal{C}}1=\frac{m!}{m_{1}!\dots m_{9n}!}
\end{equation}

\subsubsection{Coefficients relating to non-commuting matrices}
The coefficients $\cof_{m_{1},\dots,m_{9n}}$, such that least one pair of matrices in the set $\{P_{k}|m_{k}\neq0\}$ anti-commute and each $m_{k}$ is even, can be likewise bounded.  There are again precisely $\frac{m!}{m_{1}!\dots m_{9n}!}$ terms containing the product $\alpha_{1}^{m_{1}}\dots\alpha_{9n}^{m_{9n}}$ in the sum
\begin{equation}
  \sum_{k_{1},\dots,k_{m}=1}^{9n}\pm\left\langle\alpha_{k_{1}}\dots\alpha_{k_{m}}\right\rangle
\end{equation}
Let these terms be denoted by $\mathcal{A}$.  Each ones comes with a coefficient of $+1$ or $-1$.    The coefficient $\cof_{m_{1},\dots,m_{9n}}$ is the sum over all these $\frac{m!}{m_{1}!\dots m_{9n}!}$ coefficients of $\pm1$.  The triangle inequality then gives a bound for the modulus of $\cof_{m_{1},\dots,m_{9n}}$,
\begin{equation}
  \left|\cof_{m_{1},\dots,m_{9n}}\right|\leq\sum_{\mathcal{A}}\left|\pm1\right|=\frac{m!}{m_{1}!\dots m_{9n}!}
\end{equation}

\subsubsection{Number of commuting matrices}
For $n\geq 3$, given an arbitrary matrix $P_{k} = \sigma_{  j  }^{(a)}\sigma_{  j+1  }^{(b)}$, the sixteen matrices that anti-commute with it are exactly $\sigma_{  j-1  }^{(c)}\sigma_{  j  }^{(d)}$, $\sigma_{  j  }^{(d)}\sigma_{  j+1  }^{(b)}$,$\sigma_{  j  }^{(a)}\sigma_{  j+1  }^{(e)}$ and $\sigma_{  j+1  }^{(e)}\sigma_{  j+2  }^{(f)}$ where $c,d,e,f\in\{1,2,3\}$ such that $a\neq d$ and $b\neq e$.

For example with $n\geq13$, the matrices that anti-commute with $\sigma_{  11  }^{(1)}\sigma_{  12  }^{(2)}$ are precisely

\begin{center}
\begin{tabular}{ l l l l }
  $\sigma_{  10  }^{(1)}\sigma_{  11  }^{(2)}$ & $\sigma_{  11  }^{(1)}\sigma_{  12  }^{(1)}$ & $\sigma_{  11  }^{(2)}\sigma_{  12  }^{(2)}$ & $\sigma_{  12  }^{(1)}\sigma_{  13  }^{(1)}$ \\
  $\sigma_{  10  }^{(1)}\sigma_{  11  }^{(3)}$ & $\sigma_{  11  }^{(1)}\sigma_{  12  }^{(3)}$ & $\sigma_{  11  }^{(3)}\sigma_{  12  }^{(2)}$ & $\sigma_{  12  }^{(1)}\sigma_{  13  }^{(2)}$\\
  $\sigma_{  10  }^{(2)}\sigma_{  11  }^{(2)}$ &  &  & $\sigma_{  12  }^{(1)}\sigma_{  13  }^{(3)}$\\
  $\sigma_{  10  }^{(2)}\sigma_{  11  }^{(3)}$ &  &  & $\sigma_{  12  }^{(3)}\sigma_{  13  }^{(1)}$ \\
  $\sigma_{  10  }^{(3)}\sigma_{  11  }^{(2)}$ &  &  & $\sigma_{  12  }^{(3)}\sigma_{  13  }^{(2)}$ \\
  $\sigma_{  10  }^{(3)}\sigma_{  11  }^{(3)}$ &  &  & $\sigma_{  12  }^{(3)}\sigma_{  13  }^{(3)}$ \\
\end{tabular}
\end{center}

\subsubsection{Value of the average in each term}
Since the $\alpha_{k}$ correspond to the identically and independently distributed random variables $\hat{\alpha}_{k}$ in $\hat{H}_{n}$,
\begin{align}
  \left\langle\alpha_{1}^{m_{1}}\dots\alpha_{9n}^{m_{9n}}\right\rangle
  &=\left\langle\alpha_{1}^{m_{1}}\right\rangle\dots\left\langle\alpha_{9n}^{m_{9n}}\right\rangle
\end{align}
The identity (see appendix \ref{Int}), for even $m_{k}$,
\begin{equation}
  \left\langle\alpha_{k}^{m_{k}}\right\rangle=\sqrt{\frac{9n}{2\pi}}\int_{-\infty}^{\infty}\alpha_{k}^{m_{k}}\e^{-\frac{9n\alpha_{k}^{2}}{2}}\di \alpha_{k}
  =\frac{m_{k}!}{\left(\frac{m_{k}}{2}\right)!\left(18n\right)^{\frac{m_{k}}{2}}}
\end{equation}
allows the average to be computed, when all the $m_{k}$ are even, exactly as
\begin{equation}\label{termavg}
  \frac{m_{1}!}{\left(\frac{m_{1}}{2}\right)!\left(18n\right)^{\frac{m_{1}}{2}}}\dots\frac{m_{9n}!}{\left(\frac{m_{9n}}{2}\right)!\left(18n\right)^{\frac{m_{9n}}{2}}}
\end{equation}

\subsubsection{Grouping terms}
Each term in
\begin{equation}\label{2.4.24}
  \sum_{\genfrac{}{}{0pt}{}{m_{1},\dots,m_{9n}\in2\mathbb{N}_{0}}{m_{1}+\dots+m_{9n}=m}}\cof_{m_{1},\dots,m_{9n}}\left\langle\alpha_{1}^{m_{1}}\dots\alpha_{9n}^{m_{9n}}\right\rangle
\end{equation}
can be identified uniquely by the vector $\boldsymbol{m}=(m_{1},\dots,m_{9n})\in2\mathbb{N}_{0}^{9n}$.  Let $\boldsymbol{x},\boldsymbol{y}\in\mathbb{N}_{0}^{9n}$ be equivalent (denoted $x\sim y$) if and only if there exists a permutation matrix, $\Pi$, such that $x=y\Pi$.  This imposes a partition or grouping of all the terms in the sum above, terms belonging to the same group if and only if their identifying vectors $(m_{1},\dots,m_{9n})$ are equivalent.

For example, if $m=8$, $\boldsymbol{m}=(4,4,0,\dots,0)$ is equivalent to $\boldsymbol{m}^\prime=(4,0,4,0,\dots,0)$ by the permutation matrix that swaps the second and third elements.  This implies that terms containing $\left\langle\alpha_{1}^{4}\alpha_{2}^{4}\right\rangle$ and $\left\langle\alpha_{1}^{4}\alpha_{3}^{4}\right\rangle$ would be grouped together for example.

For all the terms in any such group, the value of
\begin{equation}
  \left\langle\alpha_{1}^{m_{1}}\dots\alpha_{9n}^{m_{9n}}\right\rangle
\end{equation}
is the same.  This is seen by again recalling that the $\alpha_{k}$ correspond to identically and independently distributed random variables so that
\begin{equation}
  \left\langle\alpha_{1}^{m_{1}}\dots\alpha_{9n}^{m_{9n}}\right\rangle
  =\left\langle\alpha_{1}^{m_{1}}\right\rangle\dots\left\langle\alpha_{9n}^{m_{9n}}\right\rangle
\end{equation}
Different terms in the same group differ only by a permutation of the coefficients $m_{k}$, which results in no change from the value above.

For example $\left\langle\alpha_{1}^{4}\alpha_{2}^{4}\right\rangle$ and $\left\langle\alpha_{1}^{4}\alpha_{3}^{4}\right\rangle$ have the same value of $\left\langle \alpha_{1}^{4}\right\rangle^{2}$.

\subsubsection{Number of terms in each group}
The number of terms in each group may now be counted.  Each group can be identified by a unique vector $\boldsymbol{q}=(q_{1},\dots,q_{9n})\in2\mathbb{N}_{0}^{9n}$ where $q_{1}\leq q_{2}\leq\dots\leq q_{9n}$ and $\|\boldsymbol{q}\|_{1}=m$.  Let $s$ denote the number of distinct elements in this vector and define repetition numbers $r_{1},\dots,r_{s}$ which count the number of times each one of these distinct values is repeated.

Combinatorial there are
\begin{equation}
  \frac{(9n)!}{r_{1}!\dots r_{s}!}
\end{equation}
distinct permutations of $\boldsymbol{q}$, which implies that the group contains this number of terms.

For example, if $m=8$ and $\boldsymbol{q}=(0,\dots,0,4,4)$, there are $s=2$ distinct values of the elements and therefore two repetition numbers, $r_{1}=9n-2$ and $r_{2}=2$. There are exactly $\genfrac{(}{)}{0pt}{}{9n}{2}$ equivalent vectors, that is the number of ways of picking two object from $9n$ objects.  This is consistent with the expression above as
\begin{equation}
  \frac{(9n)!}{r_{1}!\dots r_{s}!}=\frac{(9n)!}{(9n-2)!2!}=\genfrac{(}{)}{0pt}{}{9n}{2}
\end{equation}

\subsubsection{Proportion of each group that contributes a multinomial coefficient}
It has been seen that the coefficients $\cof_{m_{1},\dots,m_{9n}}$ in (\ref{2.4.24}), such that each pair of matrices in the set $\{P_{k}|m_{k}\neq0\}$ commute and each $m_{k}$ is even, are equal to
\begin{equation}
  \frac{m!}{m_{1}!\dots m_{9n}!}
\end{equation}
The number of terms in each group for which this holds will now be bounded.

Consider a group of terms characterised by the vector $\boldsymbol{q}=(q_{1},\dots,q_{9n})\in2\mathbb{N}_{0}^{9n}$ where $q_{1}\leq q_{2}\leq\dots\leq q_{9n}$ and $\|\boldsymbol{q}\|_{1}=m$. For $n>>m$ there must exist components of $\boldsymbol{q}$ which are zero as $\|\boldsymbol{q}\|_{1}=m$.  Let $r_{1}$ denote the number of components which are zero, so that the number of non-zero components, $\gamma$, is $9n-r_{1}$.

The counting of all equivalent (by permutation) vectors $\boldsymbol{m}=(m_{1}, \dots,m_{9n})$ to $\boldsymbol{q}$ such that each pair of matrices in the set $\{P_{k}|m_{k}\neq0\}$ commute proceeds as follows:  The non-zero value $q_{9n}$ in $\boldsymbol{q}$ can be permuted to any of the $9n$ coordinates.  It is recalled from above that given a matrix $P_{k}$ there are a fixed number, $\beta=16$, of other matrices $P_{k^\prime}$ which do not commute with it. Therefore the next non-zero value (if there is one) $q_{9n-1}$ in $\boldsymbol{q}$ can only be permuted to $9n-1-\beta$ of the remaining $9n-1$ coordinates to ensure the commutation of the corresponding matrices $P_{k}$.  Similarly the next non-zero value (if there is one) $q_{9n-2}$ in $\boldsymbol{q}$ can only be permuted to at least (giving rise to a possible under counting) $9n-2-2\beta$ of the $9n-2$ remaining coordinates to ensure the commutation of the corresponding matrices $P_{k}$, see Figure \ref{CombinatoricMoments}.

This process stops after all the non-zero entries of $\boldsymbol{q}$ have been placed.  Proceeding in this fashion, the minimum number of vectors $\boldsymbol{m}$ equivalent to $\boldsymbol{q}$ such that each pair of matrices in the set $\{P_{k}|m_{k}\neq0\}$ commute is
\begin{equation}
  \frac{(9n)(9n-1-\beta)\dots(9n-(\gamma-1)-(\gamma-1)\beta)}{r_{2}!\dots r_{s}!}
\end{equation}
where the factor $r_{2}!\dots r_{s}!$ in the denominator accounts for repeated values in the $m_{k}$'s.

For example, if $m=6$, $\boldsymbol{q}=(0,\dots,0,2,4)$, there are $9n$ possible locations to permute the value $4$ to. Given this choice there are only another $9n-1-16$ choices for where to place the value $2$ so that the two locations chosen correspond to Pauli basis matrices that commute.  The choice $\boldsymbol{m}=(4,0,2,0,\dots,0)$ is no good since $P_{1}=\sigma_{  1  }^{(1)}\sigma_{  2  }^{( 1 )}$ and $P_{3}=\sigma_{  1  }^{(3)}\sigma_{  2  }^{(1)}$ do not commute for example.

The fraction of the total number of elements counted in this way in the group indexed by $\boldsymbol{q}$ is
\begin{equation}
  \frac{(9n)(9n-1-\beta)\dots(9n-(\gamma-1)-(\gamma-1)\beta)}{r_{2}!\dots r_{s}!}\cdot
\frac{r_{1}!r_{2}!\dots r_{s}!}{(9n)!}
\end{equation}
Grouping factors and cancelling the factors $r_{2}!\dots r_{s}!$ reduces this to
\begin{equation}
  \frac{9n}{9n}\cdot\frac{9n-1-\beta}{9n-1}\dots\frac{9n-(\gamma-1)-(\gamma-1)\beta}{9n-(\gamma-1)}\cdot\frac{r_{1}!}{r_{1}!}
\end{equation}
This fraction is a lower bound (due to the conservative counting process) on the true ratio, $e_{q_{1},\dots,q_{9n}}$, of terms within the group characterised by $\boldsymbol{q}$, with coefficient $\cof_{m_{1},\dots,m_{9n}}=\frac{m!}{m_{1}!\dots m_{9n}!}=\frac{m!}{q_{1}!\dots q_{9n}!}$.

As $\gamma<m$ by construction, this ratio is then lower bounded by the uniform bound (only depending on $n$ and the fixed values of $m$ and $\beta$) of
\begin{equation}
  \frac{9n}{9n}\cdot\frac{9n-1-\beta}{9n-1}\dots\frac{9n-(m-1)-(m-1)\beta}{9n-(m-1)}=e_n
\end{equation}
which tends to unity as $n\to\infty$ for fixed values of $m$.  

\begin{figure}
  \centering
  \input{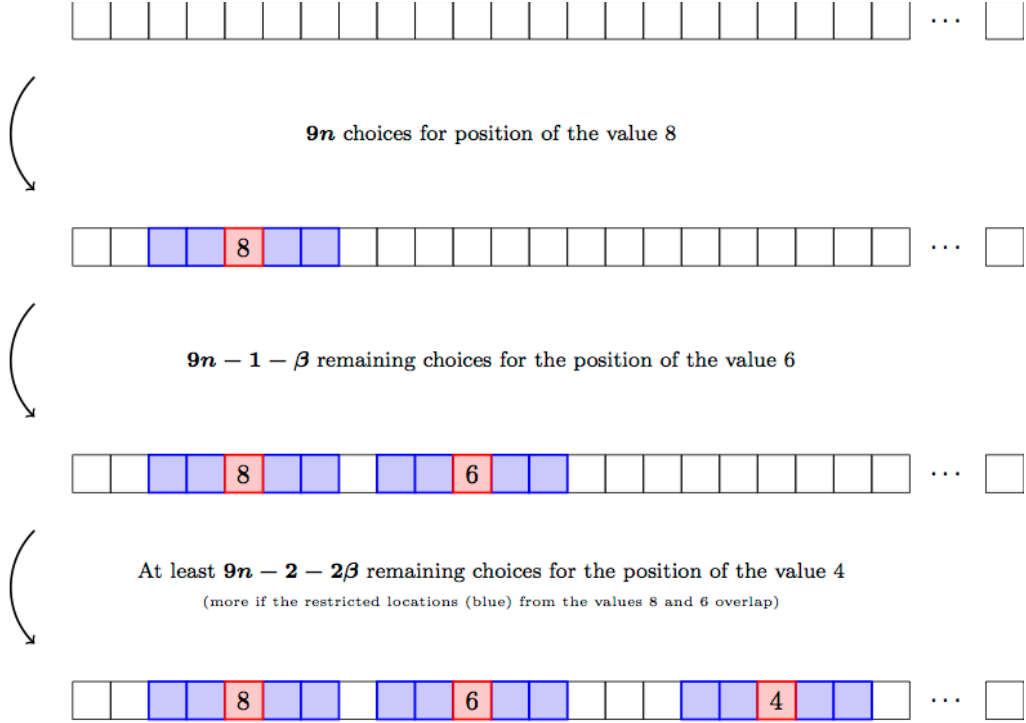}
  \caption[Combinatorial counting of vector permutations]{Diagram of the combinatoric counting process of all the vectors $\boldsymbol{m}=(m_{1},\dots,m_{9n})\in2\mathbb{N}_{0}^{9n}$ that are equivalent, up to permutation of their components, to $\boldsymbol{q}=(0,\dots,0,4,6,8)\in2\mathbb{N}_{0}^{9n}$, such that the matrices $\{P_{k}|m_{k}\neq0\}$ commute pairwise.  The boxes in each diagram represent the elements of $\boldsymbol{m}$ (in some suitable order).  Combinatorially there are $9n$ choices of location in $\boldsymbol{m}$ for the value $8$.  This choice then restricts the possible locations for the value $6$, in order to preserve the commuting property of the associated matrices $P_{k}$, by $\beta$ (blue) locations in addition to the location of the value $8$ (red).  Here $\beta=4$ for diagrammatic simplicity.  This process is then repeated for the value $4$.}
  \label{CombinatoricMoments}
\end{figure}

\subsubsection{Approximation to the expansion}
Now the expansion
\begin{equation}
  \left\langle\frac{1}{2^{n}}\Tr \left(H_{n}^{m}\right)\right\rangle=\sum_{\genfrac{}{}{0pt}{}{m_{1},\dots,m_{9n}\in2\mathbb{N}_{0}}{m_{1}+\dots+m_{9n}=m}}\cof_{m_{1},\dots,m_{9n}}\left\langle\alpha_{1}^{m_{1}}\dots\alpha_{9n}^{m_{9n}}\right\rangle
\end{equation}
can be rewritten using the facts shown above.  In summary: The terms may be grouped into groups characterised by the vector $\boldsymbol{q}=(q_{1},\dots,q_{9n})\in2\mathbb{N}_{0}^{9n}$ where $q_{1}\leq q_{2}\leq\dots\leq q_{9n}$ and $\|\boldsymbol{q}\|_{1}=m$, and where $r_{1},\dots,r_{s}$ denote the repetition numbers of each of the $s$ distinct values of the entries of $\boldsymbol{q}$.  In each group the value of the average of the product of the $\alpha_{k}$'s is constant and equal to $\left\langle \alpha_{1}^{q_{1}}\dots \alpha_{9n}^{q_{9n}}\right\rangle$.  Out of the $\frac{(9n)!}{r_{1}!\dots r_{s}!}$ terms in each group, the proportion $e_{q_{1},\dots,q_{9n}}$ have $\cof_{m_{1},\dots,m_{9n}}=\frac{m!}{m_{1}!\dots m_{9n}!}=\frac{m!}{q_{1}!\dots q_{9n}!}$ and the rest have $|\cof_{m_{1},\dots,m_{9n}}|\leq\frac{m!}{m_{1}!\dots m_{9n}!}=\frac{m!}{q_{1}!\dots q_{9n}!}$ with $e_{n}\leq e_{q_{1},\dots,q_{9n}}\leq1$.  Therefore the expansion above may be rewritten as
\begin{align}\label{expansion}
  \left\langle\frac{1}{2^{n}}\Tr \left(H_{n}^{m}\right)\right\rangle
  =\sum_{\genfrac{}{}{0pt}{}{q_{1}\leq\dots\leq q_{9n}\in2\mathbb{N}_{0}}{q_{1}+\dots+q_{9n}=m}}\left(\frac{(9n)!}{r_{1}!\dots r_{s}^{!}}\frac{m!}{q_{1}!\dots q_{9n}!}\left\langle \alpha_{1}^{q_{1}}\dots \alpha_{9n}^{q_{9n}}\right\rangle+E_{n}\right)
\end{align}
where
\begin{equation}
 \left|E_{n}\right|\leq2\left(1-e_{n}\right)\frac{(9n)!}{r_{1}!\dots r_{s}^{!}}\frac{m!}{q_{1}!\dots q_{9n}!}\left\langle \alpha_{1}^{q_{1}}\dots \alpha_{9n}^{q_{9n}}\right\rangle
\end{equation}

\subsubsection{Limiting moment value}
For fixed $m\in2\mathbb{N}$, the limiting moment $\left\langle\frac{1}{2^{n}}\Tr\left( H_{n}^{m}\right)\right\rangle$ will now be calculated by comparison to the multinomial theorem,
\begin{align}
  \left(\sum_{k=1}^{9n}\alpha_{k}\right)^{m}
  =\sum_{\genfrac{}{}{0pt}{}{m_{1},\dots,m_{9n}\in\mathbb{N}_{0}}{m_{1}+\dots+m_{9n}=m}}\frac{m!}{m_{1}!\dots m_{9n}!}\alpha_{1}^{m_{1}}\dots\alpha_{9n}^{m_{9n}}
\end{align}
The average of this expression, again by an analogous process to that already described, can be rewritten as
\begin{align}\label{2.4.37}
  \left\langle\left(\sum_{k=1}^{9n}\alpha_{k}\right)^{m}\right\rangle
  =\sum_{\genfrac{}{}{0pt}{}{q_{1}\leq\dots\leq q_{9n}\in2\mathbb{N}_{0}}{q_{1}+\dots+q_{9n}=m}}\frac{(9n)!}{r_{1}!\dots r_{s}!}\frac{m!}{q_{1}!\dots q_{9n}!}\left\langle\alpha_{1}^{q_{1}}\dots\alpha_{9n}^{q_{9n}}\right\rangle
\end{align}
All the terms on the right hand side of this expression are positive.  Therefore the identity (see appendix \ref{Int})
\begin{equation}\label{2.4.38}
  \left\langle\left(\sum_{k=1}^{9n}\alpha_{k}\right)^{m}\right\rangle=\frac{1}{\sqrt{2\pi}}\int_{-\infty}^{\infty}x^{m}\e^{-\frac{x^{2}}{2}}\di x = \frac{m!}{\left(\frac{m}{2}\right)!2^\frac{m}{2}}
\end{equation}
for even $m$ provides a bound for each term in the sum, that is
\begin{equation}
  \frac{(9n)!}{r_{1}!\dots r_{s}!}\frac{m!}{q_{1}!\dots q_{9n}!}\left\langle\alpha_{1}^{q_{1}}\dots\alpha_{9n}^{q_{9n}}\right\rangle\leq\frac{m!}{\left(\frac{m}{2}\right)!2^\frac{m}{2}}
\end{equation}
This enables the magnitude of $E_{n}$ in (\ref{expansion}) to be bounded in the large $n$ limit by
\begin{equation}
 \left|E_{n}\right|\leq2\left(1-e_{n}\right)\frac{m!}{\left(\frac{m}{2}\right)!2^\frac{m}{2}}
\end{equation}
which tends to zero as $n\to\infty$ since $e_{n}\to1$.  Hence, as $n\to\infty$, by the value for $\left\langle\frac{1}{2^{n}}\Tr\left(H_{n}^{m}\right)\right\rangle$ given in (\ref{expansion}) and the identities (\ref{2.4.37}) and (\ref{2.4.38}),
\begin{equation}
  \left\langle\frac{1}{2^{n}}\Tr\left(H_{n}^{m}\right)\right\rangle
  =\frac{1}{\sqrt{2\pi}}\int_{-\infty}^{\infty}x^{m}\e^{-\frac{x^{2}}{2}}\di x+\sum_{\genfrac{}{}{0pt}{}{q_{1}\leq\dots\leq q_{9n}\in2\mathbb{N}_{0}}{q_{1}+\dots+q_{9n}=m}}E_n
  \to\frac{1}{\sqrt{2\pi}}\int_{-\infty}^{\infty}x^{m}\e^{-\frac{x^{2}}{2}}\di x
\end{equation}
as the sum over $E_{n}$ contains a fixed maximum number of terms for all $n$, as $m$ is fixed.

If $m$ is odd it is recalled that the average of every term in the expansion of $\left\langle\frac{1}{2^{n}}\Tr\left(H_{n}^{m}\right)\right\rangle$ vanishes and the resulting moment is zero, concluding the proof.
\end{proof}

\subsection{Example: Fourth moment}
As an example of this calculation the limiting fourth moment, $\lim_{n\to\infty}\left\langle\frac{1}{2^{n}}\Tr\left(H_{n}^{4}\right)\right\rangle$, of the spectral density will be calculated.  Substituting the definition, $H_{n}=\sum_{k=1}^{9n}\alpha_{k}P_{k}$, and expanding the resulting product gives that
\begin{equation}
  \left\langle\frac{1}{2^{n}}\Tr\left(H_{n}^{4}\right)\right\rangle
  =\frac{1}{2^{n}}\sum_{k_{1},k_{2},k_{3},k_{4}=1}^{9n}\left\langle\alpha_{k_{1}}\alpha_{k_{2}}\alpha_{k_{3}}\alpha_{k_{4}}\right\rangle\Tr\left(P_{k_{1}}P_{k_{2}}P_{k_{3}}P_{k_{4}}\right)
\end{equation}
where only terms in which an even number of each $\alpha_{k}$ are present are non-zero, by the symmetry of the average.  This implies that an even number of each matrix $P_{k}$ is also present in the non-zero terms. The matrices $P_{k_{1}},P_{k_{2}},P_{k_{3}}$ and $P_{k_{4}}$ in each term are then either such that $P_{k_{1}}=P_{k_{2}}$ and $P_{k_{3}}=P_{k_{4}}$ or such that $P_{k_{1}}=P_{k_{4}}$ and $P_{k_{2}}=P_{k_{3}}$ or such that $P_{k_{1}}=P_{k_{3}}$ and $P_{k_{2}}=P_{k_{4}}$.  In the first two cases $P_{k_{1}}P_{k_{2}}P_{k_{3}}P_{k_{4}}=I_{2^{n}}$ since $P_{k}^{2}=I_{2^{n}}$ for every $k$.  In the third case $P_{k_{2}}$ and $P_{k_{3}}$ may be swapped, giving rise to a possible negative sign if $P_{k_{2}}$ and $P_{k_{3}}$ anti-commute, so that $P_{k_{1}}P_{k_{2}}P_{k_{3}}P_{k_{4}}=\pm I_{2^{n}}$.  Therefore $\frac{1}{2^{n}}\Tr\left(P_{k_{1}}P_{k_{2}}P_{k_{3}}P_{k_{4}}\right)=\pm1$ in each term.

Collecting terms then reduces this to
\begin{equation}\label{2.4.43}
  \sum_{\genfrac{}{}{0pt}{}{m_{1},\dots,m_{9n}\in2\mathbb{N}_{0}}{m_{1}+\dots+m_{9n}=4}}\cof_{m_{1},\dots,m_{9n}}\left\langle\alpha_{1}^{m_{1}}\dots\alpha_{9n}^{m_{9n}}\right\rangle
\end{equation}
for some integer coefficients $\cof_{m_{1},\dots,m_{9n}}$.  The set
\begin{equation}
	\{(m_{1},\dots,m_{9n})\in2\mathbb{N}_{0}^{9n}\,\,|\,\,m_{1}+\dots+m_{9n}=4\}
\end{equation}
has two distinct types of element; vectors $(m_{1},\dots,m_{9n})$ where only one element is non-zero and equal to $4$ and vectors $(m_{1},\dots,m_{9n})$ where exactly two elements are non-zero and both equal to $2$.  These two types of elements can be used to group the terms in (\ref{2.4.43}) into two groups, group $A$ and group $B$.

\subsubsection{Terms in group $A$}
Let the terms in group $A$ correspond to the vectors $(m_{1},\dots,m_{9n})$ where only one element is non-zero and equal to $4$.  There are exactly $\genfrac{(}{)}{0pt}{}{9n}{1}$  of these vectors. Using the notation in the generalised calculation above, there are $s=2$ distinct values in each vector, $0$ and $4$, with repetition numbers $r_{1}=9n-1$ and $r_{2}=1$ respectively.  Therefore the number of terms in this group is indeed given by the formula
\begin{equation}
  \frac{(9n)!}{r_{1}!r_{2}!}=\frac{(9n)!}{(9n-1)!1!}=\genfrac{(}{)}{0pt}{}{9n}{1}=9n
\end{equation}

For each term, the average $\left\langle\alpha_{1}^{m_{1}}\dots\alpha_{9n}^{m_{9n}}\right\rangle$, is the same and equal to
\begin{equation}
  \left\langle\alpha_{1}^{m_{1}}\dots\alpha_{9n}^{m_{9n}}\right\rangle
  =\left\langle\alpha_{1}^{4}\right\rangle
  =\sqrt{\frac{9n}{2\pi}}\int_{-\infty}^{\infty}\alpha_{1}^{4}\e^{-\frac{9n\alpha_{1}^{2}}{2}}\di\alpha_{1}=\frac{1}{27n^{2}}
\end{equation}
by the calculation in Appendix \ref{Int}.

For each term in group $A$, let $x$ denote the index of the entry of $(m_{1},\dots,m_{9n})$ that is non-zero.  By construction each term in group $A$ originates from `collecting' the term
\begin{equation}
\frac{1}{2^{n}}\left\langle\alpha_{x}\alpha_{x}\alpha_{x}\alpha_{x}\right\rangle\Tr\left(P_{x}P_{x}P_{x}P_{x}\right)
\end{equation}
where $k_{1}=k_{2}=k_{3}=k_{4}=x$ from the original sum.  In this term
\begin{equation}
 \frac{1}{2^{n}}\Tr\left(P_{k_{1}}P_{k_{2}}P_{k_{3}}P_{k_{4}}\right)=\frac{1}{2^{n}}\Tr \left(I_{2^{n}}\right)=1
\end{equation}
and the coefficient $\cof_{m_{1},\dots,m_{9n}}$ is exactly one.  This is indeed the multinomial coefficient seen before,
\begin{equation}
  \cof_{m_{1},\dots,m_{9n}}=1=\frac{4!}{4!}=\frac{m!}{m_{1}!\dots m_{9n}!}
\end{equation}

\subsubsection{Terms in group $B$}
Let the terms in group $B$ correspond to the vectors $(m_{1},\dots,m_{9n})$ where only two elements are non-zero and both equal to $2$.  There are exactly $\genfrac{(}{)}{0pt}{}{9n}{2}$ of these vectors or terms. Using the notation in the generalised calculation above, there are $s=2$ distinct values in each vector, $0$ and $2$, with repetition numbers $r_{1}=9n-2$ and $r_{2}=2$ respectively.  Therefore the number of terms in this group is indeed given by the formula
\begin{equation}
  \frac{(9n)!}{r_{1}!r_{2}!}=\frac{(9n)!}{(9n-2)!2!}=\genfrac{(}{)}{0pt}{}{9n}{2}
\end{equation}

For each term, the average $\left\langle\alpha_{1}^{m_{1}}\dots\alpha_{9n}^{m_{9n}}\right\rangle$, is the same and equal to
\begin{equation}
  \left\langle\alpha_{1}^{m_{1}}\dots\alpha_{9n}^{m_{9n}}\right\rangle
  =\left\langle\alpha_{1}^{2}\alpha_{2}^{2}\right\rangle
  =\left\langle\alpha_{1}^{2}\right\rangle\left\langle\alpha_{2}^{2}\right\rangle
  =\left(\sqrt{\frac{9n}{2\pi}}\int_{-\infty}^{\infty}\alpha_{1}^{2}\e^{-\frac{9n\alpha_{1}^{2}}{2}}\di \alpha_{1}\right)^{2}
  =\frac{1}{(9n)^{2}}
\end{equation}
by the calculation in Appendix \ref{Int}.

For each term in group $B$ let $x$ and $y$ denote the indices of the entries of $(m_{1},\dots,m_{9n})$ that are non-zero.  By construction each term in group B
originates from collecting six terms
\begin{align}
  &\frac{1}{2^{n}}\left\langle\alpha_{x}\alpha_{x}\alpha_{y}\alpha_{y}\right\rangle\Tr\left(P_{x}P_{x}P_{y}P_{y}\right),\qquad k_{1}=k_{2}=x,\quad k_{3}=k_{4}=y\nonumber\\
  &\frac{1}{2^{n}}\left\langle\alpha_{y}\alpha_{y}\alpha_{x}\alpha_{x}\right\rangle\Tr\left(P_{y}P_{y}P_{x}P_{x}\right),\qquad k_{3}=k_{4}=x,\quad k_{1}=k_{2}=y\nonumber\\
  &\frac{1}{2^{n}}\left\langle\alpha_{x}\alpha_{y}\alpha_{y}\alpha_{x}\right\rangle\Tr\left(P_{x}P_{y}P_{y}P_{x}\right),\qquad k_{1}=k_{4}=x,\quad k_{2}=k_{3}=y\nonumber\\
  &\frac{1}{2^{n}}\left\langle\alpha_{y}\alpha_{x}\alpha_{x}\alpha_{y}\right\rangle\Tr\left(P_{y}P_{x}P_{x}P_{y}\right),\qquad k_{2}=k_{3}=x,\quad k_{1}=k_{4}=y\nonumber\\
  &\frac{1}{2^{n}}\left\langle\alpha_{x}\alpha_{y}\alpha_{x}\alpha_{y}\right\rangle\Tr\left(P_{x}P_{y}P_{x}P_{y}\right),\qquad k_{1}=k_{3}=x,\quad k_{2}=k_{4}=y\nonumber\\
  &\frac{1}{2^{n}}\left\langle\alpha_{y}\alpha_{x}\alpha_{y}\alpha_{x}\right\rangle\Tr\left(P_{y}P_{x}P_{y}P_{x}\right),\qquad k_{2}=k_{4}=x,\quad k_{1}=k_{3}=y\nonumber\\
\end{align}
from the original sum.  For the first four cases
\begin{equation}
 \frac{1}{2^{n}}\Tr\left(P_{k_{1}}P_{k_{2}}P_{k_{3}}P_{k_{4}}\right)=\frac{1}{2^{n}}\Tr\left(I_{2^{n}}\right)=1
\end{equation}
but in the last two case a possible negative sign occurs,
\begin{equation}
 \frac{1}{2^{n}}\Tr\left(P_{k_{1}}P_{k_{2}}P_{k_{3}}P_{k_{4}}\right)=\begin{cases}
                                                            \frac{1}{2^{n}}\Tr\left( I_{2^{n}}\right)=1\quad&\text{if }P_{k_{2}}P_{k_{3}}=P_{k_{3}}P_{k_{2}}\\
							    -\frac{1}{2^{n}}\Tr \left(I_{2^{n}}\right)=-1\quad&\text{if }P_{k_{2}}P_{k_{3}}=-P_{k_{3}}P_{k_{2}}
                                                           \end{cases}
\end{equation}
When the matrices $P_{x}$ and $P_{y}$ commute, then the coefficient $\cof_{m_{1},\dots,m_{9n}}$ is the number of all such terms above, for a given value of $(m_{1},\dots,m_{9n})$, which is exactly six.  This is indeed the multinomial coefficient seen before,
\begin{equation}
  \cof_{m_{1},\dots,m_{9n}}=6=\frac{4!}{2!2!}=\frac{m!}{m_{1}!\dots m_{9n}!}
\end{equation}

Since any matrix $P_{k}$ anti-commutes with exactly $16$ other matrices $P_{k^\prime}$, out of the $\genfrac{(}{)}{0pt}{}{9n}{2}$ terms in the group $B$ only $9n\cdot16$ of them contain anti-commuting matrices ($9n$ choice for the first index $x$ and $16$ for the second index $y$ so that $P_{x}$ and $P_{y}$ anti-commute) so that for these terms
\begin{equation}
  |\cof_{m_{1},\dots,m_{9n}}|<6=\frac{4!}{2!2!}=\frac{m!}{m_{1}!\dots m_{9n}!}
\end{equation}

\subsubsection{Rewriting the sum}
Using the facts above, the sum for the fourth moment,
\begin{equation}
  \left\langle\frac{1}{2^{n}}\Tr\left(H_{n}^{4}\right)\right\rangle=\sum_{\genfrac{}{}{0pt}{}{m_{1},\dots,m_{9n}\in2\mathbb{N}_{0}}{m_{1}+\dots+m_{9n}=4}}\cof_{m_{1},\dots,m_{9n}}\left\langle\alpha_{1}^{m_{1}}\dots\alpha_{9n}^{m_{9n}}\right\rangle
\end{equation}
may be split into sums over each group $A$ and $B$.  Group $A$ containing $9n$ terms each equal to $1\cdot\left\langle\alpha_{1}^{4}\right\rangle$ and group $B$ containing $\genfrac{(}{)}{0pt}{}{9n}{2}-9n\cdot16$ terms equal to $6\cdot\left\langle\alpha_{1}^{2}\alpha_{2}^{2}\right\rangle$ and $9n\cdot16$ terms with modulus less than $6\cdot\left\langle\alpha_{1}^{2}\alpha_{2}^{2}\right\rangle$,
so that
\begin{align}
  \left\langle\frac{1}{2^{n}}\Tr\left(H_{n}^{4}\right)\right\rangle&=\genfrac{(}{)}{0pt}{}{9n}{1}\cdot1\cdot\left\langle\alpha_{1}^{4}\right\rangle+\genfrac{(}{)}{0pt}{}{9n}{2}\cdot6\cdot\left\langle\alpha_{1}^{2}\alpha_{2}^{2}\right\rangle+E\nonumber\\
  &=9n\cdot1\cdot\frac{1}{27n^{2}}+\frac{9n(9n-1)}{2}\cdot6\cdot\frac{1}{(9n)^{2}}+E
\end{align}
where
\begin{equation}
  |E|\leq2\cdot(9n\cdot16)\cdot6\cdot\left\langle\alpha_{1}^{2}\alpha_{2}^{2}\right\rangle=2\cdot9n\cdot16\cdot6\cdot\frac{1}{(9n)^{2}}
\end{equation}
Taking the limit $n\to\infty$ gives that
\begin{equation}
  \lim_{n\to\infty}\left\langle\frac{1}{2^{n}}\Tr\left(H_{n}^{4}\right)\right\rangle=3=\frac{1}{\sqrt{2\pi}}\int_{-\infty}^\infty x^{4}\e^{-\frac{x^{2}}{2}}\di x
\end{equation}
That is $\left\langle\frac{1}{2^{n}}\Tr\left(H_{n}^{4}\right)\right\rangle$ tends to the fourth moment of a standard normal distribution as $n\to\infty$.

\section{Limiting ensemble spectral density}\label{Separating odd and even sites}
To extend the results of the last section, the complete limiting spectral density for the ensembles $\hat{H}_{n}$ will now be calculated:

\subsection{Limiting spectral density for the ensembles \texorpdfstring{$\hat{H}_{n}$}{Hn}}\label{ThDOS}
\begin{theorem}[Convergence of the characteristic functions associated with the probability measures $\di\hat{\mu}_{n,1}$]\label{ThDOSTheorem}
  The characteristic function associated to the spectral probability measure $\di\hat{\mu}_{n,1}$, for the ensemble $\hat{H}_{n}$,
  \begin{equation}
    \hat{\psi}_{n}(t)=\int_{-\infty}^\infty\e^{\im t\lambda}\di\hat{\mu}_{n,1}(\lambda)
  \end{equation}
  satisfies
  \begin{equation}
    \left|\hat{\psi}_{n}(t)-\e^{-\frac{t^{2}}{2}}\right|\leq \frac{t^{2}\left(4\sqrt{2}+9\right)}{\sqrt{n}}
  \end{equation}
  for all $n=4,6,8,\dots$.
\end{theorem}

The generalisation to all $n=2,3,\dots$ will be made in Section \ref{Ccolour}.  The following corollary of this theorem follows directly from the continuity theorem:
\begin{corollary}[Limiting spectral density for the ensembles $\hat{H}_{n}$]
  The spectral probability measures $\di\hat{\mu}_{n,1}$ tend weakly to the probability measure of a standard normally distributed random variable, that is, for each $x\in\mathbb{R}$
  \begin{equation}
    \lim_{n\to\infty}\int_{-\infty}^{x}\di\hat{\mu}_{n,1}=\frac{1}{\sqrt{2\pi}}\int_{-\infty}^{x}\e^{-\frac{\lambda^{2}}{2}}\di\lambda
  \end{equation}
  over all $n\in2\mathbb{N}$.
\end{corollary}

\begin{proof}\hspace{-2mm}\footnote{Adapted from the proof given by the author in \cite{KLW2014_RMT}.}
For the purposes of the proof, let $n\geq4$ be an even integer.  The proof will proceed as follows:  First, the spin chain will be split into sections separating every other qubit interaction.  Then, the product of the characteristic functions associated to the spectral probability measures for each component of the split chain, $\hat{\phi}_{n}(t)$, will be shown to be close to the characteristic function associated to the full chain, $\hat{\psi}_{n}(t)$, and that $\hat{\phi}_{n}(t)$ is indeed equal to the characteristic function $\e^{-\frac{t^{2}}{2}}$ for all $n$, as desired:

\subsubsection{Separation}
Let the member
\begin{equation}
  H_{n}=\sum_{j=1}^{n}\sum_{a,b=1}^{3}\alpha_{a,b,j}\sigma_{  j  }^{(a)}\sigma_{  j+1  }^{(b)}
\end{equation}
of the ensemble $\hat{H}_{n}$ be split into two sections, $A$ and $B$, such that $H_{n}=A+B$ and where $A$ contains all the terms from $H_{n}$ that act between two sites where the lowest site index is even and likewise $B$ where the lowest site index is odd.  That is, let
\begin{align}
  A&=\sum_{\genfrac{}{}{0pt}{}{j=1}{j\text{ even}}}^{n}\sum_{a,b=1}^{3}\alpha_{a,b,j}\sigma_{  j  }^{(a)}\sigma_{  j+1  }^{(b)}\qquad\qquad
  B=\sum_{\genfrac{}{}{0pt}{}{j=1}{j\text{ odd}}}^{n}\sum_{a,b=1}^{3}\alpha_{a,b,j}\sigma_{  j  }^{(a)}\sigma_{  j+1  }^{(b)}
\end{align}
Furthermore let $A=\sum_{k}A_{k}$ where
\begin{equation}
  A_{3(b-1)+a}=\sum_{ \genfrac{}{}{0pt}{}{j=1}{j\text{ even}} }^{n} \alpha_{a,b,j}\sigma_{  j  }^{( a )}\sigma_{  j+1  }^{( b )}
\end{equation}
and let $B=\sum_{k}B_{k}$ where
\begin{equation}
  B_{3(b-1)+a}=\sum_{ \genfrac{}{}{0pt}{}{j=1}{j\text{ odd}} }^{n} \alpha_{a,b,j}\sigma_{  j  }^{( a )}\sigma_{  j+1  }^{( b )}
\end{equation}
Figure \ref{EvenOdd} gives a graphical representation of this splitting.

\begin{figure}
  \centering
  \input{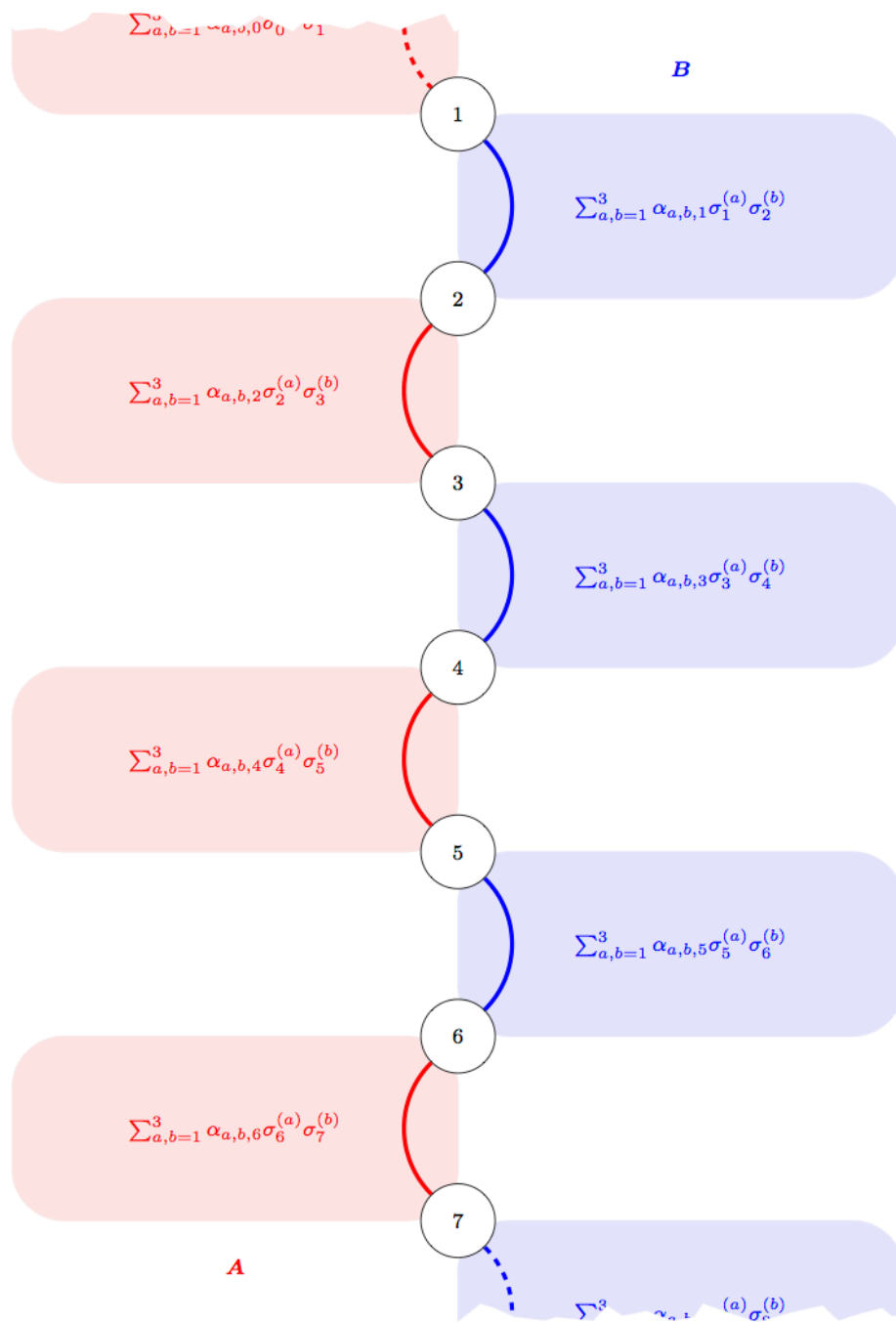}
  \caption[Splitting a chain into odd and even terms]{A representation of the splitting of the terms present in the Hamiltonian $H_{n}$.  The circles represent qubits and the links the interaction terms in $H_{n}$.  Matrices $A$ and $B$ are defined such that $H_{n}=A+B$ where $A$ contains all the terms from $H_{n}$ that act between two sites where the lowest site index is even (red terms on the left in the diagram) and likewise $B$ where the lowest site index is odd (blue terms on the right in the diagram).}
  \label{EvenOdd}
\end{figure}

\subsubsection{The characteristic function $\hat{\phi}_{n}(t)$}
The ensemble
\begin{equation}
  \hat{H}_{n}=\sum_{j=1}^{n}\sum_{a,b=1}^{3}\hat{\alpha}_{a,b,j}\sigma_{  j  }^{(a)}\sigma_{  j+1  }^{(b)}\qquad \qquad\hat{\alpha}_{a,b,j}\sim\mathcal{N}\left(0,\frac{1}{9n}\right) \iid
 \end{equation}
induces a probability measure on the matrices $A_{k}$ and $B_{k}$, the variables $\alpha_{a,b,j}$ in $A_{k}$ and $B_{k}$ corresponding to the random variables $\hat{\alpha}_{a,b,j}$.

The corresponding ensembles of matrices $\hat{A}_{k}$ and $\hat{B}_{k}$ both have an associated spectral probability measure, as defined Section \ref{PointCorrFunc}.  Let their corresponding spectral characteristic functions be $\hat{\psi}_{n}^{(A_{k})}(t)$ and $\hat{\psi}_{n}^{(B_{k})}(t)$ respectively.  Let the characteristic function $\hat{\phi}_{n}(t)$ be the product of these functions, that is, 
\begin{align}\label{OddEvenPhin}
  \hat{\phi}_{n}(t)&=\prod_{k=1}^{9}\hat{\psi}_{n}^{(A_{k})}(t)\hat{\psi}_{n}^{(B_{k})}(t)\nonumber\\
  &=\prod_{k=1}^{9}\left\langle\frac{1}{2^{n}}\Tr\left(\e^{\im tA_{k}}\right)\right\rangle\left\langle\frac{1}{2^{n}}\Tr\left(\e^{\im tB_{k}}\right)\right\rangle
\end{align}

\subsubsection{Interpretation of $\hat{\phi}_{n}(t)$}
Let $\hat{\lambda}_{A_{k}}$ and $\hat{\lambda}_{B_{k}}$ be real independent random variables with probability measures equal to the spectral probability measures for the ensembles of matrices $\hat{A}_{k}$ and $\hat{B}_{k}$ respectively.  The characteristic function of their sum, that is the characteristic function of the random variable
\begin{equation}
  \hat{\lambda}=\sum_{k=1}^{9}\left(\hat{\lambda}_{A_{k}}+\hat{\lambda}_{B_{k}}\right)
\end{equation}
is, by definition,
\begin{equation}
  \mathbb{E}\left(\e^{\im t\hat{\lambda}}\right)=\mathbb{E}\left(\prod_{k=1}^{9}\e^{\im t\hat{\lambda}_{A_{k}}}\e^{\im t\hat{\lambda}_{B_{k}}}\right)
\end{equation}
As all $\hat{\lambda}_{A_{k}}$ and $\hat{\lambda}_{B_{k}}$ are independent random variables by definition, it follows that this is equal to
\begin{equation}
  \prod_{k=1}^{9}\mathbb{E}\left(\e^{\im t\hat{\lambda}_{A_{k}}}\right)\mathbb{E}\left(\e^{\im t\hat{\lambda}_{B_{k}}}\right)  
\end{equation}
which is exactly the characteristic function $\hat{\phi}_{n}(t)$, by the definition of $\hat{\psi}_{n}^{(A_{k})}(t)$ and $\hat{\psi}_{n}^{(B_{k})}(t)$.

\subsubsection{Calculating the characteristic function $\hat{\phi}_{n}(t)$}
Swapping the trace and the average in the definition of $\hat{\phi}_{n}(t)$ gives that,
\begin{equation}\label{2.5.13}
  \hat{\phi}_{n}(t)=\prod_{k=1}^{9}\frac{1}{2^{n}}\Tr\left(\left\langle\e^{\im tA_{k}}\right\rangle\right)\frac{1}{2^{n}}\Tr\left(\left\langle\e^{\im tB_{k}}\right\rangle\right)
\end{equation}
Here, the average of a matrix is defined in an entry-wise fashion.   That is for any square matrix $M$, $\left\langle M\right\rangle_{j,k}=\left\langle M_{j,k}\right\rangle$ so that $\Tr\left(\left\langle M \right\rangle\right)=\sum_{j}\left\langle M\right\rangle_{j,j}=\sum_{j}\left\langle M_{j,j}\right\rangle=\left\langle\Tr \left(M\right) \right\rangle$.  This definition is basis independent.

The matrices $\sigma_{  j  }^{( a )}\sigma_{  j+1  }^{( b )}$, for all even values of $j$ and fixed values of $a$ and $b$, pairwise commute.  For example, two such matrices $\sigma_{  j  }^{( a )}\sigma_{  j+1  }^{( b )}$ and $\sigma_{  j+2  }^{( a)}\sigma_{  j+2  }^{( b)}$ commute as they act on completely separate qubits

Recalling that $A_{3(b-1)+a}=\sum_{ j\text{ even} } \alpha_{a,b,j}\sigma_{  j  }^{( a )}\sigma_{  j+1  }^{( b )}$ gives that
\begin{align}\label{OddEvenA1Spilt}
  \e^{\im tA_{3(b-1)+a}}&=\e^{\im t\sum_{j\text{ even}} \alpha_{a,b,j}\sigma_{  j  }^{( a )}\sigma_{  j+1  }^{( b )}}\nonumber\\
  &=\prod_{ \genfrac{}{}{0pt}{}{j=1}{j\text{ even}} }^{n}\e^{\im t\alpha_{a,b,j}\sigma_{  j  }^{( a )}\sigma_{  j+1  }^{( b )}}
\end{align}
Since all the parameters $\alpha_{a,b,j}$ correspond to the independent random variables $\hat{\alpha}_{a,b,j}$, it then follows that
\begin{equation}\label{OddEvenAvgAi}
  \left\langle\e^{\im tA_{3(b-1)+a}}\right\rangle=\prod_{ \genfrac{}{}{0pt}{}{j=1}{j\text{ even}} }^{n}\left\langle\e^{\im t\alpha_{a,b,j}\sigma_{  j  }^{( a )}\sigma_{  j+1  }^{( b )}}\right\rangle
\end{equation}
and similarly, for the odd values of $j$,
\begin{equation}
  \left\langle\e^{\im tB_{3(b-1)+a}}\right\rangle=\prod_{ \genfrac{}{}{0pt}{}{j=1}{j\text{ odd}} }^{n}\left\langle\e^{\im t\alpha_{a,b,j}\sigma_{  j  }^{( a )}\sigma_{  j+1  }^{( b )}}\right\rangle
\end{equation}

The average $\left\langle\e^{\im t\alpha_{a,b,j}\sigma_{  j  }^{( a )}\sigma_{  j+1  }^{( b )}}\right\rangle$ will now be evaluated to proceed further.  It is seen that
\begin{equation}
  \left(\sigma_{j}^{(a)}\sigma_{j+1}^{(b)}\right)^2=\sigma_{  j  }^{( a )}\sigma_{  j+1  }^{( b )}\sigma_{  j  }^{( a )}\sigma_{  j+1  }^{( b )}=\sigma_{  j  }^{( a )}\sigma_{  j  }^{( a )}\sigma_{  j+1  }^{( b )}\sigma_{  j+1  }^{( b )}=I_{2^{n}}
\end{equation}
as $\sigma_{  j+1  }^{( b )}$ and $\sigma_{  j  }^{( a )}$ commute and the square of any Pauli matrix is the identity.  Therefore, by the definition of the matrix exponential,
\begin{align}\label{EvenOddExp}
  \e^{\im t \alpha_{a,b,j}\sigma_{  j  }^{( a )}\sigma_{  j+1  }^{( b )}}
  &=\sum_{l=0}^\infty\frac{\im^{l}t^{l}\alpha_{a,b,j}^{l}\left(\sigma_{  j  }^{( a )}\sigma_{  j+1  }^{( b )}\right)^{l}}{l!}\nonumber\\
  &=\sum_{l=0}^\infty\frac{(-1)^{l}t^{2l}\alpha_{a,b,j}^{2l}}{(2l)!}I_{2^{n}}
   +\sum_{l=0}^\infty\frac{\im(-1)^{l}t^{2l+1}\alpha_{a,b,j}^{2l+1}}{(2l+1)!}\sigma_{  j  }^{( a )}\sigma_{  j+1  }^{( b )}\nonumber\\
  &=\cos(t \alpha_{a,b,j})I_{2^{n}}+\im\sin(t \alpha_{a,b,j})\sigma_{  j  }^{( a )}\sigma_{  j+1  }^{( b )}
\end{align}
so that
\begin{align}
  \left\langle\e^{\im t\alpha_{a,b,j}\sigma_{  j  }^{( a )}\sigma_{  j+1  }^{( b )}}\right\rangle
  &=\Big\langle\cos(t \alpha_{a,b,j})\Big\rangle I_{2^{n}}
    +\im\Big\langle\sin(t \alpha_{a,b,j})\Big\rangle\sigma_{  j  }^{( a )}\sigma_{  j+1  }^{( b )}\nonumber\\
  &=\Big\langle\cos(t \alpha_{a,b,j})\Big\rangle I_{2^{n}}
\end{align}
as $\left\langle\sin(t \alpha_{a,b,j})\right\rangle=0$ by the symmetry of the distributions of the corresponding random variables $\hat{\alpha}_{a,b,j}$.  Substituting this into the expression for $\left\langle\e^{\im tA_{3(b-1)+a}}\right\rangle$ above, yields that
\begin{equation}
  \left\langle\e^{\im tA_{3(b-1)+a}}\right\rangle
  =\prod_{ \genfrac{}{}{0pt}{}{j=1}{j\text{ even}} }^{n}
  \Big\langle\cos(t \alpha_{a,b,j})\Big\rangle I_{2^{n}}
\end{equation}
and likewise
\begin{equation}
  \left\langle\e^{\im tB_{3(b-1)+a}}\right\rangle=\prod_{ \genfrac{}{}{0pt}{}{j=1}{j\text{ odd}} }^{n}
  \Big\langle\cos(t \alpha_{a,b,j})\Big\rangle I_{2^{n}}
\end{equation}
The exact value of $\hat{\phi}_{n}(t)$ can now be  seen to be
\begin{equation}\label{EvenOddPhi}
  \prod_{j=1}^{n}\prod_{a,b=1}^{3}\Big\langle\cos(t \alpha_{a,b,j})\Big\rangle
  =\prod_{j=1}^{n}\prod_{a,b=1}^{3}\e^{-\frac{t^2}{2\cdot 9n}}
  =\e^{-\frac{t^{2}}{2}}
\end{equation}
as $\langle\cos(t \alpha_{a,b,j})\rangle=\e^{-\frac{t^{2}}{2\cdot 9n}}$, as seen in appendix \ref{Int}.  The $\hat{\phi}_{n}(t)$ are then equal for all $n$ and are the characteristic function of a standard normally distributed random variable, that is $\hat{\phi}_{n}(t)=\hat{\phi}(t)=\e^{-\frac{t^{2}}{2}}$.

As $\left\langle\e^{\im tA_{k}}\right\rangle$  and $\left\langle\e^{\im tB_{k}}\right\rangle$ are proportional to the identity, the expression (\ref{2.5.13})
\begin{equation}
  \prod_{k=1}^{9}\frac{1}{2^{n}}\Tr\left(\left\langle\e^{\im tA_{k}}\right\rangle\right)\frac{1}{2^{n}}\Tr\left(\left\langle\e^{\im tB_{k}}\right\rangle\right)
\end{equation}
for $\hat{\phi}_{n}(t)$, can be equivalently written as
\begin{equation}\label{2.5.24}
  \frac{1}{2^{n}}\Tr\left(\prod_{k=1}^{9}\left\langle\e^{\im tA_{k}}\right\rangle\left\langle\e^{\im tB_{k}}\right\rangle\right)
\end{equation}
where the order of the factors $\left\langle\e^{\im tA_{k}}\right\rangle$ and $\left\langle\e^{\im tB_{k}}\right\rangle$ is irrelevant.  By the statistical independence of the $\hat{A}_{k}$ and $\hat{B}_{k}$ is furthermore equivalent to
\begin{equation}
  \frac{1}{2^{n}}\Tr\left(\left\langle\prod_{k=1}^{9}\e^{\im tA_{k}}\e^{\im tB_{k}}\right\rangle\right)
\end{equation}
or, by again swapping the average and the trace, equivalent to
\begin{equation}
  \left\langle\frac{1}{2^{n}}\Tr\left(\prod_{k=1}^{9}\e^{\im tA_{k}}\e^{\im tB_{k}}\right)\right\rangle
\end{equation}

It now only remains to be seen that
\begin{equation}
     \left|\hat{\psi}_{n}(t)-\hat{\phi}_{n}(t)\right|\leq \frac{t^{2}\left(4\sqrt{2}+9\right)}{\sqrt{n}}
\end{equation}
or equivalently, as $\psi_{n}(t)=\left\langle\frac{1}{2^{n}}\Tr\left(\e^{\im tH_{n}}\right)\right\rangle$ by Section \ref{PointCorrFunc}, that
\begin{equation}
    \left|\left\langle\frac{1}{2^{n}}\Tr\left(\e^{\im tH_{n}}\right)\right\rangle-\left\langle\frac{1}{2^{n}}\Tr\left(\prod_{k=1}^{9}\e^{\im tA_{k}}\e^{\im tB_{k}}\right)\right\rangle\right|\leq \frac{t^{2}\left(4\sqrt{2}+9\right)}{\sqrt{n}}
\end{equation}
To do this, the following integral identity, similar to that used in \cite{Mahler}, will be required:

\subsubsection{Integral identity}
For any $2^{n}\times2^{n}$ Hermitian matrices $X$ and $Y$, the fundamental theorem of calculus gives that
\begin{align}
  \e^{\im t (X+Y)}-\e^{\im t X}\e^{\im t Y}
  &=-\e^{\im t(1-s)(X+Y)}\e^{\im tsX}\e^{\im tsY}\Big|_{s=1}+\e^{\im t(1-s)(X+Y)}\e^{\im tsX}\e^{\im tsY}\Big|_{s=0}\nonumber\\
  &=-\int_{0}^{1}\frac{\partial}{\partial s}\left(\e^{\im t(1-s)(X+Y)}\e^{\im tsX}\e^{\im tsY}\right)\di s
\end{align}
so that evaluating the derivative gives
\begin{align}
  \e^{\im t (X+Y)}-\e^{\im t X}\e^{\im t Y}
  &=-\int_{0}^{1} -\e^{\im t(1-s)(X+Y)}\im t(X+Y)\e^{\im tsX}\e^{\im tsY}\nonumber\\
	    &\qquad\qquad+\e^{\im t(1-s)(X+Y)}\im tX\e^{\im tsX}\e^{\im tsY}\nonumber\\
	    &\qquad\qquad\qquad+\e^{\im t(1-s)(X+Y)}\e^{\im tsX}\im tY\e^{\im tsY}\di s\nonumber\\
  &=\im t\int_{0}^{1} \e^{\im t(1-s)(X+Y)}\left[Y,\e^{\im tsX}\right]\e^{\im tsY}\di s
\end{align}
In addition to this, the fundamental theorem of calculus also gives that
\begin{align}
  \left[Y,\e^{\im tsX}\right]&=Y\e^{\im tsX}-\e^{\im tsX}Y\nonumber\\
  &=\e^{\im ts(1-r)X}Y\e^{\im tsrX}\Big|_{r=1}-\e^{\im ts(1-r)X}Y\e^{\im tsrX}\Big|_{r=0}\nonumber\\
  &=\int_{0}^{1}\frac{\partial}{\partial r}\left(\e^{\im ts(1-r)X}Y\e^{\im tsrX}\right)\di r
\end{align}
so that evaluating the derivative gives
\begin{align}
  \left[Y,\e^{\im tsX}\right]
  &=\int_{0}^{1}-\e^{\im ts(1-r)X}\im tsXY\e^{\im tsrX}\nonumber\\
  &\qquad\qquad+\e^{\im ts(1-r)X}Y\im tsX\e^{\im tsrX}\di r\nonumber\\
  &=-\im ts\int_{0}^{1}\e^{\im ts(1-r)X}[X,Y]\e^{\im tsrX}\di r
\end{align}

Combining both the last identities shows that
\begin{equation}
  \e^{\im t (X+Y)}-\e^{\im t X}\e^{\im t Y}
  =t^{2}\int_{0}^{1}s\int_{0}^{1} \e^{\im t(1-s)(X+Y)}\e^{\im ts(1-r)X}[X,Y]\e^{\im tsrX}\e^{\im tsY}\di r \di s
\end{equation}
so that by the triangle inequality, $\left|\frac{1}{2^{n}}\Tr\left(\e^{\im t (X+Y)}-\e^{\im t X}\e^{\im t Y}\right)\right|$ is bounded above by
\begin{equation}
  t^2\int_{0}^{1}|s|\int_{0}^{1} \left|\frac{1}{2^{n}}\Tr\left(\e^{\im t(1-s)(X+Y)}\e^{\im ts(1-r)X}[X,Y]\e^{\im tsrX}\e^{\im tsY}\right)\right|\di r \di s
\end{equation}
Moreover for any $2^{n}\times2^{n}$ unitary matrix $U$, multiplication through by $U$ before taking the trace and applying the triangle inequality, yields that $\left|\frac{1}{2^{n}}\Tr\left( U\left(\e^{\im t (X+Y)}-\e^{\im t X}\e^{\im t Y}\right)\right)\right|$ is bounded above by
\begin{equation}\label{2.145}
  t^{2}\int_{0}^{1}|s|\int_{0}^{1} \left|\frac{1}{2^{n}}\Tr\left(U\e^{\im t(1-s)(X+Y)}\e^{\im ts(1-r)X}[X,Y]\e^{\im tsrX}\e^{\im tsY}\right)\right|\di r \di s
\end{equation}

In order to bound this quantity a bound on the positive integrand
\begin{equation}
  \left|\frac{1}{2^{n}}\Tr \left(U\e^{\im t(1-s)(X+Y)}\e^{\im ts(1-r)X}[X,Y]\e^{\im tsrX}\e^{\im tsY}\right)\right|
\end{equation}
will be found with the aid of the Cauchy-Schwartz inequality:

\subsubsection{Cauchy-Schwartz inequality}
The Cauchy-Schwartz inequality states that for any $x_{1},\dots,x_{2^{n}}\in\mathbb{C}$ and any $y_{1},\dots,y_{2^{n}}\in\mathbb{C}$ that
\begin{equation}
  \left|\sum_{j=1}^{2^{n}} x_{j}\overline{y_{j}}\right|^{2}\leq\sum_{j=1}^{2^{n}}|x_{j}|^{2}\sum_{k=1}^{2^{n}}|y_{k}|^{2}
\end{equation}
where the bar denotes complex conjugation.

For any $2^{n}\times2^{n}$ matrix $M=\left(M_{j,k}\right)$, $\Tr \left(M\right)=\sum_{j}M_{j,j}$ and by the Cauchy-Schwartz inequality above, with $x_{j}=M_{j,j}$ and $y_{j}=1$,
\begin{equation}
  |\Tr \left(M\right)|^{2}=\left|\sum_{j=1}^{2^{n}}M_{j,j}\cdot1\right|^{2}\leq2^{n}\sum_{j=1}^{2^{n}}|M_{j,j}|^{2}
\end{equation}
A further inequality may be deduced by adding the positive terms $|M_{j,k}|^{2}$ to the right hand side of this last inequality, yielding
\begin{equation}\label{2.114}
  |\Tr \left(M\right)|^{2}\leq2^{n}\sum_{j,k=1}^{2^{n}}|M_{j,k}|^{2}=2^{n}\sum_{j,k=1}^{2^{n}}M_{j,k}\overline{M_{j,k}}=2^{n}\Tr\left(MM^\dagger\right)
\end{equation}

With $M=U\e^{\im t(1-s)(X+Y)}\e^{\im ts(1-r)X}[X,Y]\e^{\im tsrX}\e^{\im tsY}$ this inequality yields that
\begin{equation}
  \left|\frac{1}{2^{n}}\Tr\left(U\e^{\im t(1-s)(X+Y)}\e^{\im ts(1-r)X}[X,Y]\e^{\im tsrX}\e^{\im tsY}\right)\right|
  \leq\sqrt{\frac{1}{2^{n}}\Tr\left([X,Y][X,Y]^\dagger\right)}
\end{equation}
so that, substituting this into the integral bound derived in (\ref{2.145}) gives
\begin{align}
  \left|\frac{1}{2^{n}}\Tr\left(U\left(\e^{\im t (X+Y)}-\e^{\im t X}\e^{\im t Y}\right)\right)\right|
  &\leq t^{2}\int_{0}^{1}|s|\int_{0}^{1}\sqrt{\frac{1}{2^{n}}\Tr\left([X,Y][X,Y]^\dagger\right)}\di r\di s\nonumber\\
  &=\frac{t^{2}}{2}\sqrt{\frac{1}{2^{n}}\Tr\left([X,Y][X,Y]^\dagger\right)}
\end{align}
From now on let the norm a $2^{n}\times 2^{n}$ matrix, $M$, be defined as the scaled Hilbert-Schmidt norm
\begin{equation}
  \left\|M\right\|=\sqrt{\frac{1}{2^{n}}\Tr\left(MM^\dagger\right)}
\end{equation}
as in appendix \ref{Norms}.

\subsubsection{Triangle inequity}
The triangle inequality now allows the quantity of interest, $\left|\hat{\psi}_{n}(t)-\hat{\phi}_{n}(t)\right|$, to be studied.  That is, by the triangle inequality
\begin{align}
  \left|\hat{\psi}_{n}(t)-\hat{\phi}_{n}(t)\right|&=\left|\left\langle\frac{1}{2^{n}}\Tr\left(\e^{\im tH_{n}}\right)\right\rangle-\left\langle\frac{1}{2^{n}}\Tr\left(\prod_{k=1}^{9}\e^{\im tA_{k}}\e^{\im tB_{k}}\right)\right\rangle\right|\nonumber\\
  &\leq
  \left\langle\left|\frac{1}{2^{n}}\Tr\left(\e^{\im tH_{n}}\right)-\frac{1}{2^{n}}\Tr\left(\prod_{k=1}^{9}\e^{\im tA_{k}}\e^{\im tB_{k}}\right)\right|\right\rangle
\end{align}
Note that by (\ref{2.5.24}) the order of the factors $\e^{\im tA_{k}}$ and $\e^{\im tB_{k}}$ is irrelevant in this expression.  Recalling that $A=\sum_{k}A_{k}$, $B=\sum_{k}B_{k}$ and $H_{n}=A+B$, relabel the matrices $A_{k}$ and $B_{k}$ as $Q_{k}=A_{k}$ for $1\leq k\leq 9$ and $Q_{k}=A_{k-9}$ for $10\leq k\leq 18$ so that $H_{n}=\sum_{k}Q_{k}$.  The triangle inequality applied to the telescoping sum
\begin{align}
  \frac{1}{2^{n}}\Tr\left(\e^{\im tH_{n}}-\prod_{k=1}^{18}\e^{\im tQ_{k}}\right)
  &=\sum_{s=2}^{18}\frac{1}{2^{n}}\Tr\left(\left(\e^{\im t\sum_{j=1}^{s}Q_{j}}-\e^{\im t\sum_{k=1}^{s-1}Q_{k}}\e^{\im tQ_{s}}\right)\prod_{s<l\leq18}\e^{\im tQ_{l}}\right)
\end{align}
yields that
\begin{align}
  \left|\frac{1}{2^{n}}\Tr\left(\e^{\im tH_{n}}-\prod_{k=1}^{18}\e^{\im tQ_{k}}\right)\right|
  &\leq\sum_{s=2}^{18}\left|\frac{1}{2^{n}}\Tr\left(\left(\e^{\im t\sum_{j=1}^{s}Q_{j}}-\e^{\im t\sum_{k=1}^{s-1}Q_{k}}\e^{\im tQ_{s}}\right)\prod_{s<l\leq18}\e^{\im tQ_{l}}\right)\right|
\end{align}
so that by the bound just calculated,
\begin{align}
  \left|\frac{1}{2^{n}}\Tr\left(\e^{\im tH_{n}}-\prod_{k=1}^{18}\e^{\im tQ_{k}}\right)\right|
  \leq\frac{t^{2}}{2}\sum_{s=2}^{18}\left\|\left[\sum_{k=1}^{s-1}Q_{k},Q_{s}\right]\right\|
  &=\frac{t^{2}}{2}\sum_{1\leq k< k^\prime\leq18}\left\|\left[Q_{k},Q_{k^\prime}\right]\right\|
\end{align}

The averages of the norms $\left\|[Q_{k},Q_{k^\prime}]\right\|$, that is of $\left\|[A_{k},A_{k^\prime}]\right\|$, $\left\|[B_{k},B_{k^\prime}]\right\|$ and $\left\|[A_{k},B_{k^\prime}]\right\|$, will now be calculated:

\subsubsection{The norm $\left\|[A_{k},A_{k^\prime}]\right\|$}
By the definition of $A_{k}$ and the linearity of the commutator
\begin{align}
  [A_{k},A_{k^\prime}]&=\sum_{ \genfrac{}{}{0pt}{}{j=1}{j\text{ even}} }^{n}\sum_{ \genfrac{}{}{0pt}{}{j^\prime=1}{j^\prime\text{ even}} }^{n}
   \alpha_{a,b,j}\alpha_{a^\prime,b^\prime,j^\prime}\left[\sigma_{  j  }^{(a)}\sigma_{  j+1  }^{(b)},\sigma_{  j^\prime  }^{( a^\prime )}\sigma_{  j^\prime+1  }^{( b^\prime )}\right]
\end{align}
Only terms in this sum where $j^\prime=j$  may possibly have a non-zero commutator, for the other terms the matrices $\sigma_{  j  }^{(a)}\sigma_{  j+1  }^{(b)}$ and $\sigma_{  j^\prime  }^{( a^\prime )}\sigma_{  j^\prime+1  }^{( b^\prime )}$ act on completely separate qubits and therefore commute.  This enables the sum to be reduced to
\begin{align}
  [A_{k},A_{k^\prime}]&=\sum_{ \genfrac{}{}{0pt}{}{j=1}{j\text{ even}} }^{n}
   \alpha_{a,b,j}\alpha_{a^\prime,b^\prime,j}\left[\sigma_{  j  }^{(a)}\sigma_{  j+1  }^{(b)},\sigma_{  j  }^{( a^\prime )}\sigma_{  j+1  }^{( b^\prime )}\right]
\end{align}
By the definition of the norm, $\|[A_{k},A_{k^\prime}]\|$ is equal to
\begin{equation}
  \sqrt{\frac{1}{2^{n}}\Tr\left([A_{k},A_{k^\prime}][A_{k},A_{k^\prime}]^\dagger\right)}
\end{equation}
which, by substituting the expression for $[A_{k},A_{k^\prime}]$ above, is equal to
\begin{equation}
 \left(\frac{1}{2^{n}}\Tr\left(\sum_{ \genfrac{}{}{0pt}{}{j,l=1}{j,l\text{ even}}}^{n}\alpha_{a,b,j}\alpha_{a^\prime,b^\prime,j}\alpha_{a,b,l}\alpha_{a^\prime,b^\prime,l}\left[\sigma_{  j  }^{(a)}\sigma_{  j+1  }^{(b)},\sigma_{  j  }^{( a^\prime )}\sigma_{  j+1  }^{( b^\prime )}\right]\left[\sigma_{  l  }^{(a)}\sigma_{  l+1  }^{(b)},\sigma_{  l  }^{( a^\prime )}\sigma_{  l+1  }^{( b^\prime )}\right]^\dagger\right)\right)^{\frac{1}{2}}
\end{equation}

Jensen's inequality provides an upper bound for the average of this quantity of
\begin{equation}
 \left(\frac{1}{2^{n}}\Tr\left(\sum_{ \genfrac{}{}{0pt}{}{j,l=1}{j,l\text{ even}}}^{n}\left\langle\alpha_{a,b,j}\alpha_{a^\prime,b^\prime,j}\alpha_{a,b,l}\alpha_{a^\prime,b^\prime,l}\right\rangle\left[\sigma_{  j  }^{(a)}\sigma_{  j+1  }^{(b)},\sigma_{  j  }^{( a^\prime )}\sigma_{  j+1  }^{( b^\prime )}\right]\left[\sigma_{  l  }^{(a)}\sigma_{  l+1  }^{(b)},\sigma_{  l  }^{( a^\prime )}\sigma_{  l+1  }^{( b^\prime )}\right]^\dagger\right)\right)^{\frac{1}{2}}
\end{equation}
Assuming that $A_{k}$ and $A_{k^\prime}$ are distinct implies that $(a,b)\neq(a^\prime,b^\prime)$.  By the symmetry of the average the only terms for which
\begin{equation}
	\left\langle\alpha_{a,b,j}\alpha_{a^\prime,b^\prime,j}\alpha_{a,b,l}\alpha_{a^\prime,b^\prime,l}\right\rangle
\end{equation}
is non-zero are therefore those for which $j=l$.  Moreover this non-zero value is precisely
\begin{equation}
  \left\langle\alpha_{a,b,j}\alpha_{a^\prime,b^\prime,j}\alpha_{a,b,j}\alpha_{a^\prime,b^\prime,j}\right\rangle
  =\left\langle\alpha_{a,b,j}^{2}\right\rangle\left\langle\alpha_{a^\prime,b^\prime,j}^{2}\right\rangle
  =\left(\frac{1}{9n}\right)^{2}
\end{equation}
Also, by the definition of the commutator 
\begin{align}
  &\frac{1}{2^{n}}\Tr\left(\left[\sigma_{  j  }^{(a)}\sigma_{  j+1  }^{(b)},\sigma_{  j  }^{( a^\prime )}\sigma_{  j+1  }^{( b^\prime )}\right]\left[\sigma_{  l  }^{(a)}\sigma_{  l+1  }^{(b)},\sigma_{  l  }^{( a^\prime )}\sigma_{  l+1  }^{( b^\prime )}\right]^\dagger\right)\nonumber\\
  &\qquad\qquad\leq \frac{1}{2^{n}}\Tr\left(\sigma_{  j  }^{(a)}\sigma_{  j+1  }^{(b)}\sigma_{  j  }^{( a^\prime )}\sigma_{  j+1  }^{( b^\prime )}\left(\sigma_{  l  }^{(a)}\sigma_{  l+1  }^{(b)}\sigma_{  l  }^{( a^\prime )}\sigma_{  l+1  }^{( b^\prime )}\right)^\dagger\right)\nonumber\\
  &\qquad\qquad\qquad- \frac{1}{2^{n}}\Tr\left(\sigma_{  j  }^{( a^\prime )}\sigma_{  j+1  }^{( b^\prime )}\sigma_{  j  }^{(a)}\sigma_{  j+1  }^{(b)}\left(\sigma_{  l  }^{(a)}\sigma_{  l+1  }^{(b)}\sigma_{  l  }^{( a^\prime )}\sigma_{  l+1  }^{( b^\prime )}\right)^\dagger\right)\nonumber\\
  &\qquad\qquad\qquad\qquad-\frac{1}{2^{n}}\Tr\left(\sigma_{  j  }^{(a)}\sigma_{  j+1  }^{(b)}\sigma_{  j  }^{( a^\prime )}\sigma_{  j+1  }^{( b^\prime )}\left(\sigma_{  l  }^{( a^\prime )}\sigma_{  l+1  }^{( b^\prime )}\sigma_{  l  }^{(a)}\sigma_{  l+1  }^{(b)}\right)^\dagger\right)\nonumber\\
  &\qquad\qquad\qquad\qquad\qquad+\frac{1}{2^{n}}\Tr\left(\sigma_{  j  }^{( a^\prime )}\sigma_{  j+1  }^{( b^\prime )}\sigma_{  j  }^{(a)}\sigma_{  j+1  }^{(b)}\left(\sigma_{  l  }^{( a^\prime )}\sigma_{  l+1  }^{( b^\prime )}\sigma_{  l  }^{(a)}\sigma_{  l+1  }^{(b)}\right)^\dagger\right)\
\end{align}
where each of the four terms on the right hand side of this expression has modulus bounded by $1$ by the properties of the Pauli matrices.  Therefore it is concluded that
\begin{equation}
  \left\langle\|[A_{k},A_{k^\prime}]\|\right\rangle\leq\sqrt{\sum_{ \genfrac{}{}{0pt}{}{j=1}{j\text{ even}} }^{n}4\left(\frac{1}{9n}\right)^{2}}=\frac{2}{9n}\sqrt{\frac{n}{2}}
\end{equation}

\subsubsection{The norm $\left\|[B_{k},B_{k^\prime}]\right\|$}
The calculation of $\left\|[B_{k},B_{k^\prime}]\right\|$ is identical to that of $\left\|[A_{k},A_{k^\prime}]\right\|$, with the sums being over all odd $j$ instead of all even $j$.  Therefore,
\begin{equation}
  \left\langle\|[B_{k},B_{k^\prime}]\|\right\rangle\leq \frac{2}{9n}\sqrt{\frac{n}{2}}
\end{equation}

\subsubsection{The norm $\left\|[A_{k},B_{k^\prime}]\right\|$}
By the definition of $A_{k}$ and $B_{k}$ and the linearity of the commutator
\begin{align}
  [A_{k},B_{k^\prime}]=\sum_{\genfrac{}{}{0pt}{}{j=1}{j\text{ odd}}}^{n}\sum_{\genfrac{}{}{0pt}{}{j^\prime=1}{j^\prime\text{ even}}}^{n}\alpha_{a,b,j}\alpha_{a^\prime,b^\prime,j^\prime}\left[\sigma_{  j  }^{(a)}\sigma_{  j+1  }^{(b)},\sigma_{  j^\prime  }^{( a^\prime )}\sigma_{  j^\prime+1  }^{( b^\prime )}\right]
\end{align}
Only terms in this sum where $j^\prime=j\pm1$ (where the identity of indices $n+1\equiv 1$ and $n\equiv 0$ is made) may possibly have a non-zero commutator, as for the other terms the matrices $\sigma_{  j  }^{(a)}\sigma_{  j+1  }^{(b)}$ and $\sigma_{  j^\prime  }^{( a^\prime )}\sigma_{  j^\prime+1  }^{( b^\prime )}$ act on completely separate qubits and therefore commute.  This enables the sum to be reduced to
\begin{align}
  [A_{k},B_{k^\prime}]&=\sum_{\genfrac{}{}{0pt}{}{j=1}{j\text{ odd}}}^{n}\alpha_{a,b,j}\alpha_{a^\prime,b^\prime,j-1}\left[\sigma_{  j  }^{(a)}\sigma_{  j+1  }^{(b)},\sigma_{  j-1  }^{( a^\prime )}\sigma_{  j  }^{( b^\prime )}\right]\nonumber\\
  &\qquad+\sum_{\genfrac{}{}{0pt}{}{j=1}{j\text{ odd}}}^{n}\alpha_{a,b,j}\alpha_{a^\prime,b^\prime,j+1}\left[\sigma_{  j  }^{(a)}\sigma_{  j+1  }^{(b)},\sigma_{  j+1  }^{( a^\prime )}\sigma_{  j+2  }^{( b^\prime )}\right]
\end{align}
Now proceeding in an analogous way to the calculation of $\left\langle\|[A_{k},A_{k^\prime}]\|\right\rangle$ it is seen that
\begin{equation}
  \left\langle\|[A_{k},B_{k^\prime}]\|\right\rangle\leq\sqrt{2\sum_{ \genfrac{}{}{0pt}{}{j=1}{j\text{ odd}} }^{n}4\left(\frac{1}{9n}\right)^{2}}=\frac{2}{9n}\sqrt{n}
\end{equation}

\subsubsection{Collecting results}
It has been shown that
\begin{align}
  \left|\psi_{n}(t)-\phi_{n}(t)\right|
  &\leq\frac{t^{2}}{2}\sum_{1\leq k< k^\prime\leq18}\left\langle\left\|\left[Q_{k},Q_{k^\prime}\right]\right\|\right\rangle\nonumber\\
  &=\frac{t^{2}}{2}\left(\sum_{1\leq k< k^\prime\leq 9}\left\langle\left\|\left[A_{k},A_{k^\prime}\right]\right\|\right\rangle
  +\sum_{1\leq k< k^\prime\leq 9}\left\langle\left\|\left[B_{k},B_{k^\prime}\right]\right\|\right\rangle
  +\sum_{1\leq k, k^\prime\leq 9}\left\langle\left\|\left[A_{k},B_{k^\prime}\right]\right\|\right\rangle\right)\nonumber\\
  &\leq \frac{t^{2}}{2}\left(36\left(\frac{2}{9n}\sqrt{\frac{n}{2}}\right)+36\left(\frac{2}{9n}\sqrt{\frac{n}{2}}\right)+81\left(\frac{2}{9n}\sqrt{n}\right)\right)\nonumber\\
  &=\frac{t^{2}\left(4\sqrt{2}+9\right)}{\sqrt{n}}
\end{align}
or, by retaining the variance $\sigma_{n}^{2}=\frac{1}{9n}$ of the random variables $\hat{\alpha}_{a,b,j}$
\begin{equation}\label{2.169a}
	\left|\psi_{n}(t)-\phi_{n}(t)\right|\leq t^{2}\sigma_{n}^{2}\sqrt{n}\left(36\sqrt{2}+81\right)
\end{equation}
This concludes the proof as it is recalled that $\hat{\phi}_{n}(t)=\hat{\phi}(t)=\e^{-\frac{t^{2}}{2}}$.
\end{proof}

\subsection{Universality}\label{UniDOS}
The previous proof does not rely on the fact the random variables $\hat{\alpha}_{a,b,j}$ in the random matrix
\begin{equation}
  \hat{H}_{n}=\sum_{j=1}^{n}\sum_{a,b=1}^{3}\hat{\alpha}_{a,b,j}\sigma_{  j  }^{(a)}\sigma_{  j+1  }^{(b)}
\end{equation}
are normally distributed with zero mean and variance of $\frac{1}{9n}$.  In fact in the most general case, for each value of $n$ separately, they can each be taken as independent (and not necessarily identically distributed) random variables with zero mean and finite variances $\mathbb{E}\left(\hat{\alpha}_{a,b,j}^{2}\right)\leq s_{n}^{2}<\infty$ such that there exists some $\gamma\in\mathbb{N}$ for which
\begin{equation}
  \lim_{n\to\infty}\sum_{j=1}^{n}\sum_{a,b=1}^{3}\mathbb{E}\left(\left|\hat{\alpha}_{a,b,j}\right|^{2+\gamma}\right)=0
  \qquad\text{and}\qquad \lim_{n\to\infty}\sum_{j=1}^{n}\sum_{a,b=1}^{3}\mathbb{E}\left(\hat{\alpha}_{a,b,j}^{2}\right)=1
\end{equation}
holds and the distribution of $-\hat{\alpha}_{a,b,j}$ is identical to that of $+\hat{\alpha}_{a,b,j}$.  To this end the following theorem will be proved:

\begin{theorem}[Universal convergence of the characteristic functions associated with the probability measures $\di\hat{\mu}_{n,1}$]
  The characteristic functions associated to the probability measures $\di\hat{\mu}_{n,1}$ for the generalised ensembles $\hat{H}_{n}$ above,
  \begin{equation}
    \hat{\psi}_{n}(t)=\int_{-\infty}^\infty\e^{\im t\lambda}\di\hat{\mu}_{n,1}(\lambda)
  \end{equation}
  satisfy
  \begin{equation}
    \left|\hat{\psi}_{n}(t)-\e^{-\frac{t^{2}}{2}}\right|\to0
  \end{equation}
  for each $t\in\mathbb{R}$ as $n\to\infty$ over all $n\in2\mathbb{N}$ if $s_{n}^{2}$ is such that
  \begin{equation}
    \lim_{n\to\infty}s_{n}^{2}\sqrt{n}=0
  \end{equation}
\end{theorem}

As seen in the last proof, a corollary of this theorem follows directly from the continuity theorem:
\begin{corollary}
  The probability measures $\di\hat{\mu}_{n,1}$ tend weakly to the probability measure of a standard normally distributed random variable, that is, for each $x\in\mathbb{R}$
  \begin{equation}
    \lim_{n\to\infty}\int_{-\infty}^{x}\di\hat{\mu}_{n,1}=\frac{1}{\sqrt{2\pi}}\int_{-\infty}^{x}\e^{-\frac{\lambda^{2}}{2}}\di\lambda
  \end{equation}
\end{corollary}

\begin{proof}
The proof follows that in Section \ref{ThDOS}, and only the key adaptations to it will be detailed here.  First, consider the characteristic function $\hat{\phi}_{n}(t)$ from the last proof, which was defined in equation (\ref{OddEvenPhin}) to be
\begin{equation}
  \hat{\phi}_{n}(t)=\prod_{k=1}^{9}\left\langle\frac{1}{2^{n}}\Tr\left(\e^{\im tA_{k}}\right)\right\rangle\left\langle\frac{1}{2^{n}}\Tr\left(\e^{\im tB_{k}}\right)\right\rangle
\end{equation}
Using only the symmetry of the distributions of the $\hat{\alpha}_{a,b,j}$ it was proved, see equation (\ref{EvenOddPhi}), that
\begin{equation}
  \hat{\phi}_{n}(t)=\prod_{j=1}^{n}\prod_{a,b=1}^{3}\Big\langle\cos(t \alpha_{a,b,j})\Big\rangle
\end{equation}
or equivalently, as $\langle\sin(t \alpha_{a,b,j})\rangle=0$ by the symmetry of the distributions of the $\hat{\alpha}_{a,b,j}$,
\begin{equation}
  \hat{\phi}_{n}(t)=\prod_{j=1}^{n}\prod_{a,b=1}^{3}\Big\langle\cos(t \alpha_{a,b,j})+\im\sin(t \alpha_{a,b,j})\Big\rangle
  =\prod_{j=1}^{n}\prod_{a,b=1}^{3}\Big\langle\e^{\im t\alpha_{a,b,j}}\Big\rangle
\end{equation}

The characteristic function $\hat{\phi}_{n}(t)$ is therefore exactly the characteristic function of the random variable
\begin{equation}
  \hat{\alpha}=\sum_{j=1}^{n}\sum_{a,b=1}^{3}\hat{\alpha}_{a,b,j}
\end{equation}
as this is by definition
\begin{equation}
  \mathbb{E}\left(\e^{\im t\hat{\alpha}}\right)=\mathbb{E}\left(\prod_{j=1}^{n}\prod_{a,b=1}^{3}\e^{\im t\hat{\alpha}_{a,b,j}}\right)\equiv
 \left\langle\prod_{j=1}^{n}\prod_{a,b=1}^{3}\e^{\im t\alpha_{a,b,j}}\right\rangle
\end{equation}
As the $\hat{\alpha}_{a,b,j}$ are independent the average of the product of $\e^{\im t\alpha_{a,b,j}}$ is equal to the product of the averages $\left\langle\e^{\im t\alpha_{a,b,j}}\right\rangle$ and thus $\hat{\phi}_{n}(t)$ is recovered.

Lyapunov's central limit theorem implies that the distribution of $\hat{\alpha}$ tends weakly to that of a standard normal random variable.  The continuity theorem then states that the characteristic function $\hat{\phi}_{n}(t)$ tends, for every fixed $t\in\mathbb{R}$, to the characteristic function of a standard normal random variable $\hat{\phi}(t)=\e^{-\frac{t^{2}}{2}}$.

In the previous proof, equation (\ref{2.169a}), it was shown that the characteristic function $\hat{\psi}_{n}(t)$ and $\hat{\phi}_{n}(t)$ satisfy 
\begin{equation}
  \Big|\hat{\psi}_{n}(t)-\hat{\phi}_{n}(t)\Big|\leq t^{2}s_{n}^{2}\sqrt{n}\left(36\sqrt{2}+81\right)
\end{equation}
By the triangle inequality
\begin{equation}
  \left|\hat{\psi}_{n}(t)-\e^{-\frac{t^{2}}{2}}\right|\leq\Big|\hat{\psi}_{n}(t)-\hat{\phi}_{n}(t)\Big|+\left|\hat{\phi}_{n}(t)-\e^{-\frac{t^{2}}{2}}\right|
\end{equation}
where both terms on the right hand side of this equation tend to zero as $n\to\infty$ for every fixed $t\in\mathbb{R}$, by the assumption on $s_n$, concluding the proof.
\end{proof}
 
\subsection{Inclusion of local terms}\label{momentslocalterms}
The main ensemble considered so far
\begin{equation}
  \hat{H}_{n}=\sum_{j=1}^{n}\sum_{a,b=1}^{3}\hat{\alpha}_{a,b,j}\sigma_{  j  }^{(a)}\sigma_{  j+1  }^{(b)}\qquad\qquad\hat{\alpha}_{a,b,j}\sim\mathcal{N}\left(0,\frac{1}{9n}\right) \iid
  \end{equation}
did not include any `local' terms, see Section \ref{Random matrix model}, proportional to $\sigma_{  j  }^{(a)}$.  The theorems so far in this section still apply to such ensembles containing these `local' terms.  For example, for the ensembles
\begin{align}
  \hat{H}_{n}^{(local)}&=\sum_{j=1}^{n}\sum_{a=1}^{3}\sum_{b=0}^{3}\hat{\alpha}_{a,b,j}\sigma_{  j  }^{(a)}\sigma_{  j+1  }^{(b)}\qquad\qquad\hat{\alpha}_{a,b,j}\sim\mathcal{N}\left(0,\frac{1}{12n}\right) \iid
  \end{align}
for $n\in2\mathbb{N}$ the following theorem still apples:

\begin{theorem}[Convergence of the characteristic functions associated with the probability measures $\di\hat{\mu}_{n,1}^{(local)}$]
   The characteristic functions associated to the probability measures $\di\hat{\mu}_{n,1}^{(local)}$ for the ensembles $\hat{H}_{n}^{(local)}$,
  \begin{equation}
    \hat{\psi}^{(local)}_{n}(t)=\int_{-\infty}^\infty\e^{\im t\lambda}\di\hat{\mu}_{n,1}^{(local)}(\lambda)
  \end{equation}
  satisfy
  \begin{equation}
    \left|\hat{\psi}^{(local)}_{n}(t)-\e^{-\frac{t^{2}}{2}}\right|=O\left(\frac{1}{\sqrt{n}}\right)
  \end{equation}
  for each $t\in\mathbb{R}$ as $n\to\infty$ over all $n\in2\mathbb{N}$.
\end{theorem}

\begin{proof}\hspace{-2mm}\footnote{Adapted from the proof given by the author in \cite{KLW2014_RMT}.}
The proof follows that in Section \ref{ThDOS} again, and only the key adaptations to it will be detailed here.  The splitting of the terms in the sum
\begin{equation}
  H_{n}^{(local)}=\sum_{j=1}^{n}\sum_{a=1}^{3}\sum_{b=0}^{3}\alpha_{a,b,j}\sigma_{  j  }^{(a)}\sigma_{  j+1  }^{(b)}
\end{equation}
is carried out as before with
\begin{align}
  A_{3(b-1)+a}&=\sum_{ \genfrac{}{}{0pt}{}{j=1}{j\text{ even}} }^{n} \alpha_{a,b,j}\sigma_{  j  }^{( a )}\sigma_{  j+1  }^{( b )}
\end{align}
and 
\begin{align}
 B_{3(b-1)+a}&=\sum_{ \genfrac{}{}{0pt}{}{j=1}{j\text{ odd}} }^{n} \alpha_{a,b,j}\sigma_{  j  }^{( a )}\sigma_{  j+1  }^{( b )}
\end{align}
for $a,b=1,2,3$ with the addition of
\begin{align}
  C_{1}=\sum_{j=1}^{n}\alpha_{1,0,j}\sigma_{  j  }^{(1)}\qquad
  C_{2}=\sum_{j=1}^{n}\alpha_{2,0,j}\sigma_{  j  }^{(2)}\qquad
  C_{3}=\sum_{j=1}^{n}\alpha_{3,0,j}\sigma_{  j  }^{(3)}
\end{align}
with $C_{k}=0$ for $k=4,\dots,9$, so that $H_{n}^{(local)}=A+B+C$ with $A=\sum_{k}A_{k}$, $B=\sum_{k}B_{k}$ and $C=\sum_{k}C_{k}$, see Figure \ref{EvenOddLocal}.

\begin{figure}
  \centering
  \input{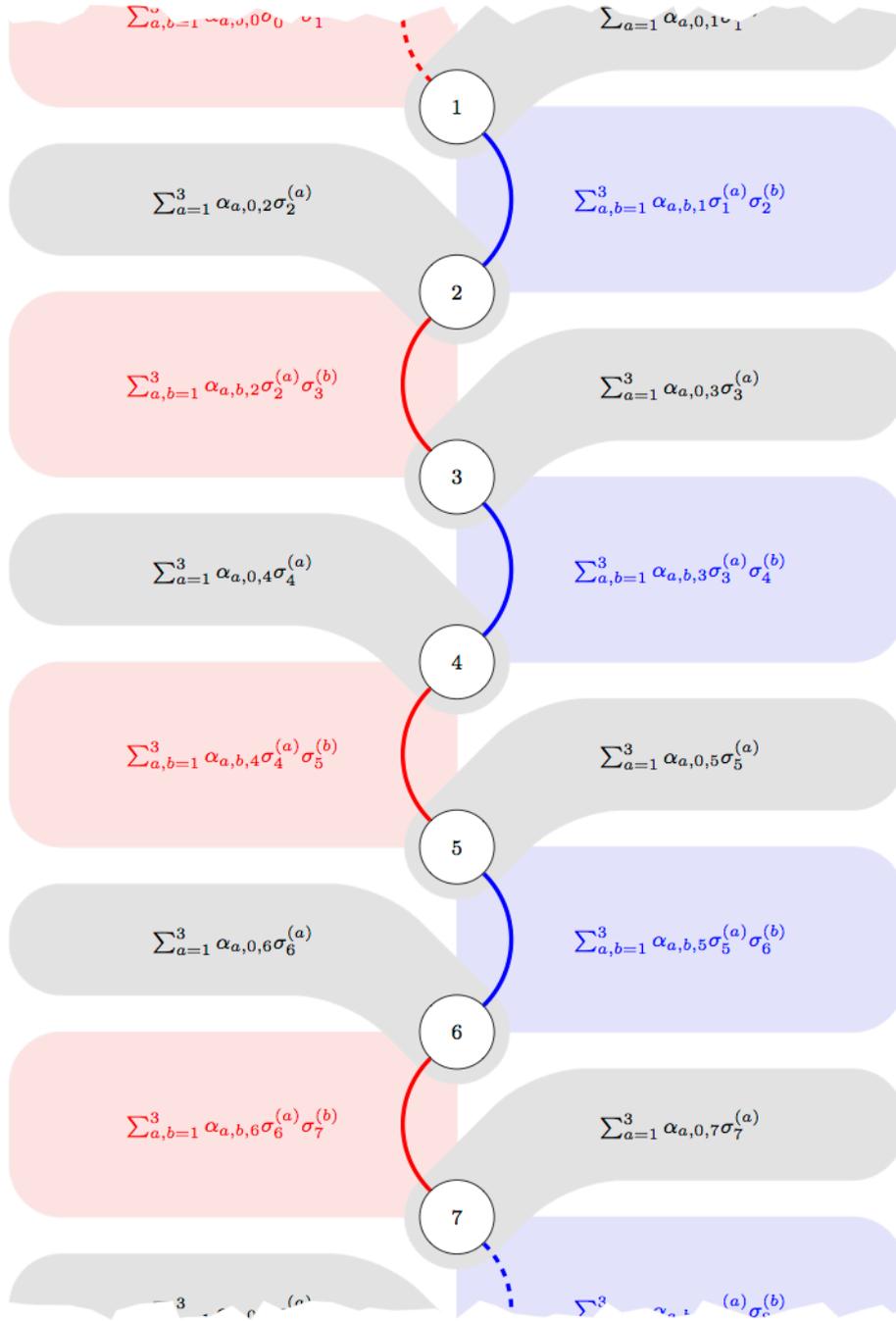}
  \caption[Splitting a chain into odd, even and local terms]{A representation of the splitting of the terms present in the Hamiltonian $H_{n}^{(local)}$.  The circles represent qubits and the links the interaction terms in $H_{n}^{(local)}$.  The matrices $A$, $B$ and $C$ are defined such that $H_{n}^{(local)}=A+B+C$ where $A$ contains all the terms from $H_{n}^{(local)}$ that act between two sites where the lowest site index is even (red terms on the left in the diagram connecting qubits), likewise $B$ where the lowest site index is odd (blue terms on the right in the diagram connecting qubits) and where $C$ contains all the local terms from $H_{n}^{(local)}$ (grey terms in the diagram at each qubit).}
  \label{EvenOddLocal}
\end{figure}

The characteristic function $\hat{\phi}_{n}(t)$, of equation (\ref{OddEvenPhin}), is then replaced with
\begin{equation}
  \hat{\phi}^{(local)}_{n}(t)=\prod_{k=1}^{9}
    \left\langle\frac{1}{2^{n}}\Tr\left(\e^{\im tA_{k}}\right)\right\rangle
    \left\langle\frac{1}{2^{n}}\Tr\left(\e^{\im tB_{k}}\right)\right\rangle
    \left\langle\frac{1}{2^{n}}\Tr\left(\e^{\im tC_{k}}\right)\right\rangle
\end{equation}
This is seen, in an analogous way to that in Section \ref{ThDOS}, to be equal to $\hat{\phi}_{n}^{(local)}(t)=\hat{\phi}^{(local)}(t)=\e^{-\frac{t^{2}}{2}}$ and equivalent to
\begin{equation}
  \left\langle\frac{1}{2^{n}}\Tr\left(\prod_{k=1}^{9}\e^{\im tA_{k}}\e^{\im tB_{k}}\e^{\im tC_{k}}\right)\right\rangle
\end{equation}
Finally, again in an analogous way to that in Section \ref{ThDOS}, it is seen that $\left|\hat{\psi}^{(local)}_{n}(t)-\hat{\phi}_{n}^{(local)}(t)\right|$, or equivalently
\begin{align}
  \left|\left\langle\frac{1}{2^{n}}\Tr\left(\e^{\im tH^{(local)}_{n}}\right)\right\rangle-\left\langle\frac{1}{2^{n}}\Tr\left(\prod_{k=1}^{9}\e^{\im tA_{k}}\e^{\im tB_{k}}\e^{\im tC_{k}}\right)\right\rangle\right|
\end{align}
is $O\left(\frac{1}{\sqrt{n}}\right)$ as $n\to\infty$ for each fixed $t\in\mathbb{R}$, concluding the proof.
\end{proof}

\subsection{Exotic geometries and \texorpdfstring{$c$}{c}-colourable graphs}\label{Ccolour}
The results so far in this section are not only restricted to rings of qubits.  For example if the qubits are arranged on a two dimension $n_{1}\times n_{2}$ lattice with (cyclic) position coordinates $(j,j^\prime)$, then the ensemble $\hat{H}_{n}$ may be generalised to
\begin{align}
  \hat{H}^{(lattice)}_{n}&=\sum_{j=1}^{n_{1}}\sum_{j^\prime=1}^{n_{2}}\sum_{a,b=1}^{3}\left(\hat{\alpha}_{a,b,j,j^\prime}\sigma_{  (j,j^\prime)  }^{(a)}\sigma_{  (j+1,j^\prime)  }^{(b)}
    +\hat{\beta}_{a,b,j,j^\prime}\sigma_{ (j,j^\prime)  }^{(a)}\sigma_{  (j,j^\prime+1) }^{(b)}\right)\nonumber\\
    \hat{\alpha}_{a,b,j,j^\prime},\hat{\beta}_{a,b,j,j^\prime}&\sim\mathcal{N}\left(0,\frac{1}{18n}\right) \iid
\end{align}
with $n_{1},n_{2}\in2\mathbb{N}$ and $n_{1}+n_{2}=n$, so that the following theorem holds:

\begin{theorem}[Convergence of the characteristic functions associated with the probability measures $\di\hat{\mu}_{n,1}^{(lattice)}$]
  The characteristic functions associated to the probability measures $\di\hat{\mu}_{n,1}^{(lattice)}$ for the ensembles $\hat{H}_{n}^{(lattice)}$,
  \begin{equation}
    \hat{\psi}^{(lattice)}_{n}(t)=\int_{-\infty}^\infty\e^{\im t\lambda}\di\hat{\mu}_{n,1}^{(lattice)}(\lambda)
  \end{equation}
  satisfy
  \begin{equation}
    \left|\hat{\psi}^{(lattice)}_{n}(t)-\e^{-\frac{t^{2}}{2}}\right|=O\left(\frac{1}{\sqrt{n}}\right)
  \end{equation}
  for each $t\in\mathbb{R}$ as $n\to\infty$ over any $n_{1},n_{2}\in2\mathbb{N}$ with $n=n_{1}+n_{2}$.
\end{theorem}

\begin{proof}\hspace{-2mm}\footnote{Adapted from the proof given by the author in \cite{KLW2014_RMT}.}
The proof follows that in Section \ref{ThDOS} again and only the key adaptations to it will be detailed here.  The splitting of the terms in the sum
\begin{equation}
  H_{n}^{(lattice)}=\sum_{j=1}^{n_{1}}\sum_{j^\prime=1}^{n_{2}}\sum_{a,b=1}^{3}\left(\alpha_{a,b,j,j^\prime}\sigma_{  (j,j^\prime)  }^{(a)}\sigma_{  (j+1,j^\prime)  }^{(b)}
    +\beta_{a,b,j,j^\prime}\sigma_{  (j,j^\prime)  }^{(a)}\sigma_{  (j,j^\prime+1)  }^{(b)}\right)
\end{equation}
is carried out similarly to before with
\begin{align}
  A_{3(b-1)+a}&=\sum_{ \genfrac{}{}{0pt}{}{j=1}{j\text{ even}} }^{n_{1}}\sum_{j^\prime=1}^{n_{2}}
		 \alpha_{a,b,j,j^\prime}\sigma_{ (j,j^\prime)  }^{( a )}\sigma_{  (j+1,j^\prime) }^{( b )}\nonumber\\
  B_{3(b-1)+a}&=\sum_{ \genfrac{}{}{0pt}{}{j=1}{j\text{ odd}} }^{n_{1}}\sum_{j^\prime=1}^{n_{2}}
		 \alpha_{a,b,j,j^\prime}\sigma_{  (j,j^\prime)  }^{( a )}\sigma_{  (j+1,j^\prime)  }^{( b )}\nonumber\\
  C_{3(b-1)+a}&=\sum_{j=1}^{n_{1}}\sum_{ \genfrac{}{}{0pt}{}{j^\prime=1}{j^\prime\text{ even}} }^{n_{2}}
		 \beta_{a,b,j,j^\prime}\sigma_{  (j,j^\prime)  }^{( a )}\sigma_{  (j,j^\prime+1)  }^{( b )}\nonumber\\
  D_{3(b-1)+a}&=\sum_{j=1}^{n_{1}}\sum_{ \genfrac{}{}{0pt}{}{j^\prime=1}{j^\prime\text{ odd}} }^{n_{2}}
		 \beta_{a,b,j,j^\prime}\sigma_{ (j,j^\prime)  }^{( a )}\sigma_{  (j,j^\prime+1)  }^{( b )}
\end{align}
so that $H_{n}^{(latticce)}=A+B+C+D$ with $A=\sum_{k}A_{k}$, $B=\sum_{k}B_{k}$, $C=\sum_{k}C_{k}$ and $D=\sum_{k}C_{k}$, see Figure \ref{EvenOddLattice}.

\begin{figure}
  \centering
  \input{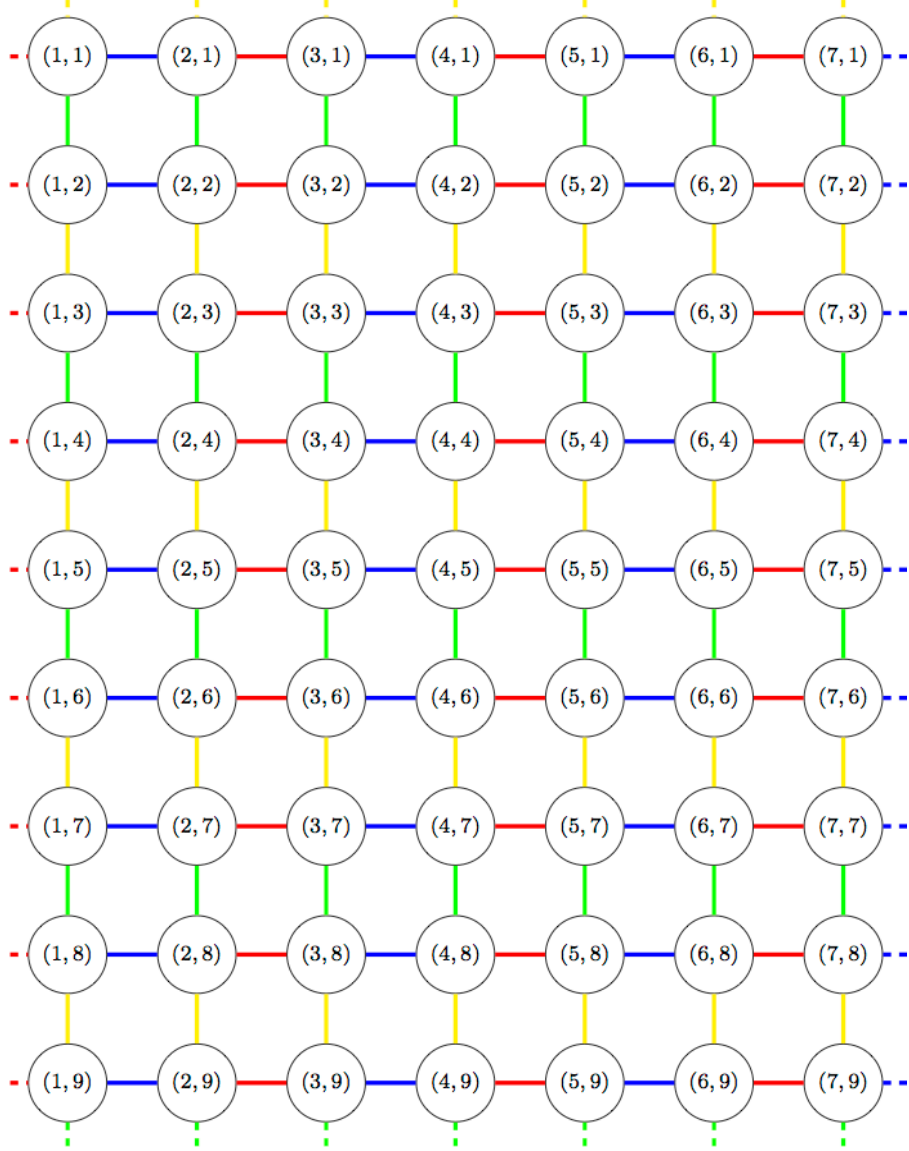}
  \caption[Splitting a lattice into odd and even terms]{A representation of the splitting of the terms present in the Hamiltonian $H_{n}^{(lattice)}$.  The circles represent qubits and the links the interaction terms in $H_{n}^{(lattice)}$.  Matrices $A$, $B$, $C$ and $D$ are defined such that $H_{n}^{(lattice)}=A+B+C+D$ where $A$ contains all the terms from $H_{n}^{(lattice)}$ that act between two sites (horizontally) where the lowest (horizontal) site index is even (red links in the diagram), likewise $B$ where the lowest (horizontal) site index is odd (blue links in the diagram), $C$ (yellow links in the diagram) and $D$ (green links in the diagram) are the vertical counterparts to $A$ and $B$ respectively.}
  \label{EvenOddLattice}
\end{figure}

The characteristic function $\hat{\phi}_{n}(t)$, of equation (\ref{OddEvenPhin}), is then replaced with
\begin{equation}
  \hat{\phi}^{(latice)}_{n}(t)=\prod_{k=1}^{9}
    \left\langle\frac{1}{2^{n}}\Tr\left(\e^{\im tA_{k}}\right)\right\rangle
    \left\langle\frac{1}{2^{n}}\Tr\left(\e^{\im tB_{k}}\right)\right\rangle
    \left\langle\frac{1}{2^{n}}\Tr\left(\e^{\im tC_{k}}\right)\right\rangle
    \left\langle\frac{1}{2^{n}}\Tr\left(\e^{\im tD_{k}}\right)\right\rangle
\end{equation}
This is seen, in an analogous way to that in Section \ref{ThDOS}, to be equal to $\hat{\phi}_{n}^{(lattice)}(t)=\hat{\phi}^{(lattice)}(t)=\e^{-\frac{t^{2}}{2}}$ and equivalent to
\begin{equation}
  \left\langle\frac{1}{2^{n}}\Tr\left(\prod_{k=1}^{9}\e^{\im tA_{k}}\e^{\im tB_{k}}\e^{\im tC_{k}}\e^{\im tD_{k}}\right)\right\rangle
\end{equation}
Finally, again in an analogous way to that in Section \ref{ThDOS}, it is seen that $\left|\hat{\psi}^{(lattice)}_{n}(t)-\hat{\phi}_{n}^{(lattice)}(t)\right|$, or equivalently
\begin{align}
  \left|\left\langle\frac{1}{2^{n}}\Tr\left(\e^{\im tH_{n}^{(lattice)}}\right)\right\rangle-\left\langle\frac{1}{2^{n}}\Tr\left(\prod_{k=1}^{9}\e^{\im tA_{k}}\e^{\im tB_{k}}\e^{\im tC_{k}}\e^{\im tD_{k}}\right)\right\rangle\right|
\end{align}
is $O\left(\frac{1}{\sqrt{n}}\right)$ as $n\to\infty$ for each $t\in\mathbb{R}$, concluding the proof.
\end{proof}

\subsubsection{C-colourable graphs}
The generalisation to a lattice is just one example in a class of more general geometric extensions.  Given any sequence of simple graphs $\Gamma_{n}$ each with $n=2,3,\dots$ vertices labelled $1$ to $n$ and $e_{n}\geq1$ edges, the ensemble
\begin{equation}
  \hat{H}^{(\Gamma_{n})}_{n}=\sum_{(j,j^\prime)\in\Gamma_{n}}\sum_{a,b=1}^{3}\hat{\alpha}_{a,b,j,j^\prime}\sigma_{  j  }^{(a)}\sigma_{  j^\prime  }^{(b)}\qquad\qquad\hat{\alpha}_{a,b,j,j^\prime}\sim\mathcal{N}\left(0,\frac{1}{9e_{n}}\right) \iid
\end{equation}
can be defined, where the sum over $(j,j^\prime)\in\Gamma_{n}$ denotes the sum over all indices $j<j^\prime$ such that the vertices labelled $j$ and $j^\prime$ of the graph $\Gamma_{n}$ are connected by an edge of the graph $\Gamma_{n}$.

The graph $\Gamma_{n}$ is c-colourable if and only if there exist $c\in\mathbb{N}$ colours such that the edges of $\Gamma_{n}$ may be coloured in these colours so that no vertex has more that one edge connected to it of any one colour.

For example, the graphs $\Gamma_{n}$ underling the lattice in the last ensemble were all $4$-colourable, see Figure \ref{EvenOddLattice}, and the graphs $\Gamma_{n}$ underling the chain in the first model were all $2$-colourable, see Figure \ref{EvenOdd}.  A particular colouring of a graph $\Gamma_{n}$, as with the lattice and chain, can therefore be used to defined a specific splitting of the terms in the matrix
\begin{equation}
  H_{n}^{(\Gamma_{n})}=\sum_{(j,j^\prime)\in\Gamma_{n}}\sum_{a,b=1}^{3}\alpha_{a,b,j,j^\prime}\sigma_{  j  }^{(a)}\sigma_{  j^\prime  }^{(b)}
\end{equation}
that is, separating the terms by their corresponding colour.

Therefore, with this separation and again by an analogous argument to that in Section \ref{ThDOS}, if there exists a c-colouring of each graph $\Gamma_{n}$, where $c\in\mathbb{N}$ is fixed independently of $n$, then the spectral probability measures, for the ensembles $\hat{H}_{n}^{(\Gamma_{n})}$, tend weakly to that of a standard normal random variable.

In particular, this extends the results of Section \ref{ThDOS} to the full sequence of ensembles $\hat{H}_{n}$ for $n=2,3,\dots$, as the underlying graphs are all $3$-colourable.  The addition of local terms to $\hat{H}_{n}^{(\Gamma_{n})}$ again also holds no difficultly.

\section{Limiting spectral density for sequences of fixed Hamiltonians}\label{Splitting method}
In this section a central limit theorem for the limiting spectral density for a sequence of fixed Hamiltonians $H_{n}^{(local)}$ from the ensembles $\hat{H}_{n}^{(local)}$ for $n=2,3,\dots$ will be proved.  This shows that the results of Section \ref{Separating odd and even sites} for the ensembles $\hat{H}_{n}^{(local)}$, hold on the level of individual ensemble members.

The spectral probability measure $\di\mu_{n,1}^{(local)}$ has been defined in Section \ref{PointCorrFunc} for the fixed matrix
\begin{equation}\label{nnEnsemble2}
  H_{n}^{(local)}=\sum_{j=1}^{n}\sum_{a=1}^{3}\sum_{b=0}^{3}\alpha_{a,b,j}\sigma_{j}^{(a)}\sigma_{j+1}^{(b)}
\end{equation}
for some fixed real constants $\alpha_{a,b,j}\in\mathbb{R}$ for each value of $n$ separately, to be
\begin{equation}
  \di\mu_{n,1}^{(local)}(\lambda)=\frac{1}{2^{n}}\sum_{k=1}^{2^{n}}\delta(\lambda-\lambda_{k})\di\lambda
\end{equation}
for the eigenvalues $\lambda_{1},\dots,\lambda_{2^{n}}$ of $H_{n}^{(local)}$.

The ideas in this section were motivated by the work of Hartmann, Mahler and Hess in \cite{Mahler}, where it was analytically shown that the energy distribution of almost every product state, over the eigenstates of some fixed spin chain Hamiltonians, becomes (weakly) normally distributed in the large chain limit.  Here a fixed spin chain Hamiltonian was split into several shorter non-interacting sections, by removing links in the chain, so that the energy distribution of a product state could be determined with the aid of the central limit theorem applied to a sum over these shorter sections, see Section \ref{Mahler}.

Exactly this type of `splitting' will be used here to prove the following theorem:

\subsection{Limiting spectral density for fixed Hamiltonians \texorpdfstring{$H_{n}^{(local)}$}{Hn(local)}}\label{DOSmembersSection}
\begin{theorem}[Limiting spectral density for fixed Hamiltonians $H_{n}^{(local)}$]\label{DOSmembers}
	For $n=2,3,\dots$ let $H_{n}^{(local)}$ be fixed spin chain Hamiltonians as defined above.  If
	\begin{equation}\label{range}
		\lim_{n\to\infty}\frac{1}{2^{n}}\Tr\left( {H_{n}^{(local)}}^{2}\right)\equiv\lim_{n\to\infty}\sum_{j=1}^{n}\sum_{a=1}^{3}\sum_{b=0}^{3}\alpha_{a,b,j}^{2}=1\qquad\text{and}\qquad
		\left|\alpha_{a,b,j}\right|<\frac{C}{\sqrt{n}}
	\end{equation}
	for some positive constant $C\in\mathbb{R}$ independent of $a$, $b$, $j$ and $n$, then the associated spectral probability measures $\di\mu_{n,1}^{(local)}$ tend weakly to that of a standard normal distribution, that is
	\begin{equation}
		\lim_{n\to\infty}\int_{-\infty}^{x^{+}}\di\mu_{n,1}^{(local)}(\lambda)=\frac{1}{\sqrt{2\pi}}\int_{-\infty}^{x}\e^{-\frac{\lambda^{2}}{2}}\di\lambda
	\end{equation}
        for all $x\in\mathbb{R}$ (the notation $x^{+}$ represents the limit as $x$ is approached from above).
\end{theorem}

In fact, throughout the following proof, the constant $C$ can be let to grow slowly with $n$.  For the choice of $l=\left\lfloor \sqrt{n}\right\rfloor$ within the proof, $C$ can be replaced with $Cn^{\frac{1}{4^{2}}}$ for example.  A constant value is taken though for clarity.

\begin{proof}\hspace{-2mm}\footnote{Adapted from the proof given by the author in \cite{KLW2014_FIXED}.}
The proof will proceed by considering the characteristic function $\psi_{n}^{(local)}(t)$ associated to the probability measure $\di\mu_{n,1}^{(local)}$.  That is, as defined in Section \ref{PointCorrFunc},
\begin{equation}
  \psi_{n}^{(local)}(t)=\int_{-\infty}^\infty\e^{\im t\lambda}\di\mu_{n,1}^{(local)}(\lambda)
\end{equation}
This will be shown to converge point-wise in $t$ to a characteristic function $\phi_{n}^{(local)}(t)$ that factorises into a large number of factors.  A probability measure associated to $\phi_{n}^{(local)}(t)$ will have, by construction, the form of a probability measure of a sum of independent random variables.  Lyapunov's central limit theorem will then be applied to show the weak convergence of $\di\mu_{n,1}^{(local)}$ to the probability measure of a standard normal random variable.

\subsubsection{Separation}
First, the Hamiltonian $H_{n}^{(local)}$ will be split into several sections.  Consider grouping the terms in $H_{n}^{(local)}$ into blocks $B_{k}$ acting on at most $l=l(n)\in\mathbb{N}$ qubits nontrivially, so that the terms acting non-trivially on the qubits labelled $1$ to $l$ are grouped together, then likewise on the qubits labelled $l+1$ to $2l$ and so on until the last group acts non-trivially on the qubits labelled $\left(\left\lceil\frac{n}{l}\right\rceil-1\right)l+1$ to $n$.  Explicitly, let
	\begin{align}\label{b}
		B_{k}&=\sum_{j=1}^{l-1}h_{(k-1)l+j}\qquad\text{and}\qquad L_{k}=h_{(k-1)l}
	\end{align}
	for $k=1,\dots,\left\lceil\frac{n}{l}\right\rceil$, where
	\begin{equation}\label{h}
		h_{0}=\sum_{a=1}^{3}\sum_{b=0}^{3}\alpha_{a,b,n}\sigma_{n}^{(a)}\sigma_{1}^{(b)}\qquad\text{and}\qquad
		h_{j}=\sum_{a=1}^{3}\sum_{b=0}^{3}\alpha_{a,b,j}\sigma_{j}^{(a)}\sigma_{j+1}^{(b)}
	\end{equation}
	for $j=1,\dots,n-1$ and $h_{j}=0$ otherwise.  Also let
	\begin{equation}
		B=\sum_{k=1}^{\left\lceil\frac{n}{l}\right\rceil}B_{k}\qquad\text{and}\qquad
		L=\sum_{k=1}^{\left\lceil\frac{n}{l}\right\rceil}L_{k}
	\end{equation}
	so that $H_{n}^{(local)}$ is a sum of the blocks $B$ and links $L$, that is $H_{n}^{(local)}=B+L$.  The `local' terms $\alpha_{a,0,j}\sigma_{j}^{(a)}$ need not be included in $L$ but could be included in $B$ instead to improve the bounds derived below by a constant multiplicative factor.  The definition above has been used for simplicity though.
	
The value $l=l(n)$ is chosen such that
	\begin{equation}\label{lconditions}
		\lim_{n\to\infty}\frac{1}{l}=0\qquad\text{and}\qquad\lim_{n\to\infty}\frac{l}{n}=0
	\end{equation}
	For example $l=\left\lfloor \sqrt{n}\right\rfloor$ would suffice.

Figure \ref{Splitting} graphically represents this splitting.

\begin{figure}
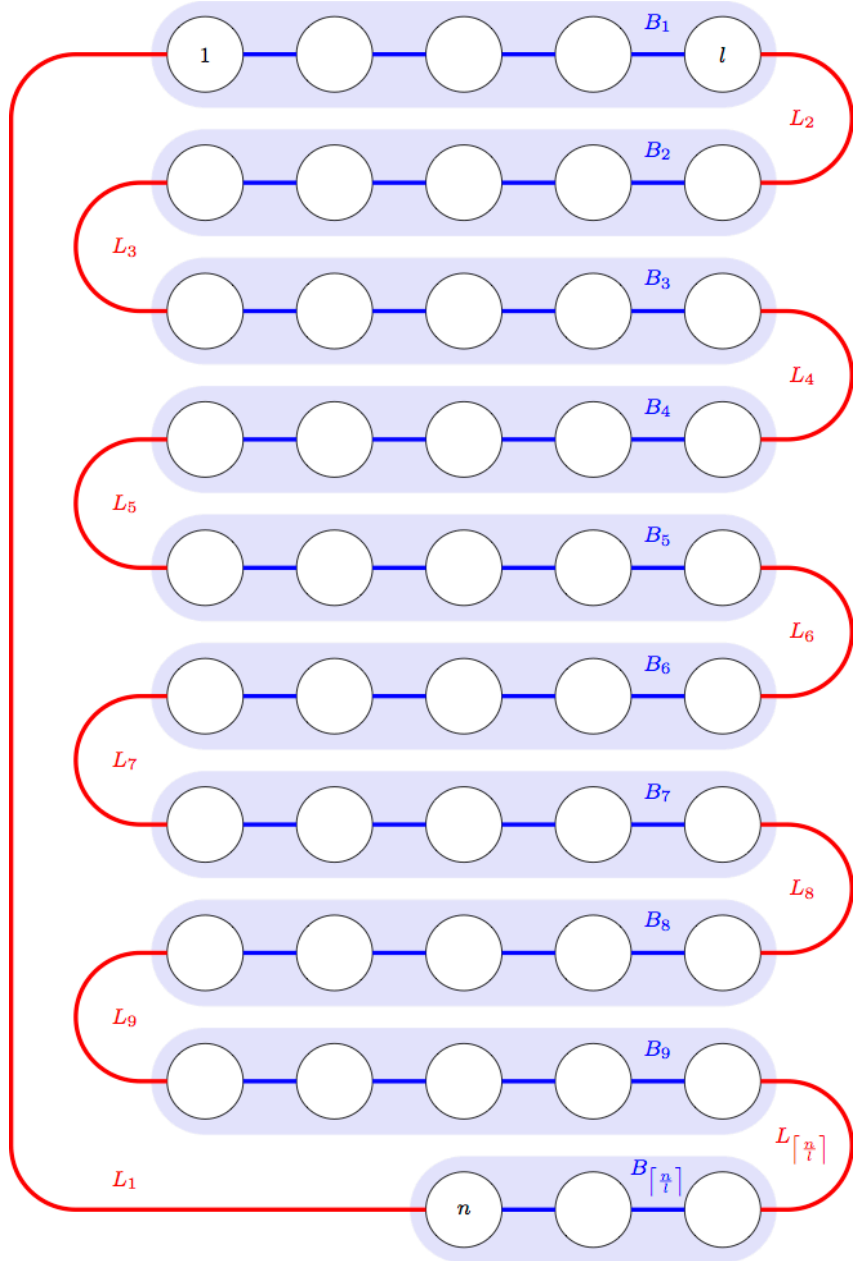

  \centering
  \include{Figures/Splitting}
  \caption[Splitting a chain into $\left\lceil\frac{n}{l}\right\rceil$ portions]{The splitting of a $n$-qubit nearest-neighbour chain into $\left\lceil\frac{n}{l}\right\rceil$ portions for $n=48$ and $l=5$.  The qubits are represented by circles and the nearest-neighbour interactions by links, which correspond to terms $B_{k}$ (shaded blue) and $L_k$ (red) for $k=1,\dots,\left\lceil\frac{n}{l}\right\rceil$ in the chain's Hamiltonian $H_{n}^{(local)}$.  By removing the red links (that is removing the corresponding terms $L_{k}=h_{(k-1)l}$ from $H_{n}^{(local)}$) the chain's Hamiltonian is split into $\left\lceil\frac{n}{l}\right\rceil$ blocks, $B_{k}$, each of which commute with each other as they act on completely separate qubits.}
  \label{Splitting}
\end{figure}

\subsubsection{The characteristic function $\phi_{n}^{(local)}(t)$}
The matrix $B$ has an associated spectral probability measure.  Let its corresponding characteristic function be the characteristic function $\phi_{n}^{(local)}(t)$, that is,
\begin{equation}
  \phi_{n}^{(local)}(t)=\frac{1}{2^{n}}\Tr\left(\e^{\im tB}\right)
\end{equation}

\subsubsection{Rewriting the characteristic function $\phi_{n}^{(local)}(t)$}
As $B$ is defined as the sum of the commuting matrices $B_{k}$, it immediately follows that
\begin{equation}
  \phi_{n}^{(local)}(t)=\frac{1}{2^{n}}\Tr\left(\prod_{k=1}^{\left\lceil\frac{n}{l}\right\rceil}\e^{\im tB_{k}}\right)
\end{equation}
Then by the calculation in Appendix \ref{tensorProd} that
\begin{equation}
  \phi_{n}^{(local)}(t)=\prod_{k=1}^{\left\lceil\frac{n}{l}\right\rceil}\frac{1}{2^{n}}\Tr\left(\e^{\im tB_{k}}\right)
\end{equation}

\subsubsection{Interpretation of the characteristic function $\phi_{n}^{(local)}(t)$}
From Section \ref{PointCorrFunc} it is again seen that the factors
\begin{equation}
  \frac{1}{2^{n}}\Tr\left(\e^{\im tB_{k}}\right)
\end{equation}
are the characteristic functions for the spectral probability measures of the matrices $B_{k}$.  If $\hat{\xi}_{n,k}$ are independent random variables distributed according to these probability measures, for each fixed $n=2,3,\dots$, then the characteristic function of the variable $\hat{\xi}_{n}=\sum_{k}\hat{\xi}_{n,k}$ is defined as
\begin{equation}
 \mathbb{E}\left(\e^{\im t\sum_{k}\hat{\xi}_{n,k}}\right)=\mathbb{E}\left(\prod_{k=1}^{\left\lceil\frac{n}{l}\right\rceil}\e^{\im t\hat{\xi}_{n,k}}\right)
\end{equation}
Since the $\hat{\xi}_{n,k}$ are independent, this is equal to
\begin{equation}
 \prod_{k=1}^{\left\lceil\frac{n}{l}\right\rceil}\mathbb{E}\left(\e^{\im t\hat{\xi}_{n,k}}\right)
\end{equation}
which, by definition of the characteristic function associated to $\hat{\xi}_{n,k}$, is precisely
\begin{equation}
  \prod_{k=1}^{\left\lceil\frac{n}{l}\right\rceil}\frac{1}{2^{n}}\Tr\left(\e^{\im tB_{k}}\right)=\phi_{n}^{(local)}(t)
\end{equation}
That is, $\phi_{n}^{(local)}(t)$ is the characteristic function of a sum of independent random variables.

\subsubsection{Lyapunov's central limit theorem}
To calculate the limiting probability measure of the random variables $\hat{\xi}_{n}=\sum_{k}\hat{\xi}_{n,k}$ as $n\to\infty$, Lyapunov's central limit theorem (Section \ref{Lyapunov}) will be used.

A sufficient Lyapunov condition on the fourth moments of the random variables $\hat{\xi}_{n,k}$ reads
\begin{equation}\label{LC}
  \lim_{n\to\infty}\frac{1}{S_{n}^{4}}\sum_{k=1}^{\left\lceil\frac{n}{l}\right\rceil}\mathbb{E}\left(\left|\hat{\xi}_{n,k}-\mathbb{E}\left(\hat{\xi}_{n,k}\right)\right|^{4}\right)=0
  \qquad\text{where}\qquad
  S_{n}^{2}=\sum_{k=1}^{\left\lceil\frac{n}{l}\right\rceil}\mathbb{E}\left(\hat{\xi}^{2}_{n,k}\right)
\end{equation}
By definition, the $m^{th}$ moment of the random variable $\hat{\xi}_{n,k}$ is given by the coefficient of $\frac{(\im t)^{m}}{m!}$ of the associated characteristic function expanded around $t=0$, so that
\begin{equation}
  \mathbb{E}\left(\hat{\xi}_{n,k}^{m}\right)=\frac{1}{2^{n}}\Tr\left(B_{k}^{m}\right)
\end{equation}
The value of $S_{n}^{2}$ is then given by
\begin{equation}
  S_{n}^{2}=\sum_{k=1}^{\left\lceil\frac{n}{l}\right\rceil}\mathbb{E}\left(\hat{\xi}_{n,k}^{2}\right)=\sum_{k=1}^{\left\lceil\frac{n}{l}\right\rceil}\frac{1}{2^{n}}\Tr \left(B_{k}^{2}\right)
\end{equation}
By the orthogonality of the Pauli matrices $\sigma_{j}^{(a)}\sigma_{j+1}^{(b)}$ under the Hilbert-Schmidt inner-product, $S_{n}^{2}$ is equal to the sum of the squares of all the coefficients $\alpha_{a,b,j}$ present in the definition of the $B_{k}$.  Similarly $\frac{1}{2^{n}}\Tr\left({H_{n}^{(local)}}^{2}\right)$ is equal to the sum of the squares of all the coefficients $\alpha_{a,b,j}$ present in the definition of $H_{n}^{(local)}$.  Therefore
\begin{equation}\label{LCBound1}
  \frac{1}{2^{n}}\Tr\left({H_{n}^{(local)}}^{2}\right)-S_{n}^{2}
  =\sum_{k=1}^{\left\lceil\frac{n}{l}\right\rceil}\sum_{a=1}^{3}\sum_{b=0}^{3}\alpha_{a,b,(k-1)l}^{2}
\end{equation}
(this is the sum of the squares of all the coefficients present in the $L_{k}$).  This expression is positive and by the assumption on the $\alpha_{a,b,j}$ in the theorem, less than
\begin{equation}
  \left\lceil\frac{n}{l}\right\rceil\frac{12C^{2}}{n}
\end{equation}
This bound tends to zero as $n\to\infty$ by the previous conditions on $l$ in (\ref{lconditions}).  As $\frac{1}{2^{n}}\Tr\left({H_{n}^{(local)}}^{2}\right)\to1$ by assumption it is concluded that $S_{n}\to1$ as $n\to\infty$.

As the Pauli matrices $\sigma_{j}^{(a)}\sigma_{j+1}^{(b)}$ in the definition of $H_{n}^{(local)}$ have zero trace it follows that $\mathbb{E}\left(\hat{\xi}_{n,k}\right)=\frac{1}{2^{n}}\Tr\left( B_{k}\right)=0$.  Then as $S_{n}^{2}\to1$ as $n\to\infty$ implies that $S_{n}^{4}\to1$ as $n\to\infty$, the Lyapunov condition (\ref{LC}) reduces to
\begin{equation}
  	\lim_{n\to\infty}\sum_{k=1}^{\left\lceil\frac{n}{l}\right\rceil}\frac{1}{2^{n}}\Tr\left( B_{k}^{4}\right)=0
  \end{equation}
  By the definition of $B_{k}$ in (\ref{b})
  \begin{equation}
  	\frac{1}{2^{n}}\Tr\left( B_{k}^{4}\right)=\sum_{q,p,r,s=1}^{l-1}\frac{1}{2^{n}}\Tr\left(h_{(k-1)l+q}h_{(k-1)l+p}h_{(k-1)l+r}h_{(k-1)l+s}\right)
	\end{equation}
Terms in this sum for which the factors $h_{j}$ do not appear in pairs where their indices differ by at most one must have zero trace, for suitably large values of $l$.  Each term containing a single factor $h_{j}$ (that is unpaired in this way) will necessarily act on some qubit as a non-identity Pauli matrix.  As the Pauli matrices have zero trace, the entire term will then have zero trace.  There are three ways to pair the four factors $h_{j}$ in the summand of the last expression and $l-1$ values that $q$, $p$, $r$ and $s$ can each take in the sum, leading to at most $3(l-1)^{2}3^2$ non-zero terms.  Each of these potentially non-zero terms can be expressed as a sum of at most $12^{4}$ terms by the definition of the $h_{j}$ in equation (\ref{h}), each containing a four fold product of factors of the form $\alpha_{a,b,j}\sigma_{j}^{(a)}\sigma_{j+1}^{(b)}$.  By the properties of the Pauli matrices and the bound on the coefficients $\alpha_{a,b,j}$ assumed in the theorem, these terms are individually bounded by $\left(\frac{C}{\sqrt{n}}\right)^{4}$ so that
\begin{equation}\label{genbound3}
	\left|\sum_{k=1}^{\left\lceil\frac{n}{l}\right\rceil}\frac{1}{2^{n}}\Tr\left( B_{k}^{4}\right)\right|
	\leq \left\lceil\frac{n}{l}\right\rceil3^3(l-1)^{2}{12}^{4}\left(\frac{C}{\sqrt{n}}\right)^{4}
\end{equation}
Bounding $\left\lceil\frac{n}{l}\right\rceil$ by $\frac{n}{l}+1$ and $l-1$ by $l$ produces the worst bound of
\begin{equation}
  3^{7}4^{4}C^{4}\left(\frac{l}{n}+\frac{l^{2}}{n^{2}}\right)
\end{equation}
which tends to zero as $n\to\infty$ by the assumptions on $l$ in (\ref{lconditions}).

Lyapunov's central limit theorem then states that for any $x\in\mathbb{R}$,
\begin{equation}
  \lim_{n\to\infty}\mathbb{P}\left(\hat{\xi}_{n}\leq x\right)=\frac{1}{\sqrt{2\pi}}\int_{-\infty}^{x}\e^{-\frac{\lambda^{2}}{2}}\di\lambda
\end{equation}
that is the distribution of $\hat{\xi}_{n}$ tends weakly to a standard normal distribution. 

\subsubsection{Point-wise convergence of $\phi_{n}^{(local)}(t)$}
The continuity theorem gives that $\phi_{n}^{(local)}(t)$ (the characteristic function associated to $\hat{\xi}_{n}$) converges point-wise to the characteristic function, $\phi^{(local)}(t)$, of a standard normal random variable.  That is, for every $t\in\mathbb{R}$,
\begin{equation}
 \lim_{n\to\infty}\phi_{n}^{(local)}(t)=\phi^{(local)}(t)=\e^{-\frac{t^{2}}{2}}
\end{equation}

In now remains to be seen that for each $t\in\mathbb{R}$, that $\lim_{n\to\infty}\psi_{n}^{(local)}(t)=\phi^{(local)}(t)$.  In this case the continuity theorem will imply that the spectral probability measure $\di\mu_{n,1}^{(local)}$ will converge in distribution to the probability measure of a standard normal random variable.

It will be shown that for each $t\in\mathbb{R}$
\begin{equation}
  \lim_{n\to\infty}\left|\psi_{n}^{(local)}(t)-\phi_{n}^{(local)}(t)\right|=0
\end{equation}
In order to show this, the following integral identity, as used by Marher \cite{Mahler}, will be needed:

\subsubsection{Integral identity}
For any Hermitian matrices $X$ and $Y$, the fundamental theorem of calculus implies that
\begin{align}
  \e^{\im t(X+Y)}-\e^{\im tX}&=\e^{\im ts(X+Y)}\e^{\im t(1-s)X}\Big|_{s=1}-\e^{\im ts(X+Y)}\e^{\im t(1-s)X}\Big|_{s=0} \nonumber\\
  &=\int_{0}^{1}\frac{\partial}{\partial s}\left(\e^{\im ts(X+Y)}\e^{\im t(1-s)X}\right)\di s
\end{align}
The evaluation of the derivative gives that
\begin{align}
  \e^{\im t(X+Y)}-\e^{\im tX}&=\int_{0}^{1}\e^{\im ts(X+Y)}\im t(X+Y)\e^{\im t(1-s)X}-\e^{\im ts(X+Y)}\im tX\e^{\im t(1-s)X}\di s\nonumber\\
  &=\im t\int_{0}^{1}\e^{\im ts(X+Y)}Y\e^{\im t(1-s)X}\di s
\end{align}

This identity now allows the absolute difference between $\phi_{n}^{(local)}(t)$ and $\psi_{n}^{(local)}(t)$ to be bounded explicitly.  Using the definitions of $\psi_{n}^{(local)}(t)$ and $\phi_{n}^{(local)}(t)$ and the identity above,
\begin{align}
  \left|\psi_{n}^{(local)}(t)-\phi_{n}^{(local)}(t)\right|&=\left|\frac{1}{2^{n}}\Tr\left(\e^{\im t(B+L)}-\e^{\im tB}\right)\right|\nonumber\\
  &=\left|\frac{1}{2^{n}}\im t\int_{0}^{1}\Tr\left(\e^{\im ts(B+L)}L\e^{\im t(1-s)B}\right)\di s\right|
\end{align}
where the trace has been taken inside the integration.  The triangle inequality allows the norm to be taken inside integral, leaving the worse bound for $\left|\psi_{n}^{(local)}(t)-\phi_{n}^{(local)}(t)\right|$ of
\begin{equation}
  \frac{|t|}{2^{n}}\int_{0}^{1}\left|\Tr\left(\e^{\im ts(B+L)}L\e^{\im t(1-s)B}\right)\right|\di s
\end{equation}
The positive integrand $\left|\Tr\left(\e^{\im ts(B+L)}L\e^{\im t(1-s)B}\right)\right|$ can then be analysed with the aid if the Cauchy-Schwartz inequality:

\subsubsection{Cauchy-Schwartz inequality}\label{CSI}
As in the last proof, see equation (\ref{2.114}), the Cauchy-Schwartz inequality yields that
\begin{equation}
  \left|\Tr\left(\e^{\im ts(B+L)}L\e^{\im t(1-s)B}\right)\right|\leq\sqrt{2^{n}\Tr\left(\left(\e^{\im ts(B+L)}L\e^{\im t(1-s)B}\right)\left(\e^{\im ts(B+L)}L\e^{\im t(1-s)B}\right)^\dagger\right)}
\end{equation}
where
\begin{align}
  &\Tr\left(\left(\e^{\im ts(B+L)}L\e^{\im t(1-s)B}\right)\left(\e^{\im ts(B+L)}L\e^{\im t(1-s)B}\right)^\dagger\right)\nonumber\\
  &\qquad\qquad=\Tr\left(\e^{\im ts(B+L)}L\e^{\im t(1-s)B}\e^{-\im t(1-s)B}L^\dagger\e^{-\im ts(B+L)}\right)\nonumber\\
  &\qquad\qquad=\Tr\left(LL^\dagger\right)
\end{align}
by the cyclic property of the trace.  Substituting this result into the last bound for $\left|\psi_{n}^{(local)}(t)-\phi_{n}^{(local)}(t)\right|$ gives that
\begin{equation}
  \left|\psi_{n}^{(local)}(t)-\phi_{n}^{(local)}(t)\right|\leq\frac{|t|}{2^{n}}\int_{0}^{1}\sqrt{2^{n}\Tr \left(LL^\dagger\right)}\di s= |t|\sqrt{\frac{1}{2^{n}}\Tr \left(LL^\dagger\right)}
\end{equation}

\subsubsection{Value of the trace}
Next, the value of $\Tr\left(LL^\dagger\right)$ will be calculated.  By definition
\begin{equation}
  L=\sum_{k=1}^{\left\lceil\frac{n}{l}\right\rceil}\sum_{a=1}^{3}\sum_{b=0}^{3}\alpha_{a,b,(k-1)l}\sigma_{  (k-1)l  }^{(a)}\sigma_{  (k-1)l+1  }^{(b)}
\end{equation}
By the orthogonality of the Pauli matrices $\sigma_{j}^{(a)}\sigma_{j+1}^{(b)}$ under the Hilbert-Schmidt inner-product it follows that
\begin{equation}
  \frac{1}{2^{n}}\Tr \left(LL^\dagger\right)=\sum_{k=1}^{\left\lceil\frac{n}{l}\right\rceil}\sum_{a=1}^{3}\sum_{b=0}^{3}\alpha^{2}_{a,b,(k-1)l}
\end{equation}
which, by the assumption on the $\alpha_{a,b,j}$ in the theorem, is bounded by
\begin{equation}
  \left\lceil\frac{n}{l}\right\rceil\frac{12C^{2}}{n}
\end{equation}
Therefore
\begin{equation}\label{2.6.39}
  \left|\psi_{n}^{(local)}(t)-\phi_{n}^{(local)}(t)\right|\leq |t|\sqrt{\left\lceil\frac{n}{l}\right\rceil\frac{12C^{2}}{n}}
\end{equation}
which tends to zero as $n\to\infty$ by the conditions on $l$ in (\ref{lconditions}).

\subsubsection{Continuity theorem}
The proof will now be concluded by making use of the continuity theorem.  It has been seen that $\left|\psi_{n}^{(local)}(t)-\phi_{n}^{(local)}(t)\right|\to0$ and $\left|\phi_{n}^{(local)}(t)-e^{-\frac{t^{2}}{2}}\right|\to0$ for each $t\in\mathbb{R}$ as $n\to\infty$.  Therefore the triangle inequality shows that
\begin{equation}
 \left|\psi_{n}^{(local)}(t)-e^{-\frac{t^{2}}{2}}\right|\leq\Big|\psi_{n}^{(local)}(t)-\phi_{n}^{(local)}(t)\Big|+\left|\phi_{n}^{(local)}(t)-e^{-\frac{t^{2}}{2}}\right|
\end{equation}
so that $\psi_{n}^{(local)}(t)$ tends to $e^{-\frac{t^{2}}{2}}$ for each $t\in\mathbb{R}$ as $n\to\infty$.

The continuity theorem then states that the probability measures $\di\mu_{n,1}^{(local)}$ associated to the characteristic functions $\psi_{n}^{(local)}(t)$ tend, in distribution, to that of a standard normal random variable.  That is, for every $x\in\mathbb{R}$
\begin{equation}
  \lim_{n\to\infty}\int_{-\infty}^{x}\di\mu_{n,1}^{(local)}(\lambda)=\frac{1}{\sqrt{2\pi}}\int_{-\infty}^{x}\e^{-\frac{t^{2}}{2}}\di\lambda
\end{equation}
which concludes the proof.
\end{proof}

\subsection{More general interactions}\label{MGen}
The results so far in this section are not only restricted to a ring of qubits.  For example, if the qubits are arranged on a two dimension lattice, a similar result applies.

For a general case, let $\Gamma_{n}$ be a simple graph with $n$ vertices labelled $1$ to $n$ and let the edge connecting the vertices labelled $j$ and $j^\prime$ with $j<j^\prime$ be labelled by $(j,j^\prime)$.  Following a similar definition to that given in Section \ref{Ccolour} for ensembles, let the fixed Hamiltonians\footnote{The local terms proportional to $\sigma_{j}^{(a)}$ have been dropped here to be consistent with the previous ensemble $\hat{H}_{n}^{(\Gamma_{n})}$, their reinclusion presents no difficulties in the following calculations.} $H_{n}^{(\Gamma_{n})}$ for $n=2,3,\dots$ be defined
\begin{equation}
	H_{n}^{(\Gamma_{n})}=\sum_{(j,j^\prime)\in\Gamma_{n}}\sum_{a,b=1}^{3}\alpha_{a,b,j,j^\prime}\sigma_{j}^{(a)}\sigma_{j^\prime}^{(b)}
\end{equation}
for some fixed real coefficients $\alpha_{a,b,j,j^\prime}$ for each value of $n$ separately.  Furthermore for this sequence, let the graphs $\Gamma_{n}$ be such that each vertex is connected to at most a constant number (independent of $n$) of other vertices.  Also let the graphs $\Gamma_{n}$ be such that $r=r(n)$ `links' or `interactions' of the form
\begin{equation}
	\sum_{a,b=1}^{3}\alpha_{a,b,j,j^\prime}\sigma_{j}^{(a)}\sigma_{j^\prime}^{(b)}
\end{equation}
can be removed from the sum for $H_{n}^{(\Gamma_{n})}$ to leave a sum of $b=b(n)$ groups of terms (blocks), each supported on non-intersecting subsets of the $n$ qubits and each containing at most $q=q(n)$ qubits.  In this case the following theorem holds:

\begin{theorem}[Limiting spectral density for the fixed Hamiltonians $H_{n}^{(\Gamma_{n})}$]
  For $n=2,3,\dots$ let $H_{n}^{(\Gamma_{n})}$ be fixed spin chain Hamiltonians as defined above.  If
	\begin{equation}
		\lim_{n\to\infty}\frac{1}{2^{n}}\Tr\left( {H_{n}^{(\Gamma_{n})}}^{2}\right)\equiv\lim_{n\to\infty}\sum_{(j,j^\prime)\in\Gamma_{n}}\sum_{a,b=1}^{3}\alpha_{a,b,j,j^\prime}^{2}=1\qquad\text{and}\qquad
		\left|\alpha_{a,b,j,j^\prime}\right|<\frac{C}{\sqrt{n}}
	\end{equation}
	for some positive constant $C\in\mathbb{R}$ independent of $a$, $b$, $j$, $j^\prime$ and $n$, and
	\begin{equation}
		\lim_{n\to\infty}\frac{r}{n}=0\qquad\text{and}\qquad\lim_{n\to\infty}\frac{bq^{2}}{n^{2}}=0
	\end{equation}	
	then the associated spectral probability measures $\di\mu_{n,1}^{(\Gamma_{n})}$ tend weakly to that of a standard normal distribution, that is
	\begin{equation}
		\lim_{n\to\infty}\int_{-\infty}^{x^{+}}\di\mu_{n,1}^{(\Gamma_{n})}(\lambda)=\frac{1}{\sqrt{2\pi}}\int_{-\infty}^{x}\e^{-\frac{\lambda^{2}}{2}}\di\lambda
	\end{equation}
  for all $x\in\mathbb{R}$ (the notation $x^{+}$ represents the limit as $x$ is approached from above).
\end{theorem}

\begin{proof}\hspace{-2mm}\footnote{Adapted from the proof given by the author in \cite{KLW2014_FIXED}.}
The proof of this theorem is analogous to that of the previous one, and only the key differences will be highlighted here:
\subsubsection{Separation}
By construction, the graph $\Gamma_{n}$ allows $H_{n}^{(\Gamma_{n})}$ to be split into a sum of $b$ blocks of terms $B_{k}$ acting on non-intersecting subsets of the $n$ qubits, by removing $r$ links $L_{k}$.  The Hamiltonian may then be written in the form
\begin{equation}
	H_{n}^{(\Gamma_{n})}=B+L,\qquad\qquad B=\sum_{k=1}^{b}B_{k},\qquad\qquad L=\sum_{k=1}^{r}L_{k}
\end{equation}

\subsubsection{Characteristic functions}
The characteristic function corresponding to the spectral density for $H_{n}^{(\Gamma_{n})}$ is again
\begin{equation}
  \psi_{n}^{(\Gamma_{n})}(t)=\frac{1}{2^{n}}\Tr\left(\e^{\im tH_{n}^{\left(\Gamma_{n}\right)}}\right)
\end{equation}
This will be shown to converge point-wise in $t$ to the characteristic function
\begin{equation}
  \phi_{n}^{(\Gamma_{n})}(t)=\frac{1}{2^{n}}\Tr\left(\e^{\im tB}\right)
\end{equation}
which in turn will be shown to converge point-wise in $t$ to the characteristic function of a standard normal random variable:

\subsubsection{Lyapunov's central limit theorem}
Lyapunov's central limit theorem  may, as before, be used show the point-wise convergence, in $t$, of $\phi_{n}^{(\Gamma_{n})}(t)$ to the characteristic function of a standard normal random variable.

The Lyapunov condition analogous to (\ref{LC}), where $\left\lceil\frac{n}{l}\right\rceil$ is replaced by $b$, will again be used.  Similarly to equation (\ref{LCBound1}) before,
\begin{equation}
	\left|\frac{1}{2^{n}}\Tr\left({H_{n}^{(\Gamma_{n})}}^{2}\right)-S_{n}^{2}\right|\leq r\frac{9C^{2}}{n}
\end{equation}
Here $r$ is the number of links removed, where each link contains $9$ terms (previously $12$ as some local terms were also included in the links), each bounded by $\frac{C^{2}}{n}$ by assumption.  By the assumptions in the theorem, this quantity then tends to zero as $n\to\infty$ and (as before) it is concluded that $S_{n}\to1$ as $n\to\infty$.

Again as $\frac{1}{2^{n}}\Tr\left(B_{k}\right)=0$, the Lyapunov condition is then equivalent to
\begin{equation}
	\lim_{n\to\infty}\sum_{k=1}^{b}\frac{1}{2^{n}}\Tr\left(B_{k}^{4}\right)=0
\end{equation}
As each qubit was only allowed to interact with at most a fixed number (independent of $n$) other qubits, an analogous argument as in the previous proof holds to bound the, necessarily positive, value $\sum_{k=1}^{b}\frac{1}{2^{n}}\Tr\left(B_{k}^{4}\right)$ by a bound proportional to
\begin{equation}
	bq^{2}\left(\frac{C}{\sqrt{n}}\right)^{4}
\end{equation}
The assumptions in the theorem then assure that this tends to zero as $n\to\infty$.  Lyapunov's central limit theorem then gives that $\left|\phi_{n}(t)-\e^{-\frac{t^{2}}{2}}\right|\to0$ for each $t\in\mathbb{R}$ as $n\to\infty$.

\subsubsection{Convergence of $\phi_{n}^{(\Gamma_{n})}$ and $\psi_{n}^{(\Gamma_{n})}$}
Following the structure of the last proof, it now only remains to be seen that
\begin{equation}
  \Big|\phi_{n}^{(\Gamma_{n})}(t)-\psi_{n}^{(\Gamma_{n})}(t)\Big|\to0
\end{equation}
for each $t\in\mathbb{R}$ as $n\to\infty$.  As in the last proof (see equation (\ref{2.6.39})), this quantity is bounded by $|t|$ multiplied by the square root of the number of coefficients $\alpha_{a,b,j,j^\prime}$ present in $L$ multiplied by the maximum values of their squares.  Precisely, this gives the bound
\begin{equation}
  |t|\sqrt{r\frac{9C^{2}}{n}}
\end{equation}
which, by the assumptions the theorem, tends to zero as $n\to\infty$, completing the proof.
\end{proof}

\subsection{Example: A two  dimensional lattice}
As an example of the previous generalisation, consider a two dimensional square lattice of $n=p\times p$ qubits.  Let interactions occur only between neighbouring horizontal or vertical qubits and let the boundary be cyclic.  The qubits may be grouped into $l\times l$ blocks as seen in Figure \ref{SplitLattice}.  Smaller blocks will be present at the boundary of the lattice if $l$ does not divide $p$, as in the case of the chain previously.  There will be precisely $b=\left\lceil\frac{p}{l}\right\rceil^{2}$ blocks, each containing at most $q=l^{2}$ qubits.  There are $r=2p\left\lceil\frac{p}{l}\right\rceil$ `links' connecting these blocks.

The necessary conditions on $r$, $b$ and $q$ then read
\begin{equation}
	\frac{r}{n}=\frac{2p\left\lceil\frac{p}{l}\right\rceil}{p^{2}}\to0\qquad\qquad\text{and}\qquad\qquad\frac{bq^{2}}{n^{2}}=\frac{\left\lceil\frac{p}{l}\right\rceil^{2}l^{4}}{p^{4}}\to0
\end{equation}
as $n\to\infty$.  These conditions are satisfied if $l=l(p)$ is chosen such that $\frac{1}{l}\to0$ and $\frac{l}{p}\to0$ as $p^{2}=n\to\infty$.

Note that these conditions on $l$ are equivalent (with $n$ replaced by $p$) to those used in the original proof for $H_{n}^{(local)}$, the example here is just the generalisation of the chain to a lattice.

\begin{figure}
  \centering
  \input{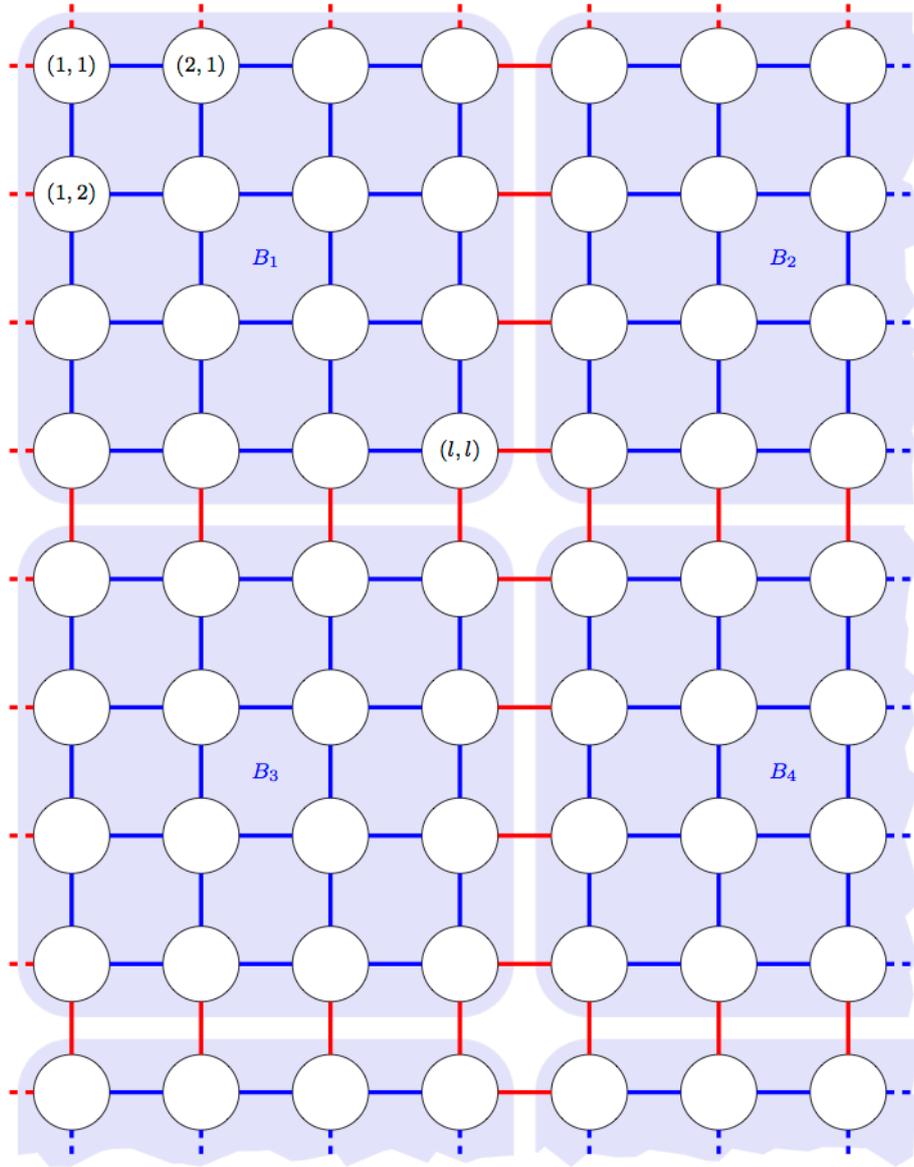}
  \caption[Splitting a lattice into blocks]{The splitting of the two dimensional $p\times p$ qubit nearest-neighbour lattice into $\left\lceil\frac{p}{l}\right\rceil$ portions for $l=4$.  The qubits are represented by circles and the nearest-neighbour interactions by links, which correspond to terms within the blocks (shaded blue) and terms linking the blocks (red) in the chain's Hamiltonian.  By removing the $r=2p\left\lceil\frac{p}{l}\right\rceil$ red links (that is removing the corresponding terms from the Hamiltonian) the chain's Hamiltonian is split into $b=\left\lceil\frac{p}{l}\right\rceil^{2}$ components, $B_{k}$, each of which commute with each other as they act on completely separate groups of $q=l^2$ qubits.}
  \label{SplitLattice}
\end{figure}

\subsection{Interacting qudits}\label{qudits}
Qudits (of fixed dimension $d$, independent of $n$) may also be used in place of qubits.  Given a matrix basis for each qudit site which is orthogonal under the Hilbert-Schmidt inner-product, an analogous proof to that of the previous theorems again holds.  In effect this involves just increasing the range of the indices $a$ and $b$ by some fixed amount in the proof, this only leads to a change of the constants in the bounds given.

The generalisation to qudits was not (easily) possible for the ensemble results of Section \ref{Separating odd and even sites} (limiting ensemble spectral density).  The techniques used there relied on specific properties of the $2\times2$ Pauli matrix basis for qubits.  It is now seen that a central limit theorem holds for the spectral density of many members of many of the ensembles considered therein.

\clearpage
     
\chapter{Eigenstate entanglement}\label{Eigenstate entanglement}
In this chapter the entanglement, of the eigenstates of fixed Hamiltonians from the ensembles defined in Section \ref{Numerics}, between two blocks of qubits will be studied.  Explicitly, a $n$-qubit nearest-neighbour chain will be split into two systems, $A$ and $B$, system $A$ containing a fixed number of qubits, independent of $n$, in a continuous block and system $B$ the remaining qubits, see Figure \ref{ABSystems}.  The entanglement between systems $A$ and $B$ will then be analysed for the eigenstates of some fixed nearest-neighbour qubit Hamiltonians.

\begin{figure}
  \centering
  \input{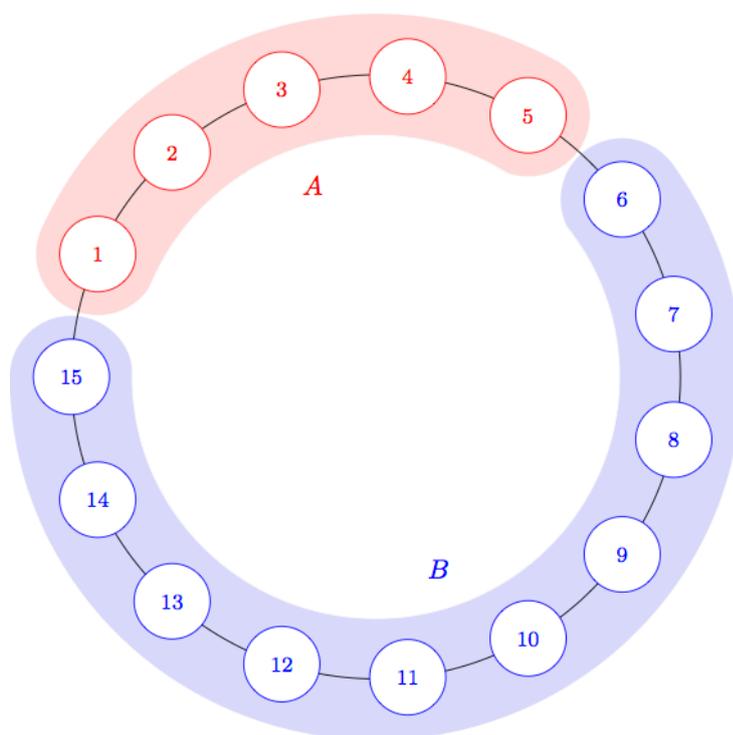}
  \caption[Subsystems $A$ and $B$ of a nearest-neighbour qubit chain]{A nearest-neighbour qubit chain of $n=15$ qubits.  The circles represent the qubits labelled $1$ to $15$ and the links the nearest-neighbour interactions.  The system is split into two subsystems; subsystem $A$ comprised of the $l=5$ qubits labelled $1$ to $5$ (red shading) and subsystem $B$ the remaining $10$ qubits (blue shading).}
  \label{ABSystems}
\end{figure}

First, in Section \ref{EigenDef} the purity of a state and its relation to eigenstate entanglement will be defined, this will become the main tool for analysing the entanglement present in the eigenstates of the fixed Hamiltonians considered.  Also the translation matrix will be defined to formally characterise the translational symmetry of each Hamiltonian from the ensembles $\hat{H}_{n}^{(inv)}$, $\hat{H}_{n}^{(inv, local)}$ and $\hat{H}_{n}^{(Heis)}$.

The results of numerical simulations of the purity of the reduced eigenstates on a fixed block of qubits for the ensembles of Section \ref{Numerics} will then be given in Section \ref{eigenNum}.

Two features of these numerical results will then be explained on the level of individual Hamiltonians in Section \ref{SingleEnt} and \ref{EntLblock}.   First the entanglement between a single qubit and the rest of the chain in each of the eigenstates of fixed generic nearest-neighbour qubit Hamiltonian, without the presence of local terms, which has a non-degenerate spectrum will be given by Theorem \ref{eigenNonInv}.  Theorem \ref{invEnt} then gives a bound for the average purity over a complete set of joint eigenstates of the translation matrix and a fixed generic nearest-neighbour translationally-invariant qubit Hamiltonian.  Extensions to more general chains are made where appropriate.

The work presented in this chapter is based on that given by the author in \cite{KLW2014_FIXED}.  Ideas initiated from discussions with \cite{Keating} and \cite{Linden} are highlighted where appropriate.
\section{Definitions}\label{EigenDef}
Two key components, the purity of a state and the translation matrix, needed for the analysis of the entanglement present in the eigenstates of relevant spin chain Hamiltonians will now be defined.

\subsection{Purity}\label{PuritySec}
To study the entanglement in the eigenstates of a generic qubit chain Hamiltonian $H_{n}^{(local)}$, that is a specific member of the ensemble $\hat{H}_{n}^{(local)}$, the purity of the reduced eigenstates on system $A$ (the qubits labelled $1$ to $l$) will be used.  As $H_{n}^{(local)}$ is Hermitian it has a (not necessarily unique) eigenstate decomposition of
\begin{equation}
  H_{n}^{(local)}=\sum_{k=1}^{2^{n}}\lambda_{k}|\psi_{k}\rangle\langle\psi_{k}|
\end{equation}
where $\lambda_{1}\leq\lambda_{2}\leq\dots\leq\lambda_{2^{n}}$ are the eigenvalues of $H_{n}^{(local)}$ and $|\psi_{k}\rangle$ are corresponding normalised eigenstates.  Let the reduced density matrix on subsystem $A$, of the $k^{th}$ eigenstate $|\psi_{k}\rangle$, be
\begin{equation}
  \rho_{l,k}=\Tr_\mathcal{B}\left(|\psi_{k}\rangle\langle\psi_{k}|\right)
\end{equation}
where $\Tr_\mathcal{B}\left(\cdot\right)$ is the partial trace over the Hilbert space of system $B$ (the qubits labelled $l+1$ to $n$) denoted by $\mathcal{B}$, see Section \ref{Quantum}.  The purity of $\rho_{l,k}$ is then defined to be
\begin{equation}
  \Tr\left(\rho_{l,k}^{2}\right)
\end{equation}
This quantity is only well defined if the eigenvalue $\lambda_{k}$ is non-degenerate.  If degenerate, the $k^{th}$ eigenstate $|\psi_k\rangle$ of $H_{n}^{(local)}$ is not well defined and then $\rho_{l,k}$ and therefore the purity of $\rho_{l,k}$ are also not.

\subsubsection{Extremal values of the purity}\label{extreme}
The Schmidt decomposition, see Section \ref{Schmidt}, states that there exists orthonormal bases of the Hilbert spaces $\mathcal{A}$ and $\mathcal{B}$ of subsystems $A$ and $B$, say $\{|a\rangle_\mathcal{A}\}_{a=1}^{2^{l}}$ and $\{|b\rangle_\mathcal{B}\}_{b=1}^{2^{n-l}}$ respectively, and real scalars $0\leq s_{j}\leq1$ such that $\sum_{j=1}^{2^{l}}s_{j}^{2}=1$ for which
\begin{equation}
  |\psi_{k}\rangle=\sum_{j=1}^{2^{l}}s_{j}|j\rangle_\mathcal{A}|j\rangle_\mathcal{B}
\end{equation}
where it is assumed that $2^{l}\leq2^{n-l}$ (in fact such a decomposition holds for any state, not just $|\psi_k\rangle$)   Therefore,
\begin{equation}
  |\psi_{k}\rangle\langle\psi_{k}|=\sum_{j,j^\prime=1}^{2^{l}}s_{j}s_{j^\prime}\left(|j\rangle_\mathcal{A}\,_\mathcal{A}\langle j^\prime|\right)\otimes\left(|j\rangle_\mathcal{B}\,_\mathcal{B}\langle j^\prime|\right)
\end{equation}
By taking the partial trace over $\mathcal{B}$ of this expression it is seen that
\begin{align}
  \rho_{l,k}&=\Tr_\mathcal{B}\left(|\psi_{k}\rangle\langle\psi_{k}|\right)\nonumber\\
  &=\sum_{b=1}^{2^{n-l}}\,_\mathcal{B}\langle b|\left(\sum_{j,j^\prime=1}^{2^{l}}s_{j}s_{j^\prime}\left(|j\rangle_\mathcal{A}\,_\mathcal{A}\langle j^\prime|\right)\otimes\left(|j\rangle_\mathcal{B}\,_\mathcal{B}\langle j^\prime|\right)\right)|b\rangle_\mathcal{B}\nonumber\\
  &=\sum_{j=1}^{2^{l}}s_{j}^{2}|j\rangle_\mathcal{A}\,_\mathcal{A}\langle j|
\end{align}
and therefore the purity $\Tr\left(\rho_{l,k}^{2}\right)$ is
\begin{equation}
  \Tr\left(\sum_{j,j^\prime=1}^{2^{l}}s_{j}^{2}s_{j^\prime}^{2}|j\rangle_\mathcal{A}\,_\mathcal{A}\langle j|j^\prime\rangle_\mathcal{A}\,_\mathcal{A}\langle j^\prime|\right)=\sum_{j=1}^{2^{l}}s_{j}^{4}
\end{equation}

The minimal value of $\Tr\left(\rho_{l,k}^{2}\right)$ can be calculated with the Cauchy-Schwartz inequality.  This states that for any $x_{1},\dots,x_{2^{l}}\in\mathbb{C}$ and any $y_{1},\dots,y_{2^{l}}\in\mathbb{C}$ that
\begin{equation}
  \left|\sum_{j=1}^{2^{l}}x_{j}\overline{y}_{j}\right|^{2}\leq\sum_{j=1}^{2^{l}}|x_{j}|^{2}\sum_{j^\prime=1}^{2^{l}}|x_{j^\prime}|^{2}
\end{equation}
with equality only when $x_{j}=cy_{j}$ for all $j$ and some fixed $c\in\mathbb{C}$.  Setting $x_{j}=s_{j}^{2}$ and $y_{j}=1$ yields that
\begin{equation}
  \left(\sum_{j=1}^{2^{l}}s_{j}^{2}\right)^{2}\leq2^{l}\sum_{j=1}^{2^{l}}s_{j}^{4}
\end{equation}
Since $\sum_{j}s_{j}^{2}=1$, the minimal value of $\Tr\left(\rho_{l,k}^{2}\right)=\sum_{j}s_{j}^{4}$ is therefore $\frac{1}{2^{l}}$.  This minimal value is only achieved when $s_{j}^{2}=c$ for all $j$ by the Cauchy-Schwartz inequality.  As $\sum_{j}s_{j}^{2}=1$, it must be that $s_{j}^{2}=c=\frac{1}{2^{l}}$ so that $s_{j}=\frac{1}{\sqrt{2^{l}}}$ for all $j$.

In addition to this, it is clear that
\begin{equation}
  \sum_{j\neq j^\prime}s_{j}^{2}s_{j^\prime}^{2}\geq0
\end{equation}
as all the terms in the sum are positive.  Therefore,
\begin{equation}
  \sum_{j=1}^{2^{l}}s_{j}^{4}\leq\sum_{j=1}^{2^{l}}s_{j}^{4}+\sum_{j\neq j^\prime}s_{j}^{2}s_{j^\prime}^{2}=\left(\sum_{j=1}^{2^{l}}s_{j}^{2}\right)^{2}
\end{equation}
so that, as $\sum_{j}s_{j}^{2}=1$, $\Tr\left(\rho_{l,k}^{2}\right)=\sum_{j}s_{j}^{4}$ is at most $1$.  If $s_{j}<1$ for all $j$ then $s_{j}^{4}<s_{j}^{2}$ for all $j$ and $\sum_{j}s_{j}^{4}<\sum_{j}s_{j}^{2}=1$.  Therefore to gain the maximal value of $1$ for $\Tr\left(\rho_{l,k}^{2}\right)$, at least one of the $s_{j}$ has to be equal to $1$,  As $\sum_{j}s_{j}^{2}=1$ there can be at most one $s_{j}$ equal to $1$ and the rest must be zero.

\subsubsection{Entanglement}
If $\Tr\left(\rho_{l,k}^{2}\right)=1$ then the original state $|\psi_{k}\rangle$ must have had the form
\begin{equation}
  |j\rangle_\mathcal{A}|j\rangle_\mathcal{B}
\end{equation}
for some $1\leq j\leq2^{l}$, by the previous Schmidt decomposition.  In this case $|\psi_{k}\rangle$ is a product state over subsystems $A$ and $B$ and hence there is no entanglement.  In particular the von-Neumann entropy of $\rho_{l,k}=|j\rangle\langle j|$ is zero, its minimal value, see Section \ref{Von}.

If $\Tr\left(\rho_{l,k}^{2}\right)=\frac{1}{2^{l}}$ then the original state $|\psi_{k}\rangle$ must have had the form
\begin{equation}
  \sum_{j=1}^{2^{l}}\frac{1}{\sqrt{2^{l}}}|j\rangle_\mathcal{A}|j\rangle_\mathcal{B}
\end{equation}
by the previous Schmidt decomposition.  This state is maximally entangled over subsystems $A$ and $B$, in particular the von-Neumann entropy of $\rho_{l,k}=\sum_{j=1}^{2^{l}}\frac{1}{2^{l}}|j\rangle\langle j|$ is
\begin{equation}
  -\sum_{j=1}^{2^{l}}\frac{1}{2^{l}}\log\left(\frac{1}{2^{l}}\right)=-\log\left(\frac{1}{2^{l}}\right)=\log\left(2^{l}\right)
\end{equation}
which is maximal, see Section \ref{Von}.

The purity $\Tr\left(\rho_{l,k}^{2}\right)$ is a polynomial in the vector elements of the state $|\psi_{k}\rangle$ in some fixed basis.  The space of all states is necessarily compact by the state normalisation condition.  If the purity associated to a sequence of such states converges to an extremal value then the states in that sequence must become arbitrary close to states (no necessarily a single unique one) with the properties associated with the extremal purity value, as described above, by continuity.

\subsection{Translation matrix}\label{TOperator}
Let $|\boldsymbol{x}\rangle=|x_{1}\rangle\otimes\dots\otimes|x_{n}\rangle$ for the multi-indices $\boldsymbol{x}=(x_{1},\dots,x_{n})\in\{0,1\}^{n}$ be the standard basis for a $n$-qubit chain's Hilbert space, $\left(\mathbb{C}^{2}\right)^{\otimes n}$, see Section \ref{stndBasis}.  Let $T$ be the unitary matrix acting on $\left(\mathbb{C}^{2}\right)^{\otimes n}$ (that is $T^{-1}=T^\dagger$) such that
\begin{equation}
  T|x_{1}\rangle\otimes|x_{2}\rangle\otimes\dots\otimes|x_{n}\rangle=|x_{n}\rangle\otimes|x_{1}\rangle\otimes\dots\otimes|x_{n-1}\rangle
\end{equation}
The matrix $T$ then translates a state of the chain by one qubit, see Figure \ref{Trans}.

\begin{figure}
  \centering
  \input{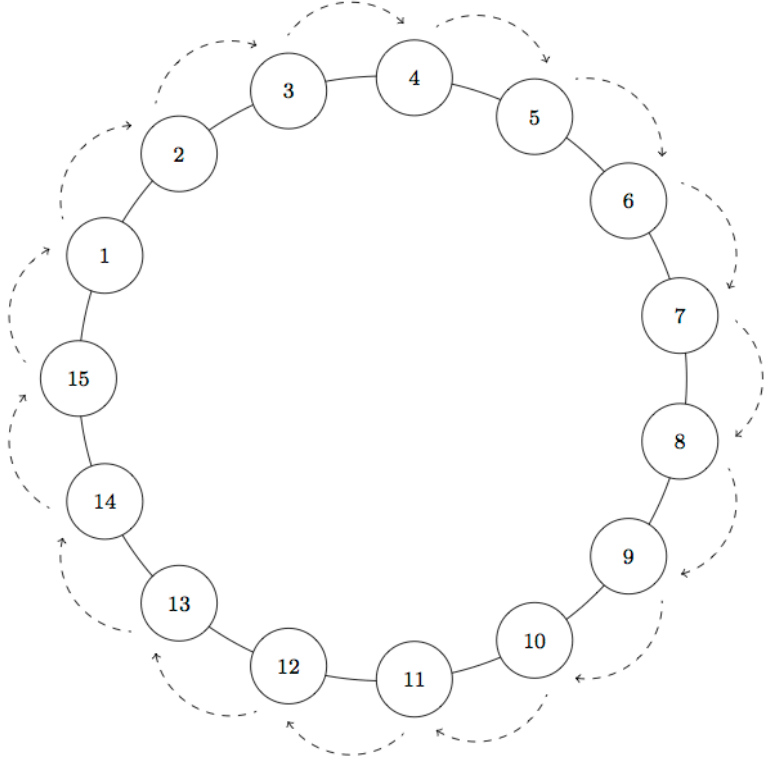}
  \caption[Action of the translation matrix $T$]{The translation matrix $T$, in effect, relabels the qubits $j$ as $j+1$, cyclically, shown here for $n=15$ qubits, represented as the circles labelled $1$ to $15$.}
  \label{Trans}
\end{figure}

The Pauli matrices $\sigma_{  j  }^{(a)}$ can also be translated by the translation matrix $T$ as
\begin{equation}
  T\sigma_{  j  }^{(a)}T^\dagger=\sigma_{  j+1  }^{(a)}
\end{equation}
Moreover, defining the interaction between the qubits labelled $1$ and $2$ in $H_{n}^{(inv)}$ as
\begin{equation}
  h_{1}=\sum_{a,b=1}^{3}\alpha_{a,b}\sigma_{  1  }^{(a)}\sigma_{  2  }^{(b)}
\end{equation}
the whole Hamiltonian $H_{n}^{(inv)}$ may then be written as
\begin{equation}
  H_{n}^{(inv)}=\sum_{j=0}^{n-1}T^{j}h_{1}T^{-j}
\end{equation}
and it is clearly seen that
\begin{equation}
  H_{n}^{(inv)}=TH_{n}^{(inv)}T^\dagger
\end{equation}

The definition of the translation matrix $T$ implies that $T^{n}=I_{2^{n}}$.  If $|\psi\rangle$ is an eigenstate of $T$ with eigenvalue $\lambda$ then
\begin{equation}
  I|\psi\rangle=T^{n}|\psi\rangle=\lambda^{n}|\psi\rangle
\end{equation}
implying that $\lambda^{n}=1$.  Since $T$ is unitary, $\lambda$ must be complex with $|\lambda|=1$, so that there exists some $\theta\in[0,2\pi)$ such that $\lambda=\e^{\im\theta}$.  As $\lambda^{n}=1$ it follows that the only eigenvalues of $T$ are $\omega^{0},\dots,\omega^{n-1}$ for $\omega=\e^{\frac{\im}{n}}$, that is the $n$, $n^{th}$ roots of unity.
\section{Numerics}\label{eigenNum}
The methodology and results of numerical simulations of the eigenstate purity for a range of relevant ensembles will now be presented.

\subsection{Models}
Hamiltonians from all the ensembles previous considered in the numerical simulations described in Section \ref{Numerics} will be considered here again.  The two ensembles of particular note, as they display different behaviours, are the ensemble of generic Hamiltonians without local terms
\begin{equation}
	\hat{H}_{n}=\sum_{j=1}^{n}\sum_{a,b=1}^{3}\hat{\alpha}_{a,b,j}\sigma_{  j  }^{(a)}\sigma_{  j+1  }^{(b)}\qquad\qquad\hat{\alpha}_{a,b,j}\sim\mathcal{N}\left(0,\frac{1}{9n}\right) \iid
\end{equation}
and the ensemble of translationally-invariant Hamiltonians with local terms
\begin{align}
	\hat{H}^{(inv,local)}_{n}&=\sum_{j=1}^{n}\sum_{a=1}^{3}\sum_{b=0}^{3}\hat{\alpha}_{a,b}\sigma_{  j  }^{(a)}\sigma_{  j+1  }^{(b)}\qquad\qquad\hat{\alpha}_{a,b}\sim\mathcal{N}\left(0,\frac{1}{12n}\right) \iid
\end{align}

It will be seen in general that the presence, or lack thereof, of local terms in the ensembles defined in Section \ref{Numerics} leads to the two distinct types of behaviours exemplified by the ensembles $\hat{H}_{n}$ and $\hat{H}_{n}^{(inv,local)}$.   These two ensembles are highlighted here as they contain the most general Hamiltonians to which the two main theorems of this chapter apply.

\subsection{Numerical methodology}\label{unfoldSec}
As in Chapter \ref{DOS} the numerical simulations of the ensembles above were carried out on a computer with a quad-core Intel i$5$ processor running at $2.9$GHz with $6$MB L$3$ cache and $8$GB of RAM.  The code was written in C++ and used the GNU Scientific Library (GSL) version 1.15, see \cite{GSL} for full documentation.  The matrices were generated and diagonalised in the same way as described in Section \ref{Numerics}.

Once each matrix sampled had been diagonalised, the reduced density matrix $\rho_{l,k}$ on the qubits labelled $1$ to $l=1,2,3,4,5$ of the $k^{th}$ ordered eigenstate (with respect to increasing eigenvalue) was calculated and from this the value of $\Tr\left(\rho_{l,k}^{2}\right)$ calculated.  The average of $\Tr\left(\rho_{l,k}^{2}\right)$ was then taken over the samples from each ensemble separately.  An average was taken here to reduce the statistical fluctuations in the results.  The analysis to follow will focus on a single generic instance of each ensemble though. 

For Hamiltonians which exhibited a degenerate spectrum, a choice of the eigenstates within each degenerate subspace was be made by the numerical algorithm implemented in the GSL numerical library.

\subsection{Results}\label{EigenResults}
The averaged values of the linear entropy, $1-\Tr \left(\rho_{l,k}^{2}\right)$, of the reduced density matrices $\rho_{l,k}$ for each eigenstate $\rho_{k}$, corresponding to the ordered eigenvalues $\lambda_{k}$, over $s$ random samples from each of the ensembles $\hat{H}_{12}$ and $\hat{H}_{12}^{(inv,local)}$ are plotted in Figures \ref{EntH_12B} and \ref{EntH_13B} respectively.  The corresponding graphs for lower values of $n$ are shown in Appendix \ref{LinEntGraphs}.  

\begin{figure}
  \centering
  \input{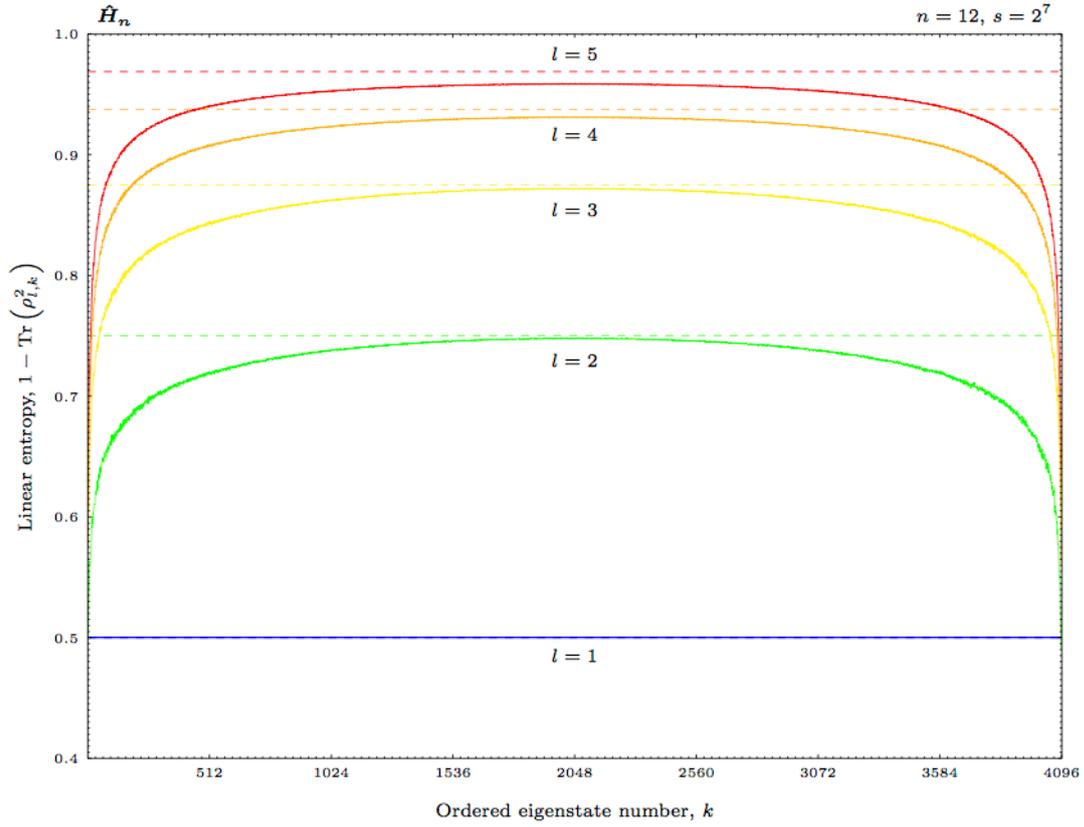}
  \caption[Average linear entropy of the reduced eigenstates over $\hat{H}_{12}$]{The average value of the linear entropy, $1-\Tr\left(\rho_{l,k}^{2}\right)$, over $s=2^{7}$ samples from $\hat{H}_{12}$, where $\rho_{l,k}$ is the reduced density matrix, on the qubits labelled $1$ to $l$, of the eigenstate $\rho_{k}$ corresponding to the numerically order eigenvalue $\lambda_{k}$, for each sample.  The points for each value of $l$ have been joined for clarity and the dashed lines are at the maximal linear entropy values, $1-\frac{1}{2^{l}}$.}
  \label{EntH_12B}
\end{figure}

\begin{figure}
  \centering
  \input{Figures/EntH_13}
  \caption[Average linear entropy of the reduced eigenstates over $\hat{H}_{12}^{(inv,local)}$]{The average value of the linear entropy, $1-\Tr\left(\rho_{l,k}^{2}\right)$, over $s=2^{7}$ samples from $\hat{H}_{12}^{(inv,local)}$, where $\rho_{l,k}$ is the reduced density matrix, on the qubits labelled $1$ to $l$, of the eigenstate $\rho_{k}$ corresponding to the numerically order eigenvalue $\lambda_{k}$, for each sample.  The points for each value of $l$ have been joined for clarity and the dashed lines are at the maximal  linear entropy values, $1-\frac{1}{2^{l}}$.}
  \label{EntH_13B}
\end{figure}

A clear difference between the two ensembles is seen in the case $l=1$.  The samples of $\hat{H}_{n}$ for even values of $n=2,\dots,12$ all had a non-degenerate spectrum (degeneracies were seen for odd values of $n$) and the corresponding unique values of the linear entropy are seen to be exactly the maximal value of one half.  For all values of $n=5,\dots,13$ the spectrum of all the samples of $\hat{H}_{n}^{(inv,local)}$ were non-degenerate and the unique values of the linear entropy are seen to approach the maximal value of one half as $n$ increases, throughout the bulk of the spectrum.

For the values $l=2,3,4,5$ the average values of the linear entropy calculated were also seen to approach their maximal respective values of $1-\frac{1}{2^{l}}$ throughout the bulk of the spectrum as $n$ increased for both $\hat{H}_{n}$ and $\hat{H}_{n}^{(inv,local)}$.  The behaviour at the edge of the spectrum is different however, with a lower value of the linear entropy being observed on average.  This could coincide with the area law for entanglement of the low lying eigenstates, as described in Section \ref{etangChain}.

The two types of behaviours exemplified by the ensembles $\hat{H}_{n}$ and $\hat{H}_{n}^{(inv,local)}$ are common to the other ensembles numerically studied, without or with local terms respectively.  The corresponding linear entropy graphs are given in Appendix \ref{LinEntGraphs} for the ensembles $\hat{H}_{n}^{(uniform)}$, $\hat{H}_{n}^{(local)}$, $\hat{H}_{n}^{(inv)}$, $\hat{H}_{n}^{(JW)}$ and $\hat{H}_{n}^{(Heis)}$ for $n=2,\dots,13$, as defined in Section \ref{Numerics}.
\section{Single qubit reduced eigenstate purity for generic chains without local terms}\label{SingleEnt}
Any fixed $2^{n}\times2^{n}$ Hermitian matrix
\begin{equation}
  H_{n}=\sum_{j=1}^{n}\sum_{a,b=1}^{3}\alpha_{a,b,j}\sigma_{  j  }^{(a)}\sigma_{  j+1  }^{(b)}
\end{equation}
where the $\alpha_{a,b,j}$ are some fixed real numbers, has a (not necessarily unique) spectral decomposition
\begin{equation}
  H_{n}=\sum_{k=1}^{2^{n}}\lambda_{k}|\psi_{k}\rangle\langle\psi_{k}|
\end{equation}
where the $\lambda_{1}\leq\lambda_{2}\leq\dots\leq\lambda_{2^{n}}$ are the eigenvalues of $H_{n}$ and the $|\psi_{k}\rangle$ are corresponding normalised eigenstates.   Note that this includes translationally-invariant Hamiltonians of the form $H_{n}^{(inv)}$.

Assuming that the eigenvalues are all non-degenerate, the purity, $\Tr\left(\rho_{1,k}^{2}\right)$, of the reduced density matrix $\rho_{1,k}=\Tr_\mathcal{B}\left(|\psi_{k}\rangle\langle\psi_{k}|\right)$ on subsystem $A$ comprised of the single qubit labelled $1$ and where subsystem $B$ is comprised of the remaining qubits (see Figure \ref{1Qubit}), may be calculated.  This calculation is performed in the following theorem:

\begin{figure}
  \centering
  \input{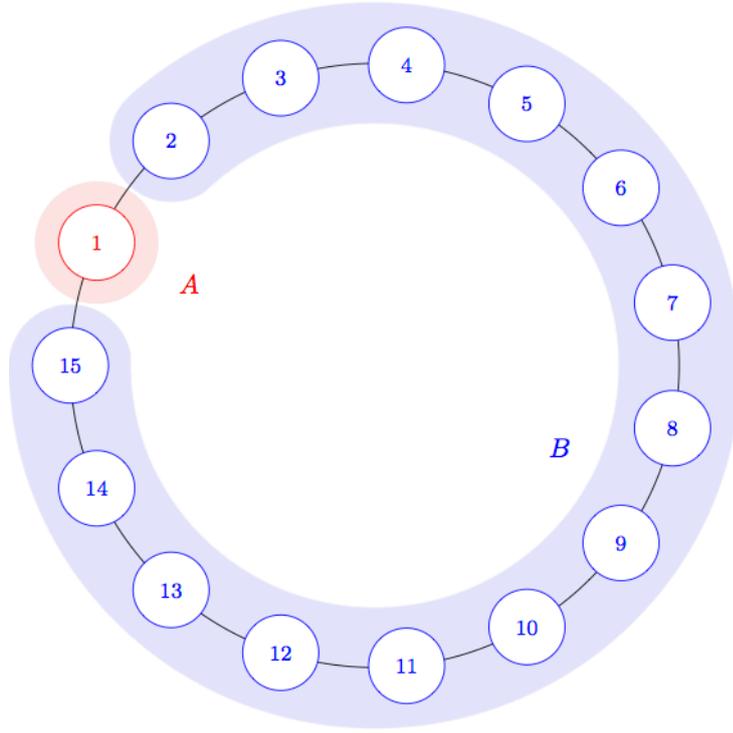}
  \caption[Splitting the $n$-qubit system into a single qubit and $n-1$ qubits]{A $n=15$ qubit nearest-neighbour chain where the qubits are represented by circles and the nearest-neighbour interactions by links.  The system is split into two subsystems, $A$ comprised of the single qubit labelled $1$ and $B$ comprised of the remaining qubits.}
  \label{1Qubit}
\end{figure}

\subsection{Single qubit reduced eigenstate purity for the Hamiltonian \texorpdfstring{$H_{n}$}{Hn}}
\begin{theorem}[Single qubit reduced eigenstate purity for the Hamiltonian $H_{n}$]\label{eigenNonInv}
  For the $2^{n}\times2^{n}$ Hermitian matrix
  \begin{equation}
    H_{n}=\sum_{j=1}^{n}\sum_{a,b=1}^{3}\alpha_{a,b,j}\sigma_{  j  }^{(a)}\sigma_{  j+1  }^{(b)}\equiv\sum_{k=1}^{2^{n}}\lambda_{k}|\psi_{k}\rangle\langle\psi_{k}|
  \end{equation}
  where $\alpha_{a,b,j}\in\mathbb{R}$ are real constants, $\lambda_{k}\in\mathbb{R}$ are the non-degenerate eigenvalues of $H_{n}$ and $|\psi_{k}\rangle$ are corresponding normalised eigenstates of $H_{n}$, the reduced density matrix on a single qubit of any eigenstate $|\psi_{k}\rangle$ of $H_{n}$ is then maximally mixed.  That is for any $k$,
  \begin{equation}
    \rho_{1,k}\equiv\Tr_\mathcal{B}\left(|\psi_{k}\rangle\langle\psi_{k}|\right)=\frac{I_{2}}{2}
  \end{equation}
\end{theorem}

This theorem directly implies that
\begin{equation}
  \Tr\left(\rho_{1,k}^{2}\right)=\frac{1}{4}\Tr \left(I_{2}^{2}\right)=\frac{1}{2}
\end{equation}
as seen in the numerical simulations.  It is also noted that the following proof also holds for the reduced density matrices of any single qubit, not just the first.

\begin{proof}\hspace{-2mm}\footnote{Adapted from the proof given by the author in \cite{KLW2014_FIXED} based on discussions with \cite{Linden}.}
The proof will be in two main parts.  First, a symmetry of $H_{n}$ will be determined.  Then this symmetry will be applied to the Pauli matrix expansion of $\rho_{1,k}$ to show that it does indeed equal $\frac{1}{2}I_{2}$:

\subsubsection{Symmetry}
Let the matrix $S$ be
\begin{equation}
  S=\prod_{j=1}^{n}\sigma_{j}^{(2)}
\end{equation}
By the definition of the Pauli matrices it is seen that $S$ is Hermitian as
\begin{align}
  S^\dagger=\left(\prod_{j=1}^{n}\sigma_{j}^{(2)}\right)^\dagger=\prod_{j=1}^{n}{\sigma_{j}^{(2)}}^\dagger=\prod_{j=1}^{n}\sigma_{j}^{(2)}=S
\end{align}
and unitary as
\begin{align}
  S^{-1}=\left(\prod_{j=1}^{n}\sigma_{j}^{(2)}\right)^{-1}=\prod_{j=1}^{n}{\sigma_{j}^{(2)}}^{-1}=\prod_{j=1}^{n}\sigma_{j}^{(2)}=S
\end{align}

Furthermore, by the definition of the Pauli matrices
\begin{align}
  \sigma_{}^{(2)\dagger}\sigma^{(1)}\sigma^{(2)}&=\begin{pmatrix}0&-\im\\\im&0\end{pmatrix}\begin{pmatrix}0&1\\1&0\end{pmatrix}\begin{pmatrix}0&-\im\\\im&0\end{pmatrix}=\begin{pmatrix}0&-1\\-1&0\end{pmatrix}=-\overline{\sigma^{(1)}}\nonumber\\
  \sigma_{}^{(2)\dagger}\sigma^{(2)}\sigma^{(2)}&=\begin{pmatrix}0&-\im\\\im&0\end{pmatrix}\begin{pmatrix}0&-\im\\\im&0\end{pmatrix}\begin{pmatrix}0&-\im\\\im&0\end{pmatrix}=\begin{pmatrix}0&-\im\\\im&0\end{pmatrix}=-\overline{\sigma^{(2)}}\nonumber\\
  \sigma_{}^{(2)\dagger}\sigma^{(3)}\sigma^{(2)}&=\begin{pmatrix}0&-\im\\\im&0\end{pmatrix}\begin{pmatrix}1&0\\0&-1\end{pmatrix}\begin{pmatrix}0&-\im\\\im&0\end{pmatrix}=\begin{pmatrix}-1&0\\0&1\end{pmatrix}=-\overline{\sigma^{(3)}}
\end{align}
where the overbar denotes entry-wise complex conjugation, so that
\begin{align}
  S^\dagger\sigma_{  j  }^{(1)}S&=-\overline{\sigma_{  j  }^{(1)}}\nonumber\\
  S^\dagger\sigma_{  j  }^{(2)}S&=-\overline{\sigma_{  j  }^{(2)}}\nonumber\\
  S^\dagger\sigma_{  j  }^{(3)}S&=-\overline{\sigma_{  j  }^{(3)}}
\end{align}
Since $S$ is unitary, that is $SS^\dagger=I_{2^{n}}$, it follows that for any $a,b=1,2,3$,
\begin{equation}
  S^\dagger\sigma_{  j  }^{(a)}\sigma_{  j+1  }^{(b)}S=S^\dagger\sigma_{  j  }^{(a)}SS^\dagger\sigma_{  j+1  }^{(b)}S=\overline{\sigma_{  j  }^{(a)}}\overline{\sigma_{  j+1  }^{(b)}}
\end{equation}
and that
\begin{equation}
  S^\dagger H_{n}S=\sum_{j=1}^{n}\sum_{a,b=1}^{3}\alpha_{a,b,j}S^\dagger\sigma_{  j  }^{(a)}\sigma_{  j+1  }^{(b)}S
  =\sum_{j=1}^{n}\sum_{a,b=1}^{3}\alpha_{a,b,j}\overline{\sigma_{  j  }^{(a)}\sigma_{  j+1  }^{(b)}}=\overline{H_{n}}
\end{equation}
as the coefficients $\alpha_{a,b,j}$ are real.

From the definition of $H_{n}$, $H_{n}|\psi_{k}\rangle=\lambda_{k}|\psi_{k}\rangle$ for any $k$.  The complex conjugate of this expression yields that $\overline{H_{n}|\psi_{k}\rangle}=\lambda_{k}\overline{|\psi_{k}\rangle}$ as $\lambda_{k}$ is real.  Since $\overline{H_{n}}=S^\dagger H_{n}S$ and $S=S^{-1}=S^\dagger$, this is equivalent to $H_{n}S\overline{|\psi_{k}\rangle}=\lambda_{k}S\overline{|\psi_{k}\rangle}$.  Therefore $S\overline{|\psi_{k}\rangle}$ and $|\psi_{k}\rangle$ are both eigenstates of $H_{n}$ with the same (non-degenerate) eigenvalue $\lambda_{k}$.  An eigenspace associated with a non-degenerate eigenvalue has complex dimension one, so $S\overline{|\psi_{k}\rangle}$ must be a complex multiple of $|\psi_{k}\rangle$, that is $S\overline{|\psi_{k}\rangle}=c|\psi_{k}\rangle$ for some $c\in\mathbb{C}$.  The state $|\psi_{k}\rangle$ is normalised and $S^\dagger S=I$ so
\begin{equation}
  1=\langle\psi_{k}|\psi_{k}\rangle=\overline{\langle\psi_{k}|\psi_{k}\rangle}=\overline{\langle\psi_{k}|}\,\overline{|\psi_{k}\rangle}=\overline{\langle\psi_{k}|}S^\dagger S\overline{|\psi_{k}\rangle}=|c|^{2}\langle\psi_{k}|\psi_{k}\rangle=|c|^{2}
\end{equation}
which implies that $S\overline{|\psi_{k}\rangle}=\e^{\im\theta_{k}}|\psi_{k}\rangle$ for some $\theta_{k}\in[0,2\pi)$.

As $\langle\psi_{k}|\sigma_{  j  }^{(a)}|\psi_{k}\rangle=\langle\psi_{k}|\e^{-\im\theta_{k}}\sigma_{  j  }^{(a)}\e^{\im\theta_{k}}|\psi_{k}\rangle$ since $\e^{\im\theta_{k}}\e^{-\im\theta_{k}}=1$, this identity can be used to see that
\begin{equation}
 \langle\psi_{k}|\sigma_{  j  }^{(a)}|\psi_{k}\rangle=\langle\psi_{k}|\e^{-\im\theta_{k}}\sigma_{  j  }^{(a)}\e^{\im\theta_{k}}|\psi_{k}\rangle=\overline{\langle\psi_{k}|}S^\dagger\sigma_{  j  }^{(a)}S\overline{|\psi_{k}\rangle}
\end{equation}
for all $a=1,2,3$.  Recalling that $S^\dagger\sigma_{  j  }^{(a)}S=-\overline{\sigma_{  j  }^{(a)}}$, then reduces this identity to
\begin{equation}
 \langle\psi_{k}|\sigma_{  j  }^{(a)}|\psi_{k}\rangle=-\overline{\langle\psi_{k}|\sigma_{  j  }^{(a)}|\psi_{k}\rangle}
\end{equation}
As $\sigma_{  j  }^{(a)}=\left(\sigma_{  j  }^{(a)}\right)^\dagger$,
\begin{equation}
  \overline{\langle\psi_{k}|\sigma_{  j  }^{(a)}|\psi_{k}\rangle}=\langle\psi_{k}|\left(\sigma_{  j  }^{(a)}\right)^\dagger|\psi_{k}\rangle=\langle\psi_{k}|\sigma_{  j  }^{(a)}|\psi_{k}\rangle
\end{equation}
and it is concluded that $\langle\psi_{k}|\sigma_{  j  }^{(a)}|\psi_{k}\rangle=-\langle\psi_{k}|\sigma_{  j  }^{(a)}|\psi_{k}\rangle=0$.

\subsubsection{Pauli basis expansion}
The $2^{n}\times2^{n}$ density matrix $\rho_{k}=|\psi_{k}\rangle\langle\psi_{k}|$ is Hermitian as $\rho_{k}^\dagger=\left(|\psi_{k}\rangle\langle\psi_{k}|\right)^\dagger=\left(\langle\psi_{k}|\right)^\dagger\left(|\psi_{k}\rangle\right)^\dagger=|\psi_{k}\rangle\langle\psi_{k}|=\rho_{k}$.  As seen in Section \ref{PauliMatrixBasis}, there must then exist real coefficients $c_{a_{1},\dots,a_{n}}$ such that
\begin{equation}\label{3.3.17}
  \rho_{k}=\sum_{a_{1},\dots,a_{n}=0}^{3}c_{a_{1},\dots,a_{n}}\sigma^{( a_{1} )}\otimes\dots\otimes\sigma^{( a_{n} )}
\end{equation}
The reduced matrix $\rho_{1,k}=\Tr_\mathcal{B}\left(\rho_{k}\right)$, on system $A$ (the single qubit labelled $1$), is then
\begin{equation}
  \rho_{1,k}=\sum_{a_{1},\dots,a_{n}=0}^{3}c_{a_{1},\dots,a_{n}}\sigma^{( a_{1} )}\Tr\left(\sigma^{( a_{2} )}\otimes\dots\otimes\sigma^{( a_{n} )}\right)
\end{equation}
seen by taking the partial trace inside the sum.  This expression may be equivalently rewritten as
\begin{equation}
  \rho_{1,k}=\sum_{b=0}^{3} d_{b}\sigma^{(b)}
\end{equation}
where the real coefficients $d_{b}$ are given by
\begin{equation}
  d_{b}=\sum_{a_{2},\dots,a_{n}=0}^{3}c_{b,a_{2},\dots,a_{n}}\Tr\left(\sigma^{( a_{2} )}\otimes\dots\otimes\sigma^{( a_{n} )}\right)
\end{equation}

As $\Tr\left(\sigma^{(b)}\sigma^{(a_{1})}\right)=2\delta_{b,a_{1}}$, this allows $d_{b}$ to be rewritten as
\begin{align}
  d_{b}&=\sum_{a_{1},\dots,a_{n}=0}^{3}c_{a_{1},\dots,a_{n}}\delta_{b,a_{1}}\Tr\left(\sigma^{( a_{2} )}\otimes\dots\otimes\sigma^{( a_{n} )}\right)\nonumber\\
  &=\frac{1}{2}\sum_{a_{1},\dots,a_{n}=0}^{3}c_{a_{1},\dots,a_{n}}\Tr\left(\sigma^{(b)}\sigma^{(a_{1})}\right)\Tr\left(\sigma^{( a_{2} )}\otimes\dots\otimes\sigma^{( a_{n} )}\right)
\end{align}
which, by rearranging the traces is equivalently written
\begin{align}
  d_{b}&=\frac{1}{2}\Tr\left(\sigma_{  1  }^{(b)}\sum_{a_{1},\dots,a_{n}=0}^{3}c_{a_{1},\dots,a_{n}}\sigma^{( a_{1} )}\otimes\dots\otimes\sigma^{( a_{n} )}\right)
\end{align}
The previous expansion of $\rho_{k}=|\psi_{k}\rangle\langle\psi_{k}|$ (\ref{3.3.17}) then implies that this is equal to $\frac{1}{2}\Tr\left(\sigma_{  1  }^{(b)}|\psi_{k}\rangle\langle\psi_{k}|\right)$ or equivalently $\frac{1}{2}\langle\psi_{k}|\sigma_{  1  }^{(b)}|\psi_{k}\rangle$.  It has already been seen that this is equal to zero, if $b=1,2,3$, and is trivially equal to $\frac{1}{2}$ if $b=0$ as $\langle\psi_{k}|\psi_{k}\rangle=1$.  Therefore,
\begin{equation}
  \rho_{1,k}=\sum_{b=0}^{3} d_{b}\sigma^{(b)}=d_{0}\sigma^{(0)}=\frac{I_{2}}{2}
\end{equation}
which concludes the proof.
\end{proof}

\subsection{Extension to more general Hamiltonians}
The proof relies on the non-degeneracy of the eigenvalues of the Hamiltonian $H_{n}$ and on the occurrence of only pairs of Pauli matrices $\sigma_{j}^{(a)}$ in its definition.  For any Hamiltonian with a non-degenerate spectrum that can be expressed as a sum of products of an even number of Pauli matrices $\sigma_{j}^{(a)}$, the preceding proof holds for the reduced eigenstates on a single qubit.
\section{Reduced eigenstate purity bounds for generic translationally-invariant Hamiltonians}\label{EntLblock}
To investigate the entanglement between a block of more than one qubit and the rest of the chain, the translational invariance of the fixed matrices
\begin{equation}
  H_{n}^{(inv,local)}=\sum_{j=1}^{n}\sum_{a=1}^{3}\sum_{b=0}^{3}\alpha_{a,b}\sigma_{  j  }^{(a)}\sigma_{  j+1  }^{(b)}
\end{equation}
will be used.  It will be seen that the methods presented in this section will not apply to the non-translationally-invariant matrices which were focused on previously.

\subsection{Reduced eigenstate purity bounds for the Hamiltonians \texorpdfstring{$H_{n}^{(inv,local)}$}{Hn(inv,local)}}
\begin{theorem}[Reduced eigenstate purity bounds for the Hamiltonians $H_{n}^{(inv,local)}$]\label{invEnt}
  For the Hamiltonian
  \begin{equation}
    H_{n}^{(local,inv)}=\sum_{j=1}^{n}\sum_{a=1}^{3}\sum_{b=0}^{3}\alpha_{a,b}\sigma_{  j  }^{(a)}\sigma_{  j+1  }^{(b)}\equiv\sum_{k=1}^{2^{n}}\lambda_{k}|\psi_{k}\rangle\langle\psi_{k}|
  \end{equation}
  where $\alpha_{a,b}\in\mathbb{R}$ are real constants, $|\psi_{k}\rangle$ are joint eigenstates of $H_{n}^{(local,inv)}$ and the translation matrix $T$ and $\lambda_{k}$ are the (not necessarily distinct) corresponding eigenvalues of $H_{n}^{(local,inv)}$, the reduced density matrices, $\rho_{l,k}=\Tr_\mathcal{B}\left(|\psi_{k}\rangle\langle\psi_{k}|\right)$, on the qubits labelled $1$ to $l$, satisfy
  \begin{equation}
	\frac{1}{2^{l}}\leq\frac{1}{2^{n}}\sum_{k=1}^{2^{n}}\Tr\left(\rho_{l,k}^{2}\right)
	\leq\frac{1}{2^{l}}+\frac{2^{l}}{n}
  \end{equation}
  where $\mathcal{B}$ is the Hilbert space of the $n-l$ qubits labelled $l+1$ to $n$ and $2l<n$.
\end{theorem}

As the $|\psi_{k}\rangle$ are eigenstates of the unitary translation matrix $T$, it must be the case that $T|\psi_{k}\rangle=\e^{\im\theta_{k}}|\psi_{k}\rangle$ for some $\theta_{k}\in[0,2\pi)$.   Therefore, by this translational symmetry, the block of $l$ qubits need not be those qubits labelled by $1$ to $l$, but could be any neighbouring block of $l$ qubits. 

It is also noted that if $H_{n}^{(inv,local)}$ has a non-degenerate spectrum, the eigenstates $|\psi_{k}\rangle$ are unique, up to phase.  The proof of Theorem \ref{invEnt} will now be given:
\begin{proof}\hspace{-2mm}\footnote{Adapted from the proof given by the author in \cite{KLW2014_FIXED}.}
The proof of this theorem will be split into two main parts, first the reduced density matrix $\rho_{l,k}$ will be expanded in a Pauli matrix basis.  Then the translational symmetry of the matrix $H_{n}^{(inv,local)}$ will be exploited to bound the coefficients in this expansion so that the value of $\Tr\left(\rho_{l,k}^{2}\right)$ can be bounded:

\subsubsection{Pauli basis expansion}
The $2^{n}\times2^{n}$ density matrix $\rho_{k}=|\psi_{k}\rangle\langle\psi_{k}|$ is Hermitian as $\rho_{k}^\dagger=\left(|\psi_{k}\rangle\langle\psi_{k}|\right)^\dagger=\left(\langle\psi_{k}|\right)^\dagger\left(|\psi_{k}\rangle\right)^\dagger=|\psi_{k}\rangle\langle\psi_{k}|=\rho_{k}$.  As seen in Section \ref{PauliMatrixBasis}, there must then exist real coefficients $c_{a_{1},\dots,a_{n}}$ such that
\begin{equation}\label{3.4.4}
  \rho_{k}=\sum_{a_{1},\dots,a_{n}=0}^{3}c_{a_{1},\dots,a_{n}}\sigma^{( a_{1} )}\otimes\dots\otimes\sigma^{( a_{n} )}
\end{equation}
The reduced density matrix $\rho_{l,k}=\Tr_\mathcal{B}\left(\rho_{k}\right)$, on the qubits labelled $1$ to $l$, is then
\begin{equation}
  \rho_{l,k}=\sum_{a_{1},\dots,a_{n}=0}^{3}c_{a_{1},\dots,a_{n}}\sigma^{( a_{1} )}\otimes\dots\otimes\sigma^{( a_{l} )}\Tr\left(\sigma^{( a_{l+1} )}\otimes\dots\otimes\sigma^{( a_{n} )}\right)
\end{equation}
where the partial trace has been taken inside the sum.  This expression may be equivalently rewritten as
\begin{equation}
  \rho_{l,k}=\sum_{b_{1},\dots,b_{l}=0}^{3} d_{b_{1}\dots,b_{l}}\sigma^{( b_{1} )}\otimes\dots\otimes\sigma^{( b_{l} )}
\end{equation}
where the real coefficients $d_{b_{1}\dots,b_{l}}$ are given by
\begin{equation}
  d_{b_{1}\dots,b_{l}}=\sum_{a_{l+1},\dots,a_{n}=0}^{3}c_{b_{1},\dots,b_{l},a_{l+1},\dots,a_{n}}\Tr\left(\sigma^{( a_{l+1} )}\otimes\dots\otimes\sigma^{( a_{n} )}\right)
\end{equation}

As $\Tr\left(\sigma^{( a_{j} )}\sigma^{(b_{j})}\right)=2\delta_{a_{j},b_{j}}$, this allows $d_{b_{1}\dots,b_{l}}$ to be rewritten as
\begin{align}
  d_{b_{1}\dots,b_{l}}
  &=\sum_{a_{1},\dots,a_{n}=0}^{3}c_{a_{1},\dots,a_{n}}\delta_{a_{1},b_{1}}\dots\delta_{a_{l},b_{l}}\Tr\left(\sigma^{( a_{l+1} )}\otimes\dots\otimes\sigma^{( a_{n} )}\right)\nonumber\\
  &=\sum_{a_{1},\dots,a_{n}=0}^{3}c_{a_{1},\dots,a_{n}}\frac{\Tr\left(\sigma^{( a_{1} )}\sigma^{(b_{1})}\right)}{2}\dots\frac{\Tr\left(\sigma^{( a_{l} )}\sigma^{(b_{l})}\right)}{2}\Tr\left(\sigma^{( a_{l+1} )}\otimes\dots\otimes\sigma^{( a_{n} )}\right)
\end{align}
which, by rearranging the traces is equivalently written
\begin{align}
  d_{b_{1},\dots,b_{l}}&=\frac{1}{2^{l}}\Tr\left(\sigma_{  1  }^{( b_{1} )}\dots\sigma_{  l  }^{( b_{l} )}\sum_{a_{1},\dots,a_{n}=0}^{3}c_{a_{1},\dots,a_{n}}\sigma^{( a_{1} )}\otimes\dots\otimes\sigma^{( a_{n} )}\right)
\end{align}
The previous expansion of $\rho_{k}=|\psi_{k}\rangle\langle\psi_{k}|$ (\ref{3.4.4}) then implies that this is equal to 
\begin{equation}\label{3.56}
 \frac{1}{2^{l}}\Tr\left(\sigma_{  1  }^{( b_{1} )}\dots\sigma_{  l  }^{( b_{l} )}|\psi_{k}\rangle\langle\psi_{k}|\right)=\frac{1}{2^{l}}\langle\psi_{k}|\sigma_{  1  }^{( b_{1} )}\dots\sigma_{  l  }^{( b_{l} )}|\psi_{k}\rangle
\end{equation}

\subsubsection{Bound}
Let
\begin{equation}
  M_{\boldsymbol{b}}=\frac{1}{\sqrt{n}}\sum_{j=0}^{n-1}T^{j}\left(\sigma_{  1  }^{( b_{1} )}\dots\sigma_{  l  }^{( b_{l} )}\right)T^{-j}\equiv\frac{1}{\sqrt{n}}\sum_{j=0}^{n-1}\sigma_{  1+j  }^{( b_{1} )}\dots\sigma_{  l+j  }^{( b_{l} )}
\end{equation}
for the coefficients $b_{1},\dots,b_{l}$ such that it is not the case that $b_{1}=b_{2}=\dots=b_{l}=0$.   This matrix is Hermitian as it is the sum of the Hermitian matrices $\sigma_{  1+j  }^{( b_{1} )}\dots\sigma_{  l+j  }^{( b_{l} )}$.  Using the facts that $T^{-1}|\psi_{k}\rangle=\e^{-\im\theta_{k}}|\psi_{k}\rangle$ for some $\theta_{k}\in[0,2\pi)$ and $\e^{\im j\theta_{k}}\e^{-\im j\theta_{k}}=1$, it is seen that
\begin{align}\label{TransUseEq}
  \langle\psi_{k}|\sigma_{  1  }^{( b_{1} )}\dots\sigma_{  l  }^{( b_{l} )}|\psi_{k}\rangle
  &=\frac{1}{n}\sum_{j=0}^{n-1}\langle\psi_{k}|\e^{\im j\theta_{k}}\sigma_{  1  }^{( b_{1} )}\dots\sigma_{  l  }^{( b_{l} )}\e^{-\im j\theta_{k}}|\psi_{k}\rangle\nonumber\\
  &=\frac{1}{n}\sum_{j=0}^{n-1}\langle\psi_{k}|T^{j}\sigma_{  1  }^{( b_{1} )}\dots\sigma_{  l  }^{( b_{l} )}T^{-j}|\psi_{k}\rangle\nonumber\\
  &=\frac{1}{\sqrt{n}}\langle\psi_{k}|M_{\boldsymbol{b}}|\psi_{k}\rangle
\end{align}
Since $M$ is Hermitian, $\langle\psi_{k}|M_{\boldsymbol{b}}|\psi_{k}\rangle$ is real and therefore $|\langle\psi_{k}|M_{\boldsymbol{b}}|\psi_{k}\rangle|^{2}=\langle\psi_{k}|M_{\boldsymbol{b}}|\psi_{k}\rangle^{2}$.  Also, since the terms $|\langle\psi_{k}|M_{\boldsymbol{b}}|\psi_{k^\prime}\rangle|^{2}$ are positive,
\begin{equation}
  \sum_{k=1}^{2^{n}}|\langle\psi_{k}|M_{\boldsymbol{b}}|\psi_{k}\rangle|^{2}\leq\sum_{k,k^\prime=1}^{2^{n}}|\langle\psi_{k}|M_{\boldsymbol{b}}|\psi_{k^\prime}\rangle|^{2}=\Tr \left(M_{\boldsymbol{b}}M_{\boldsymbol{b}}^\dagger\right)
\end{equation}
and therefore, by the previous identity (\ref{TransUseEq}),
\begin{equation}\label{3.60}
  \sum_{k=1}^{2^{n}}\langle\psi_{k}|\sigma_{  1  }^{( b_{1} )}\dots\sigma_{  l  }^{( b_{l} )}|\psi_{k}\rangle^{2}=\frac{1}{n}\sum_{k=1}^{2^{n}}|\langle\psi_{k}|M_{\boldsymbol{b}}|\psi_{k}\rangle|^{2}\leq\frac{1}{n}\Tr\left( M_{\boldsymbol{b}}M_{\boldsymbol{b}}^\dagger\right)
\end{equation}

\subsubsection{The value of $\Tr \left(M_{\boldsymbol{b}}M_{\boldsymbol{b}}^\dagger\right)$}
The value of $\Tr\left( M_{\boldsymbol{b}}M_{\boldsymbol{b}}^\dagger\right)$ will now be calculated.  For the list of coefficients $b_{1},\dots,b_{l}$, such that at least one is non-zero, let $b_{x}$ be the first one to be non-zero and $b_{y}$ be the last one to be non-zero.  The matrix $\sigma_{  1  }^{( b_{1} )}\dots\sigma_{  l  }^{( b_{l} )}$ then acts non-trivially on at least the qubits labelled $x$ and $y$ (even if $x=y$).

For $2l<n$ consider the matrices $\sigma_{  1  }^{( b_{1} )}\dots\sigma_{  l  }^{( b_{l} )}$ and its translation $\sigma_{  1+j  }^{( b_{1} )}\dots\sigma_{  l+j  }^{( b_{l} )}$.  For $j=1,\dots,n-l$, the matrix $\sigma_{  1  }^{( b_{1} )}\dots\sigma_{  l  }^{( b_{l} )}$ acts on the $x^{th}$ qubit as a non-identity Pauli matrix by definition whereas the matrix $\sigma_{  1+j  }^{( b_{1} )}\dots\sigma_{  l+j  }^{( b_{l} )}$ acts on this qubit as the identity.  For $j=l,...,n-1$ the matrix matrix $\sigma_{  1  }^{( b_{1} )}\dots\sigma_{  l  }^{( b_{l} )}$ acts on the $y^{th}$ qubit as a non-identity Pauli matrix by definition whereas the matrix $\sigma_{  1+j  }^{( b_{1} )}\dots\sigma_{  l+j  }^{( b_{l} )}$ acts on this qubit as the identity.  Therefore the matrices $\sigma_{  1  }^{( b_{1} )}\dots\sigma_{  l  }^{( b_{l} )}$ and its translation $\sigma_{  1+j  }^{( b_{1} )}\dots\sigma_{  l+j  }^{( b_{l} )}$ must be distinct for all $j=1,\dots,n-1$ (given that $2l<n$), see Figure \ref{Diff}.

\begin{figure}
  \centering
  \input{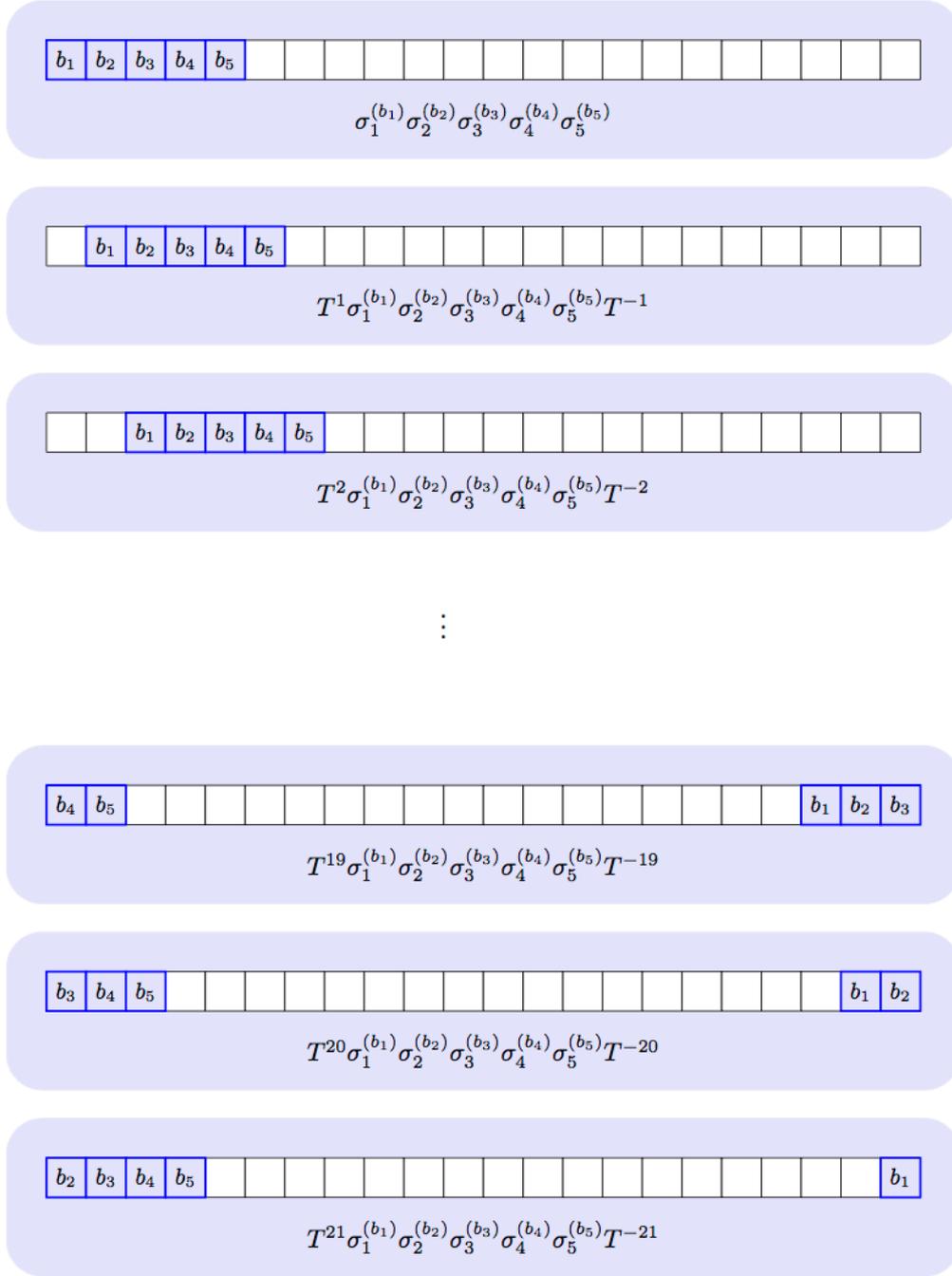}
  \caption[The uniqueness of the matrices $\sigma_{  1+j  }^{( b_{1} )}\dots\sigma_{  5+j  }^{( b_{5} )}$]{A representation of the matrices $T^{j}\sigma_{  1  }^{( b_{1} )}\dots\sigma_{  5  }^{( b_{5} )}T^{-j}=\sigma_{  1+j  }^{( b_{1} )}\dots\sigma_{  5+j  }^{( b_{5} )}$ acting on $n=22$ qubits.  Each qubit is represented by a box, in each line of boxes above, the left most is labelled $1$ and the others labelled sequentially up to $22$.  For the $n=22$ qubit matrix $\sigma_{  1+j  }^{( b_{1} )}\dots\sigma_{  5+j  }^{( b_{5} )}$, the five corresponding single qubit matrices $\sigma^{( b_{1} )},\dots,\sigma^{( b_{5} )}$ are represented by the values $b_{1},\dots,b_{5}$ placed in the corresponding box to the qubit they act upon.  It is seen, if at least one of the $\sigma^{( b_{1} )},\dots,\sigma^{( b_{5} )}$ is not the identity, that all the matrices $\sigma_{  1+j  }^{( b_{1} )}\dots\sigma_{  5+j  }^{( b_{5} )}$ for $j=0,\dots,21$ are distinct.}
  \label{Diff}
\end{figure}

This implies that all the matrices in the sum
\begin{equation}
  M_{\boldsymbol{b}}=\frac{1}{\sqrt{n}}\sum_{j=0}^{n-1}\sigma_{  1+j  }^{( b_{1} )}\dots\sigma_{  l+j  }^{( b_{l} )}
\end{equation}
are distinct for $\boldsymbol{b}\neq0$.  By the orthogonality properties of the Pauli matrices under the Hilbert-Schmidt inner-product, see Section \ref{PauliMatrixBasis}, the pairwise products of these matrices have the property that
\begin{equation}
  \Tr\left(\left(\sigma_{  1+j  }^{( b_{1} )}\dots\sigma_{  l+j  }^{( b_{l} )}\right)\left(\sigma_{  1+j^\prime  }^{( b_{1} )}\dots\sigma_{  l+j^\prime  }^{( b_{l} )}\right)\right)
   =2^{n}\delta_{j,j^\prime}\end{equation}
if $\boldsymbol{b}\neq0$.  Therefore, for $\boldsymbol{b}\neq0$
\begin{equation}\label{3.63}
  \Tr \left(M_{\boldsymbol{b}}M_{\boldsymbol{b}}^\dagger\right)=\Tr\left( M_{\boldsymbol{b}}^{2}\right)=\Tr\left(\left(\frac{1}{\sqrt{n}}\sum_{j=0}^{n-1}\sigma_{  1+j  }^{( b_{1} )}\dots\sigma_{  l+j  }^{( b_{l} )}\right)^{2}\right)=2^{n}
\end{equation}

\subsubsection{The value of $\Tr\left(\rho_{l,k}^{2}\right)$}
It has been seen that
\begin{equation}
  \rho_{l,k}=\sum_{b_{1},\dots,b_{l}=0}^{3} d_{b_{1},\dots,b_{l}}\sigma^{( b_{1} )}\otimes\dots\otimes\sigma^{( b_{l} )}
\end{equation}
where
\begin{equation}
 d_{b_{1},\dots,b_{l}}=\frac{1}{2^{l}}\langle\psi_{k}|\sigma_{  1  }^{( b_{1} )}\dots\sigma_{  l  }^{( b_{l} )}|\psi_{k}\rangle
\end{equation}
Again by the orthogonality properties of the Pauli matrices under the Hilbert-Schmidt inner-product, the pairwise products of the matrices $\sigma^{( b_{1} )}\otimes\dots\otimes\sigma^{( b_{l} )}$ have the property that
\begin{equation}
  \Tr\left(\left(\sigma^{( b_{1} )}\otimes\dots\otimes\sigma^{( b_{l} )}\right)\left(\sigma^{( b^\prime_{1} )}\otimes\dots\otimes\sigma^{( b^\prime_{l} )}\right)\right)
  =2^{l}\delta_{b_{1},b_{1}^\prime}\dots\delta_{b_{l},b_{l}^\prime}\end{equation}
so that
\begin{equation}
  \Tr\left(\rho_{l,k}^{2}\right)=\Tr\left(\left(\sum_{b_{1},\dots,b_{l}=0}^{3} d_{b_{1},\dots,b_{l}}\sigma^{( b_{1} )}\otimes\dots\otimes\sigma^{( b_{l} )}\right)^{2}\right)=2^{l}\sum_{b_{1},\dots,b_{l}=0}^{3} d_{b_{1},\dots,b_{l}}^{2}
\end{equation}
which, on substitution for the $d_{b_{1},\dots,b_{l}}$ gives that
\begin{equation}\label{3.4.22}
  \Tr\left(\rho_{l,k}^{2}\right)=\frac{1}{2^{l}}\sum_{b_{1},\dots,b_{l}=0}^{3} \langle\psi_{k}|\sigma_{  1  }^{( b_{1} )}\dots\sigma_{  l  }^{( b_{l} )}|\psi_{k}\rangle^{2}
\end{equation}
The value of $\sum_{k}\Tr\left(\rho_{l,k}^{2}\right)$ can then be bounded with the results already seen (equations (\ref{3.4.22}), (\ref{3.60}) and (\ref{3.63})), that is
\begin{align}\label{3.69}
  \frac{1}{2^{n}}\sum_{k=1}^{2^{n}}\Tr\left(\rho_{l,k}^{2}\right)&=\frac{1}{2^{n}}\sum_{k=1}^{2^{n}}\frac{1}{2^{l}}\left(\langle\psi_{k}|\psi_{k}\rangle^{2}+\sum_{\genfrac{}{}{0pt}{}{b_{1},\dots,b_{l}=0}{\text{not all zero}}}^{3} \langle\psi_{k}|\sigma_{  1  }^{( b_{1} )}\dots\sigma_{  l  }^{( b_{l} )}|\psi_{k}\rangle^{2}\right)\nonumber\\
  &\leq\frac{1}{2^{n}}\frac{1}{2^{l}}\left(2^{n}+\sum_{\genfrac{}{}{0pt}{}{b_{1},\dots,b_{l}=0}{\text{not all zero}}}^{3} \frac{1}{n}\Tr \left(M_{\boldsymbol{b}}M_{\boldsymbol{b}}^\dagger\right)\right)\nonumber\\  
  &=\frac{1}{2^{n}}\frac{1}{2^{l}}\left(2^{n}+\sum_{\genfrac{}{}{0pt}{}{b_{1},\dots,b_{l}=0}{\text{not all zero}}}^{3} \frac{2^{n}}{n}\right)
\end{align}
where the fact that the $|\psi_{k}\rangle$ are normalised, that is $\langle\psi_{k}|\psi_{k}\rangle=1$, has also been used.  The sum in this last expression contains $4^{l}-1$ terms, so that the expression can be bounded by
\begin{equation}
  \frac{1}{2^{n}}\sum_{k=1}^{2^{n}}\Tr\left(\rho_{l,k}^{2}\right)\leq\frac{1}{2^{l}}+\frac{2^{l}}{n}
\end{equation}
The left hand side of this expression is also lower bounded by $\frac{1}{2^{l}}$ (the minimal value of $\Tr\left(\rho_{l,k}^{2}\right)$, see Section \ref{EigenDef}), which concludes the proof.
\end{proof}

\subsection{Proportion of reduced eigenstates with close to minimal purity}
As the purity of $\rho_{l,k}$ is at least $\frac{1}{2^{l}}$ the next corollary follows immediately:
\begin{corollary}[Proportion of reduced eigenstates with close to minimal purity]
For any fixed $\epsilon>0$ the proportion of joint eigenstates $|\psi_{k}\rangle$ of the translation matrix $T$ and each of the Hamiltonians
\begin{equation}
H_{n}^{(inv,local)}=\sum_{j=1}^{n}\sum_{a=1}^{3}\sum_{b=0}^{3}\alpha_{a,b}\sigma_{j}^{(a)}\sigma_{j+1}^{(b)}
\end{equation}
for $n=2,3,\dots$ individually (where the $\alpha_{a,b}$ are real coefficients for each value of $n$ separately) for which $\Tr\left(\rho_{l,k}^{2}\right)>\frac{1}{2^{l}}+\epsilon$ tends to zero and $n\to\infty$.
\end{corollary}

\begin{proof}
Let $p_{n}\in[0,1]$ be the proportion of eigenstates for which the purity is greater than $\frac{1}{2^{l}}+\epsilon$ for some fixed $\epsilon>0$.  There are then $p_{n}2^{n}$ such eigenstates.  The purity for the remaining $(1-p_{n})2^{n}$ eigenstates is at least $\frac{1}{2^{l}}$ by definition.  Therefore
\begin{equation}
	\frac{1}{2^{n}}\sum_{k=1}^{2^{n}}\Tr\left(\rho_{l,k}^{2}\right)
	\geq\frac{1}{2^{n}}\left(p_{n}2^{n}\left(\frac{1}{2^{l}}+\epsilon\right)+(1-p_{n})2^{n}\frac{1}{2^{l}}\right)
	=\frac{1}{2^{l}}+p_{n}\epsilon
\end{equation}
The preceding theorem then implies that $p_{n}\leq\frac{2^{l}}{n\epsilon}$, which concludes the proof.
\end{proof}

\subsection{Extension to qudits and more general interactions}
The proof of Theorem \ref{invEnt} relies on the fact that $T|\psi_{k}\rangle=\e^{\im\theta_{k}}|\psi_{k}\rangle$ for some $\theta_{k}\in[0,2\pi)$, so that
\begin{equation}
	\langle\psi_{k}|\sigma_{1}^{(b_{1})}\dots\sigma_{l}^{(b_{l})}|\psi_{k}\rangle=\langle\psi_{k}|\sigma_{1+j}^{(b_{1})}\dots\sigma_{l+j}^{(b_{l})}|\psi_{k}\rangle
\end{equation}
where the $\sigma_{1+j}^{(b_{1})}\dots\sigma_{l+j}^{(b_{l})}$ for $j=0,\dots,n-1$ are all orthogonal under the Hilbert-Schmidt inner-product if at least one index out of $b_{1},\dots,b_{l}$ is non-zero. 

This property may be generalised to many other systems with some translational symmetry.  For example, consider a $n\times n$ two dimensional lattice of qudits.   The $d\times d$ Hermitian matrices describing the Hamiltonian of a single qudit have a basis of Hermitian matrices $P_{k}$ for $k=1,\dots,d^{2}$ which are orthogonal under the Hilbert-Schmidt inner-product.  This basis is analogous to the Pauli basis $\sigma^{(a)}$ for $a=0,1,2,3$ for qubits used in the previous calculations.  The $d^{n}\times d^{n}$ Hermitian matrices describing the Hamiltonian of the lattice of $n$ qudits then has a basis formed from the $n$ fold tensor products of the individual qudit bases.

It the Hamiltonian of this qudit lattice has a translational symmetry (see Figure \ref{quditlattice}), analogous results to that in equation (\ref{3.56}), (\ref{3.60}), (\ref{3.63}) and (\ref{3.69}) for the reduced eigenstates, of the Hamiltonian and the relevant translation matrices, on a fixed block of $l$ qudits hold, allowing the previous arguments to be generalised.

\begin{figure}
  \centering
  \input{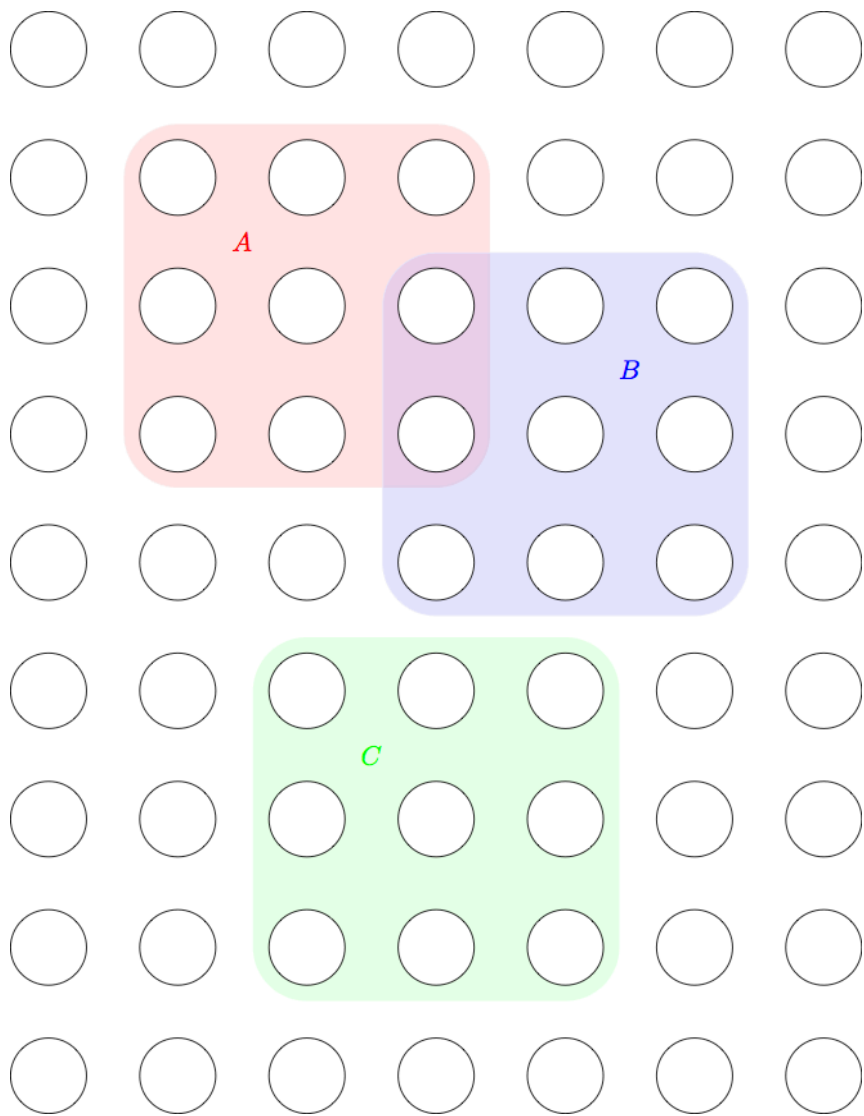}
  \caption[Translational symmetry of a lattice]{A representation of the translational symmetries of a lattice of qudits (represented by circles).  A Hamiltonian of this system that is invariant under horizontal and vertical translation, has eigenstates $|\psi_{k}\rangle$ which are similarly invariant, up to complex scaling.  An operator $X$ supported on the block of nine qubits $A$ can be translated to the block $B$ (or $C$) as $TXT^\dagger$, for some suitable unitary translation matrix $T$, so that $\langle\psi_{k}|X|\psi_{k}\rangle=\langle\psi_{k}|TXT^\dagger|\psi_{k}\rangle$ as $T^\dagger|\psi_{k}\rangle=c|\psi_{k}\rangle$ for some $c\in\mathbb{C}$ with $|c|=1$.}
  \label{quditlattice}
\end{figure}
 
\clearpage
 
\chapter{Eigenstatistics for finite length chains}\label{finite}
Attention will now be turned to the eigenstatistics of finite length qubit spin chains. 

First, the rate of convergence of the spectral densities for the spin chain Hamiltonians studied in Chapter \ref{DOS} will be considered in Section \ref{DOSrate}.  A range of numeric and analytic techniques will be used to conjecture this rate.

The focus of Section \ref{degensec} will then move away from the spectral density and start to consider the statistics of multiple eigenvalues.  To begin with, the occurrence of eigenvalue degeneracies will be looked at in the previous ensembles.  Then in Section \ref{nnnumerics} a numerically analysis of the nearest-neighbour level spacing distributions of these ensembles will be undertaken.  The limiting nearest-neighbour level spacing distributions of the GOE, GUE and GSE are numerically recovered and related to the ensembles' symmetries.

The joint spectral densities for the ensembles of non-translationally-invariant qubit spin chain Hamiltonians are then tackled in Section \ref{JDOS}.  An expression is conjectured for these using the Harish-Chandra-Itzykson-Zuber integral.  Such an expression allows any point correlation function of the eigenvalues to be calculated, in principle.  Close matches to numerical results are seen when this conjectured expression is used to generate the $1$-point and $2$-point correlation functions for the ensemble $\hat{H}_{2}$.

Finally, the tightness of the bounds on the reduced eigenstate purity for specific spin chain Hamiltonians, given in Section \ref{EntLblock}, are analysed in Section \ref{eigenstatebounds}.  In a particular case, these bounds are shown to be asymptotically tight.  Numerical evidence from finite chains suggests that this is also true in other cases too.

The work presented in this chapter is based on that given by the author in \cite{KLW2014_RMT} and \cite{KLW2014_FIXED}.  Ideas initiated from discussions with \cite{Keating} and \cite{Linden} are highlighted where appropriate.
\section{Convergence rate for the spectral density}\label{DOSrate}
The Jordan-Wigner transform \cite{Nielsen2005} allows the diagonalisation of a $2^{n}\times2^{n}$ matrix of the form
\begin{equation}\label{JWtype}
  \sum_{j=1}^{n}\sum_{a,b=1}^{2}\alpha_{a,b,j}\sigma_{j}^{(a)}\sigma_{j+1}^{(b)}+\sum_{j=1}^{n}\alpha_{3,0,j}\sigma_{j}^{(3)}
\end{equation}
where the $\alpha_{a,b,j}$ are real coefficients, to be reduced to the diagonalisation of some $2n\times2n$ matrix.  This process is summarised in \cite{Nielsen2005} and explicitly used in Sections \ref{JW2} and \ref{JW4}.  This reduction in dimension makes it numerically possible to study the spectral statistics of such matrices, which have a somewhat similar form to those investigated in the preceding chapters, for values of $n$ much larger\footnote{The $2^{n}$ eigenvalues of the original matrix can be reconstructed as a sum over the $2n$ eigenvalues of the intermediate matrix.  The limiting factor for the eigenstatistics now becomes the handling of $2^{n}$ eigenvalues.} than accessible from directly diagonalising the matrix above.  Such matrices provide examples of instances from the ensembles $\hat{H}_{n}^{(local)}$ and $\hat{H}_{n}^{(inv,local)}$.

\subsection{Lemma: The spectrum of \texorpdfstring{$H_{n}^{(\epsilon XY+Z)}$}{Hn(eXY+Z)} for prime values of \texorpdfstring{$n$}{n}}\label{JW2}
As a first application, the Jordan-Wigner transform allows the exact analytical diagonalisation of the Hamiltonian
\begin{equation}
  H_{n}^{(\epsilon XY+Z)}=\sum_{j=1}^{n}\left(\epsilon\sigma_{j}^{(1)}\sigma_{j+1}^{(2)}+\sigma_{j}^{(3)}\right)
\end{equation}
for $\epsilon\in\mathbb{R}$.

\begin{lemma}[Spectrum of $H_{n}^{(\epsilon XY+Z)}$]\label{SpecH_{n}^{(}XY+Z)}
  The spectrum of $H_{n}^{(\epsilon XY+Z)}$ for odd prime values of $n$ and $\epsilon\in\mathbb{R}$ is given by
  \begin{equation}\label{HXY+ZSpec}
    \lambda_{\boldsymbol{x}}=\sum_{j=1}^{n}(2x_{j}-1)\left(\epsilon\mu_{j}-\sqrt{\epsilon^{2}\mu_{j}^{2}+1}\right)
  \end{equation}
  for the multi-indices $\boldsymbol{x}=(x_{1},\dots,x_{n})\in\{0,1\}^{n}$.
\end{lemma}
The proof of this lemma characterises the general application of the Jordan-Wigner transform to diagonalising matrices of the form (\ref{JWtype}):

\begin{proof}\hspace{-2mm}\footnote{Adapted from the proof given by the author in \cite{KLW2014_FIXED}.}
The Jordan-Wigner transform defines the Fermi annihilation and creation matrices
\begin{align}
  \a_{j}=\left(\prod_{1\leq l<j}\sigma_{l}^{(3)}\right)S_{j}\qquad\qquad\text{with}\qquad\qquad
  \a_{j}^\dagger=\left(\prod_{1\leq l<j}\sigma_{l}^{(3)}\right)S_{j}^\dagger
\end{align}
where
\begin{equation}
  S_{j}=\frac{\sigma_{  j  }^{(   1   )}+\im\sigma_{  j  }^{(   2   )}}{2}\qquad\qquad\text{with}\qquad\qquad
  S_{j}^\dagger=\frac{\sigma_{  j  }^{(   1   )}-\im\sigma_{  j  }^{(   2   )}}{2}
\end{equation}
Full details of this transform are given in \cite{Nielsen2005} and summarised in Appendix \ref{JWT}.

\subsubsection{Applying the transform}
By the calculations in Appendix \ref{JWnn}, this transformation allows $H_{n}^{(\epsilon XY+Z)}$ to be rewritten as
\begin{equation}\label{transformedXY+Z}
	H_{n}^{(\epsilon XY+Z)}=\im\epsilon\sum_{j=1}^{n-1}\left(\a_{j}-\a_{j}^\dagger\right)\left(\a_{j+1}-\a_{j+1}^\dagger\right)-\im\epsilon\eta\Big(\a_{n}-\a_{n}^\dagger\Big)\left(\a_{1}-\a_{1}^\dagger\right)+\sum_{j=1}^{n}\left(\a_{j}\a_{j}^\dagger-\a_{j}^\dagger \a_{j}\right)
\end{equation}
for the matrix
\begin{equation}
	\eta=\prod_{j=1}^{n}\sigma_{  j  }^{(3)}
\end{equation}

\subsubsection{Block diagonalisation}
The Hamiltonian $H_{n}^{(\epsilon XY+Z)}$ commutes with the matrix $\eta$:  Firstly $\left[\eta,\sigma_{j}^{(1)}\sigma_{j+1}^{(2)}\right]=0$ as both $\sigma^{(1)}$ and $\sigma^{(2)}$ anti-commute with $\sigma^{(3)}$.  This anti-commutation occurs on two separate sites in $\sigma_{j}^{(1)}\sigma_{j+1}^{(2)}$ so that $\sigma_{j}^{(1)}\sigma_{j+1}^{(2)}$ and $\eta$ commute overall.  The local terms $\sigma_{j}^{(3)}$ trivially commute with $\eta$ so that each term in $H_{n}^{(\epsilon XY+Z)}$ commutes with $\eta$.

The standard basis is defined in Section \ref{stndBasis} to be the product basis formed from the eigenstates of $\sigma^{(3)}$ with some fixed phase.  These eigenstates are denoted $|0\rangle$ (with eigenvalue $+1$) and $|1\rangle$ (with eigenvalue $-1$).  The standard product basis consists of the states $|\boldsymbol{x}\rangle=|x_{1}\rangle\otimes\dots\otimes|x_{n}\rangle$ for the multi-indices $\boldsymbol{x}=(x_{1},\dots,x_{n})\in\{0,1\}^{n}$.

In this basis the matrix $\eta$ has the form
\begin{equation}\label{JWform}
  \eta=\left(|0\rangle\langle0|-|1\rangle\langle1|\right)^{\otimes n}=\sum_{\boldsymbol{x}}\pm|\boldsymbol{x}\rangle\langle\boldsymbol{x}|
\end{equation}
From this expression it is seen that $\eta$ is diagonal in the standard basis and has just two eigenvalues, $\pm1$.  As $\eta$ commutes with $H_{n}^{(\epsilon XY+Z)}$, the Hamiltonian $H_{n}^{(\epsilon XY+Z)}$ must be block diagonal in the standard basis with two blocks corresponding to the two eigenvalues of $\eta$.  This follows since for any basis element $|\boldsymbol{x}\rangle$ with $\eta|\boldsymbol{x}\rangle=\pm|\boldsymbol{x}\rangle$, $H_{n}^{(\epsilon XY+Z)}|\boldsymbol{x}\rangle$ is an eigenstate of $\eta$ as
\begin{equation}
  \eta\left(H_{n}^{(\epsilon XY+Z)}|\boldsymbol{x}\rangle\right)=H_{n}^{(\epsilon XY+Z)}\eta|\boldsymbol{x}\rangle=\pm \left(H_{n}^{(\epsilon XY+Z)}|\boldsymbol{x}\rangle\right)
\end{equation}
Then, for any basis elements $|\boldsymbol{x}\rangle$ and $|\boldsymbol{y}\rangle$ such that $\eta|\boldsymbol{x}\rangle=\pm|\boldsymbol{x}\rangle$ and $\eta|\boldsymbol{y}\rangle=\mp|\boldsymbol{y}\rangle$
\begin{equation}
  \langle\boldsymbol{x}|H_{n}^{(\epsilon XY+Z)}|\boldsymbol{y}\rangle=\langle\boldsymbol{x}|\left(|H_{n}^{(\epsilon XY+Z)}|\boldsymbol{y}\rangle\right)=0
\end{equation}
as $|\boldsymbol{x}\rangle$ and $H_{n}^{(\epsilon XY+Z)}|\boldsymbol{y}\rangle$ are eigenstates of $\eta$ in different eigenspaces of $\eta$, which are necessary orthogonal as $\eta$ is Hermitian.

Each block will now be considered separately.

\subsubsection{The $\eta=-1$ block}
For states in the $\eta=-1$ subspace, $H_{n}^{(\epsilon XY+Z)}$ acts as
\begin{equation}
	H_{n}^{-}=\im\epsilon\sum_{j=1}^{n}\left(\a_{j}-\a_{j}^\dagger\right)\left(\a_{j+1}-\a_{j+1}^\dagger\right)+\sum_{j=1}^{n}\left(\a_{j}\a_{j}^\dagger-\a_{j}^\dagger \a_{j}\right)
\end{equation}
with the periodic boundary conditions $\a_{n+1}=\a_{1}$ imposed, as $\eta$ in (\ref{transformedXY+Z}) always results in the factor $-1$ for any such state.  On the $\eta=+1$ subspace this matrix acts differently to $H_{n}^{(\epsilon XY+Z)}$.

From Appendix \ref{JWnn}, $\left(\a_{j}-\a_{j}^\dagger\right)\left(\a_{j+1}-\a_{j+1}^\dagger\right)$ is proportional to $\sigma_{j}^{(1)}\sigma_{j+1}^{(2)}$ up to a factor $\eta$ and $\left(\a_{j}\a_{j}^\dagger-\a_{j}^\dagger \a_{j}\right)$ is proportional to $\sigma_{j}^{(3)}$.  Therefore, by the same argument for the case of $H_{n}^{(\epsilon XY+Z)}$, $H_{n}^{-}$ commutes with $\eta$ and is therefore block diagonal in the standard basis.

The canonical commutation relations for Fermi matrices, $\a_{j}\a_{k}^\dagger=-\a_{k}^\dagger \a_{j}+\delta_{j,k}I$ and $\a_{j}\a_{k}=-\a_{k}\a_{j}$, can be used to rewrite the expression for $H_{n}^{-}$ above as
\begin{align}
	H_{n}^{-}&=\frac{\im\epsilon}{2}\sum_{j=1}^{n}\left(\a_{j}-\a_{j}^\dagger\right)\left(\a_{j+1}-\a_{j+1}^\dagger\right)-\frac{\im\epsilon}{2}\sum_{j=1}^{n}\left(\a_{j+1}-\a_{j+1}^\dagger\right)\left(\a_{j}-\a_{j}^\dagger\right)\nonumber\\
  &\qquad\qquad\qquad\qquad+\sum_{j=1}^{n}\left(\a_{j}\a_{j}^\dagger-\a_{j}^\dagger \a_{j}\right)
\end{align}
or equivalently in matrix notation
\begin{equation}
	H_{n}^{-}=
	\begin{pmatrix}
		\boldsymbol{\a}^\dagger&
		{\boldsymbol{\a}^\prime}^\dagger
	\end{pmatrix}
	\begin{pmatrix}
		A-I&-A\\
		-A&A+I
	\end{pmatrix}
	\begin{pmatrix}
		\boldsymbol{\a}\\
		\boldsymbol{\a}^\prime
	\end{pmatrix}
\end{equation}
where $\boldsymbol{\a}$ is the column vector with entries $\a_{1},\dots,\a_{n}$, $\boldsymbol{\a}^\prime$ is the column vector with entries $\a_{1}^\dagger,\dots,\a_{n}^\dagger$, $\boldsymbol{\a}^\dagger$ is the row vector with entries $\a_{1}^\dagger,\dots,\a_{n}^\dagger$, ${\boldsymbol{\a}^\prime}^\dagger$ is the row vector with entries $\a_{1},\dots,\a_{n}$ and
\begin{equation}
	A=\frac{\im\epsilon}{2}
	\begin{pmatrix}
		0&-1&0&\cdots&0&1\\
	  1&0&-1&\ddots&&0\\
	  0&1&\ddots&\ddots&&\vdots\\
	  \vdots&\ddots&\ddots&&&0\\
	  0&&&&0&-1\\
	  -1&0&\cdots&0&1&0
	\end{pmatrix}
\end{equation}

\subsubsection{Initial diagonalisation}
The matrix $A$ is proportional to a circulant matrix and by \cite[\p 388]{Bernstein2009} has the eigenstates
\begin{equation}
  \frac{1}{\sqrt{n}}\left(\omega_{1}^{k},\omega_{2}^{k},\dots,\omega_{n}^{k}\right)^{T}
\end{equation}
(where $\cdot^T$ represents the vector's transpose) for $k=1,\dots,n$ where $\omega_{j}=\e^{\frac{2\pi\im j}{n}}$, with the associated eigenvalues
\begin{equation}
  \epsilon\sin\left(\frac{2\pi k}{n}\right)
\end{equation}
Therefore $U^\dagger AU=D$ where $U$ is the unitary matrix with the eigenstates of $A$ as columns and $D$ is the diagonal matrix with the associated eigenvalues of $A$ along its diagonal.  Specifically
\begin{equation}
  U_{j,k}=\frac{1}{\sqrt{n}}\omega_{j}^{k}\qquad\qquad D_{j,k}=\delta_{j,k}\epsilon\sin\left(\frac{2\pi k}{n}\right)=\delta_{j,k}\epsilon\sin\left(\frac{2\pi j}{n}\right)
\end{equation}

As shown in Appendix \ref{FermiTrans}, the Fermi matrices $\b_{j}$ and $\b_{j}^\dagger$ can be defined through
\begin{equation}
	\begin{pmatrix}
		\boldsymbol{\a}\\
		\boldsymbol{\a}^\prime
	\end{pmatrix}
	=
	\begin{pmatrix}
		U&0\\
		0&\overline{U}
	\end{pmatrix}
	\begin{pmatrix}
		\boldsymbol{\b}\\
		\boldsymbol{\b}^\prime
	\end{pmatrix}
\end{equation}
Here $\overline{U}$ represents the unitary matrix whose elements $\overline{U}_{j,k}$ are the complex conjugate of those of $U$, that is, $\overline{U}_{j,k}=\frac{1}{\sqrt{n}}\omega_{j}^{-k}$.  This transformation leaves $H_{n}^{-}$ as
\begin{align}
	H_{n}^{-}&=\begin{pmatrix}
		\boldsymbol{\b}^\dagger&
		{\boldsymbol{\b}^\prime}^\dagger
	\end{pmatrix}
	\begin{pmatrix}
		U^\dagger&0\\0&\overline{U}^\dagger
	\end{pmatrix}
	\begin{pmatrix}
		A-I&-A \\
		-A&A+I
	\end{pmatrix}
	\begin{pmatrix}
		U&0\\0&\overline{U}
	\end{pmatrix}
	\begin{pmatrix}
		\boldsymbol{\b}\\
		\boldsymbol{\b}^\prime
	\end{pmatrix}
	\nonumber\\
	&=\begin{pmatrix}
		\boldsymbol{\b}^\dagger&
		{\boldsymbol{\b}^\prime}^\dagger
	\end{pmatrix}
	\begin{pmatrix}
		U^\dagger AU-I&-U^\dagger A\overline{U} \\
		\overline{U^\dagger A\overline{U}}&-\overline{U^\dagger AU}+I
	\end{pmatrix}
	\begin{pmatrix}
		\boldsymbol{\b}\\
		\boldsymbol{\b}^\prime
	\end{pmatrix}
\end{align}

Let $P_{j,k}=\delta_{j,\kappa(k)}$ with $\kappa(k)=n-k$ for $k=1,\dots,n-1$ and  $\kappa(n)=n$.  Note that $\kappa$ is its own inverse, it reverses the order of the labels $1,\dots,n-1$ and has no effect on the label $n$. In this case
\begin{equation}
  (UP)_{j,k}
  =\sum_{l=1}^{n}U_{j,l}P_{l,k}
  =\frac{1}{\sqrt{n}}\sum_{l=1}^{n}\omega_{j}^{l}\delta_{l,\kappa(k)}
  =\frac{1}{\sqrt{n}}\begin{cases}
                      \omega_{j}^{n-k}\quad&\text{if }k\neq n\\
		      \omega_{j}^{n}&\text{if }k=n
                     \end{cases}
  \,\,=\overline{U}_{j,k}
\end{equation}
so that $U^\dagger A\overline{U}=U^\dagger AUP=DP$.  The Hamiltonian $H_{n}^{-}$ therefore reads
\begin{equation}
	H_{n}^{-}=\begin{pmatrix}
		\boldsymbol{\b}^\dagger&
		{\boldsymbol{\b}^\prime}^\dagger
	\end{pmatrix}
	\begin{pmatrix}
		D-I&-DP \\
		DP&-D+I
	\end{pmatrix}
	\begin{pmatrix}
		\boldsymbol{\b}\\
		\boldsymbol{\b}^\prime
	\end{pmatrix}
\end{equation}
as $DP$ and $D$ both have real entries.  With $\mu_{j}=\sin\left(\frac{2\pi j}{n}\right)$ (where it is noted that $\mu_{j}=-\mu_{\kappa(j)}$), this matrix equation may be multiplied out to leave
\begin{equation}
  \sum_{j,k=1}^{n}\left(\b_{j}^\dagger(\epsilon\mu_{j}-1)\delta_{j,k}\b_{k}+\b_{j}^\dagger(-\epsilon\mu_{j})\delta_{j,\kappa(k)}\b^\dagger_{k}+\b_{j}(\epsilon\mu_{j})\delta_{j,\kappa(k)}\b_{k}+\b_{j}(-\epsilon\mu_{j}+1)\delta_{j,k}\b^\dagger_{k}\right)
\end{equation}
The summation over $k$ may be performed by evaluating of the Kronecker of delta symbols and relabelling the indices in the third and fourth terms in the summand by $j\to\kappa(j)$.  This gives the equivalent expression of
\begin{equation}
  \sum_{j=1}^{n}\left(\b_{j}^\dagger(\epsilon\mu_{j}-1)\b_{j}+\b_{j}^\dagger(-\epsilon\mu_{j})\b^\dagger_{\kappa(j)}+\b_{\kappa(j)}(\epsilon\mu_{\kappa(j)})\b_{j}+\b_{\kappa(j)}(-\epsilon\mu_{\kappa(j)}+1)\b^\dagger_{\kappa(j)}\right)
\end{equation}
or in matrix form, of
\begin{equation}
	H_{n}^{-}=\sum_{j=1}^{n}
	\begin{pmatrix}
		\b_{j}^\dagger&
		\b_{\kappa(j)}
	\end{pmatrix}
	\begin{pmatrix}
		\epsilon\mu_{j}-1&-\epsilon\mu_{j}\\
		\epsilon\mu_{\kappa(j)}&-\epsilon\mu_{\kappa(j)}+1
	\end{pmatrix}
	\begin{pmatrix}
		\b_{j}\\
		\b_{\kappa(j)}^\dagger
	\end{pmatrix}
\end{equation}

\subsubsection{Second diagonalisation (Bogoliubov transformation)}
By the symmetries of the sine function, $\mu_{j}=-\mu_{\kappa(j)}$, so that
\begin{equation}\label{2x2Bogoliubov}
	\begin{pmatrix}
		\epsilon\mu_{j}-1&-\epsilon\mu_{j}\\
		\epsilon\mu_{\kappa(j)}&-\epsilon\mu_{\kappa(j)}+1
	\end{pmatrix}=
	\begin{pmatrix}
		\epsilon\mu_{j}-1&-\epsilon\mu_{j}\\
		-\epsilon\mu_{j}&\epsilon\mu_{j}+1
	\end{pmatrix}
\end{equation}
The two eigenvalues of this matrix are given by the two roots of the characteristic equation
\begin{align}
  p_{j}(\chi_{j})=\det\begin{pmatrix}
		\epsilon\mu_{j}-1-\chi_{j}&-\epsilon\mu_{j}\\
		-\epsilon\mu_{j}&\epsilon\mu_{j}+1-\chi_{j}
	\end{pmatrix}
  &=(\epsilon\mu_{j}-\chi_{j}-1)(\epsilon\mu_{j}-\chi_{j}+1)-\epsilon^{2}\mu_{j}^{2}\nonumber\\
  &=(\epsilon\mu_{j}-\chi_{j})^{2}-1-\epsilon^{2}\mu_{j}^{2}
\end{align}
as $\chi_{j}^{\pm}=\epsilon\mu_{j}\pm\sqrt{\epsilon^{2}\mu_{j}^{2}+1}$.  Then, as the $2\times2$ matrix (\ref{2x2Bogoliubov}) is real and symmetric there exists some $\theta_{j}\in[0,2\pi)$ so that the orthogonal matrix
\begin{equation}
  U_{j}=\begin{pmatrix}\cos(\theta_{j})&\sin(\theta_{j})\\-\sin(\theta_{j})&\cos(\theta_{j})\end{pmatrix}
\end{equation}
diagonalises (\ref{2x2Bogoliubov}) as
\begin{equation}
  {U_{j}}^\dagger
  \begin{pmatrix}
    \epsilon\mu_{j}-1&-\epsilon\mu_{j}\\
    -\epsilon\mu_{j}&\epsilon\mu_{j}+1
  \end{pmatrix}U_{j}=
  \begin{pmatrix}
    \chi_{j}^{-}&0\\0&\chi_{j}^{+}
  \end{pmatrix}
\end{equation}
The four components of this equation impose the following four conditions on the value $\theta_{j}$
\begin{align}
  \chi_{j}^{-}&=2\epsilon\mu_{j}\sin(\theta_{j})\cos(\theta_{j})-2\cos^{2}(\theta_{j})+\epsilon\mu_{j}+1\nonumber\\
  \chi_{j}^{+}&=-2\epsilon\mu_{j}\sin(\theta_{j})\cos(\theta_{j})+2\cos^{2}(\theta_{j})+\epsilon\mu_{j}-1\nonumber\\
  0&=-2\sin(\theta_{j})\cos(\theta_{j})-2\epsilon\mu_{j}\cos^{2}(\theta_{j})+\epsilon\mu_{j}\nonumber\\
  0&=-2\sin(\theta_{j})\cos(\theta_{j})-2\epsilon\mu_{j}\cos^{2}(\theta_{j})+\epsilon\mu_{j}
\end{align}
For the values of $\chi_{j}^{\pm}$ given above, the first two and last two of these conditions are equivalent, they read
\begin{align}\label{4.1.30}
  -\sqrt{\epsilon^{2}\mu_{j}^{2}+1}&=2\epsilon\mu_{j}\sin(\theta_{j})\cos(\theta_{j})-2\cos^{2}(\theta_{j})+1\nonumber\\
  0&=-2\sin(\theta_{j})\cos(\theta_{j})-2\epsilon\mu_{j}\cos^{2}(\theta_{j})+\epsilon\mu_{j}
\end{align}
Note that as $\mu_{j}=-\mu_{\kappa(j)}$ the values of the $\theta_{j}$ for $j=1,\dots,n$ can aways be chosen to satisfy the conditions (\ref{4.1.30}) and the extra conditions $\theta_{j}=-\theta_{\kappa(j)}$.  This is seen by choosing $\theta_{n}=0$ and any $\theta_{j}$ for $j=1,\dots,\frac{n-1}{2}$ that satisfy (\ref{4.1.30}) and then defining $\theta_{j}$ for $j=\frac{n+1}{2},\dots,n-1$ via $\theta_{j}=-\theta_{\kappa(j)}$.  Then by the symmetries $\mu_{j}=-\mu_{\kappa(j)}$ and $\theta_{j}=-\theta_{\kappa(j)}$, the conditions in (\ref{4.1.30}) will be satisfied for $j=\frac{n+1}{2},\dots,n-1$ (that is $j=\kappa\left(\frac{n-1}{2}\right),\dots,\kappa(1)$).

As shown in Appendix \ref{FermiTrans}, new Fermi matrices $\c_{j}$ and $\c_{j}^\dagger$ can be defined via
\begin{equation}\label{cFermi}
	\begin{pmatrix}
		\boldsymbol{\b}\\
		\boldsymbol{\b}^\prime
	\end{pmatrix}
	=
	\begin{pmatrix}
		V&W\\
		\overline{W}&\overline{V}
	\end{pmatrix}
	\begin{pmatrix}
		\boldsymbol{\c}\\
		\boldsymbol{\c}^\prime
	\end{pmatrix}
\end{equation}
where $V$ and $W$ are $n\times n$ matrices such that $VV^\dagger+WW^\dagger=I$ and $VW^{T}+WV^{T}=0$.  Let $V_{j,k}=\delta_{j,k}\cos(\theta_{j})$ and $W_{j,k}=\delta_{j,\kappa(k)}\sin(\theta_{j})$ so that these conditions hold, that is using the symmetry $\theta_{j}=-\theta_{\kappa(j)}$
\begin{align}
  \left(VV^\dagger+WW^\dagger\right)_{j,k}&=\sum_{l=1}^{n}\left(\delta_{j,l}\delta_{k,l}\cos(\theta_{j})\cos(\theta_{k})+\delta_{j,\kappa(l)}\delta_{k,\kappa(l)}\sin(\theta_{j})\sin(\theta_{k})\right)=\delta_{j,k}\nonumber\\
  \left(VW^{T}+WV^{T}\right)_{j,k}&=\sum_{l=1}^{n}\left(\delta_{j,l}\delta_{k,\kappa(l)}\cos(\theta_{j})\sin(\theta_{k})+\delta_{j,\kappa(l)}\delta_{k,l}\sin(\theta_{j})\cos(\theta_{k})\right)=0
\end{align}
Moreover, equation (\ref{cFermi}) and the symmetry $\theta_{j}=-\theta_{\kappa(j)}$ give that
\begin{equation}
	\begin{pmatrix}
		\b_{j}\\
		\b_{\kappa(j)}^\dagger
	\end{pmatrix}
	=\begin{pmatrix}
		\sum_{k=1}^{n}\left(V_{j,k}\c_{k}+W_{j,k}\c_{k}^\dagger\right)\\
		\sum_{k=1}^{n}\left(\overline{W_{\kappa(j),k}}\c_{k}+\overline{V_{\kappa(j),k}}\c_{k}^\dagger\right)
	\end{pmatrix}
	=\begin{pmatrix}\cos(\theta_{j})&\sin(\theta_{j})\\-\sin(\theta_{j})&\cos(\theta_{j})\end{pmatrix}
	\begin{pmatrix}
		\c_{j}\\
		\c_{\kappa(j)}^\dagger
	\end{pmatrix}
\end{equation}
so that
\begin{equation}
	H_{n}^{-}=\sum_{j=1}^{n}
	\begin{pmatrix}
		\c_{j}^\dagger&
		\c_{\kappa(j)}
	\end{pmatrix}
	\begin{pmatrix}
		\chi_{j}^{-}&0\\0&\chi_{j}^{+}
	\end{pmatrix}
	\begin{pmatrix}
		\c_{j}\\
		\c_{\kappa(j)}^\dagger
	\end{pmatrix}
	=\sum_{j=1}^{n}\left(\chi_{j}^{-}\c_{j}^\dagger \c_{j}+\chi_{\kappa(j)}^{+}\c_{j}\c_{j}^\dagger\right)
\end{equation}
where the indices in the last term of the summand have been relabelled with $j\to\kappa(j)$.

\subsubsection{Extracting the eigenvalues of $H_{n}^{-}$}
As seen in Appendix \ref{FermiBasis}, the Hilbert space admits the orthonormal Fermi basis
\begin{equation}
	|\boldsymbol{x}\rangle_\c
	=\Big(\c_{1}^\dagger\Big)^{x_{1}}\dots\Big(\c_{n}^\dagger\Big)^{x_{n}}|\boldsymbol{0}\rangle_\c
\end{equation}
for the multi-indices $\boldsymbol{x}=(x_{1},\dots,x_{n})\in\{0,1\}^{n}$, where $|\boldsymbol{0}\rangle_\c$ is a normalised state such that $\c_{j}|\boldsymbol{0}\rangle_\c=0$ for all $j$ (see \cite[\p3]{Nielsen2005} for the existence of such a state).  It is also seen in Appendix \ref{JWFermiAction} that $\c_{j}^\dagger \c_{j}|\boldsymbol{x}\rangle_\c=x_{j}|\boldsymbol{x}\rangle_\c$ and $\c_{j}\c_{j}^\dagger|\boldsymbol{x}\rangle_\c=(1-x_{j})|\boldsymbol{x}\rangle_\c$.

The eigenvalues of $H_{n}^{-}$ can then be read off from the eigenvalue equation
\begin{equation}
  H_{n}^{-}|\boldsymbol{x}\rangle_\c
  =\sum_{j=1}^{n}\left(\chi_{j}^{-}\c_{j}^\dagger \c_{j}+\chi_{\kappa(j)}^{+}\c_{j}\c_{j}^\dagger\right)|\boldsymbol{x}\rangle_\c
  =\sum_{j=1}^{n}\left(\chi_{j}^{-}x_{j}+\chi_{\kappa(j)}^{+}(1-x_{j})\right)|\boldsymbol{x}\rangle_\c
\end{equation}
Substituting $\mu_{j}=-\mu_{\kappa(j)}$ into the expressions for $\chi_{j}^{\pm}$ simplifies the eigenvalues to be
\begin{equation}
  \lambda_{\boldsymbol{x}}=\sum_{j=1}^{n}(2x_{j}-1)\left(\epsilon\mu_{j}-\sqrt{\epsilon^{2}\mu_{j}^{2}+1}\right)
\end{equation}

\subsubsection{Relationship to the eigenvalues of $H_{n}^{(\epsilon XY+Z)}$}
Eigenstates of $H_{n}^{-}$ in the $\eta=-1$ subspace are also eigenstates of $H_{n}^{(\epsilon XY+Z)}$ by construction, as both matrices act identically on the $\eta=-1$ subspace.  The eigensubspace of $\eta$ to which the eigenstates $|\boldsymbol{x}\rangle_\c$ of $H_{n}^{-}$ belong must now be determined.

As $H_{n}^{-}$ commutes with $\eta$, and $|\boldsymbol{x}\rangle_\c$ are the eigenstates of $H_{n}^{-}$ with eigenvalues $\lambda_{\boldsymbol{x}}$,
\begin{equation}
  H_{n}^{-}\left(\eta|\boldsymbol{x}\rangle_\c\right)=\eta H_{n}^{-}|\boldsymbol{x}\rangle_\c=\lambda_{\boldsymbol{x}}\left(\eta|\boldsymbol{x}\rangle_\c\right)
\end{equation}
so that both $|\boldsymbol{x}\rangle_\c$ and $\eta|\boldsymbol{x}\rangle_\c$ are eigenstates of $H_{n}^{-}$ with eigenvalue $\lambda_{\boldsymbol{x}}$.  It will be seen in the proof of Lemma \ref{lemmaH_n^(eXY+Z)} of Section \ref{invLocalSpec} (page \pageref{lemmaH_n^(eXY+Z)}) that the values $\lambda_{\boldsymbol{x}}$ are distinct for most (apart from a set of zero Lebesgue measure) values of $\epsilon\in\mathbb{R}$ in the case that $n$ is an odd prime.  For such values it must therefore be the case that $\eta|\boldsymbol{x}\rangle_{c}=c_{\boldsymbol{x}}|\boldsymbol{x}\rangle_{c}$ for some $c_{\boldsymbol{x}}\in\mathbb{C}$.  The eigenvalues of $\eta$ are $\pm1$ so that $c_{\boldsymbol{x}}=\pm1$.

To determine $c_{\boldsymbol{x}}$, the $\c_{j}^\dagger$ can be expressed in terms of the Pauli matrices by inverting the previous transforms.  Let the unitary matrix $U^\dagger$ have entries $u_{j,k}\in\mathbb{C}$.  Then, by definition,
\begin{align}
  \c_{j}&=\cos(\theta_{j})\b_{j}-\sin(\theta_{j})\b_{\kappa(j)}^\dagger\nonumber\\
  &=\cos(\theta_{j})\sum_{k=1}^{n}u_{j,k}\a_{k}-\sin(\theta_{j})\sum_{k=1}^{n}\overline{u_{\kappa(j),k}}\a_{k}^\dagger\nonumber\\
  &=\cos(\theta_{j})\sum_{k=1}^{n}\left(u_{j,k}\left(\prod_{1\leq l<k}\sigma_{l}^{(3)}\right)S_{k}\right)-\sin(\theta_{j})\sum_{k=1}^{n}\left(\overline{u_{\kappa(j),k}}\left(\prod_{1\leq l<k}\sigma_{l}^{(3)}\right)S_{k}^\dagger\right)
\end{align}
As $\sigma^{(1)}$ and $\sigma^{(2)}$ anti-commute with $\sigma^{(3)}$, $\eta=\prod_{j=1}^{n}\sigma_{j}^{(3)}$ must anti-commute with both $2S_{k}=\sigma_{k}^{(1)}+\im\sigma_{k}^{(2)}$ and $2S_{k}^\dagger=\sigma_{k}^{(1)}-\im\sigma_{k}^{(2)}$.  Also $\prod_{1\leq l<k}\sigma_{l}^{(3)}$ trivially commutes with $\eta$, so that $\c_{j}$ and $\eta$ must anti-commute.

Therefore
\begin{equation}
  \eta|\boldsymbol{x}\rangle_\c=(-1)^{\s}\big(\c_{1}^\dagger\big)^{x_{1}}\dots\big(\c_{n}^\dagger\big)^{x_{n}}\eta|\boldsymbol{0}\rangle_\c
\end{equation}
where $\s=\sum_{j}x_{j}$.  Now as $\eta|\boldsymbol{0}\rangle_\c=\pm|\boldsymbol{0}\rangle_\c$, either all the states $|\boldsymbol{x}\rangle_\c$ for which $\s$ is even (if $\eta|\boldsymbol{0}\rangle_\c=-|\boldsymbol{0}\rangle_\c$) or odd (if $\eta|\boldsymbol{0}\rangle_\c=+|\boldsymbol{0}\rangle_\c$) are eigenstates of $H_{n}^{(\epsilon XY+Z)}$.

\subsubsection*{The $\eta=1$ block}
The Hamiltonian $H_{n}^{(\epsilon XY+Z)}$ anti-commutes with the matrix $\nu=\sum_{j=1}^{n}\sigma_{j}^{(1)}$:  Firstly $\sigma_{j}^{(1)}\sigma_{j+1}^{(2)}$ anti-commutes with $\nu$ as $\sigma^{(2)}$ anti-commutes with $\sigma^{(1)}$.  The local terms $\sigma_{j}^{(3)}$ anti-commute with $\nu$ as $\sigma^{(3)}$ anti-commutes with $\sigma^{(1)}$ so that each term in $H_{n}^{(\epsilon XY+Z)}$ anti-commutes with $\nu$.  Any eigenstate $|\psi\rangle$ of $H_{n}^{(\epsilon XY+Z)}$ in the $\eta=-1$ subspace with eigenvalue $\lambda$ must then satisfy
\begin{equation}
	H_{n}^{(\epsilon XY+Z)}\left(\nu|\psi\rangle\right)=-\nu H_{n}^{(\epsilon XY+Z)}|\psi\rangle=-\lambda\left(\nu|\psi\rangle\right)
\end{equation}
so that $\nu|\psi\rangle$ is an eigenstates of $H_{n}^{(\epsilon XY+Z)}$ with eigenvalue $-\lambda$.

As $n$ is odd, $\eta$ and $\nu$ must also anti-commute, as $\sigma^{(1)}$ and $\sigma^{(3)}$ anti-commute and this occurs on an odd number of sites.  Therefore
\begin{equation}
 \eta\left(\nu|\psi\rangle\right)=-\nu\eta|\psi\rangle=+\left(\nu|\psi\rangle\right)
\end{equation}
and the eigenstate $\nu|\psi\rangle$ must be in the $\eta=1$ eigenspace.

Given an eigenvalue $\lambda_{\boldsymbol{x}}$ its negative is given by $\lambda_{\boldsymbol{x}^{c}}$ where $\boldsymbol{x}^{c}$ is the complementary vector to $\boldsymbol{x}$ (that is the vector in which the entries of zero and one have been inverted).  This inversion ($0\leftrightarrow1$) changes the sign of each of the values $2x_{j}-1$ so that
\begin{equation}
  -\lambda_{\boldsymbol{x}}
  =-\sum_{j=1}^{n}(2x_{j}-1)\left(\epsilon\mu_{j}-\sqrt{\epsilon^{2}\mu_{j}^{2}+1}\right)
  =\lambda_{\boldsymbol{x}^{c}}
\end{equation}
This inversion also changes the parity of $\s=\sum_{j}x_{j}$ for odd values of $n$, so that the spectrum in the $\eta=1$ subspace is given by the values $\lambda_{\boldsymbol{x}}$ for which $\s=\sum_{j}x_{j}$ has the opposite parity as taken for the $\eta=-1$ block.

\subsubsection{The complete spectrum}
The entire spectrum for odd prime values of $n$ and most, apart from a set of zero Lebesgue measure, $\epsilon\in\mathbb{R}$ is therefore given by
\begin{equation}
	\lambda_{\boldsymbol{x}}=\sum_{j=1}^{n}(2x_{j}-1)\left(\epsilon\mu_{j}-\sqrt{\epsilon^{2}\mu_{j}^{2}+1}\right)
\end{equation}
for the multi-indices $\boldsymbol{x}=(x_{1},\dots,x_{n})\in\{0,1\}^{n}$.  This results then holds for all $\epsilon$ by the analyticity of the eigenvalues \cite{Wimmer1986}.
\end{proof}

\subsection{Numerical convergence of the spectral density of \texorpdfstring{$H_{n}^{(\epsilon XY+Z)}$}{Hn(eXY+Z)} in \texorpdfstring{$O(n^{-1})$}{O(1/n)}}\label{JW3}
Given the analytic expression (\ref{HXY+ZSpec}) of the spectrum for the Hamiltonians $H_{n}^{(\epsilon XY+Z)}$ for odd prime values of $n$, the rate of convergence of their spectral densities to their limiting distribution (as $n\to\infty$) can be at least numerically calculated.

In order that Theorem \ref{DOSmembers} of Section \ref{Splitting method} (page \pageref{DOSmembers}) applies, each Hamiltonian $H_{n}^{(\epsilon XY+Z)}$ will be scaled by a constant $C_{n}$ such that
\begin{equation}
  1=\frac{1}{2^{n}}\Tr\left(C_{n}^{2}{H_{n}^{(\epsilon XY+Z)}}^{2}\right)=C_{n}^{2}n\left(\epsilon^{2}+1^{2}\right)
\end{equation}
The constant $C_{n}$ is then set to be
\begin{equation}
  C_{n}=\frac{1}{\sqrt{n\left(\epsilon^{2}+1^{2}\right)}}
\end{equation}
The conditions of Theorem \ref{DOSmembers} are now satisfied for any fixed value of $\epsilon\in\mathbb{R}$ for the matrices \begin{equation}
   C_{n}H_{n}^{(\epsilon XY+Z)}=\sum_{j=1}^{n}\left(\frac{\epsilon}{\sqrt{n(\epsilon^{2}+1^{2})}}\sigma_{j}^{(1)}\sigma_{j+1}^{(2)}+\frac{1}{\sqrt{n(\epsilon^{2}+1^{2})}}\sigma_{j}^{(3)}\right)
\end{equation}
The value of $\epsilon$ will now arbitrarily be taken to be unity in order to compute numerical trends.

To quantify the rate of convergence of the spectral densities for this sequence of Hamiltonians, to the standard normal distribution, the absolute error
\begin{equation}
  E_{n}(x)=\left|\frac{1}{2^{n}}\sum_{k=1}^{2^{n}}\int_{-\infty}^{x^{+}}\delta(\lambda-\lambda_{k})\di\lambda-\frac{1}{\sqrt{2\pi}}\int_{-\infty}^{x}\e^{-\frac{\lambda^{2}}{2}}\di\lambda\right|
\end{equation}
for $x\in\mathbb{R}$ will be used.  The notation $x^{+}$ denotes the limit as $x$ is approached from above and the $\lambda_{k}$ are the $2^{n}$ eigenvalues of $C_{n}H_{n}^{(XY+Z)}$ for each $n$ individually.  The quantity $E_{n}(x)$ measures the absolute difference between the proportion of eigenvalues of $C_{n}H_{n}^{(XY+Z)}$ equal to or below the real value $x$, to the cumulative distribution function of a standard normal distribution.

To compute the function $E_{n}(x)$, numerical methods had to be resorted to.  Figure \ref{XY+ZDOSconverge} shows the numerical values of $E_{n}^{-1}(x)$ against odd prime values of $n\leq32$ for various values of $x$.  The values for other intermediate values of $n$ are also plotted where $E_{n}(x)$ is calculated using the formula (\ref{HXY+ZSpec}), even though it is not necessarily valid, to show the numeric trend in the values.

\begin{figure}
  \centering
  \input{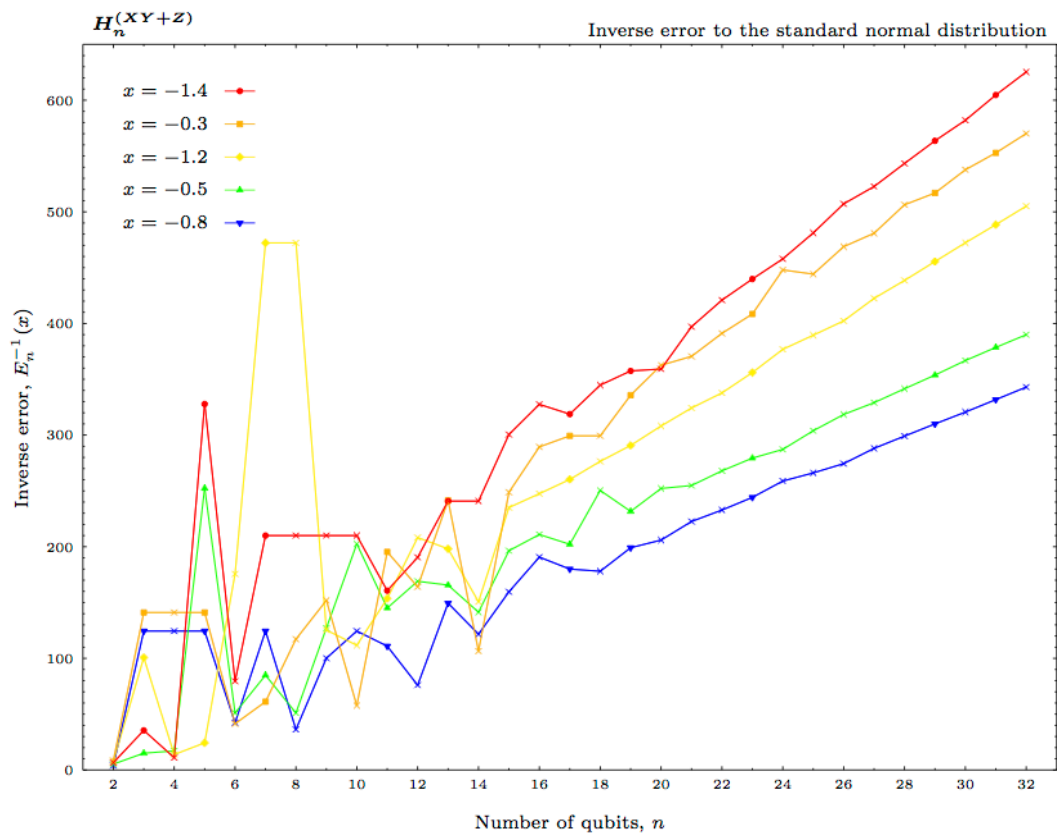}
  \caption[Convergence of the spectral density for $H_{n}^{(\epsilon XY+Z)}$]{The values of $E^{-1}_{n}(x)=\left|\frac{1}{2^{n}}\sum_{k}\int_{-\infty}^{x^{+}}\delta(\lambda-\lambda_{k})\di\lambda-\frac{1}{\sqrt{2\pi}}\int_{-\infty}^{x}\e^{-\frac{\lambda^{2}}{2}}\di\lambda\right|^{-1}$ against chain length $n$, for the eigenvalues $\lambda_{k}$ ($k=1,\dots,2^{n}$) as given by (\ref{HXY+ZSpec}) for various values of $x$.  Points for odd prime values of $n$ are denoted by the appropriate symbols whereas those for other values are denoted by crosses.}
  \label{XY+ZDOSconverge}
\end{figure}

Convincing linear behaviour is seen from this figure.  Therefore an asymptotically tight bound on these data points is conjectured to have the form
\begin{equation}
  \frac{1}{E_{n}(x)}\geq c(x)n\qquad\iff\qquad E_{n}(x)\leq \frac{1}{nc(x)}
\end{equation}
for some real function $c(x)$.  The true rate of convergence is still an open question, even though the spectrum of $H_{n}^{(XY+Z)}$ is given explicitly for odd prime values of $n$.

\subsection{Numerical convergence of the spectral density of \texorpdfstring{$H_{n}^{(JW)}$}{Hn(JW)} in \texorpdfstring{$O(n^{-1})$}{O(1/n)}}\label{JW4}
More general Hamiltonians than $H_{n}^{(\epsilon XY+Z)}$ also admit numerical analysis via the Jordan-Wigner transform.  Consider the ensemble of matrices
\begin{equation}
  \hat{H}_{n}^{(JW)}=\sum_{j=1}^{n-1}\sum_{a,b=1}^{2}\hat{\alpha}_{a,b,j}\sigma_{j}^{(a)}\sigma_{j+1}^{(b)}+\sum_{j=1}^{n}\hat{\alpha}_{3,0,j}\sigma_{j}^{(3)}\qquad\qquad\hat{\alpha}_{a,b,j}\sim\mathcal{N}\left(0,\frac{1}{5n-4}\right) \iid
\end{equation}
Specific samples from this ensemble, denoted by $H_{n}^{(JW)}$ and parametrised by the real constants $\alpha_{a,b,j}$ corresponding to the random variables $\hat{\alpha}_{a,b,j}$, can be diagonalised in a similar fashion to $H_{n}^{(\epsilon XY+Z)}$ as in Section \ref{JW2}.  In this case though the boundary terms, proportional to $\sigma_{n}^{(a)}\sigma_{1}^{(b)}$ are not present, so the calculation simplifies.  Such matrices provide non-trivial instances from the ensemble $\hat{H}_{n}^{(local)}$.

Using the results in Appendix \ref{JWnn}, the Jordan-Wigner transform transforms the sampled matrix $H_{n}^{(JW)}$ to the form
\begin{align}
  H_{n}^{(JW)}=\sum_{j=1}^{n-1}\Big(
		-&\alpha_{1,1,j}\left(\a_{j}-\a_{j}^\dagger\right)\left(\a_{j+1}+\a_{j+1}^\dagger\right)
		+\im\alpha_{1,2,j}\left(\a_{j}-\a_{j}^\dagger\right)\left(\a_{j+1}-\a_{j+1}^\dagger\right)\nonumber\\
		+\im&\alpha_{2,1,j}\left(\a_{j}+\a_{j}^\dagger\right)\left(\a_{j+1}+\a_{j+1}^\dagger\right)
		+\alpha_{2,2,j}\left(\a_{j}+\a_{j}^\dagger\right)\left(\a_{j+1}-\a_{j+1}^\dagger\right)\Big)\nonumber\\
  &\qquad+\sum_{j=1}^{n}\alpha_{3,0,j}\left(\a_{j}\a_{j}^\dagger-\a_{j}^\dagger \a_{j}\right)
\end{align}
for the Fermi matrices $\a_{j}$ and $\a_{j}^\dagger$ defined as
\begin{equation}
  \a_{j}=\left(\sum_{1\leq l<j}\sigma_{l}^{(3)}\right)\frac{\sigma_{j}^{(1)}+\im\sigma_{j}^{(2)}}{2}\qquad\qquad
  \a_{j}^\dagger=\left(\sum_{1\leq l<j}\sigma_{l}^{(3)}\right)\frac{\sigma_{j}^{(1)}-\im\sigma_{j}^{(2)}}{2}
\end{equation}
Let $\boldsymbol{\a}$ be the column vector with entries $\a_{1},\dots,\a_{n}$, $\boldsymbol{\a}^\prime$ be the column vector with entries $\a_{1}^\dagger,\dots,\a_{n}^\dagger$, $\boldsymbol{\a}^\dagger$ be the row vector with entries $\a_{1}^\dagger,\dots,\a_{n}^\dagger$, ${\boldsymbol{\a}^\prime}^\dagger$ be the row vector with entries $\a_{1},\dots,\a_{n}$.  Using the canonical commutation relations for Fermi matrices ($\a_{j}\a_{k}^\dagger=-\a_{k}^\dagger \a_{j}+\delta_{j,k}I$ and $\a_{j}\a_{k}=-\a_{k}\a_{j}$) the last expression for $H_{n}^{(JW)}$ may be represented in the form
\begin{equation}
  H_{n}^{(JW)}=\begin{pmatrix}\boldsymbol{\a}^\dagger&{\boldsymbol{\a}^\prime}^\dagger\end{pmatrix}\begin{pmatrix}A-I&-\overline{B}\\B&-\overline{A}+I\end{pmatrix}\begin{pmatrix}\boldsymbol{\a}\\\boldsymbol{\a}^\prime\end{pmatrix}
\end{equation}
where $A=-\left(-A^{(1,1)}\right)+\im\left(-A^{(1,2)}\right)+\im A^{(2,1)}+A^{(2,2)}$ and $B=-A^{(1,1)}+\im A^{(1,2)}+\im A^{(2,1)}+A^{(2,2)}$ with
\begin{equation}
  A^{(a,b)}=\frac{1}{2}\begin{pmatrix}0&\alpha_{a,b,1}&0&\dots&0\\
				      -\alpha_{a,b,1}&0&\alpha_{a,b,2}&&\vdots\\
				      0&-\alpha_{a,b,2}&\ddots&\ddots&0\\
				      \vdots&&\ddots&&\alpha_{a,b,n-1}\\
				      0&\dots&0&\alpha_{a,b,n-1}&0
                       \end{pmatrix}
\end{equation}

There exists a $2n\times2n$ unitary matrix \cite[\p6]{Nielsen2005} of the form
\begin{equation}
  T=\begin{pmatrix}U&V\\\overline{V}&\overline{U}\end{pmatrix}
\end{equation}
where $U$ and $V$ are some $n\times n$ matrices and the bar denotes complex conjugation of their elements, such that
\begin{equation}
  T\begin{pmatrix}A-I&-\overline{B}\\B&-\overline{A}+I\end{pmatrix}T^\dagger=D\equiv\text{diag}\left(\mu_{1},\dots,\mu_{2n}\right)
\end{equation}
The calculation in Appendix \ref{FermiTrans} shows that $T$ also maps the Fermi matrices $\a_{j}$ and $\a_{j}^\dagger$ to new Fermi matrices $\b_{j}$ and $\b_{j}^\dagger$ as
\begin{equation}
  \begin{pmatrix}\boldsymbol{\b}\\\boldsymbol{\b}^\prime\end{pmatrix}=T\begin{pmatrix}\boldsymbol{\a}\\\boldsymbol{\a}^\prime\end{pmatrix}
\end{equation}
Using these facts allows $H_{n}^{(JW)}$ to be rewritten in the form
\begin{equation}
  H_{n}^{(JW)}
  =\begin{pmatrix}\boldsymbol{\b}^\dagger&{\boldsymbol{\b}^\prime}^\dagger\end{pmatrix}D\begin{pmatrix}\boldsymbol{\b}\\\boldsymbol{\b}^\prime\end{pmatrix}
  =\sum_{j=1}^{n}\left(\mu_{j}\b_{j}^\dagger \b_{j}+\mu_{j+n}\b_{j}\b_{j}^\dagger\right)
\end{equation}

Let $|\boldsymbol{0}\rangle_\b$ be a state such that $\b_{j}|\boldsymbol{0}\rangle_\b=0$ for all $j$, see \cite[\p3]{Nielsen2005} for the existence of such a state.  For the multi-indices $\boldsymbol{x}=(x_{1},\dots,x_{n})\in\{0,1\}^{n}$ let
\begin{equation}
  |\boldsymbol{x}\rangle_\b=\Big(\b_{1}^\dagger\Big)^{x_{1}}\dots\Big(\b_{n}^\dagger\Big)^{x_{n}}|\boldsymbol{0}\rangle_\b
\end{equation}
These states are orthonormal as seen in Appendix \ref{FermiBasis}.  As shown in Appendix \ref{JWFermiAction}, the canonical commutation relations for Fermi matrices imply that $\b_{j}^\dagger \b_{j}|\boldsymbol{x}\rangle_\b=x_{j}|\boldsymbol{x}\rangle_\b$ and $\b_{j} \b_{j}^\dagger|\boldsymbol{x}\rangle_\b=(1-x_{j})|\boldsymbol{x}\rangle_\b$.  Therefore, the eigenvalues of $H_{n}^{(JW)}$ can be read off from the eigenvalue equation
\begin{equation}
  H_{n}^{(JW)}|\boldsymbol{x}\rangle_\b=\sum_{j=1}^{n}\left(\mu_{j}x_{j}+\mu_{j+n}(1-x_{j})\right)|\boldsymbol{x}\rangle_\b
\end{equation}

\subsubsection{Numerics}
Let $C_{n}^{2}$ be the inverse of the sum of the squares of the values $\alpha_{a,b,j}$ present in $H_{n}^{(JW)}$.  The values $\mu_{j}$, and therefore the eigenvalues of $C_{n}H_{n}^{(JW)}$ denoted by $\lambda_{k}$, can be found numerically for each value of $n=3,\dots,32$ individually following the methodology above.  The scaling by $C_{n}$ has been used to insure that the variance of each of the related spectral densities is unity, this reduces statistical fluctuations when randomly choosing a member $H_{n}^{(JW)}$ from the ensemble $\hat{H}_{n}^{(JW)}$ to study.  As before, the quantity 
\begin{equation}
  E_{n}(x)=\left|\frac{1}{2^{n}}\sum_{k=1}^{2^{n}}\int_{-\infty}^{x^{+}}\delta(\lambda-\lambda_{k})\di\lambda-\frac{1}{\sqrt{2\pi}}\int_{-\infty}^{x}\e^{-\frac{\lambda^{2}}{2}}\di\lambda\right|
\end{equation}
will be used to quantify the rate of convergence of the eigenvalue distribution for $C_nH_{n}^{(JW)}$ to the normal distribution as $n\to\infty$.  Figure \ref{JWDOSconverge} shows the numerical values of $E_{n}^{-1}(x)$ against values of $n$ for various values of $x$ for a random, scaled, instance $H_{n}^{(JW)}$ of the ensemble $\hat{H}_{n}^{(JW)}$.

As with $H_{n}^{(XY+Z)}$, conceivably linear behaviour is seen from this figure.  Therefore an asymptotically tight bound on these data points is again conjectured to have the form
\begin{equation}
  \frac{1}{E_{n}(x)}\geq c(x)n\qquad\iff\qquad E_{n}(x)\leq \frac{1}{nc(x)}
\end{equation}
for some real function $c(x)$.  The true rate of convergence is still an open question.

\begin{figure}
  \centering
  \input{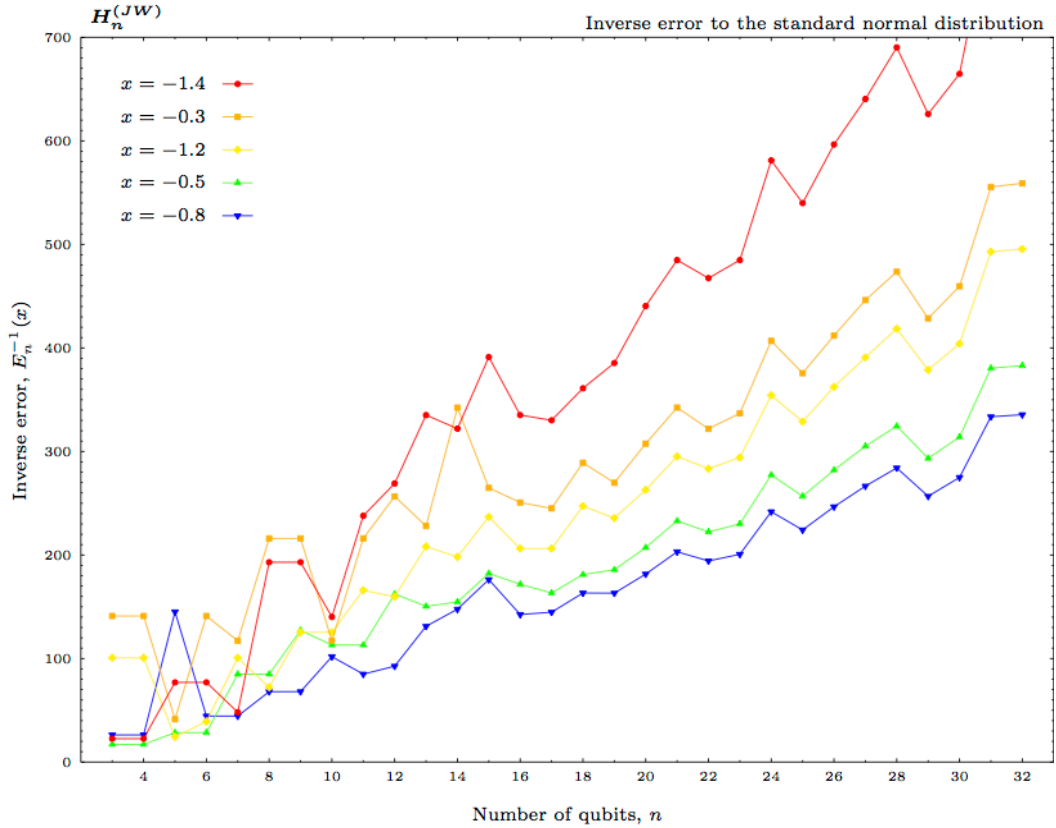}
  \caption[Convergence of the spectral density for $H_{n}^{(JW)}$]{The values of $E^{-1}_{n}(x)=\left|\frac{1}{2^{n}}\sum_{k=1}^{2^{n}}\int_{-\infty}^{x^{+}}\delta(\lambda-\lambda_{k})\di\lambda-\frac{1}{\sqrt{2\pi}}\int_{-\infty}^{x}\e^{-\frac{\lambda^{2}}{2}}\di\lambda\right|^{-1}$ against chain length $n$, for the eigenvalues $\lambda_{k}$ of a random instance $H_{n}^{(JW)}$ of the ensemble $\hat{H}_{n}^{(JW)}$, scaled so that its spectral density has unit variance, for various values of $x$.  Qualitatively similar behaviour was observed for all other instances of the ensemble $\hat{H}_{n}^{(JW)}$ tested.}
  \label{JWDOSconverge}
\end{figure}

\subsection{Unreliable convergence for the generic spin chain ensembles}\label{unsureConv}
To quantify the rate of convergence of the spectral densities of members of the ensembles $\hat{H}_{n}$, $\hat{H}_{n}^{(local)}$, $\hat{H}_{n}^{(inv)}$ and $\hat{H}_{n}^{(inv, local)}$ the quantity
\begin{equation}
  E_{n}(x)=\left|\frac{1}{2^{n}}\sum_{k=1}^{2^{n}}\int_{-\infty}^{x^{+}}\delta(\lambda-\lambda_{k})\di\lambda-\frac{1}{\sqrt{2\pi}}\int_{-\infty}^{x}\e^{-\frac{\lambda^{2}}{2}}\di\lambda\right|
\end{equation}
where the $\lambda_{k}$ are the eigenvalues of the ensemble member in question, can again be considered.  However, due to the relatively small values of $n$ accessible, compared to that in the last two sections, the fluctuations in this quaintly completely obscure any underling trend as $n$ increases.

To overcome this problem several samples from each of the ensembles in turn will be averaged over.  That is, the average proportion of the number of eigenvalues below some value $x\in\mathbb{R}$ over $s$ samples from an ensemble shall be compared against the cumulative distribution function of a standard normal random variable.  The absolute difference between these two quantities is given by
\begin{equation}
  \hat{E}_{n}(x)=\left|\frac{1}{s2^{n}}\sum_{k=1}^{s2^{n}}\int_{-\infty}^{x^{+}}\delta(\lambda-\lambda_{k})\di\lambda-\frac{1}{\sqrt{2\pi}}\int_{-\infty}^{x}\e^{-\frac{\lambda^{2}}{2}}\di\lambda\right|
\end{equation}
where the $\lambda_{k}$ are all of the eigenvalues of the $s$ ensemble members sampled.

The plot of $\hat{E}^{-1}_{n}(x)$ against $n$ for various values of $x$ for the ensembles $\hat{H}_{n}$ and $\hat{H}_{n}^{(local)}$ is shown in Figures \ref{DOSconverge_model1} and \ref{DOSconverge_model101} respectively.  In each figure a roughly linear relationship is possibly seen, but there remain too few data points for this to be conclusive.  As before, if these relationships were asymptotically linear or exponential then an asymptotically tight bound on the data point could be conjectured to have the form
\begin{align}
  \frac{1}{\hat{E}_{n}(x)}\geq c(x)n\qquad&\iff\qquad \hat{E}_{n}(x)\leq \frac{1}{nc(x)}\nonumber\\
  \text{or}\qquad\frac{1}{\hat{E}_{n}(x)}\geq c(x)2^{n}\qquad&\iff\qquad \hat{E}_{n}(x)\leq \frac{1}{2^{n}c(x)}
\end{align}
for some real function $c(x)$.  Or by taking the sample average as an indicator for the ensemble average
\begin{equation}
  \left|\int_{-\infty}^{x}\hat{\rho}_{n,1}(\lambda)\di\lambda-\frac{1}{\sqrt{2\pi}}\int_{-\infty}^{x}\e^{-\frac{\lambda^{2}}{2}}\di\lambda\right|\leq \frac{1}{nc(x)}\text{ or }\frac{1}{2^{n}c(x)}
\end{equation}
where $\hat{\rho}_{n,1}(\lambda)$ is the spectral density for the relevant ensemble.

\begin{figure}
  \centering
  \input{Figures/DOSconverge_model1.tex}
  \caption[Convergence of the spectral density for $\hat{H}_{n}$]{The values of $\hat{E}^{-1}_{n}(x)=\left|\frac{1}{s2^{n}}\sum_{k=1}^{s2^{n}}\int_{-\infty}^{x^{+}}\delta(\lambda-\lambda_{k})\di\lambda-\frac{1}{\sqrt{2\pi}}\int_{-\infty}^{x}\e^{-\frac{\lambda^{2}}{2}}\di\lambda\right|^{-1}$ against chain length $n$, for the eigenvalues $\lambda_{k}$ obtained from $s=2^{19-n}$ random samples from the ensemble $\hat{H}_{n}$.}
  \label{DOSconverge_model1}
\end{figure}

\begin{figure}
  \centering
  \input{Figures/DOSconverge_model101.tex}
  \caption[Convergence of the spectral density for $\hat{H}_{n}^{(local)}$]{The values of $\hat{E}^{-1}_{n}(x)=\left|\frac{1}{s2^{n}}\sum_{k=1}^{s2^{n}}\int_{-\infty}^{x^{+}}\delta(\lambda-\lambda_{k})\di\lambda-\frac{1}{\sqrt{2\pi}}\int_{-\infty}^{x}\e^{-\frac{\lambda^{2}}{2}}\di\lambda\right|^{-1}$ against chain length $n$, for the eigenvalues $\lambda_{k}$ obtained from $s=2^{19-n}$ random samples from the ensemble $\hat{H}^{(local)}_{n}$.}
  \label{DOSconverge_model101}
\end{figure}

This method of averaging did not remove sufficient fluctuations from the related quantities for the ensembles $\hat{H}_{n}^{(inv)}$ and $\hat{H}_{n}^{(inv, local)}$ and any underlining trends were still obscured.
\section{Spectral degeneracies}\label{degensec}
The $1$-point correlation function (or spectral density) is one of many correlation functions for the eigenvalues of random matrix ensembles, as seen in Section \ref{PointCorrFunc}.  In this section, eigenvalue degeneracies within the ensembles previously considered will be investigated.  This provides an initial step in the direction of higher correlation functions for the eigenvalues of these ensembles.

The presence of degeneracies in the spectra of Hamiltonians from $\hat{H}_{n}$, $\hat{H}_{n}^{(inv)}$ and $\hat{H}_{n}^{(inv, local)}$ are of particular importance to the applicability of Theorem \ref{eigenNonInv} of Section \ref{SingleEnt} (page \pageref{eigenNonInv}) and the interpretation of Theorem \ref{invEnt} of Section \ref{EntLblock} (page \pageref{invEnt}).  The applicability of Theorem \ref{eigenNonInv} is restricted to non-degenerate members of the ensemble $\hat{H}_{n}$ (or $\hat{H}_{n}^{(inv)}$), so an understanding of spectral degeneracies throughout these ensembles is required.  Furthermore, if almost all of the members of $\hat{H}_{n}^{(inv)}$ and $\hat{H}_{n}^{(inv, local)}$ possess a non-degenerate spectrum then Theorem \ref{invEnt} describes the unique eigenstates (up to phase) of almost all of the members of $\hat{H}_{n}^{(inv)}$ and $\hat{H}_{n}^{(inv, local)}$, so again, an understanding of spectral degeneracies throughout these ensembles is required.  

\subsection{Non-degeneracy for almost all ensemble members}
The first lemma of this section allows the non-degeneracy of almost all members of any ensemble mentioned above to be deduced from the existence of a single such example from the ensemble in question.

\begin{lemma}\label{nondeg}
  Consider the Hermitian matrix
  \begin{equation}
    H(\boldsymbol{\alpha})=\sum_{j=1}^{\r}\alpha_{j}h_{j}
  \end{equation}
  for some $N\times N$ Hermitian matrices $h_{j}$ and points $\boldsymbol{\alpha}=(\alpha_{1},\dots,\alpha_{\r})\in\mathbb{R}^{\r}$ for some $\r\in\mathbb{N}$.  If there exists some $\boldsymbol{\alpha}_{0}$ such that $H(\boldsymbol{\alpha}_{0})$ has a non-degenerate spectrum then the subset of points $\boldsymbol{\alpha}$ for which $H(\boldsymbol{\alpha})$ has a degenerate spectrum has zero Lebesgue measure.
\end{lemma}
Given a suitable probability measure\footnote{Any probability measure that assigns zero measure to a set with zero Lebesgue measure is appropriate, for example the Gaussian probability measure seen before.} on $\boldsymbol{\alpha}$, this proof then implies that almost all (with respect to this probability measure) Hamiltonians $H(\boldsymbol{\alpha})$ have a non-degenerate spectrum.

\begin{proof}\hspace{-2mm}\footnote{Adapted from the proof given by the author in \cite{KLW2014_FIXED} based on discussions with \cite{Keating}.} Let $V$ be the Vandermonde matrix with $V_{j,k}=\lambda_{j}^{k-1}$ for the eigenvalues $\lambda_{1},\dots,\lambda_{N}$ of $H(\boldsymbol{\alpha})$ (which depend of $\boldsymbol{\alpha}$).  By the standard properties of the Vandermonde matrix \cite[\p387]{Bernstein2009}
\begin{equation}
  \det_{1\leq j,k\leq N}\left(V_{j,k}\right)=\prod_{1\leq j<k\leq N}(\lambda_{k}-\lambda_{j})
\end{equation}
The determinant of a product of matrices is equal to the product of the determinants of the individual matrices.  Also the determinant of the transpose of a matrix is equal to the determinant of the original.  Therefore
\begin{equation}\label{detexp}
  \det_{1\leq j,k\leq N}\left(\left(V^\dagger V\right)_{j,k}\right)=\prod_{1\leq j<k\leq N}(\lambda_{k}-\lambda_{j})^{2}
\end{equation}
The elements of $V^\dagger V$ are given by
\begin{equation}\label{VV}
  \left(V^\dagger V\right)_{j,k}=\sum_{l=1}^{n}\lambda_{l}^{j-1}\lambda_{l}^{k-1}=\sum_{l=1}^{n}\lambda_{l}^{j+k-2}=\Tr \left(H^{j+k-2}(\boldsymbol{\alpha})\right)
\end{equation}
as the $\lambda_{l}$ are real.  By the definition of $H(\boldsymbol{\alpha})$ and the expansion (\ref{VV}), the elements of $V^\dagger V$ are polynomial functions of the variables $\alpha_{1},\dots,\alpha_{\r}$ and therefore $f(\boldsymbol{\alpha})=\det_{1\leq j,k\leq N}\left(\left(V^\dagger V\right)_{j,k}\right)$ is a polynomial function of the variables $\alpha_{1},\dots,\alpha_{\r}$.

The function $f(\boldsymbol{\alpha})$ is zero if and only if at least two eigenvalues of $H(\boldsymbol{\alpha})$ are equal by (\ref{detexp}).  The zero set of $f$ (the $\boldsymbol{\alpha}$ such that $f(\boldsymbol{\alpha})=0$) therefore coincides with the set of $\boldsymbol{\alpha}$ such that $H(\boldsymbol{\alpha})$ is degenerate.

If there exists some $\boldsymbol{\alpha}_{0}$ such that $H(\boldsymbol{\alpha}_{0})$ has a non-degenerate spectrum then $f(\boldsymbol{\alpha})$ is not the zero polynomial.  Its zero set therefore has zero Lebesgue measure.  This completes the proof.
\end{proof}

\subsection{The ensembles \texorpdfstring{$\hat{H}_{n}$}{Hn}, \texorpdfstring{$\hat{H}_{n}^{(local)}$}{Hn(local)}, \texorpdfstring{$\hat{H}_{n}^{(inv)}$}{Hn(inv)} and \texorpdfstring{$\hat{H}_{n}^{(inv, local)}$}{Hn(inv, local)} for \texorpdfstring{$2\leq n\leq13$}{1<n<14}}
For the numerical simulations of Sections \ref{Numerics} and \ref{eigenNum}, with the results summarised in Appendix \ref{collNumerics}, many samples from each of the ensembles $\hat{H}_{n}$, $\hat{H}_{n}^{(local)}$, $\hat{H}_{n}^{(inv)}$ and $\hat{H}_{n}^{(inv, local)}$ were generated.  Table \ref{degen} summarises for which of these ensembles an instance of a non-degenerate Hamiltonian was seen.

\begin{table}\footnotesize
\centering
\begin{tabular}{l|ccccccccccccccc}
  \toprule
	\multirow{1}{*}{\bf{Ensemble}}&\multicolumn{12}{ c }{\bf{Number of qubits, $\boldsymbol{n}$}}\\
	 &$2$&$3$&$4$&$5$&$6$&$7$&$8$&$9$&$10$&$11$&$12$&$13$\\
	\midrule
	$\hat{H}_{n}$&$\checkmark$&$-$&$\checkmark$&$-$&$\checkmark$&$-$&$\checkmark$&$-$&$\checkmark$&$-$&$\checkmark$&$-$\\	$\hat{H}_{n}^{(local)}$&$\checkmark$&$\checkmark$&$\checkmark$&$\checkmark$&$\checkmark$&$\checkmark$&$\checkmark$&$\checkmark$&$\checkmark$&$\checkmark$&$\checkmark$&$\checkmark$\\
	$\hat{H}_{n}^{(inv)}$&$\checkmark$&$-$&$-$&$-$&$-$&$-$&$-$&$-$&$-$&$-$&$-$&$-$\\	
	$\hat{H}_{n}^{(inv, local)}$&$\checkmark$&$\checkmark$&$-$&$\checkmark$&$\checkmark$&$\checkmark$&$\checkmark$&$\checkmark$&$\checkmark$&$\checkmark$&$\checkmark$&$\checkmark$\\
	\bottomrule
\end{tabular}
\caption[Non-degenerate instances of ensembles]{Table showing the ensembles for which an instance was found which possessed a non-degenerate spectrum (no absolute spacing less than $10^{-10}$) from the numerical samples used in Section \ref{Numerics}.  Such ensembles are marked with `$\checkmark$'.  Ensembles where this was not the case are marked with `$-$'.}
\label{degen}
\end{table}

Lemma \ref{nondeg} then implies that almost all of the members of the ensembles for which a non-degenerate Hamiltonian was found, have a non-degenerate spectrum.

\subsection{The ensemble \texorpdfstring{$\hat{H}_{n}^{(local)}$}{Hn(local)}}\label{localSpec}
In order to analytically treat the numerical lack of degeneracies in the ensemble $\hat{H}_{n}^{(local)}$, the Hamiltonian 
\begin{equation}
  H_{n}^{(\epsilon^{j}Z)}=\sum_{j=1}^{n}\epsilon^{j}\sigma_{j}^{(3)}
\end{equation}
where $\epsilon\in\mathbb{R}$ and $n\geq2$, will be analysed.  This is a specific member of the ensemble $\hat{H}_{n}^{(local)}$. The following lemma is needed first:

\begin{lemma}\label{LemmaH_{n}^{(}ejZ)}
  The set of values $\epsilon\in\mathbb{R}$ for which $H_{n}^{(\epsilon^{j}Z)}$ has a degenerate spectrum has Lebesgue measure zero for all $n\in\mathbb{N}$.
\end{lemma}

By Lemma \ref{nondeg} the following corollary then follows immediately:
\begin{corollary}
  Almost all members of the ensemble $\hat{H}_{n}^{(local)}$ have a non-degenerate spectrum for any value of $n=2,3,\dots$.
\end{corollary}

The proof of Lemma \ref{LemmaH_{n}^{(}ejZ)} will now be given:
\begin{proof}
  The standard basis is defined in Section \ref{stndBasis} to be the $n$-fold tensor products of the eigenstates $|0\rangle$ (with eigenvalue $+1$) and $|1\rangle$ (with eigenvalue $-1$), with some fixed phases, of the Pauli matrix $\sigma^{(3)}$.  That is for the multi-indices $\boldsymbol{x}=(x_{1},\dots,x_{n})\in\{0,1\}^{n}$, the basis states are given by
  \begin{equation}
    |\boldsymbol{x}\rangle=|x_{1}\rangle\otimes\dots\otimes|x_{n}\rangle
  \end{equation}
  The matrices $\sigma_{j}^{(3)}$ for $j=1,\dots,n$ then, by definition of the tensor product, satisfy
  \begin{align}
    \sigma_{j}^{(3)}|\boldsymbol{x}\rangle
    &=\left(I_{2}^{\otimes j-1}\otimes\sigma^{(3)}\otimes I_{2}^{\otimes n-j}\right)|x_{1}\rangle\otimes\dots\otimes|x_{n}\rangle\nonumber\\
    &=\left(\bigotimes_{1\leq l<j}I_{2}|x_{l}\rangle\right)\otimes\left(\sigma^{(3)}|x_{j}\rangle\right)\otimes\left(\bigotimes_{j<l\leq n}I_{2}|x_{l}\rangle\right)\nonumber\\
    &=(-1)^{x_{j}}|\boldsymbol{x}\rangle
  \end{align}
  as $\sigma^{(3)}|x_{j}\rangle=(-1)^{x_{j}}|x_{j}\rangle$ by the definition of $|x_{j}\rangle$.  The eigenvalues of $H_{n}^{(\epsilon^{j}Z)}$ can then be read off from the eigenvalue equation
  \begin{equation}
    H_{n}^{(\epsilon^{j}Z)}|\boldsymbol{x}\rangle
    =\sum_{j=1}^{n}\epsilon^{j}\sigma_{j}^{(3)}|\boldsymbol{x}\rangle
    =\sum_{j=1}^{n}\epsilon^{j}(-1)^{x_{j}}|\boldsymbol{x}\rangle
  \end{equation}
  It is seen that the eigenvalues are degree $n$ polynomial functions of $\epsilon$ given by
  \begin{equation}
    \lambda_{\boldsymbol{x}}(\epsilon)=\sum_{j=1}^{n}\epsilon^{j}(-1)^{x_{j}}
  \end{equation}
  These $2^{n}$ functions are necessarily all unique as at least one coefficient is different between each one, else the vectors $\boldsymbol{x}$ indexing them would not be unique.  The set of points $\epsilon$ where any two of these functions are equal therefore has zero Lebesgue measure.  Therefore the set of points $\epsilon$ for which $H_{n}^{(\epsilon^{j}Z)}$ has a degenerate spectrum also has zero Lebesgue measure.  This completes the proof.
\end{proof}

\subsection{The ensemble \texorpdfstring{$\hat{H}_{n}^{(inv, local)}$}{Hn(inv, local)} for odd prime values of \texorpdfstring{$n$}{n}}\label{invLocalSpec}
In order to analytically treat the numerical lack of degeneracies in the ensemble $\hat{H}_{n}^{(inv, local)}$, the Hamiltonian
\begin{equation}
  H_{n}^{(\epsilon XY+Z)}=\sum_{j=1}^{n}\left(\epsilon\sigma_{j}^{(1)}\sigma_{j+1}^{(2)}+\sigma_{j}^{(3)}\right)
\end{equation}
where $\epsilon\in\mathbb{R}$ and $n$ is an odd prime, will be used.  This is a specific member of the ensemble $\hat{H}_{n}^{(inv, local)}$. The following lemma is needed first:

\begin{lemma}\label{lemmaH_n^(eXY+Z)} 
  The set of values $\epsilon\in\mathbb{R}$ for which $H_{n}^{(\epsilon XY+Z)}$ has a degenerate spectrum has Lebesgue measure zero for all odd prime values of $n$.
\end{lemma}
Figure \ref{FigH_{n}^{(}eXY+Z)} shows the eigenvalues of $H_{7}^{(\epsilon XY+Z)}$ for $\epsilon\in[0,1.5]$.  The uniqueness of the eigenvalues for almost all values of $\epsilon$ in this range is seen.

\begin{figure}
  \centering
  \input{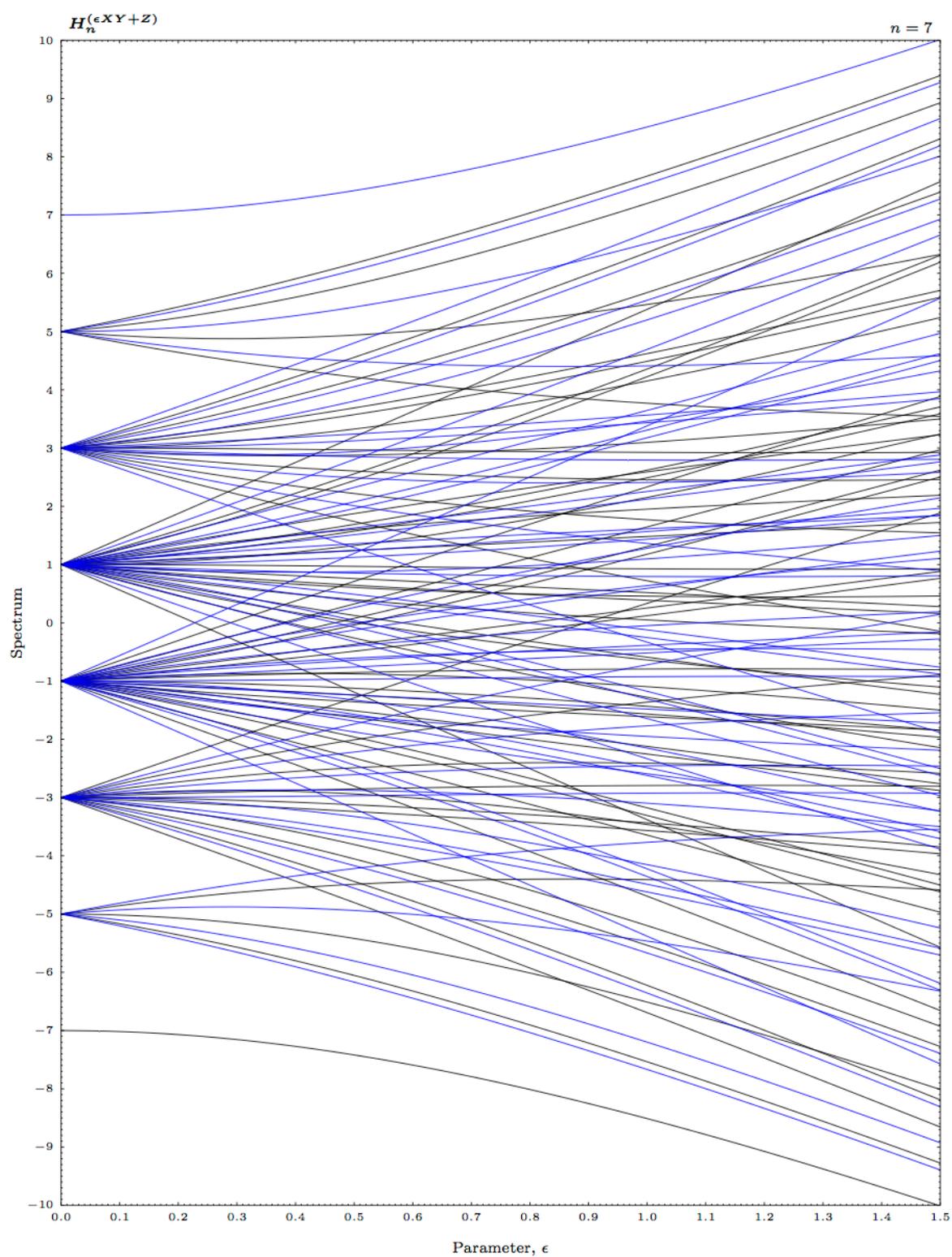}
  \caption[The eigenvalues of $H_{n}^{(\epsilon XY+Z)}$]{The $2^{7}$ eigenvalues of $H_{7}^{(\epsilon XY+Z)}$ for $\epsilon\in[0,1.5]$.  For each value of $\epsilon$ the $2^{n}$ eigenvalues $\lambda_{\boldsymbol{x}}$ are plotted as points $(\epsilon, \lambda_{\boldsymbol{x}})$, coloured blue or black for clarity.  The eigenvalues are seen to be non-degenerate for generic values of $\epsilon$.}
  \label{FigH_{n}^{(}eXY+Z)}
\end{figure}

By Lemma \ref{nondeg} the following corollary then follows immediately:
\begin{corollary}
  Almost all members of the ensemble $\hat{H}_{n}^{(inv, local)}$ have a non-degenerate spectrum for all odd prime values of $n$.
\end{corollary}

The proof of Lemma \ref{lemmaH_n^(eXY+Z)} will now be given:
\begin{proof}\hspace{-2mm}\footnote{Adapted from the proof given by the author in \cite{KLW2014_FIXED}.}
Let $n$ be an odd prime.  It will be shown that the $2^{n}$ values
\begin{equation}
	\lambda_{\boldsymbol{x}}=\sum_{j=1}^{n}(2x_{j}-1)\left(\epsilon\mu_{j}-\sqrt{\epsilon^{2}\mu_{j}^{2}+1}\right),\qquad\qquad\mu_{j}=\sin\left(\frac{2\pi j}{n}\right)
\end{equation}
for the multi-indices $\boldsymbol{x}=(x_{1},\dots,x_{n})\in\{0,1\}^{n}$, are all distinct for all real values of $\epsilon$ apart from those chosen from some subset with zero Lebesgue measure:

\subsubsection{Linear independence of the $\{\mu_{j}\}_{j=1}^{\frac{n-1}{2}}$}
First, it will be shown that for odd prime values of $n$ the values $\{\mu_{j}\}_{j=1}^{\frac{n-1}{2}}$ are linearly independent over the integers.  Let $\omega=\e^{\frac{2\pi\im}{n}}$ so that
\begin{equation}
	\mu_{j}=\sin\left(\frac{2\pi j}{n}\right)=\frac{\omega^{j}-\omega^{-j}}{2\im}
\end{equation}
For the arbitrary integers $c_{j}$, consider the linear combination $\sum_{j=1}^{\frac{n-1}{2}}c_{j}\mu_{j}$.  The expression above for $\omega$ then allows this to be rewritten as
\begin{align}
	\frac{1}{2\im}\sum_{j=1}^{\frac{n-1}{2}}c_{j}\omega^{j}-\frac{1}{2\im}\sum_{j=1}^{\frac{n-1}{2}}c_{j}\omega^{-j}
\end{align}
The indices in the second sum may be relabelled by $k=n-j$ so that the sum runs over the indices $k=\frac{n+1}{2},\dots,n-1$.  The summand is transformed to $c_{n-k}\omega^{k-n}=c_{n-k}\omega^{k}\omega^{-n}=c_{n-k}\omega^{k}$.  Therefore the linear combination considered above is rewritten as
\begin{equation}
  \frac{1}{2\im}\sum_{j=1}^{\frac{n-1}{2}}c_{j}\omega^{j}-\frac{1}{2\im}\sum_{k=\frac{n+1}{2}}^{n-1}c_{n-k}\omega^{k}
\end{equation}
The powers of $\omega$ that are not proportional to $\omega^{0}=\omega^{n}=1$ are linearly independent over the integers \cite[\p12]{Lenstra1979} so that this expression is zero if and only if all the coefficients $c_{j}$ are zero.  This implies that the $\{\mu_{j}\}_{j=1}^{\frac{n-1}{2}}$ are linearly independent over the integers.

\subsubsection{Expansion of $\lambda_{\boldsymbol{x}}(\epsilon)$ for small $\epsilon$}
For small $\epsilon$, the values $\lambda_{\boldsymbol{x}}(\epsilon)$ admit the expansion
\begin{equation}
	\lambda_{\boldsymbol{x}}(\epsilon)=\sum_{j=1}^{n}(2x_{j}-1)\left(-1+\epsilon\mu_{j}-\epsilon^{2}\frac{\mu_{j}^{2}}{2}+O\left(\epsilon^{4}\right)\right)
\end{equation}
Suppose, for a contradiction, that two eigenvalues are equal in some neighbourhood around $\epsilon=0$, that is for $\boldsymbol{x}\neq\boldsymbol{y}$,
\begin{equation}
	0=\lambda_{\boldsymbol{x}}(\epsilon)-\lambda_{\boldsymbol{y}}(\epsilon)
	=-2\sum_{j=1}^{n}(x_{j}-y_{j})+2\epsilon\sum_{j=1}^{n}\mu_{j}(x_{j}-y_{j})-\epsilon^{2}\sum_{j=1}^{n}\mu_{j}^{2}(x_{j}-y_{j})+O\left(\epsilon^{4}\right)
\end{equation}

\subsubsection{Comparing the $\epsilon^{0}$ coefficient}
Comparing the $\epsilon^{0}$ coefficient and setting $d_{j}=x_{j}-y_{j}$ gives
\begin{equation}
	0=\sum_{j=1}^{n}d_{j}
\end{equation}

\subsubsection{Comparing the $\epsilon^{1}$ coefficient}
Comparing the $\epsilon^{1}$ coefficient gives
\begin{equation}\label{4.2.18}
	0=\sum_{j=1}^{n}\mu_{j}d_{j}=\sum_{j=1}^{\frac{n-1}{2}}\mu_{j}d_{j}+\sum_{j=\frac{n+1}{2}}^{n-1}\mu_{j}d_{j}+\mu_{n}d_{n}
\end{equation}
where the right hand side of this expression represents a partitioning of the sum on the left.  By the definition of $\mu_{j}$, $\mu_{n}=\sin(2\pi)=0$.  The indices in the second sum on the right hand side of (\ref{4.2.18}) may be relabelled by $k=n-j$ so that the sum runs over the indices $k=1,\dots,\frac{n-1}{2}$.  The summand is transformed to $\mu_{n-k}d_{n-k}=\sin\left(\frac{2\pi(n-k)}{n}\right)d_{n-k}=\sin\left(\frac{-2\pi k}{n}\right)d_{n-k}=-\mu_{k}d_{n-k}$.  The equation from the $\epsilon^{1}$ coefficient then reads
\begin{equation}
  0=\sum_{j=1}^{\frac{n-1}{2}}\mu_{j}d_{j}-\sum_{k=1}^{\frac{n-1}{2}}\mu_{k}d_{n-k}=\sum_{j=1}^{\frac{n-1}{2}}\mu_{j}(d_{j}-d_{n-j})
\end{equation}
so that by the linear independence of the $\{\mu_{j}\}_{j=1}^{\frac{n-1}{2}}$ over the integers, $d_{j}=d_{n-j}$ for all $j=1,\dots,\frac{n-1}{2}$.  In particular, substituting this into the $\epsilon^{0}$ result gives
\begin{equation}
  0=2\sum_{j=1}^{\frac{n-1}{2}}d_{j}+d_{n}
\end{equation}
from which, since the sum over $j$ on the right hand side is even and the lone term $d_{n}$ is either equal to $-1$, $0$ or $1$, implies that $d_{n}=0$.

\subsubsection{Comparing the $\epsilon^{2}$ coefficient}
Comparing the $\epsilon^{2}$ coefficient gives
\begin{align}\label{e2}
  0=\sum_{j=1}^{n}\mu_{j}^{2}d_{j}=\sum_{j=1}^{n-1}\mu_{j}^{2}d_{j}&=\frac{1}{(2\im)^{2}}\sum_{j=1}^{n-1}(\omega^{j}-\omega^{-j})^{2}d_{j}\nonumber\\
  &=-\frac{1}{4}\sum_{j=1}^{n-1}\omega^{2j}d_{j}-\frac{1}{4}\sum_{j=1}^{n-1}\omega^{-{2j}}d_{j}+\frac{1}{4}\sum_{j=1}^{n-1}2d_{j}
\end{align}
where the substitution $\mu_{j}=\frac{\omega^{j}-\omega^{-j}}{2\im}$ and $d_{n}=0$ has been made and the resulting sum expanded.

It has already been seen that the third term in the last expression of (\ref{e2}) is zero from the $\epsilon^{0}$ result with $d_{n}=0$.  The indices in the second sum in the last expression of (\ref{e2}) may be relabelled by $k=n-j$ so that the sum runs over the indices $k=1,\dots,n-1$.  The summand is transformed to $\omega^{-2(n-k)}d_{n-k}=\omega^{-2n}\omega^{2k}d_{n-k}=\omega^{2k}d_{n-k}$. The equation for the $\epsilon^{2}$ coefficient then reads 
\begin{equation}
  0=-\frac{1}{4}\sum_{j=1}^{n-1}\omega^{2j}d_{j}-\frac{1}{4}\sum_{k=1}^{n-1}\omega^{2k}d_{n-k}
  =-\frac{1}{4}\sum_{j=1}^{n-1}\omega^{2j}(d_{j}+d_{n-j})
\end{equation}

It is seen that in each term of the right hand side of this expression a different power of $\omega$ is present, and that the power of $\omega$ which is proportional to $\omega^{0}=\omega^{n}=1$ is not.  This is seen from dividing the different powers of $\omega$ in the sum into the following two sets
\begin{align}
  \left\{\omega^{2j}\,\,|\,\,j=1,\dots,\frac{n-1}{2}\right\}&=\left\{\omega^{2},\omega^{4},\dots,\omega^{n-1}\right\}\nonumber\\
  \left\{\omega^{2j}\,\,|\,\,j=\frac{n+1}{2},\dots,n-1\right\}&=\left\{\omega^{1},\omega^{3},\dots,\omega^{n-2}\right\}
\end{align}

Now by the linear independence of the powers of $\omega$ (not proportional to $\omega^{0}=\omega^{n}=1$) over the integers, $d_{j}=-d_{n-j}$ for all $j=1,\dots,n-1$.

\subsubsection{The contradiction}
It has been seen that $d_{n}=0$, $d_{j}=d_{n-j}$ and that $d_{j}=-d_{n-j}$ for all $j=1,\dots,\frac{n-1}{2}$ so that it is concluded that $d_{j}=0$ for all $j$.  That is $\boldsymbol{x}=\boldsymbol{y}$, a contradiction.  Therefore there must exist some $\epsilon_{0}\in\mathbb{R}$ for which $\lambda_{\boldsymbol{x}}(\epsilon_{0})\neq\lambda_{\boldsymbol{y}}(\epsilon_{0})$.

\subsubsection{Final result}
The functions $\lambda_{\boldsymbol{x}}(\epsilon)$ and $\lambda_{\boldsymbol{y}}(\epsilon)$ are analytic and for $\boldsymbol{x}\neq\boldsymbol{y}$ necessarily distinct as there exists some $\epsilon_{0}\in\mathbb{R}$ such that $\lambda_{\boldsymbol{x}}(\epsilon_{0})\neq\lambda_{\boldsymbol{y}}(\epsilon_{0})$.  The Lebesgue measure of the set of values $\epsilon$ such that $\lambda_{\boldsymbol{x}}(\epsilon)=\lambda_{\boldsymbol{y}}(\epsilon)$ must then be zero.  Therefore the values of $\lambda_{\boldsymbol{x}}$ for all $\boldsymbol{x}$ are distinct for all real values of $\epsilon$ apart from those chosen from some subset of $\mathbb{R}$ with zero Lebesgue measure.

Lemma \ref{SpecH_{n}^{(}XY+Z)} of Section \ref{DOSrate} states that the $\lambda_{\boldsymbol{x}}$ are the eigenvalues of $H_{n}^{(\epsilon XY+Z)}$, concluding the proof.
\end{proof}

\subsection{Kramers' degeneracies in \texorpdfstring{$H_{n}$}{Hn} and \texorpdfstring{$H_{n}^{(inv)}$}{Hn(inv)} for odd \texorpdfstring{$n\geq3$}{n>2}}\label{Kramer}
Some of the degeneracies seen numerically in samples from the ensembles $\hat{H}_{n}$ and $\hat{H}_{n}^{(inv)}$ for odd values of $n$ are Kramers' degeneracies.  In fact, each member $H_{n}$ and $H_{n}^{(inv)}$ of the ensembles $\hat{H}_{n}$ and $\hat{H}_{n}^{(inv)}$ respectively has at least doubly degenerate eigenvalues for odd values of $n\geq3$.  These degeneracies can be deduced from the symmetry used in the proof of Theorem \ref{eigenNonInv} of Section \ref{SingleEnt} (page \pageref{eigenNonInv}),
\begin{equation}
  SH_{n}=\overline{H_{n}}S
\end{equation}
for the Hermitian and unitary matrix $S=\prod_{j=1}^{n}\sigma_{j}^{(2)}$.  As will be seen in Section \ref{Hevenspacigns} this symmetry is an anti-unitary (or time-reversal) symmetry.

Given that the matrices $H_{n}$ have this symmetry, their spectrum is assured to be at least doubly degenerate for odd values of $n\geq3$ as:

\begin{lemma}[Kramers' degeneracy]
  For odd values of $n\geq3$, every member ${H}_{n}$ from the ensemble $\hat{H}_{n}$ (and therefore every member $H^{(inv)}_{n}$ of the ensemble $\hat{H}^{(inv)}_{n}$) has, at least, doubly degenerate eigenvalues.
\end{lemma}

\begin{proof}\hspace{-2mm}\footnote{Adapted from the proof given by the author in \cite{KLW2014_RMT}.}
Let $|\psi\rangle$ be an eigenstate of ${H}_{n}$ with an eigenvalue of $\lambda$.  As $S{H}_{n}=\overline{{H}_{n}}S$, it follows that
\begin{equation}
  \overline{{H}_{n}}\Big(S|\psi\rangle\Big)=S {H}_{n}|\psi\rangle=\lambda\Big(S|\psi\rangle\Big)
\end{equation}
Since $\lambda$ is necessarily real as $H_{n}$ is Hermitian, taking the complex conjugate of this equation yields that $\overline{S|\psi\rangle}$ is an eigenstate of ${H}_{n}$ with eigenvalue $\lambda$.

Now the inner-product of $\overline{S|\psi\rangle}$ and $|\psi\rangle$ will be calculated.  As $S^\dagger S=I$ and 
\begin{equation}
 S\overline{S}=\prod_{j=1}^{n}\sigma_{j}^{(2)}\overline{\sigma_{j}^{(2)}}=\prod_{j=1}^{n}(-I)=(-1)^{n}I
\end{equation}
it follows that
\begin{equation}\label{kramerEqu}
  \langle\psi|\overline{S|\psi\rangle}=\langle\psi|S^\dagger S\overline{S|\psi\rangle}=(-1)^{n}\langle\psi|S^\dagger\overline{|\psi\rangle}
\end{equation}

For any vectors $|\phi_{1}\rangle$ and $|\phi_{2}\rangle$ in the Hilbert space, a standard property of the complex Euclidean inner-product is that $\langle\phi_{1}|\phi_{2}\rangle=\overline{\langle\phi_{2}|\phi_{1}\rangle}$.  With $|\phi_{1}\rangle=|\psi\rangle$ and $|\phi_{2}\rangle=S^\dagger\overline{|\psi\rangle}$
\begin{equation}
  \langle\phi_{2}|=\left(S^\dagger\overline{|\psi\rangle}\right)^\dagger=\left(\overline{|\psi\rangle}\right)^\dagger\left(S^\dagger\right)^\dagger=\overline{\langle\psi|}S
\end{equation}
and therefore $\langle\psi|S^\dagger\overline{|\psi\rangle}=\overline{\overline{\langle\psi|}S|\psi\rangle}=\langle\psi|\overline{S|\psi\rangle}$.

From (\ref{kramerEqu}) it now follows that
\begin{equation}
  \langle\psi|\overline{S|\psi\rangle}=(-1)^{n}\langle\psi|\overline{S|\psi\rangle}
\end{equation}
and it is concluded that $\langle\psi|\overline{S|\psi\rangle}=0$ for all odd values of $n\geq3$.  Hence, $|\psi\rangle$ and $\overline{S|\psi\rangle}$ are orthogonal eigenstates of ${H}_{n}$, both with the eigenvalue $\lambda$, for all odd values of $n\geq3$.
\end{proof}
\section{Numerical nearest-neighbour level spacings}\label{nnnumerics}
The ensembles $\hat{H}_{n}$, $\hat{H}_{n}^{(local)}$ and $\hat{H}_{n}^{(inv,local)}$ display a variety of nearest-neighbour level spacing statistics closely matching those from the canonical invariant ensembles.  In this section the results of numerical approximations to the nearest-neighbour level spacing distributions for these ensembles are presented and comments made as to their form in relation to the ensembles' symmetries.

To observe these spacing distributions, the eigenvalues of samples from each ensemble were rescaled (unfolded) so that on average they had a constant density on the unit interval $[0,1]$.  This rescaling of the eigenvalues removes any effect the original average spectral density has on the spacings, allowing a more unbiased comparison of the spacings between different ensembles.

The rescaling was performed by first only considering the eigenvalues in the range $[-3,3]$ of each of $s$ sampled matrices from an ensemble.  The results of Chapter \ref{DOS} show this to include the vast majority of the spectrum, so that this restriction should closely approximate the case if all the eigenvalues were considered.  The averaged proportions, $p_{j}$, of these eigenvalues in the intervals $[x_{j},x_{j+1})$, of equal length $l$, for
\begin{equation}
  -3=x_{0}<x_{1}<\dots<x_{240}=3
\end{equation}
over each of the $s$ samples was then calculated.  The eigenvalues $\lambda$ that fell into the interval $[x_{j},x_{j+1})$ being rescaled to
\begin{equation}
  \lambda\rightarrow\frac{p_{j}(\lambda-x_{j})}{l}+\sum_{0\leq k<j}p_{k}
\end{equation}

Here the eigenvalue $\lambda$ has first been shifted by $-x_{j}$ so that it lies in the interval $[0,l)$, then linearly scaled by $\frac{p_{j}}{l}$ so that it lies in the interval $[0,p_{j})$ and then shifted by $\sum_{0\leq k<j}p_{k}$ so that it lies in the interval
\begin{equation}
  \left[\sum_{0\leq k<j}p_{k},\sum_{0\leq k\leq j}p_{k}\right)
\end{equation}
The proportion of the scaled eigenvalues in this interval of width $p_{j}$ is then still $p_{j}$.  That is, the transformed eigenvalues have an approximately constant density on the interval $[0,1)$, see Figure \ref{Unfoldgraph}.

\begin{figure}
  \centering
  \input{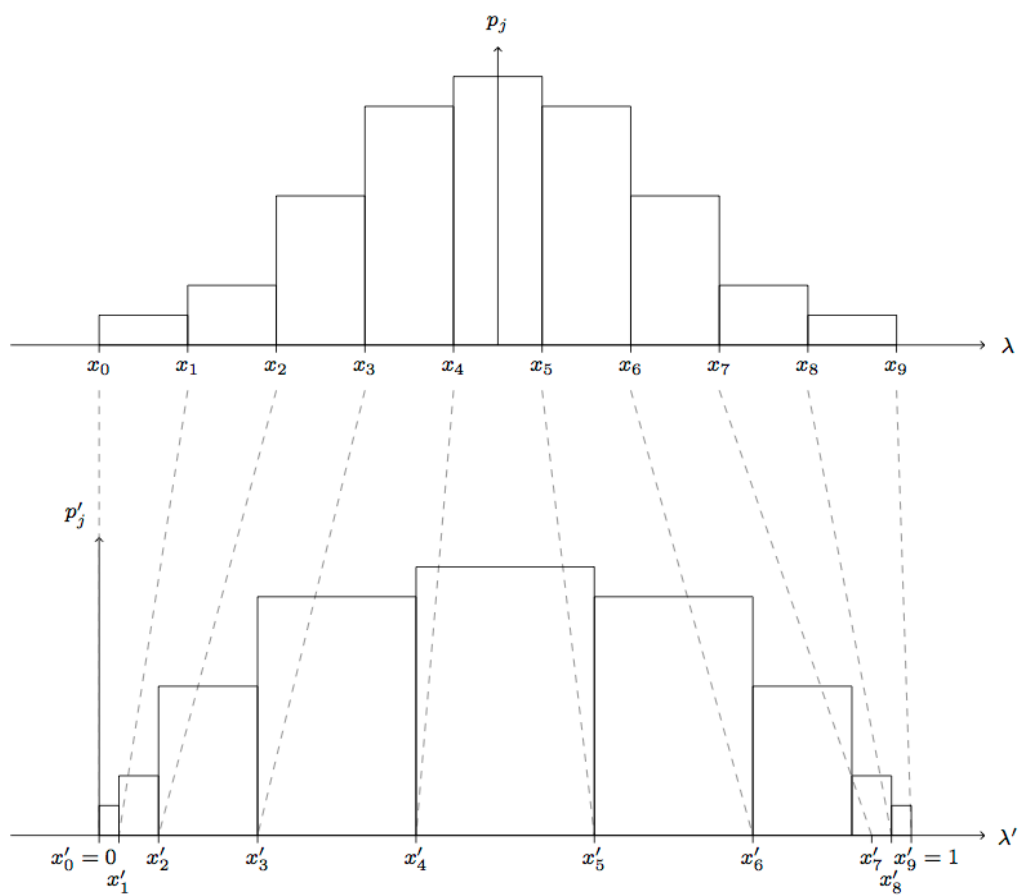}
  \caption[Unfolding of the spectrum of an ensemble]{(Top) Diagrammatic histogram showing the proportions, $p_{j}$, of eigenvalues in each of the $9$ intervals $[x_{j},x_{j+1})$, out of those present in $[x_{0},x_{9})$.  (Bottom) After unfolding, the intervals $[x_{j},x_{j+1})$ and eigenvalues therein have been linearly transformed to $[x^\prime_{j},x^\prime_{j+1})$ so that $x^\prime_{j+1}-x^\prime_{j}$ is equal to $p_{j}$.  Assuming a roughly constant density in each interval, the unfolded eigenvalues have an approximately constant density on the unit interval $[0,1)$.}
  \label{Unfoldgraph}
\end{figure}

The differences between consecutive (unfolded and numerically ordered) eigenvalues, for each of the $s$ Hamiltonians sampled from each ensemble individually, was then calculated and the total set of $s(2^{n}-1)$ resulting values scaled to have unit mean.  A normalised histogram was then made of these values.  For the ensemble $\hat{H}_{n}$, $\hat{H}_{n}^{(local)}$ and $\hat{H}_{n}^{(inv,local)}$ these are as follows:

\subsection{Poisson spacing statistics for \texorpdfstring{$\hat{H}_{n}^{(inv, local)}$}{Hn(inv, local)}}
The nearest-neighbour level spacing histogram for the unfolded eigenvalues of $s=2^{6}$ samples from the ensemble $\hat{H}_{13}^{(inv, local)}$ is shown in Figure \ref{SpacingsHinvlocal}.  Histograms for $n=2,\dots,13$ are shown in Appendix \ref{collNumericsSpacing} for completeness, with all histograms for $n=5,\dots,13$ displaying similar characteristics.

\begin{figure}
  \centering
  \input{Figures/SpacingsHinvlocal}
  \caption[Nearest-neighbour level spacing histogram for $\hat{H}_{13}^{(inv, local)}$]{The normalised nearest-neighbour level spacing histogram for $\hat{H}_{13}^{(inv, local)}$.  The spacings used were numerically obtained from each of the unfolded spectra of $s=2^{6}$ samples from the ensemble $\hat{H}_{13}^{(inv, local)}$ and scaled to have unit mean.  The unfolding was with respect to the average numerical spectral density on the interval $[-3,3]$ calculated from the same sample.  The dashed line gives the standard Poisson distribution.}
  \label{SpacingsHinvlocal}
\end{figure}

A close resemblance to the Poisson distribution is seen.  This form of spacing statistics is generally observed for statistically independent points.  It is consistent with the geometric symmetry of the Hamiltonians within the ensemble $\hat{H}_{n}^{(inv,local)}$.  Each such Hamiltonian commutes with the translation matrix $T$ that translates the Hamiltonian by one qubit around the ring, see Section \ref{EigenDef}.  The translation matrix has $n$ eigenvalues (the $n$ roots of unity) and therefore each Hamiltonian from $\hat{H}_{n}^{(inv,local)}$ is block diagonal with respect to the eigenstates of $T$.  As seen in many examples highlighted in Section \ref{nonInt}, there is generally not expected to be repulsion between the eigenvalues of these blocks (or subspaces).  In effect, eigenvalues from different blocks behave as independent points.

Similar spacing statistics are also seen for the ensemble $\hat{H}_{n}^{(JW)}$, see Appendix \ref{collNumericsSpacing}.  Again this is consistent with all instances of $\hat{H}_{n}^{(JW)}$ having geometric symmetries, that is the highly degenerate matrix $\prod_{j=1}^{n}\sigma_{j}^{(3)}$ commutes with all such instances. 

\subsection{Gaussian orthogonal spacing statistics for \texorpdfstring{$\hat{H}_{n}$}{Hn} for even values of \texorpdfstring{$n$}{n}}\label{Hevenspacigns}
The nearest-neighbour level spacing histogram for the unfolded eigenvalues of $s=2^{7}$ samples from the ensemble $\hat{H}_{12}$ is shown in Figure \ref{SpacingsHeven}.  Histograms for $n=2,\dots,13$ are shown in Appendix \ref{collNumericsSpacing} for completeness, with all histograms for even values of $n=2,\dots,12$ displaying similar characteristics.  A close resemblance to the approximate limiting standard Gaussian orthogonal (Wigner) spacing distribution is seen.  The slight discrepancy in the numerically generated data from this, closely matches the true limiting standard GOE nearest-neighbour level spacing distribution \cite[\p14]{Mehta}.

\begin{figure}
  \centering
  \input{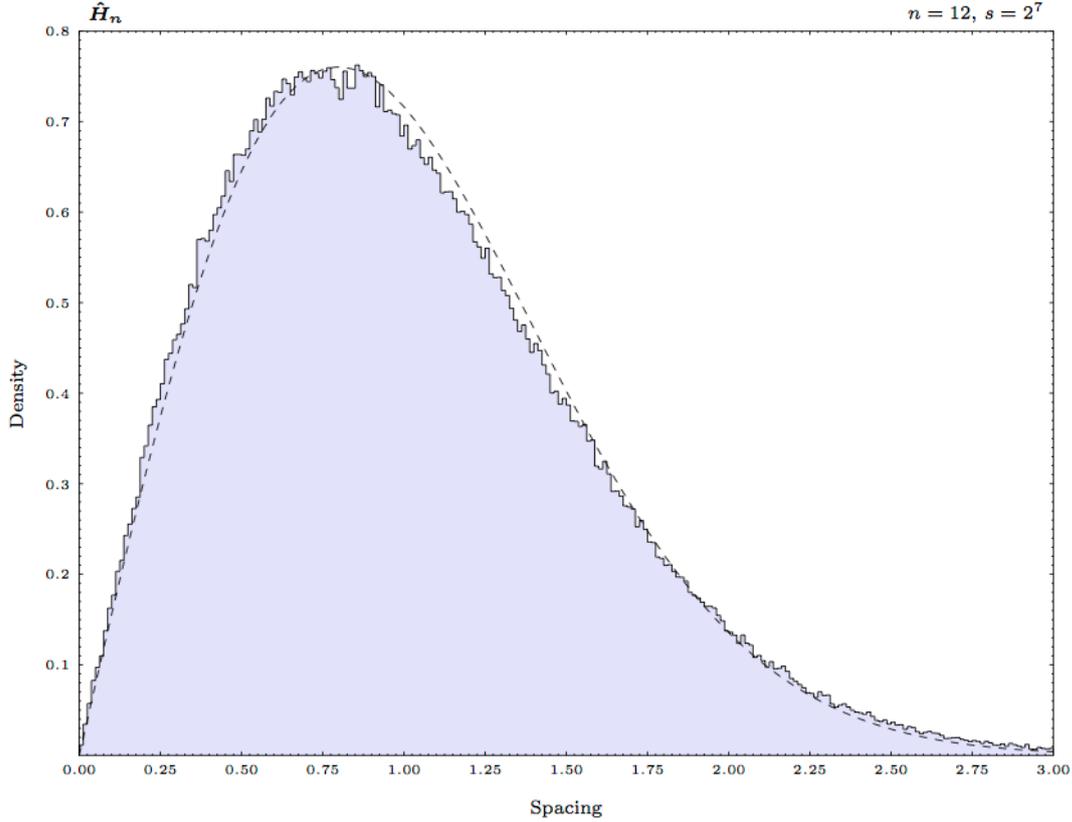}
  \caption[Nearest-neighbour level spacing histogram for $\hat{H}_{12}$]{The normalised nearest-neighbour level spacing histogram for $\hat{H}_{12}$.  The spacings used were numerically obtained from each of the unfolded spectra of $s=2^{7}$ samples from the ensemble $\hat{H}_{12}$ and scaled to have unit mean.  The unfolding was with respect to the average numerical spectral density on the interval $[-3,3]$ calculated from the same sample.  The dashed line gives the approximate limiting standard GOE (Wigner) nearest-neighbour level spacing distribution.}
  \label{SpacingsHeven}
\end{figure}

In Theorem \ref{eigenNonInv} of Section \ref{SingleEnt} (page \pageref{eigenNonInv}) it was seen that any instance $H_{n}$ of the ensemble $\hat{H}_{n}$, has the symmetry $SH_{n}=\overline{H_{n}}S$ for $S=\prod_{j=1}^{n}\sigma_{j}^{(2)}$.  To re-express this symmetry let $K$ be the operator that acts on any vector in the Hilbert space of $n$ qubits by taking the complex conjugate of the coefficients of that vector expressed in the standard basis (see Section \ref{stndBasis}).  That is, for the complex coefficients $c_{\boldsymbol{x}}$, arbitrary vector $|\phi\rangle$ and the standard basis elements $|\boldsymbol{x}\rangle$ for the multi-indices $\boldsymbol{x}\in\{0,1\}^{n}$,
\begin{equation}
  K|\phi\rangle=K\sum_{\boldsymbol{x}}c_{\boldsymbol{x}}|\boldsymbol{x}\rangle=\sum_{\boldsymbol{x}}\overline{c_{\boldsymbol{x}}}|\boldsymbol{x}\rangle\equiv\overline{|\phi\rangle}
\end{equation}
Here the bar denotes complex conjugation of a complex number, or complex conjugation of the vector or matrix elements in the standard basis.

The matrix $\overline{H_{n}}$ can now be expressed as $KH_{n}K$.  This may be checked by confirming its action on the arbitrary state $|\phi\rangle$, that is
\begin{equation}\label{4.041}
  KH_{n}K|\phi\rangle=KH_{n}\overline{|\phi\rangle}=\overline{H_{n}\overline{|\phi\rangle}}=\overline{H_{n}}|\phi\rangle
\end{equation}
Therefore the symmetry $SH_{n}=\overline{H_{n}}S$ can be rewritten as $SH_{n}=KH_{n}KS$.  As $K^{2}=I$ by definition, this symmetry is then equivalent to $KSH_{n}=H_{n}KS$, that is $H_{n}$ commutes with the operator $\Theta=KS$.  This operator is anti-unitary by definition, as for arbitrary vectors $|\phi_{1}\rangle$ and $|\phi_{2}\rangle$,
\begin{equation}
  \left(\Theta|\phi_{1}\rangle,\Theta|\phi_{2}\rangle\right)
  =\left(KS|\phi_{1}\rangle,KS|\phi_{2}\rangle\right)
  =\left(\overline{S|\phi_{1}\rangle},\overline{S|\phi_{2}\rangle}\right)
  =\overline{\left(S|\phi_{1}\rangle,S|\phi_{2}\rangle\right)}
  =\overline{\left(|\phi_{1}\rangle,|\phi_{2}\rangle\right)}
\end{equation}
for the standard complex Euclidean inner-product denoted $(\cdot,\cdot)$.  This is seen by applying the definition of $K$ and the unitarity of $S$ which implies that $\left(S|\phi_{1}\rangle,S|\phi_{2}\rangle\right)=\left(|\phi_{1}\rangle,|\phi_{2}\rangle\right)$.

For even values of $n$ the anti-unitary symmetry $\Theta$ then has the property that
\begin{equation}\label{ThetaSym}
  \Theta^{2}=KSKS=\overline{S}S=(-1)^{n}I=I
\end{equation}
as $KSK=\overline{S}$ as seen similarly to (\ref{4.041}).  The presence of GOE nearest-neighbour level spacing statistics is then consistent with Dyson's three fold way, that is the presence of this anti-unitary symmetry which squares to $I$ for all ensemble members.

\subsection{Gaussian unitary spacing statistics for \texorpdfstring{$\hat{H}_{n}^{(local)}$}{Hn(local)}}
The nearest-neighbour level spacing histogram for the unfolded eigenvalues of $s=2^{6}$ samples from the ensemble $\hat{H}_{13}^{(local)}$ is shown in Figure \ref{SpacingsHlocal}.  Histograms for $n=2,\dots,13$ are shown in Appendix \ref{collNumericsSpacing} for completeness, with all histograms for $n=3,\dots,13$ displaying similar characteristics.  A close resemblance to the approximate limiting standard GUE nearest-neighbour level spacing distribution is seen.

\begin{figure}
  \centering
  \input{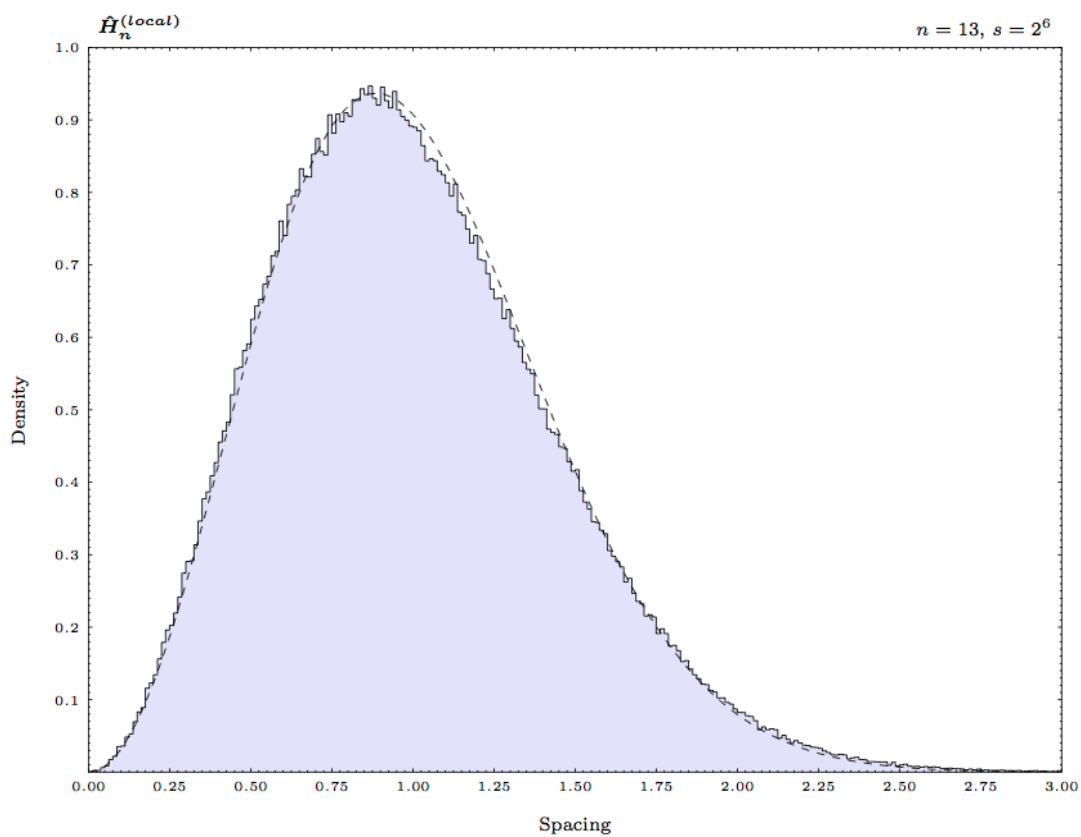}
  \caption[Nearest-neighbour level spacing histogram for $\hat{H}_{13}^{(local)}$]{The normalised nearest-neighbour level spacing histogram for $\hat{H}_{13}^{(local)}$.  The spacings used were numerically obtained from each of the unfolded spectra of $s=2^{6}$ samples from the ensemble $\hat{H}_{13}^{(local)}$ and scaled to have unit mean.  The unfolding was with respect to the average numerical spectral density on the interval $[-3,3]$ calculated from the same sample.  The dashed line gives the approximate limiting standard GUE nearest-neighbour level spacing distribution.}
  \label{SpacingsHlocal}
\end{figure}

The presence of the local terms in all most all samples of $\hat{H}_{n}^{(local)}$ breaks the anti-unitary symmetry present in samples of $\hat{H}_{n}$, discussed in the last section.  Therefore the appearance of GUE nearest-neighbour level spacing statistics is again consistent with Dyson's three fold way as no symmetries, apart from the inherent Hermitian symmetry, are present in all most all matrices from $\hat{H}_{n}^{(local)}$.

\subsection{Gaussian symplectic spacing statistics for \texorpdfstring{$\hat{H}_{n}$}{Hn} for odd values of \texorpdfstring{$n$}{n}}
The nearest-neighbour level spacing histogram for the unfolded eigenvalues of $s=2^{6}$ samples from the ensemble $\hat{H}_{13}$ is shown in Figure \ref{SpacingsHodd}.  Histograms for $n=2,\dots,13$ are shown in Appendix \ref{collNumericsSpacing} for completeness, with all histograms for odd values of $n=5,\dots,13$ displaying similar characteristics.  A close resemblance to the approximate limiting standard GSE nearest-neighbour level spacing distribution, rescaled to have mean $2$ and area $\frac{1}{2}$, is seen.

\begin{figure}
  \centering
  \input{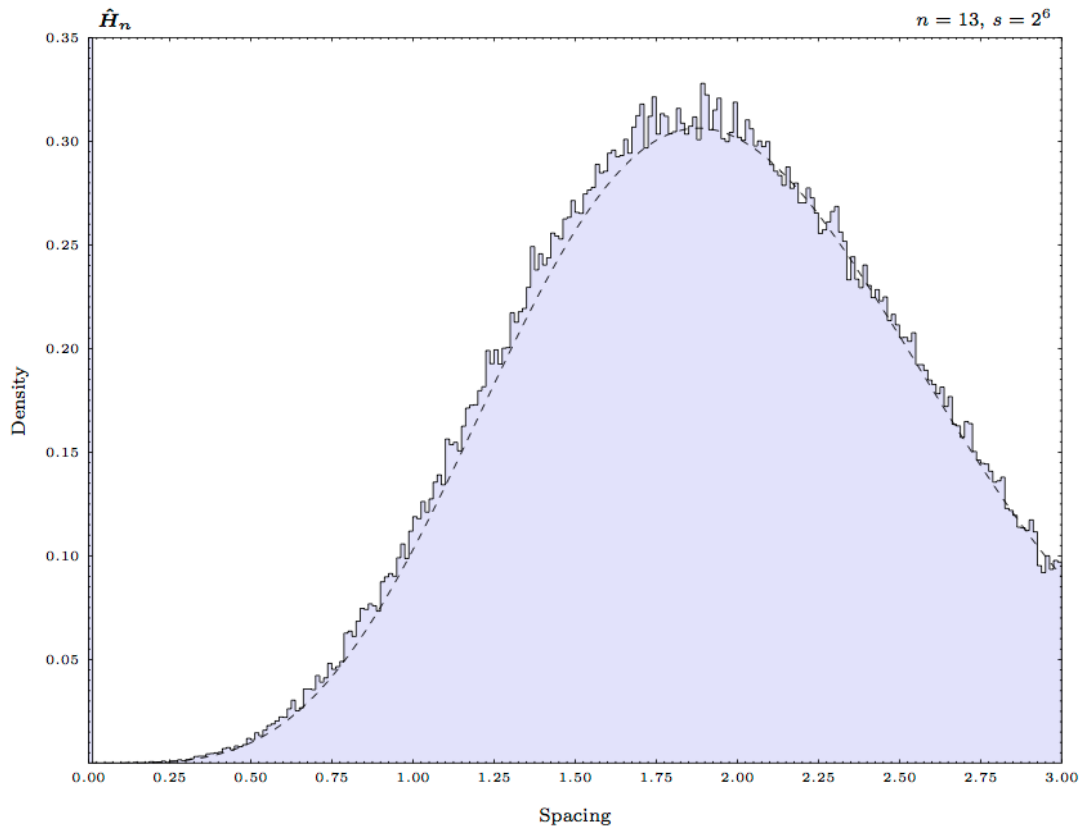}
  \caption[Nearest-neighbour level spacing histogram for $\hat{H}_{13}$]{The normalised nearest-neighbour level spacing histogram for $\hat{H}_{13}$.  The spacings used were numerically obtained from each of the unfolded spectra of $s=2^{6}$ samples from the ensemble $\hat{H}_{13}$ and scaled to have unit mean.  The unfolding was with respect to the average numerical spectral density on the interval $[-3,3]$ calculated from the same sample.  The dashed line gives the approximate limiting standard GSE nearest-neighbour level spacing distribution rescaled to have mean $2$ and area $\frac{1}{2}$.}
  \label{SpacingsHodd}
\end{figure}

As with the even values of $n$ in Section \ref{Hevenspacigns}, each member of the ensemble $\hat{H}_{n}$ has the anti-unitary symmetry $\Theta$.  In the case of odd values of $n$ though, $\Theta^{2}=-I$, as seen in equation (\ref{ThetaSym}).  The presence of GSE nearest-neighbour level spacing statistics is then consistent with Dyson's three fold way, that is the presence of this anti-unitary symmetry which squares to $-I$ for all ensemble members.

The peak seen at zero also results from this anti-unitary symmetry.  Each eigenvalue in the spectrum of a sample of $\hat{H}_{n}$ for odd values of $n\geq3$ is at least doubly degenerate as a result.  These are the Kramers' degeneracies as described in Section \ref{Kramer}.
\section{The joint spectral density}\label{JDOS}
In this section a conjecture for the joint spectral density of a large class of Hamiltonians (including those of the previous non-translationally-invariant qubit chains) will be heuristically outlined.

It has been seen in Section \ref{PauliMatrixBasis} that any $2^{n}\times2^{n}$ Hermitian matrix may be parametrised in the form
\begin{equation}
  \sum_{j=1}^{4^{n}}\alpha_{j}P_{j}
\end{equation}
for any arbitrary labelling of the $4^{n}$ Hermitian basis matrices
\begin{equation}
  \{P_{j}\,\,|\,\,j=1,\dots,4^{n}\}=\{\sigma^{( a_{1} )}\otimes\dots\otimes\sigma^{( a_{n} )}\,\,|\,\,0\leq a_{1},\dots,a_{n}\leq3\}                           
\end{equation}
and $4^{n}$ real parameters $\alpha_{j}$ for $j=1,\dots,4^{n}$.  Given such a fixed labelling, the basis $\{P_{j}\}$ is orthonormal with respect to the scaled Hilbert-Schmidt inner-product $(A,B)=2^{-n}\Tr \left(AB\right)$ on Hermitian matrices.  In this section, the ensembles
\begin{equation}
  \hat{H}_{n}^{(\r)}=\sum_{j=1}^{\r}\hat{\alpha}_{j}P_{j}\qquad\qquad\hat{\alpha}_{j}\sim\mathcal{N}\left(0,\frac{1}{\r}\right) \iid
\end{equation}
for $\r=1,\dots,4^{n}$ will be considered. In particular, this includes the GUE and the ensembles $\hat{H}_{n}$ and $\hat{H}_{n}^{(local)}$ as special cases.

\subsection{The joint spectral density via the HCIZ integral}
The joint spectral density or the $2^{n}$-point correlation function, as defined in Section \ref{PointCorrFunc}, of the ensemble $\hat{H}_{n}^{(\r)}$ is conjectured to be:

\begin{conjecture}[Joint spectral density]
  The joint spectral density, that is the joint probability density of the eigenvalues or the $2^{n}$-point correlation function, $\hat{\rho}_{n,2^{n}}^{(\r)}(\lambda_{1},\dots,\lambda_{2^{n}})=\hat{\rho}_{n,2^{n}}^{(\r)}(\boldsymbol{\lambda})$, of the ensemble $\hat{H}_{n}^{(\r)}$ is
  \begin{equation}
    C\e^{-\frac{\r\boldsymbol{\lambda}^{2}}{2^{n+1}}}\Delta^2(\boldsymbol{\lambda})\int_{\mathbb{R}^{4^{n}-\r}}\frac{\det_{1\leq j,k\leq 2^{n}}\left(\e^{\im\lambda_{j}\mu_{k}(\boldsymbol{\beta})}\right)}{\Delta(\boldsymbol{\lambda})\Delta({\boldsymbol{\mu}(\boldsymbol{\beta})})}\boldsymbol{\di}\boldsymbol{\beta}
  \end{equation}
where
\begin{equation}
  C=\frac{2^{\frac{n4^{n}}{2}}\r^{\frac{\r}2{}}(2\pi)^{\frac{\r}{2}}\im^{\frac{2^{n}}{2}}}  {2^{n}!2^{n\r}(2\pi)^{\frac{2^{n}}{2}}(2\pi)^{\frac{4^{n}}{2}}\im^{\frac{4^{n}}{2}}}\qquad
  \Delta(\boldsymbol{\lambda})=\prod_{1\leq j<k\leq2^{n}}(\lambda_{k}-\lambda_{j})
\end{equation}
and $\boldsymbol{\mu}(\boldsymbol{\beta})=(\mu_{1}(\boldsymbol{\beta}),\dots,\mu_{2^{n}}(\boldsymbol{\beta}))$ are the eigenvalues of the matrix
\begin{equation}
  B=\sum_{\r<j\leq 4^{n}}\beta_{j}P_{j}
\end{equation}
and $\boldsymbol{\beta}$ is the vector of all the real parameters $\beta_{j}$ present in $B$.  For $\r=4^n$, the integral over $\mathbb{R}^0$ is taken as a unit factor.
\end{conjecture}
 
\subsection{Heuristic argument}
The $\r$ (real) dimensional probability measure on the ensemble $\hat{H}_{n}^{(\r)}$ is explicitly written as the product of the one dimensional probability measures of the random variables $\hat{\alpha}_{j}$ for $1\leq j\leq\r$.  Explicitly, for the parameters $\alpha_{j}$ relating to the random variables $\hat{\alpha}_{j}$, the ensemble's measure is written as \begin{equation}\label{C1}
  \boldsymbol{\di}\hat{\mu}_{n}^{(\r)}(\alpha_{1},\dots,\alpha_{\r})=C_{1}\prod_{j=1}^{\r}\e^{-\frac{\r\alpha_{j}^{2}}{2}}\di\alpha_{j},\qquad\qquad C_{1}=\left(\frac{\r}{2\pi}\right)^{\frac{\r}{2}}
\end{equation}

The heuristic argument will proceed by first embedding this probability measure in the $4^{n}$ (real) dimensional space of generic Hermitian matrices, $\mathbb{R}^{4^{n}}$.  This embedded probability measure will then be seen to be equivalent to the probability measure on the Gaussian unitary ensemble multiplied by a sequence of Dirac delta distributions.  To this, the canonical diagonalisation change of variables (see Section \ref{GUE}) may be applied.  The unitary (eigenstates) degrees of freedom may then be (heuristically) integrated out with the application of the Harish-Chandra-Itzykson-Zuber (HCIZ) integral to leave the conjectured joint distribution of eigenvalues.

\subsubsection{Embedding in $\mathbb{R}^{4^{n}}$}
For the real parameters $\alpha_{j}$ for $\r<j\leq4^n$, define the matrix $H_{n}^{(\lnot \r)}$ and probability measure $\boldsymbol{\di}\hat{\mu}_{n}^{(\lnot\r)}$ formally by
\begin{align}
  H_{n}^{(\lnot \r)}&=\sum_{\r<j\leq 4^{n}}\alpha_{j}P_{j},\qquad\qquad
  \boldsymbol{\di}\hat{\mu}_{n}^{(\lnot \r)}=\prod_{\r<j\leq 4^{n}}\delta(\alpha_{j})\e^{-\frac{\r\alpha_{j}^{2}}{2}}\di\alpha_{j}
\end{align}
so that with probability one, $H_{n}^{(\lnot \r)}$ is the zero matrix, with respect to the measure $\boldsymbol{\di}\hat{\mu}_{n}^{(\lnot\r)}$.  The $\r$ (real) dimensional probability measure $\boldsymbol{\di}\hat{\mu}_{n}^{(\r)}$ may be embedded in $\mathbb{R}^{4^{n}}$ by forming the matrix
\begin{equation}
  H_{n}^{(\r+)}=H_{n}^{(\r)}+H_{n}^{(\lnot \r)}
\end{equation}
and for $\boldsymbol{\alpha}=(\alpha_{1},\dots,\alpha_{4^{n}})$ the probability measure
\begin{equation}
  \boldsymbol{\di}\hat{\mu}_{n}^{(\r+)}(\boldsymbol{\alpha})=\boldsymbol{\di}\hat{\mu}_{n}^{(\r)}\boldsymbol{\di}\hat{\mu}_{n}^{(\lnot \r)}
\end{equation}
With probability one, the matrix $H_{n}^{(\r+)}$ and $H_{n}^{(\r)}$ coincide and share eigenstatistics, with respect to the measure $\boldsymbol{\di}\hat{\mu}_{n}^{(\r+)}$.  The probability measure $\boldsymbol{\di}\hat{\mu}_{n}^{(\r+)}(\boldsymbol{\alpha})$ may be equivalently written as
\begin{equation}
  C_{1}\e^{-\frac{\r}{2^{n+1}}\Tr \left({H_{n}^{(\r+)}}^{2}\right)}\left(\prod_{\r<j\leq4^{n}}\delta\left(\frac{1}{2^{n}}\Tr\left( H_{n}^{(\r+)}P_{j}\right)\right)\right)\boldsymbol{\di}\boldsymbol{\alpha}
\end{equation}
where
$\Tr \left({H_{n}^{(\r+)}}^{2}\right)=2^{n}\sum_{j=1}^{4^{n}}\alpha_{j}^{2}$, $\Tr \left(H_{n}^{(\r+)}P_{j}\right)=2^{n}\alpha_{j}$ by the orthogonality of the $P_{j}$ under the Hilbert-Schmidt inner-product, and where $\boldsymbol{\di}\boldsymbol{\alpha}=\di\alpha_{1}\dots\di\alpha_{4^{n}}$.

\subsubsection{Change of variables from $\boldsymbol{\alpha}$ to matrix elements in the standard basis}
Any $2^{n}\times2^{n}$ Hermitian matrix may be parametrised in terms of its elements, in some basis.  In particular, $H_{n}^{(\r+)}$ may be parametrised in the form
\begin{equation}
  H_{n}^{(\r+)}=\sum_{d=1}^{2^{n}}h_{d,d}R^{(d,d)}+\sum_{1\leq k<l\leq2^{n}}h_{k,l}R^{(k,l)}+\sum_{1\leq k<l\leq2^{n}}h^\prime_{k,l}I^{(k,l)}
\end{equation}
where the $2^{n}\times2^{n}$ Hermitian matrices $R^{(d,d)}$, $R^{(k,l)}$ and $I^{(k,l)}$ in this expression have the elements
\begin{align}
	\left(R^{(d,d)}\right)_{a,b}&=\delta_{a,d}\delta_{b,d}\nonumber\\
	\left(R^{(k,l)}\right)_{a,b}&=\delta_{a,k}\delta_{b,l}+\delta_{a,l}\delta_{b,k}\nonumber\\
	\left(I^{(k,l)}\right)_{a,b}&=\im\delta_{a,k}\delta_{b,l}-\im\delta_{a,l}\delta_{b,k}
\end{align}
The corresponding real parameters $h_{d,d}$, $h_{k,l}$ and $h^\prime_{k,l}$ are then given by
\begin{align}
  h_{d,d}&=\Tr \left(H_{n}^{(\r+)}R^{(d,d)}\right)=\left(H_{n}^{(\r+)}\right)_{d,d}\nonumber\\
  h_{k,l}&=\frac{1}{2}\Tr \left(H_{n}^{(\r+)}R^{(k,l)}\right)=\Rp \left(H_{n}^{(\r+)}\right)_{k,l}\nonumber\\
  h^\prime_{k,l}&=\frac{1}{2}\Tr \left(H_{n}^{(\r+)}I^{(k,l)}\right)=\Ip \left(H_{n}^{(\r+)}\right)_{k,l}
\end{align}

The change of variables from $\boldsymbol{\alpha}=(\alpha_{1},\dots,\alpha_{4^n})$ to $\boldsymbol{h}$ (the vector of all the $h_{d,d}$, $h_{k,l}$ and $h^\prime_{k,l}$) in the probability measure $\boldsymbol{\di}\hat{\mu}_{n}^{(\r+)}(\boldsymbol{\alpha})$ will now be performed.  As $R^{(d,d)}$, $R^{(k,l)}$ and $I^{(k,l)}$ are Hermitian, they admit the expansion
\begin{align}
  R^{(d,d)}=\sum_{j=1}^{4^{n}}c_{j,d,d}P_{j},\qquad\qquad
  R^{(k,l)}=\sum_{j=1}^{4^{n}}c_{j,k,l}P_{j},\qquad\qquad
  I^{(k,l)}=\sum_{j=1}^{4^{n}}c^\prime_{j,k,l}P_{j}
\end{align}
for some real parameters $c_{j,d,d}$, $c_{j,k,l}$ and $c^{\prime}_{j,k,l}$.  Then by definition
\begin{align}
  h_{d,d}&=\Tr \left(H_{n}^{(\r+)}R^{(d,d)}\right)=\sum_{j=1}^{4^{n}}c_{j,d,d}\Tr \left(H_{n}^{(\r+)}P_{j}\right)=2^{n}\sum_{j=1}^{4^{n}}c_{j,d,d}\alpha_{j}\nonumber\\
  h_{k,l}&=\frac{1}{2}\Tr \left(H_{n}^{(\r+)}R^{(k,l)}\right)=\frac{1}{2}\sum_{j=1}^{4^{n}}c_{j,k,l}\Tr \left(H_{n}^{(\r+)}P_{j}\right)=2^{n-1}\sum_{j=1}^{4^{n}}c_{j,k,l}\alpha_{j}\nonumber\\
  h^\prime_{k,l}&=\frac{1}{2}\Tr \left(H_{n}^{(\r+)}I^{(k,l)}\right)=\frac{1}{2}\sum_{j=1}^{4^{n}}c^\prime_{j,k,l}\Tr \left(H_{n}^{(\r+)}P_{j}\right)=2^{n-1}\sum_{j=1}^{4^{n}}c^\prime_{j,k,l}\alpha_{j}
\end{align}
Therefore the transform $\boldsymbol{\alpha}\to\boldsymbol{h}$ is linear and its Jacobian $J$ is  a constant.  The transformed probability measure then reads
\begin{equation}
  JC_{1}\e^{-\frac{\r}{2^{n+1}}\Tr \left({H_{n}^{(\r+)}}^{2}\right)}\left(\prod_{\r<j\leq 4^n}\delta\left(\frac{1}{2^{n}}\Tr \left(H_{n}^{(\r+)}P_{j}\right)\right)\right)\boldsymbol{\di}\boldsymbol{h}
\end{equation}
where
\begin{equation}
 \boldsymbol{\di}\boldsymbol{h}=\prod_{d=1}^{2^{n}}\di h_{d,d}\prod_{1\leq k<l\leq2^{n}}\di h_{k,l}\di h^\prime_{k,l}
\end{equation}

To determine the constant $J$, the case $\r=4^{n}$ will be considered as the transform, and therefore $J$, is independent of $\r$.  The resulting probability measure must be normalised, that is
\begin{align}
  1&=\int_{\mathbb{R}^{4^{n}}}JC_{1}\e^{-\frac{2^{n}}{2}\Tr \left({H_{n}^{(4^{n}+)}}^{2}\right)}\boldsymbol{\di}\boldsymbol{h}\nonumber\\
  &=\int_{\mathbb{R}^{4^{n}}}JC_{1}\e^{-\frac{2^{n}}{2}\left(\sum_{d=1}^{2^{n}}h_{d,d}^{2}+2\sum_{1\leq k<l\leq2^{n}}\left(h_{k,l}^{2}+{h^\prime_{k,l}}^{2}\right)\right)}\boldsymbol{\di}\boldsymbol{h}\nonumber\\
  &=J
    \left(\frac{4^{n}}{2\pi}\right)^{\frac{4^{n}}{2}}
    \left(\frac{2\pi}{2^{n}}\right)^{\frac{2^{n}}{2}}
    \left(\frac{2\pi}{2^{n+1}}\right)^{\frac{4^{n}-2^{n}}{2}}
\end{align}
by computing the $4^{n}$ independent Gaussian integrals (see Appendix \ref{Int}) and substituting the definition of $C_{1}$ from equation (\ref{C1}).  It is then conclude that
\begin{equation}
  J=\frac{1}{2^{n4^{n}}}
      \frac{2^{\frac{n2^{n}}{2}}}{1}
      \frac{2^{\frac{n4^{n}}{2}}2^{\frac{4^{n}}{2}}}{2^\frac{n2^{n}}{2}2^\frac{2^{n}}{2}}
    =\frac{2^{\frac{4^{n}}{2}}}{2^{\frac{n4^{n}}{2}}2^{\frac{2^{n}}{2}}}
\end{equation}

\subsubsection{Change of variables from $\boldsymbol{h}$ to spectral parameters}
From the canonical case of the GUE described in Section \ref{GUE}, the change of variables from $\boldsymbol{h}$, considered as the elements of a generic Hermitian matrix, to that of its (unordered) eigenvalues $\boldsymbol{\lambda}=(\lambda_{1},\dots,\lambda_{2^{n}})$ and a unitary matrix $U$ whose columns are corresponding eigenstates, implies the following volume element transform
\begin{equation}
  \boldsymbol{\di}\boldsymbol{h}\to C_{2}\Delta^{2}(\boldsymbol{\lambda})\boldsymbol{\di}\boldsymbol{\lambda}\boldsymbol{\di}U
\end{equation}
where $C_{2}$ is some constant, $\boldsymbol{\di}\boldsymbol{\lambda}=\di\lambda_{1},\dots\di\lambda_{2^{n}}$ and $\boldsymbol{\di}U$ is the Haar measure over the unitary group.  As described in Section \ref{GUE}, this diagonalisation provides a $4^n+2^n$ real dimensional parametrisation of the $4^n$ real dimensional space of Hermitian matrices, that is a parametrisation which contains $2^n$ redundant parameters (eigenstate phases).  Further uniform over-coverings of the space of Hermitian matrices result from other freedoms in the diagonalisation, such as eigenvalue ordering, as described in Section \ref{GUE}.

Writing $H_{n}^{(\r+)}=U\Lambda(\boldsymbol{\lambda}) U^\dagger$ for the diagonal matrix $\Lambda(\boldsymbol{\lambda})=\text{diag}(\lambda_{1},\dots,\lambda_{2^{n}})$, the probability measure $\boldsymbol{\di}\hat{\mu}_{n}^{(\r+)}$ is therefore transformed to
\begin{equation}\label{4.127}
  JC_{1}C_{2}\Delta^{2}(\boldsymbol{\lambda})\e^{-\frac{\r\boldsymbol{\lambda}^{2}}{2^{n+1}}}\prod_{\r<j\leq4^{n}}\delta\left(\frac{1}{2^{n}}\Tr\left( U\Lambda U^\dagger P_{j}\right)\right)\boldsymbol{\di}\boldsymbol{\lambda}\boldsymbol{\di}U
\end{equation}
as $\Tr\left({H_{n}^{(\r+)}}^{2}\right)=\lambda_{1}^{2}+\dots+\lambda_{2^{n}}^{2}\equiv\boldsymbol{\lambda}^{2}$.

The constant $C_{2}$ can be again determined by considering the case $\r=4^{n}$ as $C_{2}$ is independent of $\r$.  In this case the transformed probability measure reads
\begin{equation}\label{4.4.23}
  JC_{1}C_{2}\Delta^{2}(\boldsymbol{\lambda})\e^{-\frac{2^{n}\boldsymbol{\lambda}^{2}}{2}}\boldsymbol{\di}\boldsymbol{\lambda}\boldsymbol{\di}U
\end{equation}
and $C_{2}$ should be chosen such that this probability measure is normalised.  The integration over $U$ can be done trivially as $\boldsymbol{\di}U$ is a normalised probability measure.  A special case of Selberg's integral, Mehta's integral, can be used to compute the remaining integration.  Mehta's integral \cite[\p321]{Mehta} reads
\begin{equation}
  \int_{\mathbb{R}^{2^{n}}}\e^{-\frac{\boldsymbol{\lambda}^{2}}{2}}\Delta^{2}(\boldsymbol{\lambda})\boldsymbol{\di}\boldsymbol{\lambda}
  =(2\pi)^{\frac{2^{n}}{2}}\prod_{j=1}^{2^{n}}j!
\end{equation}
Rescaling $\boldsymbol{\lambda}\to2^{\frac{n}{2}}\boldsymbol{\lambda}$ then yields that
\begin{equation}
  \int_{\mathbb{R}^{2^{n}}}\e^{-\frac{2^{n}\boldsymbol{\lambda}^{2}}{2}}2^{\frac{n4^{n}-n2^{n}}{2}}\Delta^{2}(\boldsymbol{\lambda})2^{\frac{n2^{n}}{2}}\boldsymbol{\di}\boldsymbol{\lambda}
  =(2\pi)^{\frac{2^{n}}{2}}\prod_{j=1}^{2^{n}}j!
\end{equation}
so that the normalisation condition on (\ref{4.4.23}) reads
\begin{equation}
  1=\int_{\mathbb{R}^{2^{n}}}JC_{1}C_{2}\Delta^{2}(\boldsymbol{\lambda})\e^{-\frac{2^{n}\boldsymbol{\lambda}^{2}}{2}}\boldsymbol{\di}\boldsymbol{\lambda}
  =\frac{2^{\frac{4^{n}}{2}}}{2^{\frac{n4^{n}}{2}}2^{\frac{2^{n}}{2}}}
    \left(\frac{4^{n}}{2\pi}\right)^{\frac{4^{n}}{2}}
    C_{2}
    \frac{(2\pi)^{\frac{2^{n}}{2}}}{2^{\frac{n4^{n}}{2}}}
    \prod_{j=1}^{2^{n}}j!
\end{equation}
This fixes $C_{2}$ to be
\begin{equation}
  C_{2}=\frac{2^{\frac{n4^{n}}{2}}2^{\frac{2^{n}}{2}}}{2^{\frac{4^{n}}{2}}}
      \frac{2^{\frac{4^{n}}{2}}\pi^{\frac{4^{n}}{2}}}{2^{n4^{n}}}
      \frac{2^{\frac{n4^{n}}{2}}}{2^{\frac{2^{n}}{2}}\pi^{\frac{2^{n}}{2}}}
      \frac{1}{\prod_{j=1}^{2^{n}}j!}
     =\frac{\pi^{\frac{4^{n}}{2}}}{\pi^{\frac{2^{n}}{2}}\prod_{j=1}^{2^{n}}j!}
\end{equation}

\subsubsection{The Harish-Chandra-Itzykson-Zuber (HCIZ) integral}
The joint probability measure for the eigenvalues of $\hat{H}_{n}^{(\r)}$ is then given by the marginal distribution found by integrating over the unitary group (characterising the eigenstates) of the measure in (\ref{4.127}).  To perform this integration the Dirac delta distribution is first formally expressed using the Fourier identity
\begin{equation}
  \delta(x)=\frac{1}{2\pi}\int_{\mathbb{R}}\e^{\im \beta x}\di \beta
\end{equation}
so that the probability measure (\ref{4.127}) reads
\begin{equation}
  JC_{1}C_{2}\Delta^{2}(\boldsymbol{\lambda})\e^{-\frac{\r\boldsymbol{\lambda}^{2}}{2^{n+1}}}\left(\prod_{\r<j\leq4^{n}}\frac{1}{2\pi}\int_{\mathbb{R}}\e^{\frac{\im\beta_{j}}{2^{n}}\Tr \left(U\Lambda U^\dagger P_{j}\right)}\di \beta_{j}\right)\boldsymbol{\di}\boldsymbol{\lambda}\boldsymbol{\di}U
\end{equation}
Collecting the powers of the exponentials and rescaling $\beta_{j}\to2^{n}\beta_{j}$ reduces this to
\begin{equation}\label{large2}
  JC_{1}C_{2}C_{3}\Delta^{2}(\boldsymbol{\lambda})\e^{-\frac{\r\boldsymbol{\lambda}^{2}}{2^{n+1}}}\int_{\mathbb{R}^{4^{n}-\r}}\e^{\im\Tr \left(U\Lambda U^\dagger B\right)}\boldsymbol{\di}\boldsymbol{\beta}\boldsymbol{\di}\boldsymbol{\lambda}\boldsymbol{\di}U
\end{equation}
where
\begin{equation}
  C_{3}=\left(\frac{2^{n}}{2\pi}\right)^{4^{n}-\r}\qquad\qquad
  B=B(\boldsymbol{\beta})=\sum_{\r<j\leq4^{n}}\beta_{j}P_{j}\qquad\qquad
  \boldsymbol{\di}\boldsymbol{\beta}=\prod_{\r<j\leq4^{n}}\di\beta_{j}
\end{equation}

After changing the order of integration (heuristically), the Harish-Chandra-Itzykson-Zuber (HCIZ) integral can be used to compute the integral of (\ref{large2}) over the unitary group.  For the $2^{n}\times2^{n}$ Hermitian matrices $X$ and $Y$ with eigenvalues $\boldsymbol{\lambda}=(\lambda_{1},\dots,\lambda_{2^{n}})$ and $\boldsymbol{\mu}=(\mu_{1},\dots,\mu_{2^{n}})$ and complex parameter $t\neq0$, the HCIZ integral reads
\begin{equation}
  \int\e^{t\Tr \left(UXU^\dagger Y\right)}\boldsymbol{\di}U=\frac{\det_{1\leq j,k\leq 2^{n}}(\e^{t\lambda_{j}\mu_{k}})}{t^{\frac{4^{n}-2^{n}}{2}}\Delta(\boldsymbol{\lambda})\Delta(\boldsymbol{\mu})}\prod_{l=1}^{2^{n}-1}l!
\end{equation}
Setting $t=\im$ allows the integration over the unity group of (\ref{large2}) to be computed.  Setting $X=\Lambda$ and $Y=B$, gives the $4^{n}$-point correlation function for the ensemble $\hat{H}_{n}^{(\r)}$ as
\begin{equation}
  \hat{\rho}_{n,4^{n}}^{(\r)}(\boldsymbol{\lambda})=C\Delta^{2}(\boldsymbol{\lambda})\e^{-\frac{\r\boldsymbol{\lambda}^{2}}{2^{n+1}}}\int_{\mathbb{R}^{4^{n}-\r}}\frac{\det_{1\leq j,k\leq 2^{n}}(\e^{\im\lambda_{j}\mu_{k}(\boldsymbol{\beta})})}{\Delta(\boldsymbol{\lambda})\Delta(\boldsymbol{\mu}(\boldsymbol{\beta}))}\boldsymbol{\di}\boldsymbol{\beta}
\end{equation}
where
\begin{align}
  C&=JC_{1}C_{2}C_{3}\frac{\im^{\frac{2^{n}}{2}}\prod_{l=1}^{2^{n}-1}l!}{\im^{\frac{4^{n}}{2}}}\nonumber\\
    &=\frac{2^{\frac{4^{n}}{2}}}{2^{\frac{n4^{n}}{2}}2^{\frac{2^{n}}{2}}}
      \frac{\r^{\frac{\r}{2}}}{(2\pi)^{\frac{\r}{2}}}
      \frac{\pi^{\frac{4^{n}}{2}}}{\pi^{\frac{2^{n}}{2}}\prod_{j=1}^{2^{n}}j!}
      \frac{(2\pi)^{\r}2^{n4^{n}}}{(2\pi)^{4^{n}}2^{n\r}}
      \frac{\im^{\frac{2^{n}}{2}}\prod_{l=1}^{2^{n}-1}l!}{\im^{\frac{4^{n}}{2}}}\nonumber\\
    &=\frac{2^{\frac{n4^{n}}{2}}\r^{\frac{\r}{2}}(2\pi)^{\frac{\r}{2}}\im^{\frac{2^{n}}{2}}}  {2^{n}!2^{n\r}(2\pi)^{\frac{2^{n}}{2}}(2\pi)^{\frac{4^{n}}{2}}\im^{\frac{4^{n}}{2}}}
\end{align}
as claimed (where the integral over $\mathbb{R}^{0}$ is defined to be unity).

\subsection{Example: The joint spectral density for \texorpdfstring{$\hat{H}_{2}$}{H2}}
To support the previous conjecture, the $1$-point correlation function (spectral density) for $\hat{H}_{2}$ will be calculated in a similar Heuristic manner.  By the definition in Section \ref{Numerics},
\begin{equation}
  \hat{H}_{2}=\sum_{a,b=1}^{3}\left(\hat{\alpha}_{a,b,1}\sigma^{(a)}\otimes\sigma^{(b)}+\hat{\alpha}_{a,b,2}\sigma^{(b)}\otimes\sigma^{(a)}\right)\qquad\qquad
  \hat{\alpha}_{a,b,j}\sim\mathcal{N}\left(0,\frac{1}{18}\right) \iid
\end{equation}
By collecting like terms, this expression can be simplified to
\begin{equation}
  \hat{H}_{2}=\sum_{a,b=1}^{3}\left(\hat{\alpha}_{a,b,1}+\hat{\alpha}_{b,a,2}\right)\sigma^{(a)}\otimes\sigma^{(b)}
  =\sum_{j=1}^{9}\hat{\alpha}_{j}P_{j}
\end{equation}
where $\hat{\alpha}_{3(b-1)+a}=\hat{\alpha}_{a,b,1}+\hat{\alpha}_{b,a,2}$ and $P_{3(b-1)+a}=\sigma^{(a)}\otimes\sigma^{(b)}$.  The random variables $\hat{\alpha}_{3(b-1)+a}$ are then independent and normally distributed with mean zero and variance of $\frac{1}{9}=\frac{1}{18}+\frac{1}{18}$ by the summation properties of independent Gaussian random variables shown in Appendix \ref{Int}.

This formulation then exactly matches that of the previous conjecture with $\hat{H}_{2}^{(9)}\equiv \hat{H}_{2}$ and with the $P_{j}$ as defined above.  The joint spectral density is therefore claimed to be equal to
  \begin{equation}\label{jdosclaim}
    C\e^{-\frac{\r\boldsymbol{\lambda}^{2}}{2^{n+1}}}\Delta(\boldsymbol{\lambda})\int_{\mathbb{R}^{4^{n}-\r}}\frac{\det_{1\leq j,k\leq 2^{n}}\left(\e^{\im\lambda_{j}\mu_{k}(\boldsymbol{\beta})}\right)}{\Delta({\boldsymbol{\mu}(\boldsymbol{\beta})})}\boldsymbol{\di}\boldsymbol{\beta}
  \end{equation}
where
\begin{equation}
  C=\frac{2^{\frac{n4^{n}}{2}}\r^{\frac{\r}{2}}(2\pi)^{\frac{\r}{2}}\im^{\frac{2^{n}}{2}}}  {2^{n}!2^{n\r}(2\pi)^{\frac{2^{n}}{2}}(2\pi)^{\frac{4^{n}}{2}}\im^{\frac{4^{n}}{2}}}\qquad
  n=2\qquad
  \r=9
\end{equation}
and where $\boldsymbol{\mu}(\boldsymbol{\beta})=(\mu_{1}(\boldsymbol{\beta}),\dots,\mu_{2^{n}}(\boldsymbol{\beta}))$ are the eigenvalues of the matrix
\begin{equation}
  B=\left(\beta_{10}\sigma^{(1)}+\beta_{11}\sigma^{(2)}+\beta_{12}\sigma^{(3)}\right)\otimes I_{2}
    +I_{2}\otimes\left(\beta_{13}\sigma^{(1)}+\beta_{14}\sigma^{(2)}+\beta_{15}\sigma^{(3)}\right)
    +\beta_{16}I_{4}
\end{equation}
for $\boldsymbol{\beta}=(\beta_{10},\dots,\beta_{16})$.

\subsubsection{The eigenvalues of B}
To calculate the eigenvalues of $B$, the eigenvalues of the matrix
\begin{equation}
  \beta_{10}\sigma^{(1)}+\beta_{11}\sigma^{(2)}+\beta_{12}\sigma^{(3)}=
  \begin{pmatrix}
    \beta_{12}&\beta_{10}-\im\beta_{11}\\\beta_{10}+\im\beta_{11}&-\beta_{12}\\
  \end{pmatrix}
\end{equation}
will be calculated first.  The characteristic polynomial equation in $\lambda$ for this matrix reads
\begin{align}
  \det\begin{pmatrix}
    \beta_{12}-\lambda&\beta_{10}-\im\beta_{11}\\\beta_{10}+\im\beta_{11}&-\beta_{12}-\lambda\\
  \end{pmatrix}
  &=(\beta_{12}-\lambda)(-\beta_{12}-\lambda)-(\beta_{10}+\im\beta_{11})(\beta_{10}-\im\beta_{11})\nonumber\\
  &=\lambda^{2}-\beta_{10}^{2}-\beta_{11}^{2}-\beta_{12}^{2}
\end{align}
The roots of this equation give the eigenvalues of the preceding matrix as $\pm\sqrt{\beta_{10}^{2}+\beta_{11}^{2}+\beta_{12}^{2}}$.  As $\beta_{10}\sigma^{(1)}+\beta_{11}\sigma^{(2)}+\beta_{12}\sigma^{(3)}$ is Hermitian the spectral decomposition theorem states that there exists a $2\times2$ unitary matrix $U_{1}$ such that
\begin{equation}
  U_{1}^\dagger\left(\beta_{10}\sigma^{(1)}+\beta_{11}\sigma^{(2)}+\beta_{12}\sigma^{(3)}\right)U_{1}=\Lambda_{1}
\end{equation}
where $\Lambda_{1}=\text{diag}\left(\sqrt{\beta_{10}^{2}+\beta_{11}^{2}+\beta_{12}^{2}},-\sqrt{\beta_{10}^{2}+\beta_{11}^{2}+\beta_{12}^{2}}\right)$.  By symmetry there also exists a $2\times2$ unitary matrix $U_{2}$ such that
\begin{equation}
  U_{2}^\dagger\left(\beta_{13}\sigma^{(1)}+\beta_{14}\sigma^{(2)}+\beta_{15}\sigma^{(3)}\right)U_{2}=\Lambda_{2}
\end{equation}
where $\Lambda_{2}=\text{diag}\left(\sqrt{\beta_{13}^{2}+\beta_{14}^{2}+\beta_{15}^{2}},-\sqrt{\beta_{13}^{2}+\beta_{14}^{2}+\beta_{15}^{2}}\right)$.

By the properties of the tensor product stated in Section \ref{Quantum} it then follows that the unitary matrix $U_{1}\otimes U_{2}$ diagonalise $B$ as
\begin{align}
  \left(U_{1}\otimes U_{2}\right)^\dagger B \left(U_{1}\otimes U_{2}\right)
  &=\left(U_{1}^\dagger\otimes U_{2}^\dagger\right) B \left(U_{1}\otimes U_{2}\right)\nonumber\\
  &=\left(U_{1}^\dagger\left(\beta_{10}\sigma^{(1)}+\beta_{11}\sigma^{(2)}+\beta_{12}\sigma^{(3)}\right)U_{1}\right)\otimes \left(U_{2}^\dagger U_{2}\right)\nonumber\\
    &\qquad+\left(U_{1}^\dagger U_{1}\right)\otimes \left(U_{2}^\dagger\left(\beta_{13}\sigma^{(1)}+\beta_{14}\sigma^{(2)}+\beta_{15}\sigma^{(3)}\right)U_{2}\right)\nonumber\\
    &\qquad\qquad+\beta_{16}\left(U_{1}^\dagger U_{1}\right)\otimes\left(U_{2}^\dagger U_{2}\right)\nonumber\\
    &=\Lambda_{1}\otimes I_{2}+I_{2}\otimes\Lambda_{2}+\beta_{16}I_{4}
\end{align}
where by definition
\begin{align}
  \Lambda_{1}\otimes I_{2}+I_{2}\otimes\Lambda_{2}+\beta_{16}I_{4}
  &=\text{diag}\Big(\sqrt{\beta_{10}^{2}+\beta_{11}^{2}+\beta_{12}^{2}}+\sqrt{\beta_{13}^{2}+\beta_{14}^{2}+\beta_{15}^{2}}+\beta_{16},\nonumber\\
  &\quad\qquad\qquad\sqrt{\beta_{10}^{2}+\beta_{11}^{2}+\beta_{12}^{2}}-\sqrt{\beta_{13}^{2}+\beta_{14}^{2}+\beta_{15}^{2}}+\beta_{16},\nonumber\\
  &\qquad\qquad\qquad-\sqrt{\beta_{10}^{2}+\beta_{11}^{2}+\beta_{12}^{2}}+\sqrt{\beta_{13}^{2}+\beta_{14}^{2}+\beta_{15}^{2}}+\beta_{16},\nonumber\\
  &\qquad\qquad\qquad\qquad-\sqrt{\beta_{10}^{2}+\beta_{11}^{2}+\beta_{12}^{2}}-\sqrt{\beta_{13}^{2}+\beta_{14}^{2}+\beta_{15}^{2}}+\beta_{16}\Big)
\end{align}
By the spectral decomposition theorem these diagonal values are exactly the eigenvalues of $B$.

\subsubsection{Change of variables to spherical coordinates}
To simplify the calculation, some of the integration variables in $\boldsymbol{\beta}$ can be transformed to spherical coordinates.  Explicitly, the standard spherical transformation will be used,
\begin{align}
  (\beta_{10},\beta_{11},\beta_{12})\in\mathbb{R}^{3}&\to(r_{1},\theta_{1},\phi_{1})\in[0,\infty)\times[0,\pi)\times[0,2\pi)\nonumber\\
  \beta_{10}&=r_{1}\sin(\theta_{1})\cos(\phi_{1})\nonumber\\
  \beta_{11}&=r_{1}\sin(\theta_{1})\sin(\phi_{1})\nonumber\\
  \beta_{12}&=r_{1}\cos(\theta_{1})
\end{align}
with the volume element $\di\beta_{10}\di\beta_{11}\di\beta_{12}=r_{1}^{2}\sin(\theta_{1})\di r_{1}\di\theta_{1}\di\phi_{1}$ and identity $r_{1}^{2}=\beta_{10}^{2}+\beta_{11}^{2}+\beta_{12}^{2}$.  The analogous standard spherical transformation
\begin{align}
  (\beta_{13},\beta_{14},\beta_{15})\in\mathbb{R}^{3}&\to(r_{2},\theta_{2},\phi_{2})\in[0,\infty)\times[0,\pi)\times[0,2\pi)\nonumber\\
  \beta_{13}&=r_{2}\sin(\theta_{2})\cos(\phi_{2})\nonumber\\
  \beta_{14}&=r_{2}\sin(\theta_{2})\sin(\phi_{2})\nonumber\\
  \beta_{15}&=r_{2}\cos(\theta_{2})
\end{align}
with the volume element $\di\beta_{13}\di\beta_{14}\di\beta_{15}=r_{2}^{2}\sin(\theta_{2})\di r_{2}\di\theta_{2}\di\phi_{2}$ and identity $r_{2}^{2}=\beta_{13}^{2}+\beta_{14}^{2}+\beta_{15}^{2}$ will also be used.

The eigenvalues of $B$ are then given by
\begin{align}
  \mu_{1}&=r_{1}+r_{2}+\beta_{16}\quad\qquad\qquad
  \mu_{2}=r_{1}-r_{2}+\beta_{16}\nonumber\\
  \mu_{3}&=-r_{1}+r_{2}+\beta_{16}\,\,\qquad\qquad
  \mu_{4}=-r_{1}-r_{2}+\beta_{16}\nonumber\\
\end{align}

\subsubsection{Rotating $r_{1}$ and $r_{2}$}
To simplify the calculation further again, $\boldsymbol{r}=(r_{1},r_{2})$ is transformed to $\boldsymbol{u}=(u_{1},u_{2})$ via the linear transform
\begin{equation}
  \begin{pmatrix}u_{1}\\u_{2}\end{pmatrix}=\begin{pmatrix}1&-1\\1&1\end{pmatrix}\begin{pmatrix}r_{1}\\r_{2}\end{pmatrix}
  =\begin{pmatrix}r_{1}-r_{2}\\r_{1}+r_{2}\end{pmatrix}
\end{equation}
with the volume element identity
\begin{align}
  r_{1}^{2}r_{2}^{2}\di r_{1}\di r_{2}
  &=\left(\frac{u_{1}+u_{2}}{2}\right)^{2}\left(\frac{u_{2}-u_{1}}{2}\right)^{2}{\det}^{-1}\begin{pmatrix}1&-1\\1&1\end{pmatrix}\di u_{1}\di u_{2}\nonumber\\
  &=\frac{(u_{1}+u_{2})^{2}(u_{2}-u_{1})^{2}}{2^{5}}\di u_{1}\di u_{2}
\end{align}
The eigenvalues of $B$ are now given by
\begin{align}\label{muvals}
  \mu_{1}&=u_{2}+\beta_{16}\qquad\qquad\quad
  \mu_{2}=u_{1}+\beta_{16}\nonumber\\
  \mu_{3}&=-u_{1}+\beta_{16}\,\,\qquad\qquad
  \mu_{4}=-u_{2}+\beta_{16}\nonumber\\
\end{align}

The domain of integration must also be considered for $\boldsymbol{u}$.  The domain of integration  $D=\{(r_{1},r_{2})\,\,|\,\,r_{1},r_{2}\geq0\}$ for $\boldsymbol{r}$ is equivalently written as $D=\{(u_{1},u_{2})\,\,|\,\,u_{2}\geq0, |u_{1}|\leq u_{2}\}$ in the variables $\boldsymbol{u}$.  This is seen most simply by studying Figure \ref{transform}.

\begin{figure}
  \centering
  \input{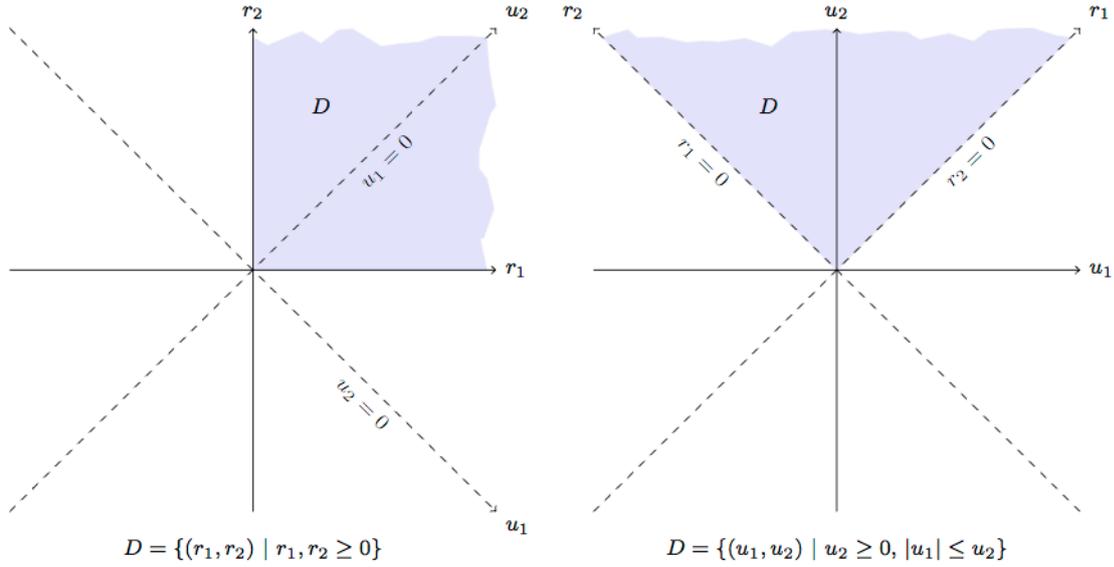}
  \caption[Domain transformation]{The domain $D=\{(r_{1},r_{2})\,\,|\,\,r_{1},r_{2}\geq0\}$ is shaded in the left graph.  The lines $u_{1}=0$ ($u_{2}$ axis) and $u_{2}=0$ ($u_{1}$ axis) are plotted as the dashed lines and orientated so that $u_{1}$ increases when either $r_{1}$ increases or $r_{2}$ decreases, and $u_{2}$ increases when either $r_{1}$ or $r_{2}$ increase.  This is consistent with the transform $\boldsymbol{r}\to\boldsymbol{u}$ in the text.  The graph on the right is a rotation of the first.  The two definitions of the domain $D$ are seen to be equivalent.}
  \label{transform}
\end{figure}

\subsubsection{The Vandermonde determinant}
The Vandermonde determinant $\Delta(\boldsymbol{\mu})$ will now be calculated.  By the definition of $\Delta(\boldsymbol{\mu})$ and the values of $\mu_{j}$ in (\ref{muvals})
\begin{align}
  \Delta(\boldsymbol{\mu})&=\prod_{1\leq j<k\leq4}(\mu_{k}-\mu_{j})\nonumber\\
  &=(-1)^{6}\prod_{1\leq j<k\leq4}(\mu_{j}-\mu_{k})\nonumber\\
  &=(\mu_{1}-\mu_{2})(\mu_{1}-\mu_{3})(\mu_{1}-\mu_{4})(\mu_{2}-\mu_{3})(\mu_{2}-\mu_{4})(\mu_{3}-\mu_{4})\nonumber\\
  &=(u_{2}-u_{1})(u_{2}+u_{1})(u_{2}+u_{2})(u_{1}+u_{1})(u_{1}+u_{2})(-u_{1}+u_{2})\nonumber\\
  &=2^{2}u_{1}u_{2}(u_{1}+u_{2})^{2}(u_{2}-u_{1})^{2}
\end{align}

\subsubsection{The determinant of phases}
The determinant $\det_{1\leq j<k\leq4}\left(\e^{\im\lambda_{j}\mu_{k}}\right)$ will now be calculated.  By the Leibniz formula for the determinant and the values of $\mu_{k}$ in (\ref{muvals})
\begin{align}
  \det_{1\leq j<k\leq4}\left(\e^{\im\lambda_{j}\mu_{k}}\right)
  &=\sum_{\tau\in \mathcal{S}_{4}}\sgn(\tau)\e^{\im\lambda_{\tau(1)}\mu_{1}}\e^{\im\lambda_{\tau(2)}\mu_{2}}\e^{\im\lambda_{\tau(3)}\mu_{3}}\e^{\im\lambda_{\tau(4)}\mu_{4}}\nonumber\\
  &=\sum_{\tau\in \mathcal{S}_{4}}\sgn(\tau)\e^{\im\lambda_{\tau(1)}(u_{2}+\beta_{16})}\e^{\im\lambda_{\tau(2)}(u_{1}+\beta_{16})}\e^{\im\lambda_{\tau(3)}(-u_{1}+\beta_{16})}\e^{\im\lambda_{\tau(4)}(-u_{2}+\beta_{16})}\nonumber\\
  &=\sum_{\tau\in \mathcal{S}_{4}}\sgn(\tau)\e^{\im\beta_{16}(\lambda_{1}+\lambda_{2}+\lambda_{3}+\lambda_{4})}\e^{\im u_{1}(\lambda_{\tau(2)}-\lambda_{\tau(3)})}\e^{\im u_{2}(\lambda_{\tau(1)}-\lambda_{\tau(4)})}
\end{align}
where $\mathcal{S}_{4}$ is the permutation group on $4$ elements and $\sgn(\tau)=\pm1$ is the sign of the permutation $\tau$.

The permutation group $\mathcal{S}_{4}$ can be split into two sets depending on the sign of each permutation.  Composition by the permutation that swaps the values $1$ and $4$ (or any two distinct values) forms a standard bijection between these two sets.  Let the set containing permutations with positive sign be denoted $\mathcal{S}_{4}^{+}$ and with negative sign $\mathcal{S}_{4}^{-}$.  A sum over all $\mathcal{S}_{4}$ can then be represented as a sum over just the element in $\mathcal{S}_{4}^{+}$ along with the corresponding element in $\mathcal{S}_{4}^{-}$ given by the bijection above.  Therefore, up to the factor $\e^{\im\beta_{16}(\lambda_{1}+\lambda_{2}+\lambda_{3}+\lambda_{4})}$, $\det_{1\leq j<k\leq4}\left(\e^{\im\lambda_{j}\mu_{k}}\right)$ is given by
\begin{align}
  &\sum_{\tau\in \mathcal{S}_{4}^{+}}\e^{\im u_{1}(\lambda_{\tau(2)}-\lambda_{\tau(3)})}\e^{\im u_{2}(\lambda_{\tau(1)}-\lambda_{\tau(4)})}-\e^{\im u_{1}(\lambda_{\tau(2)}-\lambda_{\tau(3)})}\e^{\im u_{2}(\lambda_{\tau(4)}-\lambda_{\tau(1)})}\nonumber\\
  &\qquad=\sum_{\tau\in \mathcal{S}_{4}^{+}}\e^{\im u_{1}(\lambda_{\tau(2)}-\lambda_{\tau(3)})}2\im\sin\left(u_{2}\left(\lambda_{\tau(1)}-\lambda_{\tau(4)}\right)\right)
\end{align}
The summation can be returned to a summation over $\mathcal{S}_{4}$ by the using the same bijection
\begin{align}
  &\sum_{\tau\in \mathcal{S}_{4}^{+}}\e^{\im u_{1}(\lambda_{\tau(2)}-\lambda_{\tau(3)})}2\im\sin\left(u_{2}\left(\lambda_{\tau(1)}-\lambda_{\tau(4)}\right)\right)\nonumber\\
  &\qquad=\im\sum_{\tau\in \mathcal{S}_{4}^{+}}\e^{\im u_{1}(\lambda_{\tau(2)}-\lambda_{\tau(3)})}\left(\sin\left(u_{2}\left(\lambda_{\tau(1)}-\lambda_{\tau(4)}\right)\right)-\sin\left(u_{2}\left(\lambda_{\tau(4)}-\lambda_{\tau(1)}\right)\right)\right)\nonumber\\
  &\qquad=\im\sum_{\tau\in \mathcal{S}_{4}}\sgn(\tau)\e^{\im u_{1}(\lambda_{\tau(2)}-\lambda_{\tau(3)})}\sin\left(u_{2}\left(\lambda_{\tau(1)}-\lambda_{\tau(4)}\right)\right)
\end{align}
In exactly the same fashion the exponential $\e^{\im u_{1}(\lambda_{\tau(2)}-\lambda_{\tau(3)})}$ may be replaced by $\im\sin(u_{1}(\lambda_{\tau(2)}-\lambda_{\tau(3)})$.  This gives the value of $\det_{1\leq j<k\leq4}\left(\e^{\im\lambda_{j}\mu_{k}}\right)$ to be
\begin{equation}
  \im^{2}\e^{\im\beta_{16}(\lambda_{1}+\lambda_{2}+\lambda_{3}+\lambda_{4})}\sum_{\tau\in \mathcal{S}_{4}}\sgn(\tau)\sin\left(u_{1}\left(\lambda_{\tau(2)}-\lambda_{\tau(3)}\right)\right)\sin\left(u_{2}\left(\lambda_{\tau(1)}-\lambda_{\tau(4)}\right)\right)
\end{equation}

\subsubsection{Combining results}
Substituting the results above into expression (\ref{jdosclaim}) for the joint spectral density gives
\begin{align}\label{4.4.57}
  &-C\e^{-\frac{\r\boldsymbol{\lambda}^{2}}{2^{n+1}}}\Delta(\boldsymbol{\lambda})\int\frac{\e^{\im\beta_{16}(\lambda_{1}+\lambda_{2}+\lambda_{3}+\lambda_{4})}\sum_{\tau\in \mathcal{S}_{4}}\sgn(\tau)\sin\left(u_{1}\left(\lambda_{\tau(2)}-\lambda_{\tau(3)}\right)\right)\sin\left(u_{2}\left(\lambda_{\tau(1)}-\lambda_{\tau(4)}\right)\right)}{2^{2}u_{1}u_{2}(u_{1}+u_{2})^{2}(u_{2}-u_{1})^{2}}\nonumber\\
  &\quad\qquad\qquad\qquad\qquad\cdot\frac{(u_{1}+u_{2})^{2}(u_{2}-u_{1})^{2}}{2^{5}}\sin(\theta_{1})\sin(\theta_{2})\di u_{1}\di u_{2}\di\theta_{1}\di\theta_{2}\di\phi_{1}\di\phi_{2}\di\beta_{16}
\end{align}
where $u_{1}\in[-u_{2},u_{2}]$, $u_{2}\in[0,\infty)$, $\theta_{1},\theta_{2}\in[0,\pi)$, $\phi_{1},\phi_{2}\in[0,2\pi)$, $\beta_{16}\in(-\infty,\infty)$,
\begin{equation}
  C=\frac{2^{\frac{n4^{n}}{2}}\r^{\frac{\r}{2}}(2\pi)^{\frac{\r}{2}}\im^{\frac{2^{n}}{2}}}  {2^{n}!2^{n\r}(2\pi)^{\frac{2^{n}}{2}}(2\pi)^{\frac{4^{n}}{2}}\im^{\frac{4^{n}}{2}}}\qquad
  n=2\qquad
  \r=9
\end{equation}
Note that the factors $(u_{1}+u_{2})^{2}(u_{2}-u_{1})^{2}$ in the numerator and denominator in the previous integral cancel.

\subsubsection{Integrating over the angular parameters}
The integration over the parameters $\theta_{1}$ and $\theta_{2}$ in (\ref{4.4.57}) both given the same results
\begin{equation}
  \int_{0}^{\pi}\sin(\theta_{1})\di\theta_{1}=\int_{0}^{\pi}\sin(\theta_{2})\di\theta_{2}=2
\end{equation}
Similarly for the parameters $\phi_{1}$ and $\phi_{2}$
\begin{equation}
  \int_{0}^{2\pi}\di\phi_{1}=\int_{0}^{2\pi}\di\phi_{2}=2\pi
\end{equation}
This results in the integration over the angular parameters contributing a factor of $2^{4}\pi^{2}$.

\subsubsection{Integrating over $\beta_{16}$}
The integration over $\beta_{16}$ in (\ref{4.4.57}) formally yields the Dirac delta distribution
\begin{equation}
  \int_{-\infty}^\infty\e^{\im\beta_{16}(\lambda_{1}+\lambda_{2}+\lambda_{3}+\lambda_{4})}\di\beta_{16}=2\pi\delta(\lambda_{1}+\lambda_{2}+\lambda_{3}+\lambda_{4})
\end{equation}

\subsubsection{Integrating over $\boldsymbol{u}$}
The integration over $\boldsymbol{u}=(u_{1},u_{2})$ in (\ref{4.4.57}) reads
\begin{equation}
  \int_{0}^\infty\int_{-u_{2}}^{u_{2}}\frac{\sum_{\tau\in \mathcal{S}_{4}}\sgn(\tau)\sin\left(u_{1}\left(\lambda_{\tau(2)}-\lambda_{\tau(3)}\right)\right)\sin\left(u_{2}\left(\lambda_{\tau(1)}-\lambda_{\tau(4)}\right)\right)}{u_{1}u_{2}}\di u_{1}\di u_{2}
\end{equation}
The integrand is invariant under both the transformations $u_{1}\to-u_{1}$ and $u_{2}\to-u_{2}$ separately as the integrand is an even function in each of these variables.  It is also invariant under the interchange of $u_{1}$ and $u_{2}$ as this amounts to permuting the indices $j=1,2,3,4$ corresponding to $\lambda_{\tau(j)}$ so that $1$ and $2$ are swapped and $3$ and $4$ are swapped in each term above.  This is an even ($\sgn=+1$) permutation in the group $\mathcal{S}_{4}$ so that the terms in the sum over $\mathcal{S}_{4}$ are transformed onto themselves (the composition of $\tau$ by the even permutation $(1\leftrightarrow2, 3\leftrightarrow4)$ results in a further even permutation).

Therefore the integrand is symmetric in the lines $u_{1}=\pm u_{2}$, $u_{1}=0$ and $u_{2}=0$, see Figure \ref{usym}.  The integration over $\boldsymbol{u}$ can then equivalently be taken over the domain $u_{1},u_{2}\in[0,\infty)$.  This allows the integrals with respect to $u_{1}$ and $u_{2}$ to be performed separately.  They are in fact the (improper) Dirichlet integrals which are formally computed as
\begin{align}
  \int_{0}^\infty\frac{\sin\left(u_{1}\left(\lambda_{\tau(2)}-\lambda_{\tau(3)}\right)\right)}{u_{1}}\di u_{1}
  &=\frac{\pi}{2}\sgn\left(\lambda_{\tau(2)}-\lambda_{\tau(3)}\right)\nonumber\\
  \int_{0}^\infty\frac{\sin\left(u_{2}\left(\lambda_{\tau(1)}-\lambda_{\tau(4)}\right)\right)}{u_{2}}\di u_{2}
  &=\frac{\pi}{2}\sgn\left(\lambda_{\tau(1)}-\lambda_{\tau(4)}\right)
\end{align}

\begin{figure}
  \centering
  \input{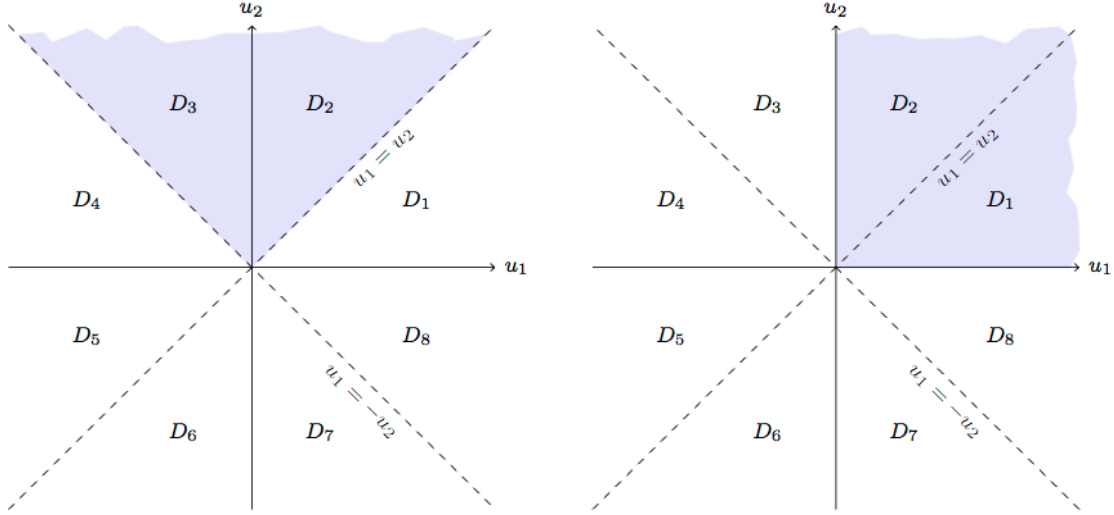}
  \caption[Integrand symmetry]{For some suitable integrand $f(u_{1},u_{2})$ with symmetries $f(u_{1},u_{2})=f(-u_{1},u_{2})=f(u_{1},-u_{2})=f(u_{2},u_{1})$ the integral over any of the domains $D_{1},\dots,D_{8}$ is identical.  This is seen by reflection in the axes and lines $u_{1}=\pm u_{2}$.  In particular the integral over $D_{2}\cup D_{3}$ is equal to that over $D_{1}\cup D_{2}$.}
  \label{usym}
\end{figure}

\subsubsection{Final result}
Collecting all of the result above yields that the joint spectral density $\hat{\rho}_{2,4}$, or $4$-point correlation function, for the ensemble $\hat{H}_{2}$ is given by
\begin{align}\label{jointDOS}
  &CC_{5}\e^{-\frac{\r\boldsymbol{\lambda}^{2}}{2^{n+1}}}\Delta(\boldsymbol{\lambda})\delta\left(\lambda_{1}+\lambda_{2}+\lambda_{3}+\lambda_{4}\right)\sum_{\tau\in \mathcal{S}_{4}}\sgn(\tau)\sgn\left(\lambda_{\tau(2)}-\lambda_{\tau(3)}\right)\sgn\left(\lambda_{\tau(1)}-\lambda_{\tau(4)}\right)
\end{align}
where
\begin{equation}
  C=\frac{2^{\frac{n4^{n}}{2}}\r^{\frac{\r}{2}}(2\pi)^{\frac{\r}{2}}\im^{\frac{2^{n}}{2}}}  {2^{n}!2^{n\r}(2\pi)^{\frac{2^{n}}{2}}(2\pi)^{\frac{4^{n}}{2}}\im^{\frac{4^{n}}{2}}}\qquad
  C_{5}=-\frac{1}{2^{2}}\frac{1}{2^{5}}2^{4}\pi^{2}2\pi\frac{\pi^{2}}{2^{2}}
  =-\frac{\pi^{5}}{2^{4}}
\end{equation}
$n=2$ and $\r=9$.

\subsection{Example: The \texorpdfstring{$1$}{1}-point correlation function for \texorpdfstring{$\hat{H}_{2}$}{H2}}
The $1$-point correlation function (spectral density) for $\hat{H}_{2}$ is, by definition,
\begin{equation}
  \hat{\rho}_{2,1}(\lambda)=\int_{\mathbb{R}^{3}}\hat{\rho}_{2,4}(\lambda,\lambda_{2},\lambda_{3},\lambda_{4})\di\lambda_{2}\di\lambda_{3}\di\lambda_{4}
\end{equation}
Using the conjectured expression (\ref{jointDOS}) for $\hat{\rho}_{2,4}$, one of the integrations here becomes trivial due to the presence of the delta factor $\delta\left(\lambda,\lambda_{2},\lambda_{3},\lambda_{4}\right)$.  For the remaining two dimensional integral, numerical methods are resorted to.

To simplify this numerical integration note that many of the $4!$ terms in the sum in the conjectured expression for $\hat{\rho}_{2,4}$ are identical.  These values are listed in Table \ref{permutations}.  With this simplification, the two dimensional integral was then computed using the software package Maple17.  The resulting function $\hat{\rho}_{2,1}(\lambda)$ is overlaid in Figure \ref{H2onepoint}.

\begin{figure}
  \centering
  \input{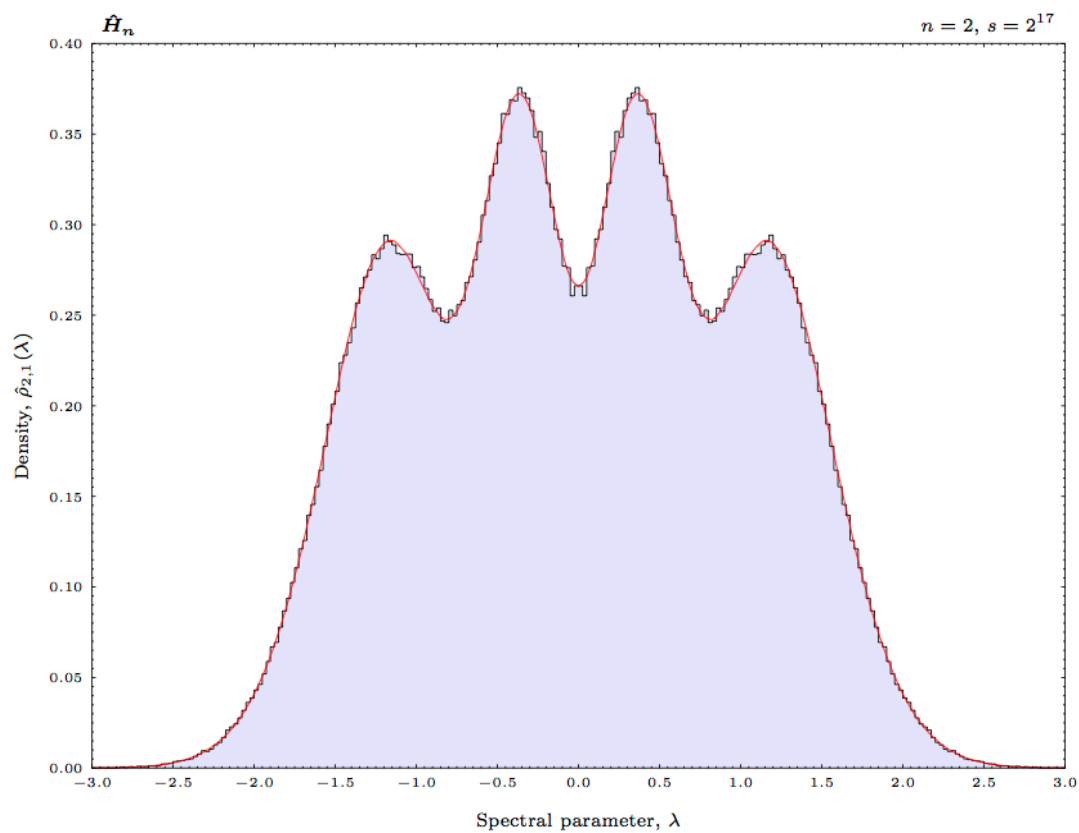}
  \caption[The $1$-point correlation function for $\hat{H}_{2}$]{The normalised spectral histogram for $\hat{H}_{2}$.  The spectra used were numerically obtained from each of $s=2^{17}$ samples of $\hat{H}_{2}$.  The $1$-point correlation function $\hat{\rho}_{2,1}(\lambda)$ found from the conjectured expression for the joint spectral density is overlaid (smooth curve).}
  \label{H2onepoint}
\end{figure}

This function may be compared against the results of a numerical simulation.  The normalised spectral histogram, obtained from $s=2^{17}$ samples from the ensemble $\hat{H}_{2}$, generated following the methodology outlined in Section \ref{Numerics}, is also shown in Figure \ref{H2onepoint}.   A good agreement is seen between the histogram and the density $\hat{\rho}_{2,1}$ calculated from the conjectured expression for $\hat{\rho}_{2,4}$.

\begin{table}\footnotesize
\centering
\begin{tabular}{cccccccc}
  \toprule
	\bf{Sign} &&& \bf{Summand} &&& \bf{Equivalent expressions} &\\
	\midrule
	$+$ &&& $\sgn\left(\lambda_{1}-\lambda_{2}\right)\sgn\left(\lambda_{3}-\lambda_{4}\right)$ &&& \multirow{8}{*}{$\sgn\left(\lambda_{1}-\lambda_{2}\right)\sgn\left(\lambda_{3}-\lambda_{4}\right)$} \\
	$-$ &&& $\sgn\left(\lambda_{2}-\lambda_{1}\right)\sgn\left(\lambda_{3}-\lambda_{4}\right)$\\
	$+$ &&& $\sgn\left(\lambda_{2}-\lambda_{1}\right)\sgn\left(\lambda_{4}-\lambda_{3}\right)$\\
	$-$ &&& $\sgn\left(\lambda_{1}-\lambda_{2}\right)\sgn\left(\lambda_{4}-\lambda_{3}\right)$\\
	$+$ &&& $\sgn\left(\lambda_{3}-\lambda_{4}\right)\sgn\left(\lambda_{1}-\lambda_{2}\right)$\\
	$-$ &&& $\sgn\left(\lambda_{4}-\lambda_{3}\right)\sgn\left(\lambda_{1}-\lambda_{2}\right)$\\
	$+$ &&& $\sgn\left(\lambda_{4}-\lambda_{3}\right)\sgn\left(\lambda_{2}-\lambda_{1}\right)$\\
	$-$ &&& $\sgn\left(\lambda_{3}-\lambda_{4}\right)\sgn\left(\lambda_{2}-\lambda_{1}\right)$\\
	\cmidrule(r){1-8}
	$+$ &&& $\sgn\left(\lambda_{1}-\lambda_{3}\right)\sgn\left(\lambda_{4}-\lambda_{2}\right)$ &&& \multirow{8}{*}{$\sgn\left(\lambda_{1}-\lambda_{3}\right)\sgn\left(\lambda_{4}-\lambda_{2}\right)$} \\
	$-$ &&& $\sgn\left(\lambda_{3}-\lambda_{1}\right)\sgn\left(\lambda_{4}-\lambda_{2}\right)$\\
	$+$ &&& $\sgn\left(\lambda_{3}-\lambda_{1}\right)\sgn\left(\lambda_{2}-\lambda_{4}\right)$\\
	$-$ &&& $\sgn\left(\lambda_{1}-\lambda_{3}\right)\sgn\left(\lambda_{2}-\lambda_{4}\right)$\\
	$+$ &&& $\sgn\left(\lambda_{4}-\lambda_{2}\right)\sgn\left(\lambda_{1}-\lambda_{3}\right)$\\
	$-$ &&& $\sgn\left(\lambda_{2}-\lambda_{4}\right)\sgn\left(\lambda_{1}-\lambda_{3}\right)$\\
	$+$ &&& $\sgn\left(\lambda_{2}-\lambda_{4}\right)\sgn\left(\lambda_{3}-\lambda_{1}\right)$\\
	$-$ &&& $\sgn\left(\lambda_{4}-\lambda_{2}\right)\sgn\left(\lambda_{3}-\lambda_{1}\right)$\\
	\cmidrule(r){1-8}
	$+$ &&& $\sgn\left(\lambda_{1}-\lambda_{4}\right)\sgn\left(\lambda_{2}-\lambda_{3}\right)$ &&& \multirow{8}{*}{$\sgn\left(\lambda_{1}-\lambda_{4}\right)\sgn\left(\lambda_{2}-\lambda_{3}\right)$} \\
	$-$ &&& $\sgn\left(\lambda_{4}-\lambda_{1}\right)\sgn\left(\lambda_{2}-\lambda_{3}\right)$\\
	$+$ &&& $\sgn\left(\lambda_{4}-\lambda_{1}\right)\sgn\left(\lambda_{3}-\lambda_{2}\right)$\\
	$-$ &&& $\sgn\left(\lambda_{1}-\lambda_{4}\right)\sgn\left(\lambda_{3}-\lambda_{2}\right)$\\
	$+$ &&& $\sgn\left(\lambda_{2}-\lambda_{3}\right)\sgn\left(\lambda_{1}-\lambda_{4}\right)$\\
	$-$ &&& $\sgn\left(\lambda_{3}-\lambda_{2}\right)\sgn\left(\lambda_{1}-\lambda_{4}\right)$\\
	$+$ &&& $\sgn\left(\lambda_{3}-\lambda_{2}\right)\sgn\left(\lambda_{4}-\lambda_{1}\right)$\\
	$-$ &&& $\sgn\left(\lambda_{2}-\lambda_{3}\right)\sgn\left(\lambda_{4}-\lambda_{1}\right)$\\
	\bottomrule
\end{tabular}
\caption[Summands in the joint density of states for $\hat{H}_{2}$]{The $4!$ summands in the sum over the permutation group $\mathcal{S}_{4}$ present in the conjectured joint probability measure for the eigenvalues of $\hat{H}_{2}$ (\ref{jointDOS}).  Each summand is listed with the sign of the associated permutation and grouped into three groups correspond to equivalent expressions.}
\label{permutations}
\end{table}

\subsection{Example: The \texorpdfstring{$2$}{2}-point correlation function for \texorpdfstring{$\hat{H}_{2}$}{H2}}
The $2$-point correlation function for $\hat{H}_{2}$ is, by definition,
\begin{equation}
  \hat{\rho}_{2,2}(\lambda_{1},\lambda_{2})=\int_{\mathbb{R}^{2}}\hat{\rho}_{2,4}(\lambda_{1},\lambda_{2},\lambda_{3},\lambda_{4})\di\lambda_{3}\di\lambda_{4}
\end{equation}
Again using the conjectured expression (\ref{jointDOS}) for $\hat{\rho}_{2,4}$, one of the integrations here becomes trivial due to the presence of the delta factor $\delta\left(\lambda_{1},\lambda_{2},\lambda_{3},\lambda_{4}\right)$.  For the one dimensional integral remaining, numerical methods are again resorted to in exactly the same fashion as for the $1$-point correlation function.  This numerically calculated function is overlaid in Figure \ref{H2twopoint}.

\begin{figure}
  \centering
  \input{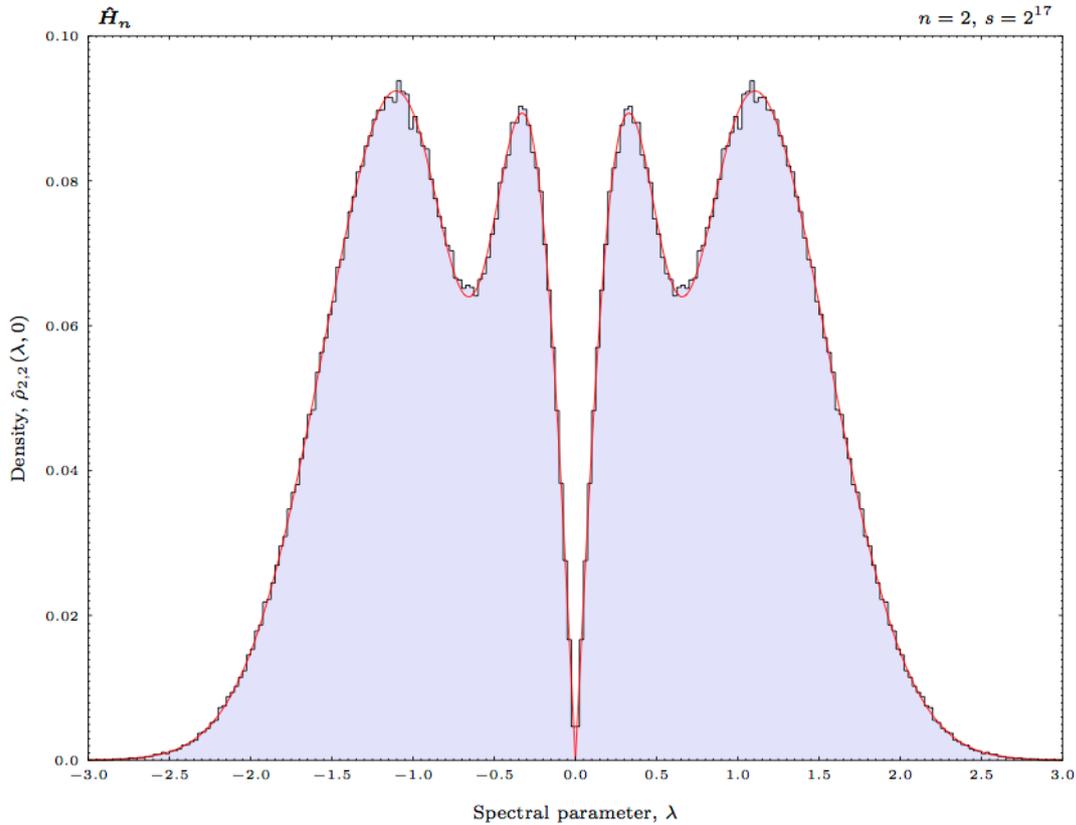}
  \caption[The $2$-point correlation function for $\hat{H}_{2}$]{The averaged histogram of the $2$-point correlation function $\hat{\rho}_{2,2}(\lambda,0)$ over $s=2^{17}$ samples of $\hat{H}_{2}$ with one eigenvalue $\lambda$ such that $|\lambda|<0.01$.  The two-point correlation function $\hat{\rho}_{2,2}(\lambda,0)$ found from the conjectured expression for the joint spectral density is overlaid (smooth curve).}
  \label{H2twopoint}
\end{figure}

This function may also be compared against a numerical simulation of the ensemble $\hat{H}_{2}$.  Random instances of $\hat{H}_{2}$ were taken until $s=2^{17}$ were found with at least one eigenvalue $\lambda$ such that $|\lambda|<0.01$.  For each such instance, the distances from the eigenvalue closest to zero to all the other eigenvalues were calculated.  A histogram of all these distances from each sample was then made.  As the probability measure of $\hat{H}_{2}$ is symmetric under $\hat{H}_{2}\to-\hat{H}_{2}$, $\hat{\rho}_{2,2}(\lambda,0)$ must be an even function of $\lambda$, therefore the values $\hat{\rho}_{2,2}(\lambda,0)$ were then replaced by $\frac{1}{2}\left(\hat{\rho}_{2,2}(\lambda,0)+\hat{\rho}_{2,2}(-\lambda,0)\right)$ in this histogram to provided further smoothing of statistical fluctuations.

As $\hat{\rho}_{2,2}(\lambda,0)$ must by definition have the normalisation
\begin{equation}
  \int_{\mathbb{R}}\hat{\rho}_{2,2}(\lambda,0)\di\lambda=\int_{\mathbb{R}^{3}}\hat{\rho}_{2,4}(\lambda,0,\lambda_{3},\lambda_{4})\di\lambda\di\lambda_{3}\di\lambda_{4}=\hat{\rho}_{2,1}(0)
\end{equation}
the resulting averaged histogram was normalised to the value of $\hat{\rho}_{2,1}(0)\approx0.266$ calculated from the results of the last section.  The resulting histogram is shown in Figure \ref{H2twopoint} and a strong agreement with $\hat{\rho}_{2,2}$, calculated from the conjectured expression for $\hat{\rho}_{2,4}$, is seen.

Furthermore the linearity of the expression for $\hat{\rho}_{2,2}(\lambda,0)$ for small values of $\lambda$ can be numerically verified.  Figure \ref{lin} plots the values
\begin{equation}
  \frac{\hat{\rho}_{2,2}(\lambda,0)}{\lambda}
\end{equation}
evaluated at $\lambda=\frac{1}{x}$, derived from the conjectured expression for the joint spectral density.  The values plotted conceivably tend to some finite value, providing numerical evidence that $\hat{\rho}_{2,2}(\lambda,0)=O(\lambda)$ as $\lambda\to0$.  This corresponds to the linear repulsion of the eigenvalues for $\hat{H}_{2}$ already observed numerically, see Section \ref{nnnumerics} and Appendix \ref{collNumericsSpacing}.

\begin{figure}
  \centering
  \input{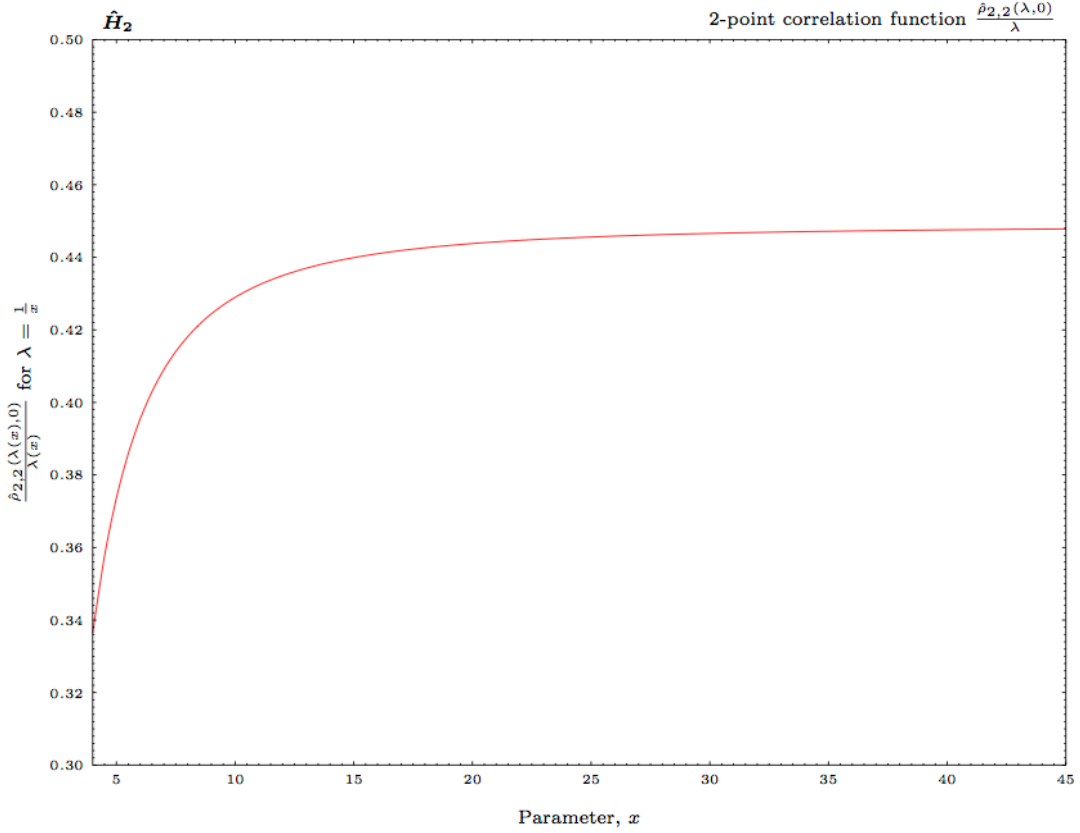}
  \caption[Linearity of $\hat{\rho}_{2,2}(\lambda,0)$ as $\lambda\to0$]{The plot of the function $\frac{\hat{\rho}_{2,2}(\lambda,0)}{\lambda}$ for $\lambda=\frac{1}{x}$ and $x\in[4,45]$, derived from the conjectured expression for the joint spectral density.}
  \label{lin}
\end{figure}
\section{Tightness of the reduced eigenstate purity bounds}\label{eigenstatebounds}
In this section it will be shown in some cases and conjectured in others that the bound given in Theorem \ref{invEnt} of Section \ref{EntLblock} (page \pageref{invEnt}) is asymptotically tight, that is linear in $\frac{1}{n}$, for different values of $l$.  To do this, the joint eigenstates of the translation matrix $T$ and a specific Hamiltonian, $H_{n}^{(Z)}=\sum_{j=1}^{n}\sigma_{j}^{(3)}$, will be constructed in Section \ref{JointEigen}.  In Section \ref{Entl=1subsection}, it will be seen that the bound given by Theorem \ref{invEnt} is asymptotically tight for $l=1$.  Sections \ref{l=2subsection} and \ref{l=3,4,5subsection} then give the results of numerical calculations which provide evidence that this is also the case for other vales of $l$.

Note that $H_{n}^{(Z)}$ has a highly non-degenerate spectrum, so that within each (large) eigensubspace it is not surprising that the bounds given in Theorem \ref{invEnt} can be satisfied.  However, without further assumption, this still shows the tightness of Theorem \ref{invEnt}, even if this tightness does not extend to the more interesting translationally-invariant nearest-neighbour Hamiltonians with non-degenerate spectra.

\subsection{Joint eigenstates of \texorpdfstring{$H_{n}^{(Z)}$}{Hn(Z)} and the translation matrix}\label{JointEigen}
The task of this section will be to find a basis of the translation matrix $T$ (see Section \ref{EigenDef}) and the Hamiltonian
\begin{equation}
  H_{n}^{(Z)}=\sum_{j=1}^{n}\sigma_{j}^{(3)}
\end{equation}

To this end, let the $n\times n$ unitary matrices $U$ and $V$ have the elements
\begin{equation}
  U_{j,k}=\frac{1}{\sqrt{n}}\omega_{j}^{k}\qquad\qquad V_{j,k}=\frac{1}{\sqrt{n}}\omega_{j-\frac{1}{2}}^{k}
\end{equation}
where $\omega_{j}=\e^{\frac{2\pi\im j}{n}}$.  The calculation in Appendix \ref{DisFour} confirms that $U$ and $V$ are unitary.  These unitary matrices represent the discrete Fourier transform with even and odd periodic boundary conditions respectively \cite[\p389]{Bernstein2009}.

The Fermi matrices $\a_{j}$ and $\a_{j}^\dagger$ for $j=1,\dots,n$ defined by the Jordan-Wigner transform (see Appendix \ref{JWT}) read
\begin{align}
  \a_{j}=\left(\prod_{1\leq l<j}\sigma_{l}^{(3)}\right)S_{j}\qquad\qquad\text{with}\qquad\qquad
  \a_{j}^\dagger=\left(\prod_{1\leq l<j}\sigma_{l}^{(3)}\right)S_{j}^\dagger
\end{align}
where
\begin{equation}
  S_{j}=\frac{\sigma_{  j  }^{(   1  )}+\im\sigma_{  j  }^{(   2   )}}{2}\qquad\qquad\text{with}\qquad\qquad
  S_{j}^\dagger=\frac{\sigma_{  j  }^{(   1   )}-\im\sigma_{  j  }^{(   2   )}}{2}
\end{equation}
As $U$ and $V$ are unitary, the Fermi matrices $\b_{j}$ and $\b_{j}^\dagger$ and $\c_{j}$ and $\c_{j}^\dagger$ can be defined to be
\begin{equation}\label{FermiBandC}
  \begin{pmatrix}\boldsymbol{\b}\\\boldsymbol{\b}^\prime\end{pmatrix}
  =\begin{pmatrix}U&0\\0&\overline{U}\end{pmatrix}
  \begin{pmatrix}\boldsymbol{\a}\\\boldsymbol{\a}^\prime\end{pmatrix}\qquad\qquad
  \begin{pmatrix}\boldsymbol{\c}\\\boldsymbol{\c}^\prime\end{pmatrix}
  =\begin{pmatrix}V&0\\0&\overline{V}\end{pmatrix}
  \begin{pmatrix}\boldsymbol{\a}\\\boldsymbol{\a}^\prime\end{pmatrix}\qquad\qquad
\end{equation}
by Appendix \ref{FermiTrans}.  Here, $\boldsymbol{\a}$ represents the column vector with entries $\a_{1},\dots,\a_{n}$ and $\boldsymbol{\a}^\prime$ represents the column vector with entries $\a_{1}^\dagger,\dots,\a_{n}^\dagger$.  The vectors $\boldsymbol{\b}$, $\boldsymbol{\b}^\prime$, $\boldsymbol{\c}$ and $\boldsymbol{\c}^\prime$ are defined similarly.

The following lemma now constructs a basis of joint eigenstates of the Hamiltonian $H_{n}^{(Z)}$ and the translation matrix $T$:
\begin{lemma}\label{EigenstatesHT}
  For the multi-indices $\boldsymbol{x}=(x_{1},\dots,x_{n})\in\{0,1\}^{n}$, the states
  \begin{equation}
    |\boldsymbol{x}\rangle_{t}=\begin{cases}
                               \Big(\b_{1}^\dagger\Big)^{x_{1}}\dots\Big(\b_{n}^\dagger\Big)^{x_{n}}|\boldsymbol{0}\rangle\qquad&\text{if }\s=\sum_{j=1}^{n}x_{j}\text{ is odd}\\
				\Big(\c_{1}^\dagger\Big)^{x_{1}}\dots\Big(\c_{n}^\dagger\Big)^{x_{n}}|\boldsymbol{0}\rangle\qquad&\text{if }\s=\sum_{j=1}^{n}x_{j}\text{ is even}
                             \end{cases}
  \end{equation}
  form an orthonormal basis of $\left(\mathbb{C}^{2}\right)^{\otimes n}$ and are joint eigenstates of the Hamiltonian $H_{n}^{(Z)}$ and the translation matrix $T$ of Section \ref{EigenDef}.  The state $|\boldsymbol{0}\rangle$ here is the member of the standard basis (see Section \ref{stndBasis}) indexed by the zero vector.
\end{lemma}

\begin{proof}
The proof will be split into three parts.  First the states $|\boldsymbol{x}\rangle_{t}$ will be shown to be eigenstates of the Hamiltonian $H_{n}^{(Z)}$, then shown to be orthonormal and finally shown to be eigenstates of the translation matrix $T$:

\subsubsection{The states $|\boldsymbol{x}\rangle_{t}$ as eigenstates of $H_{n}^{(Z)}$}
By the definition of the Fermi matrices $\b_{j}$ and $\b_{j}^\dagger$ it is seen that
\begin{align}\label{TransB}
  \Big(\b_{1}^\dagger\Big)^{x_{1}}\dots\Big(\b_{n}^\dagger\Big)^{x_{n}}|\boldsymbol{0}\rangle
  &=\frac{1}{n^{\frac{\s}{2}}n^{n-\s}}\sum_{k_{1},\dots,k_{n}=1}^{n}\Big(\omega_{1}^{k_{1}}\Big)^{x_{1}}\dots\Big(\omega_{n}^{k_{n}}\Big)^{x_{n}}\Big(\a_{k_{1}}^\dagger\Big)^{x_{1}}\dots\Big(\a_{k_{n}}^\dagger\Big)^{x_{n}}|\boldsymbol{0}\rangle
\end{align}
for the value of $\s=\sum_{j=1}^{n}x_j$ depending on the state $|\boldsymbol{x}\rangle_{t}$ under consideration.  Let $l_{1},\dots,l_\s$ denote the positions of the non-zero elements in the index vector $\boldsymbol{x}$.  By the canonical commutation relations for Fermi matrices, $\a_{j}^\dagger \a_{j}^\dagger=0$ so that the terms in which any of the indices $k_{l_{1}},\dots,k_{l_\s}$ are not distinct, are zero.  Therefore, upon relabelling $k_{l_{1}}\to k_{1},\dots,k_{l_\s}\to k_\s$ and summing over the remaining $n-\s$ indices, the last expression can be rewritten as
\begin{equation}\label{4.186}
  \Big(\b_{1}^\dagger\Big)^{x_{1}}\dots\Big(\b_{n}^\dagger\Big)^{x_{n}}|\boldsymbol{0}\rangle=\frac{1}{n^{\frac{\s}{2}}}\sum_{\genfrac{}{}{0pt}{}{k_{1},\dots,k_\s=1}{\text{distinct}}}^{n}\omega_{l_{1}}^{k_{1}}\dots\omega_{l_\s}^{k_\s}\a_{k_{1}}^\dagger\dots \a_{k_\s}^\dagger|\boldsymbol{0}\rangle
\end{equation}

For the standard basis elements $|0\rangle$ and $|1\rangle$ of $\mathbb{C}^{2}$ (eigenstates of $\sigma^{(3)}$, see Section \ref{stndBasis}) it follows from definition that
\begin{align}
  S^\dagger|0\rangle&\equiv\frac{\sigma^{(1)}-\im\sigma^{(2)}}{2}|0\rangle=\frac{|1\rangle-\im^{2}|1\rangle}{2}=|1\rangle
\end{align}
Then as $\sigma^{(3)}|y\rangle=(-1)^{y}|y\rangle$ for $y=0,1$, it follows from the definition of $\a_{j}^\dagger$ that for a multi-index $\boldsymbol{y}\in\{0,1\}^{n}$ with $y_{j}=0$
\begin{align}\label{transEl2}
  \a_{j}^\dagger|\boldsymbol{y}\rangle
  =\left(\prod_{1\leq l<j}\sigma_{l}^{(3)}\right)S_{j}^\dagger\big(|y_{1}\rangle\otimes\dots\otimes|y_{n}\rangle\big)
  =\prod_{1\leq l<j}(-1)^{y_{l}}|\boldsymbol{z}\rangle
\end{align}
where $\boldsymbol{z}$ is the multi-index formed from $\boldsymbol{y}$ by setting the $j^{th}$ entry to $1$.  

The states $\a_{k_{1}}^\dagger\dots \a_{k_\s}^\dagger|\boldsymbol{0}\rangle$ in (\ref{4.186}) are therefore each proportional to a standard basis vector $|\boldsymbol{y}\rangle$ such that $\boldsymbol{y}$ contains exactly $\s$ non-zero elements, that is $\sum_{j=1}^{n}y_{j}=\s$.  Now by the definition of such a standard basis element $|\boldsymbol{y}\rangle$ with $\sum_{j=1}^{n}y_{j}=\s$ and the identity $\sigma^{(3)}|y\rangle=(-1)^{y}|y\rangle=(1-2y)|y\rangle$ for $y=0,1$, the following eigenvalue equation holds
\begin{equation}
  H_{n}^{(Z)}|\boldsymbol{y}\rangle=\sum_{j=1}^{n}\sigma_{j}^{(3)}|\boldsymbol{y}\rangle=\sum_{j=1}^{n}(1-2y_{j})|\boldsymbol{y}\rangle=(n-2\s)|\boldsymbol{y}\rangle
\end{equation}
which implies that $|\boldsymbol{y}\rangle$ is an eigenstate of $H_{n}^{(Z)}$ with eigenvalue $n-2\s$.  Therefore each term on the right hand side of (\ref{4.186}) is an eigenstate of $H_{n}^{(Z)}$ with eigenvalue $n-2\s$, which implies that $\Big(\b_{1}^\dagger\Big)^{x_{1}}\dots\Big(\b_{n}^\dagger\Big)^{x_{n}}|\boldsymbol{0}\rangle$ also is.

An analogous argument holds for the states
\begin{equation}
  \Big(\c_{1}^\dagger\Big)^{x_{1}}\dots\Big(\c_{n}^\dagger\Big)^{x_{n}}|\boldsymbol{0}\rangle=\frac{1}{n^{\frac{\s}{2}}}\sum_{\genfrac{}{}{0pt}{}{k_{1},\dots,k_\s=1}{\text{distinct}}}^{n}\omega_{l_{1}-\frac{1}{2}}^{k_{1}}\dots\omega_{l_\s-\frac{1}{2}}^{k_\s}\a_{k_{1}}^\dagger\dots \a_{k_\s}^\dagger|\boldsymbol{0}\rangle
\end{equation}

\subsubsection{Orthonormality of the states $|\boldsymbol{x}\rangle_{t}$}
Similarly to before, the standard basis elements $|0\rangle$ and $|1\rangle$ of $\mathbb{C}^{2}$ satisfy
\begin{align}
  S|0\rangle&\equiv\frac{\sigma^{(1)}+\im\sigma^{(2)}}{2}|0\rangle=\frac{|1\rangle+\im^{2}|1\rangle}{2}=0
\end{align}
It then follows that for all $j=1,\dots,n$
\begin{align}\label{transEl}
  \a_{j}|\boldsymbol{0}\rangle
  &=\left(\prod_{1\leq l<j}\sigma_{l}^{(3)}\right)S_{j}\big(|0\rangle\otimes\dots\otimes|0\rangle\big)\nonumber\\
  &=\left(\bigotimes_{1\leq l<j}\sigma^{(3)}|0\rangle\right)\otimes S|0\rangle\otimes\left(\bigotimes_{j< l\leq n}|0\rangle\right)\nonumber\\
  &=0
\end{align}
so that $\b_{j}|\boldsymbol{0}\rangle=0$ for all $j=1,\dots,n$ by the expansion of $\b_{j}$ in terms of the $\a_{j}$ in (\ref{4.186}).  The states $\Big(\b_{1}^\dagger\Big)^{x_{1}}\dots\Big(\b_{n}^\dagger\Big)^{x_{n}}|\boldsymbol{0}\rangle$ for all $\boldsymbol{x}$ therefore form a complete orthonormal basis of the Hilbert space, see Appendix \ref{FermiBasis}.  Similarly the states $\Big(\c_{1}^\dagger\Big)^{x_{1}}\dots\Big(\c_{n}^\dagger\Big)^{x_{n}}|\boldsymbol{0}\rangle$ for all $\boldsymbol{x}$ also form a complete orthonormal basis by the same reasoning.

As states $|\boldsymbol{x}\rangle_{t}$ for even and odd values of $\s$ correspond to different eigenvalues of the Hamiltonian $H_{n}^{(Z)}$ (it has just been seen that they have eigenvalues of $n-2\s$) they must be orthogonal.  Hence, the states $|\boldsymbol{x}\rangle_{t}$ are all orthogonal.  Since there are $2^{n}$ orthogonal states $|\boldsymbol{x}\rangle_{t}$, they form an orthonormal basis of the Hilbert space.

\subsubsection{The states $|\boldsymbol{x}\rangle_{t}$ as eigenstates of $T$}
It will now be shown that the $|\boldsymbol{x}\rangle_{t}$ are eigenstates of the translation matrix $T$.  It has already been seen that
\begin{equation}
  \Big(\b_{1}^\dagger\Big)^{x_{1}}\dots\Big(\b_{n}^\dagger\Big)^{x_{n}}|\boldsymbol{0}\rangle=\frac{1}{n^{\frac{\s}{2}}}\sum_{\genfrac{}{}{0pt}{}{k_{1},\dots,k_\s=1}{\text{distinct}}}^{n}\omega_{l_{1}}^{k_{1}}\dots\omega_{l_\s}^{k_\s}\a_{k_{1}}^\dagger\dots \a_{k_\s}^\dagger|\boldsymbol{0}\rangle
\end{equation}
in (\ref{4.186}).  Applying $T$ to this state yields
\begin{equation}\label{Tstate}
  \frac{1}{n^{\frac{\s}{2}}}\sum_{\genfrac{}{}{0pt}{}{k_{1},\dots,k_\s=1}{\text{distinct}}}^{n}\omega_{l_{1}}^{k_{1}}\dots\omega_{l_\s}^{k_\s}\left(T\a_{k_{1}}^\dagger T^\dagger\right)\dots \left(T\a_{k_\s}^\dagger T^\dagger\right)|\boldsymbol{0}\rangle
\end{equation}
where the additional occurrences of $T$ and $T^\dagger$ in this expression are seen to cancel as the state $|\boldsymbol{0}\rangle$ is an eigenstates of $T$ with eigenvalue $1$ and as $T$ is unitary so that $T^\dagger T=I$.

For $k=1,\dots,n-1$, the conjugation of $\a_{k}^\dagger$ by $T$ gives, by the definition of the $\a_{k}^\dagger$ and action of $T$,
\begin{equation}
  T\a_{k}^\dagger T^\dagger
  =T\left(\prod_{1\leq l<k}\sigma_{l}^{(3)}\right)S_{k}^\dagger T^\dagger
  =\left(\prod_{1\leq l<k}\sigma_{l+1}^{(3)}\right)S_{k+1}^\dagger
  =\sigma_{1}^{(3)}\a_{k+1}^\dagger
\end{equation}
whereas for $k=n$
\begin{equation}
  T\a_{n}^\dagger T^\dagger
  =T\left(\prod_{1\leq l<n}\sigma_{l}^{(3)}\right)S_{n}^\dagger T^\dagger
  =\left(\prod_{1\leq l<n}\sigma_{l+1}^{(3)}\right)S_{1}^\dagger
  =\left(\prod_{l=2}^{n}\sigma_{l}^{(3)}\right)\a_{1}^\dagger
\end{equation}

Terms in (\ref{Tstate}) for which none of the indices $k_{1},\dots,k_\s$ are equal to $n$ are then equal to
\begin{equation}
  \frac{1}{n^{\frac{\s}{2}}}\omega_{l_{1}}^{k_{1}}\dots\omega_{l_\s}^{k_\s}\left(\sigma_{1}^{(3)}\a_{k_{1}+1}^\dagger\right) \dots \left(\sigma_{1}^{(3)}\a_{k_\s+1}^\dagger\right) |\boldsymbol{0}\rangle
\end{equation}
In this expression, all the $\sigma_{1}^{(3)}$ factors commute with the $\a_{k_{m}+1}^\dagger$ factors for $m=1,\dots,\s$ by their definitions (they both act as $\sigma^{(3)}$ on the first qubit). Therefore, the $\sigma_{1}^{(3)}$ factors can be applied to $|\boldsymbol{0}\rangle$ first, each giving $|\boldsymbol{0}\rangle$ as $\sigma^{(3)}|0\rangle=|0\rangle$, resulting in the term $\frac{1}{n^{\frac{\s}{2}}}\omega_{l_{1}}^{k_{1}}\dots\omega_{l_\s}^{k_\s}\a_{k_{1}+1}^\dagger \dots \a_{k_\s+1}^\dagger |\boldsymbol{0}\rangle$.

The remaining terms in (\ref{Tstate}) have exactly one index equal to $n$.  Let this index be $k_{m}$ for a particular value of $m=1,\dots,\s$.  These terms are then equal to
\begin{equation}
  \frac{1}{n^{\frac{\s}{2}}}\omega_{l_{1}}^{k_{1}}\dots\omega_{l_\s}^{k_\s}\left(\prod_{1\leq j<m}\sigma_{1}^{(3)}\a_{k_{j}+1}^\dagger\right)\left(\left(\prod_{l=2}^{n}\sigma_{l}^{(3)}\right)\a_{1}^\dagger\right)\left(\prod_{m< j\leq \s}\sigma_{1}^{(3)}\a_{k_{j}+1}^\dagger\right) |\boldsymbol{0}\rangle
\end{equation}
Here the factor $\left(\prod_{l=2}^{n}\sigma_{l}^{(3)}\right)$ anti-commutes with all the $\a_{k_{j}+1}^\dagger$ to its right (not the $\a_{1}^\dagger$) and each $\sigma_{1}^{(3)}$ to the left of $\a_{1}^\dagger$ anti-commutes with $\a_{1}^\dagger$.  So moving all the factors $\sigma_{1}^{(3)}$ and $\prod_{l=2}^{n}\sigma_{l}^{(3)}$ towards the $|\boldsymbol{0}\rangle$ results in an overall sign of $(-1)^{\s-1}$.  This results in the term $(-1)^{\s-1}\frac{1}{n^{\frac{\s}{2}}}\omega_{l_{1}}^{k_{1}}\dots\omega_{l_\s}^{k_\s}\a_{k_{1}+1}^\dagger \dots \a_{k_\s+1}^\dagger |\boldsymbol{0}\rangle$ where $\a_{n+1}^\dagger$ is identified with $\a_{1}^\dagger$.

For odd values of $\s$, both these cases (for values of $k_{1},\dots,k_{\s}$ with either one or no values of $n$ present) can be combined to give
\begin{align}
  T|\boldsymbol{x}\rangle_{t}
  &=\frac{1}{n^{\frac{\s}{2}}}\sum_{\genfrac{}{}{0pt}{}{k_{1},\dots,k_\s=1}{\text{distinct}}}^{n}\omega_{l_{1}}^{k_{1}}\dots\omega_{l_\s}^{k_\s}\a_{k_{1}+1}^\dagger \dots \a_{k_\s+1}^\dagger |\boldsymbol{0}\rangle
\end{align}
Relabelling $k_{1}\to k_{1}-1,\dots,k_\s\to k_\s-1$ transforms this to
\begin{align}
  \frac{1}{n^{\frac{\s}{2}}}\sum_{\genfrac{}{}{0pt}{}{k_{1},\dots,k_\s=2}{\text{distinct}}}^{n+1}\omega_{l_{1}}^{k_{1}-1}\dots\omega_{l_\s}^{k_\s-1}\a_{k_{1}}^\dagger \dots \a_{k_\s}^\dagger |\boldsymbol{0}\rangle
\end{align}
For the terms in which one index is equal to $n+1$, this index value can be replaced with the value $1$ as $\omega_{l}^{n}=\omega_{l}^{0}=1$, and $\a_{n+1}^\dagger\equiv \a_{1}^\dagger$, which gives
\begin{align}
  \frac{1}{n^{\frac{\s}{2}}}\sum_{\genfrac{}{}{0pt}{}{k_{1},\dots,k_\s=1}{\text{distinct}}}^{n}\omega_{l_{1}}^{k_{1}-1}\dots\omega_{l_\s}^{k_\s-1}\a_{k_{1}}^\dagger \dots \a_{k_\s}^\dagger |\boldsymbol{0}\rangle
\end{align}
This is precisely $\omega_{l_{1}}^{-1}\dots\omega_{l_\s}^{-1}|\boldsymbol{x}\rangle_{t}$ so that $|\boldsymbol{x}\rangle_{t}$ is indeed an eigenstate of the translation matrix $T$ for odd values of $\s$.

An analogous procedure may be applied in the case of even values of $\s$ for the states
\begin{equation}
	|\boldsymbol{x}\rangle_{t}
	=\Big(\c_{1}^\dagger\Big)^{x_{1}}\dots\Big(\c_{n}^\dagger\Big)^{x_{n}}|\boldsymbol{0}\rangle
	=\frac{1}{n^{\frac{\s}{2}}}\sum_{\genfrac{}{}{0pt}{}{k_{1},\dots,k_\s=1}{\text{distinct}}}^{n}\omega_{l_{1}-\frac{1}{2}}^{k_{1}}\dots\omega_{l_\s-\frac{1}{2}}^{k_\s}\a_{k_{1}}^\dagger\dots \a_{k_\s}^\dagger|\boldsymbol{0}\rangle
\end{equation}
Here the identification $\a_{n+1}^\dagger\equiv-\a_{1}^\dagger$ should be made so that the sign $(-1)^{\s-1}=-1$ is incorporated into the argument.  That is, following the procedure as in the case of odd values of $\s$,
\begin{align}
  T|\boldsymbol{x}\rangle_{t}
  &=\frac{1}{n^{\frac{\s}{2}}}\sum_{\genfrac{}{}{0pt}{}{k_{1},\dots,k_\s=1}{\text{distinct}}}^{n}\omega_{l_{1}-\frac{1}{2}}^{k_{1}}\dots\omega_{l_\s-\frac{1}{2}}^{k_\s}\a_{k_{1}+1}^\dagger \dots \a_{k_\s+1}^\dagger |\boldsymbol{0}\rangle
\end{align}
Relabelling $k_{1}\to k_{1}-1,\dots,k_\s\to k_\s-1$ transforms this to
\begin{align}
  \frac{1}{n^{\frac{\s}{2}}}\sum_{\genfrac{}{}{0pt}{}{k_{1},\dots,k_\s=2}{\text{distinct}}}^{n+1}\omega_{l_{1}-\frac{1}{2}}^{k_{1}-1}\dots\omega_{l_\s-\frac{1}{2}}^{k_\s-1}\a_{k_{1}}^\dagger \dots \a_{k_\s}^\dagger |\boldsymbol{0}\rangle
\end{align}
For the terms in which one index is equal to $n+1$, this index value can be replaced with the value $1$ as
\begin{equation}
  \omega_{l-\frac{1}{2}}^{n}
  =\e^{\frac{2\pi\im}{n}n(l-\frac{1}{2})}
  =\e^{2\pi\im l}\e^{-\pi\im}
  =-1
  =-\omega_{l-\frac{1}{2}}^{0}                                                                                                                          
\end{equation}
and $\a_{n+1}^\dagger\equiv -\a_{1}^\dagger$ which gives
\begin{align}
  \frac{1}{n^{\frac{\s}{2}}}\sum_{\genfrac{}{}{0pt}{}{k_{1},\dots,k_\s=1}{\text{distinct}}}^{n}\omega_{l_{1}-\frac{1}{2}}^{k_{1}-1}\dots\omega_{l_\s-\frac{1}{2}}^{k_\s-1}\a_{k_{1}}^\dagger \dots \a_{k_\s}^\dagger |\boldsymbol{0}\rangle
\end{align}
This is precisely $\omega_{l_{1}-\frac{1}{2}}^{-1}\dots\omega_{l_\s-\frac{1}{2}}^{-1}|\boldsymbol{x}\rangle_{t}$ so that $|\boldsymbol{x}\rangle_{t}$ is indeed an eigenstate of the translation matrix $T$ for even values of $\s$.

This concludes the proof.
\end{proof}

\subsection{Asymptotic tightness of the purity bound for \texorpdfstring{$l=1$}{l=1}}\label{Entl=1subsection}
The average (over all eigenstates) reduced eigenstate purity on a single qubit for the joint eigenstates of the Hamiltonian $H_{n}^{(Z)}$ and the translation matrix $T$ will be analytically calculated in this section for all values of $n\geq2$.  This will provide an explicit example for which the bound on this quantity, given in Theorem \ref{invEnt} of Section \ref{EntLblock} (page \pageref{invEnt}), is asymptotically tight, that is linear in $\frac{1}{n}$.

The following lemma gives this:
\begin{lemma}\label{l=1tightness}
  The joint eigenstates of the Hamiltonian $H_{n}^{(Z)}$ and the translation matrix $T$, given in Lemma \ref{EigenstatesHT} and denoted by $|\boldsymbol{x}\rangle_{t}$, satisfy the following equation
  \begin{equation}
    \frac{1}{2^{n}}\sum_{\boldsymbol{x}}\Tr\left(\rho_{1,\boldsymbol{x}}^{2}\right)=\frac{1}{2}+\frac{1}{2n}
  \end{equation}
  for $n=2,3,\dots$, where $\rho_{1,\boldsymbol{x}}$ are the reduced density matrices on the single qubit labelled $1$ of the eigenstates $|\boldsymbol{x}\rangle_{t}$ and the sum is over all of the $2^{n}$ multi-indices $\boldsymbol{x}\in\{0,1\}^{n}$.
\end{lemma}

Before the proof of this lemma is given, the values of the average purity it provides will be compared against that numerically obtained from the samples of the ensemble $\hat{H}_{n}^{(inv,local)}$ generated for the previous chapters, and the analytic bound given by Theorem \ref{invEnt}.

For this comparison, the values of
\begin{align}
  E_{n,1}=\frac{1}{2^{n}}\sum_{\boldsymbol{x}}\Tr\left(\rho_{1,\boldsymbol{x}}^{2}\right)-\frac{1}{2^{1}}
\end{align}
where $\rho_{1,\boldsymbol{x}}$ are the reduced density matrices, on the qubit labelled $1$, of the eigenstates $|\boldsymbol{x}\rangle_{t}$ or of the eigenstates of samples\footnote{All the matrices sampled from $\hat{H}_{n}^{(inv,local)}$ for $n=2,3$ and $n=5,\dots,13$ had non-degenerate spectra so that their unique eigenstates (up to phase) are also eigenstates of the translation matrix $T$.} of the ensemble $\hat{H}_{n}^{(inv,local)}$, will be considered.  The values of $E_{n,1}^{-1}$ are plotted against $n$ in Figure \ref{l=1Fig}.  Here it is seen that for all the values of $E_{n,1}^{-1}$ considered, the bound
\begin{equation}
  E_{n,1}\leq\frac{2^{1}}{n}\quad\iff\quad \frac{1}{E_{n,1}}\geq \frac{n}{2^{1}}
\end{equation}
from Theorem \ref{invEnt} is satisfied, although it seems not to be optimal for these cases.  Furthermore, for many members of the ensemble $\hat{H}_{n}^{(inv,local)}$ the values of $E_{n,1}$ could conceivably satisfy an exponential bound of the form
\begin{equation}
  E_{n,1}\leq\frac{1}{C2^n}\quad\iff\quad \frac{1}{E_{n,1}}\geq C2^n
\end{equation}
for some constant $C$.

\begin{figure}
  \centering
  \include{Figures/Ent1}
  \caption[Average purity of the reduced eigenstates on $l=1$ qubit]{(Solid line) The values of $E_{n,1}^{-1}=\left(\frac{1}{2^{n}}\sum_{\boldsymbol{x}}\Tr\left(\rho_{1,\boldsymbol{x}}^{2}\right)-\frac{1}{2^{1}}\right)^{-1}$ against $n=2,\dots,32$, where $\rho_{1,\boldsymbol{x}}$ are the reduced density matrices on the qubit labelled $1$ of the joint eigenstates $|\boldsymbol{x}\rangle_{t}$ of the translation matrix $T$ and spin chain Hamiltonian $H_{n}^{(Z)}$.  (Crosses) The same values for $2^6$ of the matrices sampled from the ensemble $\hat{H}_{n}^{(inv,local)}$ in Section \ref{eigenNum} for various values of $n$.  The $n=4$ values for $\hat{H}_{n}^{(inv,local)}$ are not included due to the presence of degenerate eigenvalues.  The bound from Theorem \ref{invEnt} of Section \ref{EntLblock} is given by the dashed line.}
  \label{l=1Fig}
\end{figure}

The proof of Lemma \ref{l=1tightness} will now be given:
\begin{proof}
As derived in the proof of Theorem \ref{invEnt}, the reduced density matrix of an eigenstate $|\boldsymbol{x}\rangle_{t}$ on the qubit labelled $1$ is given by
\begin{equation}\label{transonequbit}
  \rho_{1,\boldsymbol{x}}=\frac{1}{2^{1}}\sum_{a=0}^{3}\,_{t}\langle\boldsymbol{x}|\sigma_{1}^{(a)}|\boldsymbol{x}\rangle_{t}\sigma^{(a)}
\end{equation}
The qubit labelled $1$ here is not special.  By the translational symmetry of $|\boldsymbol{x}\rangle_{t}$ the corresponding expressions for any single qubit are equivalent.  That is, as $T^{-1}|\boldsymbol{x}\rangle_{t}=\e^{-\im\theta_{\boldsymbol{x}}}|\boldsymbol{x}\rangle_{t}$ for some $\theta_{\boldsymbol{x}}\in[0,\pi)$,
\begin{equation}
  \,_{t}\langle\boldsymbol{x}|\sigma_{j}^{(a)}|\boldsymbol{x}\rangle_{t}
  =\,_{t}\langle\boldsymbol{x}|T^{j}\sigma_{1}^{(1)}T^{-j}|\boldsymbol{x}\rangle_{t}
  =\,_{t}\langle\boldsymbol{x}|\sigma_{1}^{(1)}|\boldsymbol{x}\rangle_{t}
\end{equation}

\subsubsection{The value of $\,_{t}\langle\boldsymbol{x}|\sigma_{1}^{(a)}|\boldsymbol{x}\rangle_{t}$ for $a=0,1,2,3$}
The value of $\,_{t}\langle\boldsymbol{x}|\sigma_{1}^{(a)}|\boldsymbol{x}\rangle_{t}$ will now be explicitly calculated.  The value of $\,_{t}\langle\boldsymbol{x}|\sigma_{1}^{(0)}|\boldsymbol{x}\rangle_{t}=\,_{t}\langle\boldsymbol{x}|\boldsymbol{x}\rangle_{t}$ is trivially $1$ as $|\boldsymbol{x}\rangle_{t}$ is a normalised state.

The calculations in Appendix \ref{JWnn} allows $\sigma_{1}^{(a)}$, for $a=1,2,3$, to be expressed in terms of the Fermi matrices $\a_{j}$ and $\a_{j}^\dagger$ defined therein as
\begin{align}
  \sigma_{1}^{(1)}&=\a_{1}+\a_{1}^\dagger\nonumber\\
  \sigma_{1}^{(2)}&=-\im \a_{1}+\im \a_{1}^\dagger\nonumber\\
  \sigma_{1}^{(3)}&=\a_{1}\a_{1}^\dagger-\a_{1}^\dagger \a_{1}
\end{align}
These quantities can then be expressed in terms of the Fermi matrices $\b_{j}$ and $\b_{j}^\dagger$ defined in equation (\ref{FermiBandC}) as
\begin{align}
  \sigma_{1}^{(1)}&=\frac{1}{\sqrt{n}}\sum_{j=1}^{n}\left(\omega_{j}^{-1}\b_{j}+\omega_{j}^{1}\b_{j}^\dagger\right)\nonumber\\
  \sigma_{1}^{(2)}&=\frac{-\im}{\sqrt{n}}\sum_{j=1}^{n}\left(\omega_{j}^{-1}\b_{j}-\omega_{j}^{1}\b_{j}^\dagger\right)\nonumber\\
  \sigma_{1}^{(3)}&=\frac{1}{n}\sum_{j,k=1}^{n}\left(\omega_{j}^{-1}\omega_{k}^{1}\b_{j}\b_{k}^{\dagger}-\omega_{j}^{1}\omega_{k}^{-1}\b_{j}^\dagger \b_{k}\right)
\end{align}
The standard properties of Fermi matrices $\b_{j}$ and $\b_{j}^\dagger$ acting on the orthonormal Fermi basis $|\boldsymbol{x}\rangle_{{\b}}=\Big(\b_{1}^\dagger\Big)\dots\Big(\b_{n}^\dagger\Big)|\boldsymbol{0}\rangle$ yield that
\begin{align}\label{Fermiprop}
  \,_{{\b}}\langle\boldsymbol{x}|\b_{j}|\boldsymbol{x}\rangle_{{\b}}=\,_{{\b}}\langle\boldsymbol{x}|\b_{j}^\dagger|\boldsymbol{x}\rangle_{{\b}}&=0\nonumber\\
  \,_{{\b}}\langle\boldsymbol{x}|\b_{j}\b_{k}^\dagger|\boldsymbol{x}\rangle_{{\b}}&=\delta_{j,k}(1-x_{j})\nonumber\\
  \,_{{\b}}\langle\boldsymbol{x}|\b_{j}^\dagger \b_{k}|\boldsymbol{x}\rangle_{{\b}}&=\delta_{j,k}x_{j}
\end{align}
as shown in Appendix \ref{JWFermiAction}.

It now follows from these identities that $\,_{{\b}}\langle\boldsymbol{x}|\sigma_{1}^{(1)}|\boldsymbol{x}\rangle_{{\b}}=\,_{{\b}}\langle\boldsymbol{x}|\sigma_{1}^{(2)}|\boldsymbol{x}\rangle_{{\b}}=0$ and that
\begin{align}
  \,_{{\b}}\langle\boldsymbol{x}|\sigma_{1}^{(3)}|\boldsymbol{x}\rangle_{{\b}}
  =\frac{1}{n}\sum_{j=1}^{n}\left(\omega_{j}^{-1}\omega_{j}^{1}(1-x_{j})-\omega_{j}^{1}\omega_{j}^{-1}x_{j}\right)
  =\frac{1}{n}\sum_{j=1}^{n}(1-2x_{j})
  =1-\frac{2\s}{n}
\end{align}
where $\s=\sum_{j=1}^{n}x_{j}$.

An analogous procedure may be used for the Fermi matrices $\c_{j}$ and $\c_{j}^\dagger$ which yields that
\begin{align}
  \sigma_{1}^{(1)}&=\frac{1}{\sqrt{n}}\sum_{j=1}^{n}\left(\omega_{j-\frac{1}{2}}^{-1}\c_{j}+\omega_{j-\frac{1}{2}}^{1}\c_{j}^\dagger\right)\nonumber\\
  \sigma_{1}^{(2)}&=\frac{-\im}{\sqrt{n}}\sum_{j=1}^{n}\left(\omega_{j-\frac{1}{2}}^{-1}\c_{j}-\omega_{j-\frac{1}{2}}^{1}\c_{j}^\dagger\right)\nonumber\\
  \sigma_{1}^{(3)}&=\frac{1}{n}\sum_{j,k=1}^{n}\left(\omega_{j-\frac{1}{2}}^{-1}\omega_{k-\frac{1}{2}}^{1}\c_{j}\c_{k}^{\dagger}-\omega_{j-\frac{1}{2}}^{1}\omega_{k-\frac{1}{2}}^{-1}\c_{j}^\dagger \c_{k}\right)
\end{align}
so that again $\,_{{\c}}\langle\boldsymbol{x}|\sigma_{1}^{(1)}|\boldsymbol{x}\rangle_{{\c}}=\,_{{\c}}\langle\boldsymbol{x}|\sigma_{1}^{(2)}|\boldsymbol{x}\rangle_{{\c}}=0$ and
\begin{align}
  \,_{{\c}}\langle\boldsymbol{x}|\sigma_{1}^{(3)}|\boldsymbol{x}\rangle_{{\c}}
  =\frac{1}{n}\sum_{j=1}^{n}\left(\omega_{j-\frac{1}{2}}^{-1}\omega_{j-\frac{1}{2}}^{1}(1-x_{j})-\omega_{j-\frac{1}{2}}^{1}\omega_{j-\frac{1}{2}}^{-1}x_{j}\right)
  =\frac{1}{n}\sum_{j=1}^{n}(1-2x_{j})
  =1-\frac{2\s}{n}
\end{align}

As $|\boldsymbol{x}\rangle_{t}=|\boldsymbol{x}\rangle_{{\b}}$ if the value of $\s=\sum_{j=1}^{n}x_{j}$ is odd and $|\boldsymbol{x}\rangle_{t}=|\boldsymbol{x}\rangle_{{\c}}$ if the value of $\s$ is even, it follows that $\,_{t}\langle\boldsymbol{x}|\sigma_{1}^{(1)}|\boldsymbol{x}\rangle_{t}=\,_{t}\langle\boldsymbol{x}|\sigma_{1}^{(2)}|\boldsymbol{x}\rangle_{t}=0$ and
\begin{equation}
  \,_{t}\langle\boldsymbol{x}|\sigma_{1}^{(3)}|\boldsymbol{x}\rangle_{t}
  =1-\frac{2\s}{n}
\end{equation}

\subsubsection{Collecting results}
Substituting these result into equation (\ref{transonequbit}) yields that
\begin{equation}
  \rho_{1,\boldsymbol{x}}
  =\frac{1}{2}\sum_{a=0}^{3}\,_{t}\langle\boldsymbol{x}|\sigma_{1}^{(a)}|\boldsymbol{x}\rangle_{t}\sigma^{(a)}
  =\frac{1}{2}\left(\sigma^{(0)}+\left(1-\frac{2\s}{n}\right)\sigma^{(3)}\right)
\end{equation}
As $\Tr\left(\sigma^{(a)}\sigma^{(b)}\right)=2\delta_{a,b}$ for $a,b=0,1,2,3$ (see section \ref{PauliMatrixBasis}) it follows that
\begin{equation}
  \Tr\left(\rho_{1,\boldsymbol{x}}^{2}\right)
  =\frac{2}{2^{2}}\left(1+\left(1-\frac{2\s}{n}\right)^{2}\right)
  =1-\frac{2\s}{n}+\frac{2\s^{2}}{n^{2}}
\end{equation}
and then that
\begin{equation}
  \frac{1}{2^{n}}\sum_{\boldsymbol{x}}\Tr\left(\rho_{1,\boldsymbol{x}}^{2}\right)
  =\frac{1}{2^{n}}\sum_{\boldsymbol{x}}\left(1-\frac{2\s}{n}+\frac{2\s^2}{n^{2}}\right)
\end{equation}
The only values of $\s$ are $0,\dots,n$ for different values of the multi-index $\boldsymbol{x}$.  There are exactly $\genfrac{(}{)}{0pt}{}{n}{\s}$ multi-indices $\boldsymbol{x}$ with each particular value of $\s$.  Therefore the last sum can be reexpressed as a sum over these values of $\s$,
\begin{equation}
  \frac{1}{2^{n}}\sum_{\boldsymbol{x}}\Tr\left(\rho_{1,\boldsymbol{x}}^{2}\right)
  =\frac{1}{2^{n}}\sum_{\s=0}^{n}\genfrac{(}{)}{0pt}{}{n}{\s}\left(1-\frac{2\s}{n}+\frac{2\s^{2}}{n^{2}}\right)
\end{equation}
The standard results, stated in \cite[\p 14]{Boros2004},
\begin{equation}
  \sum_{\s=0}^{n}\genfrac{(}{)}{0pt}{}{n}{\s}=2^{n}\qquad
  \sum_{\s=0}^{n}\s\genfrac{(}{)}{0pt}{}{n}{\s}=n2^{n-1}\qquad
  \sum_{\s=0}^{n}\s^{2}\genfrac{(}{)}{0pt}{}{n}{\s}=n(n+1)2^{n-2}
\end{equation}
then show that
\begin{align}
  \frac{1}{2^{n}}\sum_{\boldsymbol{x}}\Tr\left(\rho_{1,\boldsymbol{x}}^{2}\right)
  &=\frac{1}{2^{n}}\left(2^{n}-\frac{2n2^{n-1}}{n}+\frac{2n(n+1)2^{n-2}}{n^{2}}\right)\nonumber\\
  &=\frac{1}{2^{n}}\left(2^{n-1}+\frac{2^{n-1}}{n}\right)\nonumber\\
  &=\frac{1}{2}+\frac{1}{2n}
\end{align}
which concludes the proof.
\end{proof}

\subsection{Numerical asymptotic tightness of the purity bound for \texorpdfstring{$l=2$}{l=2}}\label{l=2subsection}
Exactly the same procedure as in the last section can be applied for the case where $l=2$:  As derived in the proof of Theorem \ref{invEnt} of Section \ref{EntLblock} (page \pageref{invEnt}), the reduced density matrix on the qubits labelled $1$ and $2$ of a joint eigenstate of the translation matrix $T$ and the Hamiltonian $H_{n}^{(Z)}$ (see Section \ref{JointEigen}) $|\boldsymbol{x}\rangle_{t}$ is given by
\begin{equation}\label{transtwoqubits}
  \rho_{2,\boldsymbol{x}}=\frac{1}{2^{2}}\sum_{a,b=0}^{3}\,_{t}\langle\boldsymbol{x}|\sigma_{1}^{(a)}\sigma_{2}^{(b)}|\boldsymbol{x}\rangle_{t}\sigma^{(a)}\otimes\sigma^{(b)}
\end{equation}
The qubits labelled $1$ and $2$ here are again not special.  By the translational symmetry of $|\boldsymbol{x}\rangle_{t}$ ($T^{-1}|\boldsymbol{x}\rangle_{t}=\e^{-\im\theta_{\boldsymbol{x}}}|\boldsymbol{x}\rangle_{t}$ for some $\theta_{\boldsymbol{x}}\in[0,2\pi)$) the corresponding expressions for any neighbouring pair of qubits are equivalent, as seen similarly in the $l=1$ case previously.

Once the values of $\,_{t}\langle\boldsymbol{x}|\sigma_{1}^{(a)}\sigma_{2}^{(b)}|\boldsymbol{x}\rangle_{t}$ are determined it again follows that
\begin{equation}\label{tracerhoAl=2}
  \frac{1}{2^{n}}\sum_{\boldsymbol{x}}\Tr\left(\rho_{2,\boldsymbol{x}}^{2}\right)
  =\frac{1}{2^{n}}\sum_{\boldsymbol{x}}\frac{2^{2}}{2^{4}}\sum_{a,b=0}^{3}\,_{t}\langle\boldsymbol{x}|\sigma_{1}^{(a)}\sigma_{2}^{(b)}|\boldsymbol{x}\rangle_{t}^{2}
\end{equation}
as $\Tr\left(\left(\sigma^{(a)}\otimes\sigma^{(b)}\right)\left(\sigma^{(c)}\otimes\sigma^{(d)}\right)\right)=2^{2}\delta_{a,c}\delta_{b,d}$, see Section \ref{PauliMatrixBasis}.

The values of $\,_{t}\langle\boldsymbol{x}|\sigma_{1}^{(a)}\sigma_{2}^{(b)}|\boldsymbol{x}\rangle_{t}$ will now be calculated:

\subsubsection{The value of $\,_{t}\langle\boldsymbol{x}|\sigma_{1}^{(0)}\sigma_{2}^{(0)}|\boldsymbol{x}\rangle_{t}$}
As $\sigma_{1}^{(0)}\sigma_{2}^{(0)}=I$ and the $|\boldsymbol{x}\rangle_{t}$ are orthonormal states, it follows that
\begin{equation}
  \,_{t}\langle\boldsymbol{x}|\sigma_{1}^{(0)}\sigma_{2}^{(0)}|\boldsymbol{x}\rangle_{t}=\,_{t}\langle\boldsymbol{x}|\boldsymbol{x}\rangle_{t}=1
\end{equation}

\subsubsection{The values of $\,_{t}\langle\boldsymbol{x}|\sigma_{1}^{(a)}\sigma_{2}^{(0)}|\boldsymbol{x}\rangle_{t}$ for $a=1,2,3$}
As $\sigma_{2}^{(0)}=I$, $\,_{t}\langle\boldsymbol{x}|\sigma_{1}^{(a)}\sigma_{2}^{(0)}|\boldsymbol{x}\rangle_{t}=\,_{t}\langle\boldsymbol{x}|\sigma_{1}^{(a)}|\boldsymbol{x}\rangle_{t}$.  By the results for $\,_{t}\langle\boldsymbol{x}|\sigma_{1}^{(a)}|\boldsymbol{x}\rangle_{t}$ from the last section it then follows that $\,_{t}\langle\boldsymbol{x}|\sigma_{1}^{(1)}\sigma_{2}^{(0)}|\boldsymbol{x}\rangle_{t}=\,_{t}\langle\boldsymbol{x}|\sigma_{1}^{(2)}\sigma_{2}^{(0)}|\boldsymbol{x}\rangle_{t}=0$ and
\begin{equation}
  \,_{t}\langle\boldsymbol{x}|\sigma_{1}^{(3)}\sigma_{2}^{(0)}|\boldsymbol{x}\rangle_{t}=1-\frac{2\s}{n}
\end{equation}

\subsubsection{The values of $\,_{t}\langle\boldsymbol{x}|\sigma_{1}^{(0)}\sigma_{2}^{(b)}|\boldsymbol{x}\rangle_{t}$ for $b=1,2,3$}
In the last section is was seen that $\,_{t}\langle\boldsymbol{x}|\sigma_{1}^{(a)}|\boldsymbol{x}\rangle_{t}=\,_{t}\langle\boldsymbol{x}|\sigma_{2}^{(a)}|\boldsymbol{x}\rangle_{t}$ by the translational symmetry of the state $|\boldsymbol{x}\rangle_{t}$.  Therefore it again follows that $\,_{t}\langle\boldsymbol{x}|\sigma_{1}^{(0)}\sigma_{2}^{(1)}|\boldsymbol{x}\rangle_{t}=\,_{t}\langle\boldsymbol{x}|\sigma_{1}^{(0)}\sigma_{2}^{(2)}|\boldsymbol{x}\rangle_{t}=0$ and
\begin{equation}
  \,_{t}\langle\boldsymbol{x}|\sigma_{1}^{(0)}\sigma_{2}^{(3)}|\boldsymbol{x}\rangle_{t}=1-\frac{2\s}{n}
\end{equation}

\subsubsection{The values of $\,_{t}\langle\boldsymbol{x}|\sigma_{1}^{(a)}\sigma_{2}^{(b)}|\boldsymbol{x}\rangle_{t}$ for $a,b=1,2$}
From the calculations in Appendix \ref{JWnn} and the definition of the Fermi matrices $\b_{j}$ and $\b_{j}^\dagger$ in (\ref{FermiBandC}) the relevant Pauli matrices have the following expansions
\begin{align}
  \sigma_{1}^{(1)}\sigma_{2}^{(1)}
  &=-\Big(\a_{1}-\a_{1}^\dagger\Big)\Big(\a_{2}+\a_{2}^\dagger\Big)
  =-\frac{1}{n}\sum_{j,k=1}^{n}\left(\omega_{j}^{-1}\b_{j}-\omega_{j}^{1}\b_{j}^\dagger\right)\left(\omega_{k}^{-2}\b_{k}+\omega_{k}^{2}\b_{k}^\dagger\right)\nonumber\\
  \sigma_{1}^{(1)}\sigma_{2}^{(2)}
  &=\im\Big(\a_{1}-\a_{1}^\dagger\Big)\Big(\a_{2}-\a_{2}^\dagger\Big)
  =\frac{\im}{n}\sum_{j,k=1}^{n}\left(\omega_{j}^{-1}\b_{j}-\omega_{j}^{1}\b_{j}^\dagger\right)\left(\omega_{k}^{-2}\b_{k}-\omega_{k}^{2}\b_{k}^\dagger\right)\nonumber\\
  \sigma_{1}^{(2)}\sigma_{2}^{(1)}
  &=\im\Big(\a_{1}+\a_{1}^\dagger\Big)\Big(\a_{2}+\a_{2}^\dagger\Big)
  =\frac{\im}{n}\sum_{j,k=1}^{n}\left(\omega_{j}^{-1}\b_{j}+\omega_{j}^{1}\b_{j}^\dagger\right)\left(\omega_{k}^{-2}\b_{k}+\omega_{k}^{2}\b_{k}^\dagger\right)\nonumber\\
  \sigma_{1}^{(2)}\sigma_{2}^{(2)}
  &=\Big(\a_{1}+\a_{1}^\dagger\Big)\Big(\a_{2}-\a_{2}^\dagger\Big)
  =\frac{1}{n}\sum_{j,k=1}^{n}\left(\omega_{j}^{-1}\b_{j}+\omega_{j}^{1}\b_{j}^\dagger\right)\left(\omega_{k}^{-2}\b_{k}-\omega_{k}^{2}\b_{k}^\dagger\right)\nonumber\\
\end{align}
Therefore by the standard properties of the Fermi matrices $\b_{j}$ and $\b_{j}^\dagger$, as used in (\ref{Fermiprop}) and listed in Appendix \ref{JWFermiAction}, it holds that
\begin{align}
  \,_{{\b}}\langle\boldsymbol{x}|\sigma_{1}^{(1)}\sigma_{2}^{(1)}|\boldsymbol{x}\rangle_{{\b}}
  &=-\frac{1}{n}\sum_{j=1}^{n}\left(\omega_{j}^{-1}\omega_{j}^{2}(1-x_{j})-\omega_{j}^{1}\omega_{j}^{-2}x_{j}\right)\nonumber\\
  \,_{{\b}}\langle\boldsymbol{x}|\sigma_{1}^{(1)}\sigma_{2}^{(2)}|\boldsymbol{x}\rangle_{{\b}}
  &=\frac{\im}{n}\sum_{j=1}^{n}\left(-\omega_{j}^{-1}\omega_{j}^{2}(1-x_{j})-\omega_{j}^{1}\omega_{j}^{-2}x_{j}\right)\nonumber\\
  \,_{{\b}}\langle\boldsymbol{x}|\sigma_{1}^{(2)}\sigma_{2}^{(1)}|\boldsymbol{x}\rangle_{{\b}}
  &=\frac{\im}{n}\sum_{j=1}^{n}\left(\omega_{j}^{-1}\omega_{j}^{2}(1-x_{j})+\omega_{j}^{1}\omega_{j}^{-2}x_{j}\right)\nonumber\\
  \,_{{\b}}\langle\boldsymbol{x}|\sigma_{1}^{(2)}\sigma_{2}^{(2)}|\boldsymbol{x}\rangle_{{\b}}
  &=\frac{1}{n}\sum_{j=1}^{n}\left(-\omega_{j}^{-1}\omega_{j}^{2}(1-x_{j})+\omega_{j}^{1}\omega_{j}^{-2}x_{j}\right)\nonumber\\
\end{align}
for the orthonormal Fermi basis $|\boldsymbol{x}\rangle_{{\b}}=\Big(\b_{1}^\dagger\Big)^{x_{1}}\dots\Big(\b_{n}^\dagger\Big)^{x_{n}}|\boldsymbol{0}\rangle$ (it has previously been seen that $|\boldsymbol{0}\rangle=|\boldsymbol{0}\rangle_\b$ in (\ref{transEl})).

Similarly for the Fermi matrices $\c_{j}$ and $\c_{j}^\dagger$ defined in equation (\ref{FermiBandC}),
\begin{align}
  \,_{{\c}}\langle\boldsymbol{x}|\sigma_{1}^{(1)}\sigma_{2}^{(1)}|\boldsymbol{x}\rangle_{{\c}}
  &=-\frac{1}{n}\sum_{j=1}^{n}\left(\omega_{j-\frac{1}{2}}^{-1}\omega_{j-\frac{1}{2}}^{2}(1-x_{j})-\omega_{j-\frac{1}{2}}^{1}\omega_{j-\frac{1}{2}}^{-2}x_{j}\right)\nonumber\\
  \,_{{\c}}\langle\boldsymbol{x}|\sigma_{1}^{(1)}\sigma_{2}^{(2)}|\boldsymbol{x}\rangle_{{\c}}
  &=\frac{\im}{n}\sum_{j=1}^{n}\left(-\omega_{j-\frac{1}{2}}^{-1}\omega_{j-\frac{1}{2}}^{2}(1-x_{j})-\omega_{j-\frac{1}{2}}^{1}\omega_{j-\frac{1}{2}}^{-2}x_{j}\right)\nonumber\\
  \,_{{\c}}\langle\boldsymbol{x}|\sigma_{1}^{(2)}\sigma_{2}^{(1)}|\boldsymbol{x}\rangle_{{\c}}
  &=\frac{\im}{n}\sum_{j=1}^{n}\left(\omega_{j-\frac{1}{2}}^{-1}\omega_{j-\frac{1}{2}}^{2}(1-x_{j})+\omega_{j-\frac{1}{2}}^{1}\omega_{j-\frac{1}{2}}^{-2}x_{j}\right)\nonumber\\
  \,_{{\c}}\langle\boldsymbol{x}|\sigma_{1}^{(2)}\sigma_{2}^{(2)}|\boldsymbol{x}\rangle_{{\c}}
  &=\frac{1}{n}\sum_{j=1}^{n}\left(-\omega_{j-\frac{1}{2}}^{-1}\omega_{j-\frac{1}{2}}^{2}(1-x_{j})+\omega_{j-\frac{1}{2}}^{1}\omega_{j-\frac{1}{2}}^{-2}x_{j}\right)\nonumber\\
\end{align}
so that the values of $\,_{t}\langle\boldsymbol{x}|\sigma_{1}^{(a)}\sigma_{2}^{(b)}|\boldsymbol{x}\rangle_{t}$ for $a,b=1,2$ can be read off as appropriate, the values are listed in Table \ref{coeffList}.

\subsubsection{The values of $\,_{t}\langle\boldsymbol{x}|\sigma_{1}^{(3)}\sigma_{2}^{(b)}|\boldsymbol{x}\rangle_{t}$ for $b=1,2$}
From the calculations in Appendix \ref{JWnn} and the definition of the Fermi matrices $\b_{j}$ and $\b_{j}^\dagger$ in equation (\ref{FermiBandC}) the relevant Pauli matrices have the following expansions
\begin{align}
  \sigma_{1}^{(3)}\sigma_{2}^{(1)}
  &=\sigma_{1}^{(3)}\sigma_{1}^{(3)}\Big(\a_{2}+\a_{2}^\dagger\Big)
  =\a_{2}+\a_{2}^\dagger
  =\frac{1}{\sqrt{n}}\sum_{j=1}^{n}\left(\omega_{j}^{-2}\b_{j}+\omega_{j}^{2}\b_{j}^\dagger\right)\nonumber\\
  \sigma_{1}^{(3)}\sigma_{2}^{(2)}
  &=-\im\sigma_{1}^{(3)}\sigma_{1}^{(3)}\Big(\a_{2}-\a_{2}^\dagger\Big)
  =-\im \a_{2}+\im \a_{2}^\dagger
  =\frac{-\im}{\sqrt{n}}\sum_{j=1}^{n}\left(\omega_{j}^{-2}\b_{j}-\omega_{j}^{2}\b_{j}^\dagger\right)
\end{align}
Therefore by the standard properties of the Fermi matrices $\b_{j}$ and $\b_{j}^\dagger$, listed in Appendix \ref{JWFermiAction}, it holds that $\,_{{\b}}\langle\boldsymbol{x}|\sigma_{1}^{(3)}\sigma_{2}^{(1)}|\boldsymbol{x}\rangle_{{\b}}=\,_{{\b}}\langle\boldsymbol{x}|\sigma_{1}^{(3)}\sigma_{2}^{(2)}|\boldsymbol{x}\rangle_{{\b}}=0$.

Similarly for the Fermi matrices $\c_{j}$ and $\c_{j}^\dagger$, $\,_{{\c}}\langle\boldsymbol{x}|\sigma_{1}^{(3)}\sigma_{2}^{(1)}|\boldsymbol{x}\rangle_{{\c}}=\,_{{\c}}\langle\boldsymbol{x}|\sigma_{1}^{(3)}\sigma_{2}^{(2)}|\boldsymbol{x}\rangle_{{\c}}=0$.

\subsubsection{The values of $\,_{t}\langle\boldsymbol{x}|\sigma_{1}^{(a)}\sigma_{2}^{(3)}|\boldsymbol{x}\rangle_{t}$ for $a=1,2$}
From the calculations in Appendix \ref{JWnn} and the definition of the Fermi matrices $\b_{j}$ and $\b_{j}^\dagger$ in equation (\ref{FermiBandC}) the relevant Pauli matrices have the following expansions
\begin{align}
  \sigma_{1}^{(1)}\sigma_{2}^{(3)}
  &=\Big(\a_{1}+\a_{1}^\dagger\Big)\Big(\a_{2}\a_{2}^\dagger-\a_{2}^\dagger \a_{2}\Big)
  =\frac{1}{n^{\frac{3}{2}}}\sum_{j,k,l=1}^{n}\left(\omega_{j}^{-1}\b_{j}+\omega_{j}^{1}\b_{j}^\dagger\right)\Big(\omega_{k}^{-2}\omega_{l}^{2}\b_{k}\b_{l}^\dagger-\omega_{k}^{2}\omega_{l}^{-2}\b_{k}^\dagger \b_{l}\Big)\nonumber\\
  \sigma_{1}^{(2)}\sigma_{2}^{(3)}
  &=-\im\Big(\a_{1}-\a_{1}^\dagger\Big)\Big(\a_{2}\a_{2}^\dagger-\a_{2}^\dagger \a_{2}\Big)
  =\frac{-\im}{n^{\frac{3}{2}}}\sum_{j,k,l=1}^{n}\left(\omega_{j}^{-1}\b_{j}-\omega_{j}^{1}\b_{j}^\dagger\right)\Big(\omega_{k}^{-2}\omega_{l}^{2}\b_{k}\b_{l}^\dagger-\omega_{k}^{2}\omega_{l}^{-2}\b_{k}^\dagger \b_{l}\Big)
\end{align}
Therefore by the standard properties of the Fermi matrices $\b_{j}$ and $\b_{j}^\dagger$, listed in Appendix \ref{JWFermiAction}, it holds that $\,_{{\b}}\langle\boldsymbol{x}|\sigma_{1}^{(1)}\sigma_{2}^{(3)}|\boldsymbol{x}\rangle_{{\b}}=\,_{{\b}}\langle\boldsymbol{x}|\sigma_{1}^{(2)}\sigma_{2}^{(3)}|\boldsymbol{x}\rangle_{{\b}}=0$.

Similarly for the Fermi matrices $\c_{j}$ and $\c_{j}^\dagger$, $\,_{{\c}}\langle\boldsymbol{x}|\sigma_{1}^{(1)}\sigma_{2}^{(3)}|\boldsymbol{x}\rangle_{{\c}}=\,_{{\c}}\langle\boldsymbol{x}|\sigma_{1}^{(2)}\sigma_{2}^{(3)}|\boldsymbol{x}\rangle_{{\c}}=0$

\subsubsection{The value of $\,_{t}\langle\boldsymbol{x}|\sigma_{1}^{(3)}\sigma_{2}^{(3)}|\boldsymbol{x}\rangle_{t}$}
From the calculations in Appendix \ref{JWnn} and the definition of the Fermi matrices $\b_{j}$ and $\b_{j}^\dagger$ in equation (\ref{FermiBandC}) the matrix $\sigma_{1}^{(3)}\sigma_{2}^{(3)}$ has the following expansion
\begin{align}\label{ZZterm}
  \sigma_{1}^{(3)}\sigma_{2}^{(3)}
  &=\Big(2\a_{1}\a_{1}^\dagger-I\Big)\Big(2\a_{2}\a_{2}^\dagger-I\Big)\nonumber\\
  &=4\a_{1}\a_{1}^\dagger \a_{2}\a_{2}^\dagger-2\a_{1}\a_{1}^\dagger-2\a_{2}\a_{2}^\dagger+I\nonumber\\
  &=4\a_{2}\a_{1}\a_{1}^\dagger \a_{2}^\dagger-2\a_{1}\a_{1}^\dagger-2\a_{2}\a_{2}^\dagger+I\nonumber\\
  &=\frac{4}{n^{2}}\sum_{j,k,l,m=1}^{n}\omega_{j}^{-2}\omega_{k}^{-1}\omega_{l}^{1}\omega_{m}^{2}\b_{j}\b_{k}\b_{l}^\dagger \b_{m}^\dagger\nonumber\\
  &\qquad\qquad-\frac{2}{n}\sum_{j,k=1}^{n}\omega_{j}^{-1}\omega_{k}^{1}\b_{j}\b_{k}^\dagger
  -\frac{2}{n}\sum_{j,k=1}^{n}\omega_{j}^{-2}\omega_{k}^{2}\b_{j}\b_{k}^\dagger
  +I
\end{align}
Terms in the first sum on the right hand side of this last expression are zero if $j=k$ or $l=m$ as $\b_{j}\b_{j}=\b_{l}^{\dagger}\b_{l}^{\dagger}=0$ for $j,l=1,\dots,n$.  Therefore the first sum on the right hand side of this last expression can be rewritten by considering only the indices for which $j<k$ and $l<m$ and their relevant reorderings, that is
\begin{align}
  \sum_{j,k,l,m=1}^{n}\omega_{j}^{-2}\omega_{k}^{-1}&\omega_{l}^{1}\omega_{m}^{2}\b_{j}\b_{k}\b_{l}^\dagger \b_{m}^\dagger\nonumber\\
  &=\sum_{j<k}\sum_{l<m}\Bigg(
   \omega_{j}^{-2}\omega_{k}^{-1}\omega_{l}^{1}\omega_{m}^{2}\b_{j}\b_{k}\b_{l}^\dagger \b_{m}^\dagger
  +\omega_{k}^{-2}\omega_{j}^{-1}\omega_{l}^{1}\omega_{m}^{2}\b_{k}\b_{j}\b_{l}^\dagger \b_{m}^\dagger\nonumber\\
  &\qquad\qquad\qquad\qquad
  +\omega_{j}^{-2}\omega_{k}^{-1}\omega_{m}^{1}\omega_{l}^{2}\b_{j}\b_{k}\b_{m}^\dagger \b_{l}^\dagger
  +\omega_{k}^{-2}\omega_{j}^{-1}\omega_{m}^{1}\omega_{l}^{2}\b_{k}\b_{j}\b_{m}^\dagger \b_{l}^\dagger\Bigg)
\end{align}
Which, upon reordering the Fermi matrices (respecting their mutual anti-commutation), is equal to
\begin{equation}
  -\sum_{j<k}\sum_{l<m}\left(\omega_{j}^{-2}\omega_{k}^{-1}\omega_{l}^{1}\omega_{m}^{2}-\omega_{k}^{-2}\omega_{j}^{-1}\omega_{l}^{1}\omega_{m}^{2}-\omega_{j}^{-2}\omega_{k}^{-1}\omega_{m}^{1}\omega_{l}^{2}+\omega_{k}^{-2}\omega_{j}^{-1}\omega_{m}^{1}\omega_{l}^{2}\right)\b_{j}\b_{k}\b_{m}^\dagger \b_{l}^\dagger
\end{equation}
By the standard properties of the Fermi matrices $\b_{j}$ and $\b_{j}^\dagger$, listed in Appendix \ref{JWFermiAction}, in order to contribute to the value of $\,_{{\b}}\langle\boldsymbol{x}|\sigma_{1}^{(3)}\sigma_{2}^{(3)}|\boldsymbol{x}\rangle_{{\b}}$, terms in this last expression must have $j=l$ and $k=m$ ($j=m$ and $k=l$ is not possible due to the restriction on the sums).  Therefore this last expression contributes
\begin{equation}
  -\sum_{j<k}\left(\omega_{j}^{-2}\omega_{k}^{-1}\omega_{j}^{1}\omega_{k}^{2}-\omega_{k}^{-2}\omega_{j}^{-1}\omega_{j}^{1}\omega_{k}^{2}-\omega_{j}^{-2}\omega_{k}^{-1}\omega_{k}^{1}\omega_{j}^{2}+\omega_{k}^{-2}\omega_{j}^{-1}\omega_{k}^{1}\omega_{j}^{2}\right)(1-x_{j})(1-x_{k})
\end{equation}
to the value of $\,_{{\b}}\langle\boldsymbol{x}|\sigma_{1}^{(3)}\sigma_{2}^{(3)}|\boldsymbol{x}\rangle_{{\b}}$.  The remaining terms in (\ref{ZZterm}) contribute
\begin{equation}
  -\frac{2}{n}\sum_{j=1}^{n}\omega_{j}^{-1}\omega_{j}^{1}(1-x_{j})-\frac{2}{n}\sum_{j=1}^{n}\omega_{j}^{-2}\omega_{j}^{2}(1-x_{j})+1
\end{equation}
to the value of $\,_{{\b}}\langle\boldsymbol{x}|\sigma_{1}^{(3)}\sigma_{2}^{(3)}|\boldsymbol{x}\rangle_{{\b}}$, by using the same properties of the Fermi matrices $\b_{j}$ and $\b_{j}^\dagger$.

A similar expression for the orthogonal basis $|\boldsymbol{x}\rangle_{{\c}}$ holds where the indices $j$ and $k$ are replaced with $j-\frac{1}{2}$ and $k-\frac{1}{2}$ respectively.

\subsubsection{Summary of results}
Table \ref{coeffList} list the results found above, making algebraic simplifications where possible.

\begin{table}\footnotesize
  \centering
  \begin{tabular}{cc}
  \toprule
    \multicolumn{2}{c}{\bf{Odd values of $\s=\sum_{j=1}^{n}x_{j}$}}\\
    Basis Matrix, $P$ & Value of $\,_{{\b}}\langle\boldsymbol{x}|P|\boldsymbol{x}\rangle_{{\b}}=\,_{t}\langle\boldsymbol{x}|P|\boldsymbol{x}\rangle_{t}$\\
    \cline{1-2}
	$\sigma_{1}^{(0)}\sigma_{2}^{(0)}$ & $1$\\
	$\sigma_{1}^{(0)}\sigma_{2}^{(1)}$ & $0$\\
	$\sigma_{1}^{(0)}\sigma_{2}^{(2)}$ & $0$\\
	$\sigma_{1}^{(0)}\sigma_{2}^{(3)}$ & $1-\frac{2\s}{n}$\\
	$\sigma_{1}^{(1)}\sigma_{2}^{(0)}$ & $0$\\
	$\sigma_{1}^{(1)}\sigma_{2}^{(1)}$ & $-\frac{1}{n}\sum_{j=1}^{n}\left(\omega_{j}^{1}(1-x_{j})-\omega_{j}^{-1}x_{j}\right)$\\
	$\sigma_{1}^{(1)}\sigma_{2}^{(2)}$ & $\frac{\im}{n}\sum_{j=1}^{n}\left(-\omega_{j}^{1}(1-x_{j})-\omega_{j}^{-1}x_{j}\right)$\\
	$\sigma_{1}^{(1)}\sigma_{2}^{(3)}$ & $0$\\
	$\sigma_{1}^{(2)}\sigma_{2}^{(0)}$ & $0$\\
	$\sigma_{1}^{(2)}\sigma_{2}^{(1)}$ & $\frac{\im}{n}\sum_{j=1}^{n}\left(\omega_{j}^{1}(1-x_{j})+\omega_{j}^{-1}x_{j}\right)$\\
	$\sigma_{1}^{(2)}\sigma_{2}^{(2)}$ & $\frac{1}{n}\sum_{j=1}^{n}\left(-\omega_{j}^{1}(1-x_{j})+\omega_{j}^{-1}x_{j}\right)$\\
	$\sigma_{1}^{(2)}\sigma_{2}^{(3)}$ & $0$\\
	$\sigma_{1}^{(3)}\sigma_{2}^{(0)}$ & $1-\frac{2\s}{n}$\\
	$\sigma_{1}^{(3)}\sigma_{2}^{(1)}$ & $0$\\
	$\sigma_{1}^{(3)}\sigma_{2}^{(2)}$ & $0$\\
	$\sigma_{1}^{(3)}\sigma_{2}^{(3)}$ & $-\frac{4}{n^{2}}\sum_{1\leq j<k\leq n}\left(\omega_{j}^{-1}\omega_{k}^{1}-2+\omega_{j}^{1}\omega_{k}^{-1}\right)(1-x_{j})(1-x_{k})+\frac{4\s}{n}-3$\\
    \midrule\multicolumn{2}{c}{\bf{Even values of $\s=\sum_{j=1}^{n}x_{j}$}}\\
    Basis Matrix, $P$ & Value of $\,_{{\c}}\langle\boldsymbol{x}|P|\boldsymbol{x}\rangle_{{\c}}=\,_{t}\langle\boldsymbol{x}|P|\boldsymbol{x}\rangle_{t}$\\
    \cline{1-2}
	$\sigma_{1}^{(0)}\sigma_{2}^{(0)}$ & $1$\\
	$\sigma_{1}^{(0)}\sigma_{2}^{(1)}$ & $0$\\
	$\sigma_{1}^{(0)}\sigma_{2}^{(2)}$ & $0$\\
	$\sigma_{1}^{(0)}\sigma_{2}^{(3)}$ & $1-\frac{2\s}{n}$\\
	$\sigma_{1}^{(1)}\sigma_{2}^{(0)}$ & $0$\\
	$\sigma_{1}^{(1)}\sigma_{2}^{(1)}$ & $-\frac{1}{n}\sum_{j=1}^{n}\left(\omega_{j-\frac{1}{2}}^{1}(1-x_{j})-\omega_{j-\frac{1}{2}}^{-1}x_{j}\right)$\\
	$\sigma_{1}^{(1)}\sigma_{2}^{(2)}$ & $\frac{\im}{n}\sum_{j=1}^{n}\left(-\omega_{j-\frac{1}{2}}^{1}(1-x_{j})-\omega_{j-\frac{1}{2}}^{-1}x_{j}\right)$\\
	$\sigma_{1}^{(1)}\sigma_{2}^{(3)}$ & $0$\\
	$\sigma_{1}^{(2)}\sigma_{2}^{(0)}$ & $0$\\
	$\sigma_{1}^{(2)}\sigma_{2}^{(1)}$ & $\frac{\im}{n}\sum_{j=1}^{n}\left(\omega_{j-\frac{1}{2}}^{1}(1-x_{j})+\omega_{j-\frac{1}{2}}^{-1}x_{j}\right)$\\
	$\sigma_{1}^{(2)}\sigma_{2}^{(2)}$ & $\frac{1}{n}\sum_{j=1}^{n}\left(-\omega_{j-\frac{1}{2}}^{1}(1-x_{j})+\omega_{j-\frac{1}{2}}^{-1}x_{j}\right)$\\
	$\sigma_{1}^{(2)}\sigma_{2}^{(3)}$ & $0$\\
	$\sigma_{1}^{(3)}\sigma_{2}^{(0)}$ & $1-\frac{2\s}{n}$\\
	$\sigma_{1}^{(3)}\sigma_{2}^{(1)}$ & $0$\\
	$\sigma_{1}^{(3)}\sigma_{2}^{(2)}$ & $0$\\
	$\sigma_{1}^{(3)}\sigma_{2}^{(3)}$ & $-\frac{4}{n^{2}}\sum_{1\leq j<k\leq n}\left(\omega_{j-\frac{1}{2}}^{-1}\omega_{k-\frac{1}{2}}^{1}-2+\omega_{j-\frac{1}{2}}^{1}\omega_{k-\frac{1}{2}}^{-1}\right)(1-x_{j})(1-x_{k})+\frac{4\s}{n}-3$\\
    \bottomrule
\end{tabular}
\caption[The values of $\,_{t}\langle\boldsymbol{x}|\sigma_{1}^{(a)}\sigma_{2}^{(b)}|\boldsymbol{x}\rangle_{t}$]{The values of $\,_{t}\langle\boldsymbol{x}|\sigma_{1}^{(a)}\sigma_{2}^{(b)}|\boldsymbol{x}\rangle_{t}$ collected and simplified from the calculations in the text.}
\label{coeffList}
\end{table}

From the values in Table \ref{coeffList} the value of
\begin{equation}
  \frac{1}{2^{n}}\sum_{\boldsymbol{x}}\Tr\left(\rho_{2,\boldsymbol{x}}^{2}\right)
  =\frac{1}{2^{n}}\sum_{\boldsymbol{x}}\frac{2^{2}}{2^{4}}\sum_{a,b=0}^{3}\,_{t}\langle\boldsymbol{x}|\sigma_{1}^{(a)}\sigma_{2}^{(b)}|\boldsymbol{x}\rangle_{t}^{2}
\end{equation}
as given in (\ref{tracerhoAl=2}), may be numerically calculated.  The bound given in Theorem \ref{invEnt} of Section \ref{EntLblock} (page \pageref{invEnt}) can again be investigated in a similar fashion as in Section \ref{Entl=1subsection}.  That is the values of
\begin{align}
  E_{n,2}=\frac{1}{2^{n}}\sum_{\boldsymbol{x}}\Tr\left(\rho_{2,\boldsymbol{x}}^{2}\right)-\frac{1}{2^{2}}
\end{align}
where $\rho_{2,\boldsymbol{x}}$ are the reduced density matrices, on the qubits labelled $1$ and $2$, of the eigenstates $|\boldsymbol{x}\rangle_{t}$ or of the eigenstates of samples\footnote{All the matrices sampled from $\hat{H}_{n}^{(inv,local)}$ for $n=2,3$ and $n=5,\dots,13$ had non-degenerate spectra so that their unique eigenstates (up to phase) are also eigenstates of the translation matrix $T$.} of the ensemble $\hat{H}_{n}^{(inv,local)}$, will be considered.  The values of $E_{n,2}^{-1}$ are plotted against $n$ in Figure \ref{l=2Fig}.  Here it is seen that for all the values of $E_{n,2}^{-1}$ considered, the bound
\begin{equation}
  E_{n,2}\leq\frac{2^{2}}{n}\quad\iff\quad \frac{1}{E_{n,2}}\geq \frac{n}{2^{2}}
\end{equation}
from Theorem \ref{invEnt} of Section \ref{EntLblock} is satisfied.  Furthermore, it appears that the bound
\begin{equation}
  E_{n,2}\leq\frac{1}{Cn}\quad\iff\quad \frac{1}{E_{n,2}}\geq Cn
\end{equation}
for some constant $C$, is asymptotically tight for the eigenstates $|\boldsymbol{x}\rangle_{t}$.  However, for many members of the ensemble $\hat{H}_{n}^{(inv,local)}$, Figure \ref{l=2Fig} suggests that an exponential bound, perhaps of the form
\begin{equation}
  E_{n,2}\leq\frac{1}{C2^{n}}\quad\iff\quad  \frac{1}{E_{n,2}}\geq C2^{n}
\end{equation}
for some constant $C$, may be appropriate.

\begin{figure}
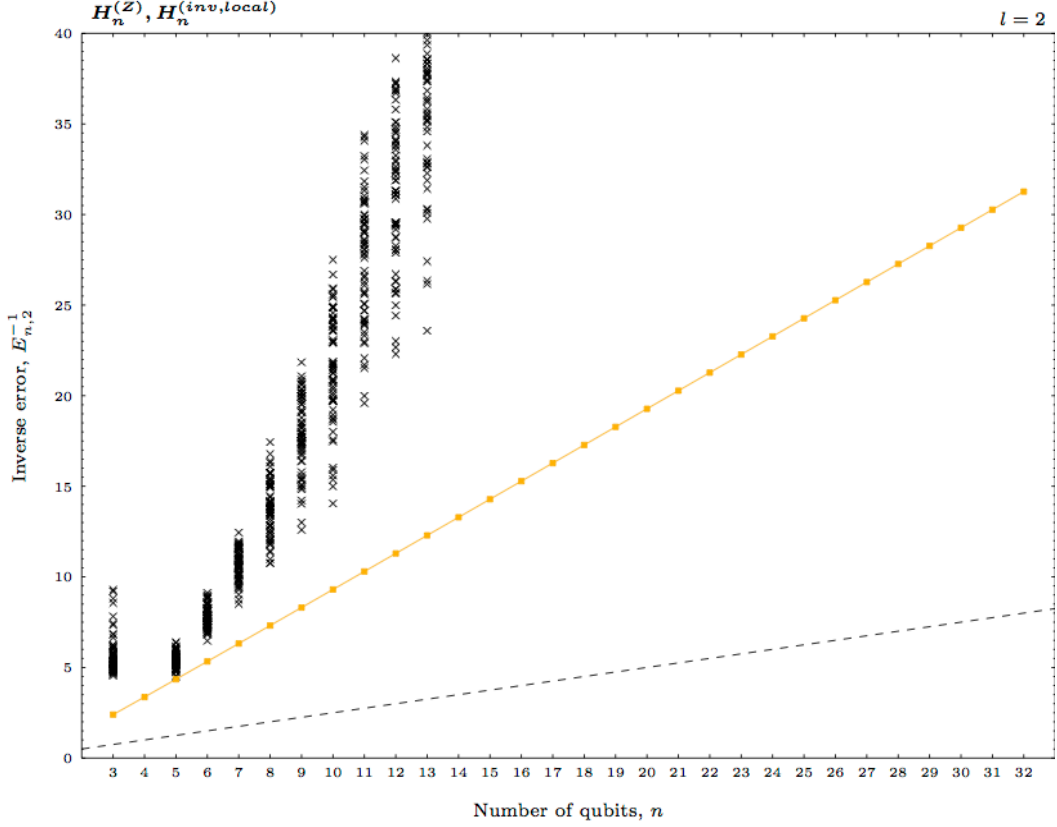

  \centering
  \include{Figures/Ent2}
  \caption[Average purity of the reduced eigenstates on $l=2$ qubits]{(Solid line) The values of $E_{n,2}^{-1}=\left(\frac{1}{2^{n}}\sum_{\boldsymbol{x}}\Tr\left(\rho_{2,\boldsymbol{x}}^{2}\right)-\frac{1}{2^{2}}\right)^{-1}$ against $n=3,\dots,32$, where $\rho_{2,\boldsymbol{x}}$ are the reduced density matrices on the qubits labelled $1$ and $2$ of the joint eigenstates $|\boldsymbol{x}\rangle_{t}$ of the translation matrix $T$ and spin chain Hamiltonian $H_{n}^{(Z)}$.  (Crosses) The same values for $2^6$ of the matrices sampled from the ensemble $\hat{H}_{n}^{(inv,local)}$ in Section \ref{eigenNum} for various values of $n$.  The $n=4$ values for $\hat{H}_{n}^{(inv,local)}$ are not included due to the presence of degenerate eigenvalues.  The bound from Theorem \ref{invEnt} of Section \ref{EntLblock} is given by the dashed line.}
  \label{l=2Fig}
\end{figure}

\subsection{Numerical values of the purity for \texorpdfstring{$l=3,4,5$}{l=3,4,5}}\label{l=3,4,5subsection}
In this section the results seen for $l=1$ and $l=2$ from the previous two sections will be extended to the values $l=3,4,5$ in order to further test the tightness of the bound in Theorem \ref{invEnt} of Section \ref{EntLblock} (page \pageref{invEnt}).  Following the method in the preceding sections leads to unwieldily algebraic equations.  However, the eigenstates $|\boldsymbol{x}\rangle_{t}$ constructed in Section \ref{JointEigen} can be numerically generated by explicitly constructing the Fermi matrices $\b_{j}$, $\b_{j}^\dagger$, $\c_{j}$ and $\c_{j}^\dagger$ from their corresponding definitions in (\ref{FermiBandC}).  To this end, the reduced density matrices $\rho_{l,\boldsymbol{x}}=\Tr_{\mathcal{B}}\left(|\boldsymbol{x}\rangle_{t}\,_{t}\langle\boldsymbol{x}|\right)$, where $\mathcal{B}$ is the Hilbert space of the qubits labelled $l+1$ to $n$, and therefore the purity $\Tr\left(\rho_{l,\boldsymbol{x}}^{2}\right)$ may be numerically calculated.  Due to the size of the Hilbert space, only values of $n$ up to $10$ for $l=1,2,3,4,5$ were numerically reached using this method.

Figure \ref{l=lFig} shows the plots of the corresponding values of $E_{n,l}^{-1}$ where
\begin{equation}
	E_{n,l}=\frac{1}{2^{n}}\sum_{\boldsymbol{x}}\Tr\left(\rho_{l,\boldsymbol{x}}^{2}\right)-\frac{1}{2^{l}}
\end{equation}
against $n$ for the values calculated.  The points for $l=1,2$ agree with the previous calculations with the points for $l=3,4,5$ producing similarly linear relations.  This suggests that the bound
\begin{equation}
  E_{n,l}\leq\frac{1}{c(l)n}\quad\iff\quad \frac{1}{E_{n,l}}\geq c(l)n
\end{equation}
for some real function $c(l)$, may indeed be asymptotically tight for this larger range of values of $l$.  This again adds evidence that the bound given in Theorem \ref{invEnt} is asymptotically tight, that is linear in $\frac{1}{n}$, for a wide range of values of $l$.

\begin{figure}
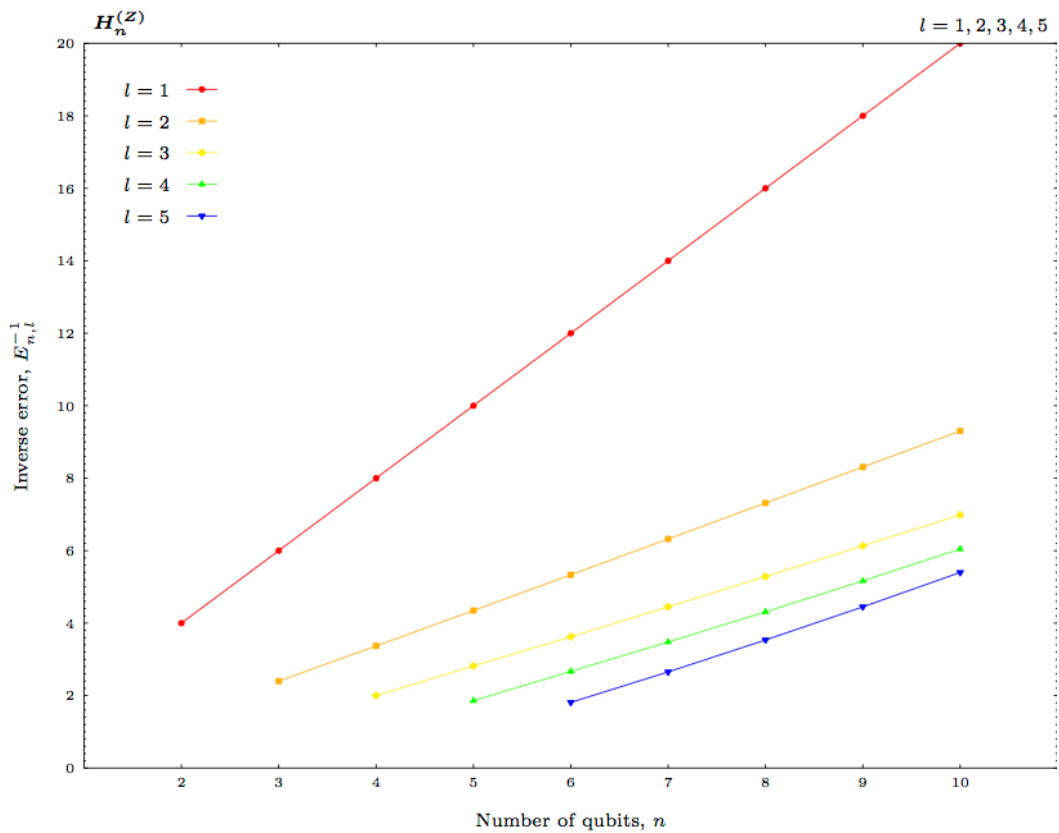

  \centering
  \include{Figures/Entl}
  \caption[Average purity of the reduced eigenstates on $l=1,2,3,4,5$ qubits]{The values of $E_{n,l}^{-1}=\left(\frac{1}{2^{n}}\sum_{\boldsymbol{x}}\Tr\left(\rho_{l,\boldsymbol{x}}^{2}\right)-\frac{1}{2^{l}}\right)^{-1}$ against $n=2,\dots,10$, where $\rho_{l,\boldsymbol{x}}$ are the reduced density matrices on the qubits labelled $1$ to $l$ of the joint eigenstates $|\boldsymbol{x}\rangle_{t}$ of the translation matrix $T$ and spin chain Hamiltonian $H_{n}^{(Z)}$.}
  \label{l=lFig}
\end{figure}
 
\clearpage
 
\chapter{Summary and outlook}
In Section \ref{Random matrix model} the most general $n$-qubit nearest-neighbour chain Hamiltonian was seen to be
\begin{equation}
	H_{n}^{(local)}=\sum_{j=1}^{n}\sum_{a=1}^{3}\sum_{b=0}^{3}\alpha_{a,b,j}\sigma_{j}^{(a)}\sigma_{j+1}^{(b)}
\end{equation}
for some real constants $\alpha_{a,b,j}$ for each value of $n=2,3,\dots$ separately, up to the addition of a term proportional to the identity.  The corresponding random matrix model
\begin{equation}
	\hat{H}_{n}^{(local)}=\sum_{j=1}^{n}\sum_{a=1}^{3}\sum_{b=0}^{3}\hat{\alpha}_{a,b,j}\sigma_{j}^{(a)}\sigma_{j+1}^{(b)}\qquad\qquad
	\hat{\alpha}_{a,b,j}\sim\mathcal{N}\left(0,\frac{1}{12n}\right) \iid
\end{equation}
was then introduced in order to analyse the spectral density of Hamiltonians of the form $H_{n}^{(local)}$.  The limiting moments of the spectral density of the simplified ensemble
\begin{equation}
	\hat{H}_{n}=\sum_{j=1}^{n}\sum_{a=1}^{3}\sum_{b=1}^{3}\hat{\alpha}_{a,b,j}\sigma_{j}^{(a)}\sigma_{j+1}^{(b)}\qquad\qquad
	\hat{\alpha}_{a,b,j}\sim\mathcal{N}\left(0,\frac{1}{9n}\right) \iid
\end{equation}
that is with the local terms proportional to $\sigma_{j}^{(a)}$ removed, were determined in Section \ref{Moment method}.  Here it was explicitly shown that these limiting moments were that of a standard normal distribution,
\begin{equation}
	\lim_{n\to\infty}\mathbb{E}\left(\Tr\left(\hat{H}_{n}^m\right)\right)\equiv\lim_{n\to\infty}\left\langle\Tr\Big(H_{n}^m\right)\Big\rangle=\begin{cases}0\qquad\qquad\qquad&\text{if $m$ is odd}\\\frac{m!}{2^{\frac{m}{2}}\left(\frac{m}{2}\right)!}&\text{if $m$ is even}\end{cases}
\end{equation}
for all $m\in\mathbb{N}_{0}$.

This result was extended to the weak convergence of the full spectral probability measure $\di\hat{\mu}_{n,1}$ for the ensemble $\hat{H}_{n}$ in Section \ref{Separating odd and even sites}.  It was shown that for any $x\in\mathbb{R}$
\begin{equation}
	\lim_{n\to\infty}\int_{-\infty}^{x}\di\hat{\mu}_{n,1}=\frac{1}{\sqrt{2\pi}}\int_{-\infty}^{x}\e^{-\frac{\lambda^2}{2}}\di\lambda
\end{equation}
with the related characteristic functions converging at a rate of at least $\frac{1}{\sqrt{n}}$.  The extension to the ensemble $\hat{H}_{n}^{(local)}$ and other more elaborate interaction geometries, such as the lattice, were also made.  The result was also seen not to depend on the probability measure of the ensemble in question, up to some reasonable conditions.  This indicates that the result is a feature of the type of Hamiltonian modelled and not just the type of ensemble chosen.

The peaks in the numerically simulated ensemble densities for finite values of $n$ (see Appendix \ref{specHist}) are reminiscent of the peaks predicted by the method of orthogonal polynomials for the GUE (see Section \ref{GUE}).  Such a connection is yet to be identified, however a first step to describing these peaks is made in Section \ref{JDOS} where the joint spectral density for such ensembles is considered.

An adaptation of the procedure used by Hartmann, Mahler and Hess \cite{Mahler} allowed the spectral convergence seen in Section \ref{Separating odd and even sites} to be shown on the level of sequences of fixed Hamiltonians $H_{n}^{(local)}$ for $n=2,3,\dots$.  In Section \ref{Splitting method} it was shown that if
\begin{equation}
		\lim_{n\to\infty}\frac{1}{2^{n}}\Tr\left( {H_{n}^{(local)}}^{2}\right)\equiv\lim_{n\to\infty}\sum_{j=1}^{n}\sum_{a=1}^{3}\sum_{b=0}^{3}\alpha_{a,b,j}^{2}=1\qquad\text{and}\qquad
		\left|\alpha_{a,b,j}\right|<\frac{C}{\sqrt{n}}
\end{equation}
for some positive constant $C\in\mathbb{R}$ independent of $a$, $b$, $j$ and $n$, then the spectral probability measures $\di\mu_{n,1}^{(local)}$ for the fixed Hamiltonians $H_{n}^{(local)}$ again tend weakly to that of a standard normal distribution.  This result was then generalised to more elaborate interaction geometries, such as the lattice, and interactions between qudits.

This limiting spectral density is exactly that conjectured by Atas and Bogomolny in \cite{Atas2014} for several spin chains of the form $H_{n}^{(local)}$.  Atas and Bogomolny though focus on approximating the spectral densities of finite length spin chains, in particular for the Ising model.  The Ising model with no magnetic field is highly degenerate.  It is shown in their paper that these degeneracies are removed by adding a transverse magnetic field and that the resulting spectral density (smoothed by integration over small intervals) is well approximated by a sum of Gaussians relating to each of the original degenerate eigensubspaces.  This produces densities containing several well defined peaks, similar to that seen for the ensemble $H_{n}^{(local)}$ in Appendix \ref{specHist}.  The connection between these two observations remains open.

The speed of the convergence of the spectral density was investigated in Section \ref{DOSrate}.  The spectrum of the Hamiltonian
\begin{equation}
	H_{n}^{(\epsilon XY+Z)}=\sum_{j=1}^{n}\left(\epsilon \sigma_{j}^{(1)}\sigma_{j+1}^{(2)}+\sigma_{j}^{(3)}\right)
\end{equation}
was explicitly determined and the rate of convergence of its spectral density to that of a Gaussian seen (numerically for $n=2,\dots,32$) to be approximately linear in $\frac{1}{n}$.  This rate is proportional to the inverse of the logarithm of the matrix dimension, much slower than for the GUE for example, for which the spectral convergence is proportional to the inverse of the matrix size \cite{Gotze2005}.  This rate was also numerically observed by diagonalising generic fixed Hamiltonians of the form
\begin{equation}
	H_{n}^{(JW)}=\sum_{j=1}^{n-1}\sum_{a,b=1}^{2}\alpha_{a,b,j}\sigma_{j}^{(a)}\sigma_{j+1}^{(b)}+\sum_{j=1}^{n}\alpha_{3,0,j}\sigma_{j}^{(3)}
\end{equation}
via the Jordan-Wigner transformation for $n=2,\dots,32$.

For the most general Hamiltonians of the form $H_{n}^{(local)}$, numerical diagonalisation was only practical for $n=2,\dots,13$.  Over this range the rate of the spectral convergence for sequences of generic fixed Hamiltonians could not be convincingly seen.  It remains an open question as to the rate at which the spectral densities of Hamiltonians of the form $H_{n}^{(local)}$ converge to a Gaussian, either for a specific sequence or on average over the ensemble $\hat{H}_{n}^{(local)}$.

\vspace{15pt}

To start analysing the correlations between eigenvalues in these spin chain models, spectral degeneracies were first looked at in Section \ref{degensec}.  It was shown that almost all members of the ensemble $\hat{H}_{n}^{(local)}$ for $n\geq2$ have simple spectra.  Whereas for $\hat{H}_{n}$ (where the local terms are removed) their always exists a Kramers' degeneracy for odd values of $n$, as in these cases each Hamiltonian exhibits an anti-unitary symmetry.  This was mirrored in the translationally-invariant ensembles
\begin{align}
    \hat{H}^{(inv)}_{n}&=\sum_{j=1}^{n}\sum_{a=1}^{3}\sum_{b=1}^{3}\hat{\alpha}_{a,b}\sigma_{j}^{(a)}\sigma_{j+1}^{(b)}\qquad\qquad
	\hat{\alpha}_{a,b}\sim\mathcal{N}\left(0,\frac{1}{9n}\right) \iid\nonumber\\
    \hat{H}_{n}^{(inv,local)}&=\sum_{j=1}^{n}\sum_{a=1}^{3}\sum_{b=0}^{3}\hat{\alpha}_{a,b}\sigma_{j}^{(a)}\sigma_{j+1}^{(b)}\qquad\qquad
	\hat{\alpha}_{a,b}\sim\mathcal{N}\left(0,\frac{1}{12n}\right) \iid
\end{align}
However, for $\hat{H}_{n}^{(inv,local)}$ the simplicity of the members' spectra was only proved in the case of odd prime values of $n$.  This is conjectured for all values of $n>4$, but remains an open problem.

In Section \ref{nnnumerics} the nearest-neighbour level statistics of the ensembles defined above were numerically simulated to further the analysis of higher eigenvalue correlations.  Close matches to the GOE statistics were seen for $\hat{H}_{n}$ when $n$ was even, all the members of these ensembles having an anti-unitary symmetry which squares to the identity.  For $\hat{H}_{n}$ when $n$ was odd, GSE statistics (along with the Kramers' degeneracy mentioned before) were seen, with all the members of the ensembles having an anti-unitary symmetry which squares to the negative of the identity.  For $\hat{H}_{n}^{(local)}$, GUE statistics were observed, the members having no symmetries apart from their inherent Hermitian symmetry.

For $\hat{H}_{n}^{(inv, local)}$ Poisson nearest-neighbour level spacing statistics were observed.  The members of this ensemble each have a translational symmetry so that they are block diagonal in a basis of eigenstates of the corresponding translation matrix, with $n$ blocks corresponding to the $n$ eigenstates of this translation matrix.  As was commonly observed in Section \ref{nonInt}, repulsion or correlations between these block are not expected, so leading to the Poisson (independent) spacing statistics.

The rigorous reasons to the appearance of the canonical nearest-neighbour level spacing distributions is still an open question.  An attempt at an answer was made in Section \ref{JDOS} where a conjecture was heuristically outlined for the joint spectral densities of ensembles, parametrised by $m=1,\dots,4^n$, of the form
\begin{equation}
	\hat{H}_{n}^{(\r)}=\sum_{j=1}^{\r}\hat{\alpha}_{j}P_{j}\qquad\qquad\hat{\alpha}_{j}\sim\mathcal{N}\left(0,\frac{1}{\r}\right) \iid
\end{equation}
where $\{P_j\}_{j=1}^{4^n}=\{\sigma^{(a_1)}\otimes\dots\otimes\sigma^{(a_n)}\}_{a_1,\dots,a_n=0}^3$.  It is conjectured that the joint spectral densities for these ensemble are given by
  \begin{equation}
    C\e^{-\frac{\r\boldsymbol{\lambda}^{2}}{2^{n+1}}}\Delta^2(\boldsymbol{\lambda})\int_{\mathbb{R}^{4^{n}-\r}}\frac{\det_{1\leq j,k\leq 2^{n}}\left(\e^{\im\lambda_{j}\mu_{k}(\boldsymbol{\beta})}\right)}{\Delta(\boldsymbol{\lambda})\Delta({\boldsymbol{\mu}(\boldsymbol{\beta})})}\boldsymbol{\di}\boldsymbol{\beta}
  \end{equation}
where
\begin{equation}
  C=\frac{2^{\frac{n4^{n}}{2}}\r^{\frac{\r}2{}}(2\pi)^{\frac{\r}{2}}\im^{\frac{2^{n}}{2}}}  {2^{n}!2^{n\r}(2\pi)^{\frac{2^{n}}{2}}(2\pi)^{\frac{4^{n}}{2}}\im^{\frac{4^{n}}{2}}}\qquad
  \Delta(\boldsymbol{\lambda})=\prod_{1\leq j<k\leq2^{n}}(\lambda_{k}-\lambda_{j})=\det_{1\leq,j,k\leq2^{n}}\left(\lambda_{j}^{k-1}\right)
\end{equation}
and $\boldsymbol{\mu}(\boldsymbol{\beta})=(\mu_{1}(\boldsymbol{\beta}),\dots,\mu_{2^{n}}(\boldsymbol{\beta}))$ are the eigenvalues of the matrix
\begin{equation}
  B=\sum_{\r<j\leq 4^{n}}\beta_{j}P_{j}
\end{equation}
and $\boldsymbol{\beta}$ is the vector of the real parameters $\beta_{j}$ present in $B$.  For $m=4^n$ the integral over $\mathbb{R}^{0}$ is taken as a unit factor.

The heuristic argument for this conjecture relies on the Harish-Chandra-Itzykson-Zuber integral.  For $\hat{H}_{2}$ this conjecture was numerical verified for the $1$-point and $2$-point correlation functions.  Linear repulsion was numerically observed in the $2$-point correlation function.  When $m=4^n$ the GUE's joint spectral density is recovered.  It was hoped that the determinantal form would allow for the integration in the expression for the joint spectral density to be generically performed, but this remains an open problem. 

\vspace{15pt}

Eigenstate entanglement was first tackled in Section \ref{SingleEnt}.  Here, entanglement between a single qubit and the rest of the chain within the eigenstates of Hamiltonians of the form $H_{n}$ with simple spectra were considered.  Entanglement between a single qubit and the rest of the chain was characterised by the purity of the reduced eigenstates on a single qubit.  The minimal value of $\frac{1}{2}$ of this purity corresponding to a maximally entangled original state, and a maximal value of $1$ corresponding to the original state being a product state.  It was found that for every such eigenstate the purity is equal to its minimal value of $\frac{1}{2}$.

The results of Section \ref{degensec} now prove useful as it is seen that almost all members of the ensemble $\hat{H}_{n}$ for even values of $n=2,\dots,12$ have simple spectra (it is conjectured that this is the case for all even values of $n$).  This then provides a large class of Hamiltonians to which the result above applies.

This result was extended to a block of a fixed number, $l$, of qubits in Section \ref{EntLblock} for joint eigenstates of the translation matrix and translationally-invariant Hamiltonians of the form $H_{n}^{(inv, local)}$.  It was seen that the average purity, over a complete orthonormal basis of such eigenstates of a fixed Hamiltonian, was at most $\frac{2^l}{n}$ away from its minimal value (this minimal value signifying maximum entanglement).  The method of proof for this result was also seen to apply to more general systems, such as a lattice, in which a suitable translational symmetry was present.

Again the results of Section \ref{degensec} prove useful as it is seen that almost all members of the ensemble $\hat{H}_{n}^{(inv, local)}$ for odd prime values of $n$ have simple spectra (although this is conjectured for all $n>4$).  This then provides a large class of Hamiltonians for which the eigenstates are unique (up to phase) to which the result above must apply.  From numerical evidence, a similar result is believed to hold for non-translationally-invariant Hamiltonians of the form $H_{n}^{(local)}$.

The tightness of the reduced eigenstate purity bound derived in Section \ref{EntLblock} was analysed in Section \ref{eigenstatebounds}.  The Hamiltonian
\begin{equation}
	H_{n}^{(Z)}=\sum_{j=1}^{n}\sigma_{j}^{(3)}
\end{equation}
was taken as an example and a complete set of joint eigenstates with the translation matrix derived.  For the case $l=1$ the purity bound derived in Section \ref{EntLblock} was analytically seen to be asymptotically tight, that is linear in $\frac{1}{n}$.  For $l=2$ this tightness was numerical verified for $n=2,\dots,32$ whereas for $l=3,4,5$ this was numerically verified for $n=2,\dots,10$.  An analytic verification for $l\neq1$ remains an open problem.

For large values of $n$ the Hamiltonian $H_{n}^{(Z)}$ has a highly degenerate spectrum so it is not surprising that eigenstates fulfilling the bound of Section \ref{EntLblock} can be found within each of the (large) eigensubspaces.  However, without further assumptions in the theorem of Section \ref{EntLblock}, this still provides evidence that the bound given is tight.

The numerical result for appropriate non-degenerate Hamiltonians do not convincingly indicate the tightness of this bound as only small matrix sizes are treatable, as seen is Section \ref{eigenstatebounds}.

\vspace{15pt}

A contrast is now seen between the spectral density and eigenstate statistics of Hamiltonians of the form $H_{n}^{(inv, local)}$ (and conceivably $H_{n}^{(local)}$ from numerical evidence).  The limiting spectral density for large $n$ is that of a chain of non-interacting qubits, as given by the results in Section \ref{Splitting method}.  For such non-interacting Hamiltonians a basis of eigenstates can always be chosen which are product states over the $n$ qubits.  The results of Section \ref{EntLblock} and \ref{degensec} prove that this is not, in some maximal way, the case for most Hamiltonians of the form $H_{n}^{(inv, local)}$, at least for odd prime values of $n$.

This observation is intuitively connected to quantum statistical mechanics and the evolution of a system to thermal equilibrium over time \cite{Linden2009}.  Consider the $n$-qubit chain with Hamiltonian $H_{n}^{(inv, local)}$ split into a subsystem $A$ (the qubits labelled $1$ to $l$) and a large bath (the qubits labelled $l+1$ to $n$).  The presence of the interaction terms in $H_{n}^{(inv,local)}$ between subsystem $A$ and the bath is seen to dramatically affect the eigenstates of the Hamiltonian, in contrast to if it were removed.  The presence of such interaction terms, even though they can be arbitrarily small compared to the system's size, allows the thermalisation between subsystem $A$ and the bath when the whole system is allowed to evolve over time.  This dramatic and sharp change in the structure of the eigenstates is also consistent with the (infinite temperature) eigenstate thermalisation hypothesis \cite{Deutsch1991,Srednicki1994}.  Here the evolution of a system to that predicted by a microcanonical ensemble is linked to the structure of the eigenstates of the system.

\clearpage

\appendix
\chapter{Standard results}
This appendix contains the sketches of several standard results required throughout this research. 

\section{Standard random matrix theory calculations}\label{stdRMT}
In this appendix some of the standard calculations needed to derive the joint probability density function for the eigenvalues of an important class of random matrix ensembles will be briefly sketched.  A summary of these calculations is contextualised in Section \ref{GUE} and are well documented in {\cite[\p31]{Snaith}}, \cite[\p66,103]{Handbook} and \cite[\p354]{Mehta}, on which the following summary is based.

\subsection{Determinantal form of the joint spectral density}\label{AppGUEdet}
Let $\hat{\rho}_{N,N}^{(det)}(\lambda_{1},\dots,\lambda_{N})$ be a joint probability density function of the form
\begin{equation}
  \hat{\rho}_{N,N}^{(det)}(\lambda_{1},\dots,\lambda_{N})=C\left(\prod_{j=1}^{N}\omega(\lambda_{j})\right)\prod_{1\leq k<l\leq N}(\lambda_{l}-\lambda_{k})^{2}
\end{equation}
on $J^{N}$ for some real interval $J\subset\mathbb{R}$, where $\omega(x)$ is a real positive function such that $\int_{J}x^{m}\omega(x)\di x<\infty$ for all $m\in\mathbb{N}_{0}$ and $C$ is some normalisation constant.  The Gaussian unitary ensemble (GUE) has this form, up to some parameter scaling, where $\omega(x)=\e^{-x^{2}}$ and $J=\mathbb{R}$.

Let $p_{j}$ be monic polynomials of order $j$ which are orthogonal with respect to $\omega$ on $J$, that is
\begin{equation}
  (p_{j},p_{k})\equiv\int_{J}\omega(x)p_{j}(x)p_{k}(x)\di x=\delta_{j,k}(p_{j},p_{j})
\end{equation}
Furthermore let the coefficients of $p_{j}$ be denoted by $c_{j,k}$ such that $p_{j}(x)=x^{j}+c_{j,j-1}x^{j-1}+\dots+c_{j,0}$.

The Vandermonde determinant is defined to be
\begin{equation}
  \det_{1\leq k,l\leq N}\left(\lambda_{l}^{k-1}\right)=\prod_{1\leq k<l\leq N}(\lambda_{l}-\lambda_{k})
\end{equation}
Adding a multiple of one row of a matrix to another does not affect that matrices' determinant.  Therefore, for each row labelled $j=N,\dots,2$ in turn, adding $c_{j-1,k-1}$ multiples of each row labelled $k=1,\dots,j-1$ to it keeps the determinant of the Vandermonde matrix $\left(\lambda_{l}^{k-1}\right)$ unchanged.  Transposition of a matrix also leaves its determinant unchanged, so that
\begin{equation}
  \det_{1\leq k,l\leq N}\left(\lambda_{l}^{k-1}\right)
  =\det_{1\leq k,l\leq N}\Big(p_{k-1}(\lambda_{l})\Big)
  =\det_{1\leq k,l\leq N}\Big(p_{l-1}(\lambda_{k})\Big)
\end{equation}

Multiplying a row or column of a matrix by a complex scalar also multiples the determinant of that matrix by the same scalar.  This fact then allows $\hat{\rho}_{N,N}^{(det)}(\lambda_{1},\dots,\lambda_{N})$ to be written as
\begin{equation}
  C\left(\prod_{j=0}^{N-1}(p_{j},p_{j})\right)
  \det_{1\leq k,l\leq N}\left(\frac{\omega^{\frac{1}{2}}(\lambda_{k})}{(p_{l-1},p_{l-1})^{\frac{1}{2}}}p_{l-1}(\lambda_{k})\right)
  \det_{1\leq k,l\leq N}\left(\frac{\omega^{\frac{1}{2}}(\lambda_{l})}{(p_{k-1},p_{k-1})^{\frac{1}{2}}}p_{k-1}(\lambda_{l})\right)
\end{equation}
Finally, the product of determinants is equal to the determinant of the product of the associated matrices so that
\begin{equation}
  \hat{\rho}_{N,N}^{(det)}(\lambda_{1},\dots,\lambda_{N})=C\left(\prod_{j=0}^{N-1}(p_{j},p_{j})\right)\det_{1\leq k,l\leq N}\big(K_{N}(\lambda_{k},\lambda_{l})\big)
\end{equation}
where the kernel $K_{N}$ is defined to be
\begin{equation}\label{kernal}
  K_{N}(x,y)=\omega^{\frac{1}{2}}(x)\omega^{\frac{1}{2}}(y)\sum_{j=0}^{N-1}\frac{p_{j}(x)p_{j}(y)}{(p_{j},p_{j})}
\end{equation}

\subsection{Point correlation functions}\label{AppGUEpoint}
The $r$-point correlation functions
\begin{equation}
  \hat{\rho}^{(det)}_{N,r}(\lambda_{1},\dots,\lambda_{r})=\int_{J^{N-r}}\hat{\rho}_{N,N}^{(det)}(\lambda_{1},\dots,\lambda_{N})\di\lambda_{r+1}\dots\di\lambda_{N}
\end{equation}
for $r=1,\dots,N-1$ can be calculated by applying Guadin's lemma, see Section \ref{GUEcorr}.  Guadin's lemma applies as
\begin{equation}
  \int_{J}K_{N}(x,y)K_{N}(y,z)\di y=\omega^{\frac{1}{2}}(x)\omega^{\frac{1}{2}}(z)\sum_{j,j^\prime=0}^{N-1}\frac{p_{j}(x)p_{j^\prime}(z)(p_{j},p_{j^\prime})}{(p_{j},p_{j})(p_{j^\prime},p_{j^\prime})}=K_{N}(x,z)
\end{equation}
by the orthogonality of the polynomials $p_{j}$, and
\begin{equation}
  \int_{J}K_{N}(x,x)\di x=\sum_{j=0}^{N-1}\frac{(p_{j},p_{j})}{(p_{j},p_{j})}=N
\end{equation}
Guadin's lemma applied $N-r$ times then yields that
\begin{equation}
  \hat{\rho}^{(det)}_{N,r}(\lambda_{1},\dots,\lambda_{r})=C\left(\prod_{j=0}^{N-1}(p_{j},p_{j})\right)(N-r)!\det_{1\leq k,l\leq r}\big(K_{N}(\lambda_{k},\lambda_{l})\big)
\end{equation}
as the integrals over $J$ produce the sequence of constants $N-(N-1), N-(N-2), \dots,N-r$.

The constant $C$ can be determined by performing all $N$ integrals.  To insure that $\hat{\rho}_{N,N}^{(det)}$ is a probability density function it must hold that
\begin{equation}
  1=\int_{J^{N}}\hat{\rho}_{N,N}^{(det)}(\lambda_{1},\dots,\lambda_{N})\di\lambda_{1}\dots\di\lambda_{N}=C\left(\prod_{j=1}^{N-1}(p_{j},p_{j})\right)N!
\end{equation}
The zero dimensional determinant, produced by applying Gaudin's lemma $N$ times to produced the identity above, evaluating to unity by convention.  It then follows that
\begin{equation}
  \hat{\rho}^{(det)}_{N,r}(\lambda_{1},\dots,\lambda_{r})=\frac{(N-r)!}{N!}\det_{1\leq k,l\leq r}\big(K_{N}(\lambda_{k},\lambda_{l})\big)
\end{equation}

\subsection{The spectral density for the GUE}\label{AppGUE1point}
For the GUE, $J=\mathbb{R}$ and $\omega(x)=\e^{-x^{2}}$ in the previous construction, up to some parameter scaling.  The polynomials which are orthogonal with respect to this weight $\omega$ on $\mathbb{R}$ are the monic Hermite polynomials
\begin{equation}
  q_{j}(x)=\frac{(-1)^{j}\e^{x^{2}}}{2^{j}}\frac{\di^{j}}{\di x^{j}}\e^{-x^{2}}
\end{equation}
and satisfy
\begin{equation}
  (q_{j},q_{k})=\frac{j!\sqrt{\pi}\delta_{j,k}}{2^{j}}
\end{equation}
In this case the exact spectral density for the GUE reads
\begin{equation}
 \hat{\rho}_{N,1}^{(GUE)}(\lambda)=\frac{1}{N}K_{N}^{(GUE)}(\lambda,\lambda)=\frac{\e^{-\lambda^{2}}}{N}\sum_{j=0}^{N-1}\frac{2^{j}q_{j}^{2}(\lambda)}{j!\sqrt{\pi}}
\end{equation}
The asymptotic spectral density for the GUE is given by the well known semicircle law
\begin{equation}
 \hat{\rho}_{N,1}^{(GUE)}=\frac{1}{N}K_{N}^{(GUE)}(\lambda,\lambda)\sim\frac{\sqrt{2}}{\pi\sqrt{N}}\begin{cases}\sqrt{1-\frac{\lambda^{2}}{2N}}\qquad&\text{if }|\lambda|<\sqrt{2N}\\
 0&\text{else}\end{cases}
\end{equation}

\subsection{The sine kernel}\label{AppGUE1kernel}
The Christoffel-Darboux formula and the asymptotics of the Hermite polynomials (see Section \ref{GUElevel}) allow the kernel $K^{(GUE)}_{N}(x,y)$ for the GUE to be approximated in the large $N$ limit.  By definition
\begin{equation}
  K_{N}^{(GUE)}(x,y)=\e^{-\frac{x^{2}}{2}}\e^{-\frac{y^{2}}{2}}\sum_{j=0}^{N-1}\frac{q_{j}(x)q_{j}(y)}{(q_{j},q_{j})}
\end{equation}
The Christoffel-Darboux formula reduces this to the sum of two terms
\begin{equation}
  \e^{-\frac{x^{2}}{2}}\e^{-\frac{y^{2}}{2}}\frac{q_{N}(x)q_{N-1}(y)-q_{N-1}(x)q_{N}(y)}{(q_{N-1},q_{N-1})(x-y)}
\end{equation}

The asymptotic formula for the Hermite polynomials approximates the kernel $K^{(GUE)}_{N}(x,y)$ for large $N$ as
\begin{align}
  &\frac{\Gamma\left(\frac{N+1}{2}\right)\Gamma\left(\frac{N}{2}\right)}{(q_{N-1},q_{N-1})\pi(x-y)}\Bigg(\cos\left(x\sqrt{2N+1}-\frac{N\pi}{2}\right)\cos\left(y\sqrt{2N-1}-\frac{N\pi}{2}+\frac{\pi}{2}\right)\nonumber\\
  &\qquad\qquad\qquad\qquad\qquad\qquad-\cos\left(x\sqrt{2N-1}-\frac{N\pi}{2}+\frac{\pi}{2}\right)\cos\left(y\sqrt{2N+1}-\frac{N\pi}{2}\right)\Bigg)
\end{align}
Using the identities
\begin{align}
  \frac{\Gamma\left(\frac{N+1}{2}\right)\Gamma\left(\frac{N}{2}\right)}{(q_{N-1},q_{N-1})}&=1\nonumber\\
  \cos\left(x+\frac{\pi}{2}\right)&=-\sin\left(x\right)\nonumber\\
  \sin(x-y)&=\sin(x)\cos(y)-\cos(x)\sin(y)\nonumber\\
  \cos\left(x\sqrt{2N\pm1}-\frac{N}{2}\right)&=\cos\left(x\sqrt{2N}-\frac{N}{2}\right)+O\left(\frac{1}{\sqrt{N}}\right)\nonumber\\
  \sin\left(x\sqrt{2N\pm1}-\frac{N}{2}\right)&=\sin\left(x\sqrt{2N}-\frac{N}{2}\right)+O\left(\frac{1}{\sqrt{N}}\right)
\end{align}
this expression can be rewritten as
\begin{equation}
  K_{N}^{(GUE)}(x,y)\sim\frac{\sin\left(\sqrt{2N}\left(x-y\right)\right)}{\pi(x-y)}
\end{equation}
for fixed $x,y\in\mathbb{R}$.  This is known as the sine kernel and is seen for many ensembles.  A different kernel holds on the edge of the support of the limiting density, that is for $x,y\sim\sqrt{2N}$.

\subsection{Level repulsion}\label{AppGUE1level}
The $2$-point correlation function for the GUE is given by
\begin{equation}
  \hat{\rho}_{N,2}^{(GUE)}(\lambda_{1},\lambda_{2})=\frac{(N-2)!}{N!}\det_{1\leq k,l\leq N}\left(K_{N}^{(GUE)}(\lambda_{k},\lambda_{l})\right)
\end{equation}
For large $N$ and $\lambda_{1}=C+o\left(\frac{1}{\sqrt{N}}\right)$ for some constant $C$ and small spacings $|\lambda_{1}-\lambda_{2}|=o\left(\frac{1}{\sqrt{N}}\right)$ it in fact holds that
\begin{align}
  K_{N}^{(GUE)}(\lambda_{1},\lambda_{2})&\sim\frac{\sin\left(\sqrt{2N}\left(\lambda_{1}-\lambda_{2}\right)\right)}{\pi(\lambda_{1}-\lambda_{2})}
  \sim\left(\frac{\sqrt{2N}}{\pi}-\frac{(2N)^{\frac{3}{2}}(\lambda_{1}-\lambda_{2})^{2}}{3!\pi}\right)
\end{align}
Then by definition
\begin{align}
  \hat{\rho}_{N,2}^{(GUE)}(\lambda_{1},\lambda_{2})&\sim\frac{1}{N(N-1)}\left(K_{N}^{(GUE)}(\lambda_{1},\lambda_{1})K_{N}^{(GUE)}(\lambda_{2},\lambda_{2})-K_{N}^{(GUE)}(\lambda_{1},\lambda_{2})K_{N}^{(GUE)}(\lambda_{2},\lambda_{1})\right)\nonumber\\
  &\sim\frac{1}{N(N-1)}\left(\frac{2N}{\pi^{2}}-\frac{2N}{\pi^{2}}+\frac{2(2N)^{2}}{\pi^{2}3!}(\lambda_{1}-\lambda_{2})^{2}\right)
\end{align}
Now $\hat{\rho}_{N,2}^{(GUE)}(\lambda_{1},\lambda_{2})$ can be seen to be approximately proportional to the small square spacing $(\lambda_{1}-\lambda_{2})^{2}$.  This highlights the quadratic nature of the level repulsion in the GUE, as observed for the $2\times2$ examples is Section \ref{spacingStat}.  Similar results hold for small spacings in the bulk of the spectrum (say $|\lambda_1|,|\lambda_2|\sim\sqrt{N}$) with different behaviour for small spacings at the edge ($|\lambda_1|,|\lambda_2|\sim\sqrt{2N}$).

\section{Free parameters in unitary matrices}\label{unitaryParam}
The following is based on \cite{Neumann1929}.  The columns (or equivalently rows) of a $N\times N$ unitary matrix form an orthonormal basis of $\mathbb{C}^{N}$, under the standard Euclidean inner-product, by definition.  The number of real parameters needed to parametrise such a basis, and therefore such a matrix, will now be calculated:

A vector in $\mathbb{C}^{N}$ has $2N$ free real parameters.  Constraining this vector to be normalised removes one degree of freedom.   The subspace of $\mathbb{C}^{N}$ orthogonal to this first vector is $\mathbb{C}^{N-1}$, it contains vectors with $2(N-1)$ free parameters.  Constraining this second vector to be normalised, again removes one degree of freedom.

Proceeding in this way, $N$ orthonormal vectors (necessarily an orthonormal basis for $\mathbb{C}^{N}$) can be chosen and the number of free parameters counted as
\begin{equation}
	2N-1+2(N-1)-1+2(N-2)-1+\dots+2(2)-1+2(1)-1=N^{2}
\end{equation}

As a side note, given an orthonormal basis as above, the complex phase of each element (a complex scale factor with unit modulus) can be described by one real parameter.  There are then $N$ parameters out of the $N^{2}$ which describe the complex phases of the basis vectors.

\section{The trace of tensor products}\label{tensorProd}
Let $a$ and $b$ be integers and then let $A$ be an arbitrary $a\times a$ matrix and $B$ be an arbitrary $b\times b$ matrix.  By the definition of the tensor product, see Section \ref{Quantum},
\begin{equation}
	\Tr\left(A\otimes B\right)=\Tr\left(A\right)\Tr\left(B\right)
\end{equation}
Similarly
\begin{align}
	\Tr\left(A\otimes I_{b}\right)=\Tr\left(A\right)\Tr\left(I_{b}\right)=b\Tr\left(A\right)\nonumber\\
	\Tr\left(I_{a}\otimes B\right)=\Tr\left(I_{a}\right)\Tr\left(B\right)=a\Tr\left(B\right)
\end{align}
where $I_{a}$ and $I_{b}$ are the  $a\times a$ and $b\times b$ identities respectively.  Therefore
\begin{equation}
	\frac{1}{ab}\Tr\left(A\otimes B\right)=\frac{1}{ab}\Tr\left(A\otimes I_{b}\right)\cdot\frac{1}{ab}\Tr\left(I_{a}\otimes B\right)
\end{equation}
By induction, similar identities hold for larger tensor products, for example with the addition of the $c\times c$ matrix $C$
\begin{equation}
	\frac{1}{abc}\Tr\left(A\otimes B\otimes C\right)=\frac{1}{abc}\Tr\left(A\otimes I_{b}\otimes I_{c}\right)\cdot\frac{1}{abc}\Tr\left(I_{a}\otimes B\otimes I_{c}\right)\cdot\frac{1}{abc}\Tr\left(I_{a}\otimes I_{b}\otimes C\right)
\end{equation}

\section{The discrete Fourier transform}\label{DisFour}
The $N\times N$ matrices $U$ and $V$ with elements given by
\begin{equation}
  U_{j,k}=\frac{1}{\sqrt{N}}\omega_{j}^{k}\qquad\qquad V_{j,k}=\frac{1}{\sqrt{N}}\omega_{j-\frac{1}{2}}^{k}
\end{equation}
where $\omega_{j}=\e^{\frac{2\pi\im j}{N}}$, represent the discrete Fourier transform with periodic and anti-periodic boundary conditions \cite[\p389]{Bernstein2009}.  The matrices $U$ and $V$ will now be shown to be unitary.  By the definition of $U$, $U^\dagger$, $\omega_{j}$ and the geometric summation formula, for $j\neq k$
\begin{equation}
  \left(UU^\dagger\right)_{j,k}
  =\frac{1}{N}\sum_{l=1}^{N}\omega_{j}^{l}\omega_{k}^{-l}
  =\frac{1}{N}\sum_{l=1}^{N}\omega_{j-k}^{l}
  =\omega_{j-k}\frac{1-\omega_{j-k}^{N}}{N(1-\omega_{j-k})}
  =\omega_{j-k}\frac{1-1}{N(1-\omega_{j-k})}
  =0
\end{equation}
If $j=k$ then $\omega_{j-k}^{l}=1$ and $\left(UU^\dagger\right)_{j,j}=1$.  Therefore $UU^\dagger=I$.  A similar calculation shows that $U^\dagger U=V^\dagger V=VV^\dagger=I$.

\section{Gaussian integrals}\label{Int}
The following is based on \cite[\p400]{Zwillinger2002}.  In this appendix some of the Gaussian integration results used previously will be briefly sketched.  It is first recalled that for some non-zero real number $s$,
\begin{equation}
	\int_{-\infty}^\infty\e^{-\frac{x^{2}}{2s^{2}}}\di x=\sqrt{2\pi s^{2}}
\end{equation}

\subsubsection{Trigonometric average}
Let $c$ and $s$ be non-zero real numbers and $\hat{x}$ be a normal random variable with mean zero and variance $s^{2}$.  The integral
\begin{equation}
	\mathbb{E}\left(\cos\left(c\hat{x}\right)\right)=\frac{1}{\sqrt{2\pi s^{2}}}\int_{-\infty}^\infty\cos\left(cx\right)\e^{-\frac{x^{2}}{2s^{2}}}\di x
\end{equation}
may be calculated by considering the cosine as the real part of a complex exponential and then evaluating its average by `completing the square'.  That is,
\begin{align}
	\frac{1}{\sqrt{2\pi s^{2}}}\int_{-\infty}^\infty\cos\left(cx\right)\e^{-\frac{x^{2}}{2s^{2}}}\di x
	&=\frac{\Rp}{\sqrt{2\pi s^{2}}}\int_{-\infty}^\infty\e^{\im cx}\e^{-\frac{x^{2}}{2s^{2}}}\di x\nonumber\\
	&=\frac{\Rp}{\sqrt{2\pi s^{2}}}\int_{-\infty}^\infty\e^{-\frac{1}{2}\left(\frac{x}{s}-\im sc\right)^{2}}\e^{-\frac{s^{2}c^{2}}{2}}\di x\nonumber\\	
	&=\frac{\Rp}{\sqrt{2\pi}}\int_{-\infty-\im sc}^{\infty-\im sc}\e^{-\frac{y^{2}}{2}}\e^{-\frac{s^{2}c^{2}}{2}}\di y\nonumber\\
	&=\e^{-\frac{s^{2}c^{2}}{2}}
\end{align}
Here the contour in the last integral may be deformed back to the real line without affecting the integral's value.

\subsubsection{Gaussian moments}
Again let $s$ be a non-zero real number and $\hat{x}$ be a normal random variable with mean zero and variance $s^{2}$.  The moments
\begin{equation}
	\mathbb{E}\left(\hat{x}^{2k}\right)=\frac{1}{\sqrt{2\pi s^{2}}}\int_{-\infty}^\infty x^{2k}\e^{-\frac{x^{2}}{2s^{2}}}\di x
\end{equation}
for $k\in\mathbb{N}_0$ will now be calculated.  For the real parameter $c$
\begin{align}
	\frac{1}{\sqrt{2\pi s^{2}}}\int_{-\infty}^\infty x^{2k}\e^{-\frac{x^{2}}{2s^{2}}}\di x
	&=\frac{1}{\sqrt{2\pi s^{2}}}\int_{-\infty}^\infty \left(-2s^{2}\right)^{k}\frac{\partial^{k}}{\partial c^{k}}\left(\e^{-\frac{cx^{2}}{2s^{2}}}\right)\bigg|_{c=1}\di x\nonumber\\
	&=\left(-2s^{2}\right)^{k}\frac{\partial^{k}}{\partial c^{k}}\frac{1}{\sqrt{c}}\bigg|_{c=1}
\end{align}
so that by computing the differentials consecutively
\begin{equation}
	\mathbb{E}\left(\hat{x}^{2k}\right)
	=\left(-2s^{2}\right)^{k}\left(-\frac{1}{2}\right)\left(-\frac{3}{2}\right)\dots\left(-\frac{2k-1}{2}\right)
	=s^{2k}(2k-1)!!
	=s^{2k}\frac{(2k)!}{2^kk!}
\end{equation}
The double factorial here represents the product of the odd positive integers equal to or less than its argument.

The odd moments $\mathbb{E}\left(\hat{x}^{2k+1}\right)$ are zero by the symmetry $x\to-x$ of the average.

\subsubsection{Sum of Gaussian random variables}
Let $s$ and $t$ be non-zero real numbers and $\hat{x}$ and $\hat{y}$ be independent normal random variables with mean zero and variances $s^{2}$ and $t^{2}$ respectively.  The random variable $\hat{z}=\hat{x}+\hat{y}$ then has a normal distribution with mean zero and variance of $s^{2}+t^{2}$.  This is seen by considering the characteristic functions of both $\hat{x}$ and $\hat{y}$ which read
\begin{align}
	\mathbb{E}\left(\e^{\im\xi\hat{x}}\right)&=\frac{1}{\sqrt{2\pi s^{2}}}\int_{-\infty}^\infty\e^{\im \xi x}\e^{-\frac{x^{2}}{2s^{2}}}\di x = \e^{-\frac{s^{2}\xi^2}{2}}\nonumber\\ 
	\mathbb{E}\left(\e^{\im\xi\hat{y}}\right)&=\frac{1}{\sqrt{2\pi t^{2}}}\int_{-\infty}^\infty\e^{\im \xi y}\e^{-\frac{y^{2}}{2t^{2}}}\di y = \e^{-\frac{t^{2}\xi^2}{2}}
\end{align}
where the integrals are computed as above.  As $\hat{x}$ and $\hat{y}$ are independent, the characteristic function of $\hat{z}=\hat{x}+\hat{y}$ is given by
\begin{equation}
	\mathbb{E}\left(\e^{\im\xi\hat{z}}\right)
	=\mathbb{E}\left(\e^{\im\xi\left(\hat{x}+\hat{y}\right)}\right)
	=\mathbb{E}\left(\e^{\im\xi\hat{x}}\right)\mathbb{E}\left(\e^{\im\xi\hat{y}}\right)
	=\e^{-\frac{(t^{2}+s^{2})\xi^2}{2}}
\end{equation}
This is the characteristic function of a normal random variable with zero mean and variance of $s^{2}+t^{2}$.  The characteristic function gives the Fourier transform of the probability density, therefore by inverting the Fourier transform, $\hat{z}$ must be a normal random variable with zero mean and variance of $s^{2}+t^{2}$.

A further identity now follows directly.  If $\hat{\alpha}_1,\dots,\hat{\alpha}_{9n}\sim\mathcal{N}\left(0,\frac{1}{9n}\right)$ are independent random variables then $\sum_{j=1}^{9n}\hat{\alpha}_j=\hat{\alpha}\sim\mathcal{N}\left(0,1\right)$ and therefore
\begin{equation}
 	\mathbb{E}\left(\left(\sum_{j=1}^{9n}\hat{\alpha}_j\right)^{2k}\right)
	=\mathbb{E}\left(\hat{\alpha}^{2k}\right)
	=\frac{(2k)!}{2^kk!}
\end{equation}
for all $k,n\in\mathbb{N}$.

\section{The Cauchy-Schwarz inequality}
The following is based on \cite[\p70]{Allan2011}.  Let $(\cdot,\cdot)$ denote a complex inner-product on some space $V$, that is a function from $V\times V\to\mathbb{C}$ which satisfies
\begin{align}\label{innprod}
	\text{Conjugate symmetry: }\quad&(A,B)=\overline{(B,A)}\nonumber\\
	\text{Linearity: }\quad&(zA+wB,C)=z(A,C)+w(B,C)\nonumber\\
	\text{Positivity: }\quad&(A,A)\geq0\text{ with }(A,A)=0\iff A=0
\end{align}
for any $A,B\in V$ and $z,w\in\mathbb{C}$ \cite[\p70]{Allan2011}.  The Cauchy-Schwarz inequality states that $|(A,B)|^{2}\leq(A,A)(B,B)$.  This is seen (in a standard way) by considering the following positive quadratic in the real variable $\lambda$
\begin{equation}
	0\leq(A-\lambda B,A-\lambda B)=(A,A)-2\lambda\Rp(A,B)+\lambda^{2}(B,B)
\end{equation}
The discriminate of this equation must be non-positive so that
\begin{equation}
	\left(\Rp(A,B)\right)^{2}\leq(A,A)(B,B)
\end{equation}
Replacing $B$ by $(A,B)B$ then results in the Cauchy-Schwarz inequality stated above multiplied by $|(A,B)|^{2}$.  If $(A,B)=0$ the Cauchy-Schwarz inequality is trivially satisfied, else it can be retrieved by dividing this last expression through by $|(A,B)|^{2}$.

\section{Matrix norms}\label{Norms}
The following is based on \cite[\p76]{Chuang2000}.  The Hilbert-Schmidt inner-product on the space of $N\times N$ complex matrices is defined as
\begin{equation}
	(A,B)_{HS}=\Tr\left(A B^\dagger\right)
\end{equation}
for any two $N\times N$ complex matrices $A$ and $B$.  It satisfies the axioms of an inner-product listed in (\ref{innprod}) by the standard properties of the trace and the conjugate transpose of a complex matrix.  Therefore the scaled Hilbert-Schmidt inner-product
\begin{equation}
	(A,B)=\frac{1}{N}\Tr\left(A B^\dagger\right)
\end{equation}
is also an inner-product on the space of $N\times N$ complex matrices.  This scaled Hilbert-Schmidt inner-product can  be used to define the norm
\begin{equation}
	\|A\|=\sqrt{(A,A)}
\end{equation}
which, by construction, is seen to satisfies the axioms of a matrix norm
\begin{align}
	\text{Triangle inequality: }\quad&\|A+B\|\leq\|A\|+\|B\|\nonumber\\
	\text{Linearity: }\quad&\|zA\|=|z|\|A\|\nonumber\\
	\text{Positivity: }\quad&\|A\|\geq0\text{ with }\|A\|=0\iff A=0
\end{align}
for any $z\in\mathbb{C}$ \cite[\p34]{Allan2011}.  The non-trivial property here is the triangle inequality.  This follows since $\|A+B\|^{2}=(A+B,A+B)\leq\|A\|^{2}+\|B\|^{2}+2\|A\|\|B\|=(\|A\|+\|B\|)^{2}$ by the Cauchy-Schwarz inequality which states that $|(A,B)|\leq\|A\|\|B\|$.

\section{Lebesgue measure on matrices}\label{Lebesgue}
Any $N\times N$ complex matrix $A$ may be parametrised by the $2N^{2}$ real parameters $\alpha_{j,k}$ and $\alpha_{j,k}^\prime$ for $1\leq j,k\leq N$ where
\begin{equation}
  A_{j,k}=\alpha_{j,k}+\im \alpha_{j,k}^\prime
\end{equation}
Therefore the space of $N\times N$ complex matrices may be identified with $\mathbb{R}^{2N^{2}}$.  It is then natural to extend the Lebesgue measure on $\mathbb{R}^{2N^{2}}$ to the space of matrices as
\begin{equation}
  \boldsymbol{\di}A=\prod_{j,k=1}^{N}\di \alpha_{j,k}\di \alpha_{j,k}^\prime
\end{equation}

For any $N\times N$ unitary matrix $U$ it will now be shown that this measure is invariant under the left and right multiplication of $A$ by $U$, up to a sign.  Let the elements of $U$ be $U_{j,k}=u_{j,k}+\im u_{j,k}^\prime$ for real parameters $u_{j,k}$ and $u_{j,k}^\prime$.  The matrix $B=UA$ has the elements
\begin{equation}
  B_{j,k}\equiv \beta_{j,k}+\im \beta_{j,k}^\prime=\sum_{l=1}^{N}\left(u_{j,l}\alpha_{l,k}-u_{j,l}^\prime \alpha_{l,k}^\prime\right)+\im\sum_{l=1}^{N}\left(u_{j,l}\alpha_{l,k}^\prime+u_{j,l}^\prime \alpha_{l,k}\right)
\end{equation}
where $\beta_{j,k}$ and $\beta_{j,k}^\prime$ are the corresponding real parameters of $B$.  Therefore
\begin{equation}\label{appLebTrans}
  \begin{pmatrix}\boldsymbol{\beta}\\\boldsymbol{\beta}^\prime\end{pmatrix}
  =\begin{pmatrix}\bigoplus_{N} u&-\bigoplus_{N} u^\prime\\\bigoplus_{N} u^\prime&\bigoplus_{N} u\end{pmatrix}
  \begin{pmatrix}\boldsymbol{\alpha}\\\boldsymbol{\alpha}^\prime\end{pmatrix}
\end{equation}
where, $\boldsymbol{\alpha}$ is the column vector with the entries $\alpha_{1,1},\alpha_{2,1},\dots,\alpha_{N,N}$, $\boldsymbol{\alpha}^\prime$ is the column vector with the entries $\alpha^\prime_{1,1},\alpha^\prime_{2,1},\dots,\alpha^\prime_{N,N}$, $\boldsymbol{\beta}$ and $\boldsymbol{\beta}^\prime$ are defined similarly, $u=\big(u_{j,k}\big)$ and $u^\prime=\big(u^\prime_{j,k}\big)$ and where $\bigoplus_{N} u$ represents the $N^{2}\times N^{2}$ matrix with $N$ copies of $u$ along its diagonal with zero elsewhere and  $\bigoplus_{N} u^\prime$ represents the $N^{2}\times N^{2}$ matrix with $N$ copies of $u^\prime$ along its diagonal with zero elsewhere.

Now as $U=u+\im u^\prime$ is unitary
\begin{equation}
  I=U^\dagger U=\Big(u^{T}-\im {u^\prime}^{T}\Big)\Big(u+\im u^\prime\Big)
  =\Big(u^{T}u+{u^\prime}^{T}u^\prime\Big)+\im\Big(u^{T}u^\prime-{u^\prime}^{T}u\Big)
\end{equation}
so that it follows that
\begin{equation}
  \begin{pmatrix}\bigoplus_{N} u&-\bigoplus_{N} u^\prime\\\bigoplus_{N} u^\prime&\bigoplus_{N} u\end{pmatrix}^{T}
  \begin{pmatrix}\bigoplus_{N} u&-\bigoplus_{N} u^\prime\\\bigoplus_{N} u^\prime&\bigoplus_{N} u\end{pmatrix}
  =\begin{pmatrix}\bigoplus_{N} \Big(u^{T}u+{u^\prime}^{T}u^\prime\Big)&\bigoplus_{N} \Big({u^\prime}^{T}u-u^{T}u^\prime\Big)\\\bigoplus_{N} \Big(u^{T}u^\prime-{u^\prime}^{T}u\Big)&\bigoplus_{N} \Big({u^\prime}^{T}u^\prime+u^{T}u\Big)\end{pmatrix}
\end{equation}
is equal to the identity matrix.  A similar result holds for the reverse ordered products.

This implies that the transform $(\boldsymbol{\alpha},\boldsymbol{\alpha}^\prime)\to(\boldsymbol{\beta},\boldsymbol{\beta}^\prime)$ in equation (\ref{appLebTrans}) is an orthogonal transform which necessarily has a Jacobian of
\begin{equation}\label{jac}
	{\det}^N\begin{pmatrix}u&-u^\prime\\u^\prime&u\end{pmatrix}=\pm1
\end{equation}
and therefore $\boldsymbol{\di}A=\pm\boldsymbol{\di}\left(UA\right)$.  An analogous results holds for left multiplication, so it is concluded that $\boldsymbol{\di}A=\pm\boldsymbol{\di}\left(UAU^\dagger\right)$ for any unitary matrix $U$.

In particular this holds on the subspace of $2^n\times 2^n$ Hermitian matrices $H$, parametrised by the $4^{n}$ real parameters $\{h_{j,j},h_{k,l},h_{k,l}^\prime\}$ for $j=1,\dots,2^n$ and $1\leq k<l\leq 2^n$ where
\begin{align}
  h_{j,j}&=H_{j,j}\nonumber\\
  h_{k,l}&=\Rp H_{k,l}=\Rp H_{l,k}\nonumber\\
  h_{k,l}^\prime&=\Ip H_{k,l}=-\Ip H_{l,k}
\end{align}
As $2^n$ is even, the sign of the Jacobian is positive in (\ref{jac}) so that
\begin{equation}
  \boldsymbol{\di}\left(UHU^\dagger\right)=\boldsymbol{\di}H=\left(\prod_{j=1}^{2^n}\di h_{j,j}\right)\prod_{1\leq k<l\leq N}\di h_{k,l}\di h_{k,l}^\prime
\end{equation}
\chapter{The Jordan-Wigner transformation}\label{JWT}
The Jordan-Wigner transform \cite{Nielsen2005} maps the matrices $\sigma_{  j  }^{(1)}$, $\sigma_{  j  }^{(2)}$ and $\sigma_{  j  }^{(3)}$ for $j=1,\dots,n$, acting on $\left(\mathbb{C}^{2}\right)^{\otimes n}$, to $n$-state Fermi annihilation and creation matrices $\a_{j}$ and $\a_{j}^\dagger$ for $j=1,\dots,n$ which satisfy the canonical commutation relations
\begin{align}\label{CCR}
  \a_{j} \a_{k}^\dagger&=-\a_{k}^\dagger \a_{j}+\delta_{j,k}I\nonumber\\
  \a_{j}\a_{k}&=-\a_{k}\a_{j}
\end{align}
characterising generic Fermi operators.  A third commutation relation, $\a_{k}^\dagger \a_{j}^\dagger =-\a_{j}^\dagger \a_{k}^\dagger$ is deduced by taking the complex conjugate transpose of the second canonical commutation relation.

An overview of this transformation and its application to nearest-neighbour qubit chains is given in \cite{Nielsen2005}.  A summary is now given for reference:

\section{The transform}\label{JWDef}
Let
\begin{equation}
  S_{j}=\frac{\sigma_{  j  }^{(1)}+\im\sigma_{  j  }^{(2)}}{2}\qquad\qquad\text{with}\qquad\qquad
  S_{j}^\dagger=\frac{\sigma_{  j  }^{(1)}-\im\sigma_{  j  }^{(2)}}{2}
\end{equation}
The Jordan-Wigner transform is then defined to be the mapping
\begin{align}
  \a_{j}=\left(\prod_{1\leq l<j}\sigma_l^{(3)}\right)S_{j}\qquad\qquad\text{with}\qquad\qquad
  \a_{j}^\dagger=\left(\prod_{1\leq l<j}\sigma_l^{(3)}\right)S_{j}^\dagger
\end{align}
for $j=1,\dots,n$.

The canonical commutation relations (\ref{CCR}) can be verified by direct calculation.  First, if $j\neq k$ then $S_{j} S_{k}^\dagger=S_{k}^\dagger S_{j}$ by the definition of $S_{j}$ and $S_{k}^\dagger$.  Also $\sigma_{  j  }^{(3)}$ commutes with both $S_{k}$ and $S_{k}^\dagger$ if $j\neq k$ and anti-commutes if $j=k$, as the matrices $\sigma^{(1)}$, $\sigma^{(2)}$ and $\sigma^{(3)}$ all pairwise anti-commute.  Therefore if $j<k$
\begin{align}
  \a_{j} \a_{k}^\dagger
  &=\left(\sum_{1\leq l<j}\sigma_{  l  }^{(3)}\right)S_{j}\left(\sum_{1\leq m<k}\sigma_{  m  }^{(3)}\right)S_{k}^\dagger\nonumber\\
  &=\left(\sum_{1\leq l<j}\sigma_{  l  }^{(3)}\right)S_{j} S_{k}^\dagger\left(\sum_{1\leq m<k}\sigma_{  m  }^{(3)}\right)\nonumber\\
  &=\left(\sum_{1\leq l<j}\sigma_{  l  }^{(3)}\right)S_{k}^\dagger S_{j}\left(\sum_{j\leq l^\prime<k}\sigma_{  l^\prime  }^{(3)}\right)\left(\sum_{1\leq m<j}\sigma_{  m  }^{(3)}\right)\nonumber\\
  &=-\left(\sum_{1\leq l<k}\sigma_{  l  }^{(3)}\right)S_{k}^\dagger S_{j}\left(\sum_{1\leq m<j}\sigma_{  m  }^{(3)}\right)\nonumber\\
  &=-\a_{k}^\dagger \a_{j}
\end{align}
Then for $j>k$, $\a_{j} \a_{k}^\dagger=(\a_{k} \a_{j}^\dagger)^\dagger=(-\a_{j}^\dagger \a_{k})^\dagger=-\a_{k}^\dagger \a_{j}$.  Whereas for $j=k$,
\begin{align}
	S_{j} S_{j}^\dagger
	&=\frac{1}{4}\left(\sigma_{  j  }^{(1)}+\im\sigma_{  j  }^{(2)}\right)\left(\sigma_{  j  }^{(1)}-\im\sigma_{  j  }^{(2)}\right)\nonumber\\
	&=\frac{1}{4}\left(\sigma_{  j  }^{(1)}\sigma_{  j  }^{(1)}+\sigma_{  j  }^{(2)}\sigma_{  j  }^{(2)}-\im\sigma_{  j  }^{(1)}\sigma_{  j  }^{(2)}+\im\sigma_{  j  }^{(2)}\sigma_{  j  }^{(1)}\right)\nonumber\\
	&=\frac{1}{4}\left(I+I+\im\sigma_{  j  }^{(2)}\sigma_{  j  }^{(1)}-\im\sigma_{  j  }^{(1)}\sigma_{  j  }^{(2)}\right)\nonumber\\
	&=I-\frac{1}{4}\left(\sigma_{  j  }^{(1)}-\im\sigma_{  j  }^{(2)}\right)\left(\sigma_{  j  }^{(1)}+\im\sigma_{  j  }^{(2)}\right)\nonumber\\
	&=I-S_{j}^\dagger S_{j}
\end{align}
so that
\begin{align}
  \a_{j} \a_{j}^\dagger
  &=\left(\sum_{1\leq l<j}\sigma_{  l  }^{(3)}\right)S_{j}\left(\sum_{1\leq m<j}\sigma_{  m  }^{(3)}\right)S_{j}^\dagger\nonumber\\
  &=\left(\sum_{1\leq l<j}\sigma_{  l  }^{(3)}\right)\left(\sum_{1\leq m<j}\sigma_{  m  }^{(3)}\right)S_{j} S_{j}^\dagger\nonumber\\
  &=\left(\sum_{1\leq l<j}\sigma_{  l  }^{(3)}\sigma_{  l  }^{(3)}\right)\left(I-S_{j}^\dagger S_{j}\right)\nonumber\\
  &=I-\left(\sum_{1\leq l<j}\sigma_{  l  }^{(3)}\right)S_{j}^\dagger\left(\sum_{1\leq m<j}\sigma_{  m  }^{(3)}\right)S_{j}\nonumber\\
  &=I-\a_{j}^\dagger \a_{j}
\end{align}
as $\sigma_{  j  }^{(3)}\sigma_{  j  }^{(3)}=I$ for all $j$.  Hence $\a_{j} \a_{k}^\dagger=-\a_{k}^\dagger \a_{j}+\delta_{j,k}I$ for all $j,k=1,\dots,n$.

The remaining canonical commutation relation is proved similarly.  For all $j$ and $k$, $S_{j}S_{k}=S_{k}S_{j}$ by definition.  Then, for $j<k$,
\begin{align}
	\a_{j}\a_{k}
	&=\left(\sum_{1\leq l<j}\sigma_{  l  }^{(3)}\right)S_{j}\left(\sum_{1\leq m<k}\sigma_{  m  }^{(3)}\right)S_{k}\nonumber\\
	&=\left(\sum_{1\leq l<j}\sigma_{  l  }^{(3)}\right)S_{j}S_{k}\left(\sum_{1\leq m<k}\sigma_{  m  }^{(3)}\right)\nonumber\\
	&=\left(\sum_{1\leq l<j}\sigma_{  l  }^{(3)}\right)S_{k}S_{j}\left(\sum_{j\leq l^\prime<k}\sigma_{  l^\prime  }^{(3)}\right)\left(\sum_{1\leq m<j}\sigma_{  m  }^{(3)}\right)\nonumber\\
	&=-\left(\sum_{1\leq l<k}\sigma_{  l  }^{(3)}\right)S_{k}S_{j}\left(\sum_{1\leq m<j}\sigma_{  m  }^{(3)}\right)\nonumber\\
	&=-\a_{k}\a_{j}
\end{align}
For $j>k$ an analogous procedure holds so that $\a_{j}\a_{k}=-\a_{k}\a_{j}$ for all $j\neq k$.  For $j=k$
\begin{align}
	S_{j} S_{j}
	&=\frac{1}{4}\left(\sigma_{  j  }^{(1)}+\im\sigma_{  j  }^{(2)}\right)\left(\sigma_{  j  }^{(1)}+\im\sigma_{  j  }^{(2)}\right)\nonumber\\
	&=\frac{1}{4}\left(\sigma_{  j  }^{(1)}\sigma_{  j  }^{(1)}-\sigma_{  j  }^{(2)}\sigma_{  j  }^{(2)}+\im\sigma_{  j  }^{(1)}\sigma_{  j  }^{(2)}+\im\sigma_{  j  }^{(2)}\sigma_{  j  }^{(1)}\right)\nonumber\\
	&=\frac{1}{4}\left(I-I+\im\sigma_{  j  }^{(1)}\sigma_{  j  }^{(2)}-\im\sigma_{  j  }^{(1)}\sigma_{  j  }^{(2)}\right)\nonumber\\
	&=0
\end{align}
so that
\begin{align}
	\a_{j}\a_{j}
	&=\left(\sum_{1\leq l<j}\sigma_{  l  }^{(3)}\right)S_{j}\left(\sum_{1\leq m<j}\sigma_{  m  }^{(3)}\right)S_{j}\nonumber\\
	&=\left(\sum_{1\leq l<j}\sigma_{  l  }^{(3)}\right)S_{j}S_{j}\left(\sum_{1\leq m<j}\sigma_{  m  }^{(3)}\right)\nonumber\\
	&=0
\end{align}
and $\a_{j}\a_{k}=-\a_{k}\a_{j}$ for all $j,k=1,\dots,n$.

\section{Nearest-neighbour interactions}\label{JWnn}
Hamiltonians containing terms proportional to $\sigma_{  j  }^{(a)}\sigma_{  j+1  }^{(b)}$ for $a,b=1,2,3$ and $j=1,\dots,n-1$  or $\sigma_j^{(3)}$ for $j=1,\dots,n$ are particularly suitable to analysis using the Jordan-Wigner transform as they are quadratic in the Fermi matrices $\a_{j}$ and $\a_{j}^\dagger$.  To this end, the aforementioned matrices will be evaluated in terms of the Fermi matrices $\a_{j}$ and $\a_{j}^\dagger$ for reference.

By applying the definition of $S_{j}$, $S_{j}^\dagger$, $\a_{j}$ and $\a_{j}^\dagger$ and recalling that $\sigma_{j}^{(3)}\sigma_{j}^{(3)}=I$ for all $j$, it is seen that
\begin{align}
	\sigma_{  j  }^{(1)}
	&=S_{j}+S_{j}^\dagger
	=\left(\prod_{1\leq l<j}\sigma_{  l  }^{(3)}\right)\left(\a_{j}+\a_{j}^\dagger\right)\nonumber\\
	\qquad\qquad \sigma_{  j  }^{(2)}
	&=-\im\left(S_{j}-S_{j}^\dagger\right)
	=-\im\left(\prod_{1\leq l<j}\sigma_{  l  }^{(3)}\right)\left(\a_{j}-\a_{j}^\dagger\right)\nonumber\\
	\text{and}\qquad\qquad \sigma_{  j  }^{(3)}
	&=-\im\sigma_{j}^{(1)}\sigma_{j}^{(2)}\nonumber\\
	&=-\left(\prod_{1\leq l<j}\sigma_{l}^{(3)}\right)\left(\a_{j}+\a_{j}^\dagger\right)\left(\prod_{1\leq m<j}\sigma_{m}^{(3)}\right)\left(\a_{j}-\a_{j}^\dagger\right)\nonumber\\
	&=-\left(\a_{j}+\a_{j}^\dagger\right)\left(\a_{j}-\a_{j}^\dagger\right)\nonumber\\
	&=\a_{j}\a_{j}^\dagger-\a_{j}^\dagger \a_{j}\nonumber\\
	&=2\a_{j}\a_{j}^\dagger-I
\end{align}
Furthermore, by the definition of $\a_{j}$ and $\a_{j}^\dagger$, the canonical commutation relations for Fermi matrices and the results above, it is also seen that for $j=1,\dots,n-1$
\begin{align}
	\sigma_{  j  }^{(1)}\sigma_{  j+1  }^{(1)}
	&=\left(\prod_{1\leq l<j}\sigma_{  l  }^{(3)}\right)\left(\a_{j}+\a_{j}^\dagger\right)\left(\prod_{1\leq m<j+1}\sigma_{  m  }^{(3)}\right)\left(\a_{j+1}+\a_{j+1}^\dagger\right)\nonumber\\
	&=\left(\a_{j}+\a_{j}^\dagger\right)\sigma_{  j  }^{(3)}\left(\a_{j+1}+\a_{j+1}^\dagger\right)\nonumber\\
	&=\left(\a_{j}+\a_{j}^\dagger\right)\left(\a_{j}\a_{j}^\dagger-\a_{j}^\dagger \a_{j}\right)\left(\a_{j+1}+\a_{j+1}^\dagger\right)\nonumber\\
	&=-\left(\a_{j}\a_{j}^\dagger \a_{j}-\a_{j}^\dagger \a_{j}\a_{j}^\dagger\right)\left(\a_{j+1}+\a_{j+1}^\dagger\right)\nonumber\\
	&=-\Big(\a_{j}-\a_{j}^\dagger\Big)\Big(\a_{j+1}+\a_{j+1}^\dagger\Big)
\end{align}
\begin{align}
	\sigma_{  j  }^{(2)}\sigma_{  j+1  }^{(2)}
	&=-\left(\prod_{1\leq l<j}\sigma_{  l  }^{(3)}\right)\left(\a_{j}-\a_{j}^\dagger\right)\left(\prod_{1\leq m<j+1}\sigma_{  m  }^{(3)}\right)\left(\a_{j+1}-\a_{j+1}^\dagger\right)\nonumber\\
	&=-\left(\a_{j}-\a_{j}^\dagger\right)\sigma_{  j  }^{(3)}\left(\a_{j+1}-\a_{j+1}^\dagger\right)\nonumber\\
	&=-\left(\a_{j}-\a_{j}^\dagger\right)\left(\a_{j}\a_{j}^\dagger-\a_{j}^\dagger \a_{j}\right)\left(\a_{j+1}-\a_{j+1}^\dagger\right)\nonumber\\
	&=\left(\a_{j}\a_{j}^\dagger \a_{j}+\a_{j}^\dagger \a_{j}\a_{j}^\dagger\right)\left(\a_{j+1}-\a_{j+1}^\dagger\right)\nonumber\\
	&=\Big(\a_{j}+\a_{j}^\dagger\Big)\Big(\a_{j+1}-\a_{j+1}^\dagger\Big)
\end{align}
\begin{align}
	\sigma_{  j  }^{(1)}\sigma_{  j+1  }^{(2)}
	&=-\im\left(\prod_{1\leq l<j}\sigma_{  l  }^{(3)}\right)\left(\a_{j}+\a_{j}^\dagger\right)\left(\prod_{1\leq m<j+1}\sigma_{  m  }^{(3)}\right)\left(\a_{j+1}-\a_{j+1}^\dagger\right)\nonumber\\
	&=-\im\left(\a_{j}+\a_{j}^\dagger\right)\sigma_{  j  }^{(3)}\left(\a_{j+1}-\a_{j+1}^\dagger\right)\nonumber\\
	&=-\im\left(\a_{j}+\a_{j}^\dagger\right)\left(\a_{j}\a_{j}^\dagger-\a_{j}^\dagger \a_{j}\right)\left(\a_{j+1}-\a_{j+1}^\dagger\right)\nonumber\\
	&=\im\left(\a_{j}\a_{j}^\dagger \a_{j}-\a_{j}^\dagger \a_{j}\a_{j}^\dagger\right)\left(\a_{j+1}-\a_{j+1}^\dagger\right)\nonumber\\
	&=\im\Big(\a_{j}-\a_{j}^\dagger\Big)\Big(\a_{j+1}-\a_{j+1}^\dagger\Big)
\end{align}
and
\begin{align}
	\sigma_{  j  }^{(2)}\sigma_{  j+1  }^{(1)}
	&=-\im\left(\prod_{1\leq l<j}\sigma_{  l  }^{(3)}\right)\left(\a_{j}-\a_{j}^\dagger\right)\left(\prod_{1\leq m<j+1}\sigma_{  m  }^{(3)}\right)\left(\a_{j+1}+\a_{j+1}^\dagger\right)\nonumber\\
	&=-\im\left(\a_{j}-\a_{j}^\dagger\right)\sigma_{  j  }^{(3)}\left(\a_{j+1}+\a_{j+1}^\dagger\right)\nonumber\\
	&=-\im\left(\a_{j}-\a_{j}^\dagger\right)\left(\a_{j}\a_{j}^\dagger-\a_{j}^\dagger \a_{j}\right)\left(\a_{j+1}+\a_{j+1}^\dagger\right)\nonumber\\
	&=\im\left(\a_{j}\a_{j}^\dagger \a_{j}+\a_{j}^\dagger \a_{j}\a_{j}^\dagger\right)\left(\a_{j+1}+\a_{j+1}^\dagger\right)\nonumber\\
	&=\im\Big(\a_{j}+\a_{j}^\dagger\Big)\Big(\a_{j+1}+\a_{j+1}^\dagger\Big)
\end{align}
For cyclic chains the boundary terms are not quadratic in the Fermi matrices $\a_{j}$ and $\a_{j}^\dagger$ but can be written in a quadratic form with the addition of the parity matrix $\eta=\prod_{j=1}^{n}\sigma_{  j  }^{(3)}$, that is,
\begin{align}
	\sigma_{  n  }^{(1)}\sigma_{  1  }^{(1)}
	&=\left(\prod_{1\leq l<n}\sigma_{  l  }^{(3)}\right)\Big(\a_{n}+\a_{n}^\dagger\Big)\left(\a_{1}+\a_{1}^\dagger\right)\nonumber\\
	&=\eta\sigma_{  n  }^{(3)}\Big(\a_{n}+\a_{n}^\dagger\Big)\left(\a_{1}+\a_{1}^\dagger\right)\nonumber\\
	&=\eta\Big(\a_{n}\a_{n}^\dagger-\a_{n}^\dagger \a_{n}\Big)\left(\a_{n}+\a_{n}^\dagger\right)\left(\a_{1}+\a_{1}^\dagger\right)\nonumber\\
	&=\eta\Big(\a_{n}\a_{n}^\dagger \a_{n}-\a_{n}^\dagger \a_{n}\a_{n}^\dagger\Big)\left(\a_{1}+\a_{1}^\dagger\right)\nonumber\\
	&=\eta\Big(\a_{n}-\a_{n}^\dagger\Big)\left(\a_{1}+\a_{1}^\dagger\right)\nonumber\\
\end{align}
\begin{align}
	\sigma_{  n  }^{(2)}\sigma_{  1  }^{(2)}
	&=-\left(\prod_{1\leq l<n}\sigma_{  l  }^{(3)}\right)\Big(\a_{n}-\a_{n}^\dagger\Big)\left(\a_{1}-\a_{1}^\dagger\right)\nonumber\\
	&=-\eta\sigma_{  n  }^{(3)}\Big(\a_{n}-\a_{n}^\dagger\Big)\left(\a_{1}-\a_{1}^\dagger\right)\nonumber\\
	&=-\eta\Big(\a_{n}\a_{n}^\dagger-\a_{n}^\dagger \a_{n}\Big)\left(\a_{n}-\a_{n}^\dagger\right)\left(\a_{1}-\a_{1}^\dagger\right)\nonumber\\
	&=-\eta\Big(\a_{n}\a_{n}^\dagger \a_{n}+\a_{n}^\dagger \a_{n}\a_{n}^\dagger\Big)\left(\a_{1}-\a_{1}^\dagger\right)\nonumber\\
	&=-\eta\Big(\a_{n}+\a_{n}^\dagger\Big)\left(\a_{1}-\a_{1}^\dagger\right)\nonumber\\
\end{align}
\begin{align}
	\sigma_{  n  }^{(1)}\sigma_{  1  }^{(2)}
	&=-\im\left(\prod_{1\leq l<n}\sigma_{  l  }^{(3)}\right)\Big(\a_{n}+\a_{n}^\dagger\Big)\left(\a_{1}-\a_{1}^\dagger\right)\nonumber\\
	&=-\im\eta\sigma_{  n  }^{(3)}\Big(\a_{n}+\a_{n}^\dagger\Big)\left(\a_{1}-\a_{1}^\dagger\right)\nonumber\\
	&=-\im\eta\Big(\a_{n}\a_{n}^\dagger-\a_{n}^\dagger \a_{n}\Big)\left(\a_{n}+\a_{n}^\dagger\right)\left(\a_{1}-\a_{1}^\dagger\right)\nonumber\\
	&=-\im\eta\Big(\a_{n}\a_{n}^\dagger \a_{n}-\a_{n}^\dagger \a_{n}\a_{n}^\dagger\Big)\left(\a_{1}-\a_{1}^\dagger\right)\nonumber\\
	&=-\im\eta\Big(\a_{n}-\a_{n}^\dagger\Big)\left(\a_{1}-\a_{1}^\dagger\right)\nonumber\\
\end{align}
and
\begin{align}
	\sigma_{  n  }^{(2)}\sigma_{  1  }^{(1)}
	&=-\im\left(\prod_{1\leq l<n}\sigma_{  l  }^{(3)}\right)\Big(\a_{n}-\a_{n}^\dagger\Big)\left(\a_{1}+\a_{1}^\dagger\right)\nonumber\\
	&=-\im\eta\sigma_{  n  }^{(3)}\Big(\a_{n}-\a_{n}^\dagger\Big)\left(\a_{1}+\a_{1}^\dagger\right)\nonumber\\
	&=-\im\eta\Big(\a_{n}\a_{n}^\dagger-\a_{n}^\dagger \a_{n}\Big)\left(\a_{n}-\a_{n}^\dagger\right)\left(\a_{1}+\a_{1}^\dagger\right)\nonumber\\
	&=-\im\eta\Big(\a_{n}\a_{n}^\dagger \a_{n}+\a_{n}^\dagger \a_{n}\a_{n}^\dagger\Big)\left(\a_{1}+\a_{1}^\dagger\right)\nonumber\\
	&=-\im\eta\Big(\a_{n}+\a_{n}^\dagger\Big)\left(\a_{1}-\a_{1}^\dagger\right)\nonumber\\
\end{align}
These relations are particularly useful for the calculations in Section \ref{DOSrate}.

\section{An orthonormal Fermi basis}\label{FermiBasis}
The Fermi matrices $\a_{j}$ and $\a_{j}^\dagger$ for $j=1,\dots,n$ (or indeed any operators satisfying the canonical commutation relations (\ref{CCR}) represented as matrices acting on $\left(\mathbb{C}^2\right)^{\otimes n}$) can be used to define an orthonormal basis for $\left(\mathbb{C}^{2}\right)^{\otimes n}$.  There exists \cite[\p3]{Nielsen2005} a normalised state $|\boldsymbol{0}\rangle_{{\a}}\in\left(\mathbb{C}^2\right)^{\otimes n}$ such that $\a_{j}|\boldsymbol{0}\rangle_{{\a}}$ for all $j$.  In the case of the Fermi matrices $\a_{j}$ and $\a_{j}^\dagger$, this is the vector $|\boldsymbol{0}\rangle=|0\rangle^{\otimes n}$ of the standard basis (see Section \ref{stndBasis}) as $(\sigma^{(1)}+\im\sigma^{(2)})|0\rangle=0$.   For the multi-indices $\boldsymbol{x}=(x_{1},\dots,x_{n})\in\{0,1\}^{n}$ let
\begin{equation}
	|\boldsymbol{x}\rangle_{{\a}}=\Big(\a_{1}^\dagger\Big)^{x_{1}}\dots\Big(\a_{n}^\dagger\Big)^{x_{n}}|\boldsymbol{0}\rangle_{{\a}}
\end{equation}
There are $2^{n}$ such states.  Their orthonormality is seen by considering the inner-product
\begin{equation}
	\,_{{\a}}\langle\boldsymbol{x}|\boldsymbol{y}\rangle_{{\a}}
	=\,_\a\langle\boldsymbol{0}|\Big(\a_{n}\Big)^{x_{n}}\dots\Big(\a_{1}\Big)^{x_{1}}\Big(\a_{1}^\dagger\Big)^{y_{1}}\dots\Big(\a_{n}^\dagger\Big)^{y_{n}}|\boldsymbol{0}\rangle_\a
\end{equation}
for $\boldsymbol{x},\boldsymbol{y}\in\{0,1\}^n$, which, by the canonical commutation relations (\ref{CCR}), is equal to
\begin{equation}
	\pm\,_\a\langle\boldsymbol{0}|\Big(\a_{1}\Big)^{x_{1}}\Big(\a_{1}^\dagger\Big)^{y_{1}}\dots\Big(\a_{n}\Big)^{x_{n}}\Big(\a_{n}^\dagger\Big)^{y_{n}}|\boldsymbol{0}\rangle_\a
\end{equation}
Since $\a_{j}\a_{j}^\dagger=I-\a_{j}^\dagger \a_{j}$, $\a_{j}|\boldsymbol{0}\rangle_\a=0$ and $\,_\a\langle\boldsymbol{0}|\a_{j}^\dagger=0$ it then follows that
\begin{equation}
	\,_{{\a}}\langle\boldsymbol{y}|\boldsymbol{x}\rangle_{{\a}}
	=\delta_{x_{1},y_{1}}\dots\delta_{x_{n},y_{n}}\,_\a\langle\boldsymbol{0}|\boldsymbol{0}\rangle_\a
	=\delta_{x_{1},y_{1}}\dots\delta_{x_{n},y_{n}}
\end{equation}
where a positive sign factor is chosen, as $\,_{{\a}}\langle\boldsymbol{x}|\boldsymbol{x}\rangle_{{\a}}\geq0$.  Therefore the $|\boldsymbol{x}\rangle_{{\a}}$ form an orthonormal basis of $\left(\mathbb{C}^{2}\right)^{\otimes n}$.

\section{The action of the Fermi matrices on the Fermi basis}\label{JWFermiAction}
For the multi-indices $\boldsymbol{x}=(x_{1},\dots,x_{n})\in\{0,1\}^{n}$ let $|\boldsymbol{x}\rangle_{{\a}}$ be the orthonormal Fermi basis as defined above, with respect to the Fermi matrices $\a_{j}$ and $\a_{j}^\dagger$ (or indeed any operators satisfying the canonical commutation relations (\ref{CCR}) represented as matrices acting  on $\left(\mathbb{C}^2\right)^{\otimes n}$).  By definition and the canonical commutation relations, the action of the matrix $\a_{j}$ on this basis is given by
\begin{align}
	\a_{j}|\boldsymbol{x}\rangle_{{\a}}
	&=\a_{j}\left(\a_{1}^\dagger\right)^{x_{1}}\dots\Big(\a_{n}^\dagger\Big)^{x_{n}}|\boldsymbol{0}\rangle_{{\a}}\nonumber\\
	&=\left(\prod_{1\leq l<j}\left(-\a_{l}^\dagger\right)^{x_{l}}\right)\a_{j}\left(\a_{j}^\dagger\right)^{x_{j}}\left(\prod_{j<m\leq n}\left(\a_{m}^\dagger\right)^{x_{m}}\right)|\boldsymbol{0}\rangle_{{\a}}\nonumber\\
	&=(-1)^{\sum_{l<j}x_{l}}x_{j}|\boldsymbol{y}\rangle_{{\a}}
\end{align}
where $\boldsymbol{y}$ is the vector constructed from $\boldsymbol{x}$ by replacing the $j^{th}$ entry with the value $0$.  In this sense, $\a_{j}$ can be thought of as the lowing matrix on the $j^{th}$ site, returning zero if a lowering is not possible.

Similarly for $\a_{j}^\dagger$,
\begin{equation}
	\a_{j}^\dagger|\boldsymbol{x}\rangle_{{\a}}=(-1)^{\sum_{l<j}x_{l}}(1-x_{j})|\boldsymbol{z}\rangle_{{\a}}
\end{equation}
where $\boldsymbol{z}$ is the vector constructed from $\boldsymbol{x}$ by replacing the $j^{th}$ entry with the value $1$.

It then follows from the orthonormality of the $|\boldsymbol{x}\rangle_{{\a}}$ that
\begin{align}
	\,_{{\a}}\langle\boldsymbol{x}|\a_{j}|\boldsymbol{x}\rangle_{{\a}}&=0\nonumber\\
	\,_{{\a}}\langle\boldsymbol{x}|\a_{j}^\dagger|\boldsymbol{x}\rangle_{{\a}}&=0\nonumber\\
	\,_{{\a}}\langle\boldsymbol{x}|\a_{j}^\dagger \a_{k}|\boldsymbol{x}\rangle_{{\a}}&=x_{j}\delta_{j,k}\nonumber\\
	\,_{{\a}}\langle\boldsymbol{x}|\a_{j} \a_{k}^\dagger|\boldsymbol{x}\rangle_{{\a}}&=(1-x_{j})\delta_{j,k}
\end{align}
In fact it follows that for any product of an odd number of the Fermi matrices $\a_{j}$ and $\a_{j}^\dagger$ denoted $A$, that
\begin{equation}
	\,_{{\a}}\langle\boldsymbol{x}|A|\boldsymbol{x}\rangle_{{\a}}=0
\end{equation}

\section{Transforms of Fermi annihilation and creation matrices}\label{FermiTrans}
The Fermi annihilation and creation matrices $\a_{j}$ and $\a_{j}^\dagger$ (or indeed any Fermi operators satisfying the canonical commutation relations (\ref{CCR})) can be `rotated' by the linear transformation
\begin{equation}\label{rotation}
	\begin{pmatrix}
		\boldsymbol{\b}\\
		\boldsymbol{\b}^\prime
	\end{pmatrix}
	=
	\begin{pmatrix}
		U&V\\
		\overline{V}&\overline{U}
	\end{pmatrix}
	\begin{pmatrix}
		\boldsymbol{\a}\\
		\boldsymbol{\a}^\prime
	\end{pmatrix}
\end{equation}
where
\begin{equation}
	\boldsymbol{\a}=
	\begin{pmatrix}
		\a_{1}\\
		\vdots\\
		\a_{n}
	\end{pmatrix}\qquad
	\boldsymbol{\a}^\prime=
	\begin{pmatrix}
		\a_{1}^\dagger\\
		\vdots\\
		\a_{n}^\dagger
	\end{pmatrix}\qquad
	\boldsymbol{\b}=
	\begin{pmatrix}
		\b_{1}\\
		\vdots\\
		\b_{n}
	\end{pmatrix}\qquad
	\boldsymbol{\b}^\prime=
	\begin{pmatrix}
		\b_{1}^\dagger\\
		\vdots\\
		\b_{n}^\dagger
	\end{pmatrix}
\end{equation}
and $U=(u_{j,k})$ and $V=(v_{j,k})$ are $n\times n$ complex matrices.  The matrices $\b_{j}^\dagger$ remain the conjugate transpose of the $\b_{j}$ as
\begin{equation}
	\b_{j}^\dagger
	=\sum_{l=1}^{n}\left(\overline{v_{j,l}}\a_{l}+\overline{u_{j,l}}\a_{l}^\dagger\right)
	=\left(\sum_{l=1}^{n}\left(v_{j,l}\a_{l}^\dagger+u_{j,l}\a_{l}\right)\right)^\dagger
	=\left(\b_{j}\right)^\dagger
\end{equation}

In order for the matrices $\b_{j}$ and $\b_{j}^\dagger$ to satisfy analogous canonical commutation relations to (\ref{CCR}), that is to be Fermi matrices, $U$ and $V$ must satisfy certain conditions.  These are deduced by direct computation.  First,
\begin{align}
	\b_{j}\b_{k}^\dagger
	&=\sum_{l=1}^{n}\Big(u_{j,l}\a_{l}+v_{j,l}\a_{l}^\dagger\Big)\sum_{m=1}^{n}\Big(\overline{v_{k,m}}\a_{m}+\overline{u_{k,m}}\a_{m}^\dagger\Big)\nonumber\\
	&=\sum_{l,m=1}^{n}\Big(u_{j,l}\overline{u_{k,m}}\a_{l}\a_{m}^\dagger+v_{j,l}\overline{v_{k,m}}\a_{l}^\dagger \a_{m}+u_{j,l}\overline{v_{k,m}}\a_{l}\a_{m}+v_{j,l}\overline{u_{k,m}}\a_{l}^\dagger \a_{m}^\dagger\Big)
\end{align}
By the canonical commutation relations (\ref{CCR}), the order of the matrices in each term may be reversed to yield
\begin{equation}
	-\sum_{m=1}^{n}\Big(\overline{v_{k,m}}\a_{m}+\overline{u_{k,m}}\a_{m}^\dagger\Big)\sum_{l=1}^{n}\Big(u_{j,l}\a_{l}+v_{j,l}\a_{l}^\dagger\Big)+\sum_{l=1}^{n}\Big(u_{j,l}\overline{u_{k,l}}+v_{j,l}\overline{v_{k,l}}\Big)
\end{equation}
or equally, upon recombining the matrix elements,
\begin{equation}
	-\b_{k}^\dagger \b_{j}+\left(UU^\dagger\right)_{j,k}+\left(VV^\dagger\right)_{j,k}
\end{equation}
Therefore $UU^\dagger+VV^\dagger=I$ if and only of $\b_{j}\b_{k}^\dagger=-\b_{k}^\dagger \b_{j}+I\delta_{j,k}$.

Furthermore,
\begin{align}
	\b_{j}\b_{k}
	&=\sum_{l=1}^{n}\Big(u_{j,l}\a_{l}+v_{j,l}\a_{l}^\dagger\Big)\sum_{m=1}^{n}\Big(u_{k,m}\a_{m}+v_{k,m}\a_{m}^\dagger\Big)\nonumber\\
	&=\sum_{l,m=1}^{n}\Big(u_{j,l}u_{k,m}\a_{l}\a_{m}+v_{j,l}v_{k,m}\a_{l}^\dagger \a_{m}^\dagger+u_{j,l}v_{k,m}\a_{l}\a_{m}^\dagger+v_{j,l}u_{k,m}\a_{l}^\dagger \a_{m}\Big)
\end{align}
in which, by the canonical commutation relations (\ref{CCR}), the order of the matrices in each term may be reversed to yield
\begin{equation}
	-\sum_{m=1}^{n}\Big(u_{k,m}\a_{m}+v_{k,m}\a_{m}^\dagger\Big)\sum_{l=1}^{n}\Big(u_{j,l}\a_{l}+v_{j,l}\a_{l}^\dagger\Big)+\sum_{l=1}^{n}\Big(u_{j,l}v_{k,l}+v_{j,l}u_{k,l}\Big)
\end{equation}
Recombining the matrix elements gives the equivalent expression of
\begin{equation}
	-\b_{k}\b_{j}+\left(UV^{T}\right)_{j,k}+\left(VU^{T}\right)_{j,k}
\end{equation}
Therefore $UV^{T}+VU^{T}=0$ if and only if $\b_{j}\b_{k}=-\b_{k}\b_{j}$.

These requirements on $U$ and $V$ are equivalent to the matrix in (\ref{rotation}) being unitary.  That is, requiring that
\begin{equation}
	\begin{pmatrix}
		U&V\\
		\overline{V}&\overline{U}
	\end{pmatrix}
	\begin{pmatrix}
		U^\dagger&\overline{V}^\dagger\\
		V^\dagger&\overline{U}^\dagger
	\end{pmatrix}
	=
	\begin{pmatrix}
		UU^\dagger+VV^\dagger & UV^{T}+VU^{T}\\
		\overline{VU^{T}+UV^{T}} & \overline{VV^\dagger+UU^\dagger}
	\end{pmatrix}
	=
	\begin{pmatrix}
		I&0\\
		0&I
	\end{pmatrix}
\end{equation}

\chapter{Collated numerical results}\label{collNumerics}
This appendix contains the graphical results of all the numerical simulations undertaken for Sections \ref{Numerics}, \ref{eigenNum} and \ref{nnnumerics}:

\section{Spectral histograms}\label{specHist}
Figures \ref{HistSpecFirst} to \ref{HistSpecLast} display the normalised spectral histograms calculated from $s=2^{19-n}$ random samples of each of the following ensembles for $n=2,\dots,13$.   Full details of the numerical procedure and discussion of the results can be seen in Section \ref{Numerics}:
\begin{alignat}{3}
	\hat{H}_{n}&=\sum_{j=1}^{n}\sum_{a,b=1}^{3}\hat{\alpha}_{a,b,j}\sigma_{  j  }^{(a)}\sigma_{  j+1  }^{(b)}\qquad\qquad&&\hat{\alpha}_{a,b,j}\sim\mathcal{N}\left(0,\frac{1}{9n}\right) \iid\nonumber\\
	\hat{H}^{(uniform)}_{n}&=\sum_{j=1}^{n}\sum_{a,b=1}^{3}\hat{\alpha}_{a,b,j}\sigma_{  j  }^{(a)}\sigma_{  j+1  }^{(b)}\qquad\qquad&&\hat{\alpha}_{a,b,j}\sim\mathcal{U}\left(-\frac{\sqrt{3}}{\sqrt{9n}},\frac{\sqrt{3}}{\sqrt{9n}}\right) \iid\nonumber\\
	\hat{H}^{(local)}_{n}&=\sum_{j=1}^{n}\sum_{a=1}^{3}\sum_{b=0}^{3}\hat{\alpha}_{a,b,j}\sigma_{  j  }^{(a)}\sigma_{  j+1  }^{(b)}\qquad\qquad&&\hat{\alpha}_{a,b,j}\sim\mathcal{N}\left(0,\frac{1}{12n}\right) \iid\nonumber\\
	\hat{H}^{(inv)}_{n}&=\sum_{j=1}^{n}\sum_{a,b=1}^{3}\hat{\alpha}_{a,b}\sigma_{  j  }^{(a)}\sigma_{  j+1  }^{(b)}\qquad\qquad&&\hat{\alpha}_{a,b}\sim\mathcal{N}\left(0,\frac{1}{9n}\right) \iid\nonumber\\
	\hat{H}^{(inv,local)}_{n}&=\sum_{j=1}^{n}\sum_{a=1}^{3}\sum_{b=0}^{3}\hat{\alpha}_{a,b}\sigma_{  j  }^{(a)}\sigma_{  j+1  }^{(b)}\qquad\qquad&&\hat{\alpha}_{a,b}\sim\mathcal{N}\left(0,\frac{1}{12n}\right) \iid\nonumber\\
	\hat{H}^{(JW)}_{n}&=\sum_{j=1}^{n-1}\sum_{a,b=1}^{2}\hat{\alpha}_{a,b,j}\sigma_{  j  }^{(a)}\sigma_{  j+1  }^{(b)}+\sum_{j=1}^{n}\alpha_{3,0,j}\sigma_{j}^{(3)}\qquad&&\hat{\alpha}_{a,b,j}\sim\mathcal{N}\left(0,\frac{1}{5n-4}\right) \iid\nonumber\\
	\hat{H}^{(Heis)}_{n}&=\sum_{j=1}^{n}\sum_{a=1}^{3}\hat{\alpha}_{a,a,j}\sigma_{  j  }^{(a)}\sigma_{  j+1  }^{(a)}\qquad\qquad\qquad\qquad&&\hat{\alpha}_{a,a,j}\sim\mathcal{N}\left(0,\frac{1}{3n}\right) \iid
\end{alignat}

\section{Linear entropy of reduced eigenstates}\label{LinEntGraphs}
Figures \ref{EntFirst} to \ref{EntLast} show the average value of the linear entropy, $1-\Tr\left(\rho_{l,k}^{2}\right)$, over $s=2^{19-n}$ samples from each of the ensembles above, in turn, where $\rho_{l,k}$ is the reduced density matrix, on the the qubits labelled $1$ to $l=1,\dots,5$, of the eigenstate $\rho_{k}$ corresponding to the numerically ordered eigenvalue $\lambda_{k}$, for each sample.

On $l$ qubits the minimal value, zero, of the linear entropy is equivalent to there being a product state between the $l$ qubits and the rest of the chain.  The maximal value of $1-\frac{1}{2^{l}}$ corresponds to there being a maximally entangled state between the $l$ qubits and the rest of the chain (see Section \ref{EigenDef}).

The value of the purity is unique when the spectrum of the corresponding Hamiltonian is simple.  If this is not the case then there is some arbitrary choice of eigenstates within an eigensubspace and the numeral procedure selects just one possible choice for this.  The ensembles for which all samples were seen to have a non-degenerate spectrum (no eigenvalue spacing less than $10^{-10}$) are shown in Table \ref{degenbig}.

Full details of the numerical procedure and discussion of the results can be found in Section \ref{eigenNum}.

\begin{table}[p]\footnotesize
\centering
\begin{tabular}{l|ccccccccccccccc}
  \toprule
	\multirow{1}{*}{\bf{Ensemble}}&\multicolumn{12}{ c }{\bf{Number of qubits, $\boldsymbol{n}$}}\\
	 &$2$&$3$&$4$&$5$&$6$&$7$&$8$&$9$&$10$&$11$&$12$&$13$\\
	\midrule
	$\hat{H}_{n}$&$\checkmark$&$-$&$\checkmark$&$-$&$\checkmark$&$-$&$\checkmark$&$-$&$\checkmark$&$-$&$\checkmark$&$-$\\
	$\hat{H}_{n}^{(uniform)}$&$\checkmark$&$-$&$\checkmark$&$-$&$\checkmark$&$-$&$\checkmark$&$-$&$\checkmark$&$-$&$\checkmark$&$-$\\
	$\hat{H}_{n}^{(local)}$&$\checkmark$&$\checkmark$&$\checkmark$&$\checkmark$&$\checkmark$&$\checkmark$&$\checkmark$&$\checkmark$&$\checkmark$&$\checkmark$&$\checkmark$&$\checkmark$\\
	$\hat{H}_{n}^{(inv)}$&$\checkmark$&$-$&$-$&$-$&$-$&$-$&$-$&$-$&$-$&$-$&$-$&$-$\\	
	$\hat{H}_{n}^{(inv, local)}$&$\checkmark$&$\checkmark$&$-$&$\checkmark$&$\checkmark$&$\checkmark$&$\checkmark$&$\checkmark$&$\checkmark$&$\checkmark$&$\checkmark$&$\checkmark$\\
	$\hat{H}_{n}^{(JW)}$&$\checkmark$&$\checkmark$&$\checkmark$&$\checkmark$&$\checkmark$&$\checkmark$&$\checkmark$&$\checkmark$&$\checkmark$&$\checkmark$&$\checkmark$&$\checkmark$\\	
	$\hat{H}_{n}^{(Heis)}$&$\checkmark$&$-$&$-$&$-$&$-$&$-$&$-$&$-$&$-$&$-$&$-$&$-$\\	
	\bottomrule
\end{tabular}
\caption[Complete non-degenerate instances of ensembles]{Table showing the ensembles for which all samples taken possessed a non-degenerate spectrum (no absolute spacing less than $10^{-10}$).  Such ensembles are marked with `$\checkmark$'.  Ensembles where this was not the case are marked with `$-$'.}
\label{degenbig}
\end{table}

\section{Nearest-neighbour level spacings}\label{collNumericsSpacing}
Figures \ref{SpacingFirst} to \ref{SpacingLast} show the normalised nearest-neighbour level spacing histograms for the ensembles above.  The spacings used were numerically obtained from each of the unfolded spectra of $s=2^{19-n}$ samples from the ensembles above, in turn, and scaled to have unit mean.  The unfolding was with respect to the average numerical spectral density on the interval $[-3,3]$ calculated from the same sample.  

Full details of the numerical methodology and discussion of the results are given in Section \ref{nnnumerics}.  Limitations in this methodology are perhaps seen in the approximated nearest-neighbour level spacing distributions for the ensembles $\hat{H}_{n}^{(inv)}$ and $\hat{H}_{n}^{(inv,local)}$ for $n=4$.  The step in the approximate spacing distribution seen may be due to the fact that the numerical ensemble spectral density, from which the spectral unfolding is performed, has a large spike (potentially a singularly) at zero.  This strongly affects the unfolding in this region and could distort the numerical unfolding results strongly.

Similar, but less pronounced, spikes in the ensemble spectral density histograms are also seen for $\hat{H}_{n}^{(inv)}$ for $n=6,8$ and $\hat{H}_{n}^{(Heis)}$ for $n=4$ in particular.  These could again affect the numerical unfolding.  On the whole though, the remaining ensembles have relatively smooth spectral histograms, allowing a reasonable confidence in the numerical unfolding process.

\clearpage
\cleartoleftpage
\rhead[]{Spectral histograms}

\def\model{1}
\def\name{\hat{H}_{n}}
\begin{figure}[ht]
	\centering
	\input{Figures/Density2-7}
	\caption[Spectral histograms for $\name$  ]{The normalised spectral histograms for $\name$.  The spectra used were numerically obtained from each of $s=2^{19-n}$ random samples of $\name$.  The dashed lines give the standard Gaussian density.}
	\label{HistSpecFirst}
\end{figure}
\begin{figure}[ht]
	\centering
	\input{Figures/Density8-13}
	\caption[]{The normalised spectral histograms for $\name$.  The spectra used were numerically obtained from each of $s=2^{19-n}$ random samples of $\name$.  The dashed lines give the standard Gaussian density.}
\end{figure}

\def\model{7}
\def\name{\hat{H}_{n}^{(uniform)}}
\begin{figure}[ht]
	\centering
	\input{Figures/Density2-7}
	\caption[Spectral histograms for $\name$  ]{The normalised spectral histograms for $\name$.  The spectra used were numerically obtained from each of $s=2^{19-n}$ random samples of $\name$.  The dashed lines give the standard Gaussian density.}
\end{figure}
\begin{figure}[ht]
	\centering
	\input{Figures/Density8-13}
	\caption[]{The normalised spectral histograms for $\name$.  The spectra used were numerically obtained from each of $s=2^{19-n}$ random samples of $\name$.  The dashed lines give the standard Gaussian density.}
\end{figure}

\def\model{101}
\def\name{\hat{H}_{n}^{(local)}}
\begin{figure}[ht]
	\centering
	\input{Figures/Density2-7}
	\caption[Spectral histograms for $\name$  ]{The normalised spectral histograms for $\name$.  The spectra used were numerically obtained from each of $s=2^{19-n}$ random samples of $\name$.  The dashed lines give the standard Gaussian density.}
\end{figure}
\begin{figure}[ht]
	\centering
	\input{Figures/Density8-13}
	\caption[]{The normalised spectral histograms for $\name$.  The spectra used were numerically obtained from each of $s=2^{19-n}$ random samples of $\name$.  The dashed lines give the standard Gaussian density.}
\end{figure}

\def\model{2}
\def\name{\hat{H}_{n}^{(inv)}}
\begin{figure}[ht]
	\centering
	\input{Figures/Density2-7}
	\caption[Spectral histograms for $\name$  ]{The normalised spectral histograms for $\name$.  The spectra used were numerically obtained from each of $s=2^{19-n}$ random samples of $\name$.  The dashed lines give the standard Gaussian density.}
\end{figure}
\begin{figure}[ht]
	\centering
	\input{Figures/Density8-13}
	\caption[]{The normalised spectral histograms for $\name$.  The spectra used were numerically obtained from each of $s=2^{19-n}$ random samples of $\name$.  The dashed lines give the standard Gaussian density.}
\end{figure}

\def\model{102}
\def\name{\hat{H}_{n}^{(inv, local)}}
\begin{figure}[ht]
	\centering
	\input{Figures/Density2-7}
	\caption[Spectral histograms for $\name$  ]{The normalised spectral histograms for $\name$.  The spectra used were numerically obtained from each of $s=2^{19-n}$ random samples of $\name$.  The dashed lines give the standard Gaussian density.}
\end{figure}
\begin{figure}[ht]
	\centering
	\input{Figures/Density8-13}
	\caption[]{The normalised spectral histograms for $\name$.  The spectra used were numerically obtained from each of $s=2^{19-n}$ random samples of $\name$.  The dashed lines give the standard Gaussian density.}
\end{figure}

\def\model{8}
\def\name{\hat{H}_{n}^{(JW)}}
\begin{figure}[ht]
	\centering
	\input{Figures/Density2-7}
	\caption[Spectral histograms for $\name$  ]{The normalised spectral histograms for $\name$.  The spectra used were numerically obtained from each of $s=2^{19-n}$ random samples of $\name$.  The dashed lines give the standard Gaussian density.}
\end{figure}
\begin{figure}[ht]
	\centering
	\input{Figures/Density8-13}
	\caption[]{The normalised spectral histograms for $\name$.  The spectra used were numerically obtained from each of $s=2^{19-n}$ random samples of $\name$.  The dashed lines give the standard Gaussian density.}
\end{figure}

\def\model{4}
\def\name{\hat{H}_{n}^{(Heis)}}
\begin{figure}[ht]
	\centering
	\input{Figures/Density2-7}
	\caption[Spectral histograms for $\name$  ]{The normalised spectral histograms for $\name$.  The spectra used were numerically obtained from each of $s=2^{19-n}$ random samples of $\name$.  The dashed lines give the standard Gaussian density.}
\end{figure}
\begin{figure}[ht]
	\centering
	\input{Figures/Density8-13}
	\caption[]{The normalised spectral histograms for $\name$.  The spectra used were numerically obtained from each of $s=2^{19-n}$ random samples of $\name$.  The dashed lines give the standard Gaussian density.}
	\label{HistSpecLast}
\end{figure}

\clearpage
\rhead[]{Linear entropy of reduced eigenstates}

\def\model{1}
\def\name{\hat{H}_{n}}
\begin{figure}
	\centering
	\input{Figures/Purity2-7}
	\caption[Average linear entropy of the reduced eigenstates over $\name$  ]{The average values of the linear entropy, $1-\Tr\left(\rho_{l,k}^{2}\right)$, over $s=2^{19-n}$ samples from $\name$, where $\rho_{l,k}$ is the reduced density matrix, on the qubits labelled $1$ to $l$, of the eigenstate $\rho_{k}$ corresponding to the numerically order eigenvalue $\lambda_{k}$, for each sample.  The points for each value of $l$ (marked with symbols for $n\leq5$) have been joined for clarity and the dashed lines are at the maximal linear entropy values of, $1-\frac{1}{2^{l}}$.  The labels for the values of $l$ in each graph correspond to the solid lines from top to bottom (with respect to the centre of the spectrum).}
	\label{EntFirst}
\end{figure}
\begin{figure}
	\centering
	\input{Figures/Purity8-13}
	\caption[]{The average values of the linear entropy, $1-\Tr\left(\rho_{l,k}^{2}\right)$, over $s=2^{19-n}$ samples from $\name$, where $\rho_{l,k}$ is the reduced density matrix, on the qubits labelled $1$ to $l$, of the eigenstate $\rho_{k}$ corresponding to the numerically order eigenvalue $\lambda_{k}$, for each sample.  The points for each value of $l$ have been joined for clarity and the dashed lines are at the maximal linear entropy values of, $1-\frac{1}{2^{l}}$.  The labels for the values of $l$ in each graph correspond to the solid lines from top to bottom (with respect to the centre of the spectrum).}
\end{figure}

\def\model{7}
\def\name{\hat{H}_{n}^{(uniform)}}
\begin{figure}
	\centering
	\input{Figures/Purity2-7}
	\caption[Average linear entropy of the reduced eigenstates over $\name$  ]{The average values of the linear entropy, $1-\Tr\left(\rho_{l,k}^{2}\right)$, over $s=2^{19-n}$ samples from $\name$, where $\rho_{l,k}$ is the reduced density matrix, on the qubits labelled $1$ to $l$, of the eigenstate $\rho_{k}$ corresponding to the numerically order eigenvalue $\lambda_{k}$, for each sample.  The points for each value of $l$ (marked with symbols for $n\leq5$) have been joined for clarity and the dashed lines are at the maximal linear entropy values of, $1-\frac{1}{2^{l}}$.  The labels for the values of $l$ in each graph correspond to the solid lines from top to bottom (with respect to the centre of the spectrum).}
\end{figure}
\begin{figure}
	\centering
	\input{Figures/Purity8-13}
	\caption[]{The average values of the linear entropy, $1-\Tr\left(\rho_{l,k}^{2}\right)$, over $s=2^{19-n}$ samples from $\name$, where $\rho_{l,k}$ is the reduced density matrix, on the qubits labelled $1$ to $l$, of the eigenstate $\rho_{k}$ corresponding to the numerically order eigenvalue $\lambda_{k}$, for each sample.  The points for each value of $l$ have been joined for clarity and the dashed lines are at the maximal linear entropy values of, $1-\frac{1}{2^{l}}$.  The labels for the values of $l$ in each graph correspond to the solid lines from top to bottom (with respect to the centre of the spectrum).}
\end{figure}

\def\model{101}
\def\name{\hat{H}_{n}^{(local)}}
\begin{figure}
	\centering
	\input{Figures/Purity2-7}
	\caption[Average linear entropy of the reduced eigenstates over $\name$  ]{The average values of the linear entropy, $1-\Tr\left(\rho_{l,k}^{2}\right)$, over $s=2^{19-n}$ samples from $\name$, where $\rho_{l,k}$ is the reduced density matrix, on the qubits labelled $1$ to $l$, of the eigenstate $\rho_{k}$ corresponding to the numerically order eigenvalue $\lambda_{k}$, for each sample.  The points for each value of $l$ (marked with symbols for $n\leq5$) have been joined for clarity and the dashed lines are at the maximal linear entropy values of, $1-\frac{1}{2^{l}}$.  The labels for the values of $l$ in each graph correspond to the solid lines from top to bottom (with respect to the centre of the spectrum).}
\end{figure}
\begin{figure}
	\centering
	\input{Figures/Purity8-13}
	\caption[]{The average values of the linear entropy, $1-\Tr\left(\rho_{l,k}^{2}\right)$, over $s=2^{19-n}$ samples from $\name$, where $\rho_{l,k}$ is the reduced density matrix, on the qubits labelled $1$ to $l$, of the eigenstate $\rho_{k}$ corresponding to the numerically order eigenvalue $\lambda_{k}$, for each sample.  The points for each value of $l$ have been joined for clarity and the dashed lines are at the maximal linear entropy values of, $1-\frac{1}{2^{l}}$.  The labels for the values of $l$ in each graph correspond to the solid lines from top to bottom (with respect to the centre of the spectrum).}
\end{figure}

\def\model{2}
\def\name{\hat{H}_{n}^{(inv)}}
\begin{figure}
	\centering
	\input{Figures/Purity2-7}
	\caption[Average linear entropy of the reduced eigenstates over $\name$  ]{The average values of the linear entropy, $1-\Tr\left(\rho_{l,k}^{2}\right)$, over $s=2^{19-n}$ samples from $\name$, where $\rho_{l,k}$ is the reduced density matrix, on the qubits labelled $1$ to $l$, of the eigenstate $\rho_{k}$ corresponding to the numerically order eigenvalue $\lambda_{k}$, for each sample.  The points for each value of $l$ (marked with symbols for $n\leq5$) have been joined for clarity and the dashed lines are at the maximal linear entropy values of, $1-\frac{1}{2^{l}}$.  The labels for the values of $l$ in each graph correspond to the solid lines from top to bottom (with respect to the centre of the spectrum).}
\end{figure}
\begin{figure}
	\centering
	\input{Figures/Purity8-13}
	\caption[]{The average values of the linear entropy, $1-\Tr\left(\rho_{l,k}^{2}\right)$, over $s=2^{19-n}$ samples from $\name$, where $\rho_{l,k}$ is the reduced density matrix, on the qubits labelled $1$ to $l$, of the eigenstate $\rho_{k}$ corresponding to the numerically order eigenvalue $\lambda_{k}$, for each sample.  The points for each value of $l$ have been joined for clarity and the dashed lines are at the maximal linear entropy values of, $1-\frac{1}{2^{l}}$.  The labels for the values of $l$ in each graph correspond to the solid lines from top to bottom (with respect to the centre of the spectrum).}
\end{figure}

\def\model{102}
\def\name{\hat{H}_{n}^{(inv,local)}}
\begin{figure}
	\centering
	\input{Figures/Purity2-7}
	\caption[Average linear entropy of the reduced eigenstates over $\name$  ]{The average values of the linear entropy, $1-\Tr\left(\rho_{l,k}^{2}\right)$, over $s=2^{19-n}$ samples from $\name$, where $\rho_{l,k}$ is the reduced density matrix, on the qubits labelled $1$ to $l$, of the eigenstate $\rho_{k}$ corresponding to the numerically order eigenvalue $\lambda_{k}$, for each sample.  The points for each value of $l$ (marked with symbols for $n\leq5$) have been joined for clarity and the dashed lines are at the maximal linear entropy values of, $1-\frac{1}{2^{l}}$.  The labels for the values of $l$ in each graph correspond to the solid lines from top to bottom (with respect to the centre of the spectrum).}
\end{figure}
\begin{figure}
	\centering
	\input{Figures/Purity8-13}
	\caption[]{The average values of the linear entropy, $1-\Tr\left(\rho_{l,k}^{2}\right)$, over $s=2^{19-n}$ samples from $\name$, where $\rho_{l,k}$ is the reduced density matrix, on the qubits labelled $1$ to $l$, of the eigenstate $\rho_{k}$ corresponding to the numerically order eigenvalue $\lambda_{k}$, for each sample.  The points for each value of $l$ have been joined for clarity and the dashed lines are at the maximal linear entropy values of, $1-\frac{1}{2^{l}}$.  The labels for the values of $l$ in each graph correspond to the solid lines from top to bottom (with respect to the centre of the spectrum).}
\end{figure}

\def\model{8}
\def\name{\hat{H}_{n}^{(JW)}}
\begin{figure}
	\centering
	\input{Figures/Purity2-7}
	\caption[Average linear entropy of the reduced eigenstates over $\name$  ]{The average values of the linear entropy, $1-\Tr\left(\rho_{l,k}^{2}\right)$, over $s=2^{19-n}$ samples from $\name$, where $\rho_{l,k}$ is the reduced density matrix, on the qubits labelled $1$ to $l$, of the eigenstate $\rho_{k}$ corresponding to the numerically order eigenvalue $\lambda_{k}$, for each sample.  The points for each value of $l$ (marked with symbols for $n\leq5$) have been joined for clarity and the dashed lines are at the maximal linear entropy values of, $1-\frac{1}{2^{l}}$.  The labels for the values of $l$ in each graph correspond to the solid lines from top to bottom (with respect to the centre of the spectrum).}
\end{figure}
\begin{figure}
	\centering
	\input{Figures/Purity8-13}
	\caption[]{The average values of the linear entropy, $1-\Tr\left(\rho_{l,k}^{2}\right)$, over $s=2^{19-n}$ samples from $\name$, where $\rho_{l,k}$ is the reduced density matrix, on the qubits labelled $1$ to $l$, of the eigenstate $\rho_{k}$ corresponding to the numerically order eigenvalue $\lambda_{k}$, for each sample.  The points for each value of $l$ have been joined for clarity and the dashed lines are at the maximal linear entropy values of, $1-\frac{1}{2^{l}}$.  The labels for the values of $l$ in each graph correspond to the solid lines from top to bottom (with respect to the centre of the spectrum).}
\end{figure}

\def\model{4}
\def\name{\hat{H}_{n}^{(Heis)}}
\begin{figure}
	\centering
	\input{Figures/Purity2-7}
	\caption[Average linear entropy of the reduced eigenstates over $\name$  ]{The average values of the linear entropy, $1-\Tr\left(\rho_{l,k}^{2}\right)$, over $s=2^{19-n}$ samples from $\name$, where $\rho_{l,k}$ is the reduced density matrix, on the qubits labelled $1$ to $l$, of the eigenstate $\rho_{k}$ corresponding to the numerically order eigenvalue $\lambda_{k}$, for each sample.  The points for each value of $l$ (marked with symbols for $n\leq5$) have been joined for clarity and the dashed lines are at the maximal linear entropy values of, $1-\frac{1}{2^{l}}$.  The labels for the values of $l$ in each graph correspond to the solid lines from top to bottom (with respect to the centre of the spectrum).}
\end{figure}
\begin{figure}
	\centering
	\input{Figures/Purity8-13}
	\caption[]{The average values of the linear entropy, $1-\Tr\left(\rho_{l,k}^{2}\right)$, over $s=2^{19-n}$ samples from $\name$, where $\rho_{l,k}$ is the reduced density matrix, on the qubits labelled $1$ to $l$, of the eigenstate $\rho_{k}$ corresponding to the numerically order eigenvalue $\lambda_{k}$, for each sample.  The points for each value of $l$ have been joined for clarity and the dashed lines are at the maximal linear entropy values of, $1-\frac{1}{2^{l}}$.  The labels for the values of $l$ in each graph correspond to the solid lines from top to bottom (with respect to the centre of the spectrum).}
	\label{EntLast}
\end{figure}

\clearpage
\rhead[]{Nearest-neighbour level spacings}

\def\model{1}
\def\name{\hat{H}_{n}}
\begin{figure}
	\centering
	\input{Figures/Spacing2-7}
	\caption[Nearest-neighbour level spacing histograms for $\name$  ]{The normalised nearest-neighbour level spacing histograms for $\name$.  The spacings used were numerically obtained from each of the unfolded spectra of $s=2^{19-n}$ samples from the ensemble $\name$ and scaled to have unit mean.  The unfolding was with respect to the average numerical spectral density on the interval $[-3,3]$ calculated from the same sample.  The dashed lines give the approximate limiting standard GOE nearest-neighbour level spacing distribution for even $n$ and the approximate limiting GSE nearest-neighbour level spacing distribution scaled to have mean $2$ and area $\frac{1}{2}$ for odd $n$.}
	\label{SpacingFirst}
\end{figure}
\begin{figure}
	\centering
	\input{Figures/Spacing8-13}
	\caption[]{The normalised nearest-neighbour level spacing histograms for $\name$.  The spacings used were numerically obtained from each of the unfolded spectra of $s=2^{19-n}$ samples from the ensemble $\name$ and scaled to have unit mean.  The unfolding was with respect to the average numerical spectral density on the interval $[-3,3]$ calculated from the same sample.  The dashed lines give the approximate limiting standard GOE nearest-neighbour level spacing distribution for even $n$ and the approximate limiting GSE nearest-neighbour level spacing distribution scaled to have mean $2$ and area $\frac{1}{2}$ for odd $n$.}
\end{figure}

\def\model{7}
\def\name{\hat{H}_{n}^{(uniform)}}
\begin{figure}
	\centering
	\input{Figures/Spacing2-7}
	\caption[Nearest-neighbour level spacing histograms for $\name$  ]{The normalised nearest-neighbour level spacing histograms for $\name$.  The spacings used were numerically obtained from each of the unfolded spectra of $s=2^{19-n}$ samples from the ensemble $\name$ and scaled to have unit mean.  The unfolding was with respect to the average numerical spectral density on the interval $[-3,3]$ calculated from the same sample.  The dashed lines give the approximate limiting standard GOE nearest-neighbour level spacing distribution for even $n$ and the approximate limiting GSE nearest-neighbour level spacing distribution scaled to have mean $2$ and area $\frac{1}{2}$ for odd $n$.}
\end{figure}
\begin{figure}
	\centering
	\input{Figures/Spacing8-13}
	\caption[]{The normalised nearest-neighbour level spacing histograms for $\name$.  The spacings used were numerically obtained from each of the unfolded spectra of $s=2^{19-n}$ samples from the ensemble $\name$ and scaled to have unit mean.  The unfolding was with respect to the average numerical spectral density on the interval $[-3,3]$ calculated from the same sample.  The dashed lines give the approximate limiting standard GOE nearest-neighbour level spacing distribution for even $n$ and the approximate limiting GSE nearest-neighbour level spacing distribution scaled to have mean $2$ and area $\frac{1}{2}$ for odd $n$.}
\end{figure}

\def\model{101}
\def\name{\hat{H}_{n}^{(local)}}
\begin{figure}
	\centering
	\input{Figures/Spacing2-7}
	\caption[Nearest-neighbour level spacing histograms for $\name$  ]{The normalised nearest-neighbour level spacing histograms for $\name$.  The spacings used were numerically obtained from each of the unfolded spectra of $s=2^{19-n}$ samples from the ensemble $\name$ and scaled to have unit mean.  The unfolding was with respect to the average numerical spectral density on the interval $[-3,3]$ calculated from the same sample.  The dashed lines give the approximate limiting standard GUE nearest-neighbour level spacing distribution.}
\end{figure}
\begin{figure}
	\centering
	\input{Figures/Spacing8-13}
	\caption[]{The normalised nearest-neighbour level spacing histograms for $\name$.  The spacings used were numerically obtained from each of the unfolded spectra of $s=2^{19-n}$ samples from the ensemble $\name$ and scaled to have unit mean.  The unfolding was with respect to the average numerical spectral density on the interval $[-3,3]$ calculated from the same sample.  The dashed lines give the approximate limiting standard GUE nearest-neighbour level spacing distribution.}
\end{figure}

\def\model{2}
\def\name{\hat{H}_{n}^{(inv)}}
\begin{figure}
	\centering
	\input{Figures/Spacing2-7}
	\caption[Nearest-neighbour level spacing histograms for $\name$  ]{The normalised nearest-neighbour level spacing histograms for $\name$.  The spacings used were numerically obtained from each of the unfolded spectra of $s=2^{19-n}$ samples from the ensemble $\name$ and scaled to have unit mean.  The unfolding was with respect to the average numerical spectral density on the interval $[-3,3]$ calculated from the same sample.}
\end{figure}
\begin{figure}
	\centering
	\input{Figures/Spacing8-13}
	\caption[]{The normalised nearest-neighbour level spacing histograms for $\name$.  The spacings used were numerically obtained from each of the unfolded spectra of $s=2^{19-n}$ samples from the ensemble $\name$ and scaled to have unit mean.  The unfolding was with respect to the average numerical spectral density on the interval $[-3,3]$ calculated from the same sample.}
\end{figure}

\def\model{102}
\def\name{\hat{H}_{n}^{(inv,local)}}
\begin{figure}
	\centering
	\input{Figures/Spacing2-7}
	\caption[Nearest-neighbour level spacing histograms for $\name$  ]{The normalised nearest-neighbour level spacing histograms for $\name$.  The spacings used were numerically obtained from each of the unfolded spectra of $s=2^{19-n}$ samples from the ensemble $\name$ and scaled to have unit mean.  The unfolding was with respect to the average numerical spectral density on the interval $[-3,3]$ calculated from the same sample.  The dashed lines give the standard Poisson spacing distribution.}
\end{figure}
\begin{figure}
	\centering
	\input{Figures/Spacing8-13}
	\caption[]{The normalised nearest-neighbour level spacing histograms for $\name$.  The spacings used were numerically obtained from each of the unfolded spectra of $s=2^{19-n}$ samples from the ensemble $\name$ and scaled to have unit mean.  The unfolding was with respect to the average numerical spectral density on the interval $[-3,3]$ calculated from the same sample.  The dashed lines give the standard Poisson spacing distribution.}
\end{figure}

\def\model{8}
\def\name{\hat{H}_{n}^{(JW)}}
\begin{figure}
	\centering
	\input{Figures/Spacing2-7}
	\caption[Nearest-neighbour level spacing histograms for $\name$  ]{The normalised nearest-neighbour level spacing histograms for $\name$.  The spacings used were numerically obtained from each of the unfolded spectra of $s=2^{19-n}$ samples from the ensemble $\name$ and scaled to have unit mean.  The unfolding was with respect to the average numerical spectral density on the interval $[-3,3]$ calculated from the same sample.  The dashed lines give the standard Poisson spacing distribution.}
\end{figure}
\begin{figure}
	\centering
	\input{Figures/Spacing8-13}
	\caption[]{The normalised nearest-neighbour level spacing histograms for $\name$.  The spacings used were numerically obtained from each of the unfolded spectra of $s=2^{19-n}$ samples from the ensemble $\name$ and scaled to have unit mean.  The unfolding was with respect to the average numerical spectral density on the interval $[-3,3]$ calculated from the same sample.  The dashed lines give the standard Poisson spacing distribution.}
\end{figure}

\def\model{4}
\def\name{\hat{H}_{n}^{(Heis)}}
\begin{figure}
	\centering
	\input{Figures/Spacing2-7}
	\caption[Nearest-neighbour level spacing histograms for $\name$  ]{The normalised nearest-neighbour level spacing histograms for $\name$.  The spacings used were numerically obtained from each of the unfolded spectra of $s=2^{19-n}$ samples from the ensemble $\name$ and scaled to have unit mean.  The unfolding was with respect to the average numerical spectral density on the interval $[-3,3]$ calculated from the same sample.}
\end{figure}
\begin{figure}
	\centering
	\input{Figures/Spacing8-13}
	\caption[]{The normalised nearest-neighbour level spacing histograms for $\name$.  The spacings used were numerically obtained from each of the unfolded spectra of $s=2^{19-n}$ samples from the ensemble $\name$ and scaled to have unit mean.  The unfolding was with respect to the average numerical spectral density on the interval $[-3,3]$ calculated from the same sample.}
	\label{SpacingLast}
\end{figure}

\footnotesize
\cleardoublepage
\rhead[]{\rightmark}
\addcontentsline{toc}{chapter}{Bibliography}
\bibliographystyle{unsrt}
\bibliography{SpinChains}

\end{document}